\documentclass[11pt]{article}
\usepackage{titling}
\usepackage{booktabs}
\usepackage{array}
\usepackage{enumitem}
\usepackage{tikz}
\usepackage{graphicx}
\usetikzlibrary{positioning, arrows.meta, calc, fit}
\usepackage{natbib}
\usepackage{calc}
\usepackage{subcaption}
\usepackage[a4paper, margin=1in]{geometry}
\usepackage{rotating}
\usepackage{mathrsfs}
\usepackage{fancyhdr}
\usepackage{amsfonts}
\usepackage{amstext}
\usepackage{amsmath}
\usepackage{amssymb}
\usepackage{algorithm}
\usepackage{algorithmic}
\usepackage{adjustbox}
\usepackage{hyperref}
\usepackage{bookmark}
\usepackage{setspace}
\usepackage{longtable}
\usepackage{float}
\usepackage{threeparttable}
\usepackage{fullpage}
\usepackage{times}
\usepackage[utf8]{inputenc}
\usepackage{units}

\usepackage{nicefrac}
\usepackage{color}
\usepackage{tabularx}
\usepackage{pdflscape}

\newtheorem{theorem}{Theorem}
\newtheorem{assumption}{Assumption}
\newtheorem{corollary}{Corollary}
\newtheorem{lemma}{Lemma}
\newtheorem{remark}{Remark}
\newtheorem{proposition}{Proposition}

\newtheorem{condition}{Condition}
\newtheorem{fact}{Fact}
\newenvironment{proof}[1][Proof]{\noindent \textbf{#1.} }{\  \rule{0.5em}{0.5em}}
\newcommand{\kerop}{\operatorname{ker}}

\newcommand{\HK}{\mathcal{H}_k}
\pretitle{\begin{center}\large\bfseries}
\posttitle{\end{center}}
\preauthor{\begin{center}
\large \lineskip 0.5em%
\begin{tabular}[t]{c}}
\postauthor{\end{tabular}\par\end{center}}
\predate{\begin{center}\large}
\postdate{\par\end{center}\vskip 0.5em}

\newcommand{\HRule}{\rule{\linewidth}{0.5mm}}
\newcommand{\BRule}{\rule{\linewidth}{0.25mm}}
\title{\HRule \\[0.1mm]
Data-Driven Tuning in Kernel Specification Tests: Selection Asymptotics, Validity, and Power \thanks{I acknowledge support from the Research Development Fund (RDF-23-02-022) of the Xi'an Jiaotong-Liverpool University. I thank valuable comments from Migule A. Delgado, Juan Carlos Escanciano, Xiaojun Song and seminar participants at the The 7th International Symposia on Nonparametric Statistics (ISNPS 2026). All errors are my own. }
\\[0.1mm]
\BRule}

\author{Yuhao Li\thanks{yuhao.li@xjtlu.edu.cn} \\ \small Xi'an Jiaotong-Liverpool University}

\date{} 

\begin{document}

\maketitle

\begin{abstract}
We develop model checks for conditional moment restrictions whose tuning parameters, the kernel bandwidth and the regularization level, are selected by cross-validation. The tests regress the estimated residuals on the covariates via kernel ridge regression (KRR), yielding a coefficient function in a reproducing kernel Hilbert space (RKHS) that vanishes under the null and is nonzero for any misspecification not annihilated by the Neyman projection. Our central contribution is a complete asymptotic theory of the data-driven tuning parameter: its limit law under the null, under fixed alternatives, and under local alternatives, and bootstrap size control under sample splitting with no convergence of the selector required. An exact energy identity shows that cross-validation maximizes the signal captured by the test (its noncentrality) and, absent signal, shrinks its noise floor. The tests are consistent against fixed alternatives, achieve $n^{-1/2}$-rate sensitivity to local alternatives, and use Neyman-orthogonal projections so the model is never re-estimated within bootstrap resamples. Simulations show size close to nominal without user-chosen tuning, power comparable to the kernel conditional moment (KCM) test at its best kernel parameter, and robustness where fixed kernel parameters fail in higher dimensions. Three empirical applications illustrate the diagnostic role of the tests; we recommend a $15$--$20\%$ selection share for the sample-splitting protocol.
\end{abstract}

\vspace{1em} 

\newpage \pagenumbering{arabic} \setcounter{page}{1} \setcounter{footnote}{0}

\section{Introduction}
\label{sec:intro}

Many models in the social sciences are characterized by conditional moment restrictions of the form
\begin{align*}
    & Y = \mathcal{M}_{\theta_0}(X) + \varepsilon_0, \\
    & \mathbb{E}(\varepsilon_0 \mid X) = 0 \quad \text{almost surely (a.s.), for some unknown } \theta_0 \in \Theta \subset \mathbb{R}^{d_\theta},
\end{align*}
where the random vectors $S_i = (Y_i, X_i)$, $i = 1, 2, \dots$, are i.i.d.\ with probability measure $\mathbb{P}_S$, $Y \in \mathcal{Y} \subset \mathbb{R}^q$, $X \in \mathcal{X} \subset \mathbb{R}^d$, and the residual function $\varepsilon: \mathcal{S} \times \Theta \to \mathbb{R}^q$ is known up to the finite-dimensional parameter $\theta_0 \in \Theta$. We focus on the univariate case $q = 1$; extension to $q > 1$ is straightforward, since the statistics aggregate over residual components (Remark~\ref{remark:qgeq1}). This framework encompasses linear and nonlinear conditional mean regression, quantile regression, treatment effect models, and instrumental variables regression.

The objective is to test
\[
    H_0: \mathbb{P}\left(\mathbb{E}[\varepsilon(S;\theta_0) \mid X] = 0\right)=1 \text{ for some } \theta_0\in \Theta
\]
against the alternative
\[
    H_1: \mathbb{P}\left(\mathbb{E}[\varepsilon(S;\theta)\mid X] \neq 0\right) >0 \text{ for all } \theta\in \Theta.
\]
A natural starting point is to regress the estimated residuals on the covariates, $\varepsilon_{\hat \theta} = X^\top \omega + e$, and test $\omega = \boldsymbol{0}$ against $\omega \neq \boldsymbol{0}$. This appealing idea has two limitations: the test is inconsistent, detecting only deviations representable by the chosen finite set of moments, and the estimated parameter $\hat \theta$ enters the null distribution, leaving the statistic non-pivotal. Modern goodness-of-fit techniques transcend it, often at the cost of complexity that hinders implementation and interpretation.

\paragraph{The tuning problem.}
Kernel-based omnibus tests, in particular the classical kernel conditional moment (KCM) test, require a bandwidth $\gamma$ and, in regularized formulations, a smoothing parameter $\lambda$, to which size and power are highly sensitive. Standard practice fixes the bandwidth a priori, typically $\gamma = 0.5$ for a Gaussian kernel \citep{bierens1990consistent}, or delegates it to a data-driven median heuristic \citep{escanciano2024gaussian}. The Gaussian-process framework of \cite{escanciano2024gaussian} unifies these kernel model checks in a single class of Neyman orthogonal kernel tests and holds the kernel fixed throughout; selecting the kernel-side tuning from data is posed there as an open problem (Section 7 of that paper). Both conventions are known failure modes in higher-dimensional settings: at $d = 20$, the fixed default drives the size of the KCM test to essentially zero, while the median rule leaves it with almost no power. The tuning parameters are therefore not a technical detail but a first-order component of the testing procedure, and a specification test intended for use without expert intervention must treat them as such. This paper develops a framework for kernel-based CMR specification testing in which the tuning parameters are selected by cross-validation and the selection itself is part of the theory: we derive the asymptotic law of the data-driven selector, and we show that the resulting bootstrap test controls size without any convergence of the selector.

\paragraph{Related work.}
Among practitioners, the most widely used specification tests are M-tests, introduced by \cite{newey1985generalized,newey1985maximum,tauchen1985diagnostic,wooldridge1990unified}. These tests assess a finite set of unconditional moment restrictions implied by the model and are inherently ``directional'', probing only a limited set of directions in the space of alternatives. By contrast, omnibus tests are consistent against any form of misspecification, which matters when the nature of potential model failure is unknown. Two approaches exist. One compares the fitted parametric model to a nonparametric regression estimate, typically smoothed, in \cite{eubank1990testing,hardle1993comparing,hong1995consistent,zheng1996consistent,ellison2000simple}; the other applies integral transforms of the residuals rather than the residuals themselves, giving rise to the Integrated Conditional Moment (ICM) tests of \cite{bierens1982consistent,bierens1997asymptotic,stute1997nonparametric,delgado2006consistent}. For a comprehensive review, see \cite{gonzalez2013updated}.

\paragraph{The KRR approach.}
In this paper we return to the regression idea and push it to its natural infinite-dimensional limit. Instead of a finite set of moments, we map the covariates into a reproducing kernel Hilbert space (RKHS) $\HK$ with reproducing kernel $k(\cdot,\cdot)$ and consider the regularized population coefficient function
\[
    w^*_\lambda = (C + \lambda I)^{-1}\,\mathbb{E}\left[\varepsilon_0 k(X,\cdot)\right],
\]
where $C = \mathbb{E}[k(X,\cdot) \otimes k(X,\cdot)]$ is the second-moment operator on $\HK$, $\lambda > 0$ is a fixed regularization parameter, and $\varepsilon_0 := \varepsilon(S;\theta_0)$; the empirical analogue of $w^*_\lambda$ is the kernel ridge regression (KRR) estimator. We show in Section~\ref{sec:stat} that the null hypothesis holds if and only if $w^*_\lambda = \boldsymbol{0}$ in $\HK$: for a universal kernel, the RKHS is dense in the space of continuous functions \citep{micchelli2006universal}, so $w^*_\lambda$ approximates the conditional mean of the residual, which vanishes exactly under $H_0$. The function $w^*_\lambda$ is thus the infinite-dimensional analogue of the regression coefficient $\omega$, and hence an interpretable diagnostic of model adequacy.

\paragraph{Contribution 1: projection test statistics.}
Our first contribution is a class of test statistics based on projecting $w^*_\lambda$ onto functions in $\HK$. We consider two projections: (i) onto $w^*_\lambda$ itself, which measures the overall magnitude of the misspecification through the RKHS norm of the fit; and (ii) onto the mean embedding of the residual, $m^* = \mathbb{E}(\varepsilon_0 k(X,\cdot))$, the cross-product of the fit with the raw residual signal. The two statistics share a common structure and differ only in their spectral weights: deviations along the eigen-directions of the kernel operator are weighted by $\mu_i/(\mu_i+\lambda)$ or $\mu_i/(\mu_i+\lambda)^2$, where $\{\mu_i\}$ are the eigenvalues of the integral operator, instead of by the eigenvalues themselves as in the classical KCM statistic; these weights define the null and local limit laws of Section~\ref{sec:stat}, and of the two the second lies closer to the classical form. The tests detect local (Pitman) alternatives at the parametric $n^{-1/2}$ rate.

Our statistics can also be viewed as a modification of the RKHS-based tests of \cite{muandet2020kernel,escanciano2024gaussian}: classical ICM statistics are V- or U-statistics with a Gaussian kernel at a fixed bandwidth, typically $0.5$; \cite{muandet2020kernel} allow general kernel choices but do not account for the estimation effect of nuisance parameters, and \cite{escanciano2024gaussian} restore validity through a Gaussian process approach with a Neyman orthogonal kernel; within that class, the present paper adds the data-driven selection of the kernel-side tuning, the object that paper poses as an open problem (Remark~\ref{rem:gp_class} and Section~\ref{sec:nuisance}). Our approach is further related to the orthogonal series regression of \cite{eubank1990testing}, who test the significance of the first $p_n$ series coefficients of an orthogonal basis expansion, with $p_n$ growing with the sample size. Their test detects local alternatives only at the rate $O(p_n^{1/4}/n^{1/2})$, slower than the $n^{-1/2}$ rate of our tests; the price is that our null distributions are non-pivotal, which necessitates the multiplier bootstrap. The RKHS formulation itself is connected to the integral transformation approach: the RKHS is isometrically isomorphic to the space of square-integrable functions \citep{carrasco2007linear}, and for shift-invariant kernels the connection is made explicit by Bochner's theorem (Theorem 2.1 of \cite{muandet2020kernel}).

\paragraph{Estimation effects.}
The statistics are completed by an estimation-effect correction. When the residuals are computed from an estimated parameter $\hat\theta$, the null distributions of the statistics are altered by the estimation effect. Following \cite{escanciano2014specification,escanciano2024gaussian,sant2019specification}, we eliminate the first-order estimation effect through a Neyman orthogonality projection onto the orthogonal complement of the conditional score. The remainder is asymptotically negligible, and the orthogonalized statistics inherit the spectral form of their infeasible counterparts, with the projected null laws of Corollary~\ref{cor:proj_null} (Section~\ref{sec:estimation_effect}). Critical values come from the multiplier bootstrap, whose conditional consistency we establish.

\paragraph{Contribution 2: the tuning problem.}
The statistics depend on two tuning parameters, the bandwidth $\gamma$ of the kernel and the regularization parameter $\lambda$, which act as nuisance parameters for the testing problem. The central contribution of this paper is a complete theory of the tuning problem. The pair $(\gamma,\lambda)$ is selected jointly by $k$-fold cross-validation on a separate selection sample and held fixed on the testing sample, under the protocol of Section~\ref{sec:nuisance}. We derive the asymptotic law of the selector in the three regimes relevant to the testing problem: under the null it converges to a random limit supported on the implemented tuning grid, so the selection is noise-driven and must not be read as an estimate of an optimal value; under a fixed alternative it is pinned down at the signal-maximizing configuration; under local alternatives it converges to a drifted random limit whose drift depends on the split between the selection and the testing samples. We further show that the bootstrap test controls size under this protocol without requiring any convergence of the selector at all. The key observation is that conditioning on the realized split makes the selected tuning a fixed grid point for the multiplier bootstrap, so size control needs only pointwise consistency of the bootstrap at that fixed tuning. The analysis also clarifies why the regularization parameter must be bounded away from zero. In contrast to the classical bandwidth-selection analysis for specification tests, which approximates the power function to second order \citep{gao2008bandwidth} or selects the smoothing adaptively without a limit theory for the selection itself \citep{horowitz2001adaptive}, the selector is here given a first-order theory in its own right. The projection family thus supplies the missing tuning layer for the kernel members of the Gaussian-process class of \cite{escanciano2024gaussian}, whose data-driven tuning that paper poses as an open problem; Remark~\ref{rem:gp_class} makes the relation precise, and the resolution is not the routine extension that the derivative-interchange route of their Remark 2 suggests, because the selector's criterion fluctuates at the degenerate $m^{-1}$ scale and retains a non-degenerate random limit.

\paragraph{Contribution 3: power and the selection mechanism.}
The literature on power properties of conditional moment goodness-of-fit tests is largely diagnostic: omnibus tests are shown to have substantial local power only against alternatives lying in a finite-dimensional space, generally unknown to the researcher \citep{escanciano2009lack}; power is documented to deteriorate with the dimension of the conditioning variables, the curse of dimensionality \citep{lavergne2008breaking,escanciano2006consistent}; minimax bounds delimit the alternatives that consistent tests can detect \cite{guerrelavergne2002optimal}; and the smoothing parameter is chosen by second-order approximations to the power function \cite{gao2008bandwidth} or by adaptive selection without a limit theory for the choice itself \cite{horowitz2001adaptive}. Power-enhancement techniques developed for high-dimensional tests \cite{fan2015power} combine test statistics rather than analyze tuning, and the kernel-choice literature for two-sample tests maximizes the signal-to-noise ratio of non-degenerate V-statistics \citep{gretton2012optimal,sutherland2021generativemodelsmodelcriticism,liu2021learningdeepkernelsnonparametric}, a setting without estimated nuisance parameters. None of these strands treats the tuning parameter as a nuisance of the test and analyzes data-driven tuning as a power mechanism. We show that the cross-validated choice does act as such a mechanism, through an exact energy decomposition of the misspecification energy into an uncaptured and a captured part: cross-validation minimizes the former and therefore maximizes the latter, namely the signal energy captured by the fitted coefficient (the test's noncentrality); under the null it instead shrinks the noise floor of the statistic. The split between the selection and the testing samples controls the magnitude of the effect: a larger selection share pulls the selected configuration toward the signal-maximizing one, at the price of attenuating the alternative that the testing sample sees.

\paragraph{Simulations and applications.}
Simulation studies confirm that the projection tests control size without any user-chosen tuning in most configurations and attain power comparable to the KCM test at its best kernel parameter. Their advantage is robustness rather than peak performance: in higher dimensions, where fixed kernel parameters (the default $\gamma = 0.5$ or the median rule) are catastrophic, the tests remain fully operational without user intervention. Three applications illustrate the method: the National Supported Work dataset, where the specification of a propensity score model and the null of a zero conditional average treatment effect are tested; the Mincer wage equation, where neither the linear schooling specification nor its augmentation is rejected; and a log-linear demand equation for cigarettes, where the specification is rejected under instrumental-variables residuals. Further simulation designs and results are reported in Appendix~\ref{app:additional_sim}.

\paragraph{Residual and projection notation.}
Throughout, $\varepsilon_0 := \varepsilon(S;\theta_0)$ denotes the residual at the true parameter, decomposing into signal plus mean-zero noise ($\varepsilon_0 = f(X) + v$, Section~\ref{sec:stat}), and the Neyman-orthogonality projection operator $\boldsymbol{\Pi}$ and its empirical counterpart $\hat{\boldsymbol{\Pi}}$, together with the projected residuals $\boldsymbol{\varepsilon}_{\perp} = \hat{\boldsymbol{\Pi}}\boldsymbol{\hat{\varepsilon}}$ and related objects, are defined in Sections~\ref{sec:stat} and~\ref{sec:estimation_effect}.

\paragraph{Notation.}
Convergence in probability is denoted $\xrightarrow{p}$ and almost-sure convergence by $\xrightarrow{a.s.}$; deterministic convergence of real sequences is written with the ordinary arrow $\to$. Convergence in distribution, of real-valued statistics, random vectors, and random elements of Hilbert spaces, is denoted $\xrightarrow{d}$, while weak convergence of stochastic processes and other random elements of function spaces is denoted $\rightsquigarrow$ in the relevant metric space, typically $\ell^\infty(\Upsilon)$, in the sense of \cite{vaart1997weak}. The symbol $\Rightarrow$ is reserved for logical implication. For the multiplier bootstrap we write $\xrightarrow{d^*}$ for almost-sure conditional weak convergence given the data (Chapter 2.9 of \cite{vaart1997weak}). The symbols $O_p$ and $o_p$ denote the usual stochastic orders.

The remainder of the paper is organized as follows. \textbf{Section~\ref{sec:stat}} (test statistics): establishes the equivalence between the null hypothesis and the vanishing of $w^*_\lambda$, introduces the projection statistics, and derives their asymptotic laws under the assumption that $\theta_0$ is known. \textbf{Section~\ref{sec:estimation_effect}} (estimation effects): extends these results to the estimated case through Neyman orthogonal projections and constructs the multiplier bootstrap. \textbf{Section~\ref{sec:nuisance}} (nuisance parameter selection), the centerpiece of the paper, develops the theory of tuning parameter selection: the cross-validation procedure and the process-level limits, and then the behavior of the data-driven test in the three regimes of the testing problem, namely its null distribution and size control, its consistency under fixed alternatives, and the laws of the selector and of the implemented statistic under local alternatives, closing with the energy-based power mechanism. \textbf{Section~\ref{sec:simulation}} (simulations and applications): provides simulation evidence and three empirical applications. \textbf{Section~\ref{sec:conclusion}} concludes. All proofs are gathered in Appendix~\ref{app:proofs}.

\section{Test Statistics and Their Asymptotic Properties}
\label{sec:stat}
This section develops the building blocks of the proposed tests. Section~\ref{sec:equivalence} establishes that, under an integrally strictly positive definite kernel, the null hypothesis holds if and only if the population coefficient function $w^*_\lambda$ vanishes in $\HK$. Section~\ref{sec:proj_stat} turns this equivalence into test statistics built on projections of $w^*_\lambda$ onto $\HK$, deriving their laws under the null, under fixed alternatives, and under local (Pitman) alternatives. Throughout, $(\gamma,\lambda)$ are held fixed; their data-driven selection is analyzed in Section~\ref{sec:nuisance}, and the multiplier bootstrap implementing the null limits derived here is constructed in Section~\ref{sec:estimation_effect}.

Throughout this section, we assume that the true parameter $\theta_0$ is known, and we denote the true residual vector by $\boldsymbol{\varepsilon}_0 = \boldsymbol{Y} - \mathcal{M}_{\theta_0}(X) \in \mathbb{R}^n$. Estimation effects are eliminated by the Neyman orthogonality projections of Section~\ref{sec:estimation_effect}.

\subsection{Equivalence of the Null Hypothesis}
\label{sec:equivalence}
The null hypothesis $\mathbb{E}(\varepsilon_0 | X) = 0$ is equivalent to the vanishing of the population coefficient function $w^*_\lambda$: a KRR fit $\hat{w}$ that is nonzero in the RKHS therefore constitutes evidence against model adequacy. This subsection proceeds in four steps: the spectral setup linking population and empirical objects ($\boldsymbol{L}_k$ with its spectrum, the second-moment operators $C$ and $C_n$, and the singular value decomposition of $\boldsymbol{R}_n$), yielding the empirical eigenbasis expansion of $\hat{w}$; a rotated basis, needed because empirical eigenvectors converge to population eigenfunctions only up to a rotation within eigenvalue clusters (Lemma~\ref{lemma:rotatedbasis}); the geometry of $\HK$ and the conditions (Assumptions~\ref{ass:A} and~\ref{ass:B}) under which $w^*_\lambda$ is well defined and consistently estimated (Lemma~\ref{lemma:estimation_consistency}); and the equivalence itself (Theorem~\ref{th:equivalence}).

\paragraph{Spectral setup.}
We define the integral operator $\boldsymbol{L}_k:L^2(\mathbb{P}_X) \to L^2(\mathbb{P}_X)$ as
\[
    (\boldsymbol{L}_k f)(y)  = \int_{\mathcal{X}} k(x,y) f(x) \, \mathrm{d} \mathbb{P}_X(x)
\]
where $\mathbb{P}_X$ is a probability measure on $\mathcal{X}$. Let $\{\phi_i\}_{i\ge1}$ and $\{\mu_i\}_{i\ge1}$ denote the eigenfunctions and eigenvalues of $\boldsymbol{L}_k$, respectively: the eigenfunctions are orthonormal in $L^2(\mathbb{P}_X)$, and the eigenvalues are non-negative. While the Gram matrix $\boldsymbol{K}$ with entries $K_{ij} = k(x_i, x_j)$ satisfies $\boldsymbol{K}/n \to \boldsymbol{L}_k$ in a suitable sense, the two objects act on fundamentally different spaces. To bridge this gap, we introduce the second-moment operator $C: \HK \to \HK$ defined as
\[
    C = \mathbb{E}\left(k(X,\cdot) \otimes k(X,\cdot)\right),
\]
and its empirical counterpart $C_n: \HK \to \HK$ defined as
\[
    C_n = \frac{1}{n} \sum_{i=1}^n k(x_i,\cdot) \otimes k(x_i,\cdot).
\]
For $f \in \HK$, the tensor product acts as $(k(x_i,\cdot) \otimes k(x_i,\cdot)) f = k(x_i,\cdot) \langle k(x_i,\cdot), f \rangle_{\HK}$.

Next, endow $\mathbb{R}^n$ with the scaled inner product $\langle a,b\rangle_{\mathbb{R}^n} = n^{-1}\sum_{i=1}^n a_i b_i$ and define the restriction operator $\boldsymbol{R}_n: \HK \to \mathbb{R}^n$ by $\boldsymbol{R}_n f = \left(f(x_1), \dots, f(x_n)\right)^\top$. Its adjoint $\boldsymbol{R}_n^\top: \mathbb{R}^n \to \HK$ (with respect to the scaled inner product on $\mathbb{R}^n$) is then given by $\boldsymbol{R}_n^\top u = n^{-1}\sum_{i=1}^{n} u_i k(x_i, \cdot)$ for $u = (u_1, \dots, u_n)^\top$. These definitions yield the factorizations $C_n = \boldsymbol{R}_n^\top \boldsymbol{R}_n$ and $\boldsymbol{R}_n \boldsymbol{R}_n^\top = \boldsymbol{K}/n$.

\paragraph{The empirical expansion and the rotated basis.}
Our analysis of the KRR estimator $\hat{w}$ uses the singular value decomposition (SVD) of the restriction operator $\boldsymbol{R}_n$:
\[
    \boldsymbol{R}_n^\top = \sum_{i=1}^n \frac{\sigma_i}{n}\, \hat{e}_i \boldsymbol{u}_i^\top
    = \sum_{i=1}^n \nu_i^{1/2} n^{-1} \hat{e}_i \tilde{\boldsymbol{u}}_i^\top,
\]
where $\nu_i = \sigma_i^2/n$ are the eigenvalues of $\boldsymbol{K}/n$, $\tilde{\boldsymbol{u}}_i$ are the corresponding eigenvectors normalized as $n^{-1}\tilde{\boldsymbol{u}}_i^\top \tilde{\boldsymbol{u}}_i = 1$, and $\hat{e}_i$ are the empirical eigenfunctions of the empirical second-moment operator $C_n$ (the factor $n^{-1}$ reflects the scaled adjoint defined above). This is equivalent to the representation $\boldsymbol{R}_n = \sum_{i=1}^n \sigma_i \boldsymbol{u}_i \hat{e}_i$ with standard-normalized eigenvectors $\boldsymbol{u}_i^\top \boldsymbol{u}_i = 1$ (since $\tilde{\boldsymbol{u}}_i = \sqrt{n} \boldsymbol{u}_i$). Details, including the spectral relationship between $\boldsymbol{K}$ and $C_n$, appear in Proposition~\ref{svd_empirical} (Appendix~\ref{app:svd_restriction}).

The KRR estimator $\hat{w}$ can be expressed as
\begin{align*}
    \hat w &= \boldsymbol{R}_n^\top (\boldsymbol{K}/n + \lambda \boldsymbol{I})^{-1} \boldsymbol{\varepsilon}_0 \\
    & =  \left( \sum_{j=1}^n \nu_j^{1/2} n^{-1} \hat{e}_j \tilde{\boldsymbol{u}}_j^\top \right)\left( \sum_{i=1}^n (\nu_i + \lambda)^{-1} n^{-1}\tilde{\boldsymbol{u}}_i \tilde{\boldsymbol{u}}_i^\top \right) \boldsymbol{\varepsilon}_0 \\
    & = \sum_{i=1}^n (\nu_i + \lambda)^{-1} \nu_i^{1/2} (n^{-1}\tilde{\boldsymbol{u}}_i^\top \boldsymbol{\varepsilon}_0 ) \hat{e}_i.
\end{align*}

The expansion is in the empirical eigenbasis $\{\hat e_i\}_{i=1}^n$. Because the population eigenvalues $\mu_i$ may have multiplicity greater than one, the estimated eigenfunctions $\hat e_i$ need not converge to the population eigenfunctions $e_i$ individually: within a cluster of empirical eigenvalues converging to the same population value, the empirical eigenvectors are identified only up to an orthogonal rotation, and what converges is the eigenspace spanned by each cluster (Appendix~\ref{orthogonal_transformations}). The following lemma re-expresses the scaled KRR estimator in a rotated basis that converges to the population eigenfunctions uniformly within each cluster. Throughout, $\hat m := n^{-1}\sum_{l=1}^n \varepsilon_{0,l}k(X_l,\cdot) = \boldsymbol{R}_n^\top\boldsymbol{\varepsilon}_0$ denotes the empirical mean embedding.

\begin{lemma}[Rotated eigenbasis of the KRR estimator]\label{lemma:rotatedbasis}
Let the distinct positive eigenvalues of the population operator $C$ be $\mu^{(1)} > \mu^{(2)} > \cdots$ with multiplicities $m_1, m_2, \ldots$, and set $v_k := \sum_{l=1}^k m_l$. Fix a cluster boundary $V$: throughout, assume the eigengap condition $\mu^{(V)} > \mu^{(V+1)}$ together with the sample-size condition of Lemma~\ref{bound_projections} ($n > 128\bar k^2\tau/(\mu^{(V)}-\mu^{(V+1)})^2$ for the relevant confidence level), which guarantees $\lVert\hat P_V-P_V\rVert_{HS}=O_p(n^{-1/2}(\mu^{(V)}-\mu^{(V+1)})^{-1})=o_p(1)$; the expansion below is stated for each such fixed $V$. Under the null hypothesis and under the local alternative $H_{1,n}$ of Theorem~\ref{th:proj_local}, the scaled KRR estimator admits the expansion
\[
    \sqrt{n}\,\hat w = \sum_{k=1}^V\sum_{j=1}^{m_k}\frac{1}{\mu^{(k)}+\lambda}\,\langle \sqrt{n}\,\hat m, \tilde e_{k,j}\rangle_{\HK}\,\tilde e_{k,j}
    \;+\; \sum_{i=v_V+1}^{n}\frac{1}{\mu_i+\lambda}\,\langle \sqrt{n}\,\hat m, \hat e_i\rangle_{\HK}\,\hat e_i + o_p(1),
\]
where $\tilde e_{k,j} := \sum_{l=1}^{m_k}(O_{k,n})_{jl}\,\hat e_{v_{k-1}+l}$ and $O_{k,n}$ is a random orthogonal matrix of size $m_k \times m_k$ constructed as in Appendix~\ref{orthogonal_transformations}, so that $\lVert\tilde e_{k,j} - e_{k,j}\rVert_{\HK} = O_p(n^{-1/2})$ for a suitable orthonormal basis $\{e_{k,j}\}_{j=1}^{m_k}$ of the $k$-th population eigenspace; the tail (indices $i > v_V$) keeps the empirical eigenfunctions, and it is stated with the population weights $\mu_i$; the empirical weights $\nu_i$ differ from these by an $o_p(1)$ term in the expansion, as shown in the proof. Relabeling the eigenvalues (accounting for multiplicity) as $\{\mu_i\}_{i\ge1}$ and the rotated eigenfunctions as $\{\tilde e_i\}_{i\ge1}$ gives the compact form
\[
    \sqrt{n}\,\hat w = \sum_{i=1}^n\frac{1}{\mu_i+\lambda}\,\langle \sqrt{n}\,\hat m, \tilde e_i\rangle_{\HK}\,\tilde e_i + o_p(1).
\]
\end{lemma}
\begin{proof}
    See Appendix~\ref{proof:rotatedbasis}.
\end{proof}

Lemma~\ref{lemma:rotatedbasis} serves the estimation-effect and nuisance-parameter analyses of Sections~\ref{sec:estimation_effect} and~\ref{sec:nuisance}, which compare empirical and population eigenbases within clusters; the proofs of Theorems~\ref{th:proj_null_fix} and~\ref{th:proj_local} bypass it, working with resolvent convergence and Slutsky's theorem in the population basis. Wherever the empirical eigenbasis enters, the rotated basis tracks the aligning rotation at rate $\lVert\tilde e_{k,j}-e_{k,j}\rVert_{\HK} = O_p(n^{-1/2})$, and the spectral weights of the limiting statistics act on the population coefficients $\langle \sqrt{n}\,\hat m, e_i\rangle_{\HK}/\sqrt{\mu_i}$.

\paragraph{The geometry of the RKHS.}
Characterizing the geometry of $\HK$ requires an orthonormal basis (ONB). Mercer's theorem provides one from the eigenfunctions $\{\phi_i\}_{i=1}^\infty$ of $\boldsymbol{L}_k$ but assumes a compact input space; the following conditions extend it to general $\mathcal{X}$.

\begin{assumption}\label{ass:A}
(i) $k(x,\cdot) \in L^2(\mathbb{P}_X)$ for every $x \in \mathcal{X}$; (ii) $\boldsymbol{L}_k$ is a bounded and positive operator on $L^2(\mathbb{P}_X)$, and for every $g \in L^2(\mathbb{P}_X)$, $\boldsymbol{L}_k(g) \in C(\mathcal{X})$, the space of continuous functions on $\mathcal{X}$; (iii) $\boldsymbol{L}_k$ has at most countably many positive eigenvalues $\{\mu_i\}_{i=1}^\infty$, with corresponding orthonormal eigenfunctions $\{\phi_i\}_{i=1}^\infty$ in $L^2(\mathbb{P}_X)$.
\end{assumption}

Assumption~\ref{ass:A} yields two results: the Mercer expansion $k(x,y) = \sum_{i \ge 1} \mu_i \phi_i(x) \phi_i(y)$, which converges absolutely and uniformly on any compact subset of $\mathcal{X} \times \mathcal{X}$, and the fact that $\{\sqrt{\mu_i}\phi_i\}_{i \ge 1}$ forms an ONB of the RKHS $\HK$.
A classical example of an input space and kernel satisfying Assumption~\ref{ass:A} is $\mathcal{X} = \mathbb{R}^d$ coupled with the Gaussian kernel $k(x,y) = \exp(-\gamma\|x-y\|^2)$ for some $\gamma > 0$. For detailed discussions and proofs of these results, see \cite{sun2005mercer}. Assumption~\ref{ass:A} is assumed throughout the remainder of the paper.

\paragraph{The structure of $w^*_\lambda$.}
The analysis of the population KRR coefficient function $w^*_\lambda$ requires additional conditions.

\begin{assumption}\label{ass:B}
(i) The reproducing kernel is universal and bounded, in the sense that $\sup_{x \in \mathcal{X}} k(x,x) = \bar{k} < \infty$; (ii) the eigenvalues of $\boldsymbol{L}_k$ satisfy $\sum_{i \geq 1}\mu_i < \infty$ (implied by (i): the Mercer diagonal $k(x,x)=\sum_{i\ge1}\mu_i\phi_i(x)^2$ and monotone convergence give $\sum_{i\ge1}\mu_i=\mathbb{E}[k(X,X)]\le\bar{k}$); (iii) $|\mathbb{E}(\varepsilon_0 \phi_i(X))| < M$ for all $i \geq 1$ (implied by (v) and $\mathbb{E}[\phi_i(X)^2]=1$ via Cauchy--Schwarz); (iv) the kernel ridge regularization parameter satisfies $\lambda>0$ and is fixed; (v) the residual has a finite second moment, $\mathbb{E}(\varepsilon_0^2) < \infty$; (vi) the kernel is integrally strictly positive definite (ISPD), i.e.\ $\iint k(x,y)f(x)f(y)\,\mathrm{d}\mathbb{P}_X(x)\mathrm{d}\mathbb{P}_X(y) > 0$ for every nonzero $f \in L^2(\mathbb{P}_X)$ (the Gaussian kernel satisfies this).
\end{assumption}

Condition (i) is standard: universality makes $\HK$ dense in $C(\mathcal{X})$, and boundedness ensures each eigenfunction of $\boldsymbol{L}_k$ is bounded. Condition (ii) ensures that $\HK$ is well defined and the eigenvalues decay sufficiently fast; it holds for most common kernels, including the Gaussian. Condition (iii) bounds the projections of the residuals onto the eigenfunctions. Condition (iv) is a design choice rather than an analytical necessity. Classical KRR sends $\lambda \to 0$ at a suitable rate to recover the unregularized population function; here the estimand is the regularized $w^*_\lambda$ itself, $\hat w$ is consistent at any fixed $\lambda$ (Lemma~\ref{lemma:estimation_consistency}), and the equivalence of Theorem~\ref{th:equivalence} does not depend on the value of $\lambda$. What fixed $\lambda$ provides is a non-degenerate limit theory: with $\lambda$ bounded away from zero, the null statistics have non-degenerate weighted chi-square limits, uniform process convergence over the tuning set holds, and the cross-validated selector admits the asymptotic characterization of Section~\ref{sec:nuisance}; each of these steps breaks down under a vanishing regularization parameter, and the admissible set $\Upsilon$ is accordingly bounded away from zero in $\lambda$.

We now record the structure of $w^*_\lambda$.
\begin{lemma}
    \label{lemma:estimation_consistency}
    Under Assumption~\ref{ass:B} and i.i.d.\ sampling (Assumption~\ref{ass:C}(i) below), $w^*_\lambda$ can be written as
    \[
        w^*_\lambda = \sum_{i=1}^\infty \frac{\mu_i}{\mu_i + \lambda} \mathbb{E}(\varepsilon_0 \phi_i(X)) \phi_i,
    \]
    and
    \[
        \| \hat w - w^*_\lambda\|_{\HK}\xrightarrow{p} 0.
    \]
\end{lemma}
\begin{proof}
    See Appendix~\ref{proof:lemma1}.
\end{proof}

The parameters defining $w^*_\lambda$ are the expectations of the residuals projected onto the eigenfunctions $\phi_i$ of $\boldsymbol{L}_k$, each weighted by the factor $\mu_i / (\mu_i + \lambda)$: the regularization parameter $\lambda$ controls the influence of each eigenfunction in the expansion of $w^*_\lambda$.

\paragraph{Equivalence with the null hypothesis.}
Under the null the expected projections $\mathbb{E}(\varepsilon_0 \phi_i(X))$ vanish, so $w^*_\lambda = \boldsymbol{0}$; conversely, $w^*_\lambda = \boldsymbol{0}$ forces every expected projection to vanish. The following theorem records this equivalence.

\begin{theorem}
    \label{th:equivalence}
    Under Assumptions~\ref{ass:A} and~\ref{ass:B}, with $k$ bounded and integrally strictly positive definite (ISPD, Assumption~\ref{ass:B}(vi)) and $\mathbb{E}(\varepsilon_0^2) < \infty$, the null hypothesis $\mathbb{E}(\varepsilon_0 | X) = 0$ is equivalent to $w^*_\lambda = \boldsymbol{0}$ in $\HK$ (i.e., $w^*_\lambda$ vanishes $\mathbb{P}_X$-almost everywhere).
\end{theorem}
\begin{proof}
    See Appendix~\ref{proof:thm1}.
\end{proof}

\subsection{Projection-Based Test Statistics}
\label{sec:proj_stat}
Given the equivalence of Section~\ref{sec:equivalence}, every departure from the null manifests itself as a nonzero element $w^*_\lambda$ of $\HK$, so a test can be built on projections of the fitted coefficient function onto $\HK$. This subsection defines the two projection statistics and derives their laws under the null and fixed alternatives (Theorem~\ref{th:proj_null_fix}) and under local (Pitman) alternatives (Theorem~\ref{th:proj_local}); Assumption~\ref{ass:C} collects the sampling conditions, and the proofs are deferred to Appendices~\ref{proof:thm2} and~\ref{proof:thm3}. The orthogonalized statistics valid when $\theta_0$ is estimated are constructed in Section~\ref{sec:estimation_effect}; the fixed-$(\gamma,\lambda)$ laws derived here anchor the bootstrap validity of Section~\ref{sec:nuisance}.

\paragraph{The projection principle.}
Under the null, we expect
\[
    \langle w^*_\lambda, f \rangle_{\HK} =0, \quad \forall f \in \HK,
\]
whereas under the alternative,
\[
    \langle w^*_\lambda, f \rangle_{\HK} \neq 0, \quad \text{for some } f \in \HK \text{ (e.g.\ } f = w^*_\lambda \text{ itself)}.
\]
We focus on two choices of $f$: (i) $f = w^*_\lambda$, giving the squared RKHS norm of the fit, and (ii) $f = m^* = \mathbb{E}(\varepsilon_0 k(X,\cdot))$, the mean embedding of the residual, giving the cross-product of the fit with the raw residual signal:
\begin{align*}
    & n \hat T_{\mathrm{proj}}^{(1)} = n \langle \hat w, \hat w \rangle_{\HK} = n \boldsymbol{\varepsilon}_0^\top \left(\boldsymbol{K} + n \lambda \boldsymbol{I}\right)^{-1} \boldsymbol{K} \left(\boldsymbol{K} + n \lambda \boldsymbol{I}\right)^{-1} \boldsymbol{\varepsilon}_0,  \\
    & n \hat T_{\mathrm{proj}}^{(2)} = n \langle \hat w, \hat m \rangle_{\HK} =  \boldsymbol{\varepsilon}_0^\top \left(\boldsymbol{K} + n \lambda \boldsymbol{I}\right)^{-1} \boldsymbol{K} \boldsymbol{\varepsilon}_0,
\end{align*}
where
\[
    \hat m =  \frac{1}{n}\sum_{i=1}^n \varepsilon_{0,i} k(X_i,\cdot) = \boldsymbol{R}_n^\top \boldsymbol{\varepsilon}_0.
\]
For brevity, write
\[
    W_n(\lambda) := (\boldsymbol{K} + n\lambda\boldsymbol{I})^{-1}\boldsymbol{K}(\boldsymbol{K} + n\lambda\boldsymbol{I})^{-1},
\]
so that $n\hat{T}_{\mathrm{proj}}^{(1)} = n\,\boldsymbol{\varepsilon}_0^\top W_n(\lambda)\,\boldsymbol{\varepsilon}_0$.

\begin{remark}
    \label{remark:qgeq1}
    In the case of $\varepsilon_{0} \in \mathbb{R}^q $ with $q \in \mathbb{N}$, we can rewrite the test statistics as
\begin{align*}
    & n \hat T_{\mathrm{proj}}^{(1)} = n \sum_{r=1}^q \langle \hat w_r, \hat w_r \rangle_{\HK} = n \sum_{r=1}^q \boldsymbol{\varepsilon}_{0,r}^\top \left(\boldsymbol{K} + n \lambda \boldsymbol{I}\right)^{-1} \boldsymbol{K} \left(\boldsymbol{K} + n \lambda \boldsymbol{I}\right)^{-1} \boldsymbol{\varepsilon}_{0,r}, \\
    & n \hat T_{\mathrm{proj}}^{(2)} = n\sum_{r=1}^q  \langle \hat w_r, \hat m_r \rangle_{\HK} = \sum_{r=1}^q  \boldsymbol{\varepsilon}_{0,r}^\top \left(\boldsymbol{K} + n \lambda \boldsymbol{I}\right)^{-1} \boldsymbol{K} \boldsymbol{\varepsilon}_{0,r},
\end{align*}
with $\boldsymbol{\varepsilon}_{0,r}$ being the $r$-th column of $\boldsymbol{\varepsilon}_0$.
\end{remark}

\paragraph{Asymptotic laws.}
\begin{assumption}\label{ass:C}
(i) The random vectors $S_i=(Y_i,X_i)$, $i=1,2,\dots$, are independent and identically distributed (i.i.d.) with probability measure $\mathbb{P}_{S}$; extensions to weakly dependent (e.g.\ mixing) data require the corresponding limit theorems and are not pursued here; (ii) \textit{Regularity Conditions}: (1) the residual function $\varepsilon: \mathcal{S} \times \Theta \longrightarrow \mathbb{R}$ is continuous on $\Theta$ for each $s \in \mathcal{S}$; (2) $\mathbb{E}(\varepsilon(S;\theta)|X=x)$ exists and is finite for every $\theta \in \Theta$ and $\mathbb{P}_X$-almost every $x \in \mathcal{X}$; (3) $\mathbb{E}(\varepsilon(S;\theta)|X=x)$ is continuous on $\Theta$ for $\mathbb{P}_X$-almost every $x \in \mathcal{X}$.
\end{assumption}

These conditions are standard in the literature on conditional moment models and model checks; see, e.g., \cite{muandet2020kernel}.

\begin{theorem}
    \label{th:proj_null_fix}
    Suppose Assumptions~\ref{ass:A}, \ref{ass:B} and~\ref{ass:C} hold. Then, under the null hypothesis, we have
    \begin{align*}
        n \hat T_{\mathrm{proj}}^{(1)} & \xrightarrow{d}  \sum_{i = 1}^\infty  \frac{\mu_i}{(\mu_i+\lambda)^2} Z_i^2,\\
        n \hat T_{\mathrm{proj}}^{(2)} & \xrightarrow{d}  \sum_{i = 1}^\infty  \frac{\mu_i}{\mu_i+\lambda} Z_i^2,
    \end{align*}
    where $(Z_1,Z_2,\dots)^\top \sim \mathcal{N}(\mathbf{0}, \Sigma)$ is a
    Gaussian sequence with covariance matrix
    $\Sigma_{ij} =
    \mathbb{E}\bigl[\varepsilon_0^2\,\phi_i(X)\,\phi_j(X)\bigr]$.
    In particular, $\Sigma_{ii} = \tau_i := \mathrm{Var}(\varepsilon_0\,
    \phi_i(X))$ (using $\mathbb{E}[\varepsilon_0 \phi_i(X)] = 0$ under $H_0$).
    Under homoscedasticity ($\mathbb{E}[v^2\mid X] = \sigma^2$
    a.s.) the $Z_i$ reduce to independent $N(0,\sigma^2)$ variables.
    Under the fixed alternative hypothesis, for any fixed $t > 0$,
    \begin{align*}
        & \mathbb{P}\left(n \hat T_{\mathrm{proj}}^{(1)} > t\right) \to 1, \\
        & \mathbb{P}\left(n \hat T_{\mathrm{proj}}^{(2)} > t\right) \to 1.
    \end{align*}
\end{theorem}
\begin{proof}
    See Appendix~\ref{proof:thm2}.
\end{proof}

We turn finally to local (Pitman) alternatives, in which the misspecification shrinks at the parametric rate:
\[
    H_{1,n}: Y - \mathcal{M}_{\theta_0}(X) = \boldsymbol{\varepsilon}_0 + \frac{R(X)}{\sqrt{n}} \quad \text{a.s.},
\]
where $\boldsymbol{\varepsilon}_0 = v$ continues to denote the mean-zero noise component, satisfying $\mathbb{E}(v|X)=0$ as under $H_0$, and $R(X)$ is a nonzero measurable function, square integrable with respect to $\mathbb{P}_X$.
\begin{theorem}
    \label{th:proj_local}
    Under the local alternative hypothesis $H_{1,n}$ and Assumptions~\ref{ass:A}, \ref{ass:B} and~\ref{ass:C}, we have
    \begin{align*}
        n \hat T_{\mathrm{proj}}^{(1)} & \xrightarrow{d}  \sum_{i\geq 1}  \frac{\mu_i}{(\mu_i+\lambda)^2} \left(Z_i + \mathbb{E}(R(X) \phi_i(X))\right)^2 ,\\
        n \hat T_{\mathrm{proj}}^{(2)} & \xrightarrow{d}  \sum_{i\geq 1}  \frac{\mu_i}{\mu_i+\lambda} \left(Z_i + \mathbb{E}(R(X) \phi_i(X))\right)^2 ,
    \end{align*}
    where $\{Z_i\}_{i\geq 1}$ are as defined in Theorem~\ref{th:proj_null_fix}.
\end{theorem}
\begin{proof}
    See Appendix~\ref{proof:thm3}.
\end{proof}

\section{Estimation Effects}
\label{sec:estimation_effect}
Thus far, $\theta_0$ has been taken as known. In practice it is estimated by some consistent $\hat{\theta}$, and plugging $\hat\theta$ into the statistics of Section~\ref{sec:stat} alters their laws unless the estimation effect is removed; this section removes it. Throughout, the bandwidth $\gamma$ of the Gaussian kernel and the regularization parameter $\lambda$ are held fixed, with their data-driven selection deferred to Section~\ref{sec:nuisance}.

The section proceeds in two steps. Section~\ref{sec:projection} eliminates the estimation effect through Neyman orthogonality projections: it introduces the required assumptions (Assumptions~\ref{ass:D}, \ref{ass:E}, and~\ref{ass:F}), constructs the population projection $\boldsymbol{\Pi}$ and its empirical counterpart $\hat{\boldsymbol{\Pi}}$, and shows that the estimation-effect remainder is asymptotically negligible, so that the residuals of Section~\ref{sec:stat} can be replaced by their orthogonalized counterparts; the resulting null laws are the projected laws of Corollaries~\ref{cor:proj_null} and~\ref{cor:proj_local}. Section~\ref{sec:bootstrap} constructs the multiplier bootstrap that approximates these null distributions in practice and establishes its conditional consistency in probability (Theorem~\ref{th:bootstrap_consistency}); the sample-splitting variant that remains valid under data-driven tuning is treated in Section~\ref{sec:nuisance}.

\subsection{Eliminating Estimation Effects via Projection}
\label{sec:projection}
We adopt the approach of \cite{escanciano2014specification,escanciano2024gaussian,sant2019specification}, which employs Neyman orthogonality projections to remove the impact of parameter estimation (extended to penalized high-dimensional settings by \cite{sancetta2022testing}). In contrast to wild-bootstrap methods \citep{delgado2006consistent}, this approach does not depend on a linear Bahadur representation for the estimator of $\theta_0$, nor does it require re-estimating parameters within each bootstrap sample, a flexibility valuable when $\hat{\theta}$ comes from non-standard procedures whose asymptotic linear expansions may be unavailable. 
\paragraph{Assumptions on the estimator and the residual function.}
The following assumptions are required for this purpose.

\begin{assumption}\label{ass:D}
(i) The parameter space $\Theta$ is a compact subset of $\mathbb{R}^{d_\theta}$; (ii) the true parameter $\theta_0$ is an interior point of $\Theta$; (iii) the consistent estimator $\hat{\theta}$ satisfies $\lVert \hat{\theta} - \theta_0 \rVert = O_p (n^{-\varrho})$, with $\varrho > 1/4$.
\end{assumption}

Assumption~\ref{ass:D} requires only convergence in probability at rate $\varrho > 1/4$, with no linear representation, accommodating non-standard estimators such as the LASSO; the proof of Lemma~\ref{lemma:estimation_remainder} confirms that this $O_p$ rate suffices for the remainder bounds.

\begin{remark}[Conditionally deterministic matrices]\label{rem:cond_det}
A standing convention of the appendix proofs treats matrices carrying the estimated objects ($\hat{\boldsymbol{\Pi}}$, $\hat{\mathbb{G}}$) as conditionally deterministic given $\mathbf{X}$. The convention is valid whenever $\hat\theta$ is a smooth functional of sample averages, in particular for $\sqrt{n}$-regular estimators; under it the projection-quadratic forms it licenses, e.g.\ $\hat{\mathbb{G}}^\top\boldsymbol{\varepsilon}_0=O_p(\sqrt{n})$ and $\boldsymbol{\varepsilon}_0^\top H\boldsymbol{\varepsilon}_0=O_p(1)$ with $H$ the score hat projector, hold as stated.
\end{remark}

Additional regularity conditions on the smoothness of the residual function are also required.

\begin{assumption}\label{ass:E}
(i) The residual $\varepsilon(s;\theta)$ is twice continuously differentiable with respect to $\theta$, with its first derivative $g_{\theta}(x) = \mathbb{E}(\nabla_{\theta} \varepsilon(s;\theta)|X=x)$ satisfying $\mathbb{E}\left(\sup_{\theta \in \Theta} \lVert g_{\theta}(X) \rVert\right) < \infty$ and the conditional score uniformly bounded, $\sup_{\theta \in \Theta} \sup_{x \in \mathcal{X}} \lVert g_{\theta}(x) \rVert \le \bar{G} < \infty$; and its second derivative satisfying $\mathbb{E}\left(\sup_{\theta \in \Theta} \lVert \nabla g_{\theta}(X) \rVert\right) < \infty$; (ii) the matrix $\Gamma_{\theta} = \mathbb{E}\left[g_{\theta}(X) g_{\theta}(X) ^\top\right]$ is nonsingular in a neighborhood of $\theta_0$; (iii) the raw Hessian is uniformly integrable in the strong sense, $\mathbb{E}\left(\sup_{\theta \in \Theta} \lVert \nabla_{\theta}^2 \varepsilon(S;\theta) \rVert^2\right) < \infty$, the raw score has a finite fourth moment, $\mathbb{E}\left(\sup_{\theta \in \Theta} \lVert \nabla_{\theta} \varepsilon(S;\theta) \rVert^4\right) < \infty$, and the score satisfies $\mathbb{E}\left(\varepsilon_0^2 \lVert g_{\theta_0}(X) \rVert^2\right) < \infty$.
\end{assumption}

\begin{remark}[Uniform score bound]\label{rem:score_bound}
The uniform score bound in Assumption~\ref{ass:E}(i) is a standard bounded-score condition: it holds automatically when $\mathcal{X}$ is compact and each $g_\theta$ is continuous, and in the sieve case of Assumption~\ref{ass:F} it is implied by $\sup_x\lVert g_{\theta_0}(x)\rVert<\infty$ together with the sup-norm deviation clause of Assumption~\ref{ass:F}(iii), the supremum over the parameter segment between $\theta_0$ and $\hat\theta$ following from the second-derivative integrability of Assumption~\ref{ass:E}(i); it supplies the leverage control in the uniform and bootstrap theory (Lemmas~\ref{lem:lipschitz_common} and~\ref{lem:cond_lip_mult}) and grounds the symbol $\bar G$ used there.
\end{remark}

\begin{assumption}\label{ass:F}
(i) The conditional score $g_{\theta_0}(X)=\mathbb{E}[\nabla_\theta \varepsilon(S;\theta_0)\mid X]$ is identifiable from $X$ and lies in the $L^2(\mathbb{P}_X)$-closure of a sieve span; (ii) the sieve basis $\{z_\ell\}_{\ell=1}^{L}$ is uniformly bounded, and, with $\boldsymbol{z} := (z_\ell(X_i))_{1\le i\le n,\,1\le\ell\le L}$ the $n\times L$ sieve design matrix, the smallest eigenvalue of the Gram matrix $\boldsymbol{z}^\top \boldsymbol{z}/n$ is bounded away from zero, $\lambda_{\min}(\boldsymbol{z}^\top \boldsymbol{z}/n) \geq c>0$ with probability tending to one; (iii) $g_{\theta_0}$ is sufficiently smooth and $L=L_n\to\infty$ at a rate such that $\hat g_{\hat\theta}(X)\xrightarrow{p} g_{\theta_0}(X)$ in $L^2(\mathbb{P}_X)$ and the sieve estimation error satisfies $\lVert \hat g_{\hat\theta}-g_{\theta_0}\rVert_{L^2(\mathbb{P}_X)} = o_p(n^{\varrho-1/2})$, the rate at which the product $\lVert\hat\theta-\theta_0\rVert^2\,n\,\lVert \hat g_{\hat\theta}-g_{\theta_0}\rVert_{L^2(\mathbb{P}_X)}^2$ consumed by the estimation-effect remainder of Lemma~\ref{lemma:estimation_remainder} vanishes; in addition, for the CLT arguments of the projected-statistic proofs that involve the estimated score at sample points, the sup-norm deviation satisfies the side condition $\lVert \hat g_{\hat\theta}-g_{\theta_0}\rVert_\infty = O_p(n^{-1/4})$.
\end{assumption}

Assumption~\ref{ass:F} is required only when the score $\nabla_\theta \varepsilon$ is not $X$-measurable (e.g., nonlinear IV/GMM or quantile residuals). When $\nabla_\theta \varepsilon$ is $X$-measurable, as in conditional mean regression $\varepsilon=Y-\mathcal{M}_\theta(X)$, the raw score equals $g_{\theta_0}(X)$ and the sieve step is vacuous. The sup-norm side condition in item (iii), invoked in the CLT arguments for the estimated score at sample points, holds under Lipschitz sieve bases: the sup-norm error decomposes into the approximation bias $L^{-\varsigma}$ and the sup-norm variance term $O_p(\sqrt{L\log n/n})$, and for any $\varsigma>1/2$ the dimension restriction $L^{2\varsigma+1}\log n=o(n)$ forces both terms to $o_p(n^{-1/4})$; the bias condition $L^{-\varsigma}=o(n^{-1/4})$ reduces to $\varsigma/(2\varsigma+1)>1/4$, while $L\log n=o(\sqrt{n})$ follows directly from the dimension restriction. The $L^2$ clause of item (iii) is attained by the same bias--variance decomposition: under the dimension restriction the variance term obeys $O_p(\sqrt{L/n})=o_p(n^{\varrho-1/2})$ for every $\varrho\ge1/4$, and the bias condition $L^{-\varsigma}=o(n^{\varrho-1/2})$ reduces to the same $\varsigma/(2\varsigma+1)>1/4$; the clause thus interpolates between consistency of $\hat g_{\hat\theta}$ at $\varrho=1/2$ and the $n^{-1/4}$ rate that \citet{escanciano2024gaussian} imposes on nonparametric score estimators at $\varrho=1/4$.

\paragraph{The population projection.}
Under Assumptions~\ref{ass:D} and~\ref{ass:E}, we introduce a projection operator $\boldsymbol{\Pi}$ defined on the Hilbert space $L^2(\mathbb{P}_S)$ of square-integrable random variables, equipped with the inner product $\langle f, g\rangle = \mathbb{E}[f(S)g(S)]$. For any $\omega \in L^2(\mathbb{P}_S)$, the value of $\boldsymbol{\Pi}\omega$ at $s = (y, x)$ is
\[
    (\boldsymbol{\Pi} \omega)(s) = \omega(s) - \mathbb{E}(\omega(S) g_{\theta_0}^\top(X)) \Gamma_{\theta_0}^{-1} g_{\theta_0}(x),
\]
$\boldsymbol{\Pi}$ thus acts on random variables; it is the residual from the population $L^2$-regression of $\omega(S)$ on $g_{\theta_0}(X)$, hence the orthogonal projection onto $\operatorname{span}\{g_{\theta_0,1}(X), \ldots, g_{\theta_0,d_\theta}(X)\}^\perp$: idempotent, self-adjoint, and score-annihilating, $\boldsymbol{\Pi} g_{\theta_0}(X) = \boldsymbol{0}$ a.s.
Applying this projection operator to the residual function leads to modified versions of $w^*_\lambda$ and $m^*$:
\begin{align*}
    & w^*_{\lambda,\perp} = \sum_{i=1}^\infty \frac{\mu_i}{\mu_i+\lambda} \mathbb{E}(\boldsymbol{\Pi} \varepsilon_0 \phi_i(X)) \phi_i, \\
    & m^*_{\perp} = \mathbb{E}(\boldsymbol{\Pi} \varepsilon_0 k(X,\cdot)).
\end{align*}
Under $H_0$, the law of iterated expectations gives $\mathbb{E}[\varepsilon_0 g_{\theta_0}^\top(X)] = \boldsymbol{0}$, hence $\boldsymbol{\Pi}\varepsilon_0 = \varepsilon_0$ a.s., and the projected population objects coincide with their unprojected counterparts; the projection exists to deliver robustness against the estimation effect of $\hat\theta$ (and, under alternatives, to reshape the mean shift through $\boldsymbol{\Pi}R$). The sample analogue does not inherit this coincidence: the correction $\hat{\boldsymbol{\Pi}}\boldsymbol{\varepsilon}_0 - \boldsymbol{\varepsilon}_0 = -\hat{\mathbb{G}}(\hat{\mathbb{G}}^\top\hat{\mathbb{G}})^{-1}\hat{\mathbb{G}}^\top\boldsymbol{\varepsilon}_0$ is $O_p(n^{-1/2})$ per component and survives at the $O_p(1)$ scale of the scaled mean embedding $\sqrt{n}\,\hat m_\perp$. The null laws below are therefore those of Corollary~\ref{cor:proj_null}, not of Theorem~\ref{th:proj_null_fix}; Remark~\ref{rem:proj_null} identifies the resulting limit.
\paragraph{Local robustness.}
To analyze the local behavior near $\theta_0$, consider the derivatives of these functions with respect to $\theta$ evaluated at $\theta_0$:
\begin{align*}
    \frac{\partial}{\partial \theta} w^*_{\lambda,\perp}(\theta)\bigg|_{\theta = \theta_0} &= \sum_{i=1}^\infty \frac{\mu_i}{\mu_i+\lambda} \mathbb{E}\left(\boldsymbol{\Pi} g_{\theta_0}(X) \phi_i(X)\right) \phi_i = \boldsymbol{0},
\end{align*}
and
\begin{align*}
    \frac{\partial}{\partial \theta} m^*_{\perp}(\theta)\bigg|_{\theta = \theta_0} &= \mathbb{E}\left(\boldsymbol{\Pi} g_{\theta_0}(X) k(X,\cdot)\right) = \boldsymbol{0}.
\end{align*}
This establishes local robustness to perturbations of $\theta$ near $\theta_0$; the vanishing derivatives follow from dominated convergence and iterated expectations. The identity holds in expectation, not pointwise: since $\boldsymbol{\Pi}$ is fixed at $\theta_0$, the pointwise derivative is
\[
    \frac{\partial}{\partial \theta} \boldsymbol{\Pi}\varepsilon(s;\theta)\bigg|_{\theta = \theta_0} = \boldsymbol{\Pi}(\nabla_\theta \varepsilon(s;\theta_0)) = \nabla_\theta \varepsilon(s;\theta_0) - g_{\theta_0}(x),
\]
which is generally nonzero: it is the score net of its conditional mean. What vanishes is the expectation against any $X$-measurable function $h(X)$:
\begin{align*}
    \frac{\partial}{\partial \theta} \mathbb{E}\!\left[\boldsymbol{\Pi}\varepsilon(S;\theta)\, h(X)\right]\bigg|_{\theta = \theta_0} &= \mathbb{E}\!\left[\boldsymbol{\Pi}(\nabla_\theta \varepsilon(S;\theta_0))\, h(X)\right] = \mathbb{E}\!\left[\boldsymbol{\Pi} g_{\theta_0}(X)\, h(X)\right] = 0,
\end{align*}
where the second equality replaces $\nabla_\theta \varepsilon(S;\theta_0)$ by $g_{\theta_0}(X) = \mathbb{E}[\nabla_\theta \varepsilon(S;\theta_0) \mid X]$ inside $\mathbb{E}[\,\cdot\, h(X)]$, legitimate because both $h(X)$ and the $g_{\theta_0}^\top(X)$ entering the regression coefficient of $\boldsymbol{\Pi}$ are $X$-measurable, and the last equality uses $\boldsymbol{\Pi} g_{\theta_0}(X) = \boldsymbol{0}$. Setting $h = \phi_i$ and $h = k(X,\cdot)$ yields the two derivatives displayed above.

\paragraph{The empirical projection.}
To build a sample analogue of $\boldsymbol{\Pi}$, which projects off $g_{\theta_0}(X)=\mathbb{E}[\nabla_\theta \varepsilon\mid X]$ (a function of $X$ alone), we project off an estimate of $g_{\theta_0}(X)$ rather than the raw scores: regressing the residuals off the raw scores $\nabla_\theta \varepsilon(s_i)$ would, for $Y$-dependent scores, also remove the idiosyncratic component $\nabla_\theta \varepsilon - g_{\theta_0}(X)$, which is correlated with $\varepsilon_0$ under $H_0$ and would make the statistic diverge. Let $\boldsymbol{Z}$ be the $n\times L$ sieve design matrix with columns $Z^{(1)}(X),\ldots,Z^{(L)}(X)$, and $\boldsymbol{P}_Z=\boldsymbol{Z}(\boldsymbol{Z}^\top \boldsymbol{Z})^{-1}\boldsymbol{Z}^\top$. With $\mathbb{G}$ the $n\times d_\theta$ raw-score matrix whose $i$th row is $\hat g_i^\top=(\nabla_\theta \varepsilon(s_i;\theta)|_{\theta=\hat\theta})^\top$, the sieve estimate of the conditional score is $\hat{\mathbb{G}}=\boldsymbol{P}_Z\mathbb{G}$, whose $i$th row $\hat g_{\hat\theta}(X_i)^\top$ estimates $g_{\theta_0}(X_i)^\top$. The empirical projection operator is
\[
    \hat{\boldsymbol{\Pi}} = \boldsymbol{I}_n - \hat{\mathbb{G}}\left(\hat{\mathbb{G}}^\top \hat{\mathbb{G}}\right)^{-1} \hat{\mathbb{G}}^\top,
\]
so that $\hat{\boldsymbol{\Pi}}\hat{\mathbb{G}}=\boldsymbol{0}$. Under Assumption~\ref{ass:F}, $\hat g_{\hat\theta}(X)\xrightarrow{p} g_{\theta_0}(X)$, so $\hat{\boldsymbol{\Pi}}$ is the empirical analogue of $\boldsymbol{\Pi}$: it projects the residual vector $\boldsymbol{\varepsilon}\in\mathbb{R}^n$ off the estimated conditional score $\hat g_{\hat\theta}(X)$. When $\nabla_\theta \varepsilon$ is $X$-measurable, the sieve step is vacuous, $\hat{\mathbb{G}}=\mathbb{G}$, and the projection reduces to the raw-score annihilator, which we denote by $\hat{\boldsymbol{\Pi}}^{\mathrm{raw}} := \boldsymbol{I}_n - \mathbb{G}(\mathbb{G}^\top \mathbb{G})^{-1}\mathbb{G}^\top$; the sieve-based $\hat{\boldsymbol{\Pi}}$ is the special case obtained from $\hat{\boldsymbol{\Pi}}^{\mathrm{raw}}$ by replacing $\mathbb{G}$ with $\hat{\mathbb{G}}$. We write $\boldsymbol{\varepsilon}_{\perp} := \hat{\boldsymbol{\Pi}}\boldsymbol{\hat{\varepsilon}}$ for the resulting orthogonalized residual vector.

\paragraph{The estimation-effect remainder.}
The following lemma establishes how this projection operator eliminates estimation effects in finite samples.

\begin{lemma}[Estimation Effect Remainder]
\label{lemma:estimation_remainder}
    Suppose Assumptions~\ref{ass:D}, \ref{ass:E}, and~\ref{ass:F} and the i.i.d.\ sampling of Assumption~\ref{ass:C}(i) hold. Let $\boldsymbol r = \hat{\boldsymbol{\Pi}} \boldsymbol{\varepsilon}_{\hat \theta} - \hat{\boldsymbol{\Pi}} \boldsymbol{\varepsilon}_0$, and define the RKHS remainders $R_w = \sqrt{n}\,\boldsymbol{R}_n^\top (\boldsymbol{K}/n + \lambda \boldsymbol{I})^{-1} \boldsymbol r$ and $R_m = \sqrt{n}\,\boldsymbol{R}_n^\top \boldsymbol r$. Then, if $\varrho > 1/4$, $\|R_w\|_{\HK} = o_p(1)$ and $\|R_m\|_{\HK} = o_p(1)$.
\end{lemma}
\begin{proof}
    See Appendix~\ref{proof:estimation_remainder}.
\end{proof}

For the projection-based test statistics, we apply the Cauchy--Schwarz inequality to the RKHS inner products:
\begin{align*}
    n \hat T_{\mathrm{proj},\perp}^{(1)} &= \| W_0 + R_w \|_{\HK}^2 \\
    &= \|W_0\|_{\HK}^2 + 2\langle W_0, R_w \rangle_{\HK} + \|R_w\|_{\HK}^2 \\
    &= n (\hat{\boldsymbol{\Pi}} \boldsymbol{\varepsilon}_0)^\top \left(\boldsymbol{K} + n \lambda \boldsymbol{I}\right)^{-1} \boldsymbol{K} \left(\boldsymbol{K} + n \lambda \boldsymbol{I}\right)^{-1} (\hat{\boldsymbol{\Pi}} \boldsymbol{\varepsilon}_0) + o_p(1),
\end{align*}
where the cross-term is bounded by $|2\langle W_0, R_w \rangle_{\HK}| \le 2 \|W_0\|_{\HK} \|R_w\|_{\HK} = O_p(1) o_p(1) = o_p(1)$, and the remainder term $\|R_w\|_{\HK}^2 = o_p(1)$. Here $W_0 := \sqrt{n}\,\boldsymbol{R}_n^\top(\boldsymbol{K}/n+\lambda\boldsymbol{I})^{-1}\hat{\boldsymbol{\Pi}}\boldsymbol{\varepsilon}_0 = (C_n+\lambda\boldsymbol{I})^{-1}\sqrt{n}\,\hat m_\perp$ and $M_0 := \sqrt{n}\,\hat m_\perp$ with $\hat m_\perp := \boldsymbol{R}_n^\top\hat{\boldsymbol{\Pi}}\boldsymbol{\varepsilon}_0$. The bound $\|W_0\|_{\HK} = O_p(1)$ (and likewise $\|M_0\|_{\HK} = O_p(1)$) holds in the null and local-alternative regimes in which the limit laws below are derived: writing $\hat\beta := (\hat{\mathbb{G}}^\top\hat{\mathbb{G}})^{-1}\hat{\mathbb{G}}^\top\boldsymbol{\varepsilon}_0$, we have $\sqrt{n}\,\hat m_\perp = \sqrt{n}\,\hat m - (\boldsymbol{R}_n^\top\hat{\mathbb{G}})\sqrt{n}\,\hat\beta$, where the second-moment computation in the proof of Lemma~\ref{lemma:rotatedbasis} gives $\mathbb{E}\|\sqrt{n}\hat m\|_{\HK}^2 = O(1)$, $\|\boldsymbol{R}_n^\top\hat{\mathbb{G}}\|_{\HK} \le \sqrt{d_\theta\bar k}\,\sup_i\|\hat g_i\| = O_p(1)$ (from $\|\boldsymbol{R}_n^\top u\|_{\HK}^2 = n^{-2}u^\top\boldsymbol K u \le n^{-1}\bar k\,\|u\|^2$ and $\lambda_{\max}(\boldsymbol K) \le n\bar k$), and $\sqrt{n}\hat\beta = O_p(1)$ by the CLT; hence $\|\sqrt{n}\hat m_\perp\|_{\HK} = O_p(1)$ and $\|W_0\|_{\HK} \le \lambda^{-1}\|\sqrt{n}\hat m_\perp\|_{\HK} = O_p(1)$. Under a fixed alternative $\|W_0\|_{\HK} = O_p(\sqrt{n})$ instead, and the divergence argument of Corollary~\ref{cor:proj_null} applies. Similarly,
\begin{align*}
    n \hat T_{\mathrm{proj},\perp}^{(2)} &= \langle W_0 + R_w, M_0 + R_m \rangle_{\HK} \\
    &= \langle W_0, M_0 \rangle_{\HK} + \langle W_0, R_m \rangle_{\HK} + \langle R_w, M_0 \rangle_{\HK} + \langle R_w, R_m \rangle_{\HK} \\
    &= (\hat{\boldsymbol{\Pi}} \boldsymbol{\varepsilon}_0)^\top \left(\boldsymbol{K} + n \lambda \boldsymbol{I}\right)^{-1} \boldsymbol{K} (\hat{\boldsymbol{\Pi}} \boldsymbol{\varepsilon}_0) + o_p(1),
\end{align*}
with all cross and remainder terms being $o_p(1)$ by the Cauchy--Schwarz inequality.

\paragraph{Asymptotic laws of the orthogonalized statistics.}
Consequently, the projected statistics inherit the laws established for known $\theta_0$, with residuals replaced by their Neyman-orthogonalized counterparts:

\begin{corollary}[Null and Fixed Alternative for Projection Tests]
\label{cor:proj_null}
Suppose Assumptions~\ref{ass:B}, \ref{ass:C}, \ref{ass:D}, \ref{ass:E}, and~\ref{ass:F} hold. Under the null hypothesis $H_0$, the projected projection-based test statistics satisfy:
\begin{align*}
    n \hat T_{\mathrm{proj},\perp}^{(1)} &\xrightarrow{d} \sum_{i=1}^\infty \frac{\mu_i}{(\mu_i+\lambda)^2} (Z_i^\perp)^2, \\
    n \hat T_{\mathrm{proj},\perp}^{(2)} &\xrightarrow{d} \sum_{i=1}^\infty \frac{\mu_i}{\mu_i+\lambda} (Z_i^\perp)^2,
\end{align*}
where $Z^\perp = (Z_1^\perp, Z_2^\perp, \dots)^\top$ is the mean-zero Gaussian sequence with covariance matrix
\[
    \Sigma^\perp_{ij} := \mathbb{E}\big[\sigma^2(X)\,\big(\phi_i(X)-\kappa_i(X)\big)\big(\phi_j(X)-\kappa_j(X)\big)\big],
    \qquad
    \kappa_i(x) := g_{\theta_0}(x)^\top\,\Gamma_{\theta_0}^{-1}\,\mathbb{E}\big[g_{\theta_0}(X)\,\phi_i(X)\big],
\]
with $\sigma^2(X) := \mathbb{E}[v^2 \mid X]$ and $\Gamma_{\theta_0} := \mathbb{E}[g_{\theta_0}(X)g_{\theta_0}(X)^\top]$ (Assumption~\ref{ass:E}(ii)). Under the fixed alternative hypothesis $H_1$, let $(\theta^*, f^*)$ denote the pseudo-true limit of $(\hat\theta, \hat f)$, with $\hat f(x) := \mathbb{E}[\varepsilon(X;\hat\theta)\mid X=x]$ and $f^*(x) := \mathbb{E}[\varepsilon(X;\theta^*) \mid X=x]$ and $\boldsymbol{\Pi}_{\theta^*}$ the projection onto the complement of $\operatorname{span}\{g_{\theta^*}\}$. Assume $\boldsymbol{\Pi}_{\theta^*}f^* \neq 0$ in $L^2(\mathbb{P}_X)$ and that the rate condition of Assumption~\ref{ass:D}(iii) holds around $\theta^*$. Then for any fixed $t > 0$, $\mathbb{P}(n \hat T_{\mathrm{proj},\perp}^{(k)} > t) \to 1$ for $k=1,2$.
\end{corollary}

\begin{remark}[The projected null law]\label{rem:proj_null}
Equivalently, the scaled mean embedding of the projected residuals satisfies $\sqrt{n}\,\hat m_\perp \xrightarrow{d} G_\perp$ in $\HK$, where $G_\perp := \xi - A^\top\Gamma_{\theta_0}^{-1}W$, $\xi$ is the Gaussian limit of Theorem~\ref{th:proj_null_fix}, $A := \mathbb{E}[g_{\theta_0}(X)\,k(X,\cdot)] \in \mathcal{H}^{d_\theta}$, and $W \sim \mathcal{N}(\boldsymbol{0}, \mathbb{E}[\sigma^2(X)g_{\theta_0}(X)g_{\theta_0}(X)^\top])$ is jointly Gaussian with $\xi$ through $\operatorname{Cov}(\langle\xi, f\rangle, W) = \mathbb{E}[\sigma^2(X)f(X)g_{\theta_0}(X)]$ for every $f \in \mathcal{H}$; then $Z_i^\perp = \langle G_\perp, e_i\rangle_{\mathcal{H}}/\sqrt{\mu_i}$. The limit is $G_\perp$ rather than $\xi$ because the $O_p(n^{-1/2})$ score-projection correction above operates at the $O_p(1)$ scale of $\sqrt{n}\,\hat m_\perp$. In the location model $g_{\theta_0} \equiv 1$, $A = \mathbb{E}[k(X,\cdot)]$ and $W$ is a scalar $Z \sim \mathcal{N}(0,\bar\sigma^{2})$ with $\bar\sigma^{2} := \mathbb{E}[\sigma^2(X)]$, so that $G_\perp = \xi - Z\,\mathbb{E}[k(X,\cdot)]$, while
$\Sigma^\perp_{ij} = \mathbb{E}\bigl[\sigma^2(X)\,(\phi_i(X)-m_i)(\phi_j(X)-m_j)\bigr]$ with $m_i := \mathbb{E}[\phi_i(X)]$; under homoscedasticity this reduces to $\Sigma^\perp_{ij} = \bar\sigma^{2}(\delta_{ij}-m_i m_j)$.
\end{remark}

\begin{corollary}[Local Alternative for Projection Tests]
\label{cor:proj_local}
Suppose Assumptions~\ref{ass:B}, \ref{ass:C}, \ref{ass:D}, \ref{ass:E}, and~\ref{ass:F} hold. Under the local alternative hypothesis $H_{1,n}$, the projected projection-based test statistics satisfy:
\begin{align*}
    n \hat T_{\mathrm{proj},\perp}^{(1)} &\xrightarrow{d} \sum_{i=1}^\infty \frac{\mu_i}{(\mu_i+\lambda)^2} \left( Z_i^\perp + \mathbb{E}(\boldsymbol{\Pi} R(X) \phi_i(X)) \right)^2, \\
    n \hat T_{\mathrm{proj},\perp}^{(2)} &\xrightarrow{d} \sum_{i=1}^\infty \frac{\mu_i}{\mu_i+\lambda} \left( Z_i^\perp + \mathbb{E}(\boldsymbol{\Pi} R(X) \phi_i(X)) \right)^2,
\end{align*}
where $\{Z_i^\perp\}_{i \ge 1}$ are as in Corollary~\ref{cor:proj_null}: the noise law, which satisfies $\mathbb{E}[\varepsilon_0\mid X]=0$, is unchanged by the alternative, while the mean shift enters reshaped, through $\mathbb{E}(\boldsymbol{\Pi} R(X) \phi_i(X))$.
\end{corollary}

\subsection{Bootstrapping the Null Distributions}
\label{sec:bootstrap}
The null limits of Section~\ref{sec:stat} are infinite weighted sums of squares of dependent Gaussians whose covariance involves the unknown law of the residuals and cannot be tabulated directly. The multiplier bootstrap draws i.i.d.\ multipliers $\{v_i\}_{i=1}^n$ with zero mean, unit variance, bounded support, independent of the data $\{s_i\}_{i=1}^n$, and recomputes the statistics on the orthogonalized, multiplied residuals: conditional on the data, the multipliers reproduce the fluctuations of the limits. Validity is conditional consistency in probability, denoted $\xrightarrow{d^*}$ (Chapter 2.9 of \cite{vaart1997weak}); a standard multiplier choice is the Mammen two-point distribution:
\[
    \mathbb{P}(V_i = 0.5(1-\sqrt{5}))=b, \quad \mathbb{P}(V_i = 0.5(1+\sqrt{5}))=1-b,
\]
where $b = (1+\sqrt{5})/(2\sqrt{5})$; see \cite{mammen1993bootstrap} for details.

\paragraph{The bootstrap statistics.}
Write $\odot$ for the element-wise (Hadamard) product. The bootstrap analogues of the test statistics are
\begin{align*}
    & n \tilde T_{\mathrm{proj}, \perp}^{(1)} = n (\hat{\boldsymbol{\Pi}} (\boldsymbol{\varepsilon}_{\hat \theta} \odot \boldsymbol{V}))^\top \left(\boldsymbol{K} + n \lambda \boldsymbol{I}\right)^{-1} \boldsymbol{K} \left(\boldsymbol{K} + n \lambda \boldsymbol{I}\right)^{-1} (\hat{\boldsymbol{\Pi}} (\boldsymbol{\varepsilon}_{\hat \theta} \odot \boldsymbol{V})), \\
    & n \tilde T_{\mathrm{proj}, \perp}^{(2)} = (\hat{\boldsymbol{\Pi}} (\boldsymbol{\varepsilon}_{\hat \theta} \odot \boldsymbol{V}))^\top \left(\boldsymbol{K} + n \lambda \boldsymbol{I}\right)^{-1} \boldsymbol{K} (\hat{\boldsymbol{\Pi}} (\boldsymbol{\varepsilon}_{\hat \theta} \odot \boldsymbol{V})),
\end{align*}
where $\boldsymbol{V} = (v_1, \ldots, v_n)^\top$.

\begin{theorem}
    \label{th:bootstrap_consistency}
    Under Assumptions~\ref{ass:B}, \ref{ass:C}, \ref{ass:D}, \ref{ass:E}, and~\ref{ass:F}, the multiplier bootstrap statistics converge conditionally on the data, in probability, to the projected null laws:
    \begin{align*}
        & n \tilde T_{\mathrm{proj}, \perp}^{(1)} \xrightarrow{d} \sum_{i\geq 1}  \frac{\mu_i}{(\mu_i+\lambda)^2} (Z_i^\perp)^2, \\
        & n \tilde T_{\mathrm{proj}, \perp}^{(2)} \xrightarrow{d} \sum_{i\geq 1}  \frac{\mu_i}{\mu_i+\lambda} (Z_i^\perp)^2,
    \end{align*}
    where $\{Z_i^\perp\}_{i\ge 1}$ are as in Corollary~\ref{cor:proj_null} and $\xrightarrow{d}$ denotes weak convergence of the conditional law given the data, holding \emph{in probability with respect to the data}. Conditional on the data, the multipliers preserve the joint covariance structure of $(\xi, W)$ (Remark~\ref{rem:proj_null}), so the bootstrap law coincides with the projected null law of Corollary~\ref{cor:proj_null}$\,$.
\end{theorem}

\begin{proof}
    See Appendix~\ref{proof:thm6}.
\end{proof}

In practice the bootstrap law is evaluated by Monte Carlo over many multiplier draws; this is the implementation used in the simulations of Section~\ref{sec:simulation}, and its sample-splitting version, valid under data-driven tuning, appears in Section~\ref{sec:nuisance}.

\section{Nuisance Parameter Selection}
\label{sec:nuisance}

The test statistics introduced in Section~\ref{sec:stat} depend on two tuning parameters that act as nuisance parameters for the testing problem: the bandwidth $\gamma$ of the Gaussian kernel $k_\gamma(x, y) = \exp(-\gamma \|x - y\|^2)$ and the ridge regularization parameter $\lambda$. Neither parameter enters the hypothesis, yet both must be fixed before the test statistic is computed, and the classical practice of fixing them a priori, typically $\gamma = 0.5$ or a median rule, is consequential: size and power are highly sensitive to this choice, and the standard defaults can fail badly in higher dimensions (Section~\ref{sec:simulation}). This section develops the theory when the data choose the tuning parameters instead, selecting them by cross-validation on a separate sample; the analysis characterizes the asymptotic behavior of the resulting selector and its effect on the two quantities that define a test, its size and its power. Relative to \cite{escanciano2024gaussian}, whose Gaussian-process class holds the kernel fixed, the present section develops the missing tuning layer for the projection family; Remark~\ref{rem:gp_class} records the precise relation.

In outline: Section~\ref{sec:cv_selection} sets out the splitting protocol, the criterion, and the selector; Section~\ref{sec:uniform_limit} proves the process-level uniform limits on which the analysis rests; Sections~\ref{sec:null_dist}, \ref{sec:fixed_consistency}, and \ref{sec:local_power} then analyze the implemented test regime by regime: its null distribution and size control, its consistency under fixed alternatives, and the laws of the selector and of the implemented statistic under local alternatives; and Section~\ref{sec:energy} closes with an exact energy identity explaining why cross-validation favors configurations favorable for power.

We focus on the Gaussian kernel throughout; the results should extend to other shift-invariant kernels.

\subsection{The Cross-Validation Selection Protocol}
\label{sec:cv_selection}

Both tuning parameters enter the limiting null distributions of Section~\ref{sec:stat} through spectral weights such as $\mu_i(\gamma)/(\mu_i(\gamma)+\lambda)^2$, where $\{\mu_i(\gamma)\}_{i \ge 1}$ are the eigenvalues of the integral operator $\boldsymbol{L}_{k_\gamma}$: the bandwidth determines the spectrum, and hence the class of departures against which the test has power, while the regularization parameter determines the shrinkage applied to each spectral direction. We select them jointly by minimizing a $k$-fold CV estimate of the prediction risk of the KRR fit. Two further components of the protocol are essential: the tuning choice is delegated to a dedicated selection sample, so that the selector and the test statistic are driven by independent noise, and the admissible set is a compact subset bounded away from zero in $\lambda$.

\paragraph{Sample splitting.}
A tuning parameter selected from the same data that are later used for testing invalidates the pointwise asymptotic argument underlying the bootstrap: under the null, the selector and the test statistic would be driven by the same noise. We therefore partition the sample into two independent subsamples,
\begin{align*}
& \mathcal{D}_m^{\mathrm{sel}} = \{(X_i^{\mathrm{sel}}, Y_i^{\mathrm{sel}})\}_{i=1}^m && \text{(selection sample)}, \\
& \mathcal{D}_n^{\mathrm{test}} = \{(X_j^{\mathrm{test}}, Y_j^{\mathrm{test}})\}_{j=1}^n && \text{(testing sample)},
\end{align*}
of total size $N = m + n$, with $\mathcal{D}_m^{\mathrm{sel}} \perp\!\!\!\perp \mathcal{D}_n^{\mathrm{test}}$ and $m, n \to \infty$ at arbitrary relative rates. Let $c := \lim m/(m+n) \in [0,1]$ denote the asymptotic selection fraction. Throughout this section, Pitman (local) alternatives are expressed at the total-sample rate, with shift $R(X)/\sqrt{m+n}$; this is the testing problem of Theorem~\ref{th:proj_local} with the sample size $n$ replaced by the total $N = m+n$, and the shift direction $R$ carries the standing second-moment convention $\mathbb{E}[R^2]<\infty$ throughout. The selection sample is used only to choose the tuning; the testing sample only to compute the statistic and the bootstrap at the chosen value. The structural parameter $\hat{\theta}$ and the Neyman-orthogonalized residuals $\boldsymbol{\varepsilon}_\perp = \hat{\boldsymbol{\Pi}}\boldsymbol{\varepsilon}_{\hat{\theta}}$ of Section~\ref{sec:estimation_effect} are formed on each subsample. We write $\varepsilon_{\perp,i}^{\mathrm{sel}}$, $i = 1, \ldots, m$, for the orthogonalized residuals on the selection sample.

\paragraph{The criterion and the selector.}
Let $\eta := (\gamma, \lambda)$ denote the tuning parameter and let
\[
\Upsilon := [\gamma_L, \gamma_U] \times [\lambda_L, \lambda_U], \qquad 0 < \gamma_L < \gamma_U < \infty, \quad 0 < \lambda_L < \lambda_U < \infty,
\]
be a compact set of admissible values, discretized in practice as a finite grid $\Upsilon_0 \subset \Upsilon$, fixed as $m, n \to \infty$. The lower bound $\lambda_L > 0$ meets the strict-positivity and fixedness requirement of Assumption~\ref{ass:B}(iv).

Let $\{\mathcal{I}_1, \ldots, \mathcal{I}_k\}$ be a partition of $\{1, \ldots, m\}$ into $k$ folds of approximately equal size, and let $\mathcal{I}(i)$ denote the fold containing observation $i$. For each fold $r$ and each candidate $\eta \in \Upsilon$, let $\hat{w}_\eta^{(-r)}$ denote the KRR estimator trained on the selection observations outside fold $r$: $\hat{w}_\eta^{(-r)} = \left(C_\gamma^{(-r)} + \lambda I\right)^{-1} \hat{m}^{(-r)}$, where $C_\gamma^{(-r)} = \frac{1}{m_r} \sum_{i \notin \mathcal{I}_r} k_\gamma(X_i^{\mathrm{sel}}, \cdot) \otimes k_\gamma(X_i^{\mathrm{sel}}, \cdot)$, and $\hat{m}^{(-r)} = \frac{1}{m_r} \sum_{i \notin \mathcal{I}_r} \varepsilon_{\perp,i}^{\mathrm{sel}} \, k_\gamma(X_i^{\mathrm{sel}}, \cdot)$, with $m_r := m - |\mathcal{I}_r|$. The CV criterion is the average held-out squared prediction error,
\begin{equation}
\mathrm{CV}_m(\eta) := \frac{1}{m} \sum_{r=1}^k \sum_{i \in \mathcal{I}_r} \left( \varepsilon_{\perp,i}^{\mathrm{sel}} - \hat{w}_\eta^{(-r)}(X_i^{\mathrm{sel}}) \right)^2,
\label{eq:cv_criterion}
\end{equation}
and the selector minimizes the criterion over the admissible set:
\begin{equation}
\hat{\eta}_m := (\hat{\gamma}_m, \hat{\lambda}_m) = \arg\min_{\eta \in \Upsilon} \mathrm{CV}_m(\eta).
\label{eq:cv_selector}
\end{equation}

\begin{remark}
\textbf{(Leave-one-out analysis in the theory.)} Although the criterion \eqref{eq:cv_criterion} is stated for a general number $k$ of folds, the asymptotic analysis of the following subsections (and of Appendix~\ref{proof:selector_validity}) is carried out for the \emph{leave-one-out} criterion, i.e.\ $k = m$ folds of size one. In that case each fold is a single observation, the training set for observation $i$ is $\{1, \ldots, m\} \setminus \{i\}$, and the criterion \eqref{eq:cv_criterion} specializes to
\begin{equation}
    \mathrm{CV}_m^{\mathrm{LOO}}(\eta) := \frac{1}{m}\sum_{i=1}^m \big(\varepsilon_{\perp,i}^{\mathrm{sel}} - \hat{w}_\eta^{(-i)}(X_i^{\mathrm{sel}})\big)^2,
    \label{eq:cv_loo}
\end{equation}
where $\hat{w}_\eta^{(-i)} := \big(C_\gamma^{(-i)} + \lambda I\big)^{-1}\hat{m}^{(-i)}$ is the KRR fit trained on the $m-1$ observations other than $i$, with $C_\gamma^{(-i)} := \frac{1}{m-1}\sum_{j \ne i} k_\gamma(X_j^{\mathrm{sel}},\cdot) \otimes k_\gamma(X_j^{\mathrm{sel}},\cdot)$ and $\hat{m}^{(-i)} := \frac{1}{m-1}\sum_{j \ne i} \varepsilon_{\perp,j}^{\mathrm{sel}} k_\gamma(X_j^{\mathrm{sel}},\cdot)$ the empirical covariance and mean embedding over $j \ne i$. The general $k$-fold case is recovered in Remark~\ref{rem:fold_proof} of the appendix.
\end{remark}

\subsection{The Uniform Limit Theory}
\label{sec:uniform_limit}

The selection rule \eqref{eq:cv_selector} is an \emph{argmin} over $\Upsilon$, and the test statistic is evaluated at the selected value, so pointwise limit theorems at fixed $\eta$ cannot describe either object: the selector's limit law requires the whole criterion process to converge (argmin continuous mapping theorem), and the limit of $S_n^\perp(\hat{\eta}_m)$ requires the statistic \emph{process} to converge uniformly, so that the random evaluation can be passed to the limit. This subsection assembles the population objects in which the limits are expressed and proves the two process-level results; we analyze the first projection statistic throughout, and every result of this section holds for the second statistic as well.

For readability we write $\varepsilon_{\perp,i} := \varepsilon_{\perp,i}^{\mathrm{sel}}$ for the residuals on the selection sample and $S_n^\perp(\eta) := n\hat{T}_{\mathrm{proj},\perp}^{(1)}(\eta)$ for the projection statistic of Section~\ref{sec:estimation_effect} evaluated at $\eta = (\gamma, \lambda)$, i.e., the quadratic form in $\boldsymbol{\varepsilon}_\perp$ with weight matrix $W_n(\eta) := (\boldsymbol{K}_\gamma + n\lambda\boldsymbol{I})^{-1}\boldsymbol{K}_\gamma(\boldsymbol{K}_\gamma + n\lambda\boldsymbol{I})^{-1}$, where $\boldsymbol{K}_\gamma$ is the Gram matrix of $k_\gamma$ on the sample in question.

\paragraph{Setup and standing conditions.}
The results below are expressed in terms of population spectral quantities of the kernel family $k_\gamma$. Let $\{\phi_i(\cdot;\gamma)\}_{i \ge 1}$ and $\{\mu_i(\gamma)\}_{i \ge 1}$ denote the eigenfunctions and eigenvalues of the integral operator $\boldsymbol{L}_{k_\gamma}$ (Section~\ref{sec:equivalence}), let $C_\gamma := \mathbb{E}[k_\gamma(X,\cdot) \otimes k_\gamma(X,\cdot)]$, and define the smoothing kernel
\begin{equation}
    g_\eta(x,y) := \sum_{\ell \ge 1} \frac{\mu_\ell(\gamma)}{\mu_\ell(\gamma) + \lambda}\,\phi_\ell(x;\gamma)\,\phi_\ell(y;\gamma),
    \label{eq:g_eta}
\end{equation}
the kernel of the population smoothing operator $S_\lambda := (C_\gamma + \lambda I)^{-1}C_\gamma$, in the sense of an integral kernel on $L^2(\mathbb{P}_X)$: for every $f \in L^2(\mathbb{P}_X)$,
\[
    (S_\lambda f)(x) \;=\; \int g_\eta(x,y)\,f(y)\,\mathrm{d}\mathbb{P}_X(y).
\]
Here $C_\gamma$ acts on $\mathcal{H}_\gamma$ and $\boldsymbol{L}_{k_\gamma}$ on $L^2(\mathbb{P}_X)$; they are intertwined by the inclusion $\iota:\mathcal{H}_\gamma\hookrightarrow L^2(\mathbb{P}_X)$ through $C_\gamma = \iota^*\iota$ and $\boldsymbol{L}_{k_\gamma} = \iota\,\iota^*$. The nonzero spectrum $\{\mu_i(\gamma)\}_{i \ge 1}$ is common to both: the eigenfunctions $\phi_i(\cdot;\gamma)$, orthonormal in $L^2(\mathbb{P}_X)$, belong to $\mathcal{H}_\gamma$ and satisfy $C_\gamma\phi_i = \mu_i(\gamma)\,\phi_i$ in $\mathcal{H}_\gamma$, with the dual normalization $\|\phi_i\|_{\mathcal{H}_\gamma} = \mu_i(\gamma)^{-1/2}$; all spectral quantities below are read from this pair.
The kernel $g_\eta$ enters the analysis through the linearized leave-one-out prediction $\hat{w}^{(-i)}_\eta(X_i^{\mathrm{sel}}) = \frac{1}{m-1}\sum_{j \ne i}\varepsilon_{\perp,j}\, g_\eta(X_i^{\mathrm{sel}}, X_j^{\mathrm{sel}}) + r_i$ of Lemma~\ref{lem:loo_prediction}, where $r_i = O_p(m^{-1})$ under $H_0$ and local alternatives and $r_i = O_p(m^{-1/2})$ under fixed alternatives.

\noindent\textit{The alternative regimes.}\; The limiting objects are indexed by the mean of the residual function in the three regimes of the testing problem; throughout, the orthogonalized residuals decompose as
\[
    \varepsilon_{\perp,i} = f_m(X_i) + v_i + o_p(1),
\]
where $f_m \equiv 0$ under $H_0$, $f_m = \boldsymbol{\Pi}f$ under a fixed alternative, and $f_m = \delta_m h$ under the local alternative, and the noise satisfies $\mathbb{E}[v \mid X] = 0$ and $\mathbb{E}[v^2 \mid X] = \sigma^2(X)$. Under a fixed alternative, the effective signal of the orthogonalized residual is $(\boldsymbol{\Pi}f)(x) := \mathbb{E}[\boldsymbol{\Pi}\varepsilon_0 \mid X = x]$, the image of the signal $f(x) := \mathbb{E}[\varepsilon_0 \mid X = x]$ of the Introduction under the projection $\boldsymbol{\Pi}$ of Section~\ref{sec:estimation_effect}. Thus $\boldsymbol{\Pi}f \equiv 0$ under $H_0$, while $\boldsymbol{\Pi}f \neq 0$ exactly when the misspecification survives the projection, that is, when $f \notin \operatorname{span}\{g_{\theta_0,1}(X), \ldots, g_{\theta_0,d_\theta}(X)\}$. Write $\delta_i(\gamma) := \mathbb{E}[(\boldsymbol{\Pi}f)(X)\phi_i(X;\gamma)]$, let $w^*_\eta := (C_\gamma + \lambda I)^{-1}\boldsymbol{L}_{k_\gamma}(\boldsymbol{\Pi}f)$ denote the population KRR solution, and set
\begin{equation}
    A(\eta) := \big\|\boldsymbol{\Pi}f - w^*_\eta\big\|^2_{L^2} = \sum_{i \ge 1} \frac{\lambda^2\,\delta_i(\gamma)^2}{(\mu_i(\gamma) + \lambda)^2},
    \label{eq:A_eta}
\end{equation}
the approximation error of the population KRR fit. For the local analysis, fix $h \in L^2(\mathbb{P}_X) \setminus \{0\}$, write $h_i(\gamma) := \mathbb{E}[h(X)\phi_i(X;\gamma)]$, $A_h(\eta) := \sum_i \lambda^2 h_i(\gamma)^2/(\mu_i(\gamma)+\lambda)^2$, $P_\eta h(x) := \sum_i \mu_i(\gamma)h_i(\gamma)\phi_i(x;\gamma)/(\mu_i(\gamma)+\lambda)$, $a_\eta := h - P_\eta h$, $b_\eta(x) := \mathbb{E}[a_\eta(X)\,g_\eta(X,x)]$, and let $G$ be a mean-zero Gaussian process on $\Upsilon$ with covariance kernel
\begin{equation}
    \operatorname{Cov}\big(G(\eta), G(\eta')\big) = \mathbb{E}\Big[\sigma^2(X)\,\big(\Pi(a_\eta - b_\eta)\big)(X)\,\big(\Pi(a_{\eta'} - b_{\eta'})\big)(X)\Big] + \Omega_G(\eta, \eta'),
    \label{eq:G_cov}
\end{equation}
where $\Omega_G(\eta,\eta')$ collects the covariance contribution of the estimation effect of the projection coefficient; its exact form is not needed, since every conclusion drawn from $G$ is a level or argmin statement.

\noindent\textit{The criterion decomposition.}\; In the null and local-alternative regimes, expanding the square of the leave-one-out error and using the linearization of Lemma~\ref{lem:loo_linearization} (the criterion \eqref{eq:cv_loo} then decomposes into the four terms (I)--(IV) defined there, stated for the infeasible orthogonalized residuals $\varepsilon_{\perp}^{0}=\hat{\boldsymbol\Pi}\boldsymbol\varepsilon_{0}$, with the transfer from the feasible residuals established by Lemma~\ref{lem:repl_cost}),
\[
    \mathrm{CV}_m^{\mathrm{LOO}}(\eta) = \underbrace{\frac{1}{m}\sum_{i=1}^m \varepsilon_{\perp,i}^2}_{(\mathrm{I})}
    \; - \; \underbrace{2\,U_m(\eta)}_{(\mathrm{II})}
    \; + \; \underbrace{Q_m(\eta)}_{(\mathrm{III})} + \underbrace{R_m(\eta)}_{(\mathrm{IV})} + o_p(m^{-1}),
\]
where $U_m(\eta) := \frac{1}{m(m-1)}\sum_{i \ne j}\varepsilon_{\perp,i}\varepsilon_{\perp,j}\,g_\eta(X_i,X_j)$ collects the off-diagonal products of the cross term; $Q_m(\eta) := \frac{1}{m(m-1)^2}\sum_{i \ne j}\varepsilon_{\perp,j}^2\,g_\eta(X_i,X_j)^2$ the diagonal part of the squared prediction; and $R_m(\eta) := \frac{2}{m(m-1)^2}\sum_{i=1}^m \sum_{\substack{j<k\\ j,k\ne i}} g_\eta(X_i, X_j)\, g_\eta(X_i, X_k)\,\varepsilon_{\perp,j}\varepsilon_{\perp,k}$ the off-diagonal ($j \ne k$) part of the squared prediction. Under a fixed alternative, the linearization remainder is only $O_p(m^{-1/2})$, negligible at the $\sqrt m$ scale.
The diagonal part (III) is governed by the noise variance: let $\bar{\sigma}^2 := \mathbb{E}[\sigma^2(X)]$, $\tau_\ell(\gamma) := \mathbb{E}[\sigma^2(X)\phi_\ell(X;\gamma)^2]$, and $V(\eta) := \sum_{\ell \ge 1} \frac{\mu_\ell(\gamma)^2}{(\mu_\ell(\gamma) + \lambda)^2}\,\tau_\ell(\gamma)$, the leading conditional mean of the diagonal part at the $m^{-1}$ scale.

\noindent\textit{The correlation structure of the residuals.}\; The selection-sample residuals are the Neyman-orthogonalized residuals of Section~\ref{sec:estimation_effect}, whose conditional covariance given $\mathbf{X}$ is $\hat{\boldsymbol{\Pi}}\operatorname{diag}(\sigma^2(X_j^{\mathrm{sel}}))\hat{\boldsymbol{\Pi}}$ rather than the diagonal matrix $D := \operatorname{diag}(\sigma^2(X_j^{\mathrm{sel}}))$: for $j \ne k$,
\[
    c_{jk} := \mathbb{E}\big[\varepsilon_{\perp,j}^{\mathrm{sel}}\,\varepsilon_{\perp,k}^{\mathrm{sel}} \mid \mathbf{X}\big] = -\sigma^2(X_j^{\mathrm{sel}})\,\hat g_j^\top M\hat g_k - \sigma^2(X_k^{\mathrm{sel}})\,\hat g_j^\top M\hat g_k + \hat g_j^\top M(\hat{\mathbb{G}}^\top D\hat{\mathbb{G}})M\hat g_k = O_p(m^{-1}),
\]
with $M := (\hat{\mathbb{G}}^\top\hat{\mathbb{G}})^{-1}$ and $\hat g_j := \hat g_{\hat\theta}(X_j^{\mathrm{sel}})$. The expansion follows from $\hat{\boldsymbol{\Pi}} = \boldsymbol{I}_m - \hat{\mathbb{G}}M\hat{\mathbb{G}}^\top$: since the noise has conditional covariance $D$ given $\mathbf{X}$, the orthogonalized residual vector has conditional covariance $\hat{\boldsymbol{\Pi}}D\hat{\boldsymbol{\Pi}}$, and expanding $(\hat{\boldsymbol{\Pi}}D\hat{\boldsymbol{\Pi}})_{jk}$ for $j \ne k$, the diagonal matrix $D$ contributes nothing, the two cross terms give $-(\hat{\mathbb{G}}M\hat{\mathbb{G}}^\top D)_{jk} = -\sigma^2(X_k^{\mathrm{sel}})\,\hat g_j^\top M\hat g_k$ and its transpose, and the quadratic term gives $\hat g_j^\top M(\hat{\mathbb{G}}^\top D\hat{\mathbb{G}})M\hat g_k$, the three displayed terms. The order is $O_p(m^{-1})$: $\hat{\mathbb{G}}^\top\hat{\mathbb{G}}/m \to \Gamma_{\theta_0}$ (Assumption~\ref{ass:E}(ii)) gives $M = O_p(m^{-1})$ and $\hat g_j^\top M\hat g_k = O_p(m^{-1})$, while $\hat{\mathbb{G}}^\top D\hat{\mathbb{G}} = \sum_j \sigma^2(X_j^{\mathrm{sel}})\hat g_j\hat g_j^\top = O_p(m)$, so each term is $O_p(m^{-1})$. The diagonal entries, by contrast, keep their leading value, $(\hat{\boldsymbol{\Pi}}D\hat{\boldsymbol{\Pi}})_{jj} = \sigma^2(X_j^{\mathrm{sel}}) + O_p(m^{-1})$: only the off-diagonal correlations act at the $m^{-1}$ scale. In the location model $g_{\theta_0} \equiv 1$ under homoscedasticity, $c_{jk} = -\sigma^2/m$ exactly, the classical negative covariance of centered residuals. Hence $m\,c_{jk} \to c^*(X_j^{\mathrm{sel}}, X_k^{\mathrm{sel}})$ where, with $\Gamma_{\theta_0} := \mathbb{E}[g_{\theta_0}(X)g_{\theta_0}(X)^\top]$ (Assumption~\ref{ass:E}(ii)) and $g_{\theta_0}$ the conditional score of Section~\ref{sec:estimation_effect},
\[
\begin{split}
    c^*(x,x') := {}&-g_{\theta_0}(x)^\top\Gamma_{\theta_0}^{-1}g_{\theta_0}(x')\big(\sigma^2(x)+\sigma^2(x')\big) + g_{\theta_0}(x)^\top\Gamma_{\theta_0}^{-1}\mathbb{E}\big[\sigma^2(X)g_{\theta_0}(X)g_{\theta_0}(X)^\top\big]\Gamma_{\theta_0}^{-1}g_{\theta_0}(x')
\end{split}
\]
Under homoscedasticity, $c^* = -\sigma^2\,g_{\theta_0}^\top\Gamma_{\theta_0}^{-1}g_{\theta_0}$. Under $H_0$ the conditional means of the two off-diagonal terms (II) and (IV) are driven by these correlations, $c_{ij} = O_p(m^{-1})$, and converge to the deterministic projection correction $\nu(\eta) := \mathbb{E}\big[c^*(X_1,X_2)\,g_\eta(X_1,X_2)\big]$, $\rho(\eta) := \mathbb{E}\big[c^*(X_2,X_3)\,g_\eta(X_1,X_2)\,g_\eta(X_1,X_3)\big]$, and $\Delta(\eta) := \rho(\eta) - 2\nu(\eta)$: indeed, $m\,\mathbb{E}\big[U_m(\eta)\mid\mathbf X\big] = \frac{1}{m-1}\sum_{i \ne j} c_{ij}\,g_\eta(X_i,X_j) \;\longrightarrow\; \nu(\eta)$ (the average $\frac{1}{m(m-1)}\sum_{i \ne j} c^*(X_i,X_j)g_\eta(X_i,X_j)$ is the $U$-statistic for $\nu(\eta)$, with $c_{ij} = m^{-1}c^*(X_i,X_j) + o_p(m^{-1})$) and $m\,\mathbb{E}\big[R_m(\eta)\mid\mathbf X\big] = \frac{2}{(m-1)^2}\sum_{i=1}^m \sum_{\substack{j<k\\ j,k\ne i}} c_{jk}\,g_\eta(X_i,X_j)g_\eta(X_i,X_k) \;\longrightarrow\; \rho(\eta)$ (each pair $(j,k)$ enters through the $m-2$ indices $i \ne j,k$ with weights averaging to $(g_\eta \circ g_\eta)(X_j,X_k)$, the finite-$m$ factor being $(m-2)/(m-1)$). The diagonal term (III) is unaffected at this scale, $mQ_m(\eta) \xrightarrow{p} V(\eta)$, so the deterministic parts of $-2mU_m + mQ_m + mR_m$ assemble into $-2\nu(\eta) + V(\eta) + \rho(\eta)$: the correction $\Delta(\eta)$ is of the same order as $V(\eta)$, so it cannot be absorbed into the remainder.

\noindent\textit{The pseudo-true convention under $H_1$.}\; Under a fixed alternative the estimated
parameter converges in general to the pseudo-true value $\theta^*\neq\theta_0$, and all
$\theta_0$-subscripted objects of this section, namely the score $g_{\theta_0}$, the projection
$\boldsymbol\Pi$ onto the complement of $\operatorname{span}\{g_{\theta_0}\}$, the signal
$f:=\mathbb E[\varepsilon(\cdot;\theta_0)\mid X]$, and the rate condition of
Assumption~\ref{ass:D}(iii), are then understood as evaluated at $\theta^*$:
$g_{\theta^*}$, $\boldsymbol\Pi_{\theta^*}$, $f^*:=\mathbb E[\varepsilon(\cdot;\theta^*)\mid X]$,
the rate condition being required around $\theta^*$. This is the convention under which Corollary~\ref{cor:proj_null}'s fixed-alternative
branch, Theorem~\ref{th:sel_fixed}, and Theorem~\ref{th:power_fixed} are stated; their proofs
(Appendix~\ref{proof:selector_validity}) run the corresponding $\theta_0$-arguments verbatim with
$\theta_0$ replaced by $\theta^*$.

\noindent\textit{The projected chaos fluctuations.}\; The fluctuations of $mU_m$ and $mR_m$ around these conditional means, governed by the fourth-moment structure of the projected noise, form the projected chaoses of the kernels below. For a kernel $f$, write $(\Pi f \Pi)(x,y) := \Pi_x[\Pi_y f](x,y)$ for the kernel obtained by applying the projection to each argument. The quadratic fluctuation of the CV criterion involves the composition kernel $(g_\eta\circ g_\eta)(x,y) := \int g_\eta(x,z)\,g_\eta(z,y)\,\mathrm{d}\mathbb{P}_X(z)$ and the projected kernels
\[
    K^{\mathrm{sym},\perp}_\eta(x,y) := \sigma(x)\,(\Pi g_\eta \Pi)(x,y)\,\sigma(y), \qquad
    K^{D,\perp}_\eta(x,y) := \sigma(x)\,\big(\Pi (g_\eta\circ g_\eta) \Pi\big)(x,y)\,\sigma(y),
\]
with eigenvalues $\{\lambda^{\perp}_k(\eta)\}_{k \ge 1}$ and $\{\lambda^{D,\perp}_k(\eta)\}_{k \ge 1}$ of the corresponding Hilbert--Schmidt operators on $L^2(\mathbb{P}_X)$; the projections act on the kernels and do not depend on $\lambda$.

\begin{assumption}[Uniform regularity]\label{ass:G}
(i) $\mathbb{E}\|X\|^{2} < \infty$; (ii) the maps $\gamma \mapsto \mu_i(\gamma)$ and $\gamma \mapsto \phi_i(\cdot;\gamma)$ are continuous on $[\gamma_L, \gamma_U]$; (iii) the conditional variance of the residual is bounded away from zero and above, $\sigma_{\min}^2 \le \sigma^2(X) := \mathbb{E}[v^2 \mid X] \le \sigma_{\max}^2 < \infty$ almost surely; (iv) the kernel family $\{k_\gamma\}$ is integrally strictly positive definite, so that all eigenvalues $\mu_i(\gamma) > 0$; (v) there exist $p\in(2,4]$ and $\bar{M}_{2p} < \infty$ with $\sup_{x \in \operatorname{supp}\mathbb{P}_X}\mathbb{E}[v^{2p} \mid X = x] \le \bar{M}_{2p}$, a conditional moment of the centered noise of order strictly above the fourth.
\end{assumption}

Assumption~\ref{ass:G} states the regularity conditions of the uniform-in-$\eta$ theory, and the clauses call for the following comments. The eigenpair conditions of clause (ii) hold for the Gaussian family by Proposition~\ref{prop:eigenpair_cont} (Appendix~\ref{proof:aux_uniform}), which derives them from the real-analyticity of $\gamma \mapsto \boldsymbol{L}_{k_\gamma}$ in Hilbert--Schmidt norm and Rellich--Kato perturbation theory \citep{kato1980}. The integral strict positive definiteness of clause (iv) is satisfied by the Gaussian kernel. Clause (v) is needed for the uniform-in-observation moment bounds of the leave-one-out linearization, where a sample maximum of conditional moments must be controlled deterministically; its $(2p)$-moment clause supplies the order-$p$ envelopes of the gradient-moment inputs of the criterion-process proofs and the projected-residual moment bounds consumed in the proofs of Theorems~\ref{th:stat_uniform} and~\ref{th:cv_uniform} and of Lemma~\ref{lem:remainder_cv} (Appendix~\ref{aux:cv_lemmas}); in those bounds a fourth moment of a centered quadratic form consumes noise moments of order $2p$, and the requirement $p>2$ is forced by the two-dimensional index set over which the suprema are taken.

The two results operate at different levels of uniformity. Theorem~\ref{th:stat_uniform} is an $\ell^\infty(\Upsilon)$ statement: its stochastic equicontinuity and tightness let the random evaluation pass to the limit, giving $S_n^\perp(\hat{\eta}_m) \xrightarrow{d} S_\infty(\eta^*)$ (resp.\ $S_\infty(\eta^*_h;h_c)$) once the selector converges; it feeds Propositions~\ref{prop:joint_null} and~\ref{prop:joint_local} and Theorem~\ref{th:power_fixed}. Theorem~\ref{th:cv_uniform} governs the \emph{selection} side, identifying the limit criterion process whose deterministic drifts, $V(\eta)+\Delta(\eta)$ under $H_0$, $\bar\sigma^2 + A(\eta)$ under fixed alternatives, and $V(\eta)+\Delta(\eta)+cA_h(\eta)$ under local alternatives, together with the projected chaoses and, where present, the $\sqrt{m}$ Gaussian fluctuations, determine the selector's law and the bias-variance trade-off of Section~\ref{sec:energy}.

\begin{theorem}[Uniform weak convergence of the statistic process]
\label{th:stat_uniform}
Under Assumptions~\ref{ass:A}--\ref{ass:F} and~\ref{ass:G}, under $H_0$, $S_n^\perp \rightsquigarrow S_\infty$ in $\ell^\infty(\Upsilon)$, where $S_\infty$ has continuous sample paths and, for each fixed $\eta$,
\[
    S_\infty(\eta) = \sum_{i \ge 1} \frac{\mu_i(\gamma)}{(\mu_i(\gamma)+\lambda)^2}\,(Z_i^\perp)^2,
\]
with $(Z_1^\perp,Z_2^\perp,\dots)^\top \sim \mathcal{N}(\boldsymbol{0}, \Sigma^\perp(\gamma))$ as in Corollary~\ref{cor:proj_null} (in the eigenbasis $\{\phi_i(\cdot;\gamma)\}$ of $\boldsymbol{L}_{k_\gamma}$). Under the local alternative at the total-sample rate with shift $R/\sqrt{m+n}$ (Section~\ref{sec:cv_selection}), $S_n^\perp \rightsquigarrow S_\infty(\cdot;h_c)$ in $\ell^\infty(\Upsilon)$ with $h_c := \sqrt{1-c}\,h$, $h := \boldsymbol{\Pi}R$, $c := \lim m/(m+n)$, $h_{c,i}(\gamma) := \mathbb E[h_c(X)\phi_i(X;\gamma)]$ (Corollary~\ref{cor:proj_local}), and
\[
    S_\infty(\eta;h_c) := \sum_{i \ge 1} \frac{\mu_i(\gamma)}{(\mu_i(\gamma)+\lambda)^2}\,\big(Z_i^\perp + h_{c,i}(\gamma)\big)^2,
\]
(cf.~Corollary~\ref{cor:proj_local}); the factor $\sqrt{1-c}$ is the asymptotic testing share of the total sample.
\end{theorem}

\begin{proof}
    See Appendix~\ref{proof:stat_uniform}.
\end{proof}

\begin{remark}[Relation to the Gaussian-process class]\label{rem:gp_class}
The construction of this section places the projection statistics within the Gaussian-process approach to model checks of \cite{escanciano2024gaussian}, with one structural difference. The two constructions share their recipe and their limit type: both are quadratic forms in Neyman-orthogonally projected residuals against a kernel, both admit a multiplier bootstrap that does not re-estimate the nuisance parameter, and both null limits are second-order weighted chi-square chaos laws determined by a kernel spectrum, the orthogonal projection onto the score orthocomplement, and the conditional variance of the noise. The difference is the plug-in, and with it the weights: the statistic of \cite{escanciano2024gaussian} evaluates the bivariate kernel $K(X_i,X_j)$ at the sample points, and its null weights are the eigenvalues of the $\sigma^2$-measured orthogonalized kernel operator on $L^2(\mu)$, $\mathrm{d}\mu = \sigma^2\,\mathrm{d}\mathbb P_X$ (Theorem 4.3 there); ours evaluate the empirical resolvent $(\boldsymbol K_\gamma/n)(\boldsymbol K_\gamma/n+\lambda I)^{-2}$, and the null weights of Theorem~\ref{th:stat_uniform} are the filtered spectral weights $w_i(\eta) = \mu_i(\gamma)/(\mu_i(\gamma)+\lambda)^2$ attached to the projected noise coordinates. The empirical-resolvent plug-in is what makes leave-one-out cross-validation intrinsic to the family, and the data-driven selection of the kernel-side tuning, which \cite{escanciano2024gaussian} pose as an open problem for their class, is what the present section supplies for it. The resolution is also not the routine extension that the derivative-interchange route of their Remark 2 suggests: the criterion fluctuations live at the $m^{-1}$ scale with degenerate chaos limits (Theorem~\ref{th:cv_uniform}), the projection-induced covariance corrections enter the selector's drift at the same order as the variance term, and the selector retains a non-degenerate random limit (Theorems~\ref{th:sel_h0} and~\ref{th:sel_local}).
\end{remark}

\begin{theorem}[Uniform convergence of the CV criterion processes]
\label{th:cv_uniform}
Under Assumptions~\ref{ass:A}--\ref{ass:F} and~\ref{ass:G}:
\textbf{(a)} Under $H_0$, the $\eta$-dependent part of the CV criterion is $O_p(m^{-1})$, and
\[
    m\Big[\mathrm{CV}_m(\eta) - \frac{1}{m}\sum_{i=1}^m (\varepsilon_{\perp,i}^{0})^{2}\Big] \;\rightsquigarrow\; \mathbb{H}(\eta) := -2\,\mathcal{C}^\perp(\eta) + V(\eta) + \mathcal{D}^\perp(\eta) + \Delta(\eta) \quad \text{in } \ell^\infty(\Upsilon),
\]
where $\mathcal{C}^\perp$ is the mean-zero Gaussian chaos of order two with $\operatorname{Cov}(\mathcal{C}^\perp(\eta), \mathcal{C}^\perp(\eta')) = 2\langle K^{\mathrm{sym},\perp}_\eta, K^{\mathrm{sym},\perp}_{\eta'}\rangle_{\mathrm{HS}}$; $\mathcal{D}^\perp$ is the mean-zero Gaussian chaos of order two with $\operatorname{Cov}(\mathcal{D}^\perp(\eta), \mathcal{D}^\perp(\eta')) = 2\langle K^{D,\perp}_\eta, K^{D,\perp}_{\eta'}\rangle_{\mathrm{HS}}$, with cross-covariance $\operatorname{Cov}(\mathcal{C}^\perp(\eta), \mathcal{D}^\perp(\eta')) = 2\langle K^{\mathrm{sym},\perp}_\eta, K^{D,\perp}_{\eta'}\rangle_{\mathrm{HS}}$, and $\Delta(\eta)$ is the deterministic projection correction of the setup. Here $\varepsilon_{\perp}^{0}=\hat{\boldsymbol\Pi}\varepsilon_{0}$ is the infeasible projected residual of Appendix~\ref{proof:cv_uniform}, at which the criterion in the display is evaluated (Lemma~\ref{lem:repl_cost}); centering at the feasible level $m^{-1}\sum_{i}\varepsilon_{\perp,i}^{2}$ instead shifts the left side by a constant that is free of $\eta$ and need not vanish under the sieve clause of Assumption~\ref{ass:D}(iii), so every argmin-based conclusion, in particular the selector laws, is unaffected.

\textbf{(b)} Under a fixed alternative ($\boldsymbol{\Pi}f \neq 0$ in $L^2(\mathbb P_X)$, $\lVert\boldsymbol{\Pi}f\rVert_\infty < \infty$), $\sup_{\eta \in \Upsilon}|\mathrm{CV}_m(\eta) - R_\infty(\eta)| = O_{p}(m^{-1/2})$ with $R_\infty(\eta) := \bar{\sigma}^2 + A(\eta)$.

\textbf{(c)} Under the local alternative $f_m = \delta_m h$ with $\delta_m \to 0$, $\sqrt{m}\,\delta_m \to a < \infty$, and $\mathbb{E}[h^4(X)] < \infty$, where $h=\boldsymbol{\Pi}h$ is an effective (score-annihilated) direction, the convention of the local-alternative pipeline,
\[
\begin{split}
    m\Big[\mathrm{CV}_m(\eta) - \frac{1}{m}\sum_{i=1}^m v_i^2\Big] &\;\rightsquigarrow\; \mathbb{H}_h^a(\eta) \\
    &:= -2\,\mathcal{C}^\perp(\eta) + V(\eta) + \mathcal{D}^\perp(\eta) + \Delta(\eta) + a^2 A_h(\eta) + 2a\,G(\eta)
\end{split}
\]
in $\ell^\infty(\Upsilon)$, where $v_i := \varepsilon_{\perp,i} - f_m(X_i)$ is the noise component of the residual.
\end{theorem}

\begin{proof}
    See Appendix~\ref{proof:cv_uniform}.
\end{proof}

Theorem~\ref{th:cv_uniform} identifies the limit criterion $\mathbb{H}$ (resp.\ $\mathbb{H}_h^a$) whose minimizer is the selector's limit, and Theorem~\ref{th:stat_uniform} identifies the limit process in which the statistic is evaluated. Sections~\ref{sec:null_dist}--\ref{sec:local_power} combine the two, one regime of the testing problem at a time: each regime contributes a law for the selector, and evaluating the statistic and the critical value at the selected tuning yields the size and the consistency of the implemented test and, under local alternatives, the limit law of its statistic.

Throughout Sections~\ref{sec:null_dist}--\ref{sec:local_power}, $\hat\eta_m$ denotes a random minimizer of the criterion over $\Upsilon$, that is, a measurable selection from the argmin set satisfying
$
    \mathrm{CV}_m(\hat\eta_m)=\inf_{\eta\in\Upsilon}\mathrm{CV}_m
$
almost surely; since $\mathrm{CV}_m$ has continuous sample paths on the compact rectangle, the argmin set is nonempty for every realization (Weierstrass). In the null and local regimes of Theorem~\ref{th:cv_uniform} the centered process $F_m$ displayed in Theorem~\ref{th:cv_uniform}(a) and (c) shares this argmin set, the centering summand being $\eta$-free; under a fixed alternative the selector analysis of Appendix~\ref{proof:selector_validity} works with $\mathrm{CV}_m$ directly.

\subsection{Null Distribution and Size Control}
\label{sec:null_dist}

The first property of the data-driven test is its behavior under $H_0$. The selector converges to a random limit, so the selected tuning remains random even asymptotically and must not be read as an estimate of an optimal value; the randomness is benign for size, because selection and testing use independent samples. This subsection assembles the chain in three steps: the null law of the selector (Theorem~\ref{th:sel_h0}), the resulting null law of the implemented statistic (Proposition~\ref{prop:joint_null}), and the size guarantee delivered by the multiplier bootstrap at the realized tuning (Theorem~\ref{th:validity_sel}), which requires no convergence of the selector at all.

\begin{theorem}[Selector under the null]
\label{th:sel_h0}
Suppose Assumptions~\ref{ass:A}--\ref{ass:F} and~\ref{ass:G} hold and the sample-splitting protocol of Section~\ref{sec:cv_selection} is in place. Under $H_0$, assume that the limit criterion $\mathbb{H}$ of Theorem~\ref{th:cv_uniform}(a) has an almost surely unique minimizer over $\Upsilon$. Then any random minimizer $\hat\eta_m$ of $\mathrm{CV}_m$ over $\Upsilon$ satisfies
\[
    \hat\eta_m \xrightarrow{d} \eta^* := \arg\min_{\eta \in \Upsilon} \mathbb{H}(\eta),
\]
Here $\eta^*$ is random and independent of the testing sample; its law depends on the kernel spectrum, the conditional variance function $\sigma^2(\cdot)$, and the joint law of the chaoses $\mathcal{C}^\perp$ and $\mathcal{D}^\perp$.
\end{theorem}

\begin{proof}
    See Appendix~\ref{proof:sel_h0}.
\end{proof}

\begin{proposition}[Joint convergence and evaluation under $H_0$]
\label{prop:joint_null}
Under the conditions of Theorem~\ref{th:sel_h0}, as $m, n \to \infty$,
\[
    \big(\hat{\eta}_m,\, S_n^\perp\big) \;\rightsquigarrow\; \big(\eta^*,\, S_\infty\big) \quad \text{in } \Upsilon \times \ell^\infty(\Upsilon),
\]
the limit components being independent because $\mathcal{D}_m^{\mathrm{sel}} \perp\!\!\!\perp \mathcal{D}_n^{\mathrm{test}}$. Consequently, by continuity of the sample paths of the limit process,
\[
    S_n^\perp(\hat{\eta}_m) \xrightarrow{d} S_\infty(\eta^*)
\]
\end{proposition}

\begin{proof}
    See Appendix~\ref{proof:joint_limit}.
\end{proof}

It remains to convert this limit law into size control. Random tuning does not break the bootstrap, provided selection and testing use independent samples: conditional on the selection sample, $\hat{\eta}_m$ is held fixed, and by split independence the law of the statistic at that point is its marginal law, so conditioning on the selection event is non-informative. It therefore suffices to construct a critical value that is pointwise valid at every fixed $\eta$ and evaluate it at the realized value; the multiplier bootstrap of Section~\ref{sec:estimation_effect} does exactly this, applied to the testing sample with the selected parameter held fixed: let $\boldsymbol{V} := (V_1, \ldots, V_n)^\top$ be an independent copy of the multiplier weights of Section~\ref{sec:estimation_effect}, independent of $\mathcal{D}_n^{\mathrm{test}}$ and of $\mathcal{D}_m^{\mathrm{sel}}$, and write
\begin{equation}
    \tilde{S}_n^\perp(\hat{\eta}_m) := n\,\tilde{T}_{\mathrm{proj},\perp}^{(1)}(\hat{\eta}_m)
    \label{eq:bootstrap_stat}
\end{equation}
for the bootstrap analogue of $S_n^\perp$ evaluated at $\hat{\eta}_m$: the quadratic form in $\hat{\boldsymbol{\Pi}}^{\mathrm{test}}(\boldsymbol{\varepsilon}_{\hat{\theta}}^{\mathrm{test}} \odot \boldsymbol{V})$ with weight matrix $W_n(\hat{\eta}_m)$. The bootstrap critical value is the conditional $(1-\alpha)$-quantile of the bootstrap law,
\[
    c_{n,\alpha}(\hat{\eta}_m) := \inf\Big\{t \ge 0 : \mathbb{P}_{\boldsymbol{V}}\big(\tilde{S}_n^\perp(\hat{\eta}_m) \le t \mid \mathcal{D}_n^{\mathrm{test}}, \hat{\eta}_m\big) \ge 1 - \alpha\Big\},
\]
and $c_\infty(\eta)$ denotes the $(1-\alpha)$-quantile of the null limit $S_\infty(\eta)$ of Theorem~\ref{th:stat_uniform}. By Theorem~\ref{th:bootstrap_consistency} the pointwise bootstrap is consistent at every fixed $\eta$, and the map $\eta \mapsto c_\infty(\eta)$ is continuous on $\Upsilon$ (Lemma~\ref{lem:quantile_cont}, Appendix~\ref{proof:validity_sel}); hence the pointwise bootstrap at the realized value estimates, by split independence, exactly the law of the statistic conditional on the selection event. No selector asymptotics and no resampling of the selection side are needed. 

\begin{theorem}[Bootstrap validity under sample splitting]
\label{th:validity_sel}
Under Assumptions~\ref{ass:A}--\ref{ass:F} and~\ref{ass:G} and the sample-splitting protocol of Section~\ref{sec:cv_selection}, under $H_0$, as $m, n \to \infty$ at arbitrary relative rates,
\[
    \mathbb{P}\big(S_n^\perp(\hat{\eta}_m) > c_{n,\alpha}(\hat{\eta}_m)\big) \;\longrightarrow\; \alpha.
\]
\end{theorem}

\begin{proof}
    See Appendix~\ref{proof:validity_sel}.
\end{proof}

\begin{remark}[No selector limit is needed.]
Theorem~\ref{th:validity_sel} does not require $\hat{\eta}_m$ to converge: the proof (Appendix~\ref{proof:validity_sel}) uses only the pointwise null limit, the pointwise bootstrap at the realized selector value, and split independence. The selector limits of Theorems~\ref{th:sel_h0}--\ref{th:sel_local} enter only the power analysis.
\end{remark}

\subsection{Consistency under Fixed Alternatives}
\label{sec:fixed_consistency}

Under a fixed alternative the criterion process converges uniformly to its deterministic limit $\bar\sigma^2 + A(\eta)$, and the selector concentrates on the minimizer of the approximation error:

\begin{theorem}[Selector under a fixed alternative]
\label{th:sel_fixed}
Suppose Assumptions~\ref{ass:A}--\ref{ass:F} and~\ref{ass:G} hold and $\boldsymbol{\Pi}f \neq 0$ in $L^2(\mathbb{P}_X)$. Let $\gamma^* := \arg\min_{\gamma \in [\gamma_L,\gamma_U]} A(\gamma, \lambda_L)$ and suppose $\gamma^*$ is unique. Then $\partial_\lambda A > 0$ on $\Upsilon$, so every minimizer of $A$ has $\lambda=\lambda_L$ and
\[
    \arg\min\nolimits_{\Upsilon} A = \{(\gamma^*, \lambda_L)\},
\]
and every random minimizer $\hat\eta_m=(\hat\gamma_m,\hat\lambda_m)$ of $\mathrm{CV}_m$ over $\Upsilon$ satisfies
\[
    \operatorname{dist}\bigl(\hat\eta_m, (\gamma^*, \lambda_L)\bigr) \xrightarrow{p} 0.
\]
\end{theorem}

\begin{proof}
    See Appendix~\ref{proof:sel_fixed}.
\end{proof}

\begin{remark}
Since $A$ is strictly increasing in $\lambda$ for every $\gamma$ ($\partial_\lambda A(\eta) = \sum_i 2\lambda\,\delta_i(\gamma)^2\mu_i(\gamma)/(\mu_i(\gamma)+\lambda)^3 > 0$), the $\lambda$-component of every minimizer of $A$ is $\lambda_L$: under a strong signal, CV selects the minimal-shrinkage configuration, the one maximizing the test's noncentrality $\|w^*_\eta(\boldsymbol{\Pi}f)\|^2_{\mathcal{H}_\gamma} = \sum_i \mu_i(\gamma)\delta_i(\gamma)^2/(\mu_i(\gamma)+\lambda)^2$. If $\gamma^*$ lies on the boundary of $[\gamma_L,\gamma_U]$ with non-vanishing slope, the selector concentrates on the boundary point: the consistency display of Theorem~\ref{th:sel_fixed} gives $\hat\gamma_m\xrightarrow{p}\gamma^*$.
\end{remark}

Under a fixed alternative the test is consistent whenever the misspecification survives the projection:

\begin{theorem}[Consistency under fixed alternatives]
\label{th:power_fixed}
Under the conditions of Theorem~\ref{th:sel_fixed} (in particular $\boldsymbol{\Pi}f \neq 0$, i.e.\ $f \notin \operatorname{span}\{g_{\theta_0,1}(X), \ldots, g_{\theta_0,d_\theta}(X)\}$, see Section~\ref{sec:uniform_limit}), for every fixed $t > 0$,
\[
    \mathbb{P}\big(n\hat{T}_{\mathrm{proj},\perp}^{(1)}(\hat{\eta}_m) > t\big) \to 1.
\]
If instead $\boldsymbol{\Pi}f = 0$ (the misspecification lies in the nuisance span, at the
pseudo-true parameter as per the convention of this section), the projection annihilates the signal
and the null CLT of Corollary~\ref{cor:proj_null} applies verbatim at $\hat\theta$: the projected
residuals $\hat{\boldsymbol{\Pi}}_{\hat\theta}\hat\varepsilon$ are then asymptotically a
mean-zero projected noise (the component $\hat{\boldsymbol\Pi}(f^*-f)$ being $o_p(1)$-negligible at
the relevant scale), the statistic stays $O_p(1)$ under the projected null law, and the rejection
probability tends to $\alpha$ by the bootstrap size control (Theorem~\ref{th:validity_sel}). The
feasible test is consistent iff the signal survives the projection at $\theta^*$,
$\boldsymbol{\Pi}_{\theta^*}f^*\neq 0$.
\end{theorem}

\begin{proof}
    See Appendix~\ref{proof:power_fixed}.
\end{proof}

\subsection{The Data-Driven Test under Local Alternatives}
\label{sec:local_power}

The third regime of the testing problem is the local one. Under a local alternative the criterion process converges to the drifted field $\mathbb{H}_h^a$, and the selector limit interpolates between the null law and the deterministic fixed-alternative configuration (Theorem~\ref{th:sel_local}); because $\hat{\eta}_m$ depends only on the selection sample and $S_n^\perp$ only on the testing sample, their joint convergence is a product structure, and the implemented statistic converges to the drifted limit evaluated at the random tuning (Proposition~\ref{prop:joint_local}).

\begin{theorem}[Selector under a local alternative]
\label{th:sel_local}
Suppose Assumptions~\ref{ass:A}--\ref{ass:F} and~\ref{ass:G} hold and the residual function in the selection sample is $f_m = \delta_m h$ with $\delta_m \to 0$ and $h \in L^2(\mathbb{P}_X) \setminus \{0\}$.
\textbf{(i)} If $\sqrt{m}\,\delta_m \to a < \infty$, assume that the drifted limit criterion $\mathbb{H}_h^a$ of Theorem~\ref{th:cv_uniform}(c) has an almost surely unique minimizer over $\Upsilon$. Then any random minimizer of $\mathrm{CV}_m$ over $\Upsilon$ satisfies
\[
    \hat\eta_m \xrightarrow{d} \eta^*_h := \arg\min_{\eta \in \Upsilon} \mathbb{H}_h^a(\eta),
\]
a drifted version of the null limit of Theorem~\ref{th:sel_h0}, shifted by the deterministic drift $a^2A_h$ and the Gaussian term $2aG$. For $a = 0$ the limit law is exactly that of Theorem~\ref{th:sel_h0}: the local alternative is invisible to the selector. The uniqueness assumption can be relaxed to convergence to the argmin set.

\end{theorem}

\begin{proof}
    See Appendix~\ref{proof:sel_local}.
\end{proof}

\begin{remark}[The Pitman alternative and the split ratio]\label{rem:split_ratio}
Under the local alternative at the total-sample rate, the residual function in the selection sample is $f_m = (m+n)^{-1/2}h$ with $h := \boldsymbol{\Pi}R$, so $\sqrt{m}\,\delta_m = \sqrt{m/(m+n)} \to a := \sqrt{c}$, $c := \lim m/(m+n) \in [0,1]$, the selection fraction of Section~\ref{sec:cv_selection}. The drift and Gaussian terms in $\mathbb{H}_h^a$ are then $c\,A_h(\eta)$ and $2\sqrt{c}\,G(\eta)$: the split ratio enters the selector's limit law. Since $\sqrt{m}\,\delta_m \le 1$, the selector always follows the drifted regime of Theorem~\ref{th:sel_local}(i) with $a = \sqrt{c}$; drifts stronger than the Pitman scale ($\sqrt m\,\delta_m\to\infty$) are never attained at the total-sample Pitman rate. A balanced split ($c \in (0,1)$) yields the drifted limit of Theorem~\ref{th:sel_local}(i); a testing-dominant split ($n \gg m$, $c \to 0$) leaves the selector asymptotically blind to the alternative (the $H_0$ law of Theorem~\ref{th:sel_h0}) while the test sees the full signal, $\sqrt{1-c}\,h \to h$; a selection-dominant split ($m \gg n$, $c \to 1$) gives the selector its maximal drift ($a = 1$) but shrinks the testing share, so the test's effective signal $h_c = \sqrt{1-c}\,h$ vanishes.
\end{remark}

\begin{proposition}[Joint convergence and evaluation under local alternatives]
\label{prop:joint_local}
Under the conditions of Theorem~\ref{th:sel_local}(i) under $H_{1,n}$, the total-sample-rate local alternative of Section~\ref{sec:cv_selection}, as $m, n \to \infty$,
\[
    \big(\hat{\eta}_m,\, S_n^\perp\big) \;\rightsquigarrow\; \big(\eta^*_h,\, S_\infty(\cdot;h_c)\big) \quad \text{in } \Upsilon \times \ell^\infty(\Upsilon),
\]
with $h_c := \sqrt{1-c}\,h$, the limit components being independent because $\mathcal{D}_m^{\mathrm{sel}} \perp\!\!\!\perp \mathcal{D}_n^{\mathrm{test}}$. Consequently, by continuity of the sample paths of the limit process,
\[
    S_n^\perp(\hat{\eta}_m) \xrightarrow{d} S_\infty(\eta^*_h; h_c).
\]
\end{proposition}

\begin{proof}
    See Appendix~\ref{proof:joint_limit}.
\end{proof}

The split ratio $c$ acts through two opposing channels. On the selection side, larger $c$ is favorable: the drift $cA_h$ increases with $\lambda$, pulling the selected configuration toward $\lambda_L$, the direction of maximal captured energy, and by Proposition~\ref{prop:snr}(i) power is non-decreasing in the captured-signal scale at fixed tuning. On the testing side, larger $c$ is unfavorable: the effective signal $h_c = \sqrt{1-c}\,h$ shrinks toward zero, and a selection-dominant split buys better tuning at the price of a weaker test. The choice of the splitting ratio is therefore a design question with an explicit trade-off between selection quality and testing power, beyond the scope of this paper; the simulation study of Section~\ref{sec:simulation} investigates the finite-sample performance of different splitting ratios, and the results seem to favor a relatively small $c$ (a testing-dominant split).

\subsection{The Energy Decomposition and the Signal-to-Noise Ratio of the Test}
\label{sec:energy}

The selector results characterize the limit law of $\hat{\eta}_m$. This subsection explains why CV-favored configurations are favorable for the test through an exact bookkeeping identity: at every $\eta$ the misspecification energy splits into an uncaptured and a captured part, and the captured part contains precisely the signal that drives the statistic. Cross-validation minimizes the uncaptured part, hence maximizes the captured part, and when no signal is present shrinks the noise floor; the same $\lambda$-dial that opens the filter to signal also admits noise, and the criterion balances the two. Theorem~\ref{th:cv_snr} is accordingly a two-channel result.

\paragraph{The energy identity.}
Fix a residual function $h \in L^2(\mathbb{P}_X)$: the effective signal $\boldsymbol{\Pi}f$ under a fixed alternative, or $h := \boldsymbol{\Pi}R$ under the local alternative of Theorem~\ref{th:proj_local}. Let $w^*_\eta(h) := (C_\gamma + \lambda I)^{-1}\boldsymbol{L}_{k_\gamma}h$ denote the population KRR solution, with the spectral representation $w^*_\eta(h) = \sum_{i\ge1}\mu_i(\gamma)h_i(\gamma)\phi_i(\cdot;\gamma)/(\mu_i(\gamma)+\lambda)$.

\begin{proposition}[Energy identity]\label{prop:energy_identity}
For every $\eta \in \Upsilon$ and every $h \in L^2(\mathbb{P}_X)$,
\begin{equation}
    \lVert h\rVert_{L^2}^2 = A_h(\eta) + \big\lVert w^*_\eta(h)\big\rVert_{L^2}^2 + 2\lambda\big\lVert w^*_\eta(h)\big\rVert_{\mathcal{H}_\gamma}^2,
    \label{eq:energy}
\end{equation}
where $A_h(\eta) = \lVert h - w^*_\eta(h)\rVert_{L^2}^2$ is the approximation error of \eqref{eq:A_eta}. Spectrally,
\[
    \big\lVert w^*_\eta(h)\big\rVert_{L^2}^2 = \sum_{i\ge1}\frac{\mu_i(\gamma)^2}{(\mu_i(\gamma)+\lambda)^2}\,h_i(\gamma)^2, \qquad
    \big\lVert w^*_\eta(h)\big\rVert_{\mathcal{H}_\gamma}^2 = \sum_{i\ge1}\frac{\mu_i(\gamma)}{(\mu_i(\gamma)+\lambda)^2}\,h_i(\gamma)^2,
\]
so for each fixed $\gamma$ and each $h \neq 0$ the approximation error is strictly increasing in $\lambda$, while the two captured components are strictly decreasing in $\lambda$ (for $h = 0$ all three components vanish identically in $\lambda$). Consequently, over any set of configurations, minimizing $A_h$ is equivalent to maximizing the captured energy
\begin{equation}
    E_h(\eta) := \lVert h\rVert_{L^2}^2 - A_h(\eta) = \big\lVert w^*_\eta(h)\big\rVert_{L^2}^2 + 2\lambda\big\lVert w^*_\eta(h)\big\rVert_{\mathcal{H}_\gamma}^2.
    \label{eq:captured_energy}
\end{equation}
\end{proposition}

\begin{proof}
    See Appendix~\ref{proof:energy_identity}.
\end{proof}

The three terms of \eqref{eq:energy} are the approximation error (the part of the misspecification the fit cannot represent), the explained variance (the captured signal measured in $L^2$), and the regularization component, equal to $2\lambda$ times the RKHS energy of the fit. Per spectral direction the budget is allocated without interference: dividing \eqref{eq:energy} directionwise,
\[
    \left(\frac{\lambda}{\mu_i(\gamma)+\lambda}\right)^{\!2} + \left(\frac{\mu_i(\gamma)}{\mu_i(\gamma)+\lambda}\right)^{\!2} + \frac{2\lambda\mu_i(\gamma)}{(\mu_i(\gamma)+\lambda)^2} = 1.
\]
The bandwidth $\gamma$ determines the spectrum and hence which directions carry the budget; the regularization $\lambda$ determines, within each direction, how the budget is split between approximation error and captured components. This is why \eqref{eq:cv_selector} chooses $(\gamma,\lambda)$ jointly: two dials of the same per-direction allocation.

\paragraph{The signal-to-noise ratio of the test.}
Under the local alternative, the limit statistic of Theorem~\ref{th:stat_uniform} is
\[
\begin{aligned}
    S_\infty(\eta;h) &= \sum_{i\ge1} w_i(\eta)\,\big(Z_i^\perp + h_i(\gamma)\big)^2, \qquad
    w_i(\eta) := \frac{\mu_i(\gamma)}{(\mu_i(\gamma)+\lambda)^2},
\end{aligned}
\]
where $(Z_1^\perp,Z_2^\perp,\dots)^\top \sim \mathcal{N}(\boldsymbol{0}, \Sigma^\perp(\gamma))$, $\Sigma^\perp(\gamma)$ is the covariance of the projected noise of Corollary~\ref{cor:proj_null} (in the eigenbasis $\{\phi_i(\cdot;\gamma)\}$ of $\boldsymbol{L}_{k_\gamma}$), and write $\tau_i^\perp(\gamma) := \Sigma^\perp_{ii}(\gamma)$ for its diagonal and $M(\eta) := \sum_{i\ge1}w_i(\eta)\tau_i^\perp(\gamma)$ for the null mean of the statistic. The signal carried by the test at $\eta$ is the noncentrality
\begin{equation}
    \mathcal{N}(\eta;h) := \sum_{i\ge1} w_i(\eta)\,h_i(\gamma)^2 = \big\lVert w^*_\eta(h)\big\rVert_{\mathcal{H}_\gamma}^2,
    \label{eq:noncentrality}
\end{equation}
the RKHS energy of the fitted signal, i.e., the regularization component of \eqref{eq:energy} divided by $2\lambda$. The second projection statistic has the same limit structure with the weights $w_i^{(2)}(\eta) := \mu_i(\gamma)/(\mu_i(\gamma)+\lambda)$, and its noncentrality is
\begin{equation}
    \mathcal{N}^{(2)}(\eta;h) := \sum_{i\ge1}\frac{\mu_i(\gamma)}{\mu_i(\gamma)+\lambda}\,h_i(\gamma)^2
    = \big\lVert w^*_\eta(h)\big\rVert_{L^2}^2 + \lambda\big\lVert w^*_\eta(h)\big\rVert_{\mathcal{H}_\gamma}^2,
    \label{eq:noncentrality2}
\end{equation}
the explained-variance component of \eqref{eq:energy} plus half of its regularization component. The two are strictly decreasing in $\lambda$; consequently item~(i) of Proposition~\ref{prop:snr}, its divergence-rate statement, and Theorem~\ref{th:cv_snr} hold for $\hat T^{(2)}_{\mathrm{proj},\perp}$ verbatim, with $S_\infty^{(2)}$ and the corresponding quantiles $c_\infty^{(2)}$ in place of their unmarked counterparts. The noise scale is the null standard deviation $\sigma_\infty(\eta) := \sqrt{\operatorname{Var}(S_\infty(\eta))} = \sqrt{2}\,\lVert W(\eta)^{1/2}\Sigma^\perp(\gamma)W(\eta)^{1/2}\rVert_{\mathrm{HS}}$ with $W(\eta) := \operatorname{diag}(w_i(\eta))$, and the signal-to-noise ratio of the limiting test at $\eta$ is
\begin{equation}
    \mathcal{S}(\eta;h) := \frac{\mathcal{N}(\eta;h)}{\sigma_\infty(\eta)}.
    \label{eq:snr}
\end{equation}
Throughout, $c_\infty(\eta)$ denotes the $(1-\alpha)$-quantile of the null law $S_\infty(\eta)=S_\infty(\eta;0)$ and $p(\eta;h):=\mathbb{P}(S_\infty(\eta;h)>c_\infty(\eta))$ the pointwise local power; at the total-rate Pitman alternative the signal enters attenuated, $h_c=\sqrt{1-c}\,h$, by the testing share.

\paragraph{The whitened form.}
For the power analysis, it is convenient to whiten the projected noise. Writing $Z^\perp = \Sigma^\perp(\gamma)^{1/2}U$ with $U \sim \mathcal{N}(\boldsymbol{0}, I)$, the limit statistic takes the diagonal form
\begin{equation}
    S_\infty(\eta; h) = \sum_{k\ge1}\varsigma_k(\eta)\big(U_k + \vartheta_k(\eta;h)\big)^2,
    \label{eq:whitened}
\end{equation}
where $\{\varsigma_k(\eta)\}$ are the eigenvalues and $\mathcal{Q}(\eta)$ the orthogonal matrix of eigenvectors of $\Sigma^\perp(\gamma)^{1/2}W(\eta)\Sigma^\perp(\gamma)^{1/2}$, and $\vartheta(\eta;h) := \mathcal{Q}(\eta)^\top\Sigma^\perp(\gamma)^{-1/2}h$ the signal coefficients in the diagonalizing basis; the square root is taken on the support of the generically singular $\Sigma^\perp(\gamma)$, which contains the projected signal $h$ by Assumption~\ref{ass:G}(iii)--(v).

\begin{proposition}[Monotone power and divergence rate]\label{prop:snr}
Let Assumptions~\ref{ass:A}--\ref{ass:F} and~\ref{ass:G} hold, and fix $\eta \in \Upsilon$.
\textbf{(i) Unbiasedness and scale monotonicity.} For every projected signal $h$, the map $t \mapsto p(\eta; th)$ is nondecreasing on $[0, \infty)$; in particular, $p(\eta; h) \ge \alpha$, so the limiting test is unbiased against every misspecification direction not annihilated by $\boldsymbol{\Pi}$.
\textbf{(ii) Divergence rate under fixed alternatives.} If $\boldsymbol{\Pi}f \neq 0$, then
\[
    S_n^\perp(\eta)/n \xrightarrow{p} \bigl\lVert w_\eta^*(\boldsymbol{\Pi}f)\bigr\rVert_{\mathcal{H}_\gamma}^2 = \mathcal{N}(\eta; \boldsymbol{\Pi}f),
\]
the RKHS component of the captured energy $E_{\boldsymbol{\Pi}f}(\eta)$ defined in \eqref{eq:captured_energy}.
\end{proposition}

\begin{proof}
    See Appendix~\ref{proof:snr}.
\end{proof}

In the local-alternative regime, $p(\eta;h)$ is the rejection probability of the limiting test at fixed $\eta$ with the projected signal $h$ added to the projected noise; under a fixed alternative only item (ii) applies, the statistic diverging at the noncentrality rate. Item (i) certifies that enlarging the captured signal never hurts power, $p(\eta;h)\ge\alpha$: the infinite-dimensional analogue of the classical unbiasedness of quadratic-form statistics with nonnegative weights against two-sided location alternatives \citep{lehmannromano2005}; it is the directional link from the selection side to the testing side. Item (ii) identifies the divergence rate with the noncentrality $\mathcal{N}(\eta;\boldsymbol{\Pi}f)$, the RKHS component of the captured energy up to the factor $2\lambda$. How much captured energy converts into power depends, beyond these signal-side facts, on the noise channel: the dispersion of the null law through $f_\eta(c_\infty(\eta))$ and $\sigma_\infty(\eta)$, the object of Theorem~\ref{th:cv_snr} and Condition~\ref{sc:monotone}. The rotation $\mathcal{Q}(\eta)$ makes directional monotonicity compatible with generically correlated coordinates: it decorrelates the projected noise, so enlarging the signal moves the mean away from the origin in every direction at once.

\paragraph{CV as a signal-to-noise optimizer.}
Two channels connect the CV choice to the signal-to-noise ratio \eqref{eq:snr}: a signal channel, through which the selected configuration maximizes the numerator $\mathcal{N}$, and a noise channel, through which it shrinks the denominator $\sigma_\infty$ when there is no signal to capture. Write $D_0(\eta) := V(\eta) + \Delta(\eta)$ for the deterministic part of the selector's limit criterion under $H_0$, and $D_c(\eta) := D_0(\eta) + cA_h(\eta)$ for the local-alternative case.

\begin{condition}[Standing monotonicity]\label{sc:monotone}
For every $\gamma \in [\gamma_L, \gamma_U]$ and every $\ell \ge 1$ with $a_\ell(\gamma) > 0$,
\[
    a_\ell(\gamma) \le \frac{\mu_\ell(\gamma)\,\tau_\ell(\gamma)}{\lambda_U},
    \qquad
    a_\ell(\gamma) := \mathbb{E}\big[c^*(X,X')\,\phi_\ell(X;\gamma)\,\phi_\ell(X';\gamma)\big],
\]
where $\mu_\ell(\gamma)$ and $\phi_\ell(\cdot;\gamma)$ are the Mercer eigenvalues and eigenfunctions of $k_\gamma$, $\tau_\ell(\gamma)$ the spectral noise variances of Section~\ref{sec:uniform_limit}, and $c^*$ the projection-correction kernel of Section~\ref{sec:nuisance}. The clause requires that no Mercer direction whose projection-overlap with the estimation error is positive carry an overlap exceeding its own noise leverage at the largest admissible regularization; directions with $a_\ell(\gamma) \le 0$ are exempt. Under this clause the drift $D_0(\eta) = V(\eta) + \Delta(\eta)$ is strictly decreasing in $\lambda$ on $[\lambda_L, \lambda_U]$ for every $\gamma$; the termwise argument is recorded in Appendix~\ref{app:monotonicity_check}. The clause holds automatically under homoscedasticity, $\sigma^2(\cdot) \equiv \bar\sigma^2$: then $c^*$ is negative semidefinite, so $a_\ell(\gamma) \le 0$ for every direction, score span, kernel, and dimension (Appendix~\ref{app:monotonicity_check}), which covers the leading case of the analysis and the primary simulation setting. Outside homoscedasticity the clause is not automatic: Appendix~\ref{app:monotonicity_check} records the exact derivative formula of which the clause is the termwise sufficient condition, the location-model configuration in which $a_\ell > 0$ genuinely occurs, and the numerical verification of the clause for the Gaussian kernel and the simulation designs of Section~\ref{sec:simulation}.
\end{condition}

\begin{theorem}[CV selection boosts the signal-to-noise ratio]\label{th:cv_snr}
Under Assumptions~\ref{ass:A}--\ref{ass:F} and~\ref{ass:G} and the sample-splitting protocol of Section~\ref{sec:cv_selection}:
\textbf{(i) Signal channel (fixed alternative).} If the conditions of Theorem~\ref{th:sel_fixed} hold, then $\hat{\eta}_m \xrightarrow{p} \eta^* := (\gamma^*, \lambda_L)$, and $\eta^*$ minimizes the approximation error, $A(\eta^*) = \min_{\Upsilon}A$, equivalently maximizes the captured energy, $E_{\boldsymbol{\Pi}f}(\eta^*) = \max_{\Upsilon}E_{\boldsymbol{\Pi}f}$. Moreover, for every $\gamma$, the maps $\lambda \mapsto E_{\boldsymbol{\Pi}f}(\gamma,\lambda)$ and $\lambda \mapsto \mathcal{N}(\gamma,\lambda;\boldsymbol{\Pi}f)$ are strictly decreasing, so $\lambda_L$ maximizes the captured energy and the divergence rate of the statistic; the resulting consistency is Theorem~\ref{th:power_fixed}.

\textbf{(ii) Noise channel (null).} Under $H_0$, the variance component $V(\eta)$, the null mean $M(\eta)$, and the noise scale $\sigma_\infty(\eta)$ of the limit statistic are strictly decreasing in $\lambda$ for every $\gamma$. Under Condition~\ref{sc:monotone}, the deterministic part $D_0(\eta)$ of the selector's limit criterion is strictly decreasing in $\lambda$, so for every $\gamma$ its $\lambda$-minimizer over $[\lambda_L, \lambda_U]$ is attained at $\lambda = \lambda_U$.

\textbf{(iii) Local alternatives and the split ratio.} Under Condition~\ref{sc:monotone}, with $c := \lim m/(m+n) \in [0,1]$, the deterministic part of the selector's limit criterion is $D_c(\eta)$, with $\partial_\lambda(V+\Delta) < 0$ and $\partial_\lambda A_h > 0$. If $D_c$ has an interior minimizer $\lambda^*_c \in (\lambda_L, \lambda_U)$ with $\partial^2_{\lambda\lambda}D_c(\lambda^*_c) > 0$, then $\lambda^*_c$ is strictly decreasing in $c$; at $c = 1$, the minimizer of $D_1$ lies strictly above $\lambda_L$ whenever $A_h'(\lambda_L) < |\partial_\lambda(V+\Delta)(\lambda_L)|$, and equals $\lambda_L$ when $A_h'(\lambda_L) \ge |\partial_\lambda(V+\Delta)(\lambda_L)|$, provided $D_1$ is convex in $\lambda$ on $[\lambda_L,\lambda_U]$.
\end{theorem}

By the homoscedasticity result of Appendix~\ref{app:monotonicity_check} ($a_\ell \le 0$ for every $\ell$ when $\sigma^2(\cdot) \equiv \bar\sigma^2$), Condition~\ref{sc:monotone} holds automatically under homoscedasticity, so items (ii)--(iii) of Theorem~\ref{th:cv_snr} are unconditional there.

\begin{proof}
    See Appendix~\ref{proof:cv_snr}.
\end{proof}

\begin{remark}
The theorem assembles the mechanism behind the selection analysis: the energy identity \eqref{eq:energy} ties the CV objective to the fitted signal, Proposition~\ref{prop:snr} ties that signal to power, and items (i)--(iii) give the direction of the CV choice in $\lambda$ through the two monotone parts of a trade-off. The signal part: the captured energy and the divergence rate are strictly decreasing in $\lambda$ (item (i)), and the uncaptured signal $A_h$ strictly increasing (item (iii)). The noise part: the noise components $V+\Delta$ are strictly decreasing in $\lambda$ under Condition~\ref{sc:monotone}, and with them the null mean $M(\eta)$ and the noise scale $\sigma_\infty(\eta)$ (item (ii)). The criterion's drift thus carries two components with opposite $\lambda$-slopes, $\partial_\lambda(V+\Delta) < 0$ against $\partial_\lambda A_h > 0$: enlarging $\lambda$ removes noise from the criterion while forfeiting signal to the approximation error. The split ratio $c := \lim m/(m+n)$ turns this opposition into the selection mechanism. The limit drift of the selector is $D_c(\eta) = D_0(\eta) + cA_h(\eta)$, the noise part present with unit weight and the uncaptured signal weighted by the selection share; when $D_c$ has an interior minimizer $\lambda^*_c$ with positive second derivative, item (iii) shows that $\lambda^*_c$ is strictly decreasing in $c$. The selected $\lambda$ therefore slides along the $\lambda$-axis from the noise-optimal configuration $\lambda_U$ at $c = 0$, where item (ii) minimizes $D_0$, $M$, and $\sigma_\infty$, toward the signal-optimal configuration $\lambda_L$, where item (i) maximizes the captured energy and the divergence rate and the test is consistent (Theorem~\ref{th:power_fixed}); the endpoint at $c = 1$ is decided by the marginal comparison of $A_h'(\lambda_L)$ against $|\partial_\lambda(V+\Delta)(\lambda_L)|$, provided $D_1$ is convex. The two channels of the theorem are the two ends of this path, item (ii) at $c = 0$ and item (i) at the $\lambda_L$ limit. The theorem is not a claim that CV maximizes the literal ratio $\mathcal{N}(\eta;h)/\sigma_\infty(\eta)$: under a fixed alternative $\lambda_L$ also maximizes $\sigma_\infty$, and under $H_0$ there is no signal to form a ratio. The exact selector laws are those of Theorems~\ref{th:sel_h0}--\ref{th:sel_local}. For the Gaussian kernel and the simulation designs of Section~\ref{sec:simulation}, Condition~\ref{sc:monotone} is verified numerically on the entire CV grid (Appendix~\ref{app:monotonicity_check}, Table~\ref{tab:monotonicity_check} and Figure~\ref{fig:monotonicity_check}).
\end{remark}

\section{Simulation Studies and Empirical Applications}
\label{sec:simulation}
This section examines the finite-sample behavior of the projection tests and illustrates them on real data. The first subsection evaluates size and power under the sample-splitting protocol of Section~\ref{sec:cv_selection} across fourteen data-generating processes, benchmarked against the classical kernel conditional moment (KCM) test, and quantifies the value of the selection sample through an oracle protocol. The second subsection applies the tests to three applications: the National Supported Work (NSW) dataset, where the specification of a probit propensity score model and the hypothesis of a zero conditional average treatment effect are tested individually and jointly; the Mincer wage equation; and a log-linear demand equation for cigarettes. Throughout, critical values come from the multiplier bootstrap of Section~\ref{sec:estimation_effect} under the sample-splitting protocol of Section~\ref{sec:nuisance}, and results are read against Sections~\ref{sec:null_dist}--\ref{sec:local_power} and~\ref{sec:energy}.

\subsection{Simulation Studies}

The projection tests are compared with the KCM test, the quadratic form $\hat T = \frac{1}{n}\,\boldsymbol{\hat{\varepsilon}}^\top \boldsymbol{K}_\gamma \boldsymbol{\hat{\varepsilon}}$ with $(\boldsymbol{K}_\gamma)_{ij} = \exp\!\left(-\gamma\lVert x_i - x_j\rVert^2\right)$, computed under the multiplier bootstrap as in \cite{escanciano2024gaussian}. Empirical size is reported in Table~\ref{tab:size}, and power in Figures~\ref{fig:method_comparison} and~\ref{fig:oracle_comparison}; the oracle protocol isolates the selection-sample value implied by Theorem~\ref{th:cv_snr}(iii).

\paragraph{Data-generating processes.}
The null DGP and the six alternative designs are each evaluated at both dimensions, $d = 10$ and $d = 20$. The null model is $\mathrm{DGP}_0$ (null): $Y = 1 + X^\top\beta + e$, with $X \sim \mathcal{N}(0, I_d)$, $\beta = 0.5\,\boldsymbol{1}_d$, and $e \sim \mathcal{N}(0,1)$ independent of $X$. The fixed alternatives under homoskedasticity are $\mathrm{DGP}_1: Y = 1 + X^\top\beta + 1.5\exp\!\big(-(X^\top\beta)^2\big) + e$, $\mathrm{DGP}_2: Y = 1 + X^\top\beta + 2\cos\!\big(1.2\pi\lVert X\rVert\big) + e$, $\mathrm{DGP}_3: Y = 1 + X^\top\beta + 0.5\,(X^\top\beta)^2 + e$, and $\mathrm{DGP}_4: Y = 1 + X^\top\beta + 1.5\exp\!\big(0.25\,X^\top\beta\big) + e$, sharing the covariates, $\beta$, and error structure of the null model: smooth global departures (DGP$_1$), oscillatory radial departures (DGP$_2$), quadratic departures (DGP$_3$), and exponential departures (DGP$_4$).

The remaining alternatives combine misspecification with heteroskedasticity. For DGP$_5$, the first $d/2$ covariates are drawn from $\mathrm{Uniform}[0,\, 1+0.1(l-1)]$, $l = 1, \ldots, d/2$, the remaining ones from $\mathcal{N}(0, (1+0.1(l-d/2))^2)$; the specification is $\mathrm{DGP}_5: Y = 1 + X^\top\boldsymbol{1}_d + \lVert X\rVert + |X^\top\boldsymbol{1}_d|\,e$, a heteroskedastic model with a fixed radial departure. For DGP$_6$, the first $d/2$ covariates are drawn from $\mathrm{Uniform}[0, l]$, $l = 1, \ldots, d/2$, the remaining ones from $\mathcal{N}(0, (1+0.1(l-d/2))^2)$; the specification is $\mathrm{DGP}_6: Y = 1 + X^\top\boldsymbol{1}_d + 1.5\,\frac{\lVert X\rVert^2}{\sqrt{n}} + \sigma(X)\,e$ with $\sigma(X) = \sqrt{0.1 + \sum_{l \le d/2} X_l + \sum_{l > d/2} X_l^2}$, a heteroskedastic model with a local (Pitman) departure of order $n^{-1/2}$, as in Theorem~\ref{th:proj_local}. In all models the structural parameter is estimated by ordinary least squares, and the Neyman-orthogonalized residuals are formed with the projection of Section~\ref{sec:estimation_effect}.

\paragraph{Protocols.}
Three protocols are compared.
\textbf{(i) Sample splitting.} Each replication draws a sample of total size $N \in \{200, 400\}$ and splits it at random into a selection sample ($15\%$) and a testing sample ($85\%$), i.e., selection sizes $30$ and $60$ against testing sizes $170$ and $340$. On the selection sample, the tuning parameters $(\gamma, \lambda)$ are chosen by $5$-fold CV over a data-adaptive grid: $20$ values of $\gamma$ spanning $[\gamma_{\max}/100, \gamma_{\max}]$ with $\gamma_{\max} = 2/\mathrm{median}\{\lVert x_i - x_j\rVert^2\}$, and $10$ values of $\lambda$ in $[1, 10]$. The projection statistics $\hat T^{(1)}_{\mathrm{proj},\perp}$ and $\hat T^{(2)}_{\mathrm{proj},\perp}$ are computed on the testing sample with the multiplier bootstrap (Mammen weights, $B = 500$ replications), and the selected tuning is held fixed, as in Section~\ref{sec:null_dist}.
\textbf{(ii) Oracle.} The tuning parameters are selected on an \emph{additional} selection sample of $200$ observations and the test is computed on a testing sample of $400$ observations, so that the testing budget coincides with that of the KCM test at $n = 400$.
\textbf{(iii) KCM.} The KCM test is applied to the full sample of size $n \in \{200, 400\}$ without splitting, at the kernel parameters $\gamma \in \{2, 1, 0.5, 0.2, 0.1, 0.05, 0.02\}$ and at the median rule $\gamma = 1/\mathrm{median}\{\lVert x_i - x_j\rVert\}$; the values $\gamma = 0.5$ and the median rule reproduce the fixed- and data-driven variants of \cite{escanciano2024gaussian}. All methods use 1000 replications and identical Monte Carlo draws within each replication, so differences in rejection rates are attributable to the method and, for the KCM test, to the kernel parameter. The figures report rejection rates at the $5\%$ level; the empirical sizes at all three levels are collected in Table~\ref{tab:size} of Appendix~\ref{app:additional_sim}.

\paragraph{Size.}
The projection tests are close to nominal size under the splitting protocol at all three levels in most configurations. At the $5\%$ level, the rejection frequency under $\mathrm{DGP}_0$ is $0.035$--$0.039$ at $d = 10$ and $0.042$--$0.071$ at $d = 20$, the $n = 200$, $d = 20$ configuration least accurate: its selection sample of $30$ observations leaves the CV choice noisy; under the oracle protocol the corresponding figures are $0.039$ and $0.053$. The two statistics $\hat T^{(1)}_{\mathrm{proj},\perp}$ and $\hat T^{(2)}_{\mathrm{proj},\perp}$ behave essentially identically throughout. The size of the KCM test is governed by the kernel parameter: at $\gamma \ge 0.5$ the statistic is effectively degenerate and never rejects (size $\approx 0$ at all three levels), while for $\gamma \le 0.1$ the size is approximately nominal in most configurations; the median rule is conservative, with rejection rates below the nominal level in all configurations.

\paragraph{Sensitivity to the kernel parameter and the CV-selected configuration.}
Figure~\ref{fig:method_comparison} reports the $5\%$ rejection rates for all methods at equal total sample size, for $n \in \{200, 400\}$ and $d \in \{10, 20\}$. The most striking feature is the sensitivity of the KCM test to the kernel parameter. At $\gamma \ge 1$ the statistic never rejects on any alternative, and at $\gamma = 0.5$ appreciable power survives only on the two strongest $d = 10$ designs ($\mathrm{DGP}_3$: $0.443$ and $\mathrm{DGP}_5$: $0.187$ at $n = 400$); the kernel is nearly diagonal and the statistic rarely exceeds its bootstrap quantiles. Power rises as $\gamma$ decreases, peaks at an intermediate value (typically $\gamma = 0.02$--$0.1$), and then declines or plateaus; the location and height of the peak vary across DGPs and dimensions. In $d = 20$ even the best kernel parameter delivers modest and uneven power: at $n = 400$ the grid maximum is $0.559$ for $\mathrm{DGP}_4$, $0.409$ for $\mathrm{DGP}_5$, and $0.251$ for $\mathrm{DGP}_6$. The median rule, the standard data-driven default, is particularly unreliable in high dimensions: for $\mathrm{DGP}_3$ at $n = 400$ it delivers $0.104$ against $0.944$ at $\gamma = 0.1$, and for $\mathrm{DGP}_4$ and $\mathrm{DGP}_6$ it is essentially powerless ($0.003$ and $0.000$).

By contrast, the projection tests, which select $(\gamma, \lambda)$ by CV on the selection sample and require no user choice, track the upper envelope of the KCM kernel-parameter family: at $n = 400$ their rejection rates are $0.871$ ($\mathrm{DGP}_1$ at $d = 10$), $0.996$ ($\mathrm{DGP}_3$ at $d = 20$), $0.462$ ($\mathrm{DGP}_4$ at $d = 20$), and $0.268$ ($\mathrm{DGP}_6$ at $d = 20$), comparable to the KCM test at its best kernel parameter ($0.968$, $0.999$, $0.559$, and $0.251$) and strictly better than the KCM test at any fixed kernel parameter in the large-$\gamma$ range. The CV-selected configuration thus adapts to the difficulty of the alternative, in line with the analysis of Sections~\ref{sec:fixed_consistency} and \ref{sec:energy}: under a signal, the criterion trades the uncaptured signal against the noise and settles near the minimal-regularization configuration that maximizes the noncentrality of the test.

\paragraph{The oracle protocol and the value of the selection sample.}
Figure~\ref{fig:oracle_comparison} compares the oracle-protocol projection tests with the KCM test at equal testing budgets ($n = 400$). The projection tests remain at nominal size and achieve power comparable to the KCM test at its best kernel parameter. They dominate the KCM test at every kernel parameter for $\mathrm{DGP}_2$ at $d = 20$ ($0.209$ vs.\ at most $0.165$) and for $\mathrm{DGP}_4$ ($0.308$ vs.\ at most $0.259$ at $d = 10$; $0.596$ vs.\ at most $0.559$ at $d = 20$); they match the best configuration for $\mathrm{DGP}_6$ at $d = 20$ ($0.251$ vs.\ $0.251$). In the remaining scenarios the two procedures land within a few points of each other, and the KCM retains an edge on the smooth alternatives $\mathrm{DGP}_1$, $\mathrm{DGP}_2$, $\mathrm{DGP}_5$, and $\mathrm{DGP}_6$ (all at $d = 10$).

The oracle results also quantify the value of the selection sample, the split-ratio channel of Section~\ref{sec:energy}. Enlarging the selection sample from $60$ to $200$ observations (with the testing sample growing from $340$ to $400$) raises the $5\%$ rejection rate from $0.871$ to $0.915$ ($\mathrm{DGP}_1$ at $d = 10$), from $0.209$ to $0.308$ and from $0.462$ to $0.596$ ($\mathrm{DGP}_4$ at $d = 10$ and $d = 20$), and from $0.686$ to $0.747$ ($\mathrm{DGP}_5$ at $d = 10$), with no material change in size. The testing sample sets the resolution of the test; the selection sample determines how closely the tuning approaches the signal-maximizing configuration, exactly as predicted by the drift analysis of Theorem~\ref{th:cv_snr}(iii).

Overall, the simulations deliver the two messages of Section~\ref{sec:nuisance}: the projection tests control size under splitting without user-chosen tuning, whereas the classical kernel test depends on a hard-to-choose kernel parameter whose defaults can be catastrophic in high dimensions; and the power of the selected test rises with the split ratio, a larger selection sample pulling the CV choice toward the captured-signal-maximizing configuration.

Appendix~\ref{app:additional_sim} reports additional designs and results, including the sensitivity of the tests to the sample-splitting ratio (Appendix~\ref{app:split_ratio}).

\begin{figure}[htbp]
    \centering
    \includegraphics[width=\textwidth]{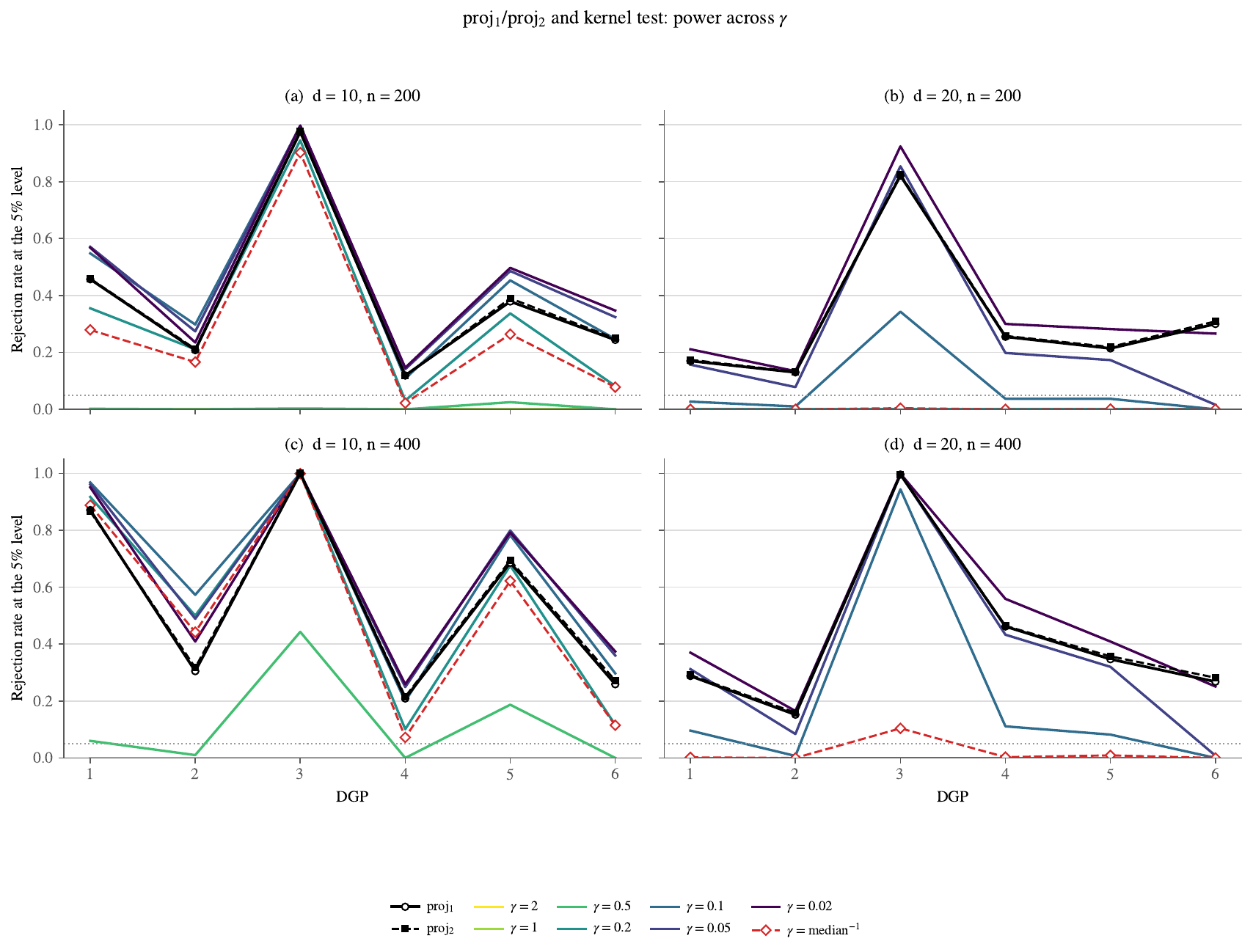}
    \caption{Rejection rates at the $5\%$ level for the projection tests $\hat T^{(1)}_{\mathrm{proj},\perp}$ and $\hat T^{(2)}_{\mathrm{proj},\perp}$ under the sample-splitting protocol and for the KCM test at kernel parameters $\gamma \in \{2, 1, 0.5, 0.2, 0.1, 0.05, 0.02\}$ and the median rule, for $n \in \{200, 400\}$ and $d \in \{10, 20\}$, at equal total sample size. Within each panel the alternatives $\mathrm{DGP}_1$--$\mathrm{DGP}_6$ are indexed $1$--$6$, at both $d = 10$ and $d = 20$; the null DGP is excluded. The dashed gray line marks the nominal $5\%$ level.}
    \label{fig:method_comparison}
\end{figure}

\begin{figure}[htbp]
    \centering
    \includegraphics[width=\textwidth]{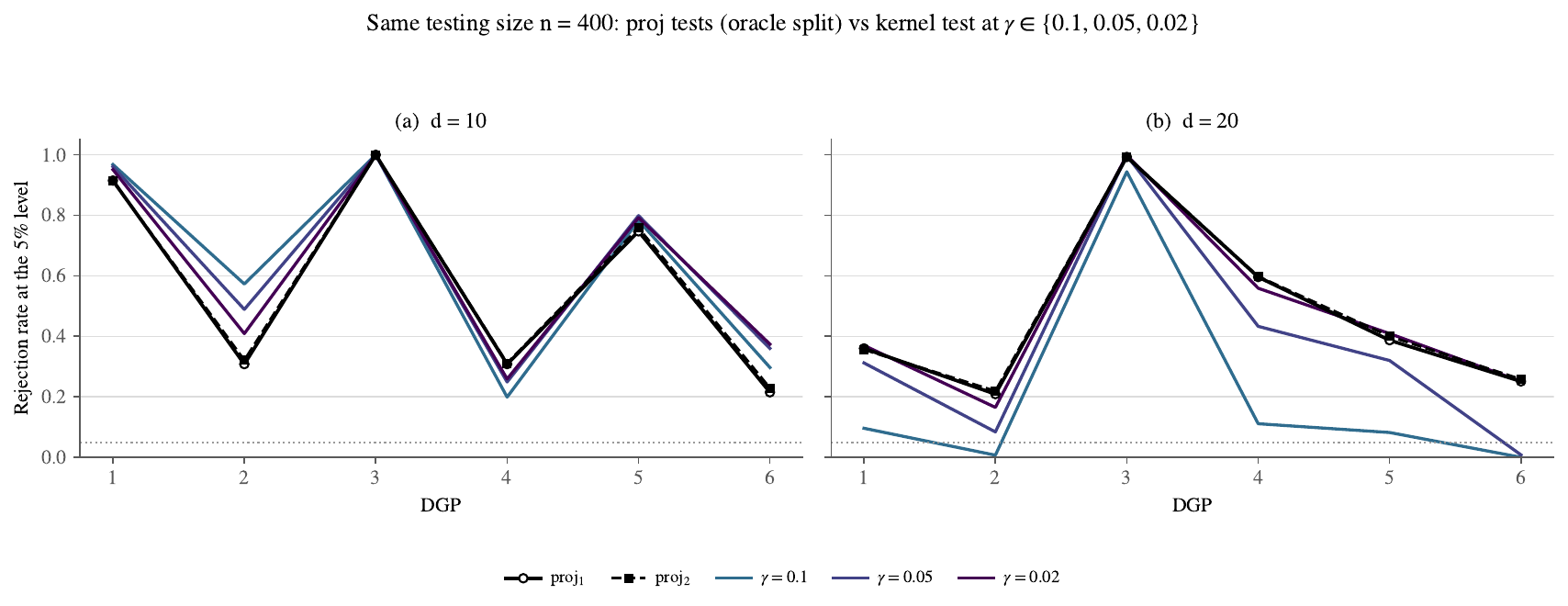}
    \caption{Rejection rates at the $5\%$ level under the oracle protocol (selection sample $200$, testing sample $400$), compared with the KCM test at $n = 400$ for kernel parameters $\gamma = 0.1, 0.05, 0.02$; the dashed gray line marks the nominal $5\%$ level.}
    \label{fig:oracle_comparison}
\end{figure}

\subsection{Empirical Applications}
Three applications are studied. The first is the National Supported Work (NSW) Demonstration, a randomized evaluation of a training program that has become standard in the causal inference literature; there we test, individually and jointly, the correct specification of a probit propensity score model and the null of a zero conditional average treatment effect (CATE). The second is the specification of the Mincer wage equation on the Wooldridge \texttt{wage2} data, the leading empirical specification of returns to schooling. The third is the specification of a log-linear demand equation for cigarettes on the pooled state panel, whose structural residual comes from two-stage least squares. In all three applications the residuals are formed as in Section~\ref{sec:estimation_effect}, whose multiplier bootstrap also approximates the null distributions.

\paragraph{First application: the National Supported Work dataset.}
The National Supported Work (NSW) Demonstration (1975--1979) randomly assigned participants across the United States to supported-work treatment or control, and its impact on labor market outcomes has been a longstanding focus in economics.

We use the Dehejia--Wahba version \citep{lalonde1986evaluating,dehejia1999causal}: 445 observations (185 treated, 260 controls), covariates including age, education, and prior earnings, and primary outcome real earnings in 1978 (\url{https://users.nber.org/~rdehejia/nswdata2.html}).

To mitigate numerical instability, we divide ``age'' and ``education'' by 10 and log-transform all earnings variables, leaving 8 covariates: 4 continuous variables and 4 binary indicators.

\paragraph{Hypotheses and residuals.}
Two hypotheses are tested individually and jointly: correct specification of the propensity score model and a zero conditional average treatment effect. Let $Y(1)$ and $Y(0)$ denote the potential outcomes under treatment and control, respectively, and let $X$ be a vector of observed covariates. The treatment indicator is $T$, so the observed outcome is $Y = Y(1)T + Y(0)(1-T)$.

The propensity score is modeled as $P(T=1|X) = \Phi(\beta^\top X)$, where $\Phi$ denotes the standard normal cumulative distribution function. The CATE function is defined as $\tau(X) = \mathbb{E}[Y(1) - Y(0) \mid X]$. Under standard unconfoundedness (i.e., $Y(1)$ and $Y(0)$ are independent of $T$ given $X$) and a correctly specified propensity score, $\mathbb{E}\left( \frac{Y(T-\Phi(\beta^\top X))}{\Phi(\beta^\top X)(1-\Phi(\beta^\top X))} \,\middle|\, X \right) = \tau(X)$ a.s. The residual for testing the propensity score model alone is $\varepsilon_0 = T - \Phi(\beta^\top X)$. For testing the joint hypothesis of correct propensity score specification and zero CATE (i.e., $\tau(X) = 0$), the residual vector is $\varepsilon_0 = \begin{pmatrix} T - \Phi(\beta^\top X) \\ \frac{Y(T-\Phi(\beta^\top X))}{\Phi(\beta^\top X)(1-\Phi(\beta^\top X))} \end{pmatrix}$. For joint testing, the statistic sums the component-wise test statistics.

\paragraph{Results.}
The projection tests are implemented under the sample-splitting protocol of Section~\ref{sec:cv_selection} ($15\%$ selection / $85\%$ testing), with the tuning parameters selected by $5$-fold CV on the selection sample over the grid of Section~\ref{sec:cv_selection}; the null distributions are approximated using the multiplier bootstrap with $B=500$ replications. As a benchmark, the KCM test is computed on the full sample, without splitting, at the median rule and at $\gamma = 0.5$, exactly as in the simulation studies; the ICM test is not reported for these data because its wild-bootstrap implementation is not straightforward in the context of a probit model. Results are presented in Table~\ref{tab:realdata}.

None of the test statistics rejects the null hypothesis of correct propensity score specification; all $p$-values substantially exceed conventional significance levels. By contrast, every statistic strongly rejects the joint null of correct propensity score specification and zero CATE, with $p$-values below the bootstrap resolution ($<0.002$). The probit model thus fits the data well, and the evidence points to a nonzero treatment effect.

\medskip\noindent To assess the sensitivity of these conclusions to the random split, we repeat the protocol on $50$ random $15\%/85\%$ splits. The individual $p$-values of the projection tests vary across splits (min, median, max: $0.030$, $0.359$, $0.996$ for $\hat T^{(1)}$ and $0.072$, $0.438$, $0.916$ for $\hat T^{(2)}$), but the verdict is stable: at the $5\%$ level the individual test rejects in $4\%$ of the splits for $\hat T^{(1)}$ and in none for $\hat T^{(2)}$, while the joint test rejects in all $50$ splits for both statistics. The KCM benchmark, which is computed on the full sample and does not depend on the split, continues not to reject the propensity-score specification ($p = 0.572$ at the median rule and $0.630$ at $\gamma = 0.5$; Table~\ref{tab:realdata}) while rejecting the joint null. The CV-selected regularization is stable as well: $\hat T^{(1)}$ selects the largest admissible $\lambda$ in every split, and $\hat T^{(2)}$ in $56\%$ of them. The conclusions of Table~\ref{tab:realdata} are therefore not an artifact of a single split.

\begin{table}[htbp]
    \centering
    \caption{Bootstrap $p$-values of different test statistics for individual and joint tests}
    \label{tab:realdata}
    \begin{tabularx}{\textwidth}{@{}l>{\centering\arraybackslash}X>{\centering\arraybackslash}X@{}}
        \toprule
         & Individual $p$-value & Joint $p$-value  \\
        \midrule
        $\hat T_{\mathrm{proj},\perp}^{(1)}$ & 0.598 & $<0.002$ \\

        $\hat T_{\mathrm{proj},\perp}^{(2)}$  & 0.632 & $<0.002$ \\

        $\hat T_{\mathrm{KCM},\mathrm{median}}$ & 0.572  & $<0.002$ \\
        $\hat T_{\mathrm{KCM},\gamma=0.5}$ & 0.630 & $<0.002$ \\
        \bottomrule
    \end{tabularx}
    \begin{tablenotes}
        \footnotesize
        \item \textit{Note:} The $p$-values are based on $B=500$ bootstrap replications; values below the bootstrap resolution $1/501 \approx 0.002$ are reported as $<0.002$.
    \end{tablenotes}
\end{table}

\paragraph{Second application: returns to schooling.}
The second application studies the Mincer wage equation on the Wooldridge \texttt{wage2} data ($n = 935$). The null is the conditional-mean restriction $\mathbb{E}[\log(\mathrm{wage}) \mid X] = X^\top\beta$, tested for two nested specifications: the parsimonious Mincer set $X = (1, \mathrm{educ}, \mathrm{exper}, \mathrm{tenure})$ and the augmented set adding $\mathrm{exper}^2$, $\mathrm{tenure}^2$, and $\mathrm{IQ}$. The OLS score is $X$, an $X$-measurable direction, so the projection tests apply verbatim with the protocol, grid, and bootstrap of the NSW application.

Neither the projection statistics nor the KCM benchmark rejects either specification at the $5\%$ level ($\hat T^{(1)}$: $p = 0.180$ for the parsimonious specification and $p = 0.544$ for the augmented; $\hat T^{(2)}$: $p = 0.178$ and $0.544$; KCM: $p = 0.602$ and $0.386$ for the parsimonious specification, $0.408$ and $0.302$ for the augmented, at the median rule and $\gamma = 0.5$; Table~\ref{tab:wage}). The two procedures agree, and the directional augmentation moves the projection statistics toward larger $p$-values.

A split-sensitivity check over $50$ random splits confirms the verdict: no split rejects either specification at the $5\%$ level for either statistic ($p$-values between $0.060$ and $0.898$). The CV-selected regularization concentrates at the largest admissible $\lambda$ ($90\%$ of the splits for $\hat T^{(1)}$ under the parsimonious specification, $98\%$ under the augmented), the noise-dominated regime to which item (ii) of Theorem~\ref{th:cv_snr} draws the selector when the signal channel is weak.

\begin{table}[htbp]
    \centering
    \caption{Bootstrap $p$-values for the Mincer wage specification tests (Wooldridge \texttt{wage2}, $n = 935$)}
    \label{tab:wage}
    \begin{tabularx}{\textwidth}{@{}l>{\centering\arraybackslash}X>{\centering\arraybackslash}X>{\centering\arraybackslash}X>{\centering\arraybackslash}X@{}}
        \toprule
        Specification & $\hat T^{(1)}_{\mathrm{proj},\perp}$ & $\hat T^{(2)}_{\mathrm{proj},\perp}$ & $\hat T_{\mathrm{KCM},\mathrm{median}}$ & $\hat T_{\mathrm{KCM},\gamma=0.5}$ \\
        \midrule
        Parsimonious $(1, \mathrm{educ}, \mathrm{exper}, \mathrm{tenure})$ & 0.180 & 0.178 & 0.602 & 0.386 \\
        Augmented ($+ \mathrm{exper}^2, \mathrm{tenure}^2, \mathrm{IQ}$) & 0.544 & 0.544 & 0.408 & 0.302 \\
        \bottomrule
    \end{tabularx}
    \begin{tablenotes}
        \footnotesize
        \item \textit{Note:} The $p$-values are based on $B=500$ bootstrap replications under the sample-splitting protocol of Table~\ref{tab:realdata}; values below the bootstrap resolution $1/501 \approx 0.002$ are reported as $<0.002$.
    \end{tablenotes}
\end{table}

\paragraph{Third application: the demand for cigarettes with instrumental variables.}
The third application tests the specification of a log-linear demand equation for cigarettes on the pooled state panel of \citet{baltagi1992taxation}: 46 states over 1963--1992, $n = 1380$, treated as a random sample, as in the textbook two-stage least squares analysis of these data. The equation is $\ln \mathrm{sales} = \beta_1 + \beta_2 \ln P + \beta_3 \ln \mathrm{ndi} + u$, where sales is per-capita packs, $P$ the real average retail price, and $\mathrm{ndi}$ real per-capita disposable income. The price is endogenous, and the equation is estimated by two-stage least squares with the real tax-inclusive minimum price as the excluded instrument; the full-sample estimates are a price elasticity of $-0.757$ and an income elasticity of $0.248$, in line with the received estimates for these data.

The null is the conditional moment restriction $\mathbb{E}[u \mid Z] = 0$ with $Z = (1, \ln \mathrm{pimin}, \ln \mathrm{ndi})$ the exogenous instruments, and the kernel is built on the instrument variables, so the null equivalence of Theorem~\ref{th:equivalence} applies with $Z$ in place of $X$ as the conditioning variable. The raw score $\nabla_\theta u = -X$ with $X = (1, \ln P, \ln \mathrm{ndi})$ is not $Z$-measurable, because the endogenous $\ln P$ is not a function of the instruments; the sieve step of Assumption~\ref{ass:F}, vacuous in the first two applications, is therefore not vacuous here. The sieve is the linear span of the instruments: with $\mathbb{G}$ the raw-score matrix, the sieve estimate of the conditional score is $\hat{\mathbb{G}} = \boldsymbol{P}_Z\mathbb{G}$, the fitted values of the regressors on the instrument set $(1, \widehat{\ln P}, \ln \mathrm{ndi})$, where the included exogenous $\ln \mathrm{ndi}$ projects onto itself and the endogenous $\ln P$ is replaced by its first-stage fitted values. The orthogonalization is the sieve annihilator $\hat{\boldsymbol{\Pi}} = \boldsymbol{I}_n - \hat{\mathbb{G}}(\hat{\mathbb{G}}^\top\hat{\mathbb{G}})^{-1}\hat{\mathbb{G}}^\top$; it removes the 2SLS estimation effect up to the asymptotically negligible RKHS remainder of Lemma~\ref{lemma:estimation_remainder}, and under the null, $\mathbb{E}[\varepsilon_0\, g_{\theta_0}(Z)^\top] = \boldsymbol{0}$ by iterated expectations, so the projected null law of Corollary~\ref{cor:proj_null} applies with $g_{\theta_0}(Z) = \mathbb{E}[\nabla_\theta u \mid Z] = -\mathbb{E}[X \mid Z]$. Projecting off the raw regressors instead would leave the orthogonalized residual with conditional mean $-\mathbb{E}[X \mid Z]^\top\beta_g$, where $\beta_g = \mathbb{E}[XX^\top]^{-1}\mathbb{E}[Xu] \neq \boldsymbol{0}$ under endogeneity: the raw annihilator, exact as it is against the estimation effect, would confound the test with the endogeneity itself.

Both projection statistics reject the log-linear specification ($p < 0.002$; Table~\ref{tab:cigar}), as does the fixed-$\gamma$ KCM benchmark, so the misspecification is strong enough that every procedure detects it. The translog augmentation (adding $(\ln P)^2$, $(\ln \mathrm{ndi})^2$, and $\ln P \cdot \ln \mathrm{ndi}$) is rejected at the reference split ($p < 0.002$), with a price elasticity at the sample mean of $-0.788$, close to the log-linear estimate. A split-sensitivity check over $50$ random splits leaves the log-linear verdict untouched, with rejections in every split for both statistics (median $p$ below the bootstrap resolution); the rejection of the translog augmentation, by contrast, is split-dependent ($48\%$ of the splits at the $5\%$ level, median $p = 0.091$ and $0.104$). Because the panel is pooled, the observations share within-state dependence; under such dependence the iid multiplier bootstrap understates the null variance, so the reported $p$-values are optimistic and the rejection should be read as indicative rather than exact.

\begin{table}[htbp]
    \centering
    \caption{Bootstrap $p$-values for the cigarette demand specification tests (pooled state panel, $n = 1380$)}
    \label{tab:cigar}
    \begin{tabularx}{\textwidth}{@{}l>{\centering\arraybackslash}X>{\centering\arraybackslash}X>{\centering\arraybackslash}X>{\centering\arraybackslash}X@{}}
        \toprule
        Specification & $\hat T^{(1)}_{\mathrm{proj},\perp}$ & $\hat T^{(2)}_{\mathrm{proj},\perp}$ & $\hat T_{\mathrm{KCM},\mathrm{median}}$ & $\hat T_{\mathrm{KCM},\gamma=0.5}$ \\
        \midrule
        Log-linear $(1, \ln P, \ln \mathrm{ndi})$ & $<0.002$ & $<0.002$ & $<0.002$ & $<0.002$ \\
        Translog ($+ (\ln P)^2, (\ln \mathrm{ndi})^2, \ln P \cdot \ln \mathrm{ndi}$) & $<0.002$ & $<0.002$ & $<0.002$ & $<0.002$ \\
        \bottomrule
    \end{tabularx}
    \begin{tablenotes}
        \footnotesize
        \item \textit{Note:} Two-stage least squares with the real tax-inclusive minimum price as the excluded instrument; orthogonalization by the sieve annihilator $\hat{\boldsymbol{\Pi}}$ with the instrument-fitted score $\hat{\mathbb{G}} = \boldsymbol{P}_Z\mathbb{G}$ (Assumption~\ref{ass:F}, Section~\ref{sec:estimation_effect}). The $p$-values are based on $B=500$ bootstrap replications under the protocol of Table~\ref{tab:realdata}; values below the bootstrap resolution $1/501 \approx 0.002$ are reported as $<0.002$.
    \end{tablenotes}
\end{table}

\section{Conclusion}
\label{sec:conclusion}
This paper develops a reproducing kernel-based test for conditional moment restrictions. The test performs kernel ridge regression of the estimated residuals on kernel functions and uses the fit as a diagnostic: under the null the population fit vanishes, while under an alternative it is a nonzero element of the RKHS. Projecting the fit onto itself and onto the mean embedding of the residuals yields two test statistics; Neyman orthogonality projections robustify them against estimation effects, and the multiplier bootstrap delivers critical values.

A distinctive feature of the method is that the kernel and regularization parameters are selected by cross-validation rather than fixed by the user. Under a sample-splitting protocol, the selector has a well-defined asymptotic law, the bootstrap test controls size without any selector limit, and an energy decomposition shows that the CV choice maximizes the captured signal (the test's noncentrality) and, under the null, shrinks the noise floor of the test. Relative to the Gaussian-process approach to model checks of \cite{escanciano2024gaussian}, whose class holds the kernel fixed and poses the data-driven selection of the kernel-side tuning as an open problem, the present results resolve that problem for the projection family. Open problems include the data-driven selection of the prior mean and of non-translation-invariant kernels within that class, selection without sample splitting, and optimality of the selector.

Simulations confirm that the test controls size without user-chosen tuning in most configurations and matches the classical kernel conditional moment test at its best kernel parameter, with a clear advantage in higher dimensions, where fixed kernel parameters can be catastrophic. An application to the National Supported Work dataset illustrates the method on the specification of a propensity score model and on the null of a zero conditional average treatment effect. A second application to the Mincer wage equation complements the first: where the NSW propensity-score specification is accepted, the linear schooling specification and its augmentation with quadratic experience, quadratic tenure, and ability are likewise not rejected, by the projection tests or by the fixed-kernel KCM benchmark. A third application tests a log-linear demand equation for cigarettes on instrumental-variables residuals: because the endogenous price is not a function of the instruments, the orthogonalization projects the regressors off the instrument span before testing, the projection tests and the fixed-$\gamma$ KCM benchmark both reject the log-linear specification, and a translog augmentation weakens the evidence against the specification to borderline levels.

\paragraph{Practical guidance.}
The empirical evidence supports the following recommendations. (i) The sample-splitting protocol with a $15\%$--$20\%$ selection share removes every user choice: the projection tests control size without a selector limit, and the two statistics behave identically throughout. (ii) The power of the selected test is governed by the split ratio, with a modest gain when the oracle protocol enlarges both the selection sample (from $60$ to $200$ observations) and the testing sample (from $340$ to $400$), and is otherwise insensitive to the allocation. (iii) In high dimensions the median rule of the classical kernel test can be catastrophic, while the CV rule continues to track the envelope of the fixed kernel parameters. The accurate reading is not that the CV rule beats the classical test at its best fixed kernel parameter, but that it matches that performance while requiring no user choices and remaining reliable where fixed rules fail.

\bibliography{paper-ref}

@inproceedings{muandet2020kernel,
  title={Kernel conditional moment test via maximum moment restriction},
  author={Muandet, Krikamol and Jitkrittum, Wittawat and K{\"u}bler, Jonas},
  booktitle={Conference on Uncertainty in Artificial Intelligence},
  pages={41--50},
  year={2020},
  organization={PMLR}
}

@article{sancetta2022testing,
  title     = {Testing subspace restrictions in the presence of high dimensional nuisance parameters},
  author    = {Sancetta, Alessio},
  journal   = {Electronic Journal of Statistics},
  volume    = {16},
  number    = {2},
  pages     = {5277--5320},
  year      = {2022},
  publisher = {The Institute of Mathematical Statistics and the Bernoulli Society}
}

@article{gretton2012optimal,
  title={Optimal kernel choice for large-scale two-sample tests},
  author={Gretton, Arthur and Sejdinovic, Dino and Strathmann, Heiko and Balakrishnan, Sivaraman and Pontil, Massimiliano and Fukumizu, Kenji and Sriperumbudur, Bharath K},
  journal={Advances in neural information processing systems},
  volume={25},
  year={2012}
}

@article{bierens1990consistent,
  title={A consistent conditional moment test of functional form},
  author={Bierens, Herman J},
  journal={Econometrica},
  volume={58},
  number={6},
  pages={1443--1458},
  year={1990},
  publisher={JSTOR}
}

@article{bierens1997asymptotic,
  title={Asymptotic theory of integrated conditional moment tests},
  author={Bierens, Herman J and Ploberger, Werner},
  journal={Econometrica: Journal of the Econometric Society},
  pages={1129--1151},
  year={1997},
  publisher={JSTOR}
}

@article{muandet2017kernel,
  title={Kernel mean embedding of distributions: A review and beyond},
  author={Muandet, Krikamol and Fukumizu, Kenji and Sriperumbudur, Bharath and Sch{\"o}lkopf, Bernhard and others},
  journal={Foundations and Trends{\textregistered} in Machine Learning},
  volume={10},
  number={1-2},
  pages={1--141},
  year={2017},
  publisher={Now Publishers, Inc.}
}

@article{zheng1996consistent,
  title={A consistent test of functional form via nonparametric estimation techniques},
  author={Zheng, John Xu},
  journal={Journal of Econometrics},
  volume={75},
  number={2},
  pages={263--289},
  year={1996},
  publisher={Elsevier}
}

@article{bierens1982consistent,
  title={Consistent model specification tests},
  author={Bierens, Herman J},
  journal={Journal of Econometrics},
  volume={20},
  number={1},
  pages={105--134},
  year={1982},
  publisher={Elsevier}
}

@article{stute1997nonparametric,
  title={Nonparametric model checks for regression},
  author={Stute, Winfried},
  journal={The Annals of Statistics},
  pages={613--641},
  year={1997},
  publisher={JSTOR}
}

@article{escanciano2014specification,
  title={Specification analysis of linear quantile models},
  author={Escanciano, Juan Carlos and Goh, Sze-Chuan},
  journal={Journal of Econometrics},
  volume={178},
  pages={495--507},
  year={2014},
  publisher={Elsevier}
}

@article{sant2019specification,
  title={Specification tests for the propensity score},
  author={Sant’Anna, Pedro HC and Song, Xiaojun},
  journal={Journal of Econometrics},
  volume={210},
  number={2},
  pages={379--404},
  year={2019},
  publisher={Elsevier}
}

@article{newey1985generalized,
  title={Generalized method of moments specification testing},
  author={Newey, Whitney K},
  journal={Journal of econometrics},
  volume={29},
  number={3},
  pages={229--256},
  year={1985},
  publisher={Elsevier}
}

@article{delgado2006consistent,
  title={Consistent tests of conditional moment restrictions},
  author={Delgado, Miguel A and Dom{\'\i}nguez, Manuel A and Lavergne, Pascal},
  journal={Annales d'{\'E}conomie et de Statistique},
  pages={33--67},
  year={2006},
  publisher={JSTOR}
}

@article{escanciano2006consistent,
  title={A consistent diagnostic test for regression models using projections},
  author={Escanciano, J Carlos},
  journal={Econometric Theory},
  volume={22},
  number={6},
  pages={1030--1051},
  year={2006},
  publisher={Cambridge University Press}
}

@article{mammen1993bootstrap,
  title={Bootstrap and wild bootstrap for high dimensional linear models},
  author={Mammen, Enno},
  journal={The annals of statistics},
  volume={21},
  number={1},
  pages={255--285},
  year={1993},
  publisher={Institute of Mathematical Statistics}
}

@book{wendland2004scattered,
  title={Scattered data approximation},
  author={Wendland, Holger},
  volume={17},
  year={2005},
  publisher={Cambridge university press}
}

@book{kato1980,
  title={Perturbation theory for linear operators},
  author={Kato, Tosio},
  edition={2nd},
  year={1980},
  publisher={Springer},
  address={Berlin}
}

@article{carrasco2007linear,
  title={Linear inverse problems in structural econometrics estimation based on spectral decomposition and regularization},
  author={Carrasco, Marine and Florens, Jean-Pierre and Renault, Eric},
  journal={Handbook of econometrics},
  volume={6},
  pages={5633--5751},
  year={2007},
  publisher={Elsevier}
}

@article{escanciano2009lack,
  title={On the lack of power of omnibus specification tests},
  author={Escanciano, J Carlos},
  journal={Econometric Theory},
  volume={25},
  number={1},
  pages={162--194},
  year={2009},
  publisher={Cambridge University Press}
}

@article{sriperumbudur2010hilbert,
  title={Hilbert space embeddings and metrics on probability measures},
  author={Sriperumbudur, Bharath K and Gretton, Arthur and Fukumizu, Kenji and Sch{\"o}lkopf, Bernhard and Lanckriet, Gert RG},
  journal={The Journal of Machine Learning Research},
  volume={11},
  pages={1517--1561},
  year={2010},
  publisher={JMLR. org}
}

@article{newey1985maximum,
  title={Maximum likelihood specification testing and conditional moment tests},
  author={Newey, Whitney K},
  journal={Econometrica: Journal of the Econometric Society},
  pages={1047--1070},
  year={1985},
  publisher={JSTOR}
}

@article{tauchen1985diagnostic,
  title={Diagnostic testing and evaluation of maximum likelihood models},
  author={Tauchen, George},
  journal={Journal of Econometrics},
  volume={30},
  number={1-2},
  pages={415--443},
  year={1985},
  publisher={Elsevier}
}

@article{wooldridge1990unified,
  title={A unified approach to robust, regression-based specification tests},
  author={Wooldridge, Jeffrey M},
  journal={Econometric Theory},
  volume={6},
  number={1},
  pages={17--43},
  year={1990},
  publisher={Cambridge University Press}
}

@article{eubank1990testing,
  title={Testing the goodness of fit of a linear model via nonparametric regression techniques},
  author={Eubank, Randall L and Spiegelman, Clifford H},
  journal={Journal of the American Statistical Association},
  volume={85},
  number={410},
  pages={387--392},
  year={1990},
  publisher={Taylor \& Francis}
}

@article{hardle1993comparing,
  title={Comparing nonparametric versus parametric regression fits},
  author={Hardle, Wolfgang and Mammen, Enno},
  journal={The Annals of Statistics},
  pages={1926--1947},
  year={1993},
  publisher={JSTOR}
}

@article{hong1995consistent,
  title={Consistent specification testing via nonparametric series regression},
  author={Hong, Yongmiao and White, Halbert},
  journal={Econometrica: Journal of the Econometric Society},
  pages={1133--1159},
  year={1995},
  publisher={JSTOR}
}

@article{ellison2000simple,
  title={A simple framework for nonparametric specification testing},
  author={Ellison, Glenn and Ellison, Sara Fisher},
  journal={Journal of Econometrics},
  volume={96},
  number={1},
  pages={1--23},
  year={2000},
  publisher={Elsevier}
}

@article{escanciano2024gaussian,
  title={A Gaussian process approach to model checks},
  author={Escanciano, Juan Carlos},
  journal={The Annals of Statistics},
  volume={52},
  number={5},
  pages={2456--2481},
  year={2024},
  publisher={Institute of Mathematical Statistics}
}

@inproceedings{sutherland2021generativemodelsmodelcriticism,
      title={Generative models and model criticism via optimized maximum mean discrepancy},
      author={Sutherland, Danica J and Tung, Hsiao-Yu and Strathmann, Heiko and De, Soumyajit and Ramdas, Aaditya and Smola, Alex and Gretton, Arthur},
      booktitle={International Conference on Learning Representations},
      year={2017},
      note={arXiv:1611.04488}
}

@inproceedings{liu2021learningdeepkernelsnonparametric,
      title={Learning deep kernels for non-parametric two-sample tests},
      author={Liu, Feng and Xu, Wenkai and Lu, Jie and Zhang, Guangquan and Gretton, Arthur and Sutherland, Danica J},
      booktitle={International conference on machine learning},
      pages={6316--6326},
      year={2020},
      organization={PMLR}
}

@book{vaart1997weak,
  title     = {Weak convergence and empirical processes with applications to statistics},
  author    = {van der Vaart, A. W. and Wellner, Jon A.},
  year      = {1996},
  publisher = {Springer}
}

@article{dehejia1999causal,
  title     = {Causal effects in nonexperimental studies: Reevaluating the evaluation of training programs},
  author    = {Dehejia, Rajeev H and Wahba, Sadek},
  journal   = {Journal of the American statistical Association},
  volume    = {94},
  number    = {448},
  pages     = {1053--1062},
  year      = {1999},
  publisher = {Taylor \& Francis}
}

@article{baltagi1992taxation,
  title     = {Cigarette Taxation: Raising Revenues and Reducing Consumption},
  author    = {Baltagi, Badi H. and Levin, Dan},
  journal   = {Structural Change and Economic Dynamics},
  volume    = {3},
  number    = {2},
  pages     = {321--335},
  year      = {1992}
}

@article{lalonde1986evaluating,
  title     = {Evaluating the econometric evaluations of training programs with experimental data},
  author    = {LaLonde, Robert J},
  journal   = {The American economic review},
  volume    = {76},
  number    = {4},
  pages     = {604--620},
  year      = {1986},
  publisher = {JSTOR}
}

@article{micchelli2006universal,
  title   = {Universal Kernels},
  author  = {Micchelli, Charles A and Xu, Yuesheng and Zhang, Haizhang},
  journal = {Journal of Machine Learning Research},
  volume  = {7},
  number  = {12},
  pages   = {2651--2667},
  year    = {2006}
}

@article{gonzalez2013updated,
  title     = {An updated review of goodness-of-fit tests for regression models},
  author    = {Gonz{\'a}lez-Manteiga, Wenceslao and Crujeiras, Rosa M},
  journal   = {Test},
  volume    = {22},
  pages     = {361--411},
  year      = {2013},
  publisher = {Springer}
}

@book{hsing2015theoretical,
  title     = {Theoretical foundations of functional data analysis, with an introduction to linear operators},
  author    = {Hsing, Tailen and Eubank, Randall L},
  year      = {2015},
  publisher = {Wiley Online Library}
}

@article{sun2005mercer,
  title     = {Mercer theorem for RKHS on noncompact sets},
  author    = {Sun, Hongwei},
  journal   = {Journal of Complexity},
  volume    = {21},
  number    = {3},
  pages     = {337--349},
  year      = {2005},
  publisher = {Elsevier}
}

@article{rosasco2010regularization,
  title     = {On Learning with Integral Operators},
  author    = {Rosasco, Lorenzo and Belkin, Mikhail and De Vito, Ernesto},
  journal   = {Journal of Machine Learning Research},
  volume    = {11},
  pages     = {905--934},
  year      = {2010}
}

@article{gao2008bandwidth,
  title={Bandwidth selection in nonparametric kernel testing},
  author={Gao, Jiti and Gijbels, Ir{\`e}ne},
  journal={Journal of the American Statistical Association},
  volume={103},
  number={484},
  pages={1584--1594},
  year={2008},
  publisher={Taylor \& Francis}
}

@article{horowitz2001adaptive,
  title={An adaptive, rate-optimal test of a parametric mean-regression model against a nonparametric alternative},
  author={Horowitz, Joel L. and Spokoiny, Vladimir G.},
  journal={Econometrica},
  volume={69},
  number={3},
  pages={599--631},
  year={2001},
  publisher={Wiley}
}

@article{lavergne2008breaking,
  title={Breaking the curse of dimensionality in nonparametric testing},
  author={Lavergne, Pascal and Patilea, Valentin},
  journal={Journal of Econometrics},
  volume={143},
  number={1},
  pages={103--122},
  year={2008},
  publisher={Elsevier}
}

@article{guerrelavergne2002optimal,
  title={Optimal minimax rates for nonparametric specification testing in regression models},
  author={Guerre, Emmanuel and Lavergne, Pascal},
  journal={Econometric Theory},
  volume={18},
  number={4},
  pages={792--814},
  year={2002},
  publisher={Cambridge University Press}
}

@article{fan2015power,
  title={Power enhancement in high-dimensional cross-sectional tests},
  author={Fan, Jianqing and Liao, Yuan and Yao, Jiawei},
  journal={Econometrica},
  volume={83},
  number={4},
  pages={1497--1541},
  year={2015},
  publisher={Wiley}
}

@book{delapena1999decoupling,
  title={Decoupling: From Dependence to Independence},
  author={de la Pe{\~n}a, Victor H. and Gin{\'e}, Evarist},
  publisher={Springer},
  address={New York},
  year={1999}
}

@book{lehmannromano2005,
  title={Testing Statistical Hypotheses},
  author={Lehmann, Erich Leo and Romano, Joseph P},
  edition={3rd},
  year={2005},
  publisher={Springer},
  address={New York}
}

\appendix
\begin{center}
    \Large\textbf{Appendix}
\end{center}
This supplement is organized as follows. Section~\ref{sec:RKHS_intro} reviews the background on reproducing kernel Hilbert spaces. Section~\ref{app:proofs} contains the proofs of all results stated in the main text, organized by the main-paper section whose results they prove; each subsection opens with the auxiliary lemmas on which its proofs rely, together with an indication of where each lemma enters. Sections~\ref{sec:lln_clt_rkhs}--\ref{orthogonal_transformations} collect the technical tools used by the proofs, each with its supporting results: random elements in Hilbert spaces, the singular value decomposition of the restriction operator, useful non-asymptotic bounds, and orthogonal transformations. Section~\ref{app:additional_sim} reports additional simulation designs and results.

\section{Introduction to Reproducing Kernel Hilbert Spaces (RKHS)}
\label{sec:RKHS_intro}
Reproducing kernel Hilbert spaces (RKHS) are central to kernel methods in machine learning and statistics. For rigorous and comprehensive treatments of this topic, we refer readers to \cite{wendland2004scattered,rosasco2010regularization,muandet2017kernel}.

Let $\mathcal{S} \subset \mathbb{R}^d$ be a non-empty set. A kernel $k: \mathcal{S} \times \mathcal{S} \to \mathbb{R}$ is a function that satisfies, for all $s, t \in \mathcal{S}$,
$$
k(s, t) = \langle \psi(s), \psi(t) \rangle_{\mathcal{H}_k},
$$
where $\psi: \mathcal{S} \to \mathcal{H}_k$ is a feature map, and $\mathcal{H}_k$ is a Hilbert space equipped with the inner product $\langle \cdot, \cdot \rangle_{\mathcal{H}_k}$. The feature map $\psi$ implicitly maps the input space $\mathcal{S}$ into a high-dimensional (possibly infinite-dimensional) feature space $\mathcal{H}_k$.

A typical example of a kernel is the Gaussian kernel, defined as
$$
k(s, t) = \exp\left(-\gamma\|s - t\|^2\right), \quad \gamma > 0.
$$
For the special case where $\gamma = 1/2$ and $d = 1$, the corresponding feature map is given by
$$
\psi(t) = \left(\exp(-t^2/2), t\exp(-t^2/2), \frac{t^2\exp(-t^2/2)}{\sqrt{2}}, \ldots, \frac{t^i\exp(-t^2/2)}{\sqrt{i!}}, \ldots\right),
$$
which maps $t$ to an infinite-dimensional space. Such kernels are called positive definite when, for any $c_i \in \mathbb{R}$, $t_i \in \mathcal{S}$, and $i = 1, \ldots, n$,
$$
\sum_{i=1}^n \sum_{j=1}^n c_i c_j k(t_i, t_j) \geq 0.
$$

The Moore--Aronszajn theorem guarantees that for every positive definite kernel $k$, there exists a unique RKHS $\mathcal{H}_k$ associated with it. This RKHS is a Hilbert space consisting of functions $f: \mathcal{S} \to \mathbb{R}$ such that, for all $t \in \mathcal{S}$, the (Dirac) evaluation functional $E_t: \mathcal{H}_k \to \mathbb{R}$, defined by $E_t(f) = f(t)$, is continuous. This continuity ensures the existence of a reproducing kernel, which satisfies the reproducing property:
$$
f(t) = \langle f, k(\cdot, t) \rangle_{\mathcal{H}_k}, \quad \forall f \in \mathcal{H}_k, \; t \in \mathcal{S}.
$$
Evaluation thus reduces to an inner product against the kernel, which simplifies computation throughout. In particular, the kernel function itself satisfies
$$
k(s, t) = \langle k(\cdot, s), k(\cdot, t) \rangle_{\mathcal{H}_k}, \quad \forall s, t \in \mathcal{S}.
$$

\section{Proofs of the Main Results}
\label{app:proofs}
Three recurring mechanisms organize the technical machinery. First, the rank-$d_{\theta}$ score projection $\boldsymbol\Pi$ reshapes every null law through the projected kernels $K^{\mathrm{sym},\perp}$ and $K^{D,\perp}$, so Section~\ref{proof:aux_uniform} develops the transport and Lipschitz calculus that tracks their dependence on $(\gamma,\lambda)$. Second, Fourier transport to a common RKHS supplies the $\gamma$-Lipschitz control underlying every uniform modulus. Third, split independence purchases bootstrap validity with no selector asymptotics: conditional on the realized tuning, the bootstrap reduces to a fixed-$\eta$ multiplier CLT.

Each subsection opens with the auxiliary lemmas on which its proofs rely, together with a brief indication of where each lemma enters, and the proofs follow the order of the corresponding main-text results. Operator-theoretic and probabilistic tools used across subsections are gathered in the subsequent appendices.

\subsection{Proofs for Section~\ref{sec:stat}: Test Statistics and Their Asymptotic Properties}\label{app:proofs_stat}

\noindent This subsection proves the five results of Section~\ref{sec:stat}. Two lemmas come first because both later proofs invoke them: Lemma~\ref{lemma:rotatedbasis} re-expresses the scaled KRR estimator in a rotated empirical eigenbasis that converges uniformly to the population eigenfunctions within each eigenvalue cluster, which is what allows the test statistics to be compared with their spectral representations despite eigenvalue multiplicity; and Lemma~\ref{lemma:estimation_consistency} gives the explicit form of the population solution $w^*_\lambda$, the object whose vanishing under $H_0$ is at stake. On this foundation, Theorem~\ref{th:equivalence} characterizes the null hypothesis as the RKHS statement $w^*_\lambda=\boldsymbol 0$,
Theorem~\ref{th:proj_null_fix} derives the limits of the projection statistic under the null hypothesis and fixed alternatives,
and Theorem~\ref{th:proj_local} derives its local-alternative limit.

\subsubsection{Proof of Lemma~\ref{lemma:rotatedbasis}}
\label{proof:rotatedbasis}

Fix the $k$-th eigenvalue cluster (multiplicity $m_k$) and set $v_{k-1} = \sum_{l=1}^{k-1}m_l$. Write $\nu_i := \sigma_i^2/n$ for the eigenvalues of $\boldsymbol K/n$ and define the vector of empirical eigenfunctions and the vector of directional coefficients in this cluster:
\[
\hat{\mathbf{e}}_k := (\hat e_{v_{k-1}+1},\ldots,\hat e_{v_k})^\top, \qquad
    \mathbf{D}_{n,k} := \bigl(\langle\sqrt{n}\hat m,\hat e_{v_{k-1}+1}\rangle_{\HK},\ldots,\langle\sqrt{n}\hat m,\hat e_{v_k}\rangle_{\HK}\bigr)^\top.
\]
The contribution of the $k$-th cluster to the empirical KRR estimator is
\[
\hat{\mathbf{e}}_k^\top W_k\mathbf{D}_{n,k}
    = \sum_{j=1}^{m_k}\frac{1}{\nu_{v_{k-1}+j}+\lambda}\,\langle\sqrt{n}\hat m,\hat e_{v_{k-1}+j}\rangle_{\HK}\,\hat e_{v_{k-1}+j},
\]
\[
W_k := \operatorname{diag}\bigl((\nu_{v_{k-1}+1}+\lambda)^{-1},\ldots,(\nu_{v_k}+\lambda)^{-1}\bigr).
\]

\noindent\textbf{Cluster reduction.}\;
The rotation is applied cluster-by-cluster. Fixing
\[
    E_k:=\operatorname{span}\{e_{v_{k-1}+1},\ldots,e_{v_k}\}
\]
and $\hat E_k$ its empirical counterpart, write $\Delta_k:=P_{E_k}-\hat P_{\hat E_k}$. The sample projector splits across the empirical ordering, $\hat P_v=\sum_{k=1}^V \hat P_{\hat E_k}$, and $P_v=\sum_k P_{E_k}$, so $\Delta_k$ telescopes into prefix differences:
\[
    \Delta_k
    =\bigl(P_{\le v_k}-\hat P_{\le v_k}\bigr)-\bigl(P_{\le v_{k-1}}-\hat P_{\le v_{k-1}}\bigr),
\]
and Lemma~\ref{bound_projections} bounds each prefix, $\lVert P_{\le v_k}-\hat P_{\le v_k}\rVert_{\mathrm{HS}}
=O_p\bigl(n^{-1/2}(\mu^{(v_k)}-\mu^{(v_k+1)})^{-1}\bigr)$, by the triangle inequality giving
$\lVert\Delta_k\rVert_{\mathrm{HS}}
=O_p\bigl(n^{-1/2}\bigl[(\mu^{(v_k)}-\mu^{(v_k+1)})^{-1}+(\mu^{(v_{k-1})}-\mu^{(v_{k-1}+1)})^{-1}\bigr]\bigr)=o_p(1)$
for each fixed cluster under the eigengap condition of Lemma~\ref{bound_projections}. No blockwise
Wedin argument, and no cross-cluster orthogonality, is required: the per-cluster convergence
$\lVert\Delta_k\rVert_{\mathrm{HS}}=o_p(1)$ is what licenses the rotation of each cluster
independently. The orthogonal matrices of Appendix~\ref{orthogonal_transformations} align the
empirical basis to the population one within the cluster at the rate
$\lVert\tilde e_{k,j}-e_{k,j}\rVert_{\HK}=O_p(n^{-1/2})$. For reference, the cluster deviations
satisfy the exact trace expansion
$\lVert P_v-\hat P_v\rVert_{\mathrm{HS}}^{2}
=\sum_k\lVert\Delta_k\rVert_{\mathrm{HS}}^{2}
-2\sum_{k<l}\bigl(\lVert P_{E_k}\hat P_{\hat E_l}\rVert_{\mathrm{HS}}^{2}
+\lVert P_{E_l}\hat P_{\hat E_k}\rVert_{\mathrm{HS}}^{2}\bigr)$,
whose cross terms are nonpositive; the total therefore does not dominate the per-cluster terms, and
the prefix argument above is the operative one.

\noindent\textbf{Rotation to an orthonormal basis.}\;
Let $O_{k,n} \in \mathbb{R}^{m_k\times m_k}$ be the random orthogonal matrix of Appendix~\ref{orthogonal_transformations}, and define the rotated eigenfunctions $\tilde{\mathbf{e}}_k := O_{k,n}\hat{\mathbf{e}}_k$ and the rotated coefficients $\tilde{\mathbf{D}}_k := O_{k,n}\mathbf{D}_{n,k}$. By orthogonality, $\hat{\mathbf{e}}_k = O_{k,n}^\top\tilde{\mathbf{e}}_k$ and $\mathbf{D}_{n,k} = O_{k,n}^\top\tilde{\mathbf{D}}_k$, so the cluster contribution becomes
\[
\hat{\mathbf{e}}_k^\top W_k\mathbf{D}_{n,k} = \tilde{\mathbf{e}}_k^\top\bigl(O_{k,n}W_kO_{k,n}^\top\bigr)\tilde{\mathbf{D}}_k.
\]

\noindent\textbf{Replacing empirical weights by population weights.}\;
By Lemma~\ref{bound_eigenvalues}, the empirical eigenvalues in the $k$-th cluster converge to the population value, $\sup_{j=1}^{m_k}|\nu_{v_{k-1}+j}-\mu^{(k)}| = O_p(n^{-1/2})$, and the weight function $t \mapsto (t+\lambda)^{-1}$ is Lipschitz with constant $\lambda^{-2}$, so
\[
W_k = \frac{1}{\mu^{(k)}+\lambda}\,I_{m_k} + \Delta W_k, \qquad \|\Delta W_k\|_{\mathrm{op}} = O_p(n^{-1/2}) = o_p(1).
\]
The scalar multiple of the identity survives the orthogonal conjugation, $O_{k,n}W_kO_{k,n}^\top = (\mu^{(k)}+\lambda)^{-1}I_{m_k} + O_{k,n}\Delta W_kO_{k,n}^\top$, so the main term is
\[
    \frac{1}{\mu^{(k)}+\lambda}\sum_{j=1}^{m_k}\langle\sqrt{n}\hat m,\tilde e_{k,j}\rangle_{\HK}\,\tilde e_{k,j},
\]
and the remainder is $o_p(1)$: orthogonal transformations preserve the operator norm, $\|\tilde{\mathbf{D}}_k\|_{\ell^2} \le \|\sqrt{n}\hat m\|_{\HK} = O_p(1)$, where the second inequality follows by Bessel's inequality. We show that $\mathbb E\|\sqrt n\,\hat m\|_{\HK}^{2}=O(1)$ in both regimes of the lemma's statement. Writing $\varepsilon_0=f+v$ with $f:=\mathbb E[\varepsilon_0\mid X]$ and $v:=\varepsilon_0-f$, so that $\varepsilon_0=v$ under the null and $\varepsilon_0=R/\sqrt n+v$ under the local alternative,
\[
    \mathbb E\|\sqrt n\,\hat m\|_{\HK}^{2}
    \;=\;\mathbb E[\varepsilon_0^{2}k(X,X)]
       +(n{-}1)\,\mathbb E[\varepsilon_0\varepsilon_0'k(X,X')]
    \;=\;A_n+B_n .
\]
Here $A_n=\mathbb E[\varepsilon_0^{2}k(X,X)]\le\bar k\,\mathbb E[\varepsilon_0^{2}]<\infty$ by Assumption~\ref{ass:B}(v). For the off-diagonal term, the cross terms and the $vv'$ term vanish: the pairs $(X,\varepsilon_0)$ and $(X',\varepsilon_0')$ are i.i.d.\ (Assumption~\ref{ass:C}(i), with $(X',\varepsilon_0')$ an independent copy) and $\mathbb E[v\mid X]=0$, so conditioning on $(X,X')$ in the tower rule gives $\mathbb E[k(X,X')\,\mathbb E[v\mid X]\,\mathbb E[v'\mid X']]=0$ and $\mathbb E[k(X,X')\,\mathbb E[v\mid X]\,f(X')]=0$. Hence, with $m_R:=\mathbb E[R(X)k(X,\cdot)]$ and $m_f:=\mathbb E[f(X)k(X,\cdot)]$,
\begin{align*}
    B_n
    &=(n{-}1)\,\mathbb E[\varepsilon_0\varepsilon_0'k(X,X')]\\
    &=(n{-}1)\iint k(x,x')f(x)f(x')\,d\mathbb P_X\,d\mathbb P_X'\\
    &=(n{-}1)\lVert m_f\rVert_{\HK}^{2}\\
    &=\begin{cases}
        0, & \text{under } H_0 \ (f=0),\\[2pt]
        (1-\tfrac1n)\lVert m_R\rVert_{\HK}^{2}, & \text{under } H_{1,n} \ (m_f=m_R/\sqrt n),
      \end{cases}
\end{align*}
so $B_n=O(1)$ in both regimes. With $\|\tilde{\mathbf{e}}_k\|_{\HK} = \sqrt{m_k}$, Cauchy--Schwarz gives
\[
    \bigl|\tilde{\mathbf{e}}_k^\top\bigl(O_{k,n}\Delta W_kO_{k,n}^\top\bigr)\tilde{\mathbf{D}}_k\bigr|
    \le \sqrt{m_k}\;o_p(1)\;O_p(1) = o_p(1).
\]

\noindent\textbf{The tail.}\;
For the indices beyond the clustered ones no rotation is needed (we keep the empirical eigenfunctions there). By Lemma~\ref{bound_eigenvalues}, $\sum_{i\ge1}(\nu_i-\mu_i)^2 = O_p(n^{-1})$ (with $\nu_i = 0$ for $i > n$), and $\sum_{i=1}^n\langle\sqrt{n}\hat m,\hat e_i\rangle_{\HK}^2 = \|\sqrt{n}\hat m\|_{\HK}^2 = O_p(1)$ by Parseval, so replacing the empirical weights $(\nu_i+\lambda)^{-1}$ by the population weights $(\mu_i+\lambda)^{-1}$ over all remaining indices changes the expansion by at most
\[
    \lambda^{-2}\Bigl(\sum_{i\ge1}(\nu_i-\mu_i)^2\Bigr)^{1/2}\Bigl(\sum_{i=1}^n\langle\sqrt{n}\hat m,\hat e_i\rangle_{\HK}^2\Bigr)^{1/2} = o_p(1),
\]
by Cauchy--Schwarz. Assembling the cluster contributions and the tail,
\[
    \sqrt{n}\,\hat w = \sum_{k=1}^V\sum_{j=1}^{m_k}\frac{1}{\mu^{(k)}+\lambda}\,\langle\sqrt{n}\hat m,\tilde e_{k,j}\rangle_{\HK}\,\tilde e_{k,j}
    \;+\; \sum_{i=v_V+1}^{n}\frac{1}{\mu_i+\lambda}\,\langle\sqrt{n}\hat m,\hat e_i\rangle_{\HK}\,\hat e_i + o_p(1),
\]
and relabeling the eigenvalue sequence (accounting for multiplicity) as $\{\mu_i\}$ and the rotated eigenfunctions as $\{\tilde e_i\}$ (with $\tilde e_i = \hat e_i$ on the tail) gives the compact form stated in the lemma.

\subsubsection{Proof of Lemma~\ref{lemma:estimation_consistency}}
\label{proof:lemma1}

\noindent\textbf{Population expansion.}\;
By Assumption~\ref{ass:A} and Proposition~\ref{prop:population_operators},
the population second-moment operator
$C = \mathbb{E}[k(X,\cdot) \otimes k(X,\cdot)]$ admits the spectral decomposition
$C = \sum_{i=1}^\infty \mu_i\,e_i \otimes e_i$, where $\{\mu_i\}_{i\ge 1}$
are the eigenvalues (identical to those of the integral operator
$\mathbf{L}_k$) and $\{e_i\}_{i\ge 1}$ forms an orthonormal basis of
$\HK$ with $e_i = \sqrt{\mu_i}\,\phi_i$, where $\{\phi_i\}_{i\ge 1}$ are
the orthonormal eigenfunctions of $\mathbf{L}_k$ in $L^2(\mathbb{P}_X)$.
Because $\sum_{i\ge 1}\mu_i < \infty$ (Assumption~\ref{ass:B}(ii)), $C$ is a trace-class,
hence Hilbert--Schmidt and compact operator. The transfer of the eigenpairs of $\boldsymbol{L}_k$ to $C$ rests on three facts. First, the two operators agree pointwise on $\HK$: for $g\in\HK\subset L^2(\mathbb{P}_X)$, $|g(X)|\le\|g\|_{\HK}\sqrt{\bar k}$ by Cauchy--Schwarz and the reproducing bound, so
$(Cg)(y)=\mathbb{E}[k(X,y)\langle k(X,\cdot),g\rangle_{\HK}]
=\mathbb{E}[k(X,y)\,g(X)]=(\boldsymbol{L}_k g)(y)$,
the middle equality being the reproducing property; hence $C(\sqrt{\mu_i}\phi_i)=\sqrt{\mu_i}\boldsymbol{L}_k\phi_i=\mu_i(\sqrt{\mu_i}\phi_i)$, each $\sqrt{\mu_i}\phi_i$ being a $C$-eigenfunction with eigenvalue $\mu_i$. Second, the list $\{\mu_i\}$ exhausts the positive spectrum of $C$: the Mercer diagonal identity $k(x,x)=\sum_{i\ge1}\mu_i\phi_i(x)^2$ and monotone convergence give $\operatorname{Tr}(C)=\mathbb{E}[k(X,X)]=\sum_{i\ge1}\mu_i\,\mathbb{E}[\phi_i(X)^2]=\sum_{i\ge1}\mu_i<\infty$, so $C$ carries no eigenvalue beyond $\{\mu_i\}$. Third, under Assumption~\ref{ass:A} the family $\{\sqrt{\mu_i}\phi_i\}_{i\ge1}$ is an orthonormal basis of $\HK$, so the spectral decomposition of the resolvent below holds on all of $\HK$.

The population coefficient function is $w^*_\lambda
= (C + \lambda I)^{-1}m^*$ with $m^* = \mathbb{E}[\varepsilon_0 k(X,\cdot)]$.
Applying the spectral decomposition of the resolvent,
\[
  w^*_\lambda = \sum_{i=1}^\infty \frac{1}{\mu_i + \lambda}\,
                \langle m^*, e_i\rangle_{\HK}\,e_i .
\]
The inner product expands as
$\langle m^*, e_i\rangle_{\HK}
 = \mathbb{E}[\varepsilon_0\,e_i(X)]
 = \sqrt{\mu_i}\,\mathbb{E}[\varepsilon_0\,\phi_i(X)]$. The first equality passes the inner product through the expectation, which is legitimate because $\varepsilon_0 k(X,\cdot)$ is Bochner integrable, $\mathbb{E}\|\varepsilon_0 k(X,\cdot)\|_{\HK}^2\le\bar k\,\mathbb{E}[\varepsilon_0^2]<\infty$ (Assumptions~\ref{ass:B}(i) and~\ref{ass:B}(v)), whence $\mathbb{E}|\varepsilon_0\,e_i(X)|\le(\mathbb{E}[\varepsilon_0^2])^{1/2}\mu_i^{1/2}<\infty$ by Cauchy--Schwarz; the second equality
uses $e_i = \sqrt{\mu_i}\,\phi_i$.  Substituting $e_i = \sqrt{\mu_i}\,\phi_i$
yields the claimed representation:
\[
  w^*_\lambda = \sum_{i=1}^\infty \frac{\mu_i}{\mu_i + \lambda}\,
                \mathbb{E}[\varepsilon_0\,\phi_i(X)]\,\phi_i .
\]
The series converges in $\HK$. With $a_i:=\mathbb{E}[\varepsilon_0\phi_i(X)]=\mathbb{E}[f(X)\phi_i(X)]$, the equality by the tower property ($\phi_i$ is $X$-measurable, $f=\mathbb{E}(\varepsilon_0\mid X)$), and $\|\phi_i\|_{\HK}^2=\mu_i^{-1}$, the squared norm of any tail satisfies
\[
  \Bigl\|\sum_{i>I}\frac{\mu_i}{\mu_i+\lambda}a_i\phi_i\Bigr\|_{\HK}^2
  = \sum_{i>I}\frac{\mu_i\,a_i^2}{(\mu_i+\lambda)^2}
  \le \lambda^{-2}\sum_{i>I}\mu_i\,a_i^2
  \longrightarrow 0,
\]
because $\sum_{i\ge1}\mu_i\,a_i^2=\langle\boldsymbol{L}_k f,f\rangle_{L^2}\le\bar k\,\mathbb{E}[f(X)^2]<\infty$ by Cauchy--Schwarz and $\|\boldsymbol{L}_k\|\le\bar k$ (Assumption~\ref{ass:B}(i)), with $\mathbb{E}[f(X)^2]\le\mathbb{E}[\varepsilon_0^2]<\infty$ by Jensen; the partial sums therefore form a Cauchy sequence in $\HK$.

\smallskip
\noindent\textbf{Empirical KRR estimator.}\;
With $C_n=\frac1n\sum_l k(x_l,\cdot)\otimes k(x_l,\cdot)$ and
$\hat m=\frac1n\sum_l\varepsilon_{0,l}\,k(x_l,\cdot)$, Proposition~\ref{svd_empirical} gives the duality
identity
\[
  \hat w=(C_n+\lambda I)^{-1}\hat m
        =\boldsymbol R_n^\top(\boldsymbol K/n+\lambda\boldsymbol I)^{-1}\boldsymbol\varepsilon_0 :
\]
Expanding over the empirical eigenpairs $i\le r:=\operatorname{rank}\boldsymbol K$, where $\hat e_i(x_l)=\sigma_i(\boldsymbol u_i)_l$ and $\hat e_i=(n/\sigma_i)\,\boldsymbol R_n^\top\boldsymbol u_i$, confirms the identity. The spectral decomposition of the inverse is $(C_n+\lambda I)^{-1}=\sum_{i\le r}(\nu_i+\lambda)^{-1}\hat e_i\otimes\hat e_i+\lambda^{-1}\bigl(I-\sum_{i\le r}\hat e_i\otimes\hat e_i\bigr)$, and the second summand annihilates $\hat m$ because $\hat m=\boldsymbol R_n^\top\boldsymbol\varepsilon_0\in\operatorname{Im}(\boldsymbol R_n^\top)=\operatorname{span}\{\hat e_1,\dots,\hat e_r\}=\operatorname{Im}(C_n)$; the first summand gives $(C_n+\lambda I)^{-1}\hat m=\sum_{i\le r}(\nu_i+\lambda)^{-1}\nu_i^{1/2}(n^{-1}\tilde{\boldsymbol u}_i^\top\boldsymbol\varepsilon_0)\hat e_i$, which is the expansion of $\boldsymbol R_n^\top(\boldsymbol K/n+\lambda\boldsymbol I)^{-1}\boldsymbol\varepsilon_0$ in the same basis.

\smallskip
\noindent\textbf{Consistency: $\|\hat w - w_\lambda^*\|_{\HK}\to 0$ in probability.}\;
Decompose the error as
\[
  \hat w - w_\lambda^*
  = (C_n + \lambda I)^{-1}(\hat m - m^*)
    + \bigl[(C_n + \lambda I)^{-1} - (C + \lambda I)^{-1}\bigr]m^* .
\]

\textit{First term.}\;
Since $\lambda > 0$, we have $\|(C_n + \lambda I)^{-1}\|_{\mathrm{op}}
\le \lambda^{-1}$ (the spectral mapping theorem applied to the
self-adjoint operator $C_n$).  Therefore
\[
  \bigl\|(C_n + \lambda I)^{-1}(\hat m - m^*)\bigr\|_{\HK}
  \le \lambda^{-1}\,\|\hat m - m^*\|_{\HK}.
\]
Under Assumption~\ref{ass:B}(i), $k$ is bounded, and by Assumption~\ref{ass:B}(v)
$\mathbb{E}\|\varepsilon_0 k(X,\cdot)\|_{\HK}^2
= \mathbb{E}[\varepsilon_0^2 k(X,X)] \le \bar{k}\,\mathbb{E}[\varepsilon_0^2]
< \infty$.  The law of large numbers in the RKHS
(Hilbert-space valued SLLN for i.i.d.\ random elements with finite
second moment; see e.g.\ \cite{hsing2015theoretical}, Theorem~7.1)
gives $\|\hat m - m^*\|_{\HK} \to 0$ in probability.  Hence the first
term is $o_p(1)$.

\textit{Second term.}\;
Apply the resolvent identity
$(A^{-1} - B^{-1}) = A^{-1}(B - A)B^{-1}$ with $A = C_n + \lambda I$,
$B = C + \lambda I$:
\[
  (C_n + \lambda I)^{-1} - (C + \lambda I)^{-1}
  = (C_n + \lambda I)^{-1}\,(C - C_n)\,(C + \lambda I)^{-1}.
\]
Using $\|(C + \lambda I)^{-1}\|_{\mathrm{op}}\le\lambda^{-1}$ and
$\|(C_n + \lambda I)^{-1}\|_{\mathrm{op}}\le\lambda^{-1}$,
\[
  \bigl\|\bigl[(C_n + \lambda I)^{-1} - (C + \lambda I)^{-1}\bigr]m^*\bigr\|_{\HK}
  \le \lambda^{-2}\,\|C - C_n\|_{\mathrm{op}}\,\|m^*\|_{\HK}.
\]

Now, $\|m^*\|_{\HK}$ is finite: with $(X',\varepsilon_0')$ an independent copy of $(X,\varepsilon_0)$ (Assumption~\ref{ass:C}(i)), the Bochner second-moment identity gives $\|m^*\|_{\HK}^2=\mathbb{E}[\varepsilon_0\varepsilon_0'\,k(X,X')]$, the interchange being covered by Tonelli's theorem since $|\varepsilon_0\varepsilon_0'k(X,X')|\le\bar k\,|\varepsilon_0\varepsilon_0'|$ and $\mathbb{E}|\varepsilon_0\varepsilon_0'|=(\mathbb{E}|\varepsilon_0|)^2\le\mathbb{E}[\varepsilon_0^2]<\infty$ by Cauchy--Schwarz; Jensen's inequality then bounds $\|m^*\|_{\HK}^2\le\mathbb{E}\|\varepsilon_0 k(X,\cdot)\|_{\HK}^2\le\bar k\,\mathbb{E}[\varepsilon_0^2]<\infty$.

For the operator convergence, Lemma~\ref{lm:C_and_Cn} (Rosasco~2010,
Theorem~7) gives the Hilbert--Schmidt bound
$\|C - C_n\|_{\mathrm{HS}} \le 2\sqrt{2}\,\bar k\,\sqrt{\tau}/\sqrt{n}$
with probability $\ge 1 - 2e^{-\tau}$.  The operator norm is dominated by
the Hilbert--Schmidt norm ($\|\cdot\|_{\mathrm{op}} \le \|\cdot\|_{\mathrm{HS}}$),
so $\|C - C_n\|_{\mathrm{op}} = O_p(n^{-1/2}) = o_p(1)$.  Consequently the
second term is $o_p(1)$.

Combining both terms yields $\|\hat w - w_\lambda^*\|_{\HK} = o_p(1)$,
which completes the proof.

\subsubsection{Proof of Theorem~\ref{th:equivalence} (Equivalence of the Null Hypothesis)}
\label{proof:thm1}

\noindent\textbf{($\Rightarrow$).}\;
Under $\mathbb E(\varepsilon_0\mid X)=0$, iterated expectations give
$\mathbb E[\varepsilon_0\,g(X)]=0$ for every bounded measurable $g$; taking $g(\cdot)=k(x,\cdot)$ yields
$m^*(x)=\mathbb E[\varepsilon_0\,k(x,X)]=0$ for every $x$, so $m^*=\mathbf 0$ in $\HK$ and
$w^*_\lambda=(C+\lambda I)^{-1}m^*=\mathbf 0$.

\smallskip
\noindent\textbf{($\Leftarrow$).}\;
Suppose $w^*_\lambda=\mathbf 0$. Since $\lambda>0$, $(C+\lambda I)$ is invertible and
$m^*=(C+\lambda I)w^*_\lambda=\mathbf 0$, i.e.\ $\mathbb E[\varepsilon_0\,k(X,y)]=0$ for every
$y\in\mathcal X$. With $f(X):=\mathbb E(\varepsilon_0\mid X)$, fix $y\in\mathcal X$: the function $X\mapsto k(X,y)f(X)$ is $X$-measurable and integrable, $\mathbb E|k(X,y)f(X)|\le\sqrt{k(y,y)}\,(\mathbb E[f(X)^2])^{1/2}<\infty$ by Cauchy--Schwarz, so the tower property and $m^*=\mathbf 0$ give
\[
    g(y):=\int_{\mathcal X}k(x,y)f(x)\,d\mathbb P_X(x)
    =\mathbb E[k(X,y)\,f(X)]
    =\mathbb E[\varepsilon_0\,k(X,y)]
    =m^*(y)=0 ;
\]
hence $\iint_{\mathcal X\times\mathcal X} k(x,y)f(x)f(y)\,d\mathbb P_X(x)d\mathbb P_X(y)=\int_{\mathcal X}f(y)\,g(y)\,d\mathbb P_X(y)=0$. By integrally strict positive definiteness of the
kernel (Assumption~\ref{ass:B}(vi); \citealp{sriperumbudur2010hilbert}), the quadratic form is strictly
positive for every nonzero $f\in L^{2}(\mathbb P_X)$, so $f=0$, i.e.\
$\mathbb E(\varepsilon_0\mid X)=0$ almost surely.

\subsubsection{Proof of Theorem~\ref{th:proj_null_fix}}
\label{proof:thm2}

The projection test statistics can be expressed as RKHS quadratic forms of the
scaled KRR estimator $\sqrt{n}\hat{w}$ and the scaled mean embedding
$\sqrt{n}\hat{m}$:
\[
  n\hat{T}_{\mathrm{proj}}^{(1)} = \|\sqrt{n}\hat{w}\|_{\HK}^2,\qquad
  n\hat{T}_{\mathrm{proj}}^{(2)} = \langle\sqrt{n}\hat{w},\sqrt{n}\hat{m}\rangle_{\HK}.
\]

\smallskip
\noindent\textbf{Weak limit of the scaled estimators.}\;
Under $H_0$ (so that $\mathbb{E}[\varepsilon_0\mid X]=0$), the Hilbert-space
CLT (Theorem~\ref{th:lln_clt_rkhs}; the random elements
$\varepsilon_{0,i}k(X_i,\cdot)$ are i.i.d.\ by Assumption~\ref{ass:C}(i) and have
finite second moment by Assumptions~\ref{ass:B}(i) and~\ref{ass:B}(v), $\mathbb{E}\lVert\varepsilon_{0}k(X,\cdot)\rVert^2 = \mathbb{E}[\varepsilon_0^2 k(X,X)] \le \bar k\,\mathbb{E}[\varepsilon_0^2] < \infty$) gives
$\sqrt{n}\hat{m} \xrightarrow{d} \xi \sim \mathcal{N}(0,\Sigma)$ in $\HK$,
where $\Sigma = \mathbb{E}[(\varepsilon_0 k(X,\cdot))\otimes
(\varepsilon_0 k(X,\cdot))]$ is the covariance operator.  In particular
$\|\sqrt{n}\hat{m}\|_{\HK} = O_p(1)$ (tightness of the limit).

The KRR estimator satisfies $\hat{w} = (C_n + \lambda I)^{-1}\hat{m}$
(see the proof of Lemma~\ref{lemma:estimation_consistency}).  By
Lemma~\ref{lm:C_and_Cn}, $\|C_n - C\|_{\mathrm{HS}} = O_p(n^{-1/2})$, hence
$\|C_n - C\|_{\mathrm{op}} = o_p(1)$; since
$\lambda > 0$, the resolvent identity
\[
  (C_n + \lambda I)^{-1} - (C + \lambda I)^{-1}
  = (C_n + \lambda I)^{-1}(C - C_n)(C + \lambda I)^{-1}
\]
yields $\|(C_n + \lambda I)^{-1} - (C + \lambda I)^{-1}\|_{\mathrm{op}}
\le \lambda^{-2}\,\|C_n - C\|_{\mathrm{op}} = o_p(1)$, so
$(C_n + \lambda I)^{-1} \xrightarrow{p} (C + \lambda I)^{-1}$ in operator
norm.  The mapping $(A,x)\mapsto Ax$ is continuous from
$\mathcal{B}(\HK)\times\HK$ to $\HK$, so Slutsky's theorem in Hilbert spaces
(the continuous mapping theorem for joint weak convergence) yields
\[
  \sqrt{n}\hat{w}
  = (C_n + \lambda I)^{-1}(\sqrt{n}\hat{m})
  \xrightarrow{d} (C + \lambda I)^{-1}\xi .
\]

\smallskip
\noindent\textbf{Convergence of $n\hat{T}_{\mathrm{proj}}^{(1)}$.}\;
The squared norm is a continuous functional on $\HK$, hence by the
continuous mapping theorem,
\[
  n\hat{T}_{\mathrm{proj}}^{(1)}
  \xrightarrow{d} \bigl\|(C + \lambda I)^{-1}\xi\bigr\|_{\HK}^2 .
\]
Let $\{\mu_i\}_{i\ge 1}$ be the eigenvalues of $C$ and $\{e_i\}_{i\ge 1}$
the corresponding orthonormal eigenfunctions, with $e_i = \sqrt{\mu_i}\phi_i$
where $\{\phi_i\}_{i\ge 1}$ are the orthonormal eigenfunctions of
$\mathbf{L}_k$ in $L^2(\mathbb{P}_X)$ (see \S\ref{sec:equivalence}).
Expanding the limit in this basis and using the orthonormality of
$\{e_i\}$,
\[
  \bigl\|(C + \lambda I)^{-1}\xi\bigr\|_{\HK}^2
  = \sum_{i=1}^\infty \frac{1}{(\mu_i+\lambda)^2}\,
    \langle \xi, e_i\rangle_{\HK}^2 .
\]
The variables $Z_i := \langle \xi, e_i\rangle_{\HK}/\sqrt{\mu_i}$ are jointly Gaussian with
$\operatorname{Cov}(Z_i,Z_j)
= \mathbb{E}[\varepsilon_0^2\,\phi_i(X)\,\phi_j(X)]$; indeed, for any
$f,g \in \HK$ the reproducing property gives
$\langle \Sigma f, g\rangle_{\HK}
= \mathbb{E}[\varepsilon_0^2\, f(X)\, g(X)]$.  Under conditional
homoscedasticity ($\mathbb{E}[v^2\mid X]=\sigma^2$ a.s.) this covariance
reduces to $\sigma^2\delta_{ij}$: under $H_0$, $\varepsilon_0=v$ a.s., so,
$\phi_i\phi_j$ being $X$-measurable,
$\Sigma_{ij}=\mathbb{E}[\phi_i(X)\phi_j(X)\,\mathbb{E}[v^2\mid X]]
=\sigma^2\,\mathbb{E}[\phi_i(X)\phi_j(X)]=\sigma^2\delta_{ij}$, and a
Gaussian sequence with covariance $\sigma^2 I$ has independent components,
which proves the reduction to independent $N(0,\sigma^2)$ variables stated
in Theorem~\ref{th:proj_null_fix}.  In particular
$\operatorname{Var}(Z_i) = \tau_i := \operatorname{Var}(\varepsilon_0\,
\phi_i(X))$, where the equality
$\mathbb{E}[\varepsilon_0^2\,\phi_i(X)^2]
= \operatorname{Var}(\varepsilon_0\,\phi_i(X))$ holds because
$\mathbb{E}[\varepsilon_0\,\phi_i(X)] = 0$ under $H_0$.  The infinite
series below are well defined almost surely: by Tonelli's theorem (the
summands are nonnegative, so no integrability precondition is needed), the
Mercer expansion and
Assumption~\ref{ass:B}(v),
$\mathbb{E}\bigl[\sum_{i=1}^\infty \frac{\mu_i}{(\mu_i+\lambda)^2} Z_i^2\bigr]
= \sum_{i=1}^\infty\frac{\mu_i}{(\mu_i+\lambda)^2}\,\mathbb{E}[Z_i^2]
\le \lambda^{-2}\sum_{i=1}^\infty\mu_i\,\mathbb{E}[\varepsilon_0^2\phi_i(X)^2]
= \lambda^{-2}\,\mathbb{E}[\varepsilon_0^2 k(X,X)] \le
\lambda^{-2}\bar{k}\,\mathbb{E}[\varepsilon_0^2] < \infty$ (and similarly
for the $\mu_i/(\mu_i+\lambda)$ weights), so the series converge almost
surely by monotone convergence.  Consequently
\[
  n\hat{T}_{\mathrm{proj}}^{(1)}
  \xrightarrow{d}
  \sum_{i=1}^\infty \frac{\mu_i}{(\mu_i+\lambda)^2}\,Z_i^2 .
\]

\smallskip
\noindent\textbf{Convergence of $n\hat{T}_{\mathrm{proj}}^{(2)}$.}\;
The inner product is a continuous functional of the pair
$(\sqrt{n}\hat{w},\sqrt{n}\hat{m})$, so the continuous mapping theorem
gives
\[
  n\hat{T}_{\mathrm{proj}}^{(2)}
  \xrightarrow{d}
  \langle (C + \lambda I)^{-1}\xi, \xi\rangle_{\HK}
  = \sum_{i=1}^\infty \frac{1}{\mu_i+\lambda}\,
    \langle \xi, e_i\rangle_{\HK}^2
  = \sum_{i=1}^\infty \frac{\mu_i}{\mu_i+\lambda}\,Z_i^2 ,
\]
with $Z = (Z_1,Z_2,\dots) \sim \mathcal{N}(0,\Sigma)$ and
$\Sigma_{ij} = \mathbb{E}[\varepsilon_0^2\,\phi_i(X)\,\phi_j(X)]$, as
claimed.

\smallskip
\noindent\textbf{Consistency under fixed alternatives.}\;
Under a fixed alternative, $\mathbb{E}[\varepsilon_0\mid X] \neq 0$, so
$m^* = \mathbb{E}[\varepsilon_0 k(X,\cdot)] \neq \mathbf{0}$ in $\HK$: by
Assumption~\ref{ass:B}(vi) the kernel is ISPD, and the double-integral argument of
Theorem~\ref{th:equivalence} shows that $m^* = 0$ would force
$\mathbb{E}[\varepsilon_0\mid X]=0$.
The law of large numbers in $\HK$ gives $\hat{m} \xrightarrow{p} m^* \neq
\mathbf{0}$ (the a.s.\ form of the SLLN holds as well, but only the in-probability mode is consumed),
so $\|\sqrt{n}\hat{m}\|_{\HK} \to \infty$ at rate $\sqrt{n}$.
The resolvent $(C + \lambda I)^{-1}$ is injective, so the population fit
$w^*_\lambda = (C + \lambda I)^{-1}m^*$ is a nonzero element of $\HK$
($m^* \neq \mathbf{0}$ as established above).  The KRR estimator satisfies
$\sqrt{n}\hat{w} = (C_n + \lambda I)^{-1}(\sqrt{n}\hat{m})$, and since
$C_n \to C$ in operator norm, the resolvents also converge: with
$\Delta_n := (C_n + \lambda I)^{-1} - (C + \lambda I)^{-1}$,
$\|\Delta_n\|_{\mathrm{op}} \le \lambda^{-2}\|C_n - C\|_{\mathrm{op}} = o_p(1)$
(resolvent identity; Lemma~\ref{lm:C_and_Cn}). Decomposing
\[
  \sqrt{n}\hat{w}
  = \sqrt{n}\,w^*_\lambda
    + \sqrt{n}\,(C + \lambda I)^{-1}(\hat m - m^*)
    + \Delta_n(\sqrt{n}\hat{m}),
\]
the two remainders are $o_p(\sqrt n)$: the first by
$\|\sqrt{n}(C+\lambda I)^{-1}(\hat m - m^*)\|_{\HK}
\le \sqrt{n}\,\lambda^{-1}\|\hat m - m^*\|_{\HK} = o_p(\sqrt n)$,
the second by
$\|\Delta_n(\sqrt{n}\hat{m})\|_{\HK} \le o_p(1)\,O_p(\sqrt n) = o_p(\sqrt n)$,
where $\|\sqrt n\hat m\|_{\HK} = \sqrt n\,\|m^*\|_{\HK} + O_p(1) = O_p(\sqrt n)$ under the fixed alternative. The reverse triangle inequality then gives
$\|\sqrt{n}\hat{w}\|_{\HK} = \sqrt{n}\,\|w^*_\lambda\|_{\HK} + o_p(\sqrt n)
= \sqrt{n}\,\|w^*_\lambda\|_{\HK}(1 + o_p(1))$,
so $\|\sqrt{n}\hat{w}\|_{\HK} \to \infty$ at rate $\sqrt{n}$.  For
$n\hat{T}_{\mathrm{proj}}^{(2)}$ the divergence additionally needs the
limiting inner product to be strictly positive, which holds since
$\langle w^*_\lambda, m^*\rangle_{\HK}
= \langle w^*_\lambda, (C+\lambda I)w^*_\lambda\rangle_{\HK}
= \langle w^*_\lambda, Cw^*_\lambda\rangle_{\HK} + \lambda\|w^*_\lambda\|_{\HK}^2
\ge \lambda\|w^*_\lambda\|_{\HK}^2 > 0$,
the inequality being $\langle w^*_\lambda, Cw^*_\lambda\rangle_{\HK} \ge 0$ by positive semidefiniteness of $C$. Moreover,
$n\hat T^{(2)}_{\mathrm{proj}} = n\langle w^*_\lambda, m^*\rangle_{\HK} + o_p(n)$:
expanding the inner product around $(w^*_\lambda, m^*)$,
\[
  \langle\hat w,\hat m\rangle_{\HK}
  = \langle w^*_\lambda, m^*\rangle_{\HK}
    + \langle\hat w - w^*_\lambda, m^*\rangle_{\HK}
    + \langle w^*_\lambda, \hat m - m^*\rangle_{\HK}
    + \langle\hat w - w^*_\lambda, \hat m - m^*\rangle_{\HK}
  = \langle w^*_\lambda, m^*\rangle_{\HK} + o_p(1),
\]
by Cauchy--Schwarz and $\|\hat w - w^*_\lambda\|_{\HK} = o_p(1)$ (Lemma~\ref{lemma:estimation_consistency}, whose proof uses no null-specific input) together with $\|\hat m - m^*\|_{\HK} = o_p(1)$ (the Hilbert-space law of large numbers). Hence both quadratic forms
$n\hat{T}_{\mathrm{proj}}^{(1)} = \|\sqrt{n}\hat{w}\|_{\HK}^2$ and
$n\hat{T}_{\mathrm{proj}}^{(2)} = n\bigl(\langle w^*_\lambda, m^*\rangle_{\HK} + o_p(1)\bigr)$ diverge to $+\infty$ in probability,
so $\mathbb{P}(n\hat{T}_{\mathrm{proj}}^{(k)} > t) \to 1$ for any fixed
$t > 0$.
\subsubsection{Proof of Theorem~\ref{th:proj_local}}
\label{proof:thm3}

Under the local alternative $H_{1,n}$, the residuals decompose as
$\varepsilon_{0,l} = \varepsilon_{0,l}^* + n^{-1/2}R(X_l)$, where
$\mathbb{E}[\varepsilon_0^*\mid X] = 0$ and $R \in L^2(\mathbb{P}_X)$ is
non-zero.  The scaled mean embedding then splits into a random fluctuation
and a deterministic shift:
\[
  \sqrt{n}\,\hat{m}
  = \underbrace{\frac{1}{\sqrt{n}}\sum_{l=1}^n
      \varepsilon_{0,l}^*\,k(X_l,\cdot)}_{\xrightarrow{d}\; \xi}
    \;+\;
    \underbrace{\frac{1}{n}\sum_{l=1}^n
      R(X_l)\,k(X_l,\cdot)}_{\xrightarrow{a.s.}\; m_R},
\]
where $\xi \sim \mathcal{N}(0,\Sigma)$ by the Hilbert-space CLT
(Theorem~\ref{th:lln_clt_rkhs}; the i.i.d.\ elements
$\varepsilon_{0,l}^*k(X_l,\cdot)$ have finite second moment by
Assumption~\ref{ass:B}(v)) and $m_R := \mathbb{E}[R(X)k(X,\cdot)]
\in \HK$ (membership by Bochner integrability:
$\mathbb{E}\|R(X)k(X,\cdot)\|_{\HK} \le \sqrt{\bar{k}}\,\mathbb{E}|R(X)|
\le \sqrt{\bar{k}}\,\|R\|_{L^2} < \infty$), with the displayed
$\xrightarrow{a.s.}$ convergence given by the Hilbert-space SLLN
(same theorem).  Hence
$\sqrt{n}\hat{m} \xrightarrow{d} \xi + m_R$ in $\HK$, and the sequence is
tight ($\|\sqrt{n}\hat{m}\|_{\HK} = O_p(1)$).

\smallskip
\noindent\textbf{Weak limit of the scaled estimators.}\;
Exactly as in the proof of Theorem~\ref{th:proj_null_fix}, the resolvent
convergence $(C_n + \lambda I)^{-1} \xrightarrow{p} (C + \lambda I)^{-1}$
in operator norm combined with Slutsky's theorem in Hilbert spaces yields
\[
  \sqrt{n}\hat{w}
  = (C_n + \lambda I)^{-1}(\sqrt{n}\hat{m})
  \xrightarrow{d} (C + \lambda I)^{-1}(\xi + m_R) .
\]

\smallskip
\noindent\textbf{Convergence of the test statistics.}\;
The test statistics are continuous functionals of
$(\sqrt{n}\hat{w}, \sqrt{n}\hat{m})$, so the continuous mapping theorem
gives, with $e_i = \sqrt{\mu_i}\,\phi_i$ as in the proof of
Theorem~\ref{th:proj_null_fix} and
$\langle m_R, e_i\rangle_{\HK}
 = \mathbb{E}[R(X)\,e_i(X)]
 = \sqrt{\mu_i}\,\mathbb{E}[R(X)\,\phi_i(X)]$,
\begin{align}
n\hat{T}_{\mathrm{proj}}^{(1)}
  &\xrightarrow{d}
    \bigl\|(C + \lambda I)^{-1}(\xi + m_R)\bigr\|_{\HK}^2 \\[4pt]
  &= \sum_{i=1}^\infty \frac{1}{(\mu_i+\lambda)^2}\,
     \bigl(\langle \xi, e_i\rangle_{\HK}
           + \langle m_R, e_i\rangle_{\HK}\bigr)^2 \\[4pt]
  &= \sum_{i=1}^\infty \frac{\mu_i}{(\mu_i+\lambda)^2}\,
     \bigl(Z_i + \mathbb{E}[R(X)\,\phi_i(X)]\bigr)^2,
\end{align}
and, similarly,
\begin{align*}
  n\hat{T}_{\mathrm{proj}}^{(2)}
  &\xrightarrow{d}
    \langle (C + \lambda I)^{-1}(\xi + m_R),\,\xi + m_R\rangle_{\HK} \\[4pt]
  &= \sum_{i=1}^\infty \frac{1}{\mu_i+\lambda}\,
     \bigl(\langle \xi, e_i\rangle_{\HK}
           + \langle m_R, e_i\rangle_{\HK}\bigr)^2 \\[4pt]
  &= \sum_{i=1}^\infty \frac{\mu_i}{\mu_i+\lambda}\,
     \bigl(Z_i + \mathbb{E}[R(X)\,\phi_i(X)]\bigr)^2,
\end{align*}
where $\{Z_i\}_{i\ge 1}$ are the jointly Gaussian variables of
Theorem~\ref{th:proj_null_fix}. \emph{Almost-sure summability.} The passage from the finite-$J$
partial sums to the infinite series requires that the limiting series converge almost surely, not
merely in distribution term-by-term; this follows by the Tonelli device used in the proof of
Proposition~\ref{th:proj_null_fix} (the unshifted version of the same series): set $S_i:=\mu_i(\mu_i+\lambda)^{-2}(Z_i+\mathbb
E[R(X)\phi_i(X)])^{2}$ and $d_i:=\mathbb E[R(X)\phi_i(X)]$. By Tonelli,
$\mathbb E[\sum_{i\ge1}S_i]=\sum_{i\ge1}\mathbb E[S_i]$, and
\[
    \sum_{i\ge1}\mathbb E[S_i]
    \;\le\;2\sum_{i\ge1}\mu_i(\mu_i+\lambda)^{-2}\,
      \bigl(\mathbb E Z_i^2+d_i^2\bigr).
\]
Both sums are finite by direct computation. First, $Z_i$ is the coordinate pairing of the noise
element against $\phi_i$, so $\mu_i Z_i^2=\varepsilon_0^2\,\langle k(X,\cdot),\phi_i\rangle^2\,\mu_i$
and the sum combines through the kernel diagonal:
\[
    \sum_{i\ge1}\mu_i\,\mathbb E Z_i^2
    \;=\;\mathbb E\Bigl[\varepsilon_0^2\textstyle\sum_{i\ge1}\mu_i\phi_i(X)^2\Bigr]
    \;\le\;\bar k\,\mathbb E[\varepsilon_0^2],
\]
since $\sum_i\mu_i\phi_i(x)^2\le k_\gamma(x,x)\le\bar k$. Second, Bessel's inequality in
$L^2(\mathbb P_X)$ applied to $R$ over the orthonormal system $\{\phi_i\}$ gives
$\sum_i d_i^2\le\lVert R\rVert_{L^2}^2<\infty$. Collecting,
$\mathbb E[\sum_{i\ge1}S_i]\le
2\,\sup_t(t+\lambda)^{-2}\,\bigl(\bar k\,\mathbb E[\varepsilon_0^2]+\lVert R\rVert_{L^2}^2\bigr)<\infty$
(with $(\mu_i+\lambda)^{-2}\le\lambda^{-2}$). Finite expectation implies summability almost surely,
so the series converges a.s.\ and the displayed limits hold as random series. This completes the proof.

\subsection{Proofs for Section~\ref{sec:estimation_effect}: Estimation Effects}\label{app:proofs_estimation}

\noindent These proofs support Section~\ref{sec:estimation_effect}, again with the supporting lemma first.
Lemma~\ref{lemma:estimation_remainder} quantifies the RKHS remainder produced by Neyman orthogonalization when
$\hat\theta$ replaces $\theta_0$ and is quoted by both corollaries that follow. Corollaries~\ref{cor:proj_null}
and~\ref{cor:proj_local} transfer the null and local-alternative laws from the infeasible statistic to its projected
counterpart, and Theorem~\ref{th:bootstrap_consistency} establishes conditional consistency of the multiplier bootstrap
for the projected statistic.

\subsubsection{Proof of Lemma~\ref{lemma:estimation_remainder}}
\label{proof:estimation_remainder}
Throughout the proof, the inverses $\bigl(\hat{\mathbb G}^\top\hat{\mathbb G}\bigr)^{-1}$ and
$\bigl(Z^\top Z\bigr)^{-1}$ are understood on the invertibility events
$\lambda_{\min}(\hat{\mathbb G}^\top\hat{\mathbb G}/n)>0$ resp.
$\lambda_{\min}(Z^\top Z/n)>0$, which hold with probability tending to one by
Assumption~\ref{ass:E}(ii) and Assumption~\ref{ass:F}(ii) respectively. For each observation $i$, we expand the estimated residual $\varepsilon(s_i; \hat\theta)$ around the true parameter $\theta_0$ using a second-order Taylor expansion. To introduce the estimated score matrix $\mathbb{G}$ (which forms the projection $\hat{\boldsymbol{\Pi}}$), we add and subtract the estimated gradient $\hat{g}_i^\top = \nabla_{\theta}\varepsilon(s_i; \hat\theta)^\top$ from the first-order term:
\begin{align*}
    \varepsilon(s_i; \hat\theta) &= \varepsilon(s_i; \theta_0) + \hat{g}_i^\top (\hat\theta - \theta_0) \\
    &\quad + \underbrace{\left[ \nabla_{\theta}\varepsilon(s_i; \theta_0) - \hat{g}_i \right]^\top (\hat\theta - \theta_0)}_{\text{Term A}} \\
    &\quad + \underbrace{\frac{1}{2} (\hat\theta - \theta_0)^\top \nabla_{\theta}^2 \varepsilon(s_i; \tilde\theta_i) (\hat\theta - \theta_0)}_{\text{Term B}},
\end{align*}
where $\tilde\theta_i$ lies on the line segment between $\hat\theta$ and $\theta_0$. Let $v_i = \text{Term A} + \text{Term B}$ denote the remainder for observation $i$.

We first determine the order of Term A. Applying the integral form of the vector Mean Value Theorem to the gradient difference yields:
\begin{align*}
    \nabla_{\theta}\varepsilon(s_i; \theta_0) - \hat{g}_i &= \nabla_{\theta}\varepsilon(s_i; \theta_0) - \nabla_{\theta}\varepsilon(s_i; \hat\theta) \\
    &= - \left( \int_0^1 \nabla_{\theta}^2 \varepsilon(s_i; \theta_0 + t(\hat\theta - \theta_0)) dt \right) (\hat\theta - \theta_0).
\end{align*}
Under Assumption~\ref{ass:E}(iii), the raw Hessian matrix $\nabla_{\theta}^2 \varepsilon(s_i; \theta)$ is uniformly bounded over the compact parameter space $\Theta$ in the sense that $\sup_{\theta \in \Theta}\lVert \nabla_{\theta}^2 \varepsilon(s_i; \theta) \rVert = O_p(1)$ pointwise in $i$. Assumption~\ref{ass:E}(iii) constrains the raw, possibly non-$X$-measurable Hessian, which is needed because $\nabla_{\theta}\varepsilon$ need not be $X$-measurable, cf.\ Assumption~\ref{ass:F}.  Since $\hat\theta, \theta_0 \in \Theta$, the integral of the Hessian along the connecting path is also $O_p(1)$. Using Assumption~\ref{ass:D}(iii), $\lVert \hat\theta - \theta_0 \rVert = O_p(n^{-\varrho})$, we obtain:
\begin{align*}
    \left\| \nabla_{\theta}\varepsilon(s_i; \theta_0) - \hat{g}_i \right\| \le O_p(1) \cdot \lVert \hat\theta - \theta_0 \rVert = O_p(n^{-\varrho}).
\end{align*}
Consequently, Term A is of order:
\[
    \text{Term A} = \left[ \nabla_{\theta}\varepsilon(s_i; \theta_0) - \hat{g}_i \right]^\top (\hat\theta - \theta_0) = O_p(n^{-\varrho}) \cdot O_p(n^{-\varrho}) = O_p(n^{-2\varrho}).
\]

For Term B, the Hessian is $O_p(1)$ by Assumption~\ref{ass:E}(iii), so:
\[
    \text{Term B} = O_p(n^{-\varrho}) \cdot O_p(1) \cdot O_p(n^{-\varrho}) = O_p(n^{-2\varrho}).
\]

Combining these, the element-wise remainder is $r_i^{(2)} = O_p(n^{-2\varrho})$. Stacking the expansions for all $n$ observations into vector notation, and letting $\boldsymbol r^{(2)} = (r_1^{(2)}, \dots, r_n^{(2)})^\top$, yields:
\begin{align*}
    \boldsymbol{\varepsilon}_{\hat \theta} = \boldsymbol{\varepsilon}_0 + \mathbb{G}(\hat\theta - \theta_0) + \boldsymbol r^{(2)},
\end{align*}
where $\mathbb{G}$ is the matrix of estimated scores with rows $\hat g_i^\top = \nabla_\theta\varepsilon(s_i;\hat\theta)^\top$ of the expansion above, so that the identity is exact with $\boldsymbol r^{(2)}$ collecting Term A and Term B, and $\boldsymbol r^{(2)} = O_p(n^{-2\varrho})$ element-wise.

Applying the projection matrix $\hat{\boldsymbol{\Pi}} = \boldsymbol{I}_n - \hat{\mathbb{G}}(\hat{\mathbb{G}}^\top \hat{\mathbb{G}})^{-1} \hat{\mathbb{G}}^\top$, where $\hat{\mathbb{G}} = \boldsymbol{P}_Z \mathbb{G}$ is the sieve estimate of $g_{\theta_0}(X)$, to both sides, and decomposing the raw score as $\mathbb{G} = \hat{\mathbb{G}} + (\mathbb{G} - \hat{\mathbb{G}})$, yields:
\begin{align*}
    \hat{\boldsymbol{\Pi}} \boldsymbol{\varepsilon}_{\hat \theta} = \hat{\boldsymbol{\Pi}} \boldsymbol{\varepsilon}_0 + \hat{\boldsymbol{\Pi}}(\mathbb{G} - \hat{\mathbb{G}})(\hat\theta - \theta_0) + \hat{\boldsymbol{\Pi}} \boldsymbol r^{(2)},
\end{align*}
using the finite-sample orthogonality $\hat{\boldsymbol{\Pi}} \hat{\mathbb{G}} = \boldsymbol{0}$.
Set $\boldsymbol r = \hat{\boldsymbol{\Pi}}(\mathbb{G} - \hat{\mathbb{G}})(\hat\theta - \theta_0) + \hat{\boldsymbol{\Pi}}\boldsymbol r^{(2)}$.
Unlike the $X$-measurable-score case, the first-order term is not annihilated pointwise by the projection: the projection annihilates only the estimated score $\hat{\mathbb{G}}$, and the idiosyncratic score $\mathbb{G} - \hat{\mathbb{G}}$ is \emph{not} exactly orthogonal to the sieve columns in the finite sample.
Elimination holds only in expectation: $\mathbb{G} - \hat{\mathbb{G}} \xrightarrow{p} \nabla_\theta \varepsilon - g_{\theta_0}(X)$, the idiosyncratic score, and $\mathbb{E}[(\nabla_\theta \varepsilon - g_{\theta_0}(X))\,\phi(X)] = 0$ for every square-integrable $\phi$, because the idiosyncratic score is conditionally mean-zero given $X$ (law of iterated expectations).
Both parts of $\boldsymbol r$ are bounded by the $\ell^2$ (spectral) route, which is exactly what the RKHS norms require.
With $\boldsymbol{K} = \sum_i \sigma_i^2 \boldsymbol{u}_i \boldsymbol{u}_i^\top$, the two remainder norms read
\begin{align*}
    \|R_w\|_{\HK}^2 = \sum_i \frac{n\,\sigma_i^2}{(\sigma_i^2 + n\lambda)^2}\,(\boldsymbol u_i^\top \boldsymbol r)^2, \qquad \|R_m\|_{\HK}^2 = \frac{1}{n}\sum_i \sigma_i^2 (\boldsymbol{u}_i^\top \boldsymbol r)^2,
\end{align*}
and each is a weighted quadratic form in the coefficient vector $\boldsymbol r$: with $w_i := \nu_i/(\nu_i+\lambda)^2 \le \lambda^{-2}\nu_i$ (so that $n\,\sigma_i^2/(\sigma_i^2+n\lambda)^2 = \nu_i/(\nu_i+\lambda)^2 = w_i$; here $\lambda$ is the single fixed regularization constant of this lemma, and no uniformity over a $\lambda$-range is needed),
\[
    \sum_{i=1}^n w_i\,(\boldsymbol u_i^\top \boldsymbol r)^2
    = \boldsymbol r^\top \boldsymbol U\operatorname{diag}(w_i)\boldsymbol U^\top \boldsymbol r
    \le \lambda^{-2}\,\boldsymbol r^\top(\boldsymbol K/n)\boldsymbol r,
\]
using $\lVert \boldsymbol K/n\rVert_{\mathrm{op}} \le \operatorname{tr}(\boldsymbol K/n) = 1$ (the Gaussian kernel has $k_\gamma(x,x) = 1$); the same bound without the weights $w_i$ covers the $R_m$ display.
For the first-order part, write $\mathbb{G} - \hat{\mathbb{G}} = (\boldsymbol I-P_Z)\boldsymbol\eta_0 + (\boldsymbol I-P_Z)\boldsymbol g + (\boldsymbol I-P_Z)\tilde{\boldsymbol H}\,\boldsymbol\delta$, with $\boldsymbol\delta := \hat\theta-\theta_0$, $\boldsymbol\eta_0$ the conditionally centered idiosyncratic score matrix (rows $\eta_{0,j}^\top = \nabla_\theta\varepsilon(s_j;\theta_0)^\top - g_{\theta_0}(X_j)^\top$), $\boldsymbol g$ the matrix of population scores, and $\tilde{\boldsymbol H}$ the matrix of path-averaged Hessians (rows $\int_0^1\nabla_\theta^2\varepsilon(s_j;\theta_0+t\boldsymbol\delta)\,\mathrm{d}t$).
The third piece is $O_p(n^{-\varrho})$ element-wise (Assumption~\ref{ass:E}(iii), pointwise in $j$) and, multiplied by $\boldsymbol\delta$, is folded into the second-order part $\boldsymbol r^{(2)}$ below (redefined to include it).
The sieve-error mean $\boldsymbol r_{\mathrm{sieve}} := \hat{\boldsymbol{\Pi}}(\boldsymbol I-P_Z)\boldsymbol g\,\boldsymbol\delta$ is controlled pathwise: since $(\boldsymbol I-P_Z)\boldsymbol g = (\boldsymbol I-P_Z)(\boldsymbol g - \hat{\boldsymbol g})$ (the fitted sieve values $\hat{\boldsymbol g} = P_Z\mathbb{G}$ are annihilated by $\boldsymbol I-P_Z$),
\[
    (\boldsymbol r_{\mathrm{sieve}})^\top(\boldsymbol K/n)\boldsymbol r_{\mathrm{sieve}}
    \le \lVert\boldsymbol\delta\rVert^2\,\lVert(\boldsymbol I-P_Z)\boldsymbol g\rVert_F^2
    \le \lVert\boldsymbol\delta\rVert^2\,\lVert\hat{\boldsymbol g}-g_{\theta_0}\rVert_F^2
    = O_p(n^{-2\varrho})\,o_p(n^{2\varrho}) = o_p(1),
\]
where $\lVert\hat{\boldsymbol g}-g_{\theta_0}\rVert_F^2 = n\,\lVert\hat g_{\hat\theta}-g_{\theta_0}\rVert_{L^2(\mathbb P_n)}^2 = o_p(n^{2\varrho})$ by Pythagoras, the bridge below, and the sieve estimation-error rate $o_p(n^{\varrho-1/2})$ of Assumption~\ref{ass:F}(iii); the sieve-mean piece is itself $o_p(1)$.

\emph{Bridge between the population and empirical norms.} Assumption~\ref{ass:F}(iii) controls
$\lVert\hat g_{\hat\theta}-g_{\theta_0}\rVert_{L^2(\mathbb P_X)}$, while the display consumes the
empirical norm $\lVert\cdot\rVert_{L^2(\mathbb P_n)}$; the two agree up to a negligible relative
error under the bridge condition $K_n=o(\sqrt n)$, where $K_n$ is the sieve dimension of Assumption~\ref{ass:F}(iii) (the standing dimension restriction $L^{2\varsigma+1}\log n=o(n)$ gives $L\log n=o(\sqrt n)$, hence $L=o(\sqrt n)$), which covers every $\varrho>0$. The
displayed rate is driven by the maximal term $O_p(B^2K_n r^2/\sqrt n)$; weaker conditions such as
$K_n\log n=o(n^{2\varrho-1/2})$ suffice only for $\varrho\le1/2$, since for $\varrho>1/2$ they
permit $K_n\gg\sqrt n$, under which the relative error need not vanish.
Writing the error in the sieve basis,
$\hat g_{\hat\theta}-g_{\theta_0}=\sum_{\ell\le K_n}(\hat c_\ell-c_\ell^*)\psi_\ell$, the difference
of squared norms is the single centered quantity $(P_n-P_X)[(\hat g_{\hat\theta}-g_{\theta_0})^{2}]$.
The class $\{h^{2}:\ h=\sum_{\ell\le K_n}b_\ell\psi_\ell,\ \lVert h\rVert_{L^{2}(\mathbb P_X)}\le r\}$
consists of functions bounded by $B^{2}r^{2}$ (basis $\|\psi_\ell\|_\infty\le B$) with VC dimension
$O(K_n)$, so the standard empirical-process maximal inequality gives
\[
    \sup_{\lVert h\rVert\le r}\bigl|(P_n-P_X)[h^{2}]\bigr|
    =O_p\!\left(\frac{B^{2}K_n\,r^{2}}{\sqrt n}\right),
\]
The displayed order applies the VC-type estimate $O_p(\sqrt{\mathrm{VC}/n}\,\|\cdot\|_{P,2})$, with generalized norm $\|h^{2}\|_{P,2}\le B^{2}r^{2}$ up to constants and $\mathrm{VC}=O(K_n)$; the VC form gives the weaker $O_p(B^{2}K_n^{1/2}r^{2}/\sqrt n)$, and both forms are $o_p(r^{2})$ under $K_n=o(\sqrt n)$. Under the side condition
$K_n=o(\sqrt n)$ this error is $o_p(r^{2})$, whence
$\lVert\cdot\rVert_{L^2(\mathbb P_n)}=\lVert\cdot\rVert_{L^2(\mathbb P_X)}
(1+o_p(1))$ on the relevant ball, and the empirical-norm consumption is justified. The same bridge is
invoked wherever an $L^2(\mathbb P_n)$ rate is consumed below.

For the centered part $\boldsymbol r_{\mathrm{center}} := \hat{\boldsymbol{\Pi}}(\boldsymbol I-P_Z)\boldsymbol\eta_0\,\boldsymbol\delta$, no higher moments of $\boldsymbol\delta$ are needed. Let $M := (\boldsymbol I-P_Z)\hat{\boldsymbol{\Pi}}(\boldsymbol K/n)\hat{\boldsymbol{\Pi}}(\boldsymbol I-P_Z)$.
The $\mathbf X$-measurability of $M$ is not immediate from the display, since $\hat{\boldsymbol\Pi}$
depends on $\boldsymbol\varepsilon_0$ through $\hat{\mathbb G}$; it follows from the identity
\[
    \hat{\boldsymbol\Pi}(\boldsymbol I-P_Z)=(\boldsymbol I-P_Z),
\]
which holds because $(\boldsymbol I-P_Z)v$ is orthogonal to
$\operatorname{range}(P_Z)\supseteq\operatorname{range}(\hat{\mathbb G})$ (the sieve dimension may exceed $d_\theta$), and
finite-sample orthogonality gives $\hat{\boldsymbol\Pi}\hat{\mathbb G}=\boldsymbol 0$, hence
$\hat{\boldsymbol\Pi}(\boldsymbol I-P_Z)v=(\boldsymbol I-P_Z)v$; and, since the projector
$\hat{\boldsymbol\Pi}$ is symmetric idempotent, transposing also gives
$(\boldsymbol I-P_Z)\hat{\boldsymbol\Pi}=(\boldsymbol I-P_Z)$. Substituting both identities,
\[
    M=(\boldsymbol I-P_Z)(\boldsymbol K/n)(\boldsymbol I-P_Z),
\]
a polynomial in the $\mathbf X$-measurable Gram matrix $\boldsymbol K$: $M$ is $\mathbf X$-measurable
as asserted. Write $A := \boldsymbol\eta_0^\top M\boldsymbol\eta_0$; with $\boldsymbol Q_Z := \boldsymbol I - P_Z$ and $M = \boldsymbol Q_Z(\boldsymbol K/n)\boldsymbol Q_Z$ (shown above), $A = W^\top(\boldsymbol K/n)W$ for $W := \boldsymbol Q_Z\boldsymbol\eta_0$. The rows of $W$ are conditionally centered given $\mathbf X$ but not conditionally independent, the $\mathbf X$-measurable $\boldsymbol Q_Z$ mixing the $\eta_0$-rows; the fourth-moment calculation therefore pairs indices through the conditional independence of the $\eta_0$-rows themselves, and each surviving four-index term is $O(1)$ by $\mathbb{E}\lVert\boldsymbol\eta_0\rVert^4 < \infty$ (Assumption~\ref{ass:E}(iii)) together with $\operatorname{tr}\big[(\boldsymbol K/n)\boldsymbol Q_Z(\boldsymbol K/n)\boldsymbol Q_Z\big]\le\operatorname{tr}[(\boldsymbol K/n)^2]\le\bar k^2$; also $\big\lVert\mathbb{E}[A\mid\mathbf X]\big\rVert_F\le\bar\Sigma\,\operatorname{tr}[(\boldsymbol K/n)\boldsymbol Q_Z]\le\bar\Sigma$, with $\bar\Sigma := \sup_x\mathbb{E}[\lVert\boldsymbol\eta_{0,1}\rVert^2\mid X=x] < \infty$. Hence $\mathbb{E}\lVert A\rVert_F^2 = O(1)$, and $\lVert A\rVert_F = O_p(1)$ by Markov. Therefore
\[
    (\boldsymbol r_{\mathrm{center}})^\top(\boldsymbol K/n)\boldsymbol r_{\mathrm{center}}
    = \boldsymbol\delta^\top A\,\boldsymbol\delta
    \le \lVert\boldsymbol\delta\rVert^2\,\lVert A\rVert_F
    = O_p(n^{-2\varrho}),
\]
using only $\lVert\boldsymbol\delta\rVert = O_p(n^{-\varrho})$ from Assumption~\ref{ass:D}(iii). Thus the centered contribution to both $\|R_w\|_{\HK}^2$ and $\|R_m\|_{\HK}^2$ is $O_p(n^{-2\varrho})$ (the two displays above differ only by the bounded weights $w_i \le \lambda^{-2}\nu_i$, with
$\lambda$ the single fixed constant of this section).

For the second-order part $\boldsymbol r_{\mathrm{proj}} := \hat{\boldsymbol{\Pi}}\boldsymbol r^{(2)}$, with $\boldsymbol r^{(2)}$ extended to include the folded-in gradient-drift piece $\hat{\boldsymbol{\Pi}}(\boldsymbol I-P_Z)\boldsymbol w = (\boldsymbol I-P_Z)\boldsymbol w$ with rows $w_j = \boldsymbol\delta^\top\tilde{\boldsymbol H}_j^\top\,\boldsymbol\delta$ (element-wise $O_p(n^{-2\varrho})$, exactly as Term A and Term B; the identity $\hat{\boldsymbol{\Pi}}(\boldsymbol I-P_Z)=(\boldsymbol I-P_Z)$ established above applies), no centering is available: $r_j^{(2)}$ has a nonzero conditional mean of order $O_p(n^{-2\varrho})$ (the leading term is a quadratic form in $\boldsymbol\delta$ built from the conditional mean of the Hessian), and its cross terms do not vanish. The pathwise bound suffices: with $\lVert\boldsymbol K/n\rVert_{\mathrm{op}} \le 1$,
\[
    (\boldsymbol r_{\mathrm{proj}})^\top(\boldsymbol K/n)\boldsymbol r_{\mathrm{proj}} \le \lVert\boldsymbol r^{(2)}\rVert^2 \le C\,\lVert\boldsymbol\delta\rVert^4\sum_{j=1}^n\sup_{\theta\in\Theta}\lVert\nabla^2_\theta\varepsilon(s_j;\theta)\rVert^2 = O_p(n^{1-4\varrho}),
\]
The strengthened Assumption~\ref{ass:E}(iii) supplies
\[
    \mathbb{E}\big[\sup_\theta\lVert\nabla^2_\theta\varepsilon(S;\theta)\rVert^2\big] < \infty,
\]
so $\sum_j\sup_\theta\lVert\nabla^2_\theta\varepsilon(s_j;\theta)\rVert^2 = O_p(n)$, and $\lVert\boldsymbol\delta\rVert^4 = O_p(n^{-4\varrho})$ by Assumption~\ref{ass:D}(iii). The universal constant $C$ collects the quadratic and cross terms of Term A, Term B, and the gradient-drift piece. This is the same $O_p(n^{1-4\varrho})$ as in the bootstrap-remainder computation of the proof of Theorem~\ref{th:bootstrap_consistency} below. Therefore $\|R_w\|_{\HK}^2 = \|R_m\|_{\HK}^2 = o_p(1)$ whenever $\varrho > 1/4$: the centered and second-order pieces are $O_p(n^{-2\varrho})$ and $O_p(n^{1-4\varrho})$, and the sieve-mean piece is $o_p(1)$ by the bound above. The between-piece cross terms of the decomposition $\boldsymbol r = \boldsymbol r_{\mathrm{sieve}}+\boldsymbol r_{\mathrm{center}}+\boldsymbol r_{\mathrm{proj}}$ are $o_p(1)$ as well: Cauchy--Schwarz for the positive semidefinite form $(\boldsymbol a,\boldsymbol b)\mapsto\boldsymbol a^\top(\boldsymbol K/n)\boldsymbol b$ bounds each mixed term by the product of the square roots of the corresponding diagonal quadratic forms, giving $o_p(1)\cdot O_p(n^{-\varrho}) = o_p(1)$ (sieve$\times$centered), $o_p(1)\cdot O_p(n^{1/2-2\varrho}) = o_p(1)$ for $\varrho>1/4$ (sieve$\times$second-order), and $O_p(n^{1/2-3\varrho}) = o_p(1)$ for $\varrho>1/4$ (centered$\times$second-order). The constants may be taken uniform in $\gamma \in [\gamma_L,\gamma_U]$: the proof's only $\gamma$-dependent inputs are $\operatorname{tr}(\boldsymbol K_\gamma/n) = 1$ and the bounds on $\lambda$.

\subsubsection{Proofs of Corollaries~\ref{cor:proj_null}--\ref{cor:proj_local}}
\label{proof:cors}
By Lemma~\ref{lemma:estimation_remainder} and the subsequent RKHS norm bounds, replacing $\hat{\boldsymbol{\Pi}}\boldsymbol{\hat{\varepsilon}}$ with $\hat{\boldsymbol{\Pi}}\boldsymbol{\varepsilon}_0$ in the scaled coefficient function and mean embedding introduces only an $o_p(1)$ remainder in the RKHS norm. Because the test statistics are continuous functionals (inner products and squared norms) of these RKHS objects, the continuous mapping theorem allows this substitution, and it remains to determine the law of the projected mean embedding $\sqrt{n}\,\hat{m}_\perp := n^{-1/2}\sum_{i=1}^n (\hat{\boldsymbol{\Pi}}\boldsymbol{\varepsilon}_0)_i\,k(X_i,\cdot)$.

The projection correction does \emph{not} vanish at the scaled level. Write $\hat \beta := (\hat{\mathbb{G}}^\top\hat{\mathbb{G}})^{-1}\hat{\mathbb{G}}^\top\boldsymbol{\varepsilon}_0 \in \mathbb{R}^{d_\theta}$ for the empirical regression coefficient of $\boldsymbol{\varepsilon}_0$ on the score matrix (a distinct symbol from the KRR estimator $\hat w$); then $\hat{\boldsymbol{\Pi}}\boldsymbol{\varepsilon}_0 = \boldsymbol{\varepsilon}_0 - \hat{\mathbb{G}}\hat \beta$ and
\[
    \sqrt{n}\,\hat{m}_\perp = \sqrt{n}\,\hat{m} - \big(R_n^\top\hat{\mathbb{G}}\big)\big(\sqrt{n}\,\hat \beta\big),
    \qquad
    \hat{m} := \frac{1}{n}\sum_{i=1}^n \varepsilon_{0,i}\,k(X_i,\cdot).
\]
Under $H_0$, $R_n^\top\hat{\mathbb{G}} = \frac{1}{n}\sum_{i=1}^n k(X_i,\cdot)\,\hat g_{\hat\theta}(X_i)^\top \xrightarrow{p} A := \mathbb{E}[g_{\theta_0}(X)\,k(X,\cdot)] \in \mathcal{H}^{d_\theta}$ (the law of large numbers in $\mathcal{H}$, with Assumptions~\ref{ass:E} and~\ref{ass:F}), $\hat{\mathbb{G}}^\top\hat{\mathbb{G}}/n \xrightarrow{p} \Gamma_{\theta_0}$, and, jointly with the CLT $\sqrt{n}\,\hat{m} \xrightarrow{d} \xi$ for the unprojected embedding ($\xi$ being the Gaussian element of Theorem~\ref{th:proj_null_fix}), the Cram\'er--Wold device applied to the $\mathcal{H}\times\mathbb{R}^{d_\theta}$-valued sums $\sum_i(\varepsilon_{0,i}k(X_i,\cdot),\,\varepsilon_{0,i}\hat g_{\hat\theta}(X_i))$ yields this joint limit, whose cross-covariance is computed below,
\[
    \sqrt{n}\,\hat \beta = \Big(\frac{\hat{\mathbb{G}}^\top\hat{\mathbb{G}}}{n}\Big)^{-1}\frac{\hat{\mathbb{G}}^\top\boldsymbol{\varepsilon}_0}{\sqrt{n}} \xrightarrow{d} \Gamma_{\theta_0}^{-1}W,
    \qquad
    W \sim \mathcal{N}\big(\boldsymbol{0}, \mathbb{E}[\sigma^2(X)g_{\theta_0}(X)g_{\theta_0}(X)^\top]\big),
\]
with $\operatorname{Cov}(\langle\xi, f\rangle, W) = \mathbb{E}[\sigma^2(X)f(X)g_{\theta_0}(X)]$ for every $f \in \mathcal{H}$, the cross-covariance following from $\mathbb{E}[v^2 \mid X] = \sigma^2(X)$.

\emph{Score estimation inside the sums.} The summands involve the estimated score
$\hat g_{\hat\theta}(X_i)$, which is data-dependent through $\hat\theta$, so the sums are not
literally built from i.i.d.\ summands. Decompose
$\hat g_{\hat\theta}=g_{\theta_0}+(\hat g_{\hat\theta}-g_{\theta_0})$: the population-score part
$n^{-1}\sum_i k(X_i,\cdot)\,g_{\theta_0}(X_i)^{\top}$ is the i.i.d.\ mean to which the LLN/CLT
statements apply, while the sieve-error part is bounded by
$\bar k\,\lVert\hat g_{\hat\theta}-g_{\theta_0}\rVert_{L^{1}(\mathbb P_n)}
\le\bar k\,\lVert\hat g_{\hat\theta}-g_{\theta_0}\rVert_{L^{2}(\mathbb P_n)}$. In the vacuous-sieve
regime ($K_n=p$ fixed), pointwise Lipschitz continuity of $x\mapsto\hat g_{\hat\theta}(x)$ in
$\hat\theta$ gives $\lVert\hat g_{\hat\theta}-g_{\theta_0}\rVert_\infty=O_p(\lVert\hat\theta-\theta_0\rVert)
=O_p(n^{-\varrho})=o_p(1)$, so the replacement is $o_p(1)$ and every displayed limit is unchanged. For a
genuine growing sieve, the sup-norm deviation requires a triangular-array treatment beyond
Assumption~\ref{ass:F}(iii): the limit statements below therefore carry the explicit side condition
$\lVert\hat g_{\hat\theta}-g_{\theta_0}\rVert_{\infty}=O_p(n^{-1/4})$
(sufficient for all uses here; implied, e.g., by Lipschitz sieve bases with
$\lVert\hat g_{\hat\theta}-g_{\theta_0}\rVert_{L^{\infty}}=O_p(K_n^{-\varsigma})$,
$\varsigma K_n/n\to0$, and $\varrho>1/4$). Hence, by the continuous mapping theorem,
\[
    \sqrt{n}\,\hat{m}_\perp \xrightarrow{d} G_\perp := \xi - A^\top\Gamma_{\theta_0}^{-1}W,
\]
a non-degenerate Gaussian element different from $\xi$. In the location model $g_{\theta_0} \equiv 1$, $G_\perp = \xi - Z\,\mathbb{E}[k(X,\cdot)]$ with $Z \sim \mathcal{N}(0,\bar\sigma^{2})$, $\bar\sigma^{2}=\mathbb E[\sigma^2(X)]$: the law of the noise with the constant direction removed. The directional projections $Z_i^\perp := \langle G_\perp, e_i\rangle_{\HK}/\sqrt{\mu_i}$ are jointly Gaussian with the covariance matrix $\Sigma^\perp$ of Corollary~\ref{cor:proj_null}: $\Sigma^\perp_{ij} = \mathbb{E}[\sigma^2(X)(\phi_i(X)-\kappa_i(X))(\phi_j(X)-\kappa_j(X))]$, where $\kappa_i$ is the population $L^2$-projection coefficient of $\phi_i$ onto $\operatorname{span}(g_{\theta_0})$, $\kappa_i(x):=g_{\theta_0}(x)^\top\Gamma_{\theta_0}^{-1}\mathbb E[g_{\theta_0}(X)\phi_i(X)]$. The exact distributional limits of Corollary~\ref{cor:proj_null} follow by the same continuous functionals as in the proof of Theorem~\ref{th:proj_null_fix}. Under the local alternative $H_{1,n}$, $\boldsymbol{\varepsilon}_0 = R/\sqrt{n} + \boldsymbol v$: the noise law is unchanged ($\mathbb{E}[v \mid X] = 0$), and the mean shift is $\mathbb{E}(\boldsymbol{\Pi} R(X)\phi_i(X))$. To see the shift, split the scaled projected embedding $\sqrt{n}\,\hat{m}_\perp = n^{-1/2}\sum_i(\hat{\boldsymbol{\Pi}}\boldsymbol{\varepsilon}_0)_i k(X_i,\cdot)$ as in the null case into the noise part (unchanged, converging to $G_\perp$) and the shift part $n^{-1/2}\sum_i(\hat{\boldsymbol{\Pi}}R)_i k(X_i,\cdot)/\sqrt{n} = n^{-1}\sum_i(\hat{\boldsymbol{\Pi}}R)_i k(X_i,\cdot) \xrightarrow{p} \mathbb{E}[\boldsymbol{\Pi}R(X)k(X,\cdot)]$.
In addition, the $\hat \beta$-part is not negligible: $\sqrt{n}\,\hat \beta = (\hat{\mathbb{G}}^\top\hat{\mathbb{G}}/n)^{-1}\hat{\mathbb{G}}^\top\boldsymbol{\varepsilon}_0/\sqrt{n} \xrightarrow{d} \Gamma_{\theta_0}^{-1}(\mathbb{E}[g_{\theta_0}(X)R(X)] + W')$, so $(R_n^\top\hat{\mathbb{G}})\sqrt{n}\,\hat \beta$ contributes the deterministic limit $A^\top\Gamma_{\theta_0}^{-1}\mathbb{E}[g_{\theta_0}(X)R(X)]$. The two deterministic parts combine into the projected shift
\[
    \mathbb{E}[R(X)k(X,\cdot)] - A^\top\Gamma_{\theta_0}^{-1}\mathbb{E}[g_{\theta_0}(X)R(X)] = \mathbb{E}[\boldsymbol{\Pi}R(X)k(X,\cdot)],
\]
whose $\phi_i$-coefficient is $\mathbb{E}[\boldsymbol{\Pi}R(X)\phi_i(X)]$, as claimed. Under a fixed alternative with pseudo-true limit $(\theta^*, f^*)$ as in Corollary~\ref{cor:proj_null}: $\hat\theta\to\theta^*$, conditional mean $f^*$, rate condition around $\theta^*$, and $\boldsymbol{\Pi}_{\theta^*}f^*\neq0$, so that $\boldsymbol{\Pi}f \neq 0$. The population KRR fit on the projected residuals is $w^*_{\lambda,\perp} \neq 0$ (the kernel integral operator is injective), and the consistency of Lemma~\ref{lemma:estimation_consistency} transfers to the projected residuals: the empirical mean embedding $n^{-1}\sum_i(\hat{\boldsymbol{\Pi}}\boldsymbol{\varepsilon}_{\hat\theta})_i k(X_i,\cdot)$ converges to $m_{\perp} := \mathbb{E}[\boldsymbol{\Pi}f\, k(X,\cdot)] \neq \boldsymbol 0$ by the same injectivity. Two laws of large numbers combine: $n^{-1}\sum_i\varepsilon_{\hat\theta,i}\,k(X_i,\cdot) \xrightarrow{p} \mathbb{E}[f^*\,k(X,\cdot)]$, the conditional mean of $\varepsilon_{\hat\theta}$ converging to $f^*$ through $\hat\theta \xrightarrow{p} \theta^*$ and the sup-norm side condition of Assumption~\ref{ass:F}(iii); and the projection correction converges, with $\hat\beta_{\hat\theta} := (\hat{\mathbb{G}}^\top\hat{\mathbb{G}})^{-1}\hat{\mathbb{G}}^\top\boldsymbol{\varepsilon}_{\hat\theta} \xrightarrow{p} \Gamma_{\theta^*}^{-1}\,\mathbb{E}[g_{\theta^*}(X)\,f^*(X)]$ and $\boldsymbol{R}_n^\top\hat{\mathbb{G}} \xrightarrow{p} \mathbb{E}[g_{\theta^*}(X)\,k(X,\cdot)]$, to $\mathbb{E}[g_{\theta^*}(X)\,k(X,\cdot)]\,\Gamma_{\theta^*}^{-1}\,\mathbb{E}[g_{\theta^*}(X)\,f^*(X)]$; by the projection identity $\mathbb{E}[\boldsymbol{\Pi}R\,k] = \mathbb{E}[R\,k] - \mathbb{E}[g\,k]\,\Gamma^{-1}\,\mathbb{E}[g\,R]$ established above, applied at $(R,\,g,\,\Gamma) = (f^*,\,g_{\theta^*},\,\Gamma_{\theta^*})$, the two limits combine into $m_\perp = \mathbb{E}[\boldsymbol{\Pi}f\,k(X,\cdot)]$; this gives $\lVert\hat w_\perp\rVert_{\HK} \xrightarrow{p} \lVert w^*_{\lambda,\perp}\rVert_{\HK} > 0$; moreover $\langle w^*_{\lambda,\perp}, m_{\perp}\rangle_{\HK} = \langle (C+\lambda I)^{-1}m_{\perp}, m_{\perp}\rangle_{\HK} \ge \lambda\lVert w^*_{\lambda,\perp}\rVert_{\HK}^2 > 0$, and both statistics diverge: $\mathbb{P}(n\hat T^{(k)}_{\mathrm{proj},\perp} > t) \to 1$ for $k = 1,2$ and every fixed $t > 0$.

\subsubsection{Proof of Theorem~\ref{th:bootstrap_consistency}}
\label{proof:thm6}

\noindent\textbf{Estimation effect is asymptotically negligible.}\;
By Lemma~\ref{lemma:estimation_remainder}, the orthogonalisation removes the
first-order estimation effect: the RKHS remainders $R_w$ and $R_m$ involving
$\hat{\theta}-\theta_0$ are $o_p(1)$.  Hence, in the asymptotic analysis we
may replace the projected bootstrap residual
$\hat{\boldsymbol{\Pi}}(\boldsymbol{\hat{\varepsilon}}\odot\mathbf{V})$
by $\hat{\boldsymbol{\Pi}}(\boldsymbol{\varepsilon}_0\odot\mathbf{V})$.
The lemma is stated for the unmultiplied difference; the multiplied
remainder requires a separate argument, because the multipliers break the
annihilation of the sieve-projected part, $\hat{\boldsymbol{\Pi}}(\hat{\mathbb{G}}(\hat\theta-\theta_0)) = \boldsymbol 0$ (the diagonal multiplier does not commute with the projection), while the projection annihilates only $\hat{\mathbb{G}}$, not the raw score.
Write $\boldsymbol{\Delta} := \boldsymbol{\hat{\varepsilon}} - \boldsymbol{\varepsilon}_0 = \mathbb{G}\boldsymbol{\Delta}\theta + \boldsymbol r^{(2)}$ with $\boldsymbol{\Delta}\theta := \hat\theta - \theta_0$, $\mathbb{G}$ the matrix of raw scores with rows $\nabla_\theta\varepsilon(s_i;\hat\theta)^\top$, and $\boldsymbol r^{(2)}$ the second-order remainder of Lemma~\ref{lemma:estimation_remainder} (Term A + Term B).
Set $\mathsf{W}_n := (\boldsymbol K/n+\lambda I)^{-1}(\boldsymbol K/n)(\boldsymbol K/n+\lambda I)^{-1}$, so that $\lVert\sqrt{n}\,\boldsymbol{R}_n^\top(\boldsymbol K/n+\lambda I)^{-1}\boldsymbol u\rVert_{\HK}^2 = \boldsymbol u^\top \mathsf{W}_n \boldsymbol u$ for every vector $\boldsymbol u$.
The conditional second moment of the scaled multiplied remainder satisfies, using $\mathbb{E}[V_i^2] = 1$,
\[
  \mathbb{E}\bigl[\lVert \sqrt{n}\,\boldsymbol{R}_n^\top(\boldsymbol K/n+\lambda I)^{-1}\hat{\boldsymbol{\Pi}}(\boldsymbol{\Delta}\odot\boldsymbol V)\rVert_{\HK}^2 \mid D_n\bigr]
  = \sum_{i=1}^n \big(\hat{\boldsymbol{\Pi}} \mathsf{W}_n \hat{\boldsymbol{\Pi}}\big)_{ii}\,\Delta_i^2 .
\]
For the first-order part: the multiplied first-order piece is
$\boldsymbol{\Delta}^{(1)} := (\mathbb{G}\boldsymbol{\Delta}\theta)\odot\boldsymbol V$, and since
$\mathbb E[V_i^{2}]=1$ its conditional second moment is computed from the unmultiplied squares; that
second moment equals $(\boldsymbol{\Delta}\theta)^\top B_n\,(\boldsymbol{\Delta}\theta)$ with $B_n := \sum_{i=1}^n\big(\hat{\boldsymbol{\Pi}} \mathsf{W}_n \hat{\boldsymbol{\Pi}}\big)_{ii}\,\nabla_\theta\varepsilon(s_i;\hat\theta)\nabla_\theta\varepsilon(s_i;\hat\theta)^\top$, and $\lVert B_n\rVert_{\mathrm{op}} = O_p(\sqrt{n})$:
by exchangeability $\mathbb{E}\big[\big(\hat{\boldsymbol{\Pi}} \mathsf{W}_n \hat{\boldsymbol{\Pi}}\big)_{11}\big] = n^{-1}\,\mathbb{E}\big[\operatorname{tr}\big(\hat{\boldsymbol{\Pi}} \mathsf{W}_n \hat{\boldsymbol{\Pi}}\big)\big] \le n^{-1}\,\mathbb{E}[\operatorname{tr}(\mathsf{W}_n)] \le n^{-1}\lambda^{-2}\,\mathbb{E}[\operatorname{tr}(\boldsymbol K/n)] = O(n^{-1})$ (using $\operatorname{tr}(\boldsymbol K/n) = 1$), and $\big(\hat{\boldsymbol{\Pi}} \mathsf{W}_n \hat{\boldsymbol{\Pi}}\big)_{11} \le \lVert \mathsf{W}_n\rVert_{\mathrm{op}} \le 1/(4\lambda)$, so $\mathbb{E}\big[\big(\hat{\boldsymbol{\Pi}} \mathsf{W}_n \hat{\boldsymbol{\Pi}}\big)_{11}^2\big] \le (1/(4\lambda))\,\mathbb{E}\big[\big(\hat{\boldsymbol{\Pi}} \mathsf{W}_n \hat{\boldsymbol{\Pi}}\big)_{11}\big] = O(n^{-1})$.
Hence $\mathbb{E}\lVert B_n\rVert_{\mathrm{op}} \le n\,\mathbb{E}\big[\big(\hat{\boldsymbol{\Pi}} \mathsf{W}_n \hat{\boldsymbol{\Pi}}\big)_{11}\,\lVert\nabla_\theta\varepsilon(S_1;\hat\theta)\rVert^2\big] \le n\,\big(\mathbb{E}\big[\big(\hat{\boldsymbol{\Pi}} \mathsf{W}_n \hat{\boldsymbol{\Pi}}\big)_{11}^2\big]\big)^{1/2}\big(\mathbb{E}\lVert\nabla_\theta\varepsilon(S_1;\hat\theta)\rVert^4\big)^{1/2} = O(\sqrt{n})$ (the fourth moment of the raw score by Assumption~\ref{ass:E}(iii)),
so $\lVert B_n\rVert_{\mathrm{op}} = O_p(\sqrt{n})$ by Markov, and the first-order contribution is $O_p\big(\lVert\boldsymbol{\Delta}\theta\rVert^2\,\sqrt{n}\big) = O_p(n^{1/2-2\varrho}) = o_p(1)$ for $\varrho > 1/4$.
For the second-order part, $\sum_i\big(\hat{\boldsymbol{\Pi}} \mathsf{W}_n \hat{\boldsymbol{\Pi}}\big)_{ii}\,r_i^{(2)2} \le \lVert \mathsf{W}_n\rVert_{\mathrm{op}}\,\lVert \boldsymbol r^{(2)}\rVert^2 = O_p(n^{1-4\varrho}) = o_p(1)$ (the second-moment version of Lemma~\ref{lemma:estimation_remainder} under Assumption~\ref{ass:E}(iii)).
The expansion in $\Delta_i = \nabla_\theta\varepsilon(s_i;\hat\theta)^\top\boldsymbol{\Delta}\theta + r_i^{(2)}$ also carries the first$\times$second-order cross term $2\sum_i\big(\hat{\boldsymbol{\Pi}} \mathsf{W}_n \hat{\boldsymbol{\Pi}}\big)_{ii}\big(\nabla_\theta\varepsilon(s_i;\hat\theta)^\top\boldsymbol{\Delta}\theta\big)r_i^{(2)}$; Cauchy--Schwarz for the weighted inner product with weights $w_i := \big(\hat{\boldsymbol{\Pi}} \mathsf{W}_n \hat{\boldsymbol{\Pi}}\big)_{ii} \ge 0$ bounds it by $2\big[(\boldsymbol{\Delta}\theta)^\top B_n(\boldsymbol{\Delta}\theta)\big]^{1/2}\big[\sum_i w_i\big(r_i^{(2)}\big)^2\big]^{1/2} = O_p\big(n^{1/4-\varrho}\big)\,O_p\big(n^{1/2-2\varrho}\big) = O_p\big(n^{3/4-3\varrho}\big) = o_p(1)$ for $\varrho > 1/4$ (Assumption~\ref{ass:D}(iii)), using the two bounds just established.
Hence the conditional second moment of the multiplied remainder is $o_p(1)$, so it is conditionally $o_p(1)$ by Markov, and the replacement is justified; the same bound holds for the $\boldsymbol{R}_m$-type remainder.  The projection itself is \emph{not} negligible at the scaled level: the
componentwise convergence
$\hat{\boldsymbol{\Pi}}\boldsymbol{\varepsilon}_0
\xrightarrow{p}\boldsymbol{\varepsilon}_0$ does not survive at the
$O_p(1)$ scale of the scaled mean embedding (see the proof of
Corollary~\ref{cor:proj_null}).  The bootstrap embedding therefore retains
the projection, in the \emph{multiply-then-project} form that matches the
bootstrap statistics of Section~\ref{sec:estimation_effect},
\[
  \tilde{m}_{\perp} := \frac{1}{n}\sum_{i=1}^n \big(\hat{\boldsymbol{\Pi}}(\boldsymbol{\varepsilon}_0\odot\boldsymbol V)\big)_i\,
                        k(X_i,\cdot),
\]
and its conditional limit is the projected law of Corollary~\ref{cor:proj_null}.

\smallskip
\noindent\textbf{Multiplier CLT in the RKHS.}\;
Using $\hat{\boldsymbol{\Pi}}\boldsymbol{u} = \boldsymbol{u} - \hat{\mathbb{G}}(\hat{\mathbb{G}}^\top\hat{\mathbb{G}})^{-1}\hat{\mathbb{G}}^\top\boldsymbol{u}$ for the score projection of the multiplied residuals,
\[
  \sqrt{n}\,\tilde{m}_{\perp} = \tilde{m}^* - \big(R_n^\top\hat{\mathbb{G}}\big)\,\hat w^*,
  \qquad
  \tilde{m}^* := \frac{1}{\sqrt{n}}\sum_{i=1}^n V_i\,\varepsilon_{0,i}\,k(X_i,\cdot),
  \qquad
  \hat w^* := \Big(\frac{\hat{\mathbb{G}}^\top\hat{\mathbb{G}}}{n}\Big)^{-1}\frac{\hat{\mathbb{G}}^\top(\boldsymbol{V}\odot\boldsymbol{\varepsilon}_0)}{\sqrt{n}} .
\]
Conditional on the data $D_n$, the pair $(\tilde{m}^*, \hat w^*)$ converges
jointly: the multiplier CLT in Hilbert spaces
(Theorem~\ref{th:multiplier_clt_rkhs}) applied to the i.i.d.\ elements
$U_i := \varepsilon_{0,i}\,k(X_i,\cdot)$ (bounded kernel,
$\mathbb{E}[\|U_i\|_{\HK}^2] \le \bar k\,\mathbb{E}[\varepsilon_0^2] < \infty$
by Assumption~\ref{ass:B}(v))
gives $\tilde{m}^* \xrightarrow{d} \xi^* \sim \mathcal{N}(0,\Sigma)$
conditionally on the data (a.s.\ by the theorem, hence in probability), with
$\Sigma = \mathbb{E}[\varepsilon_0^2\,k(X,\cdot)\otimes k(X,\cdot)]$; the
finite-dimensional conditional CLT gives
$\hat w^* \xrightarrow{d} \Gamma_{\theta_0}^{-1}W^*$, with
$W^* \sim \mathcal{N}(\boldsymbol{0}, \mathbb{E}[\sigma^2(X)g_{\theta_0}(X)g_{\theta_0}(X)^\top])$
(conditionally, in probability).  The Lindeberg input is assembled in two
steps: the tail bound $\max_{i\le n}\varepsilon_{0,i}^2 = o_p(n)$ holds by
Chebyshev's inequality and the union bound,
$\mathbb{P}(\max_{i\le n}\varepsilon_{0,i}^2 > \varepsilon^2 n) \le n\,\mathbb{E}[\varepsilon_0^2\mathbf{1}(\varepsilon_0^2 > \varepsilon^2 n)]/(\varepsilon^2 n)
= \mathbb{E}[\varepsilon_0^2\mathbf{1}(\varepsilon_0^2 > \varepsilon^2 n)]/\varepsilon^2 \to 0$
by dominated convergence under $\mathbb{E}[\varepsilon_0^2] < \infty$
(Assumption~\ref{ass:B}(v)); and the bounded support of the multipliers
($|V_i| \le M$) converts this max-bound into the Lindeberg form for the
elements $a_i := \varepsilon_{0,i}\hat g_{\hat\theta}(X_i)$,
\[
  \frac{1}{n}\sum_{i=1}^n a_i^2\,\mathbb{E}_V\bigl[\mathbf 1\bigl(|V_i a_i| > \epsilon\sqrt n\bigr)\bigr]
  \le \frac{1}{n}\sum_{i=1}^n a_i^2\,\mathbf 1\Bigl(a_i^2 > \frac{\epsilon^2 n}{M^2}\Bigr)
  \le \frac{M^2}{\epsilon^2 n}\,\Bigl(\max_{i\le n}a_i^2\Bigr)\,\frac{1}{n}\sum_{i=1}^n a_i^2
  = o_p(1),
\]
using $a_i^2\mathbf 1(a_i^2 > \epsilon^2 n/M^2) \le a_i^4\,M^2/(\epsilon^2 n)$,
$\max_{i\le n}a_i^2 \le \max_{i\le n}\varepsilon_{0,i}^2\,\lVert\hat g_{\hat\theta}\rVert_\infty^2 = o_p(n)$
by the tail bound and the sup-norm side condition of Assumption~\ref{ass:F}(iii)
together with the uniform score bound of Assumption~\ref{ass:E}(i), and
$n^{-1}\sum_i a_i^2 = O_p(1)$ by the law of large numbers under
Assumptions~\ref{ass:B}(v) and~\ref{ass:E}(i);
and the conditional cross-covariance satisfies
\[
  \frac{1}{n}\sum_{i=1}^n \varepsilon_{0,i}^2\,k(X_i,\cdot)\,\hat g_{\hat\theta}(X_i)^\top
  \xrightarrow{p}
  \mathbb{E}\big[\sigma^2(X)\,k(X,\cdot)\,g_{\theta_0}(X)^\top\big],
\]
so $(\xi^*, W^*)$ share the joint covariance structure of $(\xi, W)$ in
Remark~\ref{rem:proj_null}.  The joint convergence of the pair
$(\tilde{m}^*, \hat w^*)$ is the conditional multiplier CLT in the product
space $\HK \times \mathbb{R}^{d_\theta}$ for the elements
$z_i := (U_i, \varepsilon_{0,i}\hat g_{\hat\theta}(X_i))$, invoked through
the Cram\'er--Wold device
applied to the characteristic-functional argument of
Theorem~\ref{th:multiplier_clt_rkhs}.  Its three ingredients are as follows.
(i) Conditional Lindeberg for the joint elements: with bounded support
$|V_i| \le M$ and
$\lVert z_i\rVert^2
= \varepsilon_{0,i}^2\big(\|k(X_i,\cdot)\|_{\HK}^2 + \hat g_{\hat\theta}(X_i)^2\big)
\le \varepsilon_{0,i}^2\big(\bar k + (\bar G + \lVert\hat g_{\hat\theta}-g_{\theta_0}\rVert_\infty)^2\big)
= \varepsilon_{0,i}^2\,O_p(1)$,
the two-step chain just displayed applies verbatim to $\lVert z_i\rVert^2$
in place of $\varepsilon_{0,i}^2$.
(ii) Replacement of $\hat g_{\hat\theta}$ by $g_{\theta_0}$ inside the
$\varepsilon_0^2$-weighted sums: the replacement error in the
cross-moment display above is bounded by
$\bar k^{1/2}\,\lVert\hat g_{\hat\theta}-g_{\theta_0}\rVert_\infty\,n^{-1}\sum_i\varepsilon_{0,i}^2
= O_p(n^{-1/4}) = o_p(1)$, and the corresponding replacement inside the
linear statistic $n^{-1/2}\sum_iV_i\varepsilon_{0,i}(\hat g_{\hat\theta}-g_{\theta_0})(X_i)$
has conditional second moment at most
$\lVert\hat g_{\hat\theta}-g_{\theta_0}\rVert_\infty^2\,n^{-1}\sum_i\varepsilon_{0,i}^2
= O_p(n^{-1/2}) = o_{p^*}(1)$, both by the sup-norm side condition of
Assumption~\ref{ass:F}(iii); the unweighted corollary sums are covered by
the same replacement in the proof of
Corollaries~\ref{cor:proj_null} and~\ref{cor:proj_local}.
(iii) Diagonal blocks of the joint covariance operator:
$n^{-1}\sum_i\varepsilon_{0,i}^2 k(X_i,\cdot)\otimes k(X_i,\cdot)
\xrightarrow{a.s.} \mathbb{E}[\varepsilon_0^2 k(X,\cdot)\otimes k(X,\cdot)]$
and
$n^{-1}\sum_i\varepsilon_{0,i}^2 g_{\theta_0}(X_i)g_{\theta_0}(X_i)^\top
\xrightarrow{a.s.} \mathbb{E}[\sigma^2(X)g_{\theta_0}(X)g_{\theta_0}(X)^\top]$
by Mourier's strong law in the separable Hilbert space
$\HK\otimes\mathbb{R}^{d_\theta}$ (first moment
$\mathbb{E}[\varepsilon_0^2(\bar k + \bar G^2)] < \infty$ by
Assumptions~\ref{ass:B}(v) and~\ref{ass:E}(i), (iii)), combined with the
replacement of (ii).  Since
$R_n^\top\hat{\mathbb{G}} \xrightarrow{p} A := \mathbb{E}[g_{\theta_0}(X)k(X,\cdot)]$,
Slutsky's theorem gives, conditionally on the data and in probability with respect to it (the
strengthened almost-sure form would require summability rates for the empirical second moments along
$n$ that are not established here and are not needed below; accordingly the theorem states the
conditional limit in probability), 
\[
  \sqrt{n}\,\tilde{m}_{\perp}
  \;\xrightarrow{d}\;
  \xi^* - A^\top\Gamma_{\theta_0}^{-1}W^* =: G_\perp^*
  \qquad \text{in } \HK,
\]
where $\xrightarrow{d}$ denotes conditional weak convergence in probability over the multiplier
weights, matching the notation of the theorem statement, and $G_\perp^*$ has the projected law of
Corollary~\ref{cor:proj_null} (not the infeasible law of Theorem~\ref{th:proj_null_fix}).

\smallskip
\noindent\textbf{Spectral expansion of the limit.}\;
Let $\{\mu_i\}_{i\ge 1}$ and $\{\phi_i\}_{i\ge 1}$ be the eigenvalues and
eigenfunctions of the integral operator $\mathbf{L}_k$, with RKHS basis
$e_i = \sqrt{\mu_i}\,\phi_i$ (see \S\ref{sec:equivalence}).  The covariance
operator $\Sigma^\perp$ of $G_\perp$ satisfies
\[
  \langle\Sigma^\perp e_i, e_j\rangle_{\HK}
  = \sqrt{\mu_i\mu_j}\,
    \mathbb{E}\bigl[\sigma^2(X)\,
     \big(\phi_i(X)-\kappa_i(X)\big)\,
     \big(\phi_j(X)-\kappa_j(X)\big)\bigr]
  = \sqrt{\mu_i\mu_j}\,\Sigma^\perp_{ij},
\]
with $\Sigma^\perp_{ij}$ as in Corollary~\ref{cor:proj_null}.
In general this matrix is \textbf{not diagonal} and is generically singular
(the directions in $\operatorname{span}(g_{\theta_0})$ are annihilated by the
projection): even under homoscedasticity
$\Sigma^\perp_{ij} = \sigma^2(\delta_{ij} - \mathbb{E}[\phi_i g^\top]\Gamma^{-1}\mathbb{E}[g\phi_j])$,
which is diagonal only when the nuisance directions are aligned with the
eigenbasis (e.g.\ the location model).  Expanding $G_\perp$ in the
orthonormal basis $\{e_i\}$ gives
$G_\perp = \sum_{i=1}^\infty \zeta_i^\perp\,e_i$ with
$\zeta_i^\perp := \langle G_\perp,e_i\rangle_{\HK}$, where the vector
$(\zeta_1^\perp,\zeta_2^\perp,\dots)$ is \textbf{jointly} Gaussian with
$\operatorname{Cov}(\zeta_i^\perp,\zeta_j^\perp)
= \langle\Sigma^\perp e_i,e_j\rangle_{\HK}
= \sqrt{\mu_i\mu_j}\,\Sigma^\perp_{ij}$.
Equivalently, defining $Z_i^\perp := \zeta_i^\perp / \sqrt{\mu_i}$, we have
$\zeta_i^\perp = \sqrt{\mu_i}\,Z_i^\perp$ and
\[
  G_\perp = \sum_{i=1}^\infty \sqrt{\mu_i}\,Z_i^\perp\,e_i,
\]
where $Z^\perp = (Z_1^\perp,Z_2^\perp,\dots) \sim \mathcal{N}(0,\Sigma^\perp)$ with
$\operatorname{Cov}(Z_i^\perp,Z_j^\perp) = \Sigma^\perp_{ij}$.

\smallskip
\noindent\textbf{Weak convergence of the bootstrap KRR estimator.}\;
The bootstrap KRR estimator is $\tilde{w}_{\perp} = (C_n + \lambda I)^{-1}
\tilde{m}_{\perp}$.  By Lemma~\ref{lm:C_and_Cn},
$\|C_n - C\|_{\mathrm{HS}} = O_p(n^{-1/2})$, hence
$\|C_n - C\|_{\mathrm{op}} = o_p(1)$.  Since $\lambda>0$,
$(C_n + \lambda I)^{-1} \xrightarrow{p} (C + \lambda I)^{-1}$ in operator
norm.  The mapping $(A,x)\mapsto Ax$ is continuous from
$\mathcal{B}(\HK)\times\HK$ to $\HK$, so Slutsky's theorem in Hilbert
spaces (the continuous mapping theorem for joint weak convergence) yields,
conditionally on $D_n$,
\[
  \sqrt{n}\,\tilde{w}_{\perp}
  = (C_n + \lambda I)^{-1}(\sqrt{n}\,\tilde{m}_{\perp})
  \xrightarrow{d}
  (C + \lambda I)^{-1}G_\perp
  = \sum_{i=1}^\infty \frac{\sqrt{\mu_i}}{\mu_i + \lambda}\,
    Z_i^\perp\,e_i .
\]

\smallskip
\noindent\textbf{Convergence of the test statistics.}\;
All test statistics are continuous functions of the scaled KRR estimator
(or the pair $(\sqrt{n}\tilde{w}_{\perp},\sqrt{n}\tilde{m}_{\perp})$).
Applying the continuous mapping theorem:

\textbf{Projection-based statistics:}
\begin{align*}
  n\tilde{T}_{\mathrm{proj},\perp}^{(1)}
  &= \langle\sqrt{n}\tilde{w}_{\perp},\sqrt{n}\tilde{w}_{\perp}\rangle_{\HK}
    \xrightarrow{d^*}
    \bigl\|(C + \lambda I)^{-1}G_\perp\bigr\|_{\HK}^2
    = \sum_{i=1}^\infty
      \frac{\mu_i}{(\mu_i + \lambda)^2}\,(Z_i^\perp)^2, \\[6pt]
  n\tilde{T}_{\mathrm{proj},\perp}^{(2)}
  &= \langle\sqrt{n}\tilde{w}_{\perp},\sqrt{n}\tilde{m}_{\perp}\rangle_{\HK}
    \xrightarrow{d^*}
    \langle(C + \lambda I)^{-1}G_\perp,G_\perp\rangle_{\HK}
    = \sum_{i=1}^\infty
      \frac{\mu_i}{\mu_i + \lambda}\,(Z_i^\perp)^2 .
\end{align*}

The limit distributions coincide with those in
Corollary~\ref{cor:proj_null}, establishing the claimed bootstrap
consistency.

\subsection{Proofs for Section~\ref{sec:nuisance}: Nuisance Parameter Selection}\label{app:proofs_nuisance}
\noindent This subsection proves the results of Section~\ref{sec:nuisance}: the two uniform-in-$\eta$ limit theorems and their consequences. The organization is by theorem family, with the auxiliary lemmas of each family collected before the proof they enter. The \emph{operator and kernel tools} serve both theorems: the Fourier isometry and the cross-kernel calculus for the transport to the common space, the Sobolev maximal inequality within cells, the residual split of the orthogonalized residuals, and the Lipschitz properties of the transported empirical objects and of the resolvent. The proof of Theorem~\ref{th:stat_uniform} then proceeds from the pointwise limits through the matrix representation and its gradient moments to boundedness and asymptotic equicontinuity via the equicontinuity lemma (Lemma~\ref{lem:equi}, whose proof runs on the Sobolev maximal inequality within cells and on the containment of short pairs by shifted grids, Lemma~\ref{lem:shifts}), tightness in $C(\Upsilon)$, and the local alternative; the sample-path continuity of the limit fields is established separately in its own subsubsection and cited by the downstream results. The lemmas specific to the criterion process are collected in their own subsubsection before the proof of Theorem~\ref{th:cv_uniform}: the leave-one-out prediction linearization and the criterion decomposition built on it, the replacement cost of the feasible residuals, the conditional variance formulas, the bilinear representation of the criterion's quadratic statistics, and the chaos identification; the proof of Theorem~\ref{th:cv_uniform} follows in parts (a)--(c), each with its own road map. The selector's limit laws, bootstrap validity under sample splitting, the energy identity, the monotone-power proposition, and the two-channel selection theorem complete the subsection.

\subsubsection{Auxiliary Lemmas for the Uniform Convergence Results}
\label{proof:aux_uniform}

The following lemmas collect the spectral and analytic ingredients used in the proofs of the uniform convergence results (Theorems~\ref{th:stat_uniform} and~\ref{th:cv_uniform}): the Fourier isometry of the Gaussian RKHS (Lemma~\ref{lem:fourier_iso}), the Sobolev maximal inequality within cells (Lemma~\ref{lem:sobolev}), the containment of short pairs by shifted grids (Lemma~\ref{lem:shifts}), the equicontinuity lemma that converts a uniform gradient moment into asymptotic equicontinuity (Lemma~\ref{lem:equi}), the residual split of Fact~\ref{fact:residual_split}, the order-$p$ moment bound for centered quadratic forms in conditionally independent innovations (Lemma~\ref{lem:quad_moment}), and the Lipschitz properties of the transported empirical objects and of the resolvent (Lemma~\ref{lem:lipschitz_common}), with its corollary on Lipschitz-modulus equicontinuity (Corollary~\ref{cor:covering}); the cross-kernel closed form (Lemma~\ref{lem:crosskernel}) and the derivative-feature calculus built on it (Fact~\ref{fact:pairdiff}) are shared by the proofs of parts (a) and (b) of Lemma~\ref{lem:lipschitz_common}. The conditional Lipschitz modulus and the uniform boundedness of the multiplied-then-projected embedding, the conditional counterpart of part~(a) consumed by the bootstrap uniformity result (Lemma~\ref{lem:bootstrap_uniform}), are recorded separately in Lemma~\ref{lem:cond_lip_mult}. Notation used throughout, ahead of the proofs in which the
objects reappear: $g_{\eta}$ is the smoothing kernel of \eqref{eq:g_eta} in the main text,
$P_{\eta}:=S_{\lambda}=C_{\gamma}(C_{\gamma}+\lambda I)^{-1}$ the population smoothing operator on
$L^{2}(\mathbb P_{X})$ (the two symbols interchangeable, $P_{\eta}$ making the joint dependence on
$\eta=(\gamma,\lambda)$ explicit), $\boldsymbol\Pi$ the $L^{2}$-orthogonal projection onto the
nuisance-orthogonal complement (a contraction), and $\Upsilon$ the compact tuning rectangle.
Three further lemmas serve the CV-criterion proof: the bilinear
representation of the criterion's quadratic statistics (Lemma~\ref{lem:bilinear_cv}), which exhibits
$mU_{m}$ and the order-two part of $mR_{m}$ as quadratic forms of the single scaled mean embedding
against deterministic population operators; the uniform bound for the anchor-centered remainder
(Lemma~\ref{lem:remainder_cv}), which renders that remainder $O_{p}(m^{-1/2})$ uniformly in $\eta$;
and the chaos identification (Lemma~\ref{lem:chaos_cv}), which transfers the chaos specification of
Theorem~\ref{th:cv_uniform} to these quadratic forms by a Wick computation in the white-noise
coupling. Proposition~\ref{prop:eigenpair_cont}, finally, certifies clause (ii) of
Assumption~\ref{ass:G} for the Gaussian family and is consumed by every result of the section
through that clause.

Throughout this subsection we use the following standard notation for the empirical objects of the testing-sample analysis (proof of Theorem~\ref{th:stat_uniform}): $C_{n,\gamma} := n^{-1}\sum_{j=1}^n k_\gamma(X_j,\cdot)\otimes k_\gamma(X_j,\cdot)$ is the empirical covariance operator on $\mathcal{H}_\gamma$ with spectral decomposition $C_{n,\gamma} = \sum_{i=1}^n \nu_{i,n}(\gamma)\,\hat{e}_{i,n}^{(\gamma)}\otimes\hat{e}_{i,n}^{(\gamma)}$ (eigenvalues of $\boldsymbol{K}_\gamma/n$, Nystr\"om eigenfunctions); $M_n(\gamma) := n^{-1/2}\sum_{j=1}^n \varepsilon_{\perp,j}\,k_\gamma(X_j,\cdot) \in \mathcal{H}_\gamma$ is the scaled mean embedding of the orthogonalized residuals; and $P_{>K}^{(\gamma)} := \sum_{i=K+1}^n \hat{e}_{i,n}^{(\gamma)}\otimes\hat{e}_{i,n}^{(\gamma)}$ is the empirical projector onto the tail eigenspace $\{\hat{e}_{i,n}^{(\gamma)} : i > K\}$. The transported objects $\widetilde{C}_{n,\gamma} := J_\gamma C_{n,\gamma} J_\gamma^{-1}$, $\widetilde{M}_n(\gamma) := J_\gamma M_n(\gamma)$, and $\widetilde{P}_{>K}^{(\gamma)} := J_\gamma P_{>K}^{(\gamma)} J_\gamma^{-1}$ live on the common space $\mathcal{H}_{\gamma_0}$ via the Fourier isometry $J_\gamma$ of Lemma~\ref{lem:fourier_iso}.

\begin{proposition}[Eigenpair continuity for the Gaussian family]\label{prop:eigenpair_cont}
Let $k_\gamma(x,y) = \exp(-\gamma\|x-y\|^2)$, $\gamma\in[\gamma_L,\gamma_U]$, $0<\gamma_L\le\gamma_U<\infty$, on a probability space $(\mathcal{X},\mathcal{A},\mathbb{P}_X)$ with $\mathcal{X}\subseteq\mathbb{R}^d$ and $\operatorname{supp}\mathbb{P}_X$ infinite. Then there exist continuous parameterizations of the spectral data: maps $\gamma\mapsto\mu_i(\gamma)$ and $\gamma\mapsto\phi_i(\cdot;\gamma)\in L^2(\mathbb{P}_X)$, $i\ge1$, continuous on $[\gamma_L,\gamma_U]$, such that for every $\gamma$ the numbers $0<\mu_1(\gamma)\ge\mu_2(\gamma)\ge\cdots$ enumerate the nonzero spectrum of $\boldsymbol{L}_{k_\gamma}$ with multiplicities, each $\phi_i(\cdot;\gamma)$ is a unit-norm eigenfunction of $\boldsymbol{L}_{k_\gamma}$ for $\mu_i(\gamma)$, and $\{\phi_i(\cdot;\gamma)\}_{i\ge1}$ is orthonormal in $L^2(\mathbb{P}_X)$. In particular, clause (ii) of Assumption~\ref{ass:G} holds for the Gaussian family.
\end{proposition}
\begin{proof}
For $\alpha\in\mathbb{N}_0$ the $\alpha$-th $\gamma$-derivative of the kernel is $\partial_\gamma^\alpha k_\gamma(x,y) = (-\|x-y\|^2)^\alpha e^{-\gamma\|x-y\|^2}$, and the maximum of $r^{2\alpha}e^{-\gamma r^2}$ over $r\ge0$ is attained at $r^2 = \alpha/\gamma$, so
\[
    \sup_{x,y}\bigl|\partial_\gamma^\alpha k_\gamma(x,y)\bigr| \le \Bigl(\frac{\alpha}{e\,\gamma_L}\Bigr)^{\!\alpha}, \qquad \gamma\in[\gamma_L,\gamma_U].
\]
Since $\mathbb{P}_X$ is a probability measure, constants lie in $L^2(\mathbb{P}_X\otimes\mathbb{P}_X)$, and dominated convergence applies with these constant dominators: $\gamma\mapsto\boldsymbol{L}_{k_\gamma}$ is $C^\infty$ in Hilbert--Schmidt norm, the $\alpha$-th derivative being the integral operator with kernel $\partial_\gamma^\alpha k_\gamma$, and
\[
    \bigl\|\partial_\gamma^\alpha\boldsymbol{L}_{k_\gamma}\bigr\|_{\mathrm{HS}} \le \Bigl(\frac{\alpha}{e\,\gamma_L}\Bigr)^{\!\alpha} \le \alpha!\,\Bigl(\frac{2}{\gamma_L}\Bigr)^{\!\alpha},
\]
the last inequality from the Stirling bound $\alpha!\ge(\alpha/e)^\alpha$. The Taylor remainder of $\gamma\mapsto\boldsymbol{L}_{k_\gamma}$ of order $n$ about any point of the interval is therefore at most
$\sup_{\gamma}\|\partial_\gamma^{n+1}\boldsymbol{L}_{k_\gamma}\|_{\mathrm{HS}}\,|h|^{n+1}/(n+1)! \le (2/\gamma_L)^{n+1}|h|^{n+1}$,
which tends to zero whenever $|h|<\gamma_L/2$: the map is real-analytic with a radius uniform on $[\gamma_L,\gamma_U]$, and, by the mean-value form of the derivative bound at $\alpha=1$, $\sup_{x,y}|k_{\gamma_1}(x,y)-k_{\gamma_2}(x,y)|\le\frac{|\gamma_1-\gamma_2|}{e\,\gamma_L}$, whence
\[
    \bigl\|\boldsymbol{L}_{k_{\gamma_1}}-\boldsymbol{L}_{k_{\gamma_2}}\bigr\|_{\mathrm{HS}(L^2(\mathbb{P}_X))}
    \;\le\; \frac{|\gamma_1-\gamma_2|}{e\,\gamma_L}
    \qquad\text{for all }\gamma_1,\gamma_2\in[\gamma_L,\gamma_U].
\]
For the eigenpair parameterization, fix $\gamma$; $\boldsymbol{L}_{k_\gamma}$ is self-adjoint and positive on $L^2(\mathbb{P}_X)$; positivity follows from the Fourier representation
$e^{-\gamma\|x-y\|^2} = (2\pi)^{-d/2}\int e^{i\omega\cdot(x-y)}(2\gamma)^{-d/2}e^{-\|\omega\|^2/(4\gamma)}\,d\omega$,
which gives
$\langle\boldsymbol{L}_{k_\gamma}f,f\rangle = (2\pi)^{-d/2}(2\gamma)^{-d/2}\int\big|\int e^{i\omega\cdot x}f(x)\,d\mathbb{P}_X(x)\big|^2 e^{-\|\omega\|^2/(4\gamma)}\,d\omega \ge 0$.
The operator is Hilbert--Schmidt, hence compact, with $\operatorname{tr}\boldsymbol{L}_{k_\gamma} = \int k_\gamma(x,x)\,d\mathbb{P}_X(x) = 1 < \infty$, so its nonzero spectrum consists of positive eigenvalues of finite multiplicity accumulating only at zero, and the enumeration is infinite: the operator has infinite rank, because for every support point $x_0\in\operatorname{supp}\mathbb{P}_X$ the section $k_\gamma(\cdot,x_0)$ lies in the closure of the range, being the dominated-convergence limit of $\boldsymbol{L}_{k_\gamma}f_n$ for the normalized indicators of the balls of radius $1/n$ about $x_0$, of positive mass since $x_0$ is a support point; if that closure had dimension at most $r$, then $r+1$ distinct support points $x_1,\dots,x_{r+1}$ (which exist because $\operatorname{supp}\mathbb{P}_X$ is infinite) would admit coefficients $c\ne0$ with $\sum_jc_jk_\gamma(\cdot,x_j)=0$ in $L^2(\mathbb{P}_X)$, since $r+1$ vectors in a space of dimension at most $r$ are linearly dependent; but this combination is continuous and vanishes $\mathbb{P}_X$-a.s., hence vanishes at every support point and in particular at each $x_j$, so $c$ annihilates the Gram matrix $(k_\gamma(x_j,x_k))_{j,k}$, which is strictly positive definite for the Gaussian kernel on distinct points, whence $c=0$, a contradiction. Analyticity makes $\gamma\mapsto\boldsymbol{L}_{k_\gamma}$ a self-adjoint holomorphic family of type (a) in the real parameter $\gamma$. By Rellich's theorem for such families \citep[Ch.\ II, \S6]{kato1980}, locally around every parameter value, each isolated cluster of eigenvalues of finite total multiplicity admits finitely many real-analytic eigenvalue branches with real-analytic orthonormal eigenfunction branches spanning the cluster's eigenspace; the reduction to that theorem is finite-dimensional: positivity and the strict positivity of the individual eigenvalues at the parameter value isolate each finite cluster from the spectrum, which accumulates only at zero, and the spectral projection onto the isolated cluster \citep[Ch.\ IV, \S3]{kato1980} is norm-continuous in $\gamma$, in fact real-analytic, and has finite rank, so that the compression of the family to its range is a finite-dimensional self-adjoint analytic family to which the Rellich theorem applies verbatim. The ordered eigenvalues of a norm-continuous self-adjoint family are continuous and the spectral projections onto isolated groups of eigenvalues are norm-continuous \citep[Ch.\ IV, \S3]{kato1980}, so the locally analytic branches order consistently and glue to continuous maps on the whole interval. For the eigenfunctions, two distinct real-analytic branches coincide identically only within a constant multiplicity cluster and otherwise cross at most finitely often on the compact interval (their difference is analytic and not identically zero), and at each crossing the adjacent charts' frames span the same crossing eigenspace, in which they differ by a unitary mixing; finitely many such patchings produce continuous $L^2$-valued orthonormal eigenfunction selections. This proves the proposition.
\end{proof}

\begin{lemma}[Fourier isomorphism for the Gaussian kernel]\label{lem:fourier_iso}
Fix $\gamma_0 \in [\gamma_L,\gamma_U]$. With the Fourier convention $\hat{f}(\omega) = (2\pi)^{-d/2}\int_{\mathbb{R}^d} f(x)e^{-i\omega\cdot x}dx$, the Gaussian kernel satisfies \citep[Theorem 10.12]{wendland2004scattered}
\[
    \lVert f\rVert_{\mathcal{H}_\gamma}^2 = c_\gamma\int_{\mathbb{R}^d} |\hat{f}(\omega)|^2 e^{\lVert\omega\rVert^2/(4\gamma)}\,d\omega,
    \qquad c_\gamma := (2\pi)^{-d/2}(2\gamma)^{d/2}.
\]
The map $J_\gamma : \mathcal{H}_\gamma \to \mathcal{H}_{\gamma_0}$ defined through
\[
    \widehat{J_\gamma f}(\omega) := \sqrt{c_\gamma/c_{\gamma_0}}\; \hat{f}(\omega)\exp\Bigl(\frac{\lVert\omega\rVert^2}{8}\Bigl(\frac{1}{\gamma}-\frac{1}{\gamma_0}\Bigr)\Bigr)
\]
is an isometric isomorphism: $\lVert J_\gamma f\rVert_{\mathcal{H}_{\gamma_0}} = \lVert f\rVert_{\mathcal{H}_\gamma}$. Equivalently, $J_\gamma$ is the Fourier multiplier
\[
    J_\gamma = \mathcal{F}^{-1}w_\gamma\mathcal{F}, \qquad \text{i.e.}\quad \widehat{J_\gamma f}(\omega) = w_\gamma(\omega)\hat f(\omega),
\]
where $\mathcal{F}$ and $\mathcal{F}^{-1}$ denote the Fourier transform on $L^2(\mathbb{R}^d)$ (in the convention displayed above) and its inverse, and
\[
    w_\gamma(\omega) := \Bigl(\frac{s_{\gamma_0}(\omega)}{s_\gamma(\omega)}\Bigr)^{1/2},
    \qquad
    s_\gamma(\omega) := \hat{k}_\gamma(\omega) = (2\gamma)^{-d/2}e^{-\lVert\omega\rVert^2/(4\gamma)}
\]
is the square root of the ratio of the \emph{spectral densities} $s_\gamma$ of the kernels $k_\gamma$; by the shift rule, the Fourier transform of $k_\gamma(x,\cdot)$ at frequency $\omega$ is $e^{-i\omega\cdot x}s_\gamma(\omega)$. Consequently, the transported kernel $\tilde{k}_\gamma(x,y) := \langle J_\gamma k_\gamma(\cdot,x), J_\gamma k_\gamma(\cdot,y)\rangle_{\mathcal{H}_{\gamma_0}}$ coincides with $k_\gamma(x,y)$.
\end{lemma}

\begin{proof}
The proof proceeds in four steps: the norm formula, the isometry, the multiplier form, and the transported kernel.

\paragraph*{The norm formula.}\; The kernel $k_\gamma$ is translation invariant, and its Fourier transform is the Gaussian
\[
    \hat{k}_\gamma(\omega) = (2\pi)^{-d/2}\int_{\mathbb{R}^d}e^{-\gamma\lVert x\rVert^2}e^{-i\omega\cdot x}\,dx = (2\gamma)^{-d/2}e^{-\lVert\omega\rVert^2/(4\gamma)},
\]
using $\int_{\mathbb{R}^d}e^{-a\lVert x\rVert^2}e^{-i\omega\cdot x}dx = (\pi/a)^{d/2}e^{-\lVert\omega\rVert^2/(4a)}$ at $a = \gamma$. Wendland (2004, Theorem~10.12) gives the native-space formula $\lVert f\rVert_{\mathcal{H}_\gamma}^2 = (2\pi)^{-d/2}\int|\hat f(\omega)|^2/\hat{k}_\gamma(\omega)\,d\omega$; since $1/\hat{k}_\gamma(\omega) = (2\gamma)^{d/2}e^{\lVert\omega\rVert^2/(4\gamma)}$, this is the displayed formula $\lVert f\rVert_{\mathcal{H}_\gamma}^2 = c_\gamma\int|\hat f(\omega)|^2e^{\lVert\omega\rVert^2/(4\gamma)}d\omega$ with $c_\gamma = (2\pi)^{-d/2}(2\gamma)^{d/2}$.

\paragraph*{Isometry.}\; Apply the norm formula at $\gamma_0$ to $J_\gamma f$ and substitute the definition of $J_\gamma$:
\[
\begin{aligned}
    \lVert J_\gamma f\rVert_{\mathcal{H}_{\gamma_0}}^2
    &= c_{\gamma_0}\int_{\mathbb{R}^d} \frac{c_\gamma}{c_{\gamma_0}}|\hat f(\omega)|^2
       \exp\Bigl(\frac{\lVert\omega\rVert^2}{4\gamma_0} + \frac{\lVert\omega\rVert^2}{4}\Bigl(\frac{1}{\gamma}-\frac{1}{\gamma_0}\Bigr)\Bigr)\,d\omega \\
    &= c_\gamma\int_{\mathbb{R}^d}|\hat f(\omega)|^2 e^{\lVert\omega\rVert^2/(4\gamma)}\,d\omega = \lVert f\rVert_{\mathcal{H}_\gamma}^2.
\end{aligned}
\]
Hence $J_\gamma$ is an isometry. The swapped map is the inverse. Define $J_{\gamma_0}$ by the identical formula with $\gamma$ and $\gamma_0$ interchanged; it maps $\mathcal H_{\gamma_0}$ to $\mathcal H_\gamma$ by the same computation as above with the two parameters swapped, and its multiplier is
\[
    w_{\gamma_0}(\omega)=\Bigl(\frac{s_\gamma(\omega)}{s_{\gamma_0}(\omega)}\Bigr)^{1/2}
    =\frac{1}{w_\gamma(\omega)} :
\]
the square root of the reciprocal ratio is the reciprocal of the square root, because both densities are strictly positive. Multipliers compose pointwise, so for every $f\in\mathcal H_\gamma$,
\[
\begin{aligned}
    \widehat{J_{\gamma_0}J_\gamma f}(\omega)
    &=w_{\gamma_0}(\omega)\,w_\gamma(\omega)\,\hat f(\omega)
    =\hat f(\omega),\qquad\text{and symmetrically}\\
    \widehat{J_\gamma J_{\gamma_0} g}(\omega)&=\hat g(\omega),\quad g\in\mathcal H_{\gamma_0}.
\end{aligned}
\]
The Fourier transform being injective on $L^2$, both compositions are the identity on their respective spaces: $J_{\gamma_0}=J_\gamma^{-1}$, so $J_\gamma$ is onto, and a surjective isometry is an isometric isomorphism $\mathcal H_\gamma\to\mathcal H_{\gamma_0}$.

\paragraph*{The multiplier form.}\; The operator identity $J_\gamma=\mathcal F^{-1}w_\gamma\mathcal F$ asserts that $J_\gamma$ acts as the three-step pipeline
\emph{transform $f$, multiply pointwise by the scalar function $w_\gamma(\omega)$, transform back}.
The verification is immediate from the definition of $J_\gamma$ through its transform: substituting the explicit form of the ratio of the spectral densities computed above,
$\frac{s_{\gamma_0}}{s_\gamma}=(\frac{\gamma}{\gamma_0})^{d/2}e^{\lVert\omega\rVert^2(1/\gamma-1/\gamma_0)/4}$,
together with $\sqrt{c_\gamma/c_{\gamma_0}}=(\gamma/\gamma_0)^{d/4}$, turns the defining display into
\[
    (J_\gamma f)\widehat{\ }(\omega)
    =\Bigl(\frac{\gamma}{\gamma_0}\Bigr)^{d/4}e^{\lVert\omega\rVert^2(1/\gamma-1/\gamma_0)/8}\,\hat f(\omega)
    =w_\gamma(\omega)\,\hat f(\omega).
\]
The right-hand side is also the Fourier transform of $\mathcal F^{-1}[w_\gamma\cdot\mathcal F f]$; since the Fourier transform is injective on $L^2$, it follows that
$J_\gamma f=\mathcal F^{-1}[w_\gamma\,\mathcal F f]$ for every $f$, i.e., $J_\gamma=\mathcal F^{-1}w_\gamma\mathcal F$. By the shift rule, the transported feature function has transform
$(J_\gamma k_\gamma(x,\cdot))\widehat{\ }(\omega)=w_\gamma(\omega)s_\gamma(\omega)e^{-i\omega\cdot x}=e^{-i\omega\cdot x}\sqrt{s_{\gamma_0}(\omega)s_\gamma(\omega)}$,
the geometric mean of the two spectral densities.

\paragraph*{The transported kernel.}\; By the shift rule and the multiplier form, the Fourier transform of $J_\gamma k_\gamma(\cdot,x)$ is $w_\gamma(\omega)s_\gamma(\omega)e^{-i\omega\cdot x}$, and $\overline{\widehat{J_\gamma k_\gamma(\cdot,y)}}(\omega) = w_\gamma(\omega)s_\gamma(\omega)e^{i\omega\cdot y}$ (real-valued functions). Using the polarization of the norm formula at $\gamma_0$ and the isometry condition $w_\gamma^2s_\gamma = s_{\gamma_0}$,
\[
\begin{aligned}
    \tilde{k}_\gamma(x,y)
    &= c_{\gamma_0}\int_{\mathbb{R}^d} w_\gamma(\omega)^2 s_\gamma(\omega)^2\, e^{\lVert\omega\rVert^2/(4\gamma_0)}\, e^{i\omega\cdot(x-y)}\,d\omega \\
    &= c_{\gamma_0}\int_{\mathbb{R}^d} s_{\gamma_0}(\omega)s_\gamma(\omega)\, e^{\lVert\omega\rVert^2/(4\gamma_0)}\, e^{i\omega\cdot(x-y)}\,d\omega \\
    &= (2\pi)^{-d/2}\int_{\mathbb{R}^d} s_\gamma(\omega)\,e^{i\omega\cdot(x-y)}\,d\omega = k_\gamma(x,y),
\end{aligned}
\]
where the third equality uses $c_{\gamma_0}s_{\gamma_0}(\omega)e^{\lVert\omega\rVert^2/(4\gamma_0)} = (2\pi)^{-d/2}$ (the constants of $c_{\gamma_0}$ and $s_{\gamma_0}$ cancel exactly), and the last equality is the inversion formula $k_\gamma(x,y) = (2\pi)^{-d/2}\int_{\mathbb{R}^d}s_\gamma(\omega)e^{i\omega\cdot(x-y)}d\omega$, i.e.\ the Fourier transform of the Gaussian $s_\gamma$ evaluated at $x-y$.
\end{proof}

\begin{lemma}[Cross-kernel of two transported features]\label{lem:crosskernel}
For $\gamma_1,\gamma_2\in[\gamma_L,\gamma_U]$ and $x,y\in\mathbb R^d$, define the cross-kernel
$K_{\gamma_1,\gamma_2}(x,y):=\langle J_{\gamma_1}k_{\gamma_1}(\cdot,x),\,J_{\gamma_2}k_{\gamma_2}(\cdot,y)\rangle_{\mathcal H_{\gamma_0}}$,
where $J_\gamma:\mathcal H_\gamma\to\mathcal H_{\gamma_0}$ is the transport of Lemma~\ref{lem:fourier_iso}. Then
\begin{equation}\label{eq:crosskernel_closed}
    K_{\gamma_1,\gamma_2}(x,y)=c_{12}\,k_{\gamma_{12}}(x,y)
    =\frac{2^{d/2}(\gamma_1\gamma_2)^{d/4}}{(\gamma_1+\gamma_2)^{d/2}}\,
      \exp\Bigl(-\frac{2\gamma_1\gamma_2}{\gamma_1+\gamma_2}\,\lVert x-y\rVert^2\Bigr),
\end{equation}
a Gaussian kernel at the harmonic-mean bandwidth $\gamma_{12}:=2\gamma_1\gamma_2/(\gamma_1+\gamma_2)\in[\gamma_L,\gamma_U]$, scaled by the constant
$c_{12}:=2^{d/2}(\gamma_1\gamma_2)^{d/4}/(\gamma_1+\gamma_2)^{d/2}$,
which satisfies $c_{12}\in[c_{\min},1]$ with $c_{\min}:=2^{d/2}(\gamma_L\gamma_U)^{d/4}/(\gamma_L+\gamma_U)^{d/2}>0$, equals $1$ if and only if $\gamma_1=\gamma_2$, and depends on neither $x$ nor $y$ nor the reference bandwidth $\gamma_0$.
\end{lemma}

\begin{proof}
Write $\tilde u(x;\gamma):=J_\gamma k_\gamma(\cdot,x)$. By the shift rule and the multiplier form of the transport, $J_\gamma=\mathcal F^{-1}w_\gamma\mathcal F$ with $w_\gamma=(s_{\gamma_0}/s_\gamma)^{1/2}$ and $s_\gamma(\omega)=(2\gamma)^{-d/2}e^{-\lVert\omega\rVert^2/(4\gamma)}$ the spectral density of $k_\gamma$ in the unitary convention,
\[
    \widehat{\tilde u(x;\gamma)}(\omega)=w_\gamma(\omega)s_\gamma(\omega)e^{-i\omega\cdot x}
    =\sqrt{s_{\gamma_0}(\omega)s_\gamma(\omega)}\,e^{-i\omega\cdot x}.
\]
The norm formula of Lemma~\ref{lem:fourier_iso}, $c_{\gamma_0}e^{\lVert\omega\rVert^2/(4\gamma_0)}=(2\pi)^{-d/2}/s_{\gamma_0}(\omega)$, polarizes to
$\langle f,g\rangle_{\mathcal H_{\gamma_0}}=(2\pi)^{-d/2}\int\hat f(\omega)\overline{\hat g(\omega)}/s_{\gamma_0}(\omega)\,d\omega$.
Substituting the two transformed features and using the cancellation
$w_{\gamma_1}(\omega)w_{\gamma_2}(\omega)s_{\gamma_1}(\omega)s_{\gamma_2}(\omega)/s_{\gamma_0}(\omega)
=\sqrt{s_{\gamma_1}(\omega)s_{\gamma_2}(\omega)}$,
which records that the reference density entering through the two multipliers cancels the reference
density of the inner-product weight exactly; the result is
\begin{equation}
    K_{\gamma_1,\gamma_2}(x,y)
    =(2\pi)^{-d/2}\int_{\mathbb R^d}\sqrt{s_{\gamma_1}(\omega)s_{\gamma_2}(\omega)}\;e^{-i\omega\cdot(x-y)}\,d\omega ,
\end{equation}
already free of $\gamma_0$. For the Gaussian family the geometric mean of two spectral densities is again a spectral density, up to scaling:
$\sqrt{s_{\gamma_1}(\omega)s_{\gamma_2}(\omega)}
=(4\gamma_1\gamma_2)^{-d/4}e^{-\lVert\omega\rVert^2(1/\gamma_1+1/\gamma_2)/8}
=(4\gamma_1\gamma_2)^{-d/4}e^{-\lVert\omega\rVert^2/(4\gamma_{12})}
=c_{12}\,s_{\gamma_{12}}(\omega)$,
where $(1/\gamma_1+1/\gamma_2)/8=(1/\gamma_{12})/4$ by $1/\gamma_{12}=(1/\gamma_1+1/\gamma_2)/2$, and the prefactor checks as $c_{12}(2\gamma_{12})^{-d/2}=(4\gamma_1\gamma_2)^{-d/4}$. Inverting the transform by $(2\pi)^{-d/2}\int s_{\gamma_{12}}(\omega)e^{-i\omega\cdot t}d\omega=k_{\gamma_{12}}(t)$ yields \eqref{eq:crosskernel_closed}. Finally, $c_{12}\le1$ by AM--GM with equality iff $\gamma_1=\gamma_2$, and $c_{12}\ge c_{\min}>0$ by the explicit form $c_{12}=(1-x^2)^{d/4}$, read off the closed form \eqref{eq:crosskernel_closed}, with $x:=(\gamma_2-\gamma_1)/(\gamma_1+\gamma_2)$: on the square the smaller coordinate is at least $\gamma_L$ and the coordinate gap at most $\gamma_U-\gamma_L$, so
$|x|=\frac{|\gamma_2-\gamma_1|}{\gamma_1+\gamma_2}\le\frac{|\gamma_2-\gamma_1|}{2\gamma_L+|\gamma_2-\gamma_1|}\le\frac{\gamma_U-\gamma_L}{\gamma_L+\gamma_U}$,
the last fraction being increasing in the gap and attained at the corners $(\gamma_L,\gamma_U)$ and $(\gamma_U,\gamma_L)$; therefore $1-x^2\ge1-\bigl(\frac{\gamma_U-\gamma_L}{\gamma_U+\gamma_L}\bigr)^2=\frac{4\gamma_L\gamma_U}{(\gamma_L+\gamma_U)^2}$ and
$c_{12}=(1-x^2)^{d/4}\ge\Bigl(\frac{4\gamma_L\gamma_U}{(\gamma_L+\gamma_U)^2}\Bigr)^{d/4}=c_{\min}$;
and $\gamma_{12}\in[\gamma_L,\gamma_U]$ because a harmonic mean of points of an interval remains in the interval.
\end{proof}

The following maximal inequality is used in the proofs of Theorems~\ref{th:stat_uniform} and~\ref{th:cv_uniform}: it converts a uniform pointwise gradient moment into a within-cell oscillation bound, with no discretization, no grids, and no two-point increment bounds.

\begin{lemma}[Sobolev maximal inequality, within-cell form]\label{lem:sobolev}
Let $(X_t)_{t\in Q}$ be a real-valued random field on an axis-aligned cube $Q\subset\mathbb R^d$ of side length $\delta>0$, with $C^{1}$ paths and Euclidean norm $\|\cdot\|$, and let $K$ be a subset of $Q$. Suppose $p>d$ and
\begin{equation}\label{eq:sobolev-input}
    \sup_{t\in Q}\mathbb E\|\nabla X_t\|^{p}\le N_p<\infty.
\end{equation}
Then for every $\varepsilon>0$,
\begin{equation}\label{eq:sobolev-max}
    \mathbb E\Bigl[\sup_{t,s\in K}|X_t-X_s|^{p}\Bigr]\le C_{d,p}\,N_p\,\delta^{p},
    \qquad\text{hence}\qquad
    \mathbb P\Bigl(\sup_{t,s\in K}|X_t-X_s|\ge\varepsilon\Bigr)\le C_{d,p}\,N_p\,\delta^{p}\,\varepsilon^{-p},
\end{equation}
where $C_{d,p}$ is a universal constant (Morrey's constant), independent of the field, the law, and the set $K$.
\end{lemma}

\begin{proof}
Fix $\omega$ and put $f:=X(\omega)$. Morrey's inequality ($W^{1,p}(Q)\hookrightarrow C^{0,1-d/p}(Q)$ for $p>d$; a cube of side $\delta$) gives
\[
    \sup_{t,s\in Q}|f(t)-f(s)|\le C_{d,p}\,\delta^{1-d/p}\,\Bigl(\int_Q\|\nabla f\|^{p}\Bigr)^{1/p},
\]
and the supremum over the subset $K$ is bounded by that over $Q$; the constant depends only on $(d,p)$ (side-independent by scaling: for $g(y)=f(\delta y)$ on the unit cube, $\lVert\nabla g\rVert_{p}=\delta^{1-d/p}\lVert\nabla f\rVert_{L^{p}(Q)}$ and the suprema coincide). Raise to the power $p$ and take expectations: each $X_t$ is a random variable and $\nabla X_t$ is a pointwise limit of measurable difference quotients, so $(t,\omega)\mapsto\lVert\nabla X_t(\omega)\rVert^{p}$ is jointly measurable (continuous $t$-sections, measurable $\omega$-sections), as is $\sup_{t,s\in K}\lvert X_t-X_s\rvert^{p}$ (a countable dense subset of $K$ suffices, the paths being continuous), and Tonelli and the uniform bound \eqref{eq:sobolev-input} give
\[
    \mathbb E\Bigl[\sup_{t,s\in K}|X_t-X_s|^{p}\Bigr]
    \le C_{d,p}^{p}\,\delta^{p-d}\int_Q\mathbb E\|\nabla X_t\|^{p}\,dt
    \le C_{d,p}^{p}\,\delta^{p-d}\cdot\delta^{d}\cdot N_p
    = C_{d,p}^{p}\,N_p\,\delta^{p},
\]
since $\operatorname{vol}(Q)=\delta^{d}$; Markov's inequality yields the probability form, and the weak event "$\ge\varepsilon$" follows from the strict one by the monotone limit $\varepsilon'\uparrow\varepsilon$.
\end{proof}

The following containment lemma supplies the geometry by which a within-cell oscillation bound, such as \eqref{eq:sobolev-max}, becomes a bound on the supremum over all pairs of tuning parameters at distance at most $\delta$; its point is that the discretization is chosen after the pair, so the increment over a short pair is bounded by a single cell oscillation with no multiplicity factor.

\begin{lemma}[Shifted grids contain short pairs]\label{lem:shifts}
Let $L_{1},L_{2}$ be the side lengths of $\Upsilon$, let $0<\delta\le\min(L_{1},L_{2})/2$, put $h:=2\delta$, and for $\boldsymbol{\tau}=(\tau_{1},\tau_{2})\in\{0,h/2\}^{2}$ let $\mathcal G^{(\boldsymbol{\tau})}_{\delta}$ consist of the sets
\[
    \Upsilon\cap\bigl\{\bigl[\gamma_{L}+\tau_{1}+kh,\ \gamma_{L}+\tau_{1}+(k+1)h\bigr]\times\bigl[\lambda_{L}+\tau_{2}+lh,\ \lambda_{L}+\tau_{2}+(l+1)h\bigr]\bigr\},
    \qquad k,l\in\mathbb Z,
\]
and let $\mathcal G_{\delta}:=\bigcup_{\boldsymbol{\tau}}\mathcal G^{(\boldsymbol{\tau})}_{\delta}$. Then
\begin{enumerate}[label=(\roman*)]
\item every $K\in\mathcal G_{\delta}$ is contained in $\Upsilon$ and is contained in a square of side $h$ that is contained in $\Upsilon$, and $\lvert\mathcal G_{\delta}\rvert\le4\,(\lfloor L_{1}/h\rfloor+2)\,(\lfloor L_{2}/h\rfloor+2)\le C_{\Upsilon}\,\delta^{-2}$;
\item if $\lVert\eta-\eta'\rVert\le\delta$, then $\eta$ and $\eta'$ lie in a single element of $\mathcal G^{(\boldsymbol{\tau})}_{\delta}$ for some $\boldsymbol{\tau}\in\{0,h/2\}^{2}$; consequently
\[
    \sup_{\lVert\eta-\eta'\rVert\le\delta}\bigl\lvert S(\eta)-S(\eta')\bigr\rvert
    \le\max_{K\in\mathcal G_{\delta}}\ \sup_{t,s\in K}\bigl\lvert S(t)-S(s)\bigr\rvert
\]
for every real-valued field $S$ on $\Upsilon$.
\end{enumerate}
\end{lemma}

\begin{proof}
For one coordinate, the two families of lines $\{\gamma_{L}+kh\}_{k\in\mathbb Z}$ and $\{\gamma_{L}+h/2+kh\}_{k\in\mathbb Z}$ are the even and the odd points of the lattice $\gamma_{L}+(h/2)\mathbb Z$, so a point of the first family and a point of the second are at distance at least $h/2$. Let $I_{1}$ be the interval spanned by the two $\gamma$-coordinates; it has length at most $\lVert\eta-\eta'\rVert\le\delta=h/2$. If both families had a line strictly inside $I_{1}$, those two lines would be at distance at least $h/2$ while both lie strictly inside an interval of length at most $h/2$, which is impossible. Hence at least one $\tau_{1}\in\{0,h/2\}$ has no line of its family strictly inside $I_{1}$, and the same argument gives $\tau_{2}\in\{0,h/2\}$ for the $\lambda$-coordinate. With $\boldsymbol{\tau}=(\tau_{1},\tau_{2})$ neither coordinate interval contains a line of the corresponding family, so both points lie in one element of $\mathcal G^{(\boldsymbol{\tau})}_{\delta}$, which is (ii). For (i): cells are contained in $\Upsilon$ by construction; a cell truncated at the boundary of $\Upsilon$ is contained in the square of side $h$ obtained by extending the cell away from that boundary, since a closed interval of length $w\le h$ inside an interval of length $L\ge h$ lies in the subinterval $\bigl[\min(\alpha,L-h),\,\min(\alpha,L-h)+h\bigr]$ determined by its own left endpoint $\alpha$, clamped at $L-h$, and the two-coordinate version of the clamp applies coordinatewise to corner cells; that square lies in $\Upsilon$ because $h\le\min(L_{1},L_{2})$; the cells are the product cells of the two coordinate grids, each coordinate contributing at most $\lfloor L/h\rfloor+2$ cells per family (the two boundary cells being degenerate), whence $\lvert\mathcal G_{\delta}\rvert\le4(\lfloor L_{1}/h\rfloor+2)(\lfloor L_{2}/h\rfloor+2)\le C_{\Upsilon}\,\delta^{-2}$. Here and below $C_{\Upsilon}$ denotes a constant depending only on $\Upsilon$.
\end{proof}

\begin{lemma}[Equicontinuity and tightness from a uniform gradient moment]\label{lem:equi}
Let $(S_{m})$ be real-valued fields on $\Upsilon$ and suppose that for every $m$ there is a decomposition $S_{m}=S^{c}_{m}+a_{m}+e_{m}$ in which $S^{c}_{m}$ has $C^{1}$ paths and such that
\begin{enumerate}[label=(\roman*)]
\item there exist $p\in(2,4]$ and $N_{p}<\infty$ such that $\sup_{m}\sup_{\eta\in\Upsilon}\mathbb E\lVert\nabla_{\eta}S^{c}_{m}(\eta)\rVert^{p}\le N_{p}$;
\item there are a deterministic $\Lambda<\infty$ and random variables $r_{m}=o_{p}(1)$ such that $\lvert a_{m}(\eta)-a_{m}(\eta')\rvert\le\Lambda\lVert\eta-\eta'\rVert+2r_{m}$ for all $\eta,\eta'\in\Upsilon$;
\item $\sup_{\eta\in\Upsilon}\lvert e_{m}(\eta)\rvert=o_{p}(1)$.
\end{enumerate}
Then for every $\varepsilon>0$,
\[
    \lim_{\delta\downarrow0}\ \limsup_{m\to\infty}\ \mathbb P\Bigl(\sup_{\lVert\eta-\eta'\rVert\le\delta}\bigl\lvert S_{m}(\eta)-S_{m}(\eta')\bigr\rvert>\varepsilon\Bigr)=0 .
\]
If in addition $S_{m}(\eta_{0})=O_{p}(1)$ for some fixed $\eta_{0}\in\Upsilon$, then $\sup_{\eta\in\Upsilon}\lvert S_{m}(\eta)\rvert=O_{p}(1)$, and consequently $(S_{m})$ is asymptotically tight in $\ell^{\infty}(\Upsilon)$, so that finite-dimensional convergence of $(S_{m})$ implies weak convergence in $\ell^{\infty}(\Upsilon)$.
\end{lemma}

\begin{proof}
Fix $0<\delta\le\min(L_{1},L_{2})/2$, put $h:=2\delta$, and let $\mathcal G_{\delta}$ be as in Lemma~\ref{lem:shifts}. For $\lVert\eta-\eta'\rVert\le\delta$, Lemma~\ref{lem:shifts}(ii) applied to $S^{c}_{m}$ and the triangle inequality give
\[
    \bigl\lvert S_{m}(\eta)-S_{m}(\eta')\bigr\rvert
    \le\max_{K\in\mathcal G_{\delta}}\ \sup_{t,s\in K}\bigl\lvert S^{c}_{m}(t)-S^{c}_{m}(s)\bigr\rvert
      +\Lambda\,\delta+2r_{m}+2\sup_{\eta}\lvert e_{m}(\eta)\rvert .
\]
By Lemma~\ref{lem:shifts}(i) each $K$ is contained in a square of side $h$ contained in $\Upsilon$, on which Lemma~\ref{lem:sobolev} applies at the exponent $p\in(2,4]$, $p>d=\dim\Upsilon=2$; hypothesis (i) therefore gives
\[
    \mathbb E\Bigl[\sup_{t,s\in K}\bigl\lvert S^{c}_{m}(t)-S^{c}_{m}(s)\bigr\rvert^{p}\Bigr]
    \le C_{2,p}\,N_{p}\,h^{p}=2^{p}\,C_{2,p}\,N_{p}\,\delta^{p}
\]
for every $K$. Markov's inequality at level $\varepsilon/2$ and the union bound over the at most $C_{\Upsilon}\delta^{-2}$ cells of $\mathcal G_{\delta}$ yield
\[
    \mathbb P\Bigl(\max_{K\in\mathcal G_{\delta}}\sup_{t,s\in K}\bigl\lvert S^{c}_{m}(t)-S^{c}_{m}(s)\bigr\rvert>\tfrac{\varepsilon}{2}\Bigr)
    \le\lvert\mathcal G_{\delta}\rvert\,\frac{2^{p}C_{2,p}N_{p}\delta^{p}}{(\varepsilon/2)^{p}}
    \le\frac{2^{2p}\,C_{\Upsilon}C_{2,p}N_{p}\,\delta^{p-2}}{\varepsilon^{p}},
\]
which is independent of $m$ and vanishes as $\delta\downarrow0$, the exponent $p-2>0$ being exactly where $p>2$ enters; and for each fixed $\delta$ with $\Lambda\delta<\varepsilon/2$,
\[
    \mathbb P\Bigl(\Lambda\delta+2r_{m}+2\sup_{\eta}\lvert e_{m}(\eta)\rvert>\tfrac{\varepsilon}{2}\Bigr)\longrightarrow0
    \qquad(m\to\infty)
\]
by (ii) and (iii), the deterministic $\Lambda\delta$ being absorbed in the level. Hence the first display.

For the boundedness claim, chain along a fixed net: let $\delta_{0}:=\min(L_{1},L_{2})/2$ and let $\eta=\zeta_{0},\zeta_{1},\dots,\zeta_{J}=\eta'$ be points of $\Upsilon$ with $\lVert\zeta_{j}-\zeta_{j-1}\rVert\le\delta_{0}$ and $J\le\lceil\operatorname{diam}\Upsilon/\delta_{0}\rceil$, a deterministic constant. Then
\[
    \sup_{\eta,\eta'\in\Upsilon}\bigl\lvert S_{m}(\eta)-S_{m}(\eta')\bigr\rvert
    \le J\max_{\lVert t-s\rVert\le\delta_{0}}\bigl\lvert S_{m}(t)-S_{m}(s)\bigr\rvert ,
\]
and the display above at the fixed scale $\delta_{0}$ bounds the right-hand side in probability: it gives
\[
    \max_{\lVert t-s\rVert\le\delta_{0}}\bigl\lvert S_{m}(t)-S_{m}(s)\bigr\rvert
    \le\max_{K\in\mathcal G_{\delta_{0}}}\ \sup_{t,s\in K}\bigl\lvert S^{c}_{m}(t)-S^{c}_{m}(s)\bigr\rvert
      +\Lambda\,\delta_{0}+2r_{m}+2\sup_{\eta}\lvert e_{m}(\eta)\rvert ,
\]
whose first term is $O_{p}(N_{p}^{1/p}h)=O_{p}(N_{p}^{1/p}\delta_{0})$, the cells of $\mathcal G_{\delta_{0}}$ being finitely many and the bound $2^{p}C_{2,p}N_{p}\delta_{0}^{p}$ of the display above holding uniformly over them, and whose remaining terms are deterministic or $o_{p}(1)$, so the right-hand side is $O_{p}(1)$ and hence $\sup_{\eta,\eta'}\lvert S_{m}(\eta)-S_{m}(\eta')\rvert=O_{p}(1)$. Therefore $\sup_{\eta}\lvert S_{m}(\eta)\rvert\le\lvert S_{m}(\eta_{0})\rvert+\sup_{\eta,\eta'}\lvert S_{m}(\eta)-S_{m}(\eta')\rvert=O_{p}(1)$ by the pointwise hypothesis. Asymptotic tightness in $\ell^{\infty}(\Upsilon)$ follows from asymptotic equicontinuity together with boundedness and the total boundedness of $\Upsilon$ (the sufficiency direction of Theorem~1.5.7 there), and Theorem~1.5.4 of \cite{vaart1997weak} then converts finite-dimensional convergence and asymptotic tightness into weak convergence in $\ell^{\infty}(\Upsilon)$.
\end{proof}

The following fact records the derivative-feature calculus of the Gaussian transport: the feature difference of a bandwidth pair, normalized by the bandwidth gap, is controlled, uniformly over all coefficient vectors, by the supremum of a one-dimensional derivative-feature process, so that the two-dimensional pair square is reduced to the bandwidth interval.

\begin{fact}[Derivative features of the Gaussian transport]\label{fact:pairdiff}
Let $\tilde u_\gamma(x):=J_\gamma k_\gamma(x,\cdot)\in\mathcal H_{\gamma_0}$, $\gamma\in[\gamma_L,\gamma_U]$. Then the following hold.
\emph{(i)} The map $\gamma\mapsto\tilde u_\gamma(x)$ is continuously differentiable on $[\gamma_L,\gamma_U]$ for every $x\in\mathbb R^d$, with deterministic constants $L_d,L_{dd}$, depending only on $(\gamma_L,\gamma_U,d)$, such that
$\lVert\partial_\gamma\tilde u_\gamma(x)\rVert_{\mathcal H_{\gamma_0}}\le L_d$ and
$\lVert\partial_\gamma\tilde u_{\gamma_1}(x)-\partial_\gamma\tilde u_{\gamma_2}(x)\rVert_{\mathcal H_{\gamma_0}}\le L_{dd}\,|\gamma_1-\gamma_2|$
for all $x\in\mathbb R^d$ and all $\gamma,\gamma_1,\gamma_2\in[\gamma_L,\gamma_U]$.
\emph{(ii)} For every $n$, every weights $a_1,\dots,a_n$, and every $\gamma_1\ne\gamma_2\in[\gamma_L,\gamma_U]$ with $T:=\gamma_2-\gamma_1$,
\[n^{-1}\sum_{j,k}a_ja_k\,\bar W_{jk}(\gamma_1,\gamma_2)\;\le\;\sup_{\gamma\in[\gamma_L,\gamma_U]}D_{\boldsymbol a}(\gamma),\qquad D_{\boldsymbol a}(\gamma):=\Bigl\lVert\tfrac1{\sqrt n}\sum_{j=1}^na_j\,\partial_\gamma\tilde u_\gamma(X_j)\Bigr\rVert_{\mathcal H_{\gamma_0}}^2,\]
where $\bar{\boldsymbol W}=T^{-2}\Delta\boldsymbol G$ is the normalized feature-difference Gram of Lemma~\ref{lem:lipschitz_common}(a).
\emph{(iii)} Let $\widetilde C_{\gamma}:=\mathbb E\bigl[\tilde u_{\gamma}(X)\otimes\tilde u_{\gamma}(X)\bigr]$ on $\mathcal H_{\gamma_{0}}$, and let $\bar k:=\sup_{\gamma\in[\gamma_{L},\gamma_{U}],\,x\in\mathbb R^{d}}k_{\gamma}(x,x)$. Then $\gamma\mapsto\widetilde C_{\gamma}$ is continuously differentiable in Hilbert--Schmidt norm with $\bigl\lVert\partial_{\gamma}\widetilde C_{\gamma}\bigr\rVert_{\mathrm{HS}}\le2\,\bar k^{1/2}L_{d}$, and consequently $\bigl\lVert\widetilde C_{\gamma_{1}}-\widetilde C_{\gamma_{2}}\bigr\rVert_{\mathrm{HS}}\le2\,\bar k^{1/2}L_{d}\,\lvert\gamma_{1}-\gamma_{2}\rvert$ for all $\gamma_{1},\gamma_{2}\in[\gamma_{L},\gamma_{U}]$.
\end{fact}

\begin{proof}
\emph{(i)} Under the identification of Lemma~\ref{lem:fourier_iso}, $\mathcal H_{\gamma_0}$ is, through the Fourier transform, the weighted space $L^2\bigl(s_{\gamma_0}^{-1}(\omega)(2\pi)^{-d/2}\,d\omega\bigr)$, and by the computation of Lemma~\ref{lem:crosskernel} the transported feature has Fourier transform
\[
    \widehat{\tilde u_\gamma(x)}(\omega)=\sqrt{s_{\gamma_0}(\omega)s_\gamma(\omega)}\;e^{-i\omega\cdot x}=:\sigma_\gamma(\omega)\,e^{-i\omega\cdot x},
\]
with the inner-product formula $\langle f,g\rangle_{\mathcal H_{\gamma_0}}=(2\pi)^{-d/2}\int\hat f\,\overline{\hat g}\,/s_{\gamma_0}$ recorded in the proof of that lemma. Writing $m_\gamma(\omega):=\partial_\gamma\log s_\gamma(\omega)$, so that
\[
    m_\gamma(\omega)=\frac{\lVert\omega\rVert^2}{4\gamma^2}-\frac d{2\gamma},
    \qquad\partial_\gamma\sigma_\gamma(\omega)=\tfrac12\,\sigma_\gamma(\omega)\,m_\gamma(\omega),
\]
by the product and chain rules applied to $s_\gamma(\omega)=(2\gamma)^{-d/2}e^{-\lVert\omega\rVert^2/(4\gamma)}$, the pointwise derivative of the transformed feature is $\partial_\gamma\bigl(\sigma_\gamma e^{-i\omega\cdot x}\bigr)=\frac12\sigma_\gamma m_\gamma e^{-i\omega\cdot x}$. Since $s_\gamma=(2\pi)^{d/2}$ times the $N(0,2\gamma I_d)$ density, $(2\pi)^{-d/2}\int s_\gamma(\omega)\phi(\omega)\,d\omega=\mathbb E\phi(\Omega_\gamma)$ for $\Omega_\gamma\sim N(0,2\gamma I_d)$, whose second and fourth moments are $\mathbb E\lVert\Omega_\gamma\rVert^2=2\gamma d$ and $\mathbb E\lVert\Omega_\gamma\rVert^4=(2\gamma)^2d(d+2)$; the norm formula therefore gives, with $x$ entering only through the unimodular phase,
\[
    \bigl\lVert\partial_\gamma\tilde u_\gamma(x)\bigr\rVert_{\mathcal H_{\gamma_0}}^2
    =(2\pi)^{-d/2}\int s_\gamma\,\frac{m_\gamma^2}4\,d\omega
    =\tfrac14\,\mathbb E\Bigl(\frac{\lVert\Omega_\gamma\rVert^2}{4\gamma^2}-\frac d{2\gamma}\Bigr)^{2}
    =\frac{d}{8\gamma^2},
\]
the last equality by expanding the square and inserting the two moments. The same computation applies at every $\gamma\in[\gamma_L,\gamma_U]$, so $L_d:=\sqrt d\,/\bigl(2\sqrt2\,\gamma_L\bigr)$ bounds the derivative uniformly. For the derivative of the derivative, one further differentiation gives
\[
    \partial_\gamma^2\sigma_\gamma(\omega)=\tfrac12\,\sigma_\gamma(\omega)\Bigl(\dot m_\gamma(\omega)+\tfrac12m_\gamma(\omega)^2\Bigr),
    \qquad\dot m_\gamma(\omega):=\partial_\gamma m_\gamma(\omega)=-\frac{\lVert\omega\rVert^2}{2\gamma^3}+\frac d{2\gamma^2},
\]
and hence, by the inequality $(u+v)^2\le2u^2+2v^2$ and the Cauchy--Schwarz inequality for the cross term,
\[
    \bigl\lVert\partial_\gamma^2\tilde u_\gamma(x)\bigr\rVert_{\mathcal H_{\gamma_0}}^2
    =\tfrac14\,\mathbb E\Bigl(\dot m_\gamma+\tfrac12m_\gamma^2\Bigr)^{2}
    \le\tfrac14\Bigl(2\,\mathbb E\dot m_\gamma^2+\tfrac12\,\mathbb Em_\gamma^4\Bigr)
    =:L_{dd}^2,
\]
where the rightmost display is a polynomial in $(d,\gamma^{-1})$, $\gamma\in[\gamma_L,\gamma_U]$, built from the Gaussian moments of $\lVert\Omega_\gamma\rVert^2$ up to order four, all finite, so its supremum over the compact interval is finite and depends only on $(\gamma_L,\gamma_U,d)$; in particular $L_{dd}$ can be taken from this bound. Continuous differentiability of $\gamma\mapsto\tilde u_\gamma(x)$ into $\mathcal H_{\gamma_0}$, and of its derivative, follows from the pointwise computations above by dominated convergence in the weighted space, the dominating functions being of the form (polynomial in $\lVert\omega\rVert$)~$\times$~Gaussian, whose weighted second moments are finite by the same Gaussian-moment display; the mean value theorem, in its integral form for the $C^1$ map $\gamma\mapsto\partial_\gamma\tilde u_\gamma(x)$, then gives the announced Lipschitz bound
\[
    \bigl\lVert\partial_\gamma\tilde u_{\gamma_1}(x)-\partial_\gamma\tilde u_{\gamma_2}(x)\bigr\rVert_{\mathcal H_{\gamma_0}}
    =\Bigl\lVert\int_{\gamma_1}^{\gamma_2}\partial_\gamma^2\tilde u_\gamma(x)\,d\gamma\Bigr\rVert_{\mathcal H_{\gamma_0}}
    \le L_{dd}\,\lvert\gamma_1-\gamma_2\rvert .
\]

\emph{(ii)} Both sides are invariant under the swap $(\gamma_1,\gamma_2)\mapsto(\gamma_2,\gamma_1)$, the left because $\Delta\boldsymbol G$ is symmetric in the pair and the right because it is pair-free, so assume $T>0$. By the isometry for the two diagonal factors and the cross-kernel closed form \eqref{eq:crosskernel_closed} with its symmetry for the mixed factors, $\Delta G_{jk}=\bigl\langle(\tilde u_{\gamma_2}-\tilde u_{\gamma_1})(X_j),(\tilde u_{\gamma_2}-\tilde u_{\gamma_1})(X_k)\bigr\rangle_{\mathcal H_{\gamma_0}}$, so that
\[
    n^{-1}\sum_{j,k}a_ja_k\,\bar W_{jk}
    =\Bigl\lVert\tfrac1{\sqrt n}\sum_{j=1}^na_j\,\frac{\tilde u_{\gamma_2}(X_j)-\tilde u_{\gamma_1}(X_j)}{T}\Bigr\rVert_{\mathcal H_{\gamma_0}}^2 .
\]
By the fundamental theorem of calculus for the $C^1$ Hilbert-valued path of part~(i), applied at each $X_j$,
\[
    \frac{\tilde u_{\gamma_2}(x)-\tilde u_{\gamma_1}(x)}{T}=\int_0^1\partial_\gamma\tilde u_{\gamma_1+tT}(x)\,dt,
\]
and the finite sum commutes with the Bochner integral, so
\[
    \tfrac1{\sqrt n}\sum_{j=1}^na_j\,\frac{\tilde u_{\gamma_2}(X_j)-\tilde u_{\gamma_1}(X_j)}{T}=\int_0^1S_t\,dt,
    \qquad
    S_t:=\tfrac1{\sqrt n}\sum_{j=1}^na_j\,\partial_\gamma\tilde u_{\gamma_1+tT}(X_j) .
\]
Jensen's inequality for the probability measure $dt$ on $[0,1]$ and the convexity of the squared norm give $\bigl\lVert\int_0^1S_t\,dt\bigr\rVert^2\le\int_0^1\lVert S_t\rVert^2dt\le\sup_{t\in[0,1]}D_{\boldsymbol a}(\gamma_1+tT)$, and the segment $[\gamma_1,\gamma_2]$ lies in $[\gamma_L,\gamma_U]$, so the last supremum is at most $\sup_{\gamma\in[\gamma_L,\gamma_U]}D_{\boldsymbol a}(\gamma)$. Combining the displays proves (ii).

\emph{(iii)} By part~(i), the integrand of the Bochner expectation is continuously differentiable into the Hilbert--Schmidt class with, by the product rule for rank-one tensors and $\lVert a\otimes b\rVert_{\mathrm{HS}}=\lVert a\rVert\,\lVert b\rVert$,
\[
\begin{aligned}
    \partial_{\gamma}\bigl(\tilde u_{\gamma}(x)\otimes\tilde u_{\gamma}(x)\bigr)
    &=(\partial_{\gamma}\tilde u_{\gamma}(x))\otimes\tilde u_{\gamma}(x)+\tilde u_{\gamma}(x)\otimes(\partial_{\gamma}\tilde u_{\gamma}(x)),\\
    \Bigl\lVert\partial_{\gamma}\bigl(\tilde u_{\gamma}(x)\otimes\tilde u_{\gamma}(x)\bigr)\Bigr\rVert_{\mathrm{HS}}
    &\le2\,L_{d}\,\bar k^{1/2},
\end{aligned}
\]
using part~(i) for the first factor and $\lVert\tilde u_{\gamma}(x)\rVert_{\mathcal H_{\gamma_{0}}}^{2}=\lVert J_{\gamma}k_{\gamma}(x,\cdot)\rVert_{\mathcal H_{\gamma_{0}}}^{2}=k_{\gamma}(x,x)\le\bar k$ for the second. The right-hand side is deterministic and $\gamma$-free, so it dominates the difference quotients of the integrand, and dominated convergence for the Bochner integral differentiates under the expectation; Jensen's inequality, the norm being convex, then gives $\lVert\partial_{\gamma}\widetilde C_{\gamma}\rVert_{\mathrm{HS}}\le\mathbb E\bigl\lVert\partial_{\gamma}\bigl(\tilde u_{\gamma}(X)\otimes\tilde u_{\gamma}(X)\bigr)\bigr\rVert_{\mathrm{HS}}\le2\,\bar k^{1/2}L_{d}$, and continuous differentiability of $\gamma\mapsto\widetilde C_{\gamma}$ follows from the same domination applied to the derivative. The mean value theorem in integral form concludes:
\[
    \bigl\lVert\widetilde C_{\gamma_{1}}-\widetilde C_{\gamma_{2}}\bigr\rVert_{\mathrm{HS}}
    =\Bigl\lVert\int_{\gamma_{1}}^{\gamma_{2}}\partial_{\gamma}\widetilde C_{\gamma}\,d\gamma\Bigr\rVert_{\mathrm{HS}}
    \le2\,\bar k^{1/2}L_{d}\,\lvert\gamma_{1}-\gamma_{2}\rvert,
\]
which proves (iii). For the unit-normalized Gaussian kernel $\bar k=\sup_{\gamma,x}k_{\gamma}(x,x)=1$.
\end{proof}

The following elementary fact records the split of the orthogonalized residuals into the true residuals, the score-projection term, and an estimation remainder, together with the regime-qualified orders of that remainder; it is stated for a generic sample size.

\begin{fact}[Residual split]\label{fact:residual_split}
Let $\boldsymbol\varepsilon_\perp:=\hat{\boldsymbol\Pi}\boldsymbol{\hat\varepsilon}$ denote the orthogonalized residual vector computed on a sample of size $N\in\mathbb N$, and let $\boldsymbol\varepsilon_0$ denote the true residual vector on that sample. Under Assumptions~\ref{ass:D}(iii) and~\ref{ass:E} and, in the sieve regime, Assumption~\ref{ass:F}(ii)--(iii), with the inverses understood on the invertibility events of the proof of Lemma~\ref{lemma:estimation_remainder}, $\boldsymbol\varepsilon_\perp=\boldsymbol\varepsilon_0+\boldsymbol\delta$ with
$\boldsymbol\delta=-\hat{\mathbb G}\hat\beta+\boldsymbol r_{\mathrm{est}}$,
where $\hat\beta:=(\hat{\mathbb G}^\top\hat{\mathbb G})^{-1}\hat{\mathbb G}^\top\boldsymbol\varepsilon_0$ and $\boldsymbol r_{\mathrm{est}}$ is the estimation remainder. When the score $\nabla_\theta\varepsilon$ is $X$-measurable, the first-order effect is the score direction $\hat{\mathbb G}\boldsymbol\delta_\theta$, with $\boldsymbol\delta_\theta:=\hat\theta-\theta_0$, which the projection annihilates pointwise, $\hat{\boldsymbol\Pi}\hat{\mathbb G}=\boldsymbol 0$; the remainder $\boldsymbol r_{\mathrm{est}}$ is then the projected second-order remainder of that proof and satisfies $\lVert\boldsymbol r_{\mathrm{est}}\rVert=O_p(N^{1/2-2\varrho})$. In the sieve regime the projection annihilates only the fitted score, and the conditionally centered idiosyncratic first-order piece survives it, so $\lVert\boldsymbol r_{\mathrm{est}}\rVert=O_p(N^{1/2-\varrho})$. In both regimes $\boldsymbol r_{\mathrm{est}}^\top(\boldsymbol K_\gamma/N)\boldsymbol r_{\mathrm{est}}=o_p(1)$, uniformly in $\gamma\in[\gamma_L,\gamma_U]$, by Lemma~\ref{lemma:estimation_remainder}. The statement is sample-size-agnostic: it holds verbatim at any sample size, in particular at the selection sample of size $m$ (with $m$ in place of $N$ throughout).
\end{fact}

\begin{proof}
The split is finite-sample algebra composed from three identities recorded in the proof of Lemma~\ref{lemma:estimation_remainder}. The second-order expansion of the estimated residual around $\theta_0$ gives, exactly, $\hat{\boldsymbol\varepsilon}=\boldsymbol\varepsilon_0+\mathbb G\boldsymbol\delta_\theta+\boldsymbol r^{(2)}$, with $\boldsymbol\delta_\theta:=\hat\theta-\theta_0$, $\mathbb G$ the raw-score matrix with rows $\nabla_\theta\varepsilon(s_i;\hat\theta)^\top$, and $\boldsymbol r^{(2)}$ the Term A plus Term B remainder, element-wise $O_p(N^{-2\varrho})$. Projecting both sides with $\hat{\boldsymbol\Pi}=\boldsymbol I_N-\hat{\mathbb G}(\hat{\mathbb G}^\top\hat{\mathbb G})^{-1}\hat{\mathbb G}^\top$, where $\hat{\mathbb G}=P_Z\mathbb G$, and decomposing $\mathbb G=\hat{\mathbb G}+(\mathbb G-\hat{\mathbb G})$, the finite-sample orthogonality $\hat{\boldsymbol\Pi}\hat{\mathbb G}=\boldsymbol 0$ of that proof leaves
\[
    \boldsymbol\varepsilon_\perp=\hat{\boldsymbol\Pi}\boldsymbol\varepsilon_0+\hat{\boldsymbol\Pi}(\mathbb G-\hat{\mathbb G})\boldsymbol\delta_\theta+\hat{\boldsymbol\Pi}\boldsymbol r^{(2)} .
\]
The score-projection formula $\hat{\boldsymbol\Pi}\boldsymbol u=\boldsymbol u-\hat{\mathbb G}(\hat{\mathbb G}^\top\hat{\mathbb G})^{-1}\hat{\mathbb G}^\top\boldsymbol u$, that is, the definition of the empirical projector, applied to $\boldsymbol u=\boldsymbol\varepsilon_0$, gives $\hat{\boldsymbol\Pi}\boldsymbol\varepsilon_0=\boldsymbol\varepsilon_0-\hat{\mathbb G}\hat\beta$. Substituting this display into the previous one and collecting the last two terms into
\[
    \boldsymbol r_{\mathrm{est}}:=\hat{\boldsymbol\Pi}(\mathbb G-\hat{\mathbb G})\boldsymbol\delta_\theta+\hat{\boldsymbol\Pi}\boldsymbol r^{(2)}
\]
yields $\boldsymbol\varepsilon_\perp=\boldsymbol\varepsilon_0-\hat{\mathbb G}\hat\beta+\boldsymbol r_{\mathrm{est}}$, the claimed split with $\boldsymbol\delta=-\hat{\mathbb G}\hat\beta+\boldsymbol r_{\mathrm{est}}$. All three identities are pathwise and hold at any sample size, which also gives the sample-size-agnostic clause.

Suppose first that the score $\nabla_\theta\varepsilon$ is $X$-measurable, so that the sieve step is vacuous and $\hat{\mathbb G}=\mathbb G$. Then $\hat{\boldsymbol\Pi}(\mathbb G-\hat{\mathbb G})\boldsymbol\delta_\theta=\boldsymbol 0$, whence $\boldsymbol r_{\mathrm{est}}=\hat{\boldsymbol\Pi}\boldsymbol r^{(2)}$, and since $\hat{\boldsymbol\Pi}$ is symmetric idempotent with $\|\hat{\boldsymbol\Pi}\|_{\mathrm{op}}\le1$,
\[
    \lVert\boldsymbol r_{\mathrm{est}}\rVert=\lVert\hat{\boldsymbol\Pi}\boldsymbol r^{(2)}\rVert\le\lVert\boldsymbol r^{(2)}\rVert=O_p(N^{1/2-2\varrho}),
\]
the last step being the second-order computation of the proof of Lemma~\ref{lemma:estimation_remainder}, under Assumptions~\ref{ass:D}(iii) and~\ref{ass:E}(iii):
\[
    \lVert\boldsymbol r^{(2)}\rVert^2\le C\,\lVert\boldsymbol\delta_\theta\rVert^4\sum_{i}\sup_{\theta\in\Theta}\lVert\nabla_\theta^2\varepsilon(s_i;\theta)\rVert^2=O_p(N^{1-4\varrho}),
\]
the sum being $O_p(N)$ by the law of large numbers under $\mathbb E\big[\sup_{\theta}\lVert\nabla_\theta^2\varepsilon(S;\theta)\rVert^2\big]<\infty$.

In the sieve regime, decompose the raw score as in that proof,
\[
    \mathbb G-\hat{\mathbb G}=(\boldsymbol I-P_Z)\boldsymbol\eta_0+(\boldsymbol I-P_Z)\boldsymbol g+(\boldsymbol I-P_Z)\tilde{\boldsymbol H}\boldsymbol\delta_\theta,
\]
with $\boldsymbol\eta_0$ the conditionally centered idiosyncratic score matrix, $\boldsymbol g$ the population-score matrix with rows $g_{\theta_0}(X_i)^\top$, and $\tilde{\boldsymbol H}$ the path-averaged Hessian matrix; the third piece, multiplied by $\boldsymbol\delta_\theta$, is folded into $\boldsymbol r^{(2)}$, which is redefined accordingly and remains element-wise $O_p(N^{-2\varrho})$. The estimation remainder then splits into the three pieces of that proof,
\[
    \boldsymbol r_{\mathrm{est}}=\boldsymbol r_{\mathrm{center}}+\boldsymbol r_{\mathrm{sieve}}+\boldsymbol r_{\mathrm{proj}},
\]
with
\[
    \boldsymbol r_{\mathrm{center}}=\hat{\boldsymbol\Pi}(\boldsymbol I-P_Z)\boldsymbol\eta_0\,\boldsymbol\delta_\theta=(\boldsymbol I-P_Z)\boldsymbol\eta_0\,\boldsymbol\delta_\theta,\qquad
    \boldsymbol r_{\mathrm{sieve}}=\hat{\boldsymbol\Pi}(\boldsymbol I-P_Z)\boldsymbol g\,\boldsymbol\delta_\theta=(\boldsymbol I-P_Z)(\boldsymbol g-\hat{\boldsymbol g})\,\boldsymbol\delta_\theta,
\]
the two equalities by the identity $\hat{\boldsymbol\Pi}(\boldsymbol I-P_Z)=(\boldsymbol I-P_Z)$ recorded in that proof, the second additionally using $(\boldsymbol I-P_Z)\hat{\boldsymbol g}=\boldsymbol 0$ for the fitted sieve values $\hat{\boldsymbol g}=P_Z\mathbb G$ (projector idempotence), and $\boldsymbol r_{\mathrm{proj}}=\hat{\boldsymbol\Pi}\boldsymbol r^{(2)}$ as above. For the three norms, the Cauchy--Schwarz inequality, the contraction $\|\boldsymbol I-P_Z\|_{\mathrm{op}}\le1$, and Assumption~\ref{ass:D}(iii) give
\[
    \lVert\boldsymbol r_{\mathrm{center}}\rVert\le\lVert\boldsymbol\eta_0\rVert_{\mathrm F}\,\lVert\boldsymbol\delta_\theta\rVert=O_p(N^{1/2})\,O_p(N^{-\varrho})=O_p(N^{1/2-\varrho}),
\]
where $\lVert\boldsymbol\eta_0\rVert_{\mathrm F}=O_p(N^{1/2})$ follows from Markov's inequality, since $\mathbb E\lVert\boldsymbol\eta_0\rVert_{\mathrm F}^2\le N\,\mathbb E\lVert\eta_{0,1}\rVert^2$ and the entries of $\eta_{0,1}$ have a bounded fourth moment by Assumption~\ref{ass:E}(iii) and Jensen's inequality applied to $g_{\theta_0}(X)=\mathbb E[\nabla_\theta\varepsilon(S;\theta_0)\mid X]$; next,
\[
    \lVert\boldsymbol r_{\mathrm{sieve}}\rVert\le\lVert\hat{\boldsymbol g}-\boldsymbol g\rVert_{\mathrm F}\,\lVert\boldsymbol\delta_\theta\rVert=o_p(N^{\varrho})\,O_p(N^{-\varrho})=o_p(1),
\]
the first bound being the sieve estimation-error rate of Assumption~\ref{ass:F}(iii) through the empirical-norm bridge of the proof of Lemma~\ref{lemma:estimation_remainder}, $\lVert\hat{\boldsymbol g}-\boldsymbol g\rVert_{\mathrm F}^2=N\,\lVert\hat g_{\hat\theta}-g_{\theta_0}\rVert_{L^2(\mathbb P_n)}^2=o_p(N^{2\varrho})$; and $\lVert\boldsymbol r_{\mathrm{proj}}\rVert\le\lVert\boldsymbol r^{(2)}\rVert=O_p(N^{1/2-2\varrho})$ as in the $X$-measurable case, the second-order computation of that proof covering the folded-in gradient-drift terms. The triangle inequality yields $\lVert\boldsymbol r_{\mathrm{est}}\rVert=O_p(N^{1/2-\varrho})$ for $\varrho\le1/2$ and $O_p(1)$ in general.

Finally, at the RKHS quadratic-form level, the spectral route of the proof of Lemma~\ref{lemma:estimation_remainder} bounds each piece in the form $\boldsymbol r_k^\top(\boldsymbol K_\gamma/N)\boldsymbol r_k$: the centered piece by $O_p(N^{-2\varrho})$, the sieve-mean piece by $o_p(1)$, and the second-order piece by $O_p(N^{1-4\varrho})$; in the $X$-measurable regime only the second-order piece is present. Expanding the square and bounding each cross term by the Cauchy--Schwarz inequality for the positive semidefinite form $(\boldsymbol u,\boldsymbol v)\mapsto\boldsymbol u^\top(\boldsymbol K_\gamma/N)\boldsymbol v$, together with $2ab\le a^2+b^2$,
\[
    \boldsymbol r_{\mathrm{est}}^\top\Bigl(\frac{\boldsymbol K_\gamma}{N}\Bigr)\boldsymbol r_{\mathrm{est}}
    \le 3\sum_{k\in\{\mathrm{center},\mathrm{sieve},\mathrm{proj}\}}\boldsymbol r_k^\top\Bigl(\frac{\boldsymbol K_\gamma}{N}\Bigr)\boldsymbol r_k
    =o_p(1),
\]
the second step collecting the piecewise bounds just cited, and the last holding for the standing exponent $\varrho>1/4$, uniformly in $\gamma\in[\gamma_L,\gamma_U]$ by the closing remark of that proof. This proves the quadratic-form claim and, with it, the fact.
\end{proof}

The following lemma records the order-$p$ moment bound for centered quadratic forms in conditionally independent innovations; it is consumed in the proofs of Lemma~\ref{lem:lipschitz_common}, of Theorem~\ref{th:stat_uniform}, and of Lemma~\ref{lem:remainder_cv}, and is stated for a generic conditioning field and a generic coefficient matrix.

\begin{lemma}[Order-$p$ moment of a centered quadratic form]\label{lem:quad_moment}
Let $p\in(2,4]$, let $\mathcal G$ be a sub-$\sigma$-field, and let $\varepsilon_{0,1},\dots,\varepsilon_{0,n}$ be real random variables which, conditionally on $\mathcal G$, are independent and centered, with conditional variances $\sigma_j^2:=\mathbb E[\varepsilon_{0,j}^2\mid\mathcal G]\le\sigma_{\max}^2$ and $\sup_{j\le n}\mathbb E[\varepsilon_{0,j}^{2p}\mid\mathcal G]\le\bar M_{2p}<\infty$. Then every symmetric $\mathcal G$-measurable $A\in\mathbb R^{n\times n}$ satisfies
\begin{equation}\label{eq:quad4}
    \mathbb E\Bigl[\bigl|\boldsymbol\varepsilon_0^\top A\,\boldsymbol\varepsilon_0-\operatorname{tr}(A D)\bigr|^{p}\ \Big|\ \mathcal G\Bigr]\le C_{p}\bigl(\sigma_{\max},\bar M_{2p}\bigr)\,\|A\|_{\mathrm F}^{p},\qquad D:=\operatorname{diag}\bigl(\sigma_1^2,\dots,\sigma_n^2\bigr),
\end{equation}
for a constant $C_p$ depending only on $(p,\sigma_{\max},\bar M_{2p})$.
\end{lemma}

\begin{proof}
Work conditionally on $\mathcal G$ throughout; all coefficient matrices below are $\mathcal G$-measurable. Expanding the quadratic form and centering its diagonal part at its conditional mean,
\[
    \boldsymbol\varepsilon_0^\top A\,\boldsymbol\varepsilon_0-\operatorname{tr}(A D)
    =T_{1}+T_{2},\qquad
    T_{1}:=\sum_{j=1}^{n}A_{jj}\bigl(\varepsilon_{0,j}^{2}-\sigma_j^{2}\bigr),\qquad
    T_{2}:=\sum_{j\ne k}A_{jk}\,\varepsilon_{0,j}\varepsilon_{0,k},
\]
the diagonal mean being $\mathbb E[\varepsilon_{0,j}^{2}\mid\mathcal G]=\sigma_j^{2}$ entrywise and the off-diagonal conditional mean vanishing under the conditional centeredness and independence; the two families are bounded separately and combined by $(x+y)^{p}\le2^{p-1}(x^{p}+y^{p})$ at the end.

\emph{The diagonal family.}\; Conditionally on $\mathcal G$ the summands of $T_{1}$ are independent and centered, and Rosenthal's inequality at order $p$ gives
\[
    \mathbb E\bigl[|T_{1}|^{p}\bigm|\mathcal G\bigr]
    \le C_{p}^{R}\Bigl\{\sum_{j=1}^{n}|A_{jj}|^{p}\,\mathbb E\bigl[\bigl|\varepsilon_{0,j}^{2}-\sigma_j^{2}\bigr|^{p}\bigm|\mathcal G\bigr]
      +\Bigl(\sum_{j=1}^{n}A_{jj}^{2}\,\nu_j\Bigr)^{p/2}\Bigr\},
    \qquad\nu_j:=\mathbb E\bigl[\bigl(\varepsilon_{0,j}^{2}-\sigma_j^{2}\bigr)^{2}\bigm|\mathcal G\bigr],
\]
with $C_{p}^{R}$ the Rosenthal constant. The summand moments obey $\mathbb E[|\varepsilon_{0,j}^{2}-\sigma_j^{2}|^{p}\mid\mathcal G]\le2^{p-1}\bigl(\bar M_{2p}+\sigma_{\max}^{2p}\bigr)$, by $|a-b|^{p}\le2^{p-1}(a^{p}+b^{p})$, the hypothesis $\mathbb E[\varepsilon_{0,j}^{2p}\mid\mathcal G]\le\bar M_{2p}$, and Jensen's inequality $\sigma_j^{2p}\le\sigma_{\max}^{2p}$; and $\nu_j\le\mathbb E[\varepsilon_{0,j}^{4}\mid\mathcal G]\le\bar M_{2p}^{2/p}$ by the Lyapunov interpolation $\mathbb E[|\varepsilon_{0,j}|^{q}\mid\mathcal G]\le\bar M_{2p}^{q/(2p)}$ at $q=4$. The power mean $\sum_{j}|A_{jj}|^{p}\le\bigl(\sum_{j}A_{jj}^{2}\bigr)^{p/2}$ for $p\ge2$, the diagonal bound $\sum_{j}A_{jj}^{2}\le\|A\|_{\mathrm F}^{2}$, and $\bigl(\sum_{j}A_{jj}^{2}\nu_j\bigr)^{p/2}\le\bar M_{2p}\|A\|_{\mathrm F}^{p}$ then give
\[
    \mathbb E\bigl[|T_{1}|^{p}\bigm|\mathcal G\bigr]\le C_{p}\bigl(\sigma_{\max},\bar M_{2p}\bigr)\,\|A\|_{\mathrm F}^{p}.
\]

\emph{The off-diagonal family.}\; Let $\boldsymbol\varepsilon_0'$ be an independent copy of $\boldsymbol\varepsilon_0$ given $\mathcal G$ and set $S_j':=\sum_{k\ne j}A_{jk}\varepsilon_{0,k}'$, so that $\sum_{j}\varepsilon_{0,j}S_j'$ is the decoupled form of $T_{2}$. The decoupling inequality for homogeneous chaos forms of order two \citep[Ch.~3]{delapena1999decoupling} gives
\[
    \mathbb E\bigl[|T_{2}|^{p}\bigm|\mathcal G\bigr]
    \le C_{p}^{dec}\,\mathbb E\Bigl[\Bigl|\sum_{j=1}^{n}\varepsilon_{0,j}S_j'\Bigr|^{p}\Bigm|\mathcal G\Bigr],
\]
with $C_{p}^{dec}$ the decoupling constant. Conditioning further on $\boldsymbol\varepsilon_0'$ and applying Rosenthal's inequality to the linear form in the conditionally independent, centered variables $\varepsilon_{0,j}$, whose $p$-th moments satisfy $\mathbb E[|\varepsilon_{0,j}|^{p}\mid\mathcal G]\le\bar M_{2p}^{1/2}$ by the Lyapunov interpolation at $q=p$,
\[
    \mathbb E\Bigl[\Bigl|\sum_{j}\varepsilon_{0,j}S_j'\Bigr|^{p}\Bigm|\mathcal G,\boldsymbol\varepsilon_0'\Bigr]
    \le C_{p}^{R}\Bigl\{\bar M_{2p}^{1/2}\sum_{j}|S_j'|^{p}
      +\sigma_{\max}^{p}\Bigl(\sum_{j}(S_j')^{2}\Bigr)^{p/2}\Bigr\},
\]
the second summand using $\sum_{j}\sigma_j^{2}(S_j')^{2}\le\sigma_{\max}^{2}\sum_{j}(S_j')^{2}$. The two summands are unconditioned over $\boldsymbol\varepsilon_0'$ separately. For the first, Rosenthal's inequality applied to each linear form $S_j'$ in the independent, centered variables $\varepsilon_{0,k}'$ gives
\[
    \mathbb E\bigl[|S_j'|^{p}\bigm|\mathcal G\bigr]
    \le C_{p}^{R}\Bigl\{\bar M_{2p}^{1/2}\sum_{k\ne j}|A_{jk}|^{p}+\sigma_{\max}^{p}\Bigl(\sum_{k\ne j}A_{jk}^{2}\Bigr)^{p/2}\Bigr\},
\]
and summing over $j$, the power mean $\sum_{j,k}|A_{jk}|^{p}\le\|A\|_{\mathrm F}^{p}$ and the superadditivity $\sum_{j}\bigl(\sum_{k}A_{jk}^{2}\bigr)^{p/2}\le\bigl(\sum_{j,k}A_{jk}^{2}\bigr)^{p/2}=\|A\|_{\mathrm F}^{p}$ for $p/2\ge1$ yield
\[
    \bar M_{2p}^{1/2}\,\mathbb E\Bigl[\sum_{j}|S_j'|^{p}\Bigm|\mathcal G\Bigr]\le C_{p}\,\bar M_{2p}\,\|A\|_{\mathrm F}^{p}.
\]
For the second, $\sum_{j}(S_j')^{2}=\boldsymbol\varepsilon_0'^{\top}M\boldsymbol\varepsilon_0'$ with $M:=(A^{\mathrm{od}})^{\top}A^{\mathrm{od}}$ symmetric positive semidefinite, $A^{\mathrm{od}}$ denoting the off-diagonal part of $A$; at $q:=p/2\in(1,2]$, splitting around the conditional mean, then applying $(a+b)^{q}\le2^{q-1}(a^{q}+b^{q})$ and Lyapunov's inequality $\bigl(\mathbb E|W|^{q}\bigr)^{1/q}\le\bigl(\mathbb EW^{2}\bigr)^{1/2}$ for $q\le2$ to the centered part, gives
\[
    \mathbb E\Bigl[\bigl(\boldsymbol\varepsilon_0'^{\top}M\boldsymbol\varepsilon_0'\bigr)^{q}\Bigm|\mathcal G\Bigr]
    \le2^{q-1}\Bigl\{\Bigl(\operatorname{Var}\bigl(\boldsymbol\varepsilon_0'^{\top}M\boldsymbol\varepsilon_0'\bigm|\mathcal G\bigr)\Bigr)^{q/2}+\bigl|\operatorname{tr}(M D)\bigr|^{q}\Bigr\}.
\]
The variance is computed by expanding the square and keeping the surviving pairings, the conditional independence and centeredness leaving only the diagonal squares and the two orderings of each off-diagonal index pair:
\[
    \operatorname{Var}\bigl(\boldsymbol\varepsilon_0'^{\top}M\boldsymbol\varepsilon_0'\bigm|\mathcal G\bigr)
    =\sum_{j}M_{jj}^{2}\,\mathbb E\bigl[\bigl(\varepsilon_{0,j}'^{2}-\sigma_j^{2}\bigr)^{2}\bigm|\mathcal G\bigr]+2\sum_{j\ne k}M_{jk}^{2}\,\sigma_j^{2}\sigma_k^{2}
    \le\bigl(\bar M_{2p}^{2/p}+2\sigma_{\max}^{4}\bigr)\,\|M\|_{\mathrm F}^{2},
\]
the bound using $\mathbb E[(\varepsilon_{0,j}'^{2}-\sigma_j^{2})^{2}\mid\mathcal G]\le\mathbb E[\varepsilon_{0,j}'^{4}\mid\mathcal G]\le\bar M_{2p}^{2/p}$ and $\sigma_j^{2}\sigma_k^{2}\le\sigma_{\max}^{4}$; the trace term obeys $|\operatorname{tr}(MD)|\le\|D\|_{\mathrm{op}}\operatorname{tr}M\le\sigma_{\max}^{2}\|A^{\mathrm{od}}\|_{\mathrm F}^{2}\le\sigma_{\max}^{2}\|A\|_{\mathrm F}^{2}$, since $M$ and $D$ are positive semidefinite, and $\|M\|_{\mathrm F}=\|(A^{\mathrm{od}})^{\top}A^{\mathrm{od}}\|_{\mathrm F}\le\|A^{\mathrm{od}}\|_{\mathrm F}^{2}\le\|A\|_{\mathrm F}^{2}$ by submultiplicativity. Hence
\[
    \sigma_{\max}^{p}\,\mathbb E\Bigl[\Bigl(\sum_{j}(S_j')^{2}\Bigr)^{p/2}\Bigm|\mathcal G\Bigr]\le C_{p}\bigl(\sigma_{\max},\bar M_{2p}\bigr)\,\|A\|_{\mathrm F}^{p},
\]
and the two unconditioned summands combine into $\mathbb E[|T_{2}|^{p}\mid\mathcal G]\le C_{p}(\sigma_{\max},\bar M_{2p})\,\|A\|_{\mathrm F}^{p}$.

\emph{Assembly.}\; The split display and $(x+y)^{p}\le2^{p-1}(x^{p}+y^{p})$ give
\[
    \mathbb E\Bigl[\bigl|\boldsymbol\varepsilon_0^\top A\,\boldsymbol\varepsilon_0-\operatorname{tr}(A D)\bigr|^{p}\ \Big|\ \mathcal G\Bigr]\le C_{p}\bigl(\sigma_{\max},\bar M_{2p}\bigr)\,\|A\|_{\mathrm F}^{p},
\]
which is \eqref{eq:quad4}.
\end{proof}

\begin{lemma}[Lipschitz properties on the common space]\label{lem:lipschitz_common}
Let
\[
\begin{aligned}
    \widetilde C_{n,\gamma}&:=J_\gamma C_{n,\gamma}J_\gamma^{-1},\\
    \widetilde M_n(\gamma)&:=J_\gamma M_n(\gamma)
      =n^{-1/2}\sum_{j=1}^n\varepsilon_{\perp,j}\,J_\gamma k_\gamma(\cdot,X_j),\\
    \widetilde P^{(\gamma)}_{>K}&:=J_\gamma P^{(\gamma)}_{>K}J_\gamma^{-1}
\end{aligned}
\]
all three objects act on the common space $\mathcal{H}_{\gamma_0}$. Assume Assumptions~\ref{ass:A}--\ref{ass:F} and~\ref{ass:G}(iii),(v), and, under the local alternative, the Pitman calibration $N\lVert g\rVert_{L^2}^2=O(1)$ of the regime-coverage paragraph below. Then:
\begin{enumerate}[label=(\alph*)]
    \item there exists $L_M=O_p(1)$ with $\lVert\widetilde M_n(\gamma_1)-\widetilde M_n(\gamma_2)\rVert_{\mathcal H_{\gamma_0}}
          \le L_M|\gamma_1-\gamma_2|$ for all $\gamma_1,\gamma_2\in[\gamma_L,\gamma_U]$, and
          $\overline M_n:=\sup_{\gamma\in[\gamma_L,\gamma_U]}\lVert\widetilde M_n(\gamma)\rVert_{\mathcal H_{\gamma_0}}=O_p(1)$;
    \item there exists $L_C=O_p(1)$ with $\lVert\widetilde C_{n,\gamma_1}-\widetilde C_{n,\gamma_2}\rVert_{\mathrm{HS}}
          \le L_C|\gamma_1-\gamma_2|$ for all $\gamma_1,\gamma_2\in[\gamma_L,\gamma_U]$, where the Hilbert--Schmidt difference
          is taken on the common space; consequently $r_n:=\sup_{\gamma\in[\gamma_L,\gamma_U]}
          \lVert C_{n,\gamma}-C_\gamma\rVert_{\mathrm{op}}=o_p(1)$;
    \item the transported resolvent $\widetilde A_n(\eta):=(\widetilde C_{n,\gamma}+\lambda I)^{-1}$ satisfies, for all
          $\eta_1=(\gamma_1,\lambda_1),\eta_2=(\gamma_2,\lambda_2)\in\Upsilon$,
    \[
        \lVert\widetilde A_n(\eta_1)-\widetilde A_n(\eta_2)\rVert_{\mathrm{op}}
        \le\lambda_L^{-2}\bigl(L_C|\gamma_1-\gamma_2|+|\lambda_1-\lambda_2|\bigr),
        \qquad
        \lVert\widetilde A_n(\eta)\rVert_{\mathrm{op}}\le\lambda_L^{-1},
    \]
    with $L_C$ as in part (b).
\end{enumerate}
Under fixed alternatives the same estimates hold with both sides inflated by powers of $\sqrt n$.
\end{lemma}

\begin{proof}

\medskip
\noindent\textbf{(a) Mean embedding.}\; The argument combines the feature-distance geometry of the transported kernels with the derivative-feature calculus of Fact~\ref{fact:pairdiff}: the pair supremum of the modulus is reduced to a one-dimensional derivative-process supremum and closed by a conditional second-moment modulus with dyadic chaining, and the boundedness half is read off at a base point through the quadratic-form clause of Fact~\ref{fact:residual_split}.

\emph{Gram--matrix reduction.} By linearity and the binomial identity on norms,
\[
\begin{aligned}
    \lVert\widetilde M_n(\gamma_1)-\widetilde M_n(\gamma_2)\rVert_{\mathcal H_{\gamma_0}}^2
    &=\frac1n\sum_{j,k}\varepsilon_{\perp,j}\varepsilon_{\perp,k}\,\Delta G_{jk},\\
    \Delta G_{jk}&:=k_{\gamma_1}(r_{jk})+k_{\gamma_2}(r_{jk})-2K_{\gamma_1,\gamma_2}(X_j,X_k),
\end{aligned}
\]
with $r_{jk}:=\lVert X_j-X_k\rVert$: the two same-bandwidth terms reduce to native kernel entries because $J_\gamma$ is an isometry,
$\langle J_\gamma k_\gamma(\cdot,x),J_\gamma k_\gamma(\cdot,y)\rangle_{\mathcal H_{\gamma_0}}=\langle k_\gamma(\cdot,x),k_\gamma(\cdot,y)\rangle_{\mathcal H_\gamma}=k_\gamma(x,y)$,
and the mixed term is the cross-kernel
$K_{\gamma_1,\gamma_2}(x,y):=\langle J_{\gamma_1}k_{\gamma_1}(\cdot,x),J_{\gamma_2}k_{\gamma_2}(\cdot,y)\rangle_{\mathcal H_{\gamma_0}}$.
By Lemma~\ref{lem:crosskernel}, $K_{\gamma_1,\gamma_2}=c_{12}\,k_{\gamma_{12}}$
with $\gamma_{12}\in[\gamma_L,\gamma_U]$ and $c_{12}\in[c_{\min},1]$ as stated there;
in particular $c_{12}=1$ at $\gamma_1=\gamma_2$, recovering the isometry. Note that $\Delta\boldsymbol G$ is positive semidefinite: it is the Gram matrix of the feature differences
$J_{\gamma_1}k_{\gamma_1}(\cdot,X_j)-J_{\gamma_2}k_{\gamma_2}(\cdot,X_j)$, the symmetry of the cross term being part of \eqref{eq:crosskernel_closed}.

\emph{Residual decomposition.} Throughout part (a), write $\boldsymbol\varepsilon_\perp=\boldsymbol\varepsilon_0-\hat{\mathbb G}\hat\beta+\boldsymbol r_{\mathrm{est}}$ with the three pieces of Fact~\ref{fact:residual_split}, applied at the testing sample of size $n$; the uniformity step below consumes the three pieces separately (the projected score term through the rank-$d_\theta$ projection identity $\hat{\mathbb G}\hat\beta=H\boldsymbol\varepsilon_0$ with $H:=\hat{\mathbb G}(\hat{\mathbb G}^\top\hat{\mathbb G})^{-1}\hat{\mathbb G}^\top$, the estimation remainder through the quadratic-form clause of the fact), never the Euclidean norm of $\boldsymbol\delta$, which is not bounded in the sieve regime.

\emph{Regime coverage.} The stochastic arguments below assume that the noise part of $\boldsymbol\varepsilon_0$ is conditionally centered given $\mathbf X$; this paragraph records the split and the envelopes by which the local alternative reduces to that case. Under $H_0$ write $\boldsymbol\varepsilon_0=\boldsymbol v$, a vector of conditionally independent, conditionally centered entries with $\mathbb E[v^2\mid X]=\sigma^2(X)$. Under $H_{1,n}$ write $\boldsymbol\varepsilon_0=\boldsymbol g+\boldsymbol v$ with the raw shift vector $\boldsymbol g:=(g(X_1),\dots,g(X_n))^\top$, $g:=c_NR$, and $\boldsymbol v$ the conditionally centered noise. The effective signal of the orthogonalized residuals is $c_N\boldsymbol\Pi R$: at the testing sample, of size $N=n$, $c_N=(m+n)^{-1/2}$ and $h:=\boldsymbol\Pi R$, the convention of the uniform-split paragraph of the proof of Theorem~\ref{th:stat_uniform}; there, under $H_{1,n}$, the symbol $\boldsymbol\varepsilon_0$ denotes the centered noise (the two readings coincide under $H_0$), while the raw reading $\boldsymbol\varepsilon_0=\boldsymbol g+\boldsymbol v$ is local to the present paragraph; at the selection sample, of size $N=m$ (Lemma~\ref{lem:bilinear_cv} below), $c_N=\delta_m$ with $\sqrt m\,\delta_m\to a<\infty$ and $h$ the fixed local-alternative direction. The score-span component of the raw shift is carried asymptotically by the projected-score piece $-\hat{\mathbb G}\hat\beta$ of the residual decomposition, and every envelope below is stated for the raw vector because both components share the same integrability, $\mathbb E[(\boldsymbol\Pi R)^2]\le\mathbb E[R^2]<\infty$ and $\mathbb E[((\boldsymbol I-\boldsymbol\Pi)R)^2]\le\mathbb E[R^2]<\infty$, by the contraction property of the $L^2$-projection $\boldsymbol\Pi$ and the standing convention $\mathbb E[R^2]<\infty$ of Section~\ref{sec:cv_selection}. In both regimes the Pitman calibration holds: by the law of large numbers,
\begin{equation}\label{eq:pitman_lln}
    N^{-1}\sum_{j=1}^{N}h^2(X_j)\;\xrightarrow{p}\;\mathbb E\bigl[h^2\bigr]\;\le\;\mathbb E\bigl[R^2\bigr]\;<\;\infty,
\end{equation}
and likewise $N^{-1}\sum_jR^2(X_j)\xrightarrow{p}\mathbb E[R^2]$, so $\lVert\boldsymbol g\rVert^2=Nc_N^2\,\lVert g\rVert_{L^2}^2\,(1+o_p(1))=O_p(1)$, the factor $Nc_N^2$ being $n/(m+n)\to1-c$ at the testing sample and $m\delta_m^2\to a^2<\infty$ at the selection sample. The mean part costs only entrywise envelopes. By Cauchy--Schwarz against $|\Delta G_{jk}|\le WT^2$, the entrywise bound of the feature-distance paragraph below,
\[
    \bigl\lvert n^{-1}\boldsymbol g^\top\Delta\boldsymbol G\,\boldsymbol g\bigr\rvert
    \;\le\; W T^2\lVert\boldsymbol g\rVert^2
    \;=\;O_p(T^2),
\]
uniformly in the pair index, and at fixed alternatives the same bound is $O_p(nT^2)$, the inflation recorded in the closing statement of the lemma. On the $\bar{\boldsymbol W}$-scale, the trace bound
\[
    \lambda_{\max}\bigl(n^{-1}\bar{\boldsymbol W}\bigr)\;\le\;n^{-1}\operatorname{tr}\bar{\boldsymbol W}\;=\;n^{-1}\sum_j\bar W_{jj}\;\le\;W,
\]
holds, the diagonal entries sharing the entrywise bound, so $n^{-1}\boldsymbol g^\top\bar{\boldsymbol W}\boldsymbol g\le W\lVert\boldsymbol g\rVert^2=O_p(1)$, and the cross term follows from the Cauchy--Schwarz inequality for the positive semidefinite form $(\boldsymbol u,\boldsymbol v)\mapsto n^{-1}\boldsymbol u^\top\bar{\boldsymbol W}\boldsymbol v$,
\[
    \bigl\lvert n^{-1}\boldsymbol g^\top\bar{\boldsymbol W}\,\boldsymbol v\bigr\rvert
    \;\le\;\bigl(W\lVert\boldsymbol g\rVert^2\bigr)^{1/2}\Bigl(\sup_{\gamma_1\ne\gamma_2}n^{-1}\boldsymbol v^\top\bar{\boldsymbol W}\boldsymbol v\Bigr)^{1/2}
    \;=\;O_p(1),
\]
the second factor being bounded in the uniformity step below; the $\boldsymbol v$-part enters the derivative-process argument there verbatim.

\emph{Feature distance and the second-order entrywise bound.} By symmetry of the statement in $(\gamma_1,\gamma_2)$, assume $\gamma_2\ge\gamma_1$ and put $T:=\gamma_2-\gamma_1$. Write $\tilde u_a(x):=J_{\gamma_a}k_{\gamma_a}(x,\cdot)$ for the transported features, so that $\lVert\tilde u_a(x)\rVert_{\mathcal H_{\gamma_0}}=\lVert k_{\gamma_a}(x,\cdot)\rVert_{\mathcal H_{\gamma_a}}=\sqrt{k_{\gamma_a}(x,x)}=1$, and $f(x):=\tilde u_1(x)-\tilde u_2(x)$ for the feature difference. The feature distance is exact and second-order:
\[
    \lVert f(x)\rVert_{\mathcal H_{\gamma_0}}^2
    =k_{\gamma_1}(x,x)+k_{\gamma_2}(x,x)-2K_{\gamma_1,\gamma_2}(x,x)
    =2\bigl(1-c_{12}\bigr),
\]
the mixed term being the cross-kernel on the diagonal, $K_{\gamma_1,\gamma_2}(x,x)=c_{12}\,k_{\gamma_{12}}(x,x)=c_{12}$ by Lemma~\ref{lem:crosskernel}. With the closed form $c_{12}=2^{d/2}(\gamma_1\gamma_2)^{d/4}/(\gamma_1+\gamma_2)^{d/2}$ of that lemma, equivalently $c_{12}=\bigl(1-x^2\bigr)^{d/4}$ by $4\gamma_1\gamma_2=(\gamma_1+\gamma_2)^2(1-x^2)$, where $x:=(\gamma_2-\gamma_1)/(\gamma_1+\gamma_2)$, the elementary inequality $1-(1-u)^{\alpha}\le\max(1,\alpha)\,u$ for $u\in[0,1]$ (concavity at $\alpha\le1$, Bernoulli at $\alpha\ge1$) and $x\le T/(2\gamma_L)$ give
\[
    \lVert f(x)\rVert_{\mathcal H_{\gamma_0}}\le L_k'\,T,
    \qquad L_k':=\frac{1}{\gamma_L}\sqrt{\frac{c_d}{2}},\qquad c_d:=\max\Bigl(1,\frac d4\Bigr),
\]
a deterministic, $\mathbf X$-free constant. The entrywise bound follows without any expansion in the bandwidth gap: by exact Gram algebra and the symmetry of the cross-kernel (part of \eqref{eq:crosskernel_closed}),
\[
    \Delta G_{jk}=\bigl\langle f(X_j),f(X_k)\bigr\rangle_{\mathcal H_{\gamma_0}},
\]
so that, by Cauchy--Schwarz,
\[
    |\Delta G_{jk}|\le\lVert f(X_j)\rVert\,\lVert f(X_k)\rVert\le W\,T^2,
    \qquad W:=L_k'^2,
\]
uniformly in $j,k$ and $\mathbf X$; on the diagonal the bound is sharp in structure, $\Delta G_{jj}=2(1-c_{12})\le WT^2$. Since $\Delta\boldsymbol G$ is, as recorded after the Gram--matrix reduction, the Gram matrix of the feature differences $f(X_j)$, scaling by the positive number $T^{-2}$ preserves symmetry and positive semidefiniteness in
\[
    \bar{\boldsymbol W}:=T^{-2}\Delta\boldsymbol G,
    \qquad
    |\bar W_{jk}|\le W .
\]
The Gram--matrix reduction therefore yields the exact pathwise identity
$\lVert\widetilde M_n(\gamma_1)-\widetilde M_n(\gamma_2)\rVert_{\mathcal H_{\gamma_0}}^2=|\gamma_1-\gamma_2|^2\,n^{-1}\boldsymbol\varepsilon_\perp^\top\bar{\boldsymbol W}\boldsymbol\varepsilon_\perp$,
with no remainder term; consistency of signs is visible on both sides, since the left side is a squared norm while the right side is a nonnegative quadratic form in a positive semidefinite matrix. It remains to bound the quadratic form uniformly over pairs. Split the residuals by the decomposition recorded above, $\boldsymbol\varepsilon_\perp=\boldsymbol\varepsilon_0-H\boldsymbol\varepsilon_0+\boldsymbol r_{\mathrm{est}}$, and expand the square, bounding each cross term by the Cauchy--Schwarz inequality for the positive semidefinite form $(\boldsymbol u,\boldsymbol v)\mapsto n^{-1}\boldsymbol u^\top\bar{\boldsymbol W}\boldsymbol v$, together with $2ab\le a^2+b^2$ applied twice, as in the proof of Fact~\ref{fact:residual_split}:
\[
    n^{-1}\boldsymbol\varepsilon_\perp^\top\bar{\boldsymbol W}\boldsymbol\varepsilon_\perp
    \;\le\;3\,n^{-1}\boldsymbol\varepsilon_0^\top\bar{\boldsymbol W}\boldsymbol\varepsilon_0
    \;+\;3\,n^{-1}(H\boldsymbol\varepsilon_0)^\top\bar{\boldsymbol W}(H\boldsymbol\varepsilon_0)
    \;+\;3\,n^{-1}\boldsymbol r_{\mathrm{est}}^\top\bar{\boldsymbol W}\boldsymbol r_{\mathrm{est}} ,
\]
and the three terms on the right are bounded separately, uniformly over the pair.

\emph{The noise term.} For the first term, Fact~\ref{fact:pairdiff}(ii) reduces the pair supremum to a one-dimensional supremum of the derivative-feature process,
\[
    \sup_{\gamma_1\ne\gamma_2}n^{-1}\boldsymbol a^\top\bar{\boldsymbol W}\boldsymbol a\;\le\;\sup_{\gamma\in[\gamma_L,\gamma_U]}D_{\boldsymbol a}(\gamma)
    \qquad\text{at}\qquad\boldsymbol a\in\{\boldsymbol\varepsilon_0,\boldsymbol v\} ,
\]
where $\boldsymbol v=\boldsymbol\varepsilon_0$ under $H_0$; under $H_{1,n}$, expanding the quadratic form of the regime-coverage split $\boldsymbol\varepsilon_0=\boldsymbol g+\boldsymbol v$ with $2ab\le a^2+b^2$ gives
$\sup_{\gamma_1\ne\gamma_2}n^{-1}\boldsymbol\varepsilon_0^\top\bar{\boldsymbol W}\boldsymbol\varepsilon_0\le2W\lVert\boldsymbol g\rVert^2+2\sup_{\gamma_1\ne\gamma_2}n^{-1}\boldsymbol v^\top\bar{\boldsymbol W}\boldsymbol v$,
the first summand by the mean-part envelope of that paragraph. It therefore suffices to bound $\sup_\gamma D_v(\gamma)$, where $D_v(\gamma):=D_{\boldsymbol a}(\gamma)\big|_{\boldsymbol a=\boldsymbol v}=\bigl\lVert\tfrac1{\sqrt n}\sum_jv_j\,\partial_\gamma\tilde u_\gamma(X_j)\bigr\rVert_{\mathcal H_{\gamma_0}}^2$; the argument is a conditional second-moment modulus with dyadic chaining. Conditionally on $\mathbf X$, the summands $v_j\,\partial_\gamma\tilde u_{\gamma_L}(X_j)$ are independent and centered, with $\mathbb E[v_j^2\mid\mathbf X]=\sigma^2(X_j)\le\sigma_{\max}^2$ (Assumption~\ref{ass:G}(iii)), so expanding the squared norm and keeping the diagonal pairings gives the base-point bound
\[
    \mathbb E\bigl[D_v(\gamma_L)\bigm|\mathbf X\bigr]
    =n^{-1}\sum_{j=1}^n\sigma^2(X_j)\,\bigl\lVert\partial_\gamma\tilde u_{\gamma_L}(X_j)\bigr\rVert_{\mathcal H_{\gamma_0}}^2
    \;\le\;\sigma_{\max}^2L_d^2
\]
by Fact~\ref{fact:pairdiff}(i). For the modulus, introduce the derivative-feature Gram and its increments,
\[
    G^{(d)}_{\gamma,jk}:=\bigl\langle\partial_\gamma\tilde u_\gamma(X_j),\partial_\gamma\tilde u_\gamma(X_k)\bigr\rangle_{\mathcal H_{\gamma_0}},
    \qquad
    \Delta^{(d)}_{\gamma,jk}:=G^{(d)}_{\gamma_1,jk}-G^{(d)}_{\gamma_2,jk},
\]
so that $D_v(\gamma)=n^{-1}\sum_{j,k}v_jv_kG^{(d)}_{\gamma,jk}$ by the Gram representation of Fact~\ref{fact:pairdiff}(ii); adding and subtracting $\partial_\gamma\tilde u_{\gamma_2}(X_j)$ in the first factor of each inner product, then applying the triangle inequality and Fact~\ref{fact:pairdiff}(i),
\[
\begin{aligned}
    \bigl|\Delta^{(d)}_{\gamma,jk}\bigr|
    &\le\bigl\lVert\partial_\gamma\tilde u_{\gamma_1}(X_j)-\partial_\gamma\tilde u_{\gamma_2}(X_j)\bigr\rVert\,\bigl\lVert\partial_\gamma\tilde u_{\gamma_1}(X_k)\bigr\rVert\\
    &\quad+\bigl\lVert\partial_\gamma\tilde u_{\gamma_2}(X_j)\bigr\rVert\,\bigl\lVert\partial_\gamma\tilde u_{\gamma_1}(X_k)-\partial_\gamma\tilde u_{\gamma_2}(X_k)\bigr\rVert\\
    &\le L_{dd}'\,\lvert\gamma_1-\gamma_2\rvert ,
\end{aligned}
\]
uniformly in $j,k$ and $\mathbf X$, with the deterministic constant $L_{dd}':=2L_dL_{dd}$. Splitting $D_v(\gamma_1)-D_v(\gamma_2)=n^{-1}\sum_jv_j^2\Delta^{(d)}_{\gamma,jj}+n^{-1}\sum_{j\ne k}v_jv_k\Delta^{(d)}_{\gamma,jk}$ and using $(a+b)^2\le2a^2+2b^2$, the two parts have conditional second moments
\[
\begin{aligned}
    \mathbb E\Bigl[\Bigl(n^{-1}\textstyle\sum_jv_j^2\Delta^{(d)}_{\gamma,jj}\Bigr)^{2}\Bigm|\mathbf X\Bigr]
    &\le n^{-1}\textstyle\sum_j\mathbb E\bigl[v_j^4\bigm|\mathbf X\bigr]\bigl(\Delta^{(d)}_{\gamma,jj}\bigr)^{2}
    \;\le\;\bar M_4L_{dd}'^{\,2}\,\lvert\gamma_1-\gamma_2\rvert^{2},\\
    \mathbb E\Bigl[\Bigl(n^{-1}\textstyle\sum_{j\ne k}v_jv_k\Delta^{(d)}_{\gamma,jk}\Bigr)^{2}\Bigm|\mathbf X\Bigr]
    &=2n^{-2}\textstyle\sum_{j\ne k}\sigma^2(X_j)\sigma^2(X_k)\bigl(\Delta^{(d)}_{\gamma,jk}\bigr)^{2}
    \;\le\;2\sigma_{\max}^4L_{dd}'^{\,2}\,\lvert\gamma_1-\gamma_2\rvert^{2},
\end{aligned}
\]
the first line by the Cauchy--Schwarz inequality $(\sum_jx_j)^2\le n\sum_jx_j^2$ and the interpolated fourth moment $\mathbb E[v^4\mid X]\le\bar M_4:=\bar M_{2p}^{2/p}$ of Assumption~\ref{ass:G}(v) (Lyapunov interpolation), the second by the conditional independence and centeredness of the entries, which leave the two orderings of each index pair as the only surviving pairings. Hence
\[
    \mathbb E\Bigl[\bigl(D_v(\gamma_1)-D_v(\gamma_2)\bigr)^{2}\Bigm|\mathbf X\Bigr]\;\le\;C_D\,\lvert\gamma_1-\gamma_2\rvert^{2},
    \qquad C_D:=\bigl(2\bar M_4+4\sigma_{\max}^4\bigr)L_{dd}'^{\,2} .
\]
For the chaining, let $\mathcal G_J:=\{\gamma_L+kh_J:0\le k\le2^J\}$ with $h_J:=(\gamma_U-\gamma_L)2^{-J}$ be the dyadic grids; the paths $\gamma\mapsto D_v(\gamma)$ are continuous by Fact~\ref{fact:pairdiff}(i), and the dyadic expansion of $\gamma-\gamma_L$ in binary digits, each partial sum and its successor being adjacent points of $\mathcal G_J$, telescopes the increment into level maxima,
\[
    \sup_{\gamma\in[\gamma_L,\gamma_U]}\bigl|D_v(\gamma)-D_v(\gamma_L)\bigr|
    \;\le\;\sum_{J\ge0}\ \max_{\substack{\gamma,\gamma'\in\mathcal G_J\\ \lvert\gamma-\gamma'\rvert=h_J}}\bigl|D_v(\gamma)-D_v(\gamma')\bigr| ,
\]
the limit term vanishing by path continuity; each level contains $2^J$ adjacent pairs, so the Cauchy--Schwarz inequality for maxima within a level, followed by integration over $\mathbf X$ of the modulus display, gives
\[
    \mathbb E\Bigl[\max_{\lvert\gamma-\gamma'\rvert=h_J}\bigl|D_v(\gamma)-D_v(\gamma')\bigr|\Bigr]
    \;\le\;\bigl(2^J\bigr)^{1/2}\bigl(C_Dh_J^2\bigr)^{1/2}
    \;=\;C_D^{1/2}(\gamma_U-\gamma_L)\,2^{-J/2},
\]
and the geometric series at ratio $2^{-1/2}$ sums to $(1-2^{-1/2})^{-1}$, whence $\mathbb E\bigl[\sup_\gamma|D_v(\gamma)-D_v(\gamma_L)|\bigr]\le K_D:=C_D^{1/2}(\gamma_U-\gamma_L)(1-2^{-1/2})^{-1}$. Combining with the base-point bound, $\mathbb E[\sup_\gamma D_v(\gamma)]\le\sigma_{\max}^2L_d^2+K_D<\infty$, so $\sup_\gamma D_v(\gamma)=O_p(1)$ by Markov's inequality; with the mean-part envelope of the regime-coverage paragraph this bounds the first term, $\sup_{\gamma_1\ne\gamma_2}n^{-1}\boldsymbol\varepsilon_0^\top\bar{\boldsymbol W}\boldsymbol\varepsilon_0=O_p(1)$ in both regimes.

\emph{The projected-score term.} For the second term, put $\varphi_j:=(\tilde u_{\gamma_1}(X_j)-\tilde u_{\gamma_2}(X_j))/T$, so that $\bar{\boldsymbol W}$ is the Gram matrix of the $\varphi_j$, of norms at most $L_k'$ by the feature-distance bound, and $n^{-1}\boldsymbol a^\top\bar{\boldsymbol W}\boldsymbol a=\bigl\lVert n^{-1/2}\sum_ja_j\varphi_j\bigr\rVert^2$ for every $\boldsymbol a\in\mathbb R^n$; the tight entrywise form of the Cauchy--Schwarz inequality, $n^{-1}(\sum_j|x_j|)^2\le\sum_jx_j^2$, then gives the deterministic majorization
\[
    n^{-1}\boldsymbol a^\top\bar{\boldsymbol W}\boldsymbol a
    \;\le\;n^{-1}\Bigl(\sum_j|a_j|\,\lVert\varphi_j\rVert\Bigr)^{2}
    \;\le\;W\,\lVert\boldsymbol a\rVert^{2},
    \qquad\text{that is,}\qquad n^{-1}\bar{\boldsymbol W}\preceq W\boldsymbol I_n ,
\]
uniformly over the pair. With $\boldsymbol a=H\boldsymbol\varepsilon_0$, and $H$ symmetric idempotent (the orthogonal projector onto $\operatorname{range}(\hat{\mathbb G})$, of rank $d_\theta$ on the invertibility events of Fact~\ref{fact:residual_split}),
\[
    n^{-1}(H\boldsymbol\varepsilon_0)^\top\bar{\boldsymbol W}(H\boldsymbol\varepsilon_0)
    \;\le\;W\,\lVert H\boldsymbol\varepsilon_0\rVert^{2}
    \;=\;W\,\boldsymbol\varepsilon_0^\top H\boldsymbol\varepsilon_0 ,
\]
uniformly over the pair, and $\boldsymbol\varepsilon_0^\top H\boldsymbol\varepsilon_0=O_p(1)$ under the convention of Remark~\ref{rem:cond_det}, which licenses the adaptive projector: by \eqref{eq:quad4} at $\lVert\cdot\rVert_{\mathrm F}=\sqrt{d_\theta}$ for the centered part and the tracial centering $\lvert\operatorname{tr}(HD)\rvert\le\sigma_{\max}^{2}d_\theta$, in both regimes.

\emph{The estimation remainder.} For the third term, the majorization just displayed extends the quadratic-form clause of Fact~\ref{fact:residual_split} from $\boldsymbol K_\gamma/n$ to $n^{-1}\bar{\boldsymbol W}$: the recorded splits of the proof of Lemma~\ref{lemma:estimation_remainder} consume, beside the projector identities, only positive semidefiniteness and the contraction of the quadratic-form matrix, and $n^{-1}\bar{\boldsymbol W}$ shares both, with $W$ in place of one ($\lambda_{\max}(n^{-1}\bar{\boldsymbol W})\le W$). In the $X$-measurable regime only the second-order piece is present, and
\[
    n^{-1}\boldsymbol r_{\mathrm{est}}^\top\bar{\boldsymbol W}\boldsymbol r_{\mathrm{est}}
    \;\le\;W\,\lVert\boldsymbol r_{\mathrm{proj}}\rVert^{2}
    \;=\;O_p\bigl(n^{1-4\varrho}\bigr)
    \;=\;o_p(1) ,
\]
the rate being the second-order computation of that proof, whose exponent is negative under $\varrho>1/4$. In the sieve regime the second-order and sieve-mean pieces obey the same majorization,
\[
\begin{aligned}
    n^{-1}\boldsymbol r_{\mathrm{proj}}^\top\bar{\boldsymbol W}\boldsymbol r_{\mathrm{proj}}
    &\le W\,\lVert\boldsymbol r_{\mathrm{proj}}\rVert^{2}=O_p\bigl(n^{1-4\varrho}\bigr)=o_p(1),\\
    n^{-1}\boldsymbol r_{\mathrm{sieve}}^\top\bar{\boldsymbol W}\boldsymbol r_{\mathrm{sieve}}
    &\le W\,\lVert\boldsymbol g-\hat{\boldsymbol g}\rVert_{\mathrm F}^{2}\,\lVert\boldsymbol\delta_\theta\rVert^{2}
    =o_p\bigl(n^{2\varrho}\bigr)\,O_p\bigl(n^{-2\varrho}\bigr)=o_p(1),
\end{aligned}
\]
the second line by the contraction $\lVert\boldsymbol I-P_Z\rVert_{\mathrm{op}}\le1$ and the sieve estimation-error rate of Assumption~\ref{ass:F}(iii) in the form recorded in that proof. The centered piece requires its fluctuation structure and is reduced to one dimension first: by Fact~\ref{fact:pairdiff}(ii) at $\boldsymbol a=\boldsymbol r_{\mathrm{center}}$ and the symmetry of $\boldsymbol I-P_Z$,
\[
    \sup_{\gamma_1\ne\gamma_2}\boldsymbol r_{\mathrm{center}}^\top\bigl(n^{-1}\bar{\boldsymbol W}\bigr)\boldsymbol r_{\mathrm{center}}
    \;\le\;\sup_{\gamma}\ \boldsymbol\delta_\theta^\top\tilde A(\gamma)\,\boldsymbol\delta_\theta ,
    \qquad \tilde A(\gamma):=\boldsymbol\eta_0^\top M_\gamma\boldsymbol\eta_0 ,
\]
with $M_\gamma:=\bigl(\boldsymbol I-P_Z\bigr)\bigl(n^{-1}G^{(d)}_\gamma\bigr)\bigl(\boldsymbol I-P_Z\bigr)$,
where $\boldsymbol r_{\mathrm{center}}=(\boldsymbol I-P_Z)\boldsymbol\eta_0\boldsymbol\delta_\theta$ is the split of that proof, the identity following from the Gram representation of Fact~\ref{fact:pairdiff}(ii) applied to $\boldsymbol a=(\boldsymbol I-P_Z)\boldsymbol\eta_0\boldsymbol\delta_\theta$, the symmetry of $\boldsymbol I-P_Z$ transferring the projector from the weights to the Gram. The same fourth-moment computation as in the centered part of the proof of Lemma~\ref{lemma:estimation_remainder}, applied conditionally on $\mathbf X$ to the $\mathbf X$-measurable matrices $M_\gamma$, gives, with finite constants $C_A,C_A'$ depending only on $(d_\theta,L_d,L_{dd}')$ and the score moments of Assumption~\ref{ass:E}(iii),
\[
    \mathbb E\bigl\lVert\tilde A(\gamma_L)\bigr\rVert_{\mathrm F}^{2}\le C_A,
    \qquad
    \mathbb E\bigl\lVert\tilde A(\gamma_1)-\tilde A(\gamma_2)\bigr\rVert_{\mathrm F}^{2}\le C_A'\,\lvert\gamma_1-\gamma_2\rvert^{2} ,
\]
using $\lVert M_{\gamma_L}\rVert_{\mathrm F}\le\lVert n^{-1}G^{(d)}_{\gamma_L}\rVert_{\mathrm F}\le L_d^2$ and $\lVert M_{\gamma_1}-M_{\gamma_2}\rVert_{\mathrm F}\le\lVert n^{-1}(G^{(d)}_{\gamma_1}-G^{(d)}_{\gamma_2})\rVert_{\mathrm F}\le L_{dd}'\lvert\gamma_1-\gamma_2\rvert$ entrywise, the entrywise increment bound having been displayed above; the same dyadic chaining as for $D_v$, applied to the field $\gamma\mapsto\lVert\tilde A(\gamma)\rVert_{\mathrm F}$ (continuous paths by the same increment bound), yields $\sup_\gamma\lVert\tilde A(\gamma)\rVert_{\mathrm F}=O_p(1)$, whence
\[
    \sup_{\gamma_1\ne\gamma_2}\boldsymbol r_{\mathrm{center}}^\top\bigl(n^{-1}\bar{\boldsymbol W}\bigr)\boldsymbol r_{\mathrm{center}}
    \;\le\;\lVert\boldsymbol\delta_\theta\rVert^{2}\ \sup_\gamma\lVert\tilde A(\gamma)\rVert_{\mathrm F}
    \;=\;O_p\bigl(n^{-2\varrho}\bigr)\,O_p(1)
    \;=\;o_p(1) .
\]
The cross terms between the three pieces are absorbed by the Cauchy--Schwarz inequality and $2ab\le a^2+b^2$ with the factor $3$, as in the proof of the fact. Hence $\sup_{\gamma_1\ne\gamma_2}n^{-1}\boldsymbol r_{\mathrm{est}}^\top\bar{\boldsymbol W}\boldsymbol r_{\mathrm{est}}=o_p(1)$, uniformly over the pair, the majorization and the one-dimensional reduction being pair-free.

Combining the three terms in the three-term bound,
\[
\begin{gathered}
    \lVert\widetilde M_n(\gamma_1)-\widetilde M_n(\gamma_2)\rVert_{\mathcal H_{\gamma_0}}
    \;\le\;L_M\,\lvert\gamma_1-\gamma_2\rvert,\\
    L_M^{2}:=3\sup_{\gamma_1\ne\gamma_2}\Bigl[n^{-1}\boldsymbol\varepsilon_0^\top\bar{\boldsymbol W}\boldsymbol\varepsilon_0
    +W\,\boldsymbol\varepsilon_0^\top H\boldsymbol\varepsilon_0
    +n^{-1}\boldsymbol r_{\mathrm{est}}^\top\bar{\boldsymbol W}\boldsymbol r_{\mathrm{est}}\Bigr],
\end{gathered}
\]
whose bracket is $O_p(1)+W\,O_p(1)+o_p(1)$ in both regimes, the middle summand being the convention-licensed projected-score term of the preceding step (Remark~\ref{rem:cond_det}); this proves the first half of part~(a).

\medskip\noindent\emph{Uniform boundedness.} By the modulus just established it suffices to bound the process at a single point: for any $\gamma_*\in[\gamma_L,\gamma_U]$,
\[
    \sup_{\gamma\in[\gamma_L,\gamma_U]}\lVert\widetilde M_n(\gamma)\rVert_{\mathcal H_{\gamma_0}}
    \;\le\;\lVert\widetilde M_n(\gamma_*)\rVert_{\mathcal H_{\gamma_0}}+L_M(\gamma_U-\gamma_L) .
\]
By the isometry $J_{\gamma_*}$ and the reproducing property, $\lVert\widetilde M_n(\gamma_*)\rVert_{\mathcal H_{\gamma_0}}^2=\lVert M_n(\gamma_*)\rVert_{\mathcal H_{\gamma_*}}^2=n^{-1}\boldsymbol\varepsilon_\perp^\top\boldsymbol K_{\gamma_*}\boldsymbol\varepsilon_\perp$, the normalized Gram matrix being positive semidefinite and contracting, $\lambda_{\max}(\boldsymbol K_{\gamma_*}/n)\le\operatorname{tr}(\boldsymbol K_{\gamma_*}/n)=1$ ($k_\gamma(x,x)=1$ by unit normalization). Expanding the square of the residual decomposition with the Cauchy--Schwarz inequality and $2ab\le a^2+b^2$, as in the uniformity step,
\[
    n^{-1}\boldsymbol\varepsilon_\perp^\top\boldsymbol K_{\gamma_*}\boldsymbol\varepsilon_\perp
    \;\le\;3\,Q(\gamma_*)+3\,\boldsymbol\varepsilon_0^\top H\boldsymbol\varepsilon_0
    \;+\;3\,\boldsymbol r_{\mathrm{est}}^\top\bigl(\boldsymbol K_{\gamma_*}/n\bigr)\boldsymbol r_{\mathrm{est}},
    \qquad Q(\gamma):=n^{-1}\boldsymbol\varepsilon_0^\top\boldsymbol K_\gamma\boldsymbol\varepsilon_0\ge0 ,
\]
where the third term is $o_p(1)$ uniformly in $\gamma$ by the quadratic-form clause of Fact~\ref{fact:residual_split} and the second is the convention-licensed projected-score term of the projected-score step. For the first, under $H_0$ the diagonal carries the whole expectation, the off-diagonal expectations vanishing under $\mathbb E[\varepsilon_0\mid X]=0$:
\[
    \mathbb E\bigl[Q(\gamma_*)\bigr]
    =n^{-1}\sum_j\mathbb E\bigl[\varepsilon_{0,j}^2\bigr]\,k_{\gamma_*}(X_j,X_j)
    =\mathbb E\bigl[\varepsilon_0^2\bigr] ,
\]
finite by Assumption~\ref{ass:B}(v), so $Q(\gamma_*)=O_p(1)$ by Markov's inequality. Under $H_{1,n}$, write $\boldsymbol\varepsilon_0=\boldsymbol g+\boldsymbol v$ as in the regime-coverage paragraph and split $Q(\gamma_*)=Q^{\boldsymbol v}(\gamma_*)+2\,n^{-1}\boldsymbol g^\top\boldsymbol K_{\gamma_*}\boldsymbol v+Q^{\boldsymbol g}(\gamma_*)$; the same diagonal computation, with $\mathbb E[v^2\mid X]=\sigma^2(X)$ and the off-diagonal expectations vanishing, gives $\mathbb E[Q^{\boldsymbol v}(\gamma_*)]\le\sigma_{\max}^2$ (Assumption~\ref{ass:G}(iii)), hence $Q^{\boldsymbol v}(\gamma_*)=O_p(1)$ by Markov's inequality; next,
\[
    Q^{\boldsymbol g}(\gamma_*)\;\le\;\lambda_{\max}\bigl(\boldsymbol K_{\gamma_*}/n\bigr)\,\lVert\boldsymbol g\rVert^2
    \;\le\;\lVert\boldsymbol g\rVert^2
    \;=\;O_p(1)
\]
by \eqref{eq:pitman_lln} read at the raw shift function; and the cross term is conditionally centered given $\mathbf X$, with conditional variance
\[
    \operatorname{Var}\bigl(n^{-1}\boldsymbol g^\top\boldsymbol K_{\gamma_*}\boldsymbol v\bigm|\mathbf X\bigr)
    =n^{-2}\sum_j\sigma^2(X_j)\,\bigl[(\boldsymbol K_{\gamma_*}\boldsymbol g)_j\bigr]^2
    \;\le\;\sigma_{\max}^2\,n^{-1}\lVert\boldsymbol g\rVert^2 ,
\]
by the row bound $(\boldsymbol K_{\gamma_*}\boldsymbol g)_j\le\lVert\boldsymbol g\rVert$ (the Gram rows have unit norm, $k_\gamma(x,x)=1$), a data-random bound whose expectation is $\sigma_{\max}^2\,n^{-1}\,\mathbb E\lVert\boldsymbol g\rVert^2=O(1/n)$ by the expectation form of \eqref{eq:pitman_lln}; since the cross term is conditionally centered, the iterated expectation gives $\mathbb E\bigl[(n^{-1}\boldsymbol g^\top\boldsymbol K_{\gamma_*}\boldsymbol v)^2\bigr]=O(1/n)$, and Chebyshev's inequality yields $n^{-1}\boldsymbol g^\top\boldsymbol K_{\gamma_*}\boldsymbol v=o_p(1)$. Hence $\lVert\widetilde M_n(\gamma_*)\rVert_{\mathcal H_{\gamma_0}}^2=O_p(1)+W\,O_p(1)+o_p(1)$, and with the base-point triangle this completes the boundedness half.

\noindent\textbf{(b) Covariance.}\; The operators $C_{n,\gamma_1}$ and $C_{n,\gamma_2}$ act on the different spaces $\mathcal{H}_{\gamma_1}$, $\mathcal{H}_{\gamma_2}$, so their Hilbert--Schmidt difference is taken after transporting both to the common space $\mathcal{H}_{\gamma_0}$. Writing $\tilde u_j(\gamma) := J_\gamma k_\gamma(X_j,\cdot)$ for the transported kernels, $\lVert\tilde u_j(\gamma)\rVert_{\mathcal{H}_{\gamma_0}} = \lVert k_\gamma(X_j,\cdot)\rVert_{\mathcal{H}_\gamma} = \sqrt{k_\gamma(X_j,X_j)} = 1$ (the transport is an isometry, Lemma~\ref{lem:fourier_iso}, and the Gaussian kernel is unit-normalized), one has $\widetilde{C}_{n,\gamma} = n^{-1}\sum_j \tilde u_j(\gamma)\otimes \tilde u_j(\gamma)$. Rank-one telescoping makes the Lipschitz bound deterministic. With $f_j:=\tilde u_j(\gamma_1)-\tilde u_j(\gamma_2)$, the feature-distance bound of part (a) gives $\lVert f_j\rVert_{\mathcal H_{\gamma_0}}\le L_k'\,|\gamma_1-\gamma_2|$ for every $j$, and the exact operator identity
$\tilde u_j(\gamma_1)\otimes\tilde u_j(\gamma_1)-\tilde u_j(\gamma_2)\otimes\tilde u_j(\gamma_2)
 =f_j\otimes\tilde u_j(\gamma_1)+\tilde u_j(\gamma_2)\otimes f_j$
(check by expanding the right side) together with the triangle inequality for the Hilbert--Schmidt norm gives
\[
    \lVert\widetilde C_{n,\gamma_1}-\widetilde C_{n,\gamma_2}\rVert_{\mathrm{HS}}
    \le\frac1n\sum_j\bigl(\lVert f_j\rVert\,\lVert\tilde u_j(\gamma_1)\rVert+\lVert\tilde u_j(\gamma_2)\rVert\,\lVert f_j\rVert\bigr)
    \le 2L_k'\,T,
\]
using $\lVert\tilde u_j(\gamma)\rVert_{\mathcal H_{\gamma_0}}=1$. Hence
\[
    \lVert\widetilde C_{n,\gamma_1}-\widetilde C_{n,\gamma_2}\rVert_{\mathrm{HS}}\le L_C\,T
\]
with the \emph{deterministic} constant $L_C:=2L_k'$: no stochastic input enters the Lipschitz constant.
For the uniform operator convergence, cover $[\gamma_L,\gamma_U]$ by a grid of mesh $\delta$. At each grid point the two operators act on the same space $\mathcal{H}_{\gamma_j}$, and $\lVert C_{n,\gamma_j}-C_{\gamma_j}\rVert_{\mathrm{op}} \le \lVert C_{n,\gamma_j}-C_{\gamma_j}\rVert_{\mathrm{HS}} = o_p(1)$ (Lemma~\ref{lm:C_and_Cn}; the operator norm is dominated by the Hilbert--Schmidt norm). Between grid points, the Lipschitz bound just established applies verbatim to the population operators, by the pointwise kernel Lipschitz bound: $\lVert C_{\gamma_a}-C_{\gamma_b}\rVert_{\mathrm{op}} \le \lVert C_{\gamma_a}-C_{\gamma_b}\rVert_{\mathrm{HS}} \le L_C^\circ|\gamma_a-\gamma_b|$ with the deterministic population constant $L_C^\circ := C(\gamma_L)$, since
$\lVert C_{\gamma_a}-C_{\gamma_b}\rVert_{\mathrm{HS}}^2
=\mathbb E\bigl[\bigl(k_{\gamma_a}(X,X')-k_{\gamma_b}(X,X')\bigr)^2\bigr]\le C(\gamma_L)^2|\gamma_a-\gamma_b|^2$
by the same pointwise bound. Hence, for $r_n := \sup_\gamma\lVert C_{n,\gamma}-C_\gamma\rVert_{\mathrm{op}}$
\[
    r_n \le \max_j\lVert C_{n,\gamma_j}-C_{\gamma_j}\rVert_{\mathrm{op}} + (L_C + L_C^\circ)\delta = o_p(1) + O_p(1)\delta,
\]
and since $\limsup_{n\to\infty} r_n \le (L_C + L_C^\circ)\delta$ for every $\delta > 0$, sending $\delta \to 0$ after $n \to \infty$ yields $r_n = o_p(1)$.

\noindent\textbf{(c) Resolvent.}\; Transport to the common space: $\widetilde{A}_n(\eta) = (\widetilde{C}_{n,\gamma}+\lambda I)^{-1}$, where $\widetilde{C}_{n,\gamma} = J_\gamma C_{n,\gamma} J_\gamma^{-1}$ is self-adjoint and positive semidefinite on $\mathcal{H}_{\gamma_0}$ (unitary conjugation, Lemma~\ref{lem:fourier_iso}), so $\widetilde{C}_{n,\gamma}+\lambda I \ge \lambda I$ and $\lVert\widetilde{A}_n(\eta)\rVert_{\mathrm{op}} \le \lambda^{-1} \le \lambda_L^{-1}$. For $\eta_1,\eta_2$, set $A_j := \widetilde{C}_{n,\gamma_j}+\lambda_j I$; the resolvent identity $A_1^{-1}-A_2^{-1} = A_1^{-1}(A_2-A_1)A_2^{-1}$ gives
\[
    \lVert\widetilde{A}_n(\eta_1)-\widetilde{A}_n(\eta_2)\rVert_{\mathrm{op}} \le \lVert\widetilde{A}_n(\eta_1)\rVert_{\mathrm{op}}\,\lVert A_2-A_1\rVert_{\mathrm{op}}\,\lVert\widetilde{A}_n(\eta_2)\rVert_{\mathrm{op}} \le \lambda_L^{-2}\,\lVert(\widetilde{C}_{n,\gamma_1}-\widetilde{C}_{n,\gamma_2})+(\lambda_1-\lambda_2)I\rVert_{\mathrm{op}},
\]
and $\lVert\widetilde{C}_{n,\gamma_1}-\widetilde{C}_{n,\gamma_2}\rVert_{\mathrm{op}}\le\lVert\widetilde{C}_{n,\gamma_1}-\widetilde{C}_{n,\gamma_2}\rVert_{\mathrm{HS}}\le L_C|\gamma_1-\gamma_2|$ by~(b), while $\lVert(\lambda_1-\lambda_2)I\rVert_{\mathrm{op}}=|\lambda_1-\lambda_2|$, yielding the displayed bound. No spectral gap or eigenfunction tracking enters: $\widetilde{A}_n(\eta)=f(\widetilde{C}_{n,\gamma})$ for the smooth map $f(t)=(t+\lambda)^{-1}$, which is operator-Lipschitz.

Finally, the closing statement on fixed alternatives holds by the entrywise envelope $n^{-1}\boldsymbol\varepsilon_\perp^\top\Delta\boldsymbol G\,\boldsymbol\varepsilon_\perp\le W\lVert\boldsymbol\varepsilon_\perp\rVert^2=O_p(n)$ (so $L_M=O_p(\sqrt n)$) together with the Gram contraction $\sup_\gamma n^{-1}\boldsymbol\varepsilon_\perp^\top\boldsymbol K_\gamma\boldsymbol\varepsilon_\perp\le\lVert\boldsymbol\varepsilon_\perp\rVert^2=O_p(n)$ (so $\overline M_n=O_p(\sqrt n)$); parts (b) and (c) carry no residuals.
\end{proof}

The following corollary packages the Lipschitz-modulus equicontinuity of the transported statistic process for standalone citation.

\begin{corollary}[Lipschitz-modulus equicontinuity]\label{cor:covering}
The process $S_n(\eta):=\lVert\widetilde A_n(\eta)\widetilde M_n(\gamma)\rVert_{\mathcal H_{\gamma_0}}^2$, $\eta=(\gamma,\lambda)\in\Upsilon$, satisfies, for all $\eta_1,\eta_2\in\Upsilon$,
\[
    |S_n(\eta_1)-S_n(\eta_2)|\le L_n\lVert\eta_1-\eta_2\rVert,
    \qquad
    L_n:=2\lambda_L^{-1}\overline M_n\bigl(\lambda_L^{-2}L_C\overline M_n+\lambda_L^{-1}L_M\bigr)+2\lambda_L^{-3}\overline M_n^{\,2}=O_p(1),
\]
where $\overline M_n:=\sup_{\gamma\in[\gamma_L,\gamma_U]}\lVert\widetilde M_n(\gamma)\rVert_{\mathcal H_{\gamma_0}}=O_p(1)$ and $L_M,L_C$ are as in parts (a) and (b) of Lemma~\ref{lem:lipschitz_common}; consequently $\sup_{\eta\in\Upsilon}S_n(\eta)\le\lambda_L^{-2}\overline M_n^{\,2}=O_p(1)$, and the process is asymptotically stochastically equicontinuous: for every $\varepsilon>0$,
\[
    \lim_{\delta\to0}\limsup_{n\to\infty}\mathbb P\Bigl(\sup_{\lVert\eta_1-\eta_2\rVert\le\delta}|S_n(\eta_1)-S_n(\eta_2)|>\varepsilon\Bigr)=0.
\]
\end{corollary}

\begin{proof}
Write $\widetilde A_j:=\widetilde A_n(\eta_j)$ and $\widetilde M_j:=\widetilde M_n(\gamma_j)$ for $\eta_j=(\gamma_j,\lambda_j)$, $j=1,2$. The three gap-free bounds are as follows.
(R) $\lVert\widetilde A_1-\widetilde A_2\rVert_{\mathrm{op}}\le\lambda_L^{-2}(L_C|\gamma_1-\gamma_2|+|\lambda_1-\lambda_2|)$ and $\lVert\widetilde A_n(\eta)\rVert_{\mathrm{op}}\le\lambda_L^{-1}$ (Lemma~\ref{lem:lipschitz_common}(c)).
(E) $\lVert\widetilde M_1-\widetilde M_2\rVert\le L_M|\gamma_1-\gamma_2|$ (Lemma~\ref{lem:lipschitz_common}(a)).
(B) $\overline M_n=O_p(1)$ (Lemma~\ref{lem:lipschitz_common}(a)).
The reverse triangle inequality on the norm and the telescope give
\[
    |S_n(\eta_1)-S_n(\eta_2)| \le 2\lambda_L^{-1}\overline M_n\,\lVert\widetilde A_1\widetilde M_1-\widetilde A_2\widetilde M_2\rVert,
\]
and $\lVert\widetilde A_1\widetilde M_1-\widetilde A_2\widetilde M_2\rVert \le \lVert\widetilde A_1-\widetilde A_2\rVert_{\mathrm{op}}\lVert\widetilde M_1\rVert + \lVert\widetilde A_2\rVert_{\mathrm{op}}\lVert\widetilde M_1-\widetilde M_2\rVert \le (\lambda_L^{-2}L_C\overline M_n+\lambda_L^{-1}L_M)|\gamma_1-\gamma_2| + \lambda_L^{-2}\overline M_n|\lambda_1-\lambda_2|$, whence the displayed Lipschitz bound with $L_n$ as defined; all ingredients are $O_p(1)$ and $\lambda_L$ is deterministic, so $L_n=O_p(1)$. The uniform bound $\sup_{\eta\in\Upsilon}S_n(\eta)\le\lambda_L^{-2}\overline M_n^{\,2}$ combines $\lVert\widetilde A_n(\eta)\rVert_{\mathrm{op}}\le\lambda_L^{-1}$ with the definition of $\overline M_n$. Covering argument: on $\{L_n\le M\}$, $S_n$ is $M$-Lipschitz on the compact set $\Upsilon$, so with mesh $\delta=\varepsilon/(2M)$, $\sup_{\lVert\eta_1-\eta_2\rVert\le\delta}|S_n(\eta_1)-S_n(\eta_2)|\le\varepsilon/2$; choosing $M$ with $\mathbb P(L_n>M)\le\alpha$ (possible as $L_n=O_p(1)$) gives $\limsup_{n\to\infty}\mathbb P(\sup_{\lVert\eta_1-\eta_2\rVert\le\delta}|S_n(\eta_1)-S_n(\eta_2)|>\varepsilon)\le\alpha$; letting $\alpha\downarrow0$ (with $\delta\downarrow0$) yields the displayed asymptotic stochastic equicontinuity.
\end{proof}

\begin{lemma}[Conditional Lipschitz modulus of the multiplied embedding]\label{lem:cond_lip_mult}
Let
\[
    \widetilde M_n^*(\gamma)
    :=J_\gamma\,n^{-1/2}\sum_{j=1}^n\bigl(\hat{\boldsymbol{\Pi}}(\boldsymbol{\varepsilon}_0\odot\boldsymbol V)\bigr)_j\,k_\gamma(X_j,\cdot)\in\mathcal H_{\gamma_0}
\]
be the multiplied-then-projected embedding transported to the common space, where $J_\gamma$ is the isometry of Lemma~\ref{lem:lipschitz_common} and $V_1,\dots,V_n$ are i.i.d.\ multipliers with $\mathbb E[V]=0$, $\mathbb E[V^2]=1$, and $|V|\le\bar V$ a.s.\ (e.g.\ the Mammen multipliers). Under Assumptions~\ref{ass:B}(v), \ref{ass:C}(i), and~\ref{ass:D}(iii), the score conditions of Assumptions~\ref{ass:E}(i)--(ii) (the first through the uniform score bound $\sup_{\theta\in\Theta}\sup_{x\in\mathcal X}\lVert g_\theta(x)\rVert\le\bar G$ of Assumption~\ref{ass:E}(i)), and, in the sieve regime, Assumption~\ref{ass:F}(iii); under the local alternative, Assumption~\ref{ass:G}(iii) together with the standing convention $\mathbb E[R^{2}]<\infty$ of Section~\ref{sec:cv_selection} is additionally invoked. All statements are conditional on the data:
\begin{enumerate}[label=(\roman*)]
    \item there exists $L_M^*=O_p(1)$ with
    $\lVert\widetilde M_n^*(\gamma_1)-\widetilde M_n^*(\gamma_2)\rVert_{\mathcal H_{\gamma_0}}\le L_M^*|\gamma_1-\gamma_2|$ for all $\gamma_1,\gamma_2\in[\gamma_L,\gamma_U]$, the bound holding conditionally on the data in $\mathbb P^*$-probability and hence jointly over the data and the multipliers;
    \item $\overline M_n^*:=\sup_{\gamma\in[\gamma_L,\gamma_U]}\lVert\widetilde M_n^*(\gamma)\rVert_{\mathcal H_{\gamma_0}}=O_p(1)$.
\end{enumerate}
\end{lemma}

\begin{proof}
Abbreviate $\boldsymbol a:=\hat{\boldsymbol\Pi}(\boldsymbol{\varepsilon}_0\odot\boldsymbol V)$ and write $\mathbb E^*$ and $\mathbb P^*$ for the conditional expectation and conditional probability given the data $\mathcal D_n:=(\mathbf X,\boldsymbol{\varepsilon}_0,\hat\theta)$, of which the multipliers are independent and which fixes $\hat{\boldsymbol\Pi}$, through $\hat{\mathbb G}$ and $M:=(\hat{\mathbb G}^\top\hat{\mathbb G})^{-1}$, and $\boldsymbol{\varepsilon}_0$; the inverse is understood on the invertibility events of Fact~\ref{fact:residual_split}, as there. For a pair $(\gamma_1,\gamma_2)$ with $\gamma_1\ne\gamma_2$, write $T:=|\gamma_1-\gamma_2|$.

\emph{Multiplier covariance.} By the independence of the multipliers, $\mathbb E^*[V_lV_m]=\delta_{lm}$, so the $(j,m)$-entry of the conditional second-moment matrix of $\boldsymbol a$ is
\[
    \bigl(\mathbb E^*[\boldsymbol a\boldsymbol a^\top]\bigr)_{jm}
    =\sum_{l,l'=1}^n\hat{\boldsymbol\Pi}_{jl}\hat{\boldsymbol\Pi}_{ml'}\,\varepsilon_{0,l}\varepsilon_{0,l'}\,\mathbb E^*[V_lV_{l'}]
    =\sum_{l=1}^n\hat{\boldsymbol\Pi}_{jl}\hat{\boldsymbol\Pi}_{ml}\,\varepsilon_{0,l}^{2},
\]
that is, $\mathbb E^*[\boldsymbol a\boldsymbol a^\top]=\hat{\boldsymbol\Pi}\boldsymbol D_0\hat{\boldsymbol\Pi}$ with $\boldsymbol D_0:=\operatorname{diag}(\varepsilon_{0,1}^{2},\dots,\varepsilon_{0,n}^{2})$.

\emph{Gram reduction.} The Gram--matrix reduction of Lemma~\ref{lem:lipschitz_common}(a) is coefficient-free, the feature geometry $\Delta\boldsymbol G_{jk}=\langle f(X_j),f(X_k)\rangle_{\mathcal H_{\gamma_0}}$ being derived there for arbitrary coefficient vectors, so
\[
    \lVert\widetilde M_n^*(\gamma_1)-\widetilde M_n^*(\gamma_2)\rVert_{\mathcal H_{\gamma_0}}^{2}
    =T^{2}\,n^{-1}\,\boldsymbol a^\top\bar{\boldsymbol W}\boldsymbol a,
    \qquad
    \bar{\boldsymbol W}:=T^{-2}\Delta\boldsymbol G,
\]
where $\bar{\boldsymbol W}$ is symmetric positive semidefinite and free of the weights $\boldsymbol a$. By Fact~\ref{fact:pairdiff}(ii), applied pathwise in the realized multipliers at the weights $\boldsymbol a$, the pair supremum of the quadratic form is bounded by the one-dimensional supremum of the derivative-feature process,
\begin{equation}\label{eq:condlip_pairdiff}
    \sup_{\gamma_1\ne\gamma_2}n^{-1}\boldsymbol a^\top\bar{\boldsymbol W}(\gamma_1,\gamma_2)\boldsymbol a
    \;\le\;\sup_{\gamma\in[\gamma_L,\gamma_U]}D_{\boldsymbol a}(\gamma),
    \qquad
    D_{\boldsymbol a}(\gamma):=\Bigl\lVert n^{-1/2}\sum_{j=1}^na_j\,\partial_\gamma\tilde u_\gamma(X_j)\Bigr\rVert_{\mathcal H_{\gamma_0}}^{2},
\end{equation}
so the pair square never enters as an index set. The remainder of the proof bounds $\sup_{\gamma\in[\gamma_L,\gamma_U]}D_{\boldsymbol a}(\gamma)$ conditionally on the data, mirroring the noise chain of Lemma~\ref{lem:lipschitz_common}(a) with a data-random modulus constant.

\emph{Leverage control.} All coefficient matrices below are data-measurable and reach the multipliers only through $\hat{\boldsymbol\Pi}$; a single entrywise envelope covers them, and no positive definiteness is required. Write $\boldsymbol H:=\hat{\mathbb G}M\hat{\mathbb G}^\top$ for the hat matrix, so that $\hat{\boldsymbol\Pi}=\boldsymbol I-\boldsymbol H$; $\boldsymbol H$ is symmetric because $M$ is, and idempotent because $M(\hat{\mathbb G}^\top\hat{\mathbb G})=\boldsymbol I$. The claim is
\begin{equation}\label{eq:leverage_generic}
    \bigl\lvert(\hat{\boldsymbol\Pi}A\hat{\boldsymbol\Pi})_{lm}\bigr\rvert
    \;\le\;W^{(e)}\bigl(1+2\hbar_n^{1/2}+\hbar_n\bigr)
    \quad\text{for all }l,m,
    \qquad
    \hbar_n:=n\lVert M\rVert_{\mathrm{op}}\sup_{l\le n}\lVert\hat g_l\rVert^{2},
\end{equation}
for every symmetric data-measurable $A\in\mathbb R^{n\times n}$ with entrywise envelope $|A_{lm}|\le W^{(e)}$. Expanding $\hat{\boldsymbol\Pi}A\hat{\boldsymbol\Pi}=(\boldsymbol I-\boldsymbol H)A(\boldsymbol I-\boldsymbol H)$,
\[
    (\hat{\boldsymbol\Pi}A\hat{\boldsymbol\Pi})_{lm}
    =A_{lm}-(\boldsymbol HA)_{lm}-(\boldsymbol AH)_{lm}+(\boldsymbol HAH)_{lm},
\]
the first term is at most $W^{(e)}$ in absolute value by the envelope. By symmetry of $A$ the two middle terms are transposes, $(\boldsymbol AH)_{lm}=(\boldsymbol HA)_{ml}$, and the Cauchy--Schwarz inequality with the idempotence gives
\[
    \bigl\lvert(\boldsymbol HA)_{lm}\bigr\rvert
    =\bigl\lvert\langle\boldsymbol H\boldsymbol e_l,A\boldsymbol e_m\rangle\bigr\rvert
    \le\lVert\boldsymbol H\boldsymbol e_l\rVert\,\lVert A\boldsymbol e_m\rVert
    =H_{ll}^{1/2}\,\lVert A\boldsymbol e_m\rVert
    \le W^{(e)}(nH_{ll})^{1/2},
\]
since $\lVert\boldsymbol H\boldsymbol e_l\rVert^{2}=\boldsymbol e_l^\top\boldsymbol H^{2}\boldsymbol e_l=H_{ll}$ and $\lVert A\boldsymbol e_m\rVert^{2}=\sum_jA_{jm}^{2}\le nW^{(e)2}$, the column sharing the envelope by symmetry. The corner term follows from the Cauchy--Schwarz inequality and the operator norm,
\[
    \bigl\lvert(\boldsymbol HAH)_{lm}\bigr\rvert
    =\bigl\lvert\langle\boldsymbol H\boldsymbol e_l,A\boldsymbol H\boldsymbol e_m\rangle\bigr\rvert
    \le\lVert A\rVert_{\mathrm{op}}\,H_{ll}^{1/2}H_{mm}^{1/2}
    \le nW^{(e)}(H_{ll}H_{mm})^{1/2},
\]
where $\lVert A\rVert_{\mathrm{op}}\le\lVert A\rVert_{\mathrm F}\le\bigl(\sum_{j,k}A_{jk}^{2}\bigr)^{1/2}\le nW^{(e)}$, again without definiteness of $A$. The leverages are uniformly small: $H_{ll}=\hat g_l^\top M\hat g_l\le\lVert M\rVert_{\mathrm{op}}\lVert\hat g_l\rVert^{2}$ by Rayleigh--Ritz, so $nH_{ll}\le\hbar_n$ for every $l$, and collecting the four terms by the triangle inequality proves \eqref{eq:leverage_generic}.

The leverage scale is $O_p(1)$ in both regimes. The score part is bounded first: in the $X$-measurable regime $\hat g_l=g_{\hat\theta}(X_l)$, so $\sup_l\lVert\hat g_l\rVert\le\bar G$ deterministically by the uniform score bound of Assumption~\ref{ass:E}(i); in the sieve regime $\hat g_l=\hat g_{\hat\theta}(X_l)$ and the sup-norm clause of Assumption~\ref{ass:F}(iii) gives $\sup_l\lVert\hat g_l\rVert\le\sup_x\lVert g_{\theta_0}(x)\rVert+\lVert\hat g_{\hat\theta}-g_{\theta_0}\rVert_\infty\le\bar G+O_p(n^{-1/4})$, the population score sharing the bound at $\theta_0$. For the inverse, write $g_i:=g_{\theta_0}(X_i)$ and telescope the Gram entries,
\[
    \hat g_i\hat g_i^\top-g_ig_i^\top=(\hat g_i-g_i)\hat g_i^\top+g_i(\hat g_i-g_i)^\top,
\]
so that, by the triangle inequality for the Frobenius norm, which dominates the operator norm, and the rank-one identity $\lVert xy^\top\rVert_{\mathrm F}=\lVert x\rVert\,\lVert y\rVert$,
\[
    \Bigl\lVert\frac{\hat{\mathbb G}^\top\hat{\mathbb G}}{n}-\Gamma_{\theta_0}\Bigr\rVert_{\mathrm{op}}
    \le n^{-1}\sum_{i=1}^n\lVert\hat g_i-g_i\rVert\bigl(\lVert\hat g_i\rVert+\lVert g_i\rVert\bigr)
    +\Bigl\lVert n^{-1}\sum_{i=1}^ng_ig_i^\top-\Gamma_{\theta_0}\Bigr\rVert_{\mathrm{op}} .
\]
The first term is $o_p(1)$ in both regimes: in the $X$-measurable one the mean-value inequality in $\theta$ and Assumption~\ref{ass:D}(iii) give $\lVert\hat g_i-g_i\rVert\le\lVert\hat\theta-\theta_0\rVert\sup_\theta\lVert\nabla g_\theta(X_i)\rVert$, and $n^{-1}\sum_i\sup_\theta\lVert\nabla g_\theta(X_i)\rVert=O_p(1)$ by the law of large numbers under the second-derivative clause of Assumption~\ref{ass:E}(i) and Assumption~\ref{ass:C}(i); in the sieve one $\sup_i\lVert\hat g_i-g_i\rVert\le\lVert\hat g_{\hat\theta}-g_{\theta_0}\rVert_\infty=O_p(n^{-1/4})$; and $\lVert\hat g_i\rVert+\lVert g_i\rVert\le2\bar G+O_p(n^{-1/4})$ throughout. The second term is $o_p(1)$ by the law of large numbers under Assumptions~\ref{ass:E}(i) and~\ref{ass:C}(i), the entries of $g_ig_i^\top$ being integrable by the first clause of Assumption~\ref{ass:E}(i). Since $\Gamma_{\theta_0}$ is nonsingular (Assumption~\ref{ass:E}(ii)), Weyl's inequality bounds $\lambda_{\min}(\hat{\mathbb G}^\top\hat{\mathbb G}/n)$ away from zero with probability tending to one, whence $\lVert M\rVert_{\mathrm{op}}=\lVert n^{-1}(\hat{\mathbb G}^\top\hat{\mathbb G}/n)^{-1}\rVert_{\mathrm{op}}=O_p(n^{-1})$, and therefore $\hbar_n=n\lVert M\rVert_{\mathrm{op}}\sup_l\lVert\hat g_l\rVert^{2}=O_p(1)$ in both regimes.

\emph{Base point.} Introduce the derivative-feature Gram
\[
    G^{(d)}_{\gamma,jk}:=\bigl\langle\partial_\gamma\tilde u_\gamma(X_j),\partial_\gamma\tilde u_\gamma(X_k)\bigr\rangle_{\mathcal H_{\gamma_0}},
\]
so that $D_{\boldsymbol a}(\gamma)=n^{-1}\boldsymbol a^\top G^{(d)}_{\gamma}\boldsymbol a$ by the Gram representation of Fact~\ref{fact:pairdiff}(ii), and abbreviate $q_n:=n^{-1}\sum_{l=1}^n\varepsilon_{0,l}^{2}$. At the left endpoint, by the covariance step and the cyclicity of the trace,
\[
    \mathbb E^*\bigl[D_{\boldsymbol a}(\gamma_L)\bigr]
    =n^{-1}\operatorname{tr}\bigl(G^{(d)}_{\gamma_L}\hat{\boldsymbol\Pi}\boldsymbol D_0\hat{\boldsymbol\Pi}\bigr)
    =n^{-1}\sum_{l=1}^n\varepsilon_{0,l}^{2}\bigl(\hat{\boldsymbol\Pi}G^{(d)}_{\gamma_L}\hat{\boldsymbol\Pi}\bigr)_{ll}
    \le L_d^{2}\bigl(1+2\hbar_n^{1/2}+\hbar_n\bigr)q_n,
\]
the first equality by the Gram representation and the covariance step, the second moving both projectors onto the Gram by the cyclicity of the trace and reading off the diagonal of $\boldsymbol D_0$, and the inequality applying \eqref{eq:leverage_generic} at $A=G^{(d)}_{\gamma_L}$, which is symmetric, data-measurable, and of entrywise envelope $W^{(e)}=L_d^{2}$ by the Cauchy--Schwarz inequality against the derivative bound of Fact~\ref{fact:pairdiff}(i).

\emph{Conditional modulus.} For the increment, put $\Delta^{(d)}:=G^{(d)}_{\gamma_1}-G^{(d)}_{\gamma_2}$ and $B^{(d)}:=\hat{\boldsymbol\Pi}\Delta^{(d)}\hat{\boldsymbol\Pi}$, and split the difference at the diagonal of the multiplier quadratic form,
\[
    D_{\boldsymbol a}(\gamma_1)-D_{\boldsymbol a}(\gamma_2)
    =n^{-1}(\boldsymbol{\varepsilon}_0\odot\boldsymbol V)^\top B^{(d)}(\boldsymbol{\varepsilon}_0\odot\boldsymbol V)
    =\Xi_d+\Xi_o,
\]
with $\Xi_d:=n^{-1}\sum_{l=1}^n\varepsilon_{0,l}^{2}B^{(d)}_{ll}V_l^{2}$ and $\Xi_o:=n^{-1}\sum_{l\ne m}\varepsilon_{0,l}\varepsilon_{0,m}B^{(d)}_{lm}V_lV_m$, the expansion substituting $\boldsymbol a=\hat{\boldsymbol\Pi}(\boldsymbol{\varepsilon}_0\odot\boldsymbol V)$, which transfers both projections onto the Gram by the symmetry of $\hat{\boldsymbol\Pi}$, and collecting the diagonal and off-diagonal summands. The increment of the derivative Gram is Lipschitz: adding and subtracting $\partial_\gamma\tilde u_{\gamma_2}(X_l)$ in the first factor of each inner product, then applying the triangle inequality and the Cauchy--Schwarz inequality,
\[
\begin{aligned}
    \bigl|\Delta^{(d)}_{lm}\bigr|
    &\le\bigl\lVert\partial_\gamma\tilde u_{\gamma_1}(X_l)-\partial_\gamma\tilde u_{\gamma_2}(X_l)\bigr\rVert\,\bigl\lVert\partial_\gamma\tilde u_{\gamma_1}(X_m)\bigr\rVert
    +\bigl\lVert\partial_\gamma\tilde u_{\gamma_2}(X_l)\bigr\rVert\,\bigl\lVert\partial_\gamma\tilde u_{\gamma_1}(X_m)-\partial_\gamma\tilde u_{\gamma_2}(X_m)\bigr\rVert\\
    &\le2L_dL_{dd}\,\lvert\gamma_1-\gamma_2\rvert,
\end{aligned}
\]
uniformly in $l,m$, the last line by Fact~\ref{fact:pairdiff}(i); since $\Delta^{(d)}$ is symmetric and data-measurable, \eqref{eq:leverage_generic} at the envelope $2L_dL_{dd}\lvert\gamma_1-\gamma_2\rvert$ gives
\[
    \bigl|B^{(d)}_{lm}\bigr|\le2L_dL_{dd}\bigl(1+2\hbar_n^{1/2}+\hbar_n\bigr)\lvert\gamma_1-\gamma_2\rvert .
\]
By the independence of the multipliers, $\mathbb E^*[V_l^{2}V_m^{2}]=\delta_{lm}\mathbb E[V^{4}]+(1-\delta_{lm})$ with $\mathbb E[V^{4}]\le\bar V^{4}$, and the two parts of the split have conditional second moments
\[
\begin{aligned}
    \mathbb E^*\bigl[\Xi_d^{2}\bigr]
    &=n^{-2}\sum_{l,m}\varepsilon_{0,l}^{2}\varepsilon_{0,m}^{2}B^{(d)}_{ll}B^{(d)}_{mm}\,\mathbb E^*[V_l^{2}V_m^{2}]\\
    &\le\bar V^{4}n^{-2}\sum_l\varepsilon_{0,l}^{4}\bigl(B^{(d)}_{ll}\bigr)^{2}
    +n^{-2}\Bigl(\sum_l\varepsilon_{0,l}^{2}\bigl|B^{(d)}_{ll}\bigr|\Bigr)^{2}\\
    &\le\bigl(\bar V^{4}+1\bigr)\bigl(2L_dL_{dd}\bigr)^{2}\bigl(1+2\hbar_n^{1/2}+\hbar_n\bigr)^{2}\lvert\gamma_1-\gamma_2\rvert^{2}q_n^{2},
\end{aligned}
\]
the second step discarding the nonpositive diagonal subtraction inside the square of the off-diagonal pairings, and the last inserting the entrywise bound together with the inequality $n^{-2}\sum_lx_l^{4}\le\bigl(n^{-1}\sum_lx_l^{2}\bigr)^{2}$ (a consequence of $\bigl(\sum_ly_l\bigr)^{2}-\sum_ly_l^{2}=2\sum_{l<m}y_ly_m\ge0$ at $y_l=x_l^{2}$) for the first summand and the definition of $q_n$ for the second; and
\[
    \mathbb E^*\bigl[\Xi_o^{2}\bigr]
    =2n^{-2}\sum_{l\ne m}\varepsilon_{0,l}^{2}\varepsilon_{0,m}^{2}\bigl(B^{(d)}_{lm}\bigr)^{2}
    \le2\bigl(2L_dL_{dd}\bigr)^{2}\bigl(1+2\hbar_n^{1/2}+\hbar_n\bigr)^{2}\lvert\gamma_1-\gamma_2\rvert^{2}q_n^{2},
\]
the identity keeping only the two orderings $l'=l,m'=m$ and $l'=m,m'=l$ of the pairings $\mathbb E^*[V_lV_mV_{l'}V_{m'}$, the only ones in which every multiplier appears twice, each equal to one, and the inequality inserting the entrywise bound and $\sum_{l\ne m}x_l^{2}x_m^{2}\le\bigl(\sum_lx_l^{2}\bigr)^{2}$. Combining the two displays through $(a+b)^{2}\le2a^{2}+2b^{2}$,
\[
    \mathbb E^*\bigl[\bigl(D_{\boldsymbol a}(\gamma_1)-D_{\boldsymbol a}(\gamma_2)\bigr)^{2}\bigr]
    \;\le\;C_D^{*}\,\lvert\gamma_1-\gamma_2\rvert^{2},
    \qquad
    C_D^{*}:=\bigl(2\bar V^{4}+6\bigr)\bigl(2L_dL_{dd}\bigr)^{2}\bigl(1+2\hbar_n^{1/2}+\hbar_n\bigr)^{2}q_n^{2},
\]
a deterministic function of the data. Only second moments of $\boldsymbol{\varepsilon}_0$ are consumed, the multiplier fourth moment entering solely through the diagonal pairing; and $q_n=O_p(1)$: under $H_0$ the law of large numbers applies directly (Assumptions~\ref{ass:B}(v) and~\ref{ass:C}(i)), while under the local alternative $\varepsilon_{0,l}=R(X_l)/\sqrt{m+n}+v_l$ and $2ab\le a^{2}+b^{2}$ give
\[
    q_n\;\le\;\frac{2}{m+n}\,n^{-1}\sum_{l=1}^nR^{2}(X_l)+2\,n^{-1}\sum_{l=1}^nv_l^{2}\;=\;O_p(1),
\]
the two averages obeying the law of large numbers under the standing convention $\mathbb E[R^{2}]<\infty$ of Section~\ref{sec:cv_selection} and Assumption~\ref{ass:G}(iii). Hence $C_D^{*}=O_p(1)$, being a continuous function of the $O_p(1)$ quantities $\hbar_n$ and $q_n^{2}$ with deterministic constants.

\emph{Conditional chaining.} The paths $\gamma\mapsto D_{\boldsymbol a}(\gamma)$ are continuous at the realized weights, by Fact~\ref{fact:pairdiff}(i) and the finiteness of the sum. Let $\mathcal G_J:=\{\gamma_L+kh_J:0\le k\le2^J\}$ with $h_J:=(\gamma_U-\gamma_L)2^{-J}$ be the dyadic grids; the dyadic expansion of $\gamma-\gamma_L$ in binary digits, each partial sum and its successor being adjacent points of $\mathcal G_J$, telescopes the increment into level maxima,
\[
    \sup_{\gamma\in[\gamma_L,\gamma_U]}\bigl|D_{\boldsymbol a}(\gamma)-D_{\boldsymbol a}(\gamma_L)\bigr|
    \;\le\;\sum_{J\ge0}\ \max_{\substack{\gamma,\gamma'\in\mathcal G_J\\ \lvert\gamma-\gamma'\rvert=h_J}}\bigl|D_{\boldsymbol a}(\gamma)-D_{\boldsymbol a}(\gamma')\bigr|,
\]
the limit term vanishing by path continuity. Each level contains $2^J$ adjacent pairs, so the Cauchy--Schwarz inequality for maxima within a level, followed by the conditional expectation of the modulus display, gives
\[
    \mathbb E^*\Bigl[\max_{\lvert\gamma-\gamma'\rvert=h_J}\bigl|D_{\boldsymbol a}(\gamma)-D_{\boldsymbol a}(\gamma')\bigr|\Bigr]
    \;\le\;\bigl(2^J\bigr)^{1/2}\bigl(C_D^{*}h_J^{2}\bigr)^{1/2}
    \;=\;C_D^{*1/2}(\gamma_U-\gamma_L)\,2^{-J/2},
\]
and the geometric series at ratio $2^{-1/2}$ sums to $(1-2^{-1/2})^{-1}$, whence, combining with the base-point bound,
\[
    \mathbb E^*\Bigl[\sup_{\gamma\in[\gamma_L,\gamma_U]}D_{\boldsymbol a}(\gamma)\Bigr]
    \;\le\;L_d^{2}\bigl(1+2\hbar_n^{1/2}+\hbar_n\bigr)q_n+K_D^{*}
    \;=:\;B_n,
\]
with $K_D^{*}:=C_D^{*1/2}(\gamma_U-\gamma_L)\bigl(1-2^{-1/2}\bigr)^{-1}$, so that $B_n$ is a data-measurable $O_p(1)$ bound. Markov's inequality applied conditionally gives $\mathbb P^*\bigl(\sup_{\gamma}D_{\boldsymbol a}(\gamma)>M\bigr)\le B_n/M$ for every $M>0$, so
\[
    \mathbb P\Bigl(\sup_{\gamma}D_{\boldsymbol a}(\gamma)>M\Bigr)
    =\mathbb E\Bigl[\mathbb P^*\Bigl(\sup_{\gamma}D_{\boldsymbol a}(\gamma)>M\Bigr)\Bigr]
    \le\mathbb E\bigl[\min\{1,B_n/M\}\bigr]
    \le\varepsilon+\mathbb P(B_n>\varepsilon M)
\]
for every $\varepsilon,M>0$; given $\eta>0$, choosing $\varepsilon=\eta/2$ and $M$ so large that $\limsup_n\mathbb P(B_n>\varepsilon M)\le\eta/2$, possible since $B_n=O_p(1)$, yields $\limsup_n\mathbb P\bigl(\sup_\gamma D_{\boldsymbol a}(\gamma)>M\bigr)\le\eta$. Thus $\sup_\gamma D_{\boldsymbol a}(\gamma)=O_{p^{*}}(1)$ conditionally on the data and $O_p(1)$ jointly over the data and the multipliers, and, by \eqref{eq:condlip_pairdiff},
\[
    L_M^{*2}:=\sup_{\gamma_1\ne\gamma_2}n^{-1}\boldsymbol a^\top\bar{\boldsymbol W}\boldsymbol a
    \;\le\;\sup_{\gamma\in[\gamma_L,\gamma_U]}D_{\boldsymbol a}(\gamma),
\]
so $\lVert\widetilde M_n^*(\gamma_1)-\widetilde M_n^*(\gamma_2)\rVert_{\mathcal H_{\gamma_0}}=T\bigl(n^{-1}\boldsymbol a^\top\bar{\boldsymbol W}\boldsymbol a\bigr)^{1/2}\le L_M^*\lvert\gamma_1-\gamma_2\rvert$ with $L_M^*=O_p(1)$, conditionally on the data in $\mathbb P^*$-probability and hence jointly, which proves (i).

\emph{Boundedness.} The leverage control \eqref{eq:leverage_generic}, at $A=\boldsymbol K_\gamma$ of envelope one, bounds the pointwise second moment. By the isometry of $J_\gamma$,
\[
    \bigl\lVert\widetilde M_n^*(\gamma)\bigr\rVert_{\mathcal H_{\gamma_0}}^{2}
    =\Bigl\lVert n^{-1/2}\sum_{j=1}^na_jk_\gamma(X_j,\cdot)\Bigr\rVert_{\mathcal H_\gamma}^{2}
    =n^{-1}\boldsymbol a^\top\boldsymbol K_\gamma\boldsymbol a,
\]
the second equality expanding the squared RKHS norm through the Gram matrix, so that, by the covariance step and the cyclicity of the trace,
\[
    \mathbb E^*\bigl\lVert\widetilde M_n^*(\gamma)\bigr\rVert_{\mathcal H_{\gamma_0}}^{2}
    =n^{-1}\sum_{l=1}^n\varepsilon_{0,l}^{2}\,(\hat{\boldsymbol\Pi}\boldsymbol K_\gamma\hat{\boldsymbol\Pi})_{ll}.
\]
The Gram matrix is symmetric, data-measurable, and of envelope one, $|(\boldsymbol K_\gamma)_{jk}|=|k_\gamma(X_j,X_k)|\le1$ with diagonal $(\boldsymbol K_\gamma)_{ll}=k_\gamma(X_l,X_l)=1$; \eqref{eq:leverage_generic} at $W^{(e)}=1$ yields
\[
    (\hat{\boldsymbol\Pi}\boldsymbol K_\gamma\hat{\boldsymbol\Pi})_{ll}
    \le1+2\hbar_n^{1/2}+\hbar_n
    =1+O_p(1)
\]
uniformly in $l$ and $\gamma$, whence
\[
    \sup_{\gamma\in[\gamma_L,\gamma_U]}\mathbb E^*\bigl\lVert\widetilde M_n^*(\gamma)\bigr\rVert_{\mathcal H_{\gamma_0}}^{2}
    \le\bigl(1+O_p(1)\bigr)\,q_n
    =O_p(1)
\]
by the law of large numbers, as in the conditional-modulus step. Markov's inequality applied conditionally gives $\lVert\widetilde M_n^*(\gamma_*)\rVert_{\mathcal H_{\gamma_0}}=O_{p^{*}}(1)$, hence jointly $O_p(1)$ by the upgrade of the conditional-chaining step, at any fixed base point $\gamma_*\in[\gamma_L,\gamma_U]$, and the base-point triangle with the Lipschitz bound of part (i),
\[
    \overline M_n^*
    \le\lVert\widetilde M_n^*(\gamma_*)\rVert_{\mathcal H_{\gamma_0}}+L_M^*\,(\gamma_U-\gamma_L)
    =O_p(1)
\]
at any fixed base point $\gamma_*$, proves (ii).
\end{proof}

\subsubsection{Proof of Theorem~\ref{th:stat_uniform} (Uniform Weak Convergence of the Statistic Process)}
\label{proof:stat_uniform}

The theorem has two parts: the null clause, under which $S_n^\perp \rightsquigarrow S_\infty$ in $\ell^\infty(\Upsilon)$ with $S_\infty(\eta) = \sum_{i \ge 1} \mu_i(\gamma)(\mu_i(\gamma)+\lambda)^{-2}(Z_i^\perp)^2$; and the local-alternative clause, under which, at the total-sample rate, $S_n^\perp \rightsquigarrow S_\infty(\cdot;h_c)$ with $h_c := \sqrt{1-c}\,h$ and the noncentral limit of Corollary~\ref{cor:proj_local} attenuated by the testing share (see the local-alternative paragraph below). The proof proceeds through the following paragraphs: pointwise convergence; the matrix representation and path regularity; gradient moments; boundedness; asymptotic equicontinuity; tightness, finite-dimensional convergence, and the limit; and the local alternative. Throughout, $S_n^\perp$ is computed on the testing sample of size $n$, and all statements are uniform in $\eta \in \Upsilon$ unless otherwise noted.

\medskip\noindent\textbf{Notation.}\; Let $\boldsymbol{\varepsilon}_\perp = \hat{\boldsymbol{\Pi}}\boldsymbol{\hat{\varepsilon}}$ denote the orthogonalized residual vector on the testing sample (Section~\ref{sec:estimation_effect}) and write $\boldsymbol{\varepsilon}_0$ for the true residual vector. Let $\boldsymbol{K}_\gamma$ be the $n\times n$ Gram matrix, $(\boldsymbol{K}_\gamma)_{jl} = k_\gamma(X_j,X_l)$, and let $\boldsymbol{R}_{n,\gamma}:\mathcal{H}_\gamma\to\mathbb{R}^n$ be the restriction operator, $(\boldsymbol{R}_{n,\gamma}f)_j = f(X_j)$, so that $\boldsymbol{R}_{n,\gamma}\boldsymbol{R}_{n,\gamma}^\top = \boldsymbol{K}_\gamma$ and $C_{n,\gamma} = n^{-1}\boldsymbol{R}_{n,\gamma}^\top\boldsymbol{R}_{n,\gamma}$, the last identity being the operator form of the definition $C_{n,\gamma} = n^{-1}\sum_{j}k_\gamma(X_j,\cdot)\otimes k_\gamma(X_j,\cdot)$. Let $C_\gamma = \mathbb{E}[k_\gamma(X,\cdot) \otimes k_\gamma(X,\cdot)]$ denote the population second-moment operator on $\mathcal{H}_\gamma$. The empirical covariance admits the spectral decomposition
\[
    C_{n,\gamma} = \sum_{i=1}^n \nu_{i,n}(\gamma)\,\hat{e}_{i,n}^{(\gamma)} \otimes \hat{e}_{i,n}^{(\gamma)},
    \qquad \langle \hat{e}_{i,n}^{(\gamma)}, \hat{e}_{j,n}^{(\gamma)} \rangle_{\mathcal{H}_\gamma} = \delta_{ij},
\]
where $\{\nu_{i,n}(\gamma)\}_{i=1}^n$ are the eigenvalues of $\boldsymbol{K}_\gamma/n$ (Proposition~\ref{svd_empirical}) and $\{\hat{e}_{i,n}^{(\gamma)}\}$ the Nystr\"om eigenfunctions (these objects are introduced in the notation paragraph of Appendix~\ref{proof:aux_uniform}). Define the scaled mean embedding
\[
    M_n(\gamma) := \frac{1}{\sqrt{n}}\sum_{j=1}^n \varepsilon_{\perp,j}\, k_\gamma(X_j,\cdot) = n^{-1/2}\,\boldsymbol{R}_{n,\gamma}^\top\boldsymbol{\varepsilon}_\perp \in \mathcal{H}_\gamma,
\]
and, for the oracle residuals, $M_n^{(0)}(\gamma) := n^{-1/2}\sum_j \varepsilon_{0,j} k_\gamma(X_j,\cdot)$. The statistic admits the RKHS representation
\begin{equation*}
    S_n^\perp(\eta)
    = \big\| (C_{n,\gamma}+\lambda I)^{-1} M_n(\gamma) \big\|_{\mathcal{H}_\gamma}^2
    = \sum_{i=1}^n \frac{1}{(\nu_{i,n}(\gamma)+\lambda)^2}\,\big\langle M_n(\gamma), \hat{e}_{i,n}^{(\gamma)} \big\rangle_{\mathcal{H}_\gamma}^2 .
\end{equation*}
We now prove the theorem.

\begin{proof}
\medskip\noindent\textbf{Pointwise convergence.}\; At each fixed $\eta \in \Upsilon$ the residual vector $\boldsymbol{\varepsilon}_\perp$ and the Gram matrix $\boldsymbol{K}_\gamma$ are those of a sample of size $n$ from the fixed kernel $k_\gamma$, so the pointwise limit is already known from the corollaries. Under $H_0$, $S_n^\perp(\eta) \xrightarrow{d} S_\infty(\eta) = \sum_{i\ge1}\mu_i(\gamma)(\mu_i(\gamma)+\lambda)^{-2}(Z_i^\perp)^2$ by Corollary~\ref{cor:proj_null}; under the local alternative at the total-sample rate the limit is the projected law with the attenuated signal $h_c := \sqrt{1-c}\,h$ (see the local-alternative paragraph below). The corollaries are stated for exactly the projected residuals, so they already account for the projection. What the theorem requires beyond these pointwise limits is \emph{uniformity in $\eta$}: $S_n^\perp \rightsquigarrow S_\infty$ in $\ell^\infty(\Upsilon)$, which needs the gradient-moment and equicontinuity argument below. For this we must know the orders of the two remainder sources that distinguish $\boldsymbol{\varepsilon}_\perp$ from the ideal residuals $\boldsymbol{\varepsilon}_0$:
\[
    \boldsymbol{\varepsilon}_\perp = \boldsymbol{\varepsilon}_0 + \boldsymbol{r}_{\mathrm{est}} + \boldsymbol{\delta},
    \qquad
    \boldsymbol{r}_{\mathrm{est}} := \hat{\boldsymbol{\Pi}}\boldsymbol{\hat{\varepsilon}} - \hat{\boldsymbol{\Pi}}\boldsymbol{\varepsilon}_0 \ \ (\text{estimation}), \qquad
    \boldsymbol{\delta} := \hat{\boldsymbol{\Pi}}\boldsymbol{\varepsilon}_0 - \boldsymbol{\varepsilon}_0 \ \ (\text{projection});
\]
for any vector $\boldsymbol{v}$, write $M_n^{(\boldsymbol{v})}(\gamma) := n^{-1/2}\sum_j v_j k_\gamma(X_j,\cdot)$ for the corresponding scaled mean embedding. These two remainders behave very differently at the embedding level. The \emph{estimation} remainder is negligible. By Lemma~\ref{lemma:estimation_remainder}, $\lVert R_w\rVert_{\mathcal{H}_\gamma}^2 = o_p(1)$ for $R_w := \sqrt{n}\,\boldsymbol{R}_{n,\gamma}^\top(\boldsymbol{K}_\gamma/n+\lambda\boldsymbol{I})^{-1}\boldsymbol{r}_{\mathrm{est}}$, the standing exponent being $\varrho>1/4$. Passing from $R_w$ to the embedding is immediate at the quadratic-form level: $\lVert M_n^{(\boldsymbol{r}_{\mathrm{est}})}(\gamma)\rVert_{\mathcal{H}_\gamma}^{2}=\boldsymbol{r}_{\mathrm{est}}^\top(\boldsymbol{K}_\gamma/n)\boldsymbol{r}_{\mathrm{est}}=o_p(1)$ uniformly in $\gamma\in[\gamma_L,\gamma_U]$, which is the quadratic-form clause of Fact~\ref{fact:residual_split} at sample size $n$. Estimating $\theta_0$ by $\hat\theta$ leaves no trace in the limit. The \emph{projection} remainder, in contrast, is \emph{not} negligible: although $\hat{\boldsymbol{\Pi}}\boldsymbol{\varepsilon}_0 - \boldsymbol{\varepsilon}_0 = O_p(n^{-1/2})$ componentwise, the scaled embedding of the difference is of order one. Indeed, with $\hat\beta := (\hat{\mathbb{G}}^\top\hat{\mathbb{G}})^{-1}\hat{\mathbb{G}}^\top\boldsymbol{\varepsilon}_0$ (so that $\boldsymbol{\delta} = -\hat{\mathbb{G}}\hat\beta$; $\hat{\mathbb G}^\top\boldsymbol{\varepsilon}_0/\sqrt{n}=O_p(1)$ by the CLT for the fixed-score matrix, extended to $\hat{\mathbb G}$ by the convention of Remark~\ref{rem:cond_det}, and $(\hat{\mathbb G}^\top\hat{\mathbb G}/n)^{-1}=O_p(1)$ under Assumption~\ref{ass:E}, hence $\lVert\hat\beta\rVert=O_p(n^{-1/2})$),
\[
    \lVert M_n^{(\boldsymbol{\delta})}(\gamma)\rVert_{\mathcal{H}_\gamma}^{2}
    = \boldsymbol{\delta}^\top(\boldsymbol{K}_\gamma/n)\boldsymbol{\delta}
    = \boldsymbol{\varepsilon}_0^\top H(\boldsymbol{K}_\gamma/n)H\,\boldsymbol{\varepsilon}_0
    \le \boldsymbol{\varepsilon}_0^\top H\,\boldsymbol{\varepsilon}_0
    = O_p(1),
\]
uniformly in $\gamma$: the identity is $\boldsymbol{\delta}=-H\boldsymbol{\varepsilon}_0$ with the hat projector $H := \hat{\mathbb G}(\hat{\mathbb G}^\top\hat{\mathbb G})^{-1}\hat{\mathbb G}^\top$, of rank $d_\theta$ on the invertibility events of Fact~\ref{fact:residual_split}, so that $\hat{\mathbb G}\hat\beta = H\boldsymbol{\varepsilon}_0$, and the inequality is Loewner, $H(\boldsymbol{K}_\gamma/n)H\preceq H$ (congruence preserving the order, $\boldsymbol{K}_\gamma/n\preceq\boldsymbol{I}$ by Step A, idempotence of $H$); the terminal quantity is the single flagged object of Lemma~\ref{lem:lipschitz_common}(a), there $O_p(1)$ under the convention of Remark~\ref{rem:cond_det}. Consequently $S_n^\perp(\eta)$ and the ``raw'' statistic $S_n(\eta)$ built on $\boldsymbol{\varepsilon}_0$ differ by $O_p(1)$ at the process level, not $o_p(1)$: the projection reshapes the limit law into the $Z^\perp$-law of Corollary~\ref{cor:proj_null} rather than vanishing asymptotically.

\medskip\noindent\textbf{The shift embedding.}\; The local alternative of Section~\ref{sec:cv_selection} shifts the residual function by $\boldsymbol{R}/\sqrt{m+n}$; the following record underlies the second envelope of the uniform-split paragraph below and is consumed in the local-alternative paragraph below. With $M_n^{(b)}(\gamma) := n^{-1/2}\sum_j(\hat{\boldsymbol{\Pi}}\boldsymbol{R})_jk_\gamma(X_j,\cdot)$: pathwise, $\|M_n^{(b)}(\gamma)\|^{2} = (\hat{\boldsymbol{\Pi}}\boldsymbol{R})^{\top}(\boldsymbol{K}_\gamma/n)(\hat{\boldsymbol{\Pi}}\boldsymbol{R}) \le \|\hat{\boldsymbol{\Pi}}\boldsymbol{R}\|^{2} \le \boldsymbol{R}^{\top}\boldsymbol{R} = O_p(n)$ by $\|\boldsymbol{K}_\gamma/n\|_{\mathrm{op}} \le \operatorname{tr}(\boldsymbol{K}_\gamma/n) = 1$ and the LLN under $\mathbb{E}[R^{2}] < \infty$; hence $\|M_n^{(b)}(\gamma)\|^2/(m+n) = O_p(n/(m+n)) = O_p(1)$, $\gamma$-free, the form in which the uniform-split paragraph consumes it.

\medskip\noindent\emph{Step A: the matrix representation and path regularity.}\;
This step derives the native quadratic form of the statistic and the deterministic regularity of its sample paths. Fix $\eta = (\gamma,\lambda)\in\Upsilon$ and abbreviate $R_\lambda := (\boldsymbol{K}_\gamma/n+\lambda\boldsymbol{I})^{-1}$. The derivation enters through the push-through identity
\[
    \bigl(\lambda\boldsymbol{I}+n^{-1}\boldsymbol{R}_{n,\gamma}^\top\boldsymbol{R}_{n,\gamma}\bigr)^{-1}\boldsymbol{R}_{n,\gamma}^\top
    = \boldsymbol{R}_{n,\gamma}^\top\bigl(\lambda\boldsymbol{I}+\boldsymbol{K}_\gamma/n\bigr)^{-1},
\]
which is verified by multiplying both sides on the left by $(\lambda\boldsymbol{I}+n^{-1}\boldsymbol{R}_{n,\gamma}^\top\boldsymbol{R}_{n,\gamma})$ and using $\boldsymbol{R}_{n,\gamma}\boldsymbol{R}_{n,\gamma}^\top = \boldsymbol{K}_\gamma$. Since $M_n(\gamma) = n^{-1/2}\boldsymbol{R}_{n,\gamma}^\top\boldsymbol{\varepsilon}_\perp$, the push-through identity applied to $(C_{n,\gamma}+\lambda\boldsymbol{I})^{-1}M_n(\gamma) = n^{-1/2}\bigl(\lambda\boldsymbol{I}+n^{-1}\boldsymbol{R}_{n,\gamma}^{\top}\boldsymbol{R}_{n,\gamma}\bigr)^{-1}\boldsymbol{R}_{n,\gamma}^{\top}\boldsymbol{\varepsilon}_\perp$ gives
\[
    \bigl(C_{n,\gamma}+\lambda\boldsymbol{I}\bigr)^{-1}M_n(\gamma)
    = n^{-1/2}\,\boldsymbol{R}_{n,\gamma}^{\top}\bigl(\lambda\boldsymbol{I}+\boldsymbol{K}_\gamma/n\bigr)^{-1}\boldsymbol{\varepsilon}_\perp,
\]
and squaring in $\mathcal{H}_\gamma$,
\[
\begin{aligned}
    S_n^\perp(\eta)
    &= n^{-1}\,\boldsymbol{\varepsilon}_\perp^\top\bigl(\lambda\boldsymbol{I}+\boldsymbol{K}_\gamma/n\bigr)^{-1}\boldsymbol{R}_{n,\gamma}\boldsymbol{R}_{n,\gamma}^\top\bigl(\lambda\boldsymbol{I}+\boldsymbol{K}_\gamma/n\bigr)^{-1}\boldsymbol{\varepsilon}_\perp\\
    &= n^{-1}\,\boldsymbol{\varepsilon}_\perp^\top\bigl(\lambda\boldsymbol{I}+\boldsymbol{K}_\gamma/n\bigr)^{-1}\boldsymbol{K}_\gamma\bigl(\lambda\boldsymbol{I}+\boldsymbol{K}_\gamma/n\bigr)^{-1}\boldsymbol{\varepsilon}_\perp\\
    &= \boldsymbol{\varepsilon}_\perp^\top\Bigl(\frac{\boldsymbol{K}_\gamma}{n}\Bigr)\Bigl(\frac{\boldsymbol{K}_\gamma}{n}+\lambda\boldsymbol{I}\Bigr)^{-2}\boldsymbol{\varepsilon}_\perp,
\end{aligned}
\]
the second line substituting $\boldsymbol{R}_{n,\gamma}\boldsymbol{R}_{n,\gamma}^\top = \boldsymbol{K}_\gamma$ and the third extracting the factor $\boldsymbol{K}_\gamma = n\,(\boldsymbol{K}_\gamma/n)$ against the prefactor $n^{-1}$. This resolves the normalization: with
\begin{equation}\label{eq:Wn-def}
    \mathsf{W}_n(\eta)
    := \Bigl(\frac{\boldsymbol{K}_\gamma}{n}+\lambda\boldsymbol{I}\Bigr)^{-1}\Bigl(\frac{\boldsymbol{K}_\gamma}{n}\Bigr)\Bigl(\frac{\boldsymbol{K}_\gamma}{n}+\lambda\boldsymbol{I}\Bigr)^{-1}
    = \Bigl(\frac{\boldsymbol{K}_\gamma}{n}\Bigr)\Bigl(\frac{\boldsymbol{K}_\gamma}{n}+\lambda\boldsymbol{I}\Bigr)^{-2},
\end{equation}
the two forms agreeing because the resolvents are functions of $\boldsymbol{K}_\gamma/n$ and commute with it, the representation
\[
    S_n^\perp(\eta) = \boldsymbol{\varepsilon}_\perp^\top\,\mathsf{W}_n(\eta)\,\boldsymbol{\varepsilon}_\perp
\]
holds with prefactor exactly $1$. The matrix $\mathsf{W}_n(\eta)$ is symmetric, being a function of the symmetric matrix $\boldsymbol K_\gamma/n$, and the spectral mapping theorem gives it the eigenvalue--eigenvector pairs $\bigl(f_\lambda(\nu_{i,n}(\gamma)),\boldsymbol u_{i,n}(\gamma)\bigr)$ for $f_\lambda(t) := t\,(t+\lambda)^{-2}$, so
\[
\begin{aligned}
    \|\mathsf{W}_n(\eta)\|_{\mathrm{op}}
    &= \max_{i\le n}\frac{\nu_{i,n}(\gamma)}{(\nu_{i,n}(\gamma)+\lambda)^2}
    \le \lambda_L^{-2},\\
    \operatorname{tr}\mathsf{W}_n(\eta)
    &= \sum_{i\le n}\frac{\nu_{i,n}(\gamma)}{(\nu_{i,n}(\gamma)+\lambda)^2}
    \le \lambda_L^{-2}\sum_{i\le n}\nu_{i,n}(\gamma)
    = \lambda_L^{-2},
\end{aligned}
\]
using $(\nu+\lambda)^2\ge\lambda^2$, $\nu_{i,n}\le1$ (the largest eigenvalue of a positive semidefinite matrix is at most its trace), and $\sum_{i\le n}\nu_{i,n}(\gamma) = \operatorname{tr}(\boldsymbol{K}_\gamma/n) = n^{-1}\sum_jk_\gamma(X_j,X_j) = 1$ for the Gaussian kernel.

For the path regularity, the entries of $\boldsymbol{K}_\gamma$ are real-analytic in $\gamma$ with
\[
    \bigl|\partial_\gamma k_\gamma(x,y)\bigr|
    = \|x-y\|^2e^{-\gamma\|x-y\|^2}
    \le (e\,\gamma_L)^{-1},
    \qquad x,y\in\mathcal{X},
\]
the maximum of $s\mapsto se^{-\gamma s}$ over $s\ge0$ being attained at $s = \gamma^{-1}$ with value $(e\gamma)^{-1}$, and $1/\gamma \le 1/\gamma_L$. For a symmetric matrix $A$, $\|A\|_{\mathrm{op}}\le\|A\|_{\infty} = \max_j\sum_l|A_{jl}|$, so the entrywise envelope gives, by row averaging,
\[
\begin{aligned}
    \Bigl\|\partial_\gamma\Bigl(\frac{\boldsymbol{K}_\gamma}{n}\Bigr)\Bigr\|_{\mathrm{op}}
    &\le \max_j\sum_l\frac{1}{n}\bigl|\partial_\gamma k_\gamma(X_j,X_l)\bigr|
    \le (e\,\gamma_L)^{-1},\\
    \Bigl\|\partial_\gamma\Bigl(\frac{\boldsymbol{K}_\gamma}{n}\Bigr)\Bigr\|_{\mathrm{F}}
    &\le n\cdot(e\,\gamma_L\,n)^{-1}
    = (e\,\gamma_L)^{-1},
\end{aligned}
\]
deterministic, design-free bounds. The product rule through the three factors of \eqref{eq:Wn-def} runs on two spectral inputs: $\|R_\lambda\|_{\mathrm{op}}\le\lambda_L^{-1}$, and the Frobenius--operator pairing
\[
    \Bigl\|\Bigl(\frac{\boldsymbol{K}_\gamma}{n}\Bigr)R_\lambda\Bigr\|_{\mathrm{F}}
    = \Bigl(\sum_{i\le n}\Bigl(\frac{\nu_{i,n}(\gamma)}{\nu_{i,n}(\gamma)+\lambda}\Bigr)^{2}\Bigr)^{1/2}
    \le \lambda_L^{-1}\Bigl(\sum_{i\le n}\nu_{i,n}(\gamma)^{2}\Bigr)^{1/2}
    \le \lambda_L^{-1}\sum_{i\le n}\nu_{i,n}(\gamma)
    = \lambda_L^{-1},
\]
the first step using $\nu_{i,n}(\gamma)+\lambda\ge\lambda\ge\lambda_L$, so that $\nu_{i,n}(\gamma)/(\nu_{i,n}(\gamma)+\lambda)\le\lambda_L^{-1}\nu_{i,n}(\gamma)$, the second $\sum_{i\le n}\nu_{i,n}^{2}\le\bigl(\sum_{i\le n}\nu_{i,n}\bigr)^{2} = 1$, and the last the trace identity $\sum_{i\le n}\nu_{i,n}(\gamma) = \operatorname{tr}(\boldsymbol{K}_\gamma/n) = 1$ established above; in particular $\bigl\|\bigl(\boldsymbol{K}_\gamma/n\bigr)R_\lambda\bigr\|_{\mathrm{op}}\le\lambda_L^{-1}$. Each term of the product rule is then bounded by pairing one factor in Frobenius norm against the other two in operator norm; the factor $R_\lambda$ enters only through its operator norm, never through $\|R_\lambda\|_{\mathrm{F}} = O(\sqrt n)$:
\[
\begin{aligned}
    \partial_\gamma\mathsf{W}_n(\eta)
    &= \bigl(\partial_\gamma R_\lambda\bigr)\Bigl(\frac{\boldsymbol{K}_\gamma}{n}\Bigr)R_\lambda
       + R_\lambda\bigl(\partial_\gamma\tfrac{\boldsymbol{K}_\gamma}{n}\bigr)R_\lambda
       + R_\lambda\Bigl(\frac{\boldsymbol{K}_\gamma}{n}\Bigr)\bigl(\partial_\gamma R_\lambda\bigr),\\[2pt]
    \bigl\|\partial_\gamma R_\lambda\bigr\|_{\mathrm{op}}
    &\le \lambda_L^{-2}(e\,\gamma_L)^{-1},
    \qquad
    \bigl\|\partial_\gamma R_\lambda\bigr\|_{\mathrm{F}}
    \le \lambda_L^{-2}(e\,\gamma_L)^{-1},\\
    \bigl\|(\partial_\gamma R_\lambda)\bigl(\tfrac{\boldsymbol{K}_\gamma}{n}\bigr)R_\lambda\bigr\|_{\mathrm{F}}
    &\le \bigl\|\partial_\gamma R_\lambda\bigr\|_{\mathrm{F}}\,\Bigl\|\Bigl(\tfrac{\boldsymbol{K}_\gamma}{n}\Bigr)R_\lambda\Bigr\|_{\mathrm{op}}
    \le \lambda_L^{-3}(e\,\gamma_L)^{-1},\\
    \bigl\|R_\lambda\bigl(\partial_\gamma\tfrac{\boldsymbol{K}_\gamma}{n}\bigr)R_\lambda\bigr\|_{\mathrm{F}}
    &\le \bigl\|\partial_\gamma\tfrac{\boldsymbol{K}_\gamma}{n}\bigr\|_{\mathrm{F}}\,\|R_\lambda\|_{\mathrm{op}}^{2}
    \le \lambda_L^{-2}(e\,\gamma_L)^{-1},\\
    \bigl\|\partial_\gamma\mathsf{W}_n(\eta)\bigr\|_{\mathrm{F}}
    &\le 2\lambda_L^{-3}(e\,\gamma_L)^{-1} + \lambda_L^{-2}(e\,\gamma_L)^{-1},\\
    \bigl\|\partial_\gamma\mathsf{W}_n(\eta)\bigr\|_{\mathrm{op}}
    &\le \bigl(2\lambda_L^{-3}+\lambda_L^{-2}\bigr)(e\,\gamma_L)^{-1},
\end{aligned}
\]
the third line covering the last product term by the transpose-invariance of the Frobenius norm (all three factors are symmetric), both resolvent bounds following from $\partial_\gamma R_\lambda = -R_\lambda(\partial_\gamma\tfrac{\boldsymbol{K}_\gamma}{n})R_\lambda$ by submultiplicativity, and the operator-norm total following by the same submultiplicativity through $\|\partial_\gamma R_\lambda\|_{\mathrm{op}}\le\lambda_L^{-2}(e\,\gamma_L)^{-1}$, $\|\boldsymbol{K}_\gamma/n\|_{\mathrm{op}}\le1$, and $\|R_\lambda\|_{\mathrm{op}}\le\lambda_L^{-1}$. Likewise, with $\partial_\lambda R_\lambda = -R_\lambda^2$, the same pairing gives $\|\partial_\lambda\mathsf{W}_n(\eta)\|_{\mathrm{F}}\le2\|R_\lambda\|_{\mathrm{op}}^{2}\bigl\|\bigl(\tfrac{\boldsymbol{K}_\gamma}{n}\bigr)R_\lambda\bigr\|_{\mathrm{op}}\le2\lambda_L^{-3}$:
\[
    \bigl\|\partial_\lambda\mathsf{W}_n(\eta)\bigr\|_{\mathrm{F}}
    \le 2\lambda_L^{-3},
    \qquad
    \bigl\|\partial_\lambda\mathsf{W}_n(\eta)\bigr\|_{\mathrm{op}}
    \le 2\lambda_L^{-3}.
\]
Writing $C_{\mathsf{W}} := \bigl(2\lambda_L^{-3}+\lambda_L^{-2}\bigr)(e\,\gamma_L)^{-1}+2\lambda_L^{-3}$, all four derivative norms are bounded by $C_{\mathsf{W}}$, uniformly in $(n,\eta)$ and the design. Finally, the paths are continuously differentiable in $\eta$ on $\Upsilon$. The entries of $\boldsymbol{K}_\gamma/n$ are continuously differentiable in $\gamma$, with $\partial_\gamma k_\gamma(x,y) = -\|x-y\|^{2}k_\gamma(x,y)$ continuous in $(\gamma,x,y)$ and $\bigl|\partial_\gamma^{2}k_\gamma(x,y)\bigr| = \|x-y\|^{4}e^{-\gamma\|x-y\|^{2}}\le4(e^{2}\gamma_L^{2})^{-1}$ (the maximum of $s\mapsto s^{2}e^{-\gamma s}$ over $s\ge0$ being attained at $s=2/\gamma$), so the mean-value theorem and row averaging give
\[
    \Bigl\|\partial_\gamma\Bigl(\frac{\boldsymbol{K}_\gamma}{n}\Bigr)-\partial_\gamma\Bigl(\frac{\boldsymbol{K}_{\gamma'}}{n}\Bigr)\Bigr\|_{\mathrm{op}}\le4(e^{2}\gamma_L^{2})^{-1}\,\bigl|\gamma-\gamma'\bigr|,
\]
deterministic and design-free; the resolvent increments obey, by the resolvent identity together with $\bigl\|\boldsymbol{K}_\gamma/n-\boldsymbol{K}_{\gamma'}/n\bigr\|_{\mathrm{op}}\le(e\gamma_L)^{-1}\bigl|\gamma-\gamma'\bigr|$ (the same device at the first derivative),
\[
    \bigl\|R_\lambda-R_{\lambda'}\bigr\|_{\mathrm{op}}\le\lambda_L^{-2}\Bigl[(e\,\gamma_L)^{-1}\bigl|\gamma-\gamma'\bigr|+\bigl|\lambda-\lambda'\bigr|\Bigr],\qquad\eta=(\gamma,\lambda),\ \eta'=(\gamma',\lambda'),
\]
so $\eta\mapsto R_\lambda$ is operator-norm continuous on $\Upsilon$. Each term of the product-rule displays for $\partial_\gamma\mathsf{W}_n$ and $\partial_\lambda\mathsf{W}_n$ is a product of factors drawn from $R_\lambda$, $\boldsymbol{K}_\gamma/n$, and $\partial_\gamma(\boldsymbol{K}_\gamma/n)$, all continuous in $\eta$ in operator norm by the two increment bounds just recorded and uniformly bounded by this step; hence $\eta\mapsto\partial_\eta\mathsf{W}_n(\eta)$ is operator-norm continuous and $\eta\mapsto\mathsf{W}_n(\eta)$ is $C^{1}$ on $\Upsilon$. Since the representation $S_n^\perp(\eta) = \boldsymbol{x}^{\top}\mathsf{W}_n(\eta)\boldsymbol{x}$ at $\boldsymbol{x} = \boldsymbol{\varepsilon}_\perp$ is $\boldsymbol{x}$-generic, $\eta\mapsto S_n^\perp(\eta)$ is $C^{1}$ for every realization of $\boldsymbol{\varepsilon}_\perp$, and the sample paths are continuous on the compact set $\Upsilon$.

\medskip\noindent\emph{Step B: gradient moments.}\;
This step establishes the fourth moments of the gradient process. Write
\[
    \partial_\eta S_n^\perp(\eta) = \boldsymbol{\varepsilon}_\perp^\top B_n(\eta)\,\boldsymbol{\varepsilon}_\perp,
    \qquad
    B_n(\eta) := \partial_\eta\mathsf{W}_n(\eta),
\]
a conditionally deterministic symmetric matrix given $\mathbf{X}$, with $\|B_n(\eta)\|_{\mathrm{F}}\le C_{\mathsf{W}}$ and $\|B_n(\eta)\|_{\mathrm{op}}\le C_{\mathsf{W}}$ for every $(n,\eta)$ by Step A. Throughout this step, matrices carrying the estimated objects $\hat{\boldsymbol{\Pi}}$, $\hat{\mathbb{G}}$ are treated as conditionally deterministic given $\mathbf{X}$, the convention of Remark~\ref{rem:cond_det}; conditional independence refers to the innovations $\varepsilon_{0,1},\dots,\varepsilon_{0,n}$, which given $\mathbf{X}$ are independent and centered with $\mathbb{E}[\varepsilon_{0,j}^2\mid X_j] = \sigma^2(X_j)\in[\sigma_{\min}^2,\sigma_{\max}^2]$ and $\bar{M}_{2p} := \sup_j\mathbb{E}[\varepsilon_{0,j}^{2p}\mid X_j]<\infty$ for the exponent $p\in(2,4]$ of Assumption~\ref{ass:G}(v) (Assumption~\ref{ass:G}(iii), (v)). One input is used, the order-$p$ moment of a centered quadratic form: Lemma~\ref{lem:quad_moment}, applied conditionally on $\mathbf{X}$ to these innovations and to any symmetric $\mathbf{X}$-measurable $A$, gives \eqref{eq:quad4} at $\mathcal G = \sigma(\mathbf X)$, that is,
\[
\begin{gathered}
    \mathbb{E}\Bigl[\bigl|\boldsymbol{\varepsilon}_0^\top A\,\boldsymbol{\varepsilon}_0-\operatorname{tr}(A\,D)\bigr|^{p}\,\Big|\,\mathbf{X}\Bigr]
    \le C_{p}\bigl(\sigma_{\max},\bar{M}_{2p}\bigr)\,\|A\|_{\mathrm F}^{p},\\
    D := \operatorname{diag}\bigl(\sigma^2(X_1),\dots,\sigma^2(X_n)\bigr),
    \qquad
    \boldsymbol{\Sigma}_n := \hat{\boldsymbol{\Pi}}D\hat{\boldsymbol{\Pi}}.
\end{gathered}
\]

\medskip\noindent\textbf{The uniform split of the statistic.}\; In both regimes the statistic splits exactly into a null-type field, two shift terms, and a uniformly negligible remainder; this paragraph sets up the split and records the envelopes consumed by the later steps. By Fact~\ref{fact:residual_split} and the definition of the shift,
\[
    \boldsymbol{\varepsilon}_\perp
    = \hat{\boldsymbol{\Pi}}\boldsymbol{\varepsilon}_0 + \boldsymbol{f}_{\mathrm{sh}} + \boldsymbol{r}_{\mathrm{est}},
    \qquad
    \boldsymbol{f}_{\mathrm{sh}} := \frac{\hat{\boldsymbol{\Pi}}\boldsymbol{R}}{\sqrt{m+n}},
\]
with $\boldsymbol{f}_{\mathrm{sh}} := \boldsymbol{0}$ under $H_0$; $\boldsymbol{f}_{\mathrm{sh}}$ is $\mathbf{X}$-measurable in the $X$-measurable regime and conditionally deterministic under the convention of Remark~\ref{rem:cond_det} in the sieve regime. Under $H_{1,n}$ the symbol $\boldsymbol{\varepsilon}_0$ in this split denotes the centered noise, $\boldsymbol{\varepsilon}_0=\boldsymbol{\varepsilon}(\theta_0)-\boldsymbol{R}/\sqrt{m+n}$ pointwise; the two readings of the symbol coincide under $H_0$, and all stochastic statements below (Steps B--E) use the centered reading, under which the entries of $\boldsymbol{\varepsilon}_0$ are conditionally independent and centered given $\mathbf X$. Here $\boldsymbol{r}_{\mathrm{est}} := \boldsymbol{\varepsilon}_\perp - \hat{\boldsymbol{\Pi}}\boldsymbol{\varepsilon}_0 - \boldsymbol{f}_{\mathrm{sh}}$ is the pure estimation remainder: under $H_0$ it is the $\hat{\boldsymbol{\Pi}}\hat{\boldsymbol{\varepsilon}} - \hat{\boldsymbol{\Pi}}\boldsymbol{\varepsilon}_0$ of the pointwise-convergence paragraph, and under $H_{1,n}$ it is $\hat{\boldsymbol{\Pi}}[\hat{\boldsymbol{\varepsilon}} - \boldsymbol{\varepsilon}(\theta_0)]$ with $\boldsymbol{\varepsilon}(\theta_0)$ the true residual at $\theta_0$: the shift enters $\hat{\boldsymbol{\varepsilon}}$ and $\boldsymbol{\varepsilon}(\theta_0)$ identically through the observed responses and cancels, so the quantity is the same estimation effect as under $H_0$ and Lemma~\ref{lemma:estimation_remainder}'s bound applies verbatim. For the shift component, $\hat{\boldsymbol{\Pi}}\boldsymbol{R} = \boldsymbol{\Pi}R(X_j) + o_p(1)$ componentwise: writing $d_j := (\hat{\boldsymbol{\Pi}}\boldsymbol{R})_j - (\boldsymbol{\Pi}R)(X_j)$, Assumption~\ref{ass:F}(iii) gives $n^{-1}\sum_jd_j^{2} = o_p(1)$, and the embedding of the difference is negligible uniformly in $\gamma$: indeed $\boldsymbol{d}^{\top}\boldsymbol{K}_\gamma\boldsymbol{d}\le\bigl(\sum_j|d_j|\bigr)^{2}\le n\sum_jd_j^{2}$ by $|k_\gamma|\le1$ and Cauchy--Schwarz, so
\[
    \Bigl\|\sqrt{\tfrac{n}{m+n}}\;\frac{1}{n}\sum_{j} d_j\,k_\gamma(X_j,\cdot)\Bigr\|_{\mathcal{H}_\gamma}^{2}
    = \frac{1}{n(m+n)}\,\boldsymbol{d}^{\top}\boldsymbol{K}_\gamma\,\boldsymbol{d}
    \le \frac{1}{m+n}\sum_{j}d_j^{2}
    = o_p(1),
\]
uniformly in $\gamma$. Define the projected-noise and shift embeddings
\[
    M_n^{(\perp 0)}(\gamma) := \frac{1}{\sqrt{n}}\sum_{j=1}^n(\hat{\boldsymbol{\Pi}}\boldsymbol{\varepsilon}_0)_j\,k_\gamma(X_j,\cdot),
    \qquad
    M_n^{(h)}(\gamma) := \sqrt{\frac{n}{m+n}}\;\frac{1}{n}\sum_{j}(\hat{\boldsymbol{\Pi}}\boldsymbol{R})(X_j)\,k_\gamma(X_j,\cdot),
\]
so that $M_n^{(h)}$ is the embedding of $\boldsymbol{f}_{\mathrm{sh}}$ and $M_n^{(h)} \equiv \boldsymbol{0}$ under $H_0$, and the exact pathwise remainder
\[
    r_n(\gamma) := M_n(\gamma) - M_n^{(\perp 0)}(\gamma) - M_n^{(h)}(\gamma)
    = M_n^{(\boldsymbol{r}_{\mathrm{est}})}(\gamma),
    \qquad
    M_n^{(d)}(\gamma) := \sqrt{\frac{n}{m+n}}\;\frac{1}{n}\sum_{j}d_j\,k_\gamma(X_j,\cdot),
\]
where $M_n^{(d)}$ is the embedding of the hat--true shift discrepancy $\boldsymbol{d}/\sqrt{m+n}$, negligible by the $d$-display above and consumed in the local-alternative paragraph below; the equality $r_n = M_n^{(\boldsymbol{r}_{\mathrm{est}})}$ is exact by linearity, $M_n^{(h)}$ being the embedding of $\boldsymbol{f}_{\mathrm{sh}}$ and $(\hat{\boldsymbol{\Pi}}\boldsymbol{R})(X_j) = (\hat{\boldsymbol{\Pi}}\boldsymbol{R})_j$. Four envelopes are read off. First, the definitional split $M_n^{(\perp 0)} = M_n - M_n^{(h)} - r_n$ and the triangle inequality with $\sup(f+g)\le\sup f+\sup g$ give
\[
\begin{aligned}
    \sup_{\gamma\in[\gamma_L,\gamma_U]}\bigl\|M_n^{(\perp 0)}(\gamma)\bigr\|
    &\le \sup_{\gamma\in[\gamma_L,\gamma_U]}\bigl\|M_n(\gamma)\bigr\| + \sup_{\gamma\in[\gamma_L,\gamma_U]}\bigl\|M_n^{(h)}(\gamma)\bigr\| + \sup_{\gamma\in[\gamma_L,\gamma_U]}\bigl\|r_n(\gamma)\bigr\|\\
    &= O_p(1) + O_p(1) + o_p(1)
    = O_p(1),
\end{aligned}
\]
the first term by the uniform boundedness conclusion of Lemma~\ref{lem:lipschitz_common}(a), which is regime-free: its proof bounds the base point by Markov's inequality in both regimes, the local-alternative shift being carried by the $\mathbf X$-measurable envelope of the regime-coverage paragraph, and its projected-score step is licensed by the convention of Remark~\ref{rem:cond_det}; the second by envelope two below (absent under $H_0$); the third by envelope three below. Second,
\[
    \sup_{\gamma\in[\gamma_L,\gamma_U]}\bigl\|M_n^{(h)}(\gamma)\bigr\|_{\mathcal{H}_\gamma}^{2}
    = \sup_{\gamma\in[\gamma_L,\gamma_U]}\boldsymbol{f}_{\mathrm{sh}}^{\top}\Bigl(\frac{\boldsymbol{K}_\gamma}{n}\Bigr)\boldsymbol{f}_{\mathrm{sh}}
    \le \boldsymbol{f}_{\mathrm{sh}}^{\top}\boldsymbol{f}_{\mathrm{sh}}
    = \frac{\boldsymbol{R}^{\top}\hat{\boldsymbol{\Pi}}\boldsymbol{R}}{m+n}
    = O_p(1),
\]
by $\hat{\boldsymbol{\Pi}} \preceq \boldsymbol{I}$ and the pathwise bound $\boldsymbol{R}^{\top}\hat{\boldsymbol{\Pi}}\boldsymbol{R} \le \boldsymbol{R}^{\top}\boldsymbol{R} = O_p(n)$, the LLN at the testing sample under the standing convention $\mathbb{E}[R^{2}] < \infty$ of Section~\ref{sec:cv_selection} (the shift-embedding record), $\gamma$-free. Third,
\[
    \sup_{\gamma\in[\gamma_L,\gamma_U]}\bigl\|r_n(\gamma)\bigr\|
    = \sup_{\gamma\in[\gamma_L,\gamma_U]}\bigl\|M_n^{(\boldsymbol{r}_{\mathrm{est}})}(\gamma)\bigr\|_{\mathcal{H}_\gamma}
    = o_p(1),
\]
the first envelope being the $\eta$-uniform estimation-remainder bound of the pointwise-convergence paragraph and the second the $d$-display above. Fourth, writing $\mathsf{A}_n := (\boldsymbol{C}_{n,\gamma}+\lambda\boldsymbol{I})^{-1}$ for the resolvent, $\bigl\|\mathsf{A}_n M_n^{(\boldsymbol{x})}(\gamma)\bigr\| \le \lambda_L^{-1}\bigl\|M_n^{(\boldsymbol{x})}(\gamma)\bigr\|$ for every $\boldsymbol{x}\in\mathbb{R}^n$ and every $\gamma$: by the push-through identity of Step A applied to $\mathsf{A}_n M_n^{(\boldsymbol{x})}(\gamma) = n^{-1/2}\boldsymbol{R}_{n,\gamma}^{\top}R_\lambda\boldsymbol{x}$ with $R_\lambda := (\boldsymbol{K}_\gamma/n+\lambda\boldsymbol{I})^{-1}$,
\[
\begin{aligned}
    \bigl\|\mathsf{A}_n M_n^{(\boldsymbol{x})}(\gamma)\bigr\|_{\mathcal{H}_\gamma}^{2}
    &= \boldsymbol{x}^{\top}R_\lambda\Bigl(\frac{\boldsymbol{K}_\gamma}{n}\Bigr)R_\lambda\,\boldsymbol{x}
    = \boldsymbol{x}^{\top}\mathsf{W}_n(\eta)\,\boldsymbol{x}
    \le \lambda_L^{-2}\,\boldsymbol{x}^{\top}\Bigl(\frac{\boldsymbol{K}_\gamma}{n}\Bigr)\boldsymbol{x}\\
    &= \lambda_L^{-2}\bigl\|M_n^{(\boldsymbol{x})}(\gamma)\bigr\|_{\mathcal{H}_\gamma}^{2},
\end{aligned}
\]
the first identity squaring the norm formula $\|n^{-1/2}\boldsymbol{R}_{n,\gamma}^{\top}\boldsymbol{y}\|^{2} = \boldsymbol{y}^{\top}(\boldsymbol{K}_\gamma/n)\boldsymbol{y}$ at $\boldsymbol{y} = R_\lambda\boldsymbol{x}$, the middle identity being the definition \eqref{eq:Wn-def} of $\mathsf{W}_n$, the inequality the spectral comparison $t(t+\lambda)^{-2}\le\lambda_L^{-2}t$, valid since $(t+\lambda)^{2}\ge\lambda^{2}\ge\lambda_L^{2}$, and the last identity the embedding-norm identity of the pointwise-convergence paragraph. The exact split of the statistic reads
\[
\begin{split}
    S_n^\perp(\eta)
    &= \Bigl\|\mathsf{A}_n\bigl(M_n^{(\perp 0)}(\gamma) + M_n^{(h)}(\gamma) + r_n(\gamma)\bigr)\Bigr\|^{2}\\
    &= \underbrace{\bigl\|\mathsf{A}_n M_n^{(\perp 0)}(\gamma)\bigr\|^{2}}_{=:\,A_n^{\mathrm{sh}}(\eta)}
      + \underbrace{2\bigl\langle\mathsf{A}_n M_n^{(\perp 0)}(\gamma),\,\mathsf{A}_n M_n^{(h)}(\gamma)\bigr\rangle_{\mathcal{H}_\gamma}}_{=:\,2B_n^{\mathrm{sh}}(\eta)}\\
    &\quad + \underbrace{\bigl\|\mathsf{A}_n M_n^{(h)}(\gamma)\bigr\|^{2}}_{=:\,C_n^{\mathrm{sh}}(\eta)}
      + \underbrace{2\bigl\langle\mathsf{A}_n\bigl(M_n^{(\perp 0)}(\gamma)+M_n^{(h)}(\gamma)\bigr),\,\mathsf{A}_n r_n(\gamma)\bigr\rangle_{\mathcal{H}_\gamma} + \bigl\|\mathsf{A}_n r_n(\gamma)\bigr\|^{2}}_{=:\,\rho_n(\eta)},
\end{split}
\]
exact in both regimes: the first line substitutes $M_n = M_n^{(\perp 0)} + M_n^{(h)} + r_n$ and the remaining lines expand the squared norm of the three-term sum by polarization. The remainder is uniformly negligible: with $\overline{M}_{\perp 0} := \sup_\gamma\|M_n^{(\perp 0)}(\gamma)\|$, $\overline{M}_h := \sup_\gamma\|M_n^{(h)}(\gamma)\|$, and $\overline{r}_n := \sup_\gamma\|r_n(\gamma)\|$,
\[
\begin{aligned}
    \sup_{\eta\in\Upsilon}\bigl|\rho_n(\eta)\bigr|
    &\le 2\sup_{\eta\in\Upsilon}\bigl|\bigl\langle\mathsf{A}_n\bigl(M_n^{(\perp 0)}(\gamma)+M_n^{(h)}(\gamma)\bigr),\,\mathsf{A}_n r_n(\gamma)\bigr\rangle\bigr|
    + \sup_{\eta\in\Upsilon}\bigl\|\mathsf{A}_n r_n(\gamma)\bigr\|^{2}\\
    &\le 2\sup_{\eta\in\Upsilon}\bigl\|\mathsf{A}_n\bigl(M_n^{(\perp 0)}(\gamma)+M_n^{(h)}(\gamma)\bigr)\bigr\|\,\sup_{\eta\in\Upsilon}\bigl\|\mathsf{A}_n r_n(\gamma)\bigr\|
    + \sup_{\eta\in\Upsilon}\bigl\|\mathsf{A}_n r_n(\gamma)\bigr\|^{2}\\
    &\le 2\lambda_L^{-2}\,\sup_\gamma\bigl\|M_n^{(\perp 0)}(\gamma)+M_n^{(h)}(\gamma)\bigr\|\,\sup_\gamma\bigl\|r_n(\gamma)\bigr\|
    + \lambda_L^{-2}\sup_\gamma\bigl\|r_n(\gamma)\bigr\|^{2}\\
    &\le 2\lambda_L^{-2}\bigl(\overline{M}_{\perp 0} + \overline{M}_h\bigr)\overline{r}_n
    + \lambda_L^{-2}\,\overline{r}_n^{\,2},
\end{aligned}
\]
by the triangle inequality in the first row, Cauchy--Schwarz in the second, envelope four in the third, and the triangle inequality in the fourth; reading envelopes one to three in the last row gives $\sup_{\eta\in\Upsilon}|\rho_n(\eta)| = o_p(1)\,O_p(1) = o_p(1)$. Term $A_n^{\mathrm{sh}}$ is the noise statistic: Step B takes its gradient moments and Step D its equicontinuity, in both regimes; the shift terms and the $H_{1,n}$ process-level conclusion are closed in the local-alternative paragraph below; $\rho_n$ is the uniform remainder just bounded.

\medskip\noindent\emph{Gradient moments of the null-type field (Step B, continued).}\; Now condition on $\mathbf{X}$ and split by the residual decomposition of Fact~\ref{fact:residual_split}, $\boldsymbol{\varepsilon}_\perp = \hat{\boldsymbol{\Pi}}\boldsymbol{\varepsilon}_0 + \boldsymbol{r}_{\mathrm{est}}$ with $\hat{\boldsymbol{\Pi}}\boldsymbol{\varepsilon}_0 = \boldsymbol{\varepsilon}_0 + \boldsymbol{\delta}$, $\boldsymbol{\delta} := -\hat{\mathbb{G}}\hat\beta = -H\boldsymbol{\varepsilon}_0$, $\hat\beta$ as in the pointwise-convergence paragraph, and $H := \hat{\mathbb{G}}(\hat{\mathbb{G}}^{\top}\hat{\mathbb{G}})^{-1}\hat{\mathbb{G}}^{\top}$ the hat projector of rank $d_\theta$, with Frobenius norm $\sqrt{d_\theta}$ and operator norm $\|H\|_{\mathrm{op}}\le1$:
\[
    \partial_\eta A_n^{\mathrm{sh}}(\eta)
    = \underbrace{\boldsymbol{\varepsilon}_0^\top B_n(\eta)\,\boldsymbol{\varepsilon}_0}_{\text{oracle term}}
    + \underbrace{2\,\boldsymbol{\varepsilon}_0^\top B_n(\eta)\,\boldsymbol{\delta}}_{\text{cross term}}
    + \underbrace{\boldsymbol{\delta}^\top B_n(\eta)\,\boldsymbol{\delta}}_{\text{projection term}},
\]
the expansion of $\nabla_\eta A_n^{\mathrm{sh}}(\eta) = (\boldsymbol{\varepsilon}_0+\boldsymbol{\delta})^\top B_n(\eta)(\boldsymbol{\varepsilon}_0+\boldsymbol{\delta})$, which holds by Step A's quadratic-form representation applied to the embedding of $\hat{\boldsymbol{\Pi}}\boldsymbol{\varepsilon}_0 = \boldsymbol{\varepsilon}_0 + \boldsymbol{\delta}$, the representation being $\boldsymbol{x}$-generic and the vector $\boldsymbol{\varepsilon}_0+\boldsymbol{\delta}$ being $\eta$-free; and $\nabla_\eta S_n^c(\eta) = \nabla_\eta A_n^{\mathrm{sh}}(\eta) - \operatorname{tr}(B_n(\eta)\boldsymbol{\Sigma}_n)$ with $S_n^c := A_n^{\mathrm{sh}} - a_n$.
The \emph{oracle term} is centered by construction, $\mathbb{E}[\boldsymbol{\varepsilon}_0^\top B_n\boldsymbol{\varepsilon}_0\mid\mathbf{X}] = \operatorname{tr}(B_nD)$, and \eqref{eq:quad4} with $A = B_n$ gives
\[
    \sup_{\eta\in\Upsilon}\mathbb{E}\Bigl[\bigl|\boldsymbol{\varepsilon}_0^\top B_n(\eta)\boldsymbol{\varepsilon}_0-\operatorname{tr}(B_n(\eta)D)\bigr|^{p}\,\Big|\,\mathbf{X}\Bigr]
    \le C_{p}\,C_{\mathsf{W}}^{p}.
\]
The \emph{cross term}. Its $\boldsymbol{\delta}$-part is $\boldsymbol{\varepsilon}_0^\top B_n\boldsymbol{\delta} = -\,\boldsymbol{\varepsilon}_0^\top B_nH\boldsymbol{\varepsilon}_0 = -\,\boldsymbol{\varepsilon}_0^\top A_H\boldsymbol{\varepsilon}_0$ with $A_H := \operatorname{sym}(B_nH)$; only the symmetric part enters a quadratic form, and $A_H = \operatorname{sym}\bigl(B_n\hat{\mathbb{G}}(\hat{\mathbb{G}}^\top\hat{\mathbb{G}})^{-1}\hat{\mathbb{G}}^\top\bigr)$ under the identity $H = \hat{\mathbb{G}}(\hat{\mathbb{G}}^\top\hat{\mathbb{G}})^{-1}\hat{\mathbb{G}}^\top$. The Frobenius norm satisfies $\|A_H\|_{\mathrm{F}} \le \|B_nH\|_{\mathrm{F}}\le\|B_n\|_{\mathrm{op}}\|H\|_{\mathrm{F}} = \|B_n\|_{\mathrm{op}}\sqrt{d_\theta}$; \eqref{eq:quad4} with $A = A_H$ therefore gives $\mathbb{E}[|\boldsymbol{\varepsilon}_0^\top A_H\boldsymbol{\varepsilon}_0-\operatorname{tr}(A_HD)|^{p}\mid\mathbf{X}]\le C_{p}\,d_\theta^{p/2}\,C_{\mathsf{W}}^{p}$. The \emph{projection term} is a quadratic form in $\boldsymbol{\varepsilon}_0$ with conditionally deterministic symmetric matrix: $\boldsymbol{\delta}^\top B_n(\eta)\boldsymbol{\delta} = \boldsymbol{\varepsilon}_0^\top HB_n(\eta)H\,\boldsymbol{\varepsilon}_0$ with $\|HB_n(\eta)H\|_{\mathrm{F}} \le \|B_n(\eta)\|_{\mathrm{op}}\,\|H\|_{\mathrm{F}} = C_{\mathsf{W}}\sqrt{d_\theta}$ by $\|H\|_{\mathrm{op}} \le 1$. Centering at the conditional mean $\operatorname{tr}(B_n(\eta) HDH) = \operatorname{tr}(HB_n(\eta)HD)$, whose modulus is bounded tracially, $|\operatorname{tr}(B_n HDH)| \le \|B_n(\eta)\|_{\mathrm{op}} \operatorname{tr}(HDH) \le C_{\mathsf{W}} \sigma_{\max}^{2} d_\theta$ since $HDH \preceq \sigma_{\max}^{2} H$ and $\operatorname{tr}(H) = d_\theta$, \eqref{eq:quad4} with $A = HB_n(\eta) H$ gives
\[
    \sup_{\eta\in\Upsilon}\mathbb{E}\Bigl[\bigl|\boldsymbol{\delta}^\top B_n(\eta)\boldsymbol{\delta} - \operatorname{tr}(B_n(\eta) HDH)\bigr|^{p} \,\Big|\, \mathbf{X}\Bigr]
    \le C_{p}\, C_{\mathsf{W}}^{p} d_\theta^{p/2}.
\]

Collecting the pieces defines the centered field and the drift,
\[
    S_n^c(\eta) := A_n^{\mathrm{sh}}(\eta) - a_n(\eta),
    \qquad
    a_n(\eta) := \operatorname{tr}\bigl(\mathsf{W}_n(\eta)\,\boldsymbol{\Sigma}_n\bigr),
\]
with $\boldsymbol{\Sigma}_n := \hat{\boldsymbol{\Pi}} D \hat{\boldsymbol{\Pi}}$ of Step B. The role of the field $a_n$ is confined to hypothesis (ii) of Lemma~\ref{lem:equi}, recorded after the assembly.
Writing $\nabla_\eta S_n^c(\eta) = (\boldsymbol{\varepsilon}_0+\boldsymbol{\delta})^{\top} B_n(\eta) (\boldsymbol{\varepsilon}_0+\boldsymbol{\delta}) - \operatorname{tr}(B_n(\eta)\boldsymbol{\Sigma}_n)$, the identity of the assembly above, and decomposing $\operatorname{tr}(B_n(\eta)\boldsymbol{\Sigma}_n) = \operatorname{tr}(B_n(\eta) D) - \operatorname{tr}(B_n(\eta)(D - \boldsymbol{\Sigma}_n))$, the gradient splits into three groups,
\[
\begin{aligned}
    \nabla_\eta S_n^c(\eta)
    &= \bigl(\boldsymbol{\varepsilon}_0^\top B_n(\eta)\boldsymbol{\varepsilon}_0-\operatorname{tr}(B_n(\eta)D)\bigr)
    +2\,\boldsymbol{\varepsilon}_0^\top B_n(\eta)\boldsymbol{\delta}\\
    &\quad +\Bigl(\boldsymbol{\delta}^\top B_n(\eta)\boldsymbol{\delta}
    +\operatorname{tr}\bigl(B_n(\eta)(D-\boldsymbol{\Sigma}_n)\bigr)\Bigr),
\end{aligned}
\]
Group 1 is \eqref{eq:quad4} with $A = B_n(\eta)$, the oracle-term display: at most $C_{p}\,C_{\mathsf{W}}^{p}$, conditionally. Group 2 is the $A_H$-treatment: $2\,\boldsymbol{\varepsilon}_0^\top B_n\boldsymbol{\delta} = -2\,\boldsymbol{\varepsilon}_0^\top A_H\boldsymbol{\varepsilon}_0$, centered at $-2\operatorname{tr}(A_HD)$, \eqref{eq:quad4} with $A = A_H$ giving $C_{p}\,d_\theta^{p/2}\,C_{\mathsf{W}}^{p}$; the centering constant is absorbed deterministically, $\bigl|\operatorname{tr}(A_HD)\bigr| \le \tfrac{1}{2}\bigl|\operatorname{tr}(B_nHD)\bigr| + \tfrac{1}{2}\bigl|\operatorname{tr}(B_nDH)\bigr|$ with $\operatorname{tr}(B_nDH) = \operatorname{tr}(HB_nD)$, and
\[
    \bigl|\operatorname{tr}(B_nHD)\bigr|
    \le \bigl\|B_n\bigr\|_{\mathrm{op}}\,\bigl\|HD\bigr\|_{(1)}
    \le \sqrt{d_\theta}\,\bigl\|B_n\bigr\|_{\mathrm{op}}\,\bigl\|HD\bigr\|_{\mathrm{F}}
    \le \sigma_{\max}^{2}\,d_\theta\,\bigl\|B_n\bigr\|_{\mathrm{op}},
\]
by operator--nuclear duality, the rank bound $\|M\|_{(1)}\le\sqrt{\operatorname{rank}(M)}\,\|M\|_{\mathrm{F}}$ with $\operatorname{rank}(HD)\le\operatorname{rank}(H) = d_\theta$, and $\|HD\|_{\mathrm{F}}\le\|H\|_{\mathrm{F}}\|D\|_{\mathrm{op}} = \sigma_{\max}^{2}\sqrt{d_\theta}$, the modulus of $\operatorname{tr}(B_nDH)$ admitting the same bounds since $\|DH\|_{\mathrm{F}} = \|HD\|_{\mathrm{F}}$. Group 3 combines the projection-term bound $\sup_\eta\mathbb{E}\bigl[\bigl|\boldsymbol{\delta}^\top B_n(\eta)\boldsymbol{\delta}-\operatorname{tr}(B_n(\eta)HDH)\bigr|^{p}\,\big|\,\mathbf{X}\bigr]\le C_{p}\,C_{\mathsf{W}}^{p}d_\theta^{p/2}$ with the bias envelope $\bigl|\operatorname{tr}(A(D-\boldsymbol{\Sigma}_n))\bigr|\le3\sigma_{\max}^{2}\sqrt{d_\theta}\,\|A\|_{\mathrm{F}}$: since $D-\boldsymbol{\Sigma}_n = HD+DH-HDH$ with the hat projector $H$ of the split, Cauchy--Schwarz in the Hilbert--Schmidt inner product gives $\bigl|\operatorname{tr}(A(D-\boldsymbol{\Sigma}_n))\bigr|\le\|A\|_{\mathrm{F}}\bigl(\|HD\|_{\mathrm{F}}+\|DH\|_{\mathrm{F}}+\|HDH\|_{\mathrm{F}}\bigr)\le3\sigma_{\max}^{2}\sqrt{d_\theta}\,\|A\|_{\mathrm{F}}$, each Frobenius norm being at most $\sigma_{\max}^{2}\sqrt{d_\theta}$: the first two by $\|HD\|_{\mathrm{F}}\le\|H\|_{\mathrm{F}}\|D\|_{\mathrm{op}}$ and symmetry, the third by the intermediate bound $\|HDH\|_{\mathrm{F}}\le\|HD\|_{\mathrm{F}}\|H\|_{\mathrm{op}}$, through $\|H\|_{\mathrm{F}} = \sqrt{d_\theta}$, $\|H\|_{\mathrm{op}}\le1$, and $\|D\|_{\mathrm{op}}\le\sigma_{\max}^{2}$; the pathwise piece is therefore at most $3\sigma_{\max}^{2}\sqrt{d_\theta}\,C_{\mathsf{W}}$.Consequently,
\[
    \sup_{\eta\in\Upsilon}\mathbb{E}\Bigl[\bigl\|\nabla_\eta S_n^c(\eta)\bigr\|^{p}\,\Big|\,\mathbf{X}\Bigr]
    \le N_p,
    \qquad
    \sup_{\eta\in\Upsilon}\mathbb{E}\bigl\|\nabla_\eta S_n^c(\eta)\bigr\|^{p}\le N_p,
\]
with $N_p$ deterministic, depending only on $(p,\gamma_L,\lambda_L,\sigma_{\max},\bar{M}_{2p},d_\theta)$, the conditional bound combining the three groups through $(a+b+c)^{p}\le3^{p-1}(a^{p}+b^{p}+c^{p})$ and, within group 3, $(x+y)^{p}\le2^{p-1}(|x|^{p}+y^{p})$ for the deterministic bias piece, a fixed $p$-constant; the passage from the two coordinate bounds to the gradient norm uses $\lVert\nabla_{\eta}S_n^c\rVert^{p}=\bigl(|\partial_{\gamma}S_n^c|^{2}+|\partial_{\lambda}S_n^c|^{2}\bigr)^{p/2}\le2^{p/2-1}\bigl(|\partial_{\gamma}S_n^c|^{p}+|\partial_{\lambda}S_n^c|^{p}\bigr)$ by $(a+b)^{q}\le2^{q-1}(a^{q}+b^{q})$ at $q=p/2$, the factor $2^{p/2-1}$ being absorbed into $N_p$; and, identically with $\mathsf{W}_n(\eta)$ in place of $B_n(\eta)$ at the level of the statistic, using $\|\mathsf{W}_n\|_{\mathrm{F}}^{p}\le\lambda_L^{-2p}$ (spectrally, $\sum_i\bigl(\nu_{i,n}(\gamma)/(\nu_{i,n}(\gamma)+\lambda)^{2}\bigr)^{2}\le\lambda_L^{-4}\sum_i\nu_{i,n}(\gamma)^{2}\le\lambda_L^{-4}$), $\|\mathsf{W}_nH\|_{\mathrm{F}}\le\|\mathsf{W}_n\|_{\mathrm{op}}\sqrt{d_\theta}\le\lambda_L^{-2}\sqrt{d_\theta}$, and the tracial centering $\bigl|\operatorname{tr}(\mathsf{W}_nHDH)\bigr|\le\lambda_L^{-2}\sigma_{\max}^{2}d_\theta$,
\[
    \sup_{\eta\in\Upsilon}\mathbb{E}\bigl|S_n^c(\eta)\bigr|^{p}\le N_p^{(0)}.
\]

For the drift's gradient, since $\boldsymbol{\varepsilon}_\perp$ and $\boldsymbol{\Sigma}_n$ do not depend on $\eta$, differentiating $a_n(\eta) = \operatorname{tr}(\mathsf{W}_n(\eta)\boldsymbol{\Sigma}_n)$ gives $\partial_{\eta_l}a_n(\eta) = \operatorname{tr}\bigl((\partial_{\eta_l}\mathsf{W}_n)(\eta)\,\boldsymbol{\Sigma}_n\bigr)$ for $l\in\{\gamma,\lambda\}$; the two coordinates are bounded separately. In the $\lambda$-coordinate, $\partial_\lambda\mathsf{W}_n = -2\bigl(\tfrac{\boldsymbol K_\gamma}{n}\bigr)R_\lambda^{3}$ by the second form of \eqref{eq:Wn-def} and $\partial_\lambda R_\lambda = -R_\lambda^{2}$, a symmetric function of $\boldsymbol K_\gamma/n$ with nonpositive spectrum $\{-2\nu_{i,n}(\gamma)(\nu_{i,n}(\gamma)+\lambda)^{-3}\}_{i\le n}$ (spectral mapping theorem), so its nuclear norm is
\[
    \bigl\|\partial_\lambda\mathsf{W}_n\bigr\|_{(1)} = 2\sum_{i\le n}\frac{\nu_{i,n}(\gamma)}{(\nu_{i,n}(\gamma)+\lambda)^{3}}\le2\lambda_L^{-1}\,\operatorname{tr}\mathsf{W}_n\le2\lambda_L^{-3},
\]
termwise by $(\nu+\lambda)^{-1}\le\lambda_L^{-1}$ and $\operatorname{tr}\mathsf{W}_n\le\lambda_L^{-2}$ (Step A); the trace--nuclear duality $|\operatorname{tr}(NS)|\le\|N\|_{(1)}\|S\|_{\mathrm{op}}$ with $\|\boldsymbol{\Sigma}_n\|_{\mathrm{op}}\le\|\hat{\boldsymbol{\Pi}}\|_{\mathrm{op}}^{2}\|D\|_{\mathrm{op}}\le\sigma_{\max}^{2}$ then gives $|\partial_\lambda a_n(\eta)|\le2\sigma_{\max}^{2}\lambda_L^{-3}$, uniformly in $(n,\eta)$. In the $\gamma$-coordinate, abbreviate $\mathsf T_\gamma := \partial_\gamma(\boldsymbol K_\gamma/n)$ for the tangent matrix; the product rule through the three factors of $\mathsf{W}_n = R_\lambda(\boldsymbol K_\gamma/n)R_\lambda$, with the resolvent derivative $\partial_\gamma R_\lambda = -R_\lambda\mathsf T_\gamma R_\lambda$ of the resolvent identity, reads
\[
    \partial_\gamma\mathsf{W}_n
    = -R_\lambda\mathsf T_\gamma R_\lambda\Bigl(\frac{\boldsymbol K_\gamma}{n}\Bigr)R_\lambda + R_\lambda\mathsf T_\gamma R_\lambda - R_\lambda\Bigl(\frac{\boldsymbol K_\gamma}{n}\Bigr)R_\lambda\mathsf T_\gamma R_\lambda ,
\]
and cyclicity of the trace isolates the tangent matrix against a bracket, $\partial_\gamma a_n(\eta) = \operatorname{tr}(\mathsf T_\gamma\mathcal B_\gamma)$ with
\[
    \mathcal B_\gamma := R_\lambda\boldsymbol{\Sigma}_nR_\lambda - R_\lambda\Bigl(\frac{\boldsymbol K_\gamma}{n}\Bigr)R_\lambda\boldsymbol{\Sigma}_nR_\lambda - R_\lambda\boldsymbol{\Sigma}_nR_\lambda\Bigl(\frac{\boldsymbol K_\gamma}{n}\Bigr)R_\lambda .
\]
Since $\mathsf T_\gamma$ is symmetric, only the symmetric part of the bracket contributes; with $\mathcal S_\gamma := \operatorname{sym}(\mathcal B_\gamma) = R_\lambda\boldsymbol{\Sigma}_nR_\lambda - 2\operatorname{sym}\bigl(\mathsf{W}_n\boldsymbol{\Sigma}_nR_\lambda\bigr)$, using $R_\lambda(\boldsymbol K_\gamma/n)R_\lambda = \mathsf W_n$ and the symmetry of its three factors,
\[
    \bigl\|\mathcal S_\gamma\bigr\|_{\mathrm{op}}
    \le\bigl\|R_\lambda\bigr\|_{\mathrm{op}}^{2}\bigl\|\boldsymbol{\Sigma}_n\bigr\|_{\mathrm{op}}+2\bigl\|\mathsf{W}_n\bigr\|_{\mathrm{op}}\bigl\|\boldsymbol{\Sigma}_n\bigr\|_{\mathrm{op}}\bigl\|R_\lambda\bigr\|_{\mathrm{op}}
    \le\sigma_{\max}^{2}\bigl(\lambda_L^{-2}+2\lambda_L^{-3}\bigr),
\]
by $\|\operatorname{sym}(M)\|_{\mathrm{op}}\le\|M\|_{\mathrm{op}}$, submultiplicativity, $\|R_\lambda\|_{\mathrm{op}}\le\lambda_L^{-1}$, $\|\mathsf W_n\|_{\mathrm{op}}\le\lambda_L^{-2}$ (Step A), and $\|\boldsymbol{\Sigma}_n\|_{\mathrm{op}}\le\sigma_{\max}^{2}$. It remains to pair $\mathsf T_\gamma$ against $\mathcal S_\gamma$. Transporting to the common space, write $u_\gamma(x) := J_\gamma k_\gamma(x,\cdot)\in\mathcal H_{\gamma_0}$ for the transported features of Fact~\ref{fact:pairdiff}, so that $\langle u_\gamma(X_j),u_\gamma(X_k)\rangle_{\mathcal H_{\gamma_0}} = k_\gamma(X_j,X_k)$ and $\|u_\gamma(x)\|^{2} = k_\gamma(x,x) = 1$; differentiating the Gram entries,
\[
    (\mathsf T_\gamma)_{jk} = \frac1n\,\partial_\gamma k_\gamma(X_j,X_k) = \frac2n\,\bigl\langle u_\gamma(X_j),\,\partial_\gamma u_\gamma(X_k)\bigr\rangle_{\mathcal H_{\gamma_0}},
\]
the product rule for the inner product giving the two half-derivatives, which coincide in the present pairing: only the symmetric part of the coefficient matrix enters against the symmetric $\mathcal S_\gamma$, and relabeling the summation indices $(j,k)\leftrightarrow(k,j)$ exchanges the two halves. Therefore, by the Cauchy--Schwarz inequality in $\mathcal H_{\gamma_0}^{n}$ applied to the tuples $\bigl(\sum_j(\mathcal S_\gamma)_{jk}u_\gamma(X_j)\bigr)_{k\le n}$ and $\bigl(\partial_\gamma u_\gamma(X_k)\bigr)_{k\le n}$,
\[
    \bigl|\operatorname{tr}(\mathsf T_\gamma\mathcal S_\gamma)\bigr|
    \le\frac2n\,\Bigl(\sum_{k}\Bigl\|\sum_{j}(\mathcal S_\gamma)_{jk}\,u_\gamma(X_j)\Bigr\|^{2}\Bigr)^{1/2}\Bigl(\sum_{k}\bigl\|\partial_\gamma u_\gamma(X_k)\bigr\|^{2}\Bigr)^{1/2}.
\]
The second factor is at most $\sqrt nL_d$ by Fact~\ref{fact:pairdiff}(i). For the first, expanding the squared norms through the Gram representation gives $\sum_k\bigl\|\sum_j(\mathcal S_\gamma)_{jk}u_\gamma(X_j)\bigr\|^{2} = \operatorname{tr}\bigl(\mathcal S_\gamma\mathcal S_\gamma^{\top}\boldsymbol K_\gamma\bigr)\le\|\mathcal S_\gamma\|_{\mathrm{op}}^{2}\operatorname{tr}\boldsymbol K_\gamma = n\|\mathcal S_\gamma\|_{\mathrm{op}}^{2}$, by $\operatorname{tr}(PQ)\le\|P\|_{\mathrm{op}}\operatorname{tr}(Q)$ for positive semidefinite $Q$ and $\operatorname{tr}\boldsymbol K_\gamma = n$; hence
\[
    \bigl|\partial_\gamma a_n(\eta)\bigr|\le2L_d\,\sigma_{\max}^{2}\bigl(\lambda_L^{-2}+2\lambda_L^{-3}\bigr),
    \qquad\text{uniformly in }(n,\eta),
\]
with $L_d$ the deterministic constant of Fact~\ref{fact:pairdiff}(i), depending only on $(\gamma_L,\gamma_U,d)$.

\medskip\noindent\emph{Step C: boundedness at one point.}\;
Fix $\eta_{0}\in\Upsilon$. The drift is bounded at $\eta_{0}$: under $H_0$,
\[
\begin{aligned}
    a_n(\eta_{0})
    &= \operatorname{tr}\bigl(\mathsf{W}_n(\eta_{0})\,\boldsymbol{\Sigma}_n\bigr)
    \le \bigl\|\boldsymbol{\Sigma}_n\bigr\|_{\mathrm{op}}\,\operatorname{tr}\mathsf{W}_n(\eta_{0})
    \le \sigma_{\max}^{2}\lambda_L^{-2},
\end{aligned}
\]
using $\boldsymbol{\Sigma}_n\preceq\sigma_{\max}^{2}\hat{\boldsymbol{\Pi}}\preceq\sigma_{\max}^{2}\boldsymbol{I}$ (Step B) and $\operatorname{tr}\mathsf{W}_n\le\lambda_L^{-2}$ (Step A), while $\sup_\eta\mathbb{E}|S_n^c(\eta)|^{p}\le N_p^{(0)}$ by Step B and $\sup_{\eta\in\Upsilon}|\rho_n(\eta)| = o_p(1)$ by the uniform-split paragraph, so $S_n^\perp(\eta_{0}) = S_n^c(\eta_{0}) + a_n(\eta_{0}) + \rho_n(\eta_{0}) = O_p(1) + O_p(1) + o_p(1) = O_p(1)$; under $H_{1,n}$ the exact split adds the shift contribution $C_n^{\mathrm{sh}}(\eta) = \bigl(\hat{\boldsymbol{\Pi}}\boldsymbol{R}/\sqrt{m+n}\bigr)^{\top}\mathsf{W}_n(\eta)\bigl(\hat{\boldsymbol{\Pi}}\boldsymbol{R}/\sqrt{m+n}\bigr)$, i.e.
\[
    \Bigl(\frac{\hat{\boldsymbol{\Pi}}\boldsymbol{R}}{\sqrt{m+n}}\Bigr)^{\!\top}\!\mathsf{W}_n(\eta)\Bigl(\frac{\hat{\boldsymbol{\Pi}}\boldsymbol{R}}{\sqrt{m+n}}\Bigr)
    \le \lambda_L^{-2}\,\frac{\boldsymbol{R}^\top\hat{\boldsymbol{\Pi}}\boldsymbol{R}}{m+n}
    = O_p\bigl(n/(m+n)\bigr)
    = O_p(1),
\]
by $\boldsymbol{R}^\top\hat{\boldsymbol{\Pi}}\boldsymbol{R}\le\boldsymbol{R}^\top\boldsymbol{R}$ and the LLN under $\mathbb{E}[R^{2}] < \infty$, while
\[
    \bigl|B_n^{\mathrm{sh}}(\eta)\bigr|
    \le \bigl\|\mathsf{A}_n M_n^{(\perp 0)}(\gamma)\bigr\|\,\bigl\|\mathsf{A}_n M_n^{(h)}(\gamma)\bigr\|
    \le \lambda_L^{-2}\sup_{\gamma\in[\gamma_L,\gamma_U]}\bigl\|M_n^{(\perp 0)}\bigr\|\,\sup_{\gamma\in[\gamma_L,\gamma_U]}\bigl\|M_n^{(h)}\bigr\|
    = O_p(1)
\]
by Cauchy--Schwarz and envelope four, the envelopes of the uniform-split paragraph, so $S_n^\perp(\eta_{0}) = S_n^c(\eta_{0}) + a_n(\eta_{0}) + 2B_n^{\mathrm{sh}}(\eta_{0}) + C_n^{\mathrm{sh}}(\eta_{0}) + \rho_n(\eta_{0}) = O_p(1)$ in both regimes.

\medskip\noindent\emph{Step D: asymptotic equicontinuity and tightness.}\;
By Step A the field $A_n^{\mathrm{sh}}$ is $C^{1}$ in $\eta$ on $\Upsilon$, Step A's quadratic-form representation being $\boldsymbol{x}$-generic, and so is $S_n^c = A_n^{\mathrm{sh}} - a_n$. Lemma~\ref{lem:equi} applies to $S_{m}:=A_n^{\mathrm{sh}}$ with the decomposition $S_{m} = S_{m}^{c} + a_{m} + e_{m}$, $S_{m}^{c} := S_n^c$, $a_{m} := a_n$, $e_{m} := \boldsymbol{0}$: hypothesis (i) is Step B's $\sup_{\eta\in\Upsilon}\mathbb{E}\bigl\|\nabla_\eta S_n^c(\eta)\bigr\|^{p}\le N_p$ with $N_p$ deterministic and regime-free, the three groups consuming only the conditional law of $\boldsymbol{\varepsilon}_0$ given $\mathbf{X}$, identical in both regimes, and the centering being built from $\boldsymbol{\Sigma}_n = \hat{\boldsymbol{\Pi}}D\hat{\boldsymbol{\Pi}}$, which is shift-free; hypothesis (ii) holds with $\Lambda := 2\sigma_{\max}^{2}\lambda_L^{-3}+2L_d\sigma_{\max}^{2}\bigl(\lambda_L^{-2}+2\lambda_L^{-3}\bigr)$ and $r_{m} := 0$, since by the mean-value form on the convex rectangle $\Upsilon$ and the drift-gradient bounds displayed in Step B, $\bigl|a_n(\eta)-a_n(\eta')\bigr| \le \sup_{\xi\in\Upsilon}\bigl\|\nabla_\eta a_n(\xi)\bigr\|\,\bigl\|\eta-\eta'\bigr\| \le \Lambda\,\bigl\|\eta-\eta'\bigr\|$, both coordinate bounds being deterministic and $(n,\eta)$-uniform; hypothesis (iii) is vacuous with $e_{m} := \boldsymbol{0}$. The pointwise bound is $\bigl|A_n^{\mathrm{sh}}(\eta_{0})\bigr| \le \lambda_L^{-2}\bigl(\sup_{\gamma\in[\gamma_L,\gamma_U]}\bigl\|M_n^{(\perp 0)}\bigr\|\bigr)^{2} = O_p(1)$ by envelopes four and one of the uniform-split paragraph. Its conclusions are asymptotic equicontinuity of $(A_n^{\mathrm{sh}})$,
\[
    \lim_{\delta\downarrow0}\ \limsup_{n\to\infty}\ \mathbb{P}\Bigl(\sup_{\|\eta-\eta'\|\le\delta}\bigl|A_n^{\mathrm{sh}}(\eta)-A_n^{\mathrm{sh}}(\eta')\bigr|>\varepsilon\Bigr)=0
    \qquad\text{for every }\varepsilon>0,
\]
and $\sup_{\eta\in\Upsilon}\bigl|A_n^{\mathrm{sh}}(\eta)\bigr|=O_p(1)$. Under $H_0$, $S_n^\perp = A_n^{\mathrm{sh}} + \rho_n$ with $\sup_{\eta\in\Upsilon}\bigl|\rho_n\bigr| = o_p(1)$, so the equicontinuity display and $\sup_{\eta\in\Upsilon}\bigl|S_n^\perp(\eta)\bigr| = O_p(1)$ follow for $S_n^\perp$; under $H_{1,n}$ the process-level assembly is completed in the local-alternative paragraph below.

\medskip\noindent\emph{Step E: tightness, finite-dimensional convergence, and the limit.}\;
By Step A the sample paths of $S_n^\perp$ are a.s.\ continuous on the compact set $\Upsilon$; with boundedness and asymptotic equicontinuity (Step D), Theorem~1.5.4 of \cite{vaart1997weak}, in its $C(\Upsilon)$ form for processes with continuous paths, yields asymptotic tightness of $(S_n^\perp)$ in $C(\Upsilon)\subset\ell^\infty(\Upsilon)$, and every subsequential limit has a.s.\ continuous sample paths; this is the source of the theorem's continuity claim for $S_\infty$.

For the finite-dimensional marginals, fix $\eta_1,\dots,\eta_J$ and consider the random elements $U_j := \varepsilon_{0,j}\bigl(k_{\gamma_1}(X_j,\cdot),\dots,k_{\gamma_J}(X_j,\cdot)\bigr)$ in the separable Hilbert space $\bigoplus_{l=1}^J\mathcal{H}_{\gamma_l}$; under $H_0$, $\mathbb{E}[U_j] = 0$ and $\mathbb{E}\lVert U_j\rVert^2 = \sum_l\mathbb{E}[\varepsilon_0^2 k_{\gamma_l}(X,X)] = J\mathbb{E}[\varepsilon_0^2] < \infty$, so the Hilbert-space CLT (Theorem~\ref{th:lln_clt_rkhs}) gives $(M_n^{(0)}(\gamma_1),\dots,M_n^{(0)}(\gamma_J)) \xrightarrow{d} (\xi_1,\dots,\xi_J)$, a Gaussian element of the product space.
By the pointwise-convergence paragraph and the proof of Corollary~\ref{cor:proj_null}, $M_n(\gamma_l) = M_n^{(0)}(\gamma_l) + M_n^{(\boldsymbol{\delta})}(\gamma_l) + o_p(1)$ with $M_n^{(\boldsymbol{\delta})}(\gamma_l) = -n^{-1/2}\boldsymbol{R}_{n,\gamma_l}^\top\hat{\mathbb{G}}\,\hat\beta = -\bigl(n^{-1}\boldsymbol{R}_{n,\gamma_l}^\top\hat{\mathbb{G}}\bigr)\sqrt{n}\,\hat\beta \xrightarrow{d} -A_{\gamma_l}^\top\Gamma_{\theta_0}^{-1}W$ (jointly in $l$: the same $\hat\beta$ enters for every $l$, and $(M_n^{(0)},\sqrt{n}\,\hat\beta) \xrightarrow{d} (\xi,\Gamma_{\theta_0}^{-1}W)$ jointly), so the projected embedding converges to the projected Gaussian element,
\[
    (M_n(\gamma_1),\dots,M_n(\gamma_J)) \xrightarrow{d} (G_{\perp,\gamma_1},\dots,G_{\perp,\gamma_J}),
    \qquad G_{\perp,\gamma_l} := \xi_l - A_{\gamma_l}^\top\Gamma_{\theta_0}^{-1}W .
\]
By Lemma~\ref{lm:C_and_Cn}, $C_{n,\gamma_l} \xrightarrow{p} C_{\gamma_l}$ for each fixed $l$ (no uniformity is needed at fixed $\gamma_l$), and the map $(x_1,\dots,x_J) \mapsto \bigl(\lVert(C_{\gamma_1}+\lambda_1I)^{-1}x_1\rVert^2,\allowbreak\dots,\lVert(C_{\gamma_J}+\lambda_JI)^{-1}x_J\rVert^2\bigr)$ is continuous; Slutsky's theorem yields the joint convergence of $(S_n^\perp(\eta_1),\dots,S_n^\perp(\eta_J))$ to the corresponding vector of quadratic forms in $(G_{\perp,\gamma_1},\dots,G_{\perp,\gamma_J})$.

It remains to identify the law of each coordinate with the spectral form of Corollary~\ref{cor:proj_null}. Fix $\eta$ and write $w_i(\gamma) := \mu_i(\gamma)(\mu_i(\gamma)+\lambda)^{-2}$ and $Z_i^\perp(\gamma) := \langle G_{\perp,\gamma},e_i(\gamma)\rangle_{\mathcal{H}_\gamma}/\sqrt{\mu_i(\gamma)}$ in the eigenbasis of $\boldsymbol{L}_{k_\gamma}$, so that $\lVert(C_\gamma+\lambda I)^{-1}G_{\perp,\gamma}\rVert^2 = \sum_{i\ge1}w_i(\gamma)(Z_i^\perp(\gamma))^{2}$, with $(C_\gamma+\lambda I)^{-1}$ acting by multiplication by $(\mu_i(\gamma)+\lambda)^{-1}$ and $\langle G_{\perp,\gamma},e_i(\gamma)\rangle^2 = \mu_i(\gamma)(Z_i^\perp(\gamma))^2$. The tail of the series is controlled pointwise in $\gamma$: the directional variances satisfy $\tau_i^\perp(\gamma) := \mathbb{E}[(Z_i^\perp(\gamma))^2] = \Sigma^\perp_{ii}(\gamma)\le\sigma_{\max}^{2}$ (Corollary~\ref{cor:proj_null}; the projection onto $\operatorname{span}(g_{\theta_0})$ contracts the $L^2$-norm), hence, for every fixed $\gamma$,
\[
    \mathbb{E}\Bigl[\sum_{i>K}w_i(\gamma)\,\bigl(Z_i^\perp(\gamma)\bigr)^{2}\Bigr]
    = \sum_{i>K}w_i(\gamma)\,\tau_i^\perp(\gamma)
    \le \sigma_{\max}^{2}\,\lambda_L^{-2}\sum_{i>K}\mu_i(\gamma)
    \;\longrightarrow\;0
    \qquad (K\to\infty),
\]
the last convergence being the term test for the convergent series $\sum_{i\ge1}\mu_i(\gamma) = 1$ at the fixed $\gamma$; no uniformity in $\gamma$ is claimed or needed at this step. The spectral form of the limit at each fixed $\eta$ is therefore that of Corollary~\ref{cor:proj_null}. Theorem~1.5.4 of \cite{vaart1997weak} now gives $S_n^\perp \rightsquigarrow S_\infty$ in $\ell^\infty(\Upsilon)$.

\medskip\noindent\textbf{The local alternative.}\; Under the local alternative at the total-sample rate (Section~\ref{sec:cv_selection}), the residual function is $\varepsilon_0 + R/\sqrt{m+n}$, and the exact split of the uniform-split paragraph already separates the statistic into the null-type field $A_n^{\mathrm{sh}}$, the shift terms $2B_n^{\mathrm{sh}}$ and $C_n^{\mathrm{sh}}$, and the uniformly negligible remainder $\rho_n$. Term $C_n^{\mathrm{sh}}$ is deterministic conditional on $\mathbf{X}$: $M_n^{(h)}(\gamma)$, defined in the uniform-split paragraph together with the $d$-display that licenses the replacement $\hat{\boldsymbol{\Pi}}R \to \Pi R$ at the embedding level, is $\sqrt{n/(m+n)}$ times a sample mean of the functions $x \mapsto (\Pi R)(x)k_\gamma(x,\cdot)$ up to that $o_p(1)$ replacement. The uniform law of large numbers over $\gamma$ is organized on the common space: transporting by the isometry $J_\gamma$ (which preserves norms), the objects
$J_\gamma M_n^{(h)}(\gamma)=\sqrt{n/(m+n)}\;n^{-1}\sum_j(\hat{\boldsymbol\Pi}R)(X_j)\,J_\gamma k_\gamma(X_j,\cdot)$
all live in $\mathcal H_{\gamma_0}$, and the class is Lipschitz in $\gamma$ with the $O_p(1)$ envelope
$n^{-1}\sum_j|\Pi R(X_j)|$, by the feature moduli recorded below (the feature-distance bound of the proof of Lemma~\ref{lem:lipschitz_common}(a), via the closed form of Lemma~\ref{lem:crosskernel}). At each fixed $\gamma$, $n^{-1}\sum_j(\boldsymbol{\Pi}R)(X_j)u_\gamma(X_j) \xrightarrow{p} m_h(\gamma)$ in $\mathcal{H}_{\gamma_0}$ by the Hilbert-space LLN, the summands being i.i.d.\ with $\mathbb{E}\|(\boldsymbol{\Pi}R)(X)u_\gamma(X)\|^2 \le \bar{k}\,\mathbb{E}[(\boldsymbol{\Pi}R)^2] < \infty$. Covering $[\gamma_L,\gamma_U]$ by finitely many cells of width $t$, the Lipschitz bound and the triangle inequality across the mesh give, for every fixed $t>0$, $\sup_\gamma\|M_n^{(h)}(\gamma) - \sqrt{1-c}\,m_h(\gamma)\| \le \max_l\|\,\cdot\,\|_{\gamma_l} + O_p(1)\,t$; letting $n\to\infty$ at fixed $t$ and then $t\downarrow0$ gives the uniform convergence in probability. Hence $M_n^{(h)}(\gamma) \xrightarrow{p} \sqrt{1-c}\,m_h(\gamma)$, $m_h(\gamma) := \mathbb{E}[(\Pi R)(X)k_\gamma(X,\cdot)]$, $c := \lim m/(m+n)$, uniformly in $\gamma$. The resolvent converges uniformly as well: moreover, $\sup_{\gamma\in[\gamma_L,\gamma_U]}\lVert C_{n,\gamma}-C_\gamma\rVert_{\mathrm{op}} = o_p(1)$ by Lemma~\ref{lem:lipschitz_common}(b), whose inputs are design-only. The population bound is the same in expectation, and it is here that the feature moduli enter: with $u_\gamma(x) := J_\gamma k_\gamma(x,\cdot)$, $\sup_x\|u_\gamma(x)\|\le\bar k^{1/2}$ and $\|u_\gamma(x)-u_{\gamma'}(x)\|\le L_k'|\gamma-\gamma'|$ (the feature-distance bound of the proof of Lemma~\ref{lem:lipschitz_common}(a), via the closed form of Lemma~\ref{lem:crosskernel}), the rank-one telescoping of part (b) of that lemma gives $\|C_\gamma-C_{\gamma'}\|_{\mathrm{op}}\le\|C_\gamma-C_{\gamma'}\|_{\mathrm{HS}}\le2\bar k^{1/2}L_k'\,|\gamma-\gamma'|$, in expectation. Since $(\boldsymbol{C}_{n,\gamma}+\lambda I)^{-1}-(\boldsymbol{C}_\gamma+\lambda I)^{-1} = (\boldsymbol{C}_{n,\gamma}+\lambda I)^{-1}(\boldsymbol{C}_\gamma-\boldsymbol{C}_{n,\gamma})(\boldsymbol{C}_\gamma+\lambda I)^{-1}$, both resolvent factors being of operator norm at most $\lambda_L^{-1}$, $\bigl\|(\boldsymbol{C}_{n,\gamma}+\lambda I)^{-1}-(\boldsymbol{C}_\gamma+\lambda I)^{-1}\bigr\|_{\mathrm{op}}\le\lambda_L^{-2}\lVert C_{n,\gamma}-C_\gamma\rVert_{\mathrm{op}}$ by the resolvent identity, $C_n^{\mathrm{sh}}(\eta) \xrightarrow{p} \lVert(C_\gamma+\lambda I)^{-1}\sqrt{1-c}\,m_h(\gamma)\rVert^2$ uniformly in $\eta$, a continuous deterministic limit.

Both shift terms are asymptotically equicontinuous, which closes the process level under $H_{1,n}$. For $2B_n^{\mathrm{sh}}$: transporting by the isometries $J_\gamma$ (which preserve norms and inner products), write $\widetilde{\mathsf{A}}_n(\eta) := J_\gamma\mathsf{A}_nJ_\gamma^{\top}$, $\widetilde{M}_n^{(\perp 0)}(\gamma) := J_\gamma M_n^{(\perp 0)}(\gamma)$, $\widetilde{M}_n^{(h)}(\gamma) := J_\gamma M_n^{(h)}(\gamma)$, and define the two factor maps on the common space $\mathcal{H}_{\gamma_0}$,
\[
    \widetilde{F}_{1,n}(\eta) := \widetilde{\mathsf{A}}_n(\eta)\,\widetilde{M}_n^{(\perp 0)}(\gamma),
    \qquad
    \widetilde{F}_{2,n}(\eta) := \widetilde{\mathsf{A}}_n(\eta)\,\widetilde{M}_n^{(h)}(\gamma),
\]
so that $B_n^{\mathrm{sh}}(\eta) = \bigl\langle\widetilde{F}_{1,n}(\eta),\,\widetilde{F}_{2,n}(\eta)\bigr\rangle_{\mathcal{H}_{\gamma_0}}$. Each factor is bounded in $\sup_\eta$, $\bigl\|\widetilde{F}_{1,n}(\eta)\bigr\| \le \lambda_L^{-1}\sup_{\gamma\in[\gamma_L,\gamma_U]}\bigl\|M_n^{(\perp 0)}\bigr\| = O_p(1)$ and $\bigl\|\widetilde{F}_{2,n}(\eta)\bigr\| \le \lambda_L^{-1}\sup_{\gamma\in[\gamma_L,\gamma_U]}\bigl\|M_n^{(h)}\bigr\| = O_p(1)$ by envelope four and envelopes one and two of the uniform-split paragraph. Each factor is asymptotically equicontinuous with an $O_p(1)$ modulus, $n$-uniform: writing $\Delta\widetilde{F}_{i,n} := \widetilde{F}_{i,n}(\eta)-\widetilde{F}_{i,n}(\eta')$ and inserting and subtracting $\widetilde{\mathsf{A}}_n(\eta')\widetilde{M}_n^{(i)}(\gamma)$,
\[
    \bigl\|\Delta\widetilde{F}_{i,n}\bigr\|
    \le \bigl\|\widetilde{\mathsf{A}}_n(\eta)-\widetilde{\mathsf{A}}_n(\eta')\bigr\|_{\mathrm{op}}\,\sup_{\eta\in\Upsilon}\bigl\|\widetilde{M}_n^{(i)}\bigr\|
    + \lambda_L^{-1}\bigl\|\widetilde{M}_n^{(i)}(\gamma)-\widetilde{M}_n^{(i)}(\gamma')\bigr\|,
\]
the resolvent increment being at most $\lambda_L^{-2}\bigl(L_C\,\bigl|\gamma-\gamma'\bigr| + \bigl|\lambda-\lambda'\bigr|\bigr)$ by Lemma~\ref{lem:lipschitz_common}(c) and the transported resolvent having operator norm at most $\lambda_L^{-1}$ by Step A. For $i=2$ the embedding increment satisfies
\[
    \bigl\|\widetilde{M}_n^{(h)}(\gamma)-\widetilde{M}_n^{(h)}(\gamma')\bigr\|
    \le \sqrt{\tfrac{n}{m+n}}\;L_k'\,n^{-1}\sum_j\bigl|(\hat{\boldsymbol{\Pi}}R)(X_j)\bigr|\,\bigl|\gamma-\gamma'\bigr|,
\]
by the feature Lipschitz bound, with
\[
    n^{-1}\sum_j\bigl|(\hat{\boldsymbol{\Pi}}R)(X_j)\bigr|
    \le \Bigl(n^{-1}\sum_jd_j^{2}\Bigr)^{1/2} + n^{-1}\sum_j\bigl|(\boldsymbol{\Pi}R)(X_j)\bigr|
    = o_p(1) + O_p(1)
    = O_p(1),
\]
the first term by Cauchy--Schwarz and the $d$-display of the uniform-split paragraph, the second by the LLN under $\mathbb{E}\bigl[(\boldsymbol{\Pi}R)^{2}\bigr] \le \mathbb{E}[R^{2}] < \infty$. For $i=1$ the embedding increment is at most $O_p(1)\bigl|\gamma-\gamma'\bigr| + o_p(1)$ in both regimes: in either regime the exact split $\widetilde{M}_n = \widetilde{M}_n^{(\perp 0)} + \widetilde{M}_n^{(h)} + \widetilde{r}_n$ holds ($\widetilde{M}_n^{(h)} \equiv \boldsymbol{0}$ under $H_0$), so the triangle inequality adds to the $O_p(1)$ Lipschitz modulus of $\widetilde{M}_n$ provided by Lemma~\ref{lem:lipschitz_common}(a) the modulus of $\widetilde{M}_n^{(h)}$ displayed above and the uniform envelope $\sup_{\gamma\in[\gamma_L,\gamma_U]}\bigl\|\widetilde{r}_n\bigr\| = o_p(1)$ (envelope three). The product identity gives, on $\bigl\{\|\eta-\eta'\|\le\delta\bigr\}$,
\[
\begin{aligned}
    \sup_{\|\eta-\eta'\|\le\delta}\bigl|B_n^{\mathrm{sh}}(\eta)-B_n^{\mathrm{sh}}(\eta')\bigr|
    &\le \sup_{\|\eta-\eta'\|\le\delta}\Bigl(\bigl\|\widetilde{F}_{1,n}(\eta)\bigr\|\,\bigl\|\Delta\widetilde{F}_{2,n}\bigr\|
    + \bigl\|\Delta\widetilde{F}_{1,n}\bigr\|\,\bigl\|\widetilde{F}_{2,n}(\eta')\bigr\|\Bigr)\\
    &\le \delta\,O_p(1) + o_p(1),
\end{aligned}
\]
with the $O_p(1)$ $n$-uniform, by adding and subtracting $\bigl\langle\widetilde{F}_{1,n}(\eta),\widetilde{F}_{2,n}(\eta')\bigr\rangle$, Cauchy--Schwarz, and the boundedness and modulus estimates above; the $o_p(1)$ term vanishes in $\limsup_{n\to\infty}$ for every fixed $\delta$, so $2B_n^{\mathrm{sh}}$ is asymptotically equicontinuous. For $C_n^{\mathrm{sh}}$: the uniform convergence to the continuous deterministic limit just established gives, by the triangle inequality and the modulus of continuity of the limit,
\[
    \sup_{\|\eta-\eta'\|\le\delta}\bigl|C_n^{\mathrm{sh}}(\eta)-C_n^{\mathrm{sh}}(\eta')\bigr|
    \le 2\sup_{\eta\in\Upsilon}\bigl|C_n^{\mathrm{sh}}(\eta) - C^{\mathrm{sh},\infty}(\eta)\bigr| + \omega_{C^{\mathrm{sh},\infty}}(\delta),
\]
with $C^{\mathrm{sh},\infty}(\eta) := \bigl\|(\boldsymbol{C}_\gamma+\lambda I)^{-1}\sqrt{1-c}\,m_h(\gamma)\bigr\|^{2}$ and $\omega_{C^{\mathrm{sh},\infty}}(\delta) := \sup_{\|\eta-\eta'\|\le\delta}\bigl|C^{\mathrm{sh},\infty}(\eta)-C^{\mathrm{sh},\infty}(\eta')\bigr| \downarrow 0$ as $\delta\downarrow0$ by continuity; so $C_n^{\mathrm{sh}}$ is asymptotically equicontinuous. Assembling the exact split, writing $\Delta T := T(\eta)-T(\eta')$ and taking the outer suprema over $\|\eta-\eta'\|\le\delta$, for $\|\eta-\eta'\|\le\delta$,
\[
    \bigl|S_n^\perp(\eta)-S_n^\perp(\eta')\bigr|
    \le \sup_{\|\eta-\eta'\|\le\delta}\bigl|\Delta A_n^{\mathrm{sh}}\bigr|
    + 2\sup_{\|\eta-\eta'\|\le\delta}\bigl|\Delta B_n^{\mathrm{sh}}\bigr|
    + \sup_{\|\eta-\eta'\|\le\delta}\bigl|\Delta C_n^{\mathrm{sh}}\bigr|
    + 2\sup_{\eta\in\Upsilon}\bigl|\rho_n\bigr|,
\]
by the triangle inequality, the four pieces being licensed by Step D, whose conclusions hold for $A_n^{\mathrm{sh}}$ in both regimes, by the $B_n^{\mathrm{sh}}$ and $C_n^{\mathrm{sh}}$ blocks above, and by the uniform-split paragraph; hence $S_n^\perp$ is asymptotically equicontinuous under $H_{1,n}$, with $\sup_{\eta\in\Upsilon}\bigl|S_n^\perp(\eta)\bigr| = O_p(1)$ the pointwise bound of Step C.

The finite-dimensional convergence follows from the same product-space argument with the Gaussian limit shifted; the multi-$\gamma$ joint CLT is the stacked Lindeberg CLT in
$\bigl(\bigoplus_{l}\mathcal H_{\gamma_l}\bigr)\times\mathbb R^{d_\theta}$: every linear combination of the coordinates is a univariate sum of independent summands whose Lindeberg ratios vanish: the summands are bounded by $\bar{k}\,\max_{j\le n}\lvert\varepsilon_{0,j}\rvert$ times unit-sphere coefficients in the Hilbert coordinates, and by $\bar G\max_{j\le n}\lvert\varepsilon_{0,j}\rvert$ times unit-sphere coefficients in the $\mathbb R^{d_\theta}$ coordinates, the score factors $\hat g_{\hat\theta}(X_j)$ entering the latter being conditionally deterministic by Remark~\ref{rem:cond_det} and uniformly bounded by the score bound of Assumption~\ref{ass:E}(i) in the $X$-measurable regime and by $\bar G+O_p(n^{-1/4})$ through the sup-norm clause of Assumption~\ref{ass:F}(iii) in the sieve regime; and $\max_{j\le n}\lvert\varepsilon_{0,j}\rvert=o_{p}(n^{1/2})$ under $\mathbb{E}[\varepsilon_{0}^{2}]<\infty$ (Assumption~\ref{ass:G}(iii)) by Chebyshev's inequality applied to the truncated variables, $\mathbb P(\max_{j\le n}\lvert\varepsilon_{0,j}\rvert>\varepsilon n^{1/2})\le n\,\mathbb E[\varepsilon_{0}^{2}\mathbf{1}(\lvert\varepsilon_{0}\rvert>\varepsilon n^{1/2})]/(\varepsilon^{2}n) = \mathbb E[\varepsilon_{0}^{2}\mathbf{1}(\lvert\varepsilon_{0}\rvert>\varepsilon n^{1/2})]/\varepsilon^{2}\to0$ by dominated convergence, which is the classical Lindeberg condition for independent summands with finite second moments, and the Cram\'er--Wold device lifts the univariate conclusions.
Precisely, $(M_n^{(\perp 0)}(\gamma_1),\dots,M_n^{(\perp 0)}(\gamma_J)) \xrightarrow{d} (G_{\perp,\gamma_1},\dots,G_{\perp,\gamma_J})$ (proof of Corollary~\ref{cor:proj_null}) and $M_n^{(h)}(\gamma_l) \xrightarrow{p} \sqrt{1-c}\,m_h(\gamma_l)$.
Hence $S_n^\perp(\eta_l) \xrightarrow{d} \lVert(C_{\gamma_l}+\lambda_lI)^{-1}(G_{\perp,\gamma_l} + \sqrt{1-c}\,m_h(\gamma_l))\rVert^2 = \sum_{i\ge1}\mu_i(\gamma_l)(\mu_i(\gamma_l)+\lambda_l)^{-2}(Z_{i,l}^\perp + \sqrt{1-c}\,h_i(\gamma_l))^2$ with $h_i(\gamma) := \mathbb{E}[(\Pi R)(X)\phi_i(X;\gamma)]$, matching $S_\infty(\eta_l;h_c)$ of Theorem~\ref{th:stat_uniform} ($h_c := \sqrt{1-c}\,h$). Asymptotic equicontinuity of $S_n^\perp$ under $H_{1,n}$ (the assembly above), the pointwise bound of Step C, and the finite-dimensional convergence of the preceding paragraph give, by Theorem~1.5.4 of \cite{vaart1997weak}, $S_n^\perp \rightsquigarrow S_\infty(\cdot;h_c)$ in $\ell^\infty(\Upsilon)$; the sample paths of the limit are continuous by the tightness in $C(\Upsilon)$ established in Step D and the assembly above.
\end{proof}

\subsubsection{Sample-Path Continuity of the Limit Fields}
\label{app:limit_continuity}
The limit fields of Theorems~\ref{th:stat_uniform} and~\ref{th:cv_uniform} admit modifications with H\"older-$\beta$ sample paths for every $\beta<1/2$; the derivation below is self-contained and is cited by the convergence proofs.
\emph{Continuity of the limit fields via functional calculus.}\;
The sample-path regularity of the limiting Gaussian fields follows from the functional calculus of
the analytic operator family $\gamma\mapsto\boldsymbol L_\gamma$, and no eigenpair regularity enters.
\emph{White-noise coupling and the operator representation.}\;
Couple the projected limit of Corollary~\ref{cor:proj_null} with an isonormal Gaussian noise $\zeta$ on
$L^{2}(\mathbb P_{X})$, i.e.\ a Gaussian white noise with
$\mathbb E[\langle\zeta,f\rangle\langle\zeta,g\rangle]=\langle f,g\rangle_{L^{2}(\mathbb P_{X})}$:
under this coupling the directional coordinates $\langle G_{\perp,\gamma},e_i(\gamma)\rangle_{\mathcal
H_{\gamma}}/\sqrt{\mu_i(\gamma)}$ and the score limit $W$ are the evaluations of the single noise
$\zeta$ on the features $\boldsymbol M_{\sigma}\boldsymbol\Pi\phi_i(\cdot;\gamma)$ and
$\boldsymbol M_{\sigma}g_{\theta_0}$, so $G_{\perp,\gamma}\overset{d}{=}\boldsymbol\Pi\boldsymbol
M_{\sigma}\zeta$ at process level and
\[
    Z_{i}^{\perp}(\gamma)=\langle\zeta,\boldsymbol M_{\sigma}\boldsymbol\Pi\phi_{i}(\cdot;\gamma)\rangle_{L^{2}(\mathbb P_{X})},
\]
with $\operatorname{Cov}(Z_{i}^{\perp}(\gamma),Z_{j}^{\perp}(\gamma))=\mathbb E[\sigma^{2}(X)(\boldsymbol\Pi\phi_{i})(X)(\boldsymbol\Pi\phi_{j})(X)]=\Sigma^{\perp}_{ij}(\gamma)$,
the covariance of the projected limit ($\boldsymbol\Pi\phi_i=\phi_i-\kappa_i g_{\theta_0}$ by
Corollary~\ref{cor:proj_null}). At differing bandwidths the same coupling gives the cross-$\gamma$ covariance
\[
\begin{gathered}
    \operatorname{Cov}\bigl(Z_{i}^{\perp}(\gamma),Z_{j}^{\perp}(\gamma')\bigr)
    =\bigl\langle\boldsymbol M_{\sigma}\boldsymbol\Pi\phi_{i}(\cdot;\gamma),\ \boldsymbol M_{\sigma}\boldsymbol\Pi\phi_{j}(\cdot;\gamma')\bigr\rangle_{L^{2}}\\
    =\mathbb E\bigl[\sigma^{2}(X)\,\bigl(\boldsymbol\Pi\phi_{i}\bigr)(X;\gamma)\,\bigl(\boldsymbol\Pi\phi_{j}\bigr)(X;\gamma')\bigr],
\end{gathered}
\]
reducing to $\Sigma_{ij}^{\perp}(\gamma)$ at equal parameters. Substituting into the spectral form,
\[
\begin{aligned}
    S_{\infty}(\gamma)&=\sum_{i\ge1}w_{i}(\gamma)\bigl(Z_{i}^{\perp}(\gamma)\bigr)^{2}
    =\langle\zeta,\boldsymbol A_{\gamma}\zeta\rangle_{L^{2}(\mathbb P_{X})},\\
    \boldsymbol A_{\gamma}&:=\boldsymbol M_{\sigma}\boldsymbol\Pi\boldsymbol W(\gamma)\boldsymbol\Pi\boldsymbol M_{\sigma},
    \qquad
    \boldsymbol W(\gamma):=\boldsymbol L_{\gamma}\bigl(\boldsymbol L_{\gamma}+\lambda\boldsymbol I\bigr)^{-2},
\end{aligned}
\]
where $w_{i}(\gamma):=\mu_{i}(\gamma)/(\mu_{i}(\gamma)+\lambda)^{2}$. The operator $\boldsymbol A_{\gamma}$ is
nonnegative, self-adjoint, and trace class with $\lVert\boldsymbol A_{\gamma}\rVert_{1}=\operatorname{tr}\boldsymbol
A_{\gamma}=\sum_{i\ge1}w_{i}(\gamma)\,\tau_{i}^{\perp}(\gamma)\le\sigma_{\max}^{2}\lambda_{L}^{-2}$
uniformly in $\gamma$, since $\mu_{i}/(\mu_{i}+\lambda)^{2}\le\mu_{i}\lambda_{L}^{-2}$ and
$\operatorname{tr}\boldsymbol L_{\gamma}=\sum_{i}\mu_{i}(\gamma)=\mathbb E[k_{\gamma}(X,X)]=1$; in particular
$\mathbb E\,S_{\infty}(\gamma)=\operatorname{tr}\boldsymbol A_{\gamma}\le\sigma_{\max}^{2}\lambda_{L}^{-2}$ and the
quadratic form is finite almost surely.

\emph{The increment bound: divided differences and the trace norm.}\;
The parameter increments are controlled by the Birman--Solomyak form of the spectral calculus. For self-adjoint $S,T$ with spectra in $[0,\bar k]$ and $f(t):=t/(t+\lambda)^{2}$, the increment admits the double-operator-integral representation
\[
    f(S)-f(T)=T_{f^{[1]}}(S-T),\qquad f^{[1]}(s,t):=\frac{f(s)-f(t)}{s-t},
\]
and on the Hilbert--Schmidt and trace classes the double operator integral is a contraction up to its symbol's supremum, $\lVert T_{f^{[1]}}(X)\rVert_{\mathrm{HS}}\le\sup_{s\ne t}\lvert f^{[1]}(s,t)\rvert\,\lVert X\rVert_{\mathrm{HS}}$ and likewise for the trace norm (Birman--Solomyak); the operator-level statement with the same constant fails for noncommuting self-adjoints and is not used. Here $\sup_{s\ne t}\lvert f^{[1]}(s,t)\rvert\le\sup_{[0,\bar k]}\lvert f'\rvert\le\lambda_{L}^{-2}$ by the mean value theorem. With $S=\boldsymbol L_{\gamma}$, $T=\boldsymbol L_{\gamma'}$, hence, using $\lVert\boldsymbol\Pi\rVert_{\mathrm{op}}\le1$,
$\lVert\boldsymbol M_{\sigma}\rVert_{\mathrm{op}}\le\sigma_{\max}$, and the kernel-level Hilbert--Schmidt bound of
Lemma~\ref{lem:lipschitz_common}(b),
\[
    \lVert\boldsymbol A_{\gamma}-\boldsymbol A_{\gamma'}\rVert_{\mathrm{HS}}
    \le\sigma_{\max}^{2}\lVert\boldsymbol W(\gamma)-\boldsymbol W(\gamma')\rVert_{\mathrm{HS}}
    \le\sigma_{\max}^{2}\lambda_{L}^{-2}\lVert\boldsymbol L_{\gamma}-\boldsymbol L_{\gamma'}\rVert_{\mathrm{HS}}
    \le C\,|\gamma-\gamma'|,
\]
and likewise, writing $\boldsymbol L_{\gamma}=\boldsymbol P_{\gamma}^{*}\boldsymbol P_{\gamma}$ with the transported-feature operator $(\boldsymbol P_{\gamma}h)(x)=\langle h,\boldsymbol u_{\gamma}(x,\cdot)\rangle_{\mathcal H_{\gamma_0}}$ (so that $\boldsymbol P_{\gamma}^{*}\boldsymbol P_{\gamma'}$ has kernel $\boldsymbol K_{\gamma\gamma'}=c_{\gamma\gamma'}k_{\gamma_{\gamma\gamma'}}$, Lemma~\ref{lem:crosskernel}), the trace-class difference satisfies $\lVert\boldsymbol L_{\gamma}-\boldsymbol L_{\gamma'}\rVert_{1}\le2\lVert\boldsymbol P_{\gamma}-\boldsymbol P_{\gamma'}\rVert_{\mathrm{HS}}=2\sqrt{2\,(1-c_{\gamma\gamma'})}\le C|\gamma-\gamma'|$.
The last inequality is the closed form of Lemma~\ref{lem:crosskernel}: $c_{\gamma\gamma'}=\bigl(4\gamma\gamma'/({\gamma+\gamma'})^{2}\bigr)^{d/4}$ with $4\gamma\gamma'/({\gamma+\gamma'})^{2}\in[c_{\min}^{4/d},1]$ by AM--GM, and $1-4\gamma\gamma'/({\gamma+\gamma'})^{2}=({\gamma-\gamma'})^{2}/({\gamma+\gamma'})^{2}\le(2\gamma_{L})^{-2}|\gamma-\gamma'|^{2}$, so $1-c_{\gamma\gamma'}\le(d/4)\max\{1,c_{\min}^{1-4/d}\}(2\gamma_{L})^{-2}|\gamma-\gamma'|^{2}$ by the mean value theorem applied to $t\mapsto t^{d/4}$ on $[c_{\min}^{4/d},1]$.
Cyclicity then gives $|\operatorname{tr}(\boldsymbol A_{\gamma}-\boldsymbol A_{\gamma'})|\le\sigma_{\max}^{2}\lambda_{L}^{-2}\lVert\boldsymbol L_{\gamma}-\boldsymbol L_{\gamma'}\rVert_{1}\le C|\gamma-\gamma'|$.
The $\lambda$-direction is identical, with the derivative of the rational weights bounded by a constant proportional to $\lambda_{L}^{-3}$, so both bounds hold uniformly over $\Upsilon$ with $\lVert\eta-\eta'\rVert$ in place of $|\gamma-\gamma'|$.
\emph{Fourth moments and Kolmogorov--Chentsov.}\;
The Wick identity (Isserlis' theorem) applied to the difference of the quadratic forms gives the exact increment
\[
    \mathbb E\bigl[\bigl(S_{\infty}(\gamma)-S_{\infty}(\gamma')\bigr)^{2}\bigr]
    =2\,\lVert\boldsymbol A_{\gamma}-\boldsymbol A_{\gamma'}\rVert_{\mathrm{HS}}^{2}
      +\bigl(\operatorname{tr}(\boldsymbol A_{\gamma}-\boldsymbol A_{\gamma'})\bigr)^{2}
    \le C\,\lVert\eta-\eta'\rVert^{2},
\]
so the canonical metric of the field is Lipschitz in the parameter. The centered increment is a second-order chaos; hypercontractivity of second-order chaoses gives $\mathbb E[(\Delta S-\mathbb E\Delta S)^{4}]\le C\,(\mathbb E[(\Delta S-\mathbb E\Delta S)^{2}])^{2}\le C\,\lVert\Delta\eta\rVert^{4}$ for a universal constant, and adding back the deterministic mean increment via the trace bound and Jensen's inequality yields $\mathbb E[(\Delta S)^{4}]\le C\,\lVert\Delta\eta\rVert^{4}$. The moment criterion of Kolmogorov--Chentsov with $p=4>d=\dim\Upsilon=2$ and exponent $4$ then yields a modification of $S_{\infty}$ whose sample paths are H\"older-$\beta$ for every $\beta<1/2$.

\emph{The remaining fields: drifts and the shifted fields.}\;
The same scheme governs the remaining limit fields. The projected chaos kernels are $K_{\eta}:=\boldsymbol M_{\sigma}\boldsymbol\Pi P_{\eta}\boldsymbol\Pi\boldsymbol M_{\sigma}$ with $P_{\eta}:=\boldsymbol L_{\gamma}(\boldsymbol L_{\gamma}+\lambda\boldsymbol I)^{-1}$, the isometry $\mathbb E[(\mathcal C(\eta)-\mathcal C(\eta'))^{2}]=2\lVert K_{\eta}-K_{\eta'}\rVert_{\mathrm{HS}}^{2}$ holds, and the resolvent identity gives the $K$-free modulus $\lVert P_{\eta}-P_{\eta'}\rVert_{\mathrm{HS}}\le\lambda_{U}\lambda_{L}^{-2}\lVert\boldsymbol L_{\gamma}-\boldsymbol L_{\gamma'}\rVert_{\mathrm{HS}}+\bigl(\lambda_{U}\lambda_{L}^{-2}+\lambda_{L}^{-1}\bigr)|\lambda-\lambda'|\le C\lVert\eta-\eta'\rVert$. The drift terms $\nu(\eta)$, $\rho(\eta)$, $\Delta(\eta)$, $V(\eta)$, and $A_{h}(\eta)$ are Hilbert--Schmidt inner products or $L^{2}$-norms of $P_{\eta}$ and of the resolvent against fixed functions and kernels, hence Lipschitz by the same modulus. The shifted fields of the local alternative split as $S_{\infty}(\gamma;h_{c})=\langle\zeta,\boldsymbol A_{\gamma}\zeta\rangle+2\langle\zeta,\boldsymbol b_{\gamma}\rangle+m_{\gamma}$ with $\boldsymbol b_{\gamma}:=\sqrt{1-c}\,\boldsymbol M_{\sigma}\boldsymbol\Pi\boldsymbol W(\gamma)\boldsymbol\Pi R$ and $m_{\gamma}:=(1-c)\langle\boldsymbol\Pi R,\boldsymbol W(\gamma)\boldsymbol\Pi R\rangle_{L^{2}}$, both Lipschitz in $\gamma$ by the divided difference. The conclusion therefore covers $\mathbb H$ and $\mathbb H_{h}^{a}$; the level statement of Theorem~\ref{th:cv_uniform}(b) and the demoted-fluctuation remark (Remark~\ref{rem:cv_fluct_demoted}) of its proof cover the fixed-alternative regime, no limit process being asserted there.

\subsubsection{Auxiliary lemmas for the cross-validation criterion process}
\label{aux:cv_lemmas}
The seven lemmas below are ordered by their first use in the proof of Theorem~\ref{th:cv_uniform}: the leave-one-out prediction linearization enters at the outset, the criterion decomposition built on it identifies the objects the theory consumes, the replacement cost of the feasible residuals transfers the feasible criterion to its infeasible counterpart, the conditional variance formulas control the centered pieces, and the bilinear representation with the chaos identification deliver the limit. Throughout the lemmas of this subsubsection, $\varepsilon_{\perp}^{0}:=\hat{\boldsymbol\Pi}\boldsymbol\varepsilon_{0}$ denotes the infeasible (oracle) orthogonalized residual vector of the selection sample, in the manuscript's standing notation, with components $\varepsilon_{\perp,i}^{0}$; the feasible vector of \eqref{eq:cv_loo} is $\varepsilon_{\perp}=\hat{\boldsymbol\Pi}\hat{\boldsymbol\varepsilon}$, and the transfer between the two criteria is performed once, by Lemma~\ref{lem:repl_cost}. Proof-level citations that run forward (the decomposition's proof invokes the variance bound \eqref{eq:quadvar_bound} of Lemma~\ref{lem:pairvar}) are unproblematic, since the cited statements are proved independently of the referring proofs. The operator and kernel tools shared with Theorem~\ref{th:stat_uniform} are collected in the preceding subsubsection.

\begin{lemma}[Leave-one-out prediction linearization]\label{lem:loo_prediction}
For $i=1,\dots,m$, let
\[
    C_\gamma^{(-i)}:=\frac{1}{m-1}\sum_{j\ne i}k_\gamma(X_j^{\mathrm{sel}},\cdot)\otimes k_\gamma(X_j^{\mathrm{sel}},\cdot),
    \qquad
    \hat m^{(-i)}:=\frac{1}{m-1}\sum_{j\ne i}\varepsilon_{\perp,j}^{0}\,k_\gamma(X_j^{\mathrm{sel}},\cdot)
\]
be the leave-one-out covariance operator and mean embedding, both built on the orthogonalized residuals $\varepsilon_{\perp,j}^{0}$, and let
$\hat w^{(-i)}_\eta:=(C^{(-i)}_\gamma+\lambda I)^{-1}\hat m^{(-i)}$
be the leave-one-out predictor; write $g_\eta$ for the smoothing kernel of \eqref{eq:g_eta}. Then, uniformly in $i=1,\dots,m$ and $\eta\in\Upsilon$,
\[
    \hat w^{(-i)}_\eta(X_i^{\mathrm{sel}})
    =L_i(\eta)+r_i,
    \qquad
    L_i(\eta):=\frac{1}{m-1}\sum_{j\ne i}\varepsilon_{\perp,j}^{0}\,g_\eta(X_i^{\mathrm{sel}},X_j^{\mathrm{sel}}),
\]
where $r_i=O_p(m^{-1})$ uniformly under $H_0$ and under local alternatives with signal $O(m^{-1/2})$, and $r_i=O_p(m^{-1/2})$ uniformly under a fixed alternative; the implied constants are of order $\lambda_L^{-2}$, entering through the resolvent identity.
\end{lemma}
\begin{proof}

Throughout the proof, bounds stated pointwise in $\eta$ are upgraded to their
$\sup_{\eta}$ form by a Lipschitz-plus-$\varepsilon$-net argument over the compact
$\Upsilon=[\gamma_L,\gamma_U]\times[\lambda_L,\lambda_U]$, exactly as in
Lemma~\ref{lem:lipschitz_common}(b) and Theorem~\ref{th:stat_uniform};
we record pointwise rates and note the Lipschitz input when it is not obvious.
We first realize $\hat w^{(-i)}_\eta(X_i^{\mathrm{sel}})$ as a smoothed average plus a
perturbation term. Put
\[
    \Delta_{j}:=k_{\gamma}(X_{j}^{\mathrm{sel}},\cdot)\otimes k_{\gamma}(X_{j}^{\mathrm{sel}},\cdot)-C_{\gamma},
    \qquad \mathbb{E}\Delta_{j}=0,
\]
so that, since $C_{\gamma}$ is the population second-moment operator of the standing notation, which is what makes the $\Delta_{j}$ centered, $C_{\gamma}^{(-i)}=C_{\gamma}+E_{i}$ with
\[
    E_{i}:=\frac{1}{m-1}\sum_{j\ne i}\Delta_{j},
\]
and
\[
    \hat m^{(-i)}:=\frac{1}{m-1}\sum_{j\ne i}\varepsilon_{\perp,j}^{0}k_{\gamma}(X_{j}^{\mathrm{sel}},\cdot)\in\mathcal{H}_{\gamma}.
\]
On $\mathcal{H}_{\gamma}$ the resolvent identity gives
\[
    (C_{\gamma}^{(-i)}+\lambda I)^{-1}
    =(C_{\gamma}+\lambda I)^{-1}
    -(C_{\gamma}^{(-i)}+\lambda I)^{-1}E_{i}(C_{\gamma}+\lambda I)^{-1}.
\]
Write
\[
    \mathsf R_{i}:=-(C_{\gamma}^{(-i)}+\lambda I)^{-1}E_{i}(C_{\gamma}+\lambda I)^{-1},
\]
so that $\hat w^{(-i)}_{\eta}:=(C_{\gamma}^{(-i)}+\lambda I)^{-1}\hat m^{(-i)}$ satisfies
\[
    \hat w^{(-i)}_{\eta}=(C_{\gamma}+\lambda I)^{-1}\hat m^{(-i)}+\mathsf R_{i}\hat m^{(-i)}.
\]
Evaluation at $X_{i}^{\mathrm{sel}}$ is the $\mathcal{H}_{\gamma}$ inner
product with the representer: for any $f\in\mathcal{H}_{\gamma}$,
$f(X_{i}^{\mathrm{sel}})=\langle k_{\gamma}(X_{i}^{\mathrm{sel}},\cdot),f\rangle_{\mathcal{H}_{\gamma}}$
by the reproducing property. Expanding the leading term by linearity and self-adjointness of
$(C_{\gamma}+\lambda I)^{-1}$ gives
\begin{align*}
    \bigl((C_{\gamma}+\lambda I)^{-1}\hat m^{(-i)}\bigr)(X_{i}^{\mathrm{sel}})
    &=\Big\langle k_{\gamma}(X_{i}^{\mathrm{sel}},\cdot),
        (C_{\gamma}+\lambda I)^{-1}\hat m^{(-i)}\Big\rangle_{\mathcal{H}_{\gamma}} \\
    &=\frac{1}{m-1}\sum_{j\ne i}\varepsilon_{\perp,j}^{0}
      \Big\langle k_{\gamma}(X_{i}^{\mathrm{sel}},\cdot),
              (C_{\gamma}+\lambda I)^{-1}k_{\gamma}(X_{j}^{\mathrm{sel}},\cdot)
      \Big\rangle_{\mathcal{H}_{\gamma}}.
\end{align*}
By definition \eqref{eq:g_eta} the inner product is the population
smoothing kernel
\[
    g_{\eta}(x,y):=\big\langle k_{\gamma}(x,\cdot),
        (C_{\gamma}+\lambda I)^{-1}k_{\gamma}(y,\cdot)\big\rangle_{\mathcal{H}_{\gamma}},
\]
the kernel of $(C_{\gamma}+\lambda I)^{-1}C_{\gamma}$ on
$L^{2}(\mathbb{P}_{X})$, whence
\[
    \bigl((C_{\gamma}+\lambda I)^{-1}\hat m^{(-i)}\bigr)(X_{i}^{\mathrm{sel}})
    =\frac{1}{m-1}\sum_{j\ne i}\varepsilon_{\perp,j}^{0}\,
      g_{\eta}(X_{i}^{\mathrm{sel}},X_{j}^{\mathrm{sel}}).
\]
Putting $r_{i}:=(\mathsf R_{i}\hat m^{(-i)})(X_{i}^{\mathrm{sel}})
 =\langle k_{\gamma}(X_{i}^{\mathrm{sel}},\cdot),\mathsf R_{i}\hat m^{(-i)}\rangle_{\mathcal{H}_{\gamma}}$
gives the claimed algebraic linearization
\[
    \hat w^{(-i)}_{\eta}(X_{i}^{\mathrm{sel}})
    =\frac{1}{m-1}\sum_{j\ne i}\varepsilon_{\perp,j}^{0}\,
      g_{\eta}(X_{i}^{\mathrm{sel}},X_{j}^{\mathrm{sel}})+r_{i},
\]
whose leading term is exactly $L_{i}(\eta)$ of the statement; hence $\hat w^{(-i)}_{\eta}(X_{i}^{\mathrm{sel}})=L_{i}(\eta)+r_{i}$ as an identity.

It remains to bound $r_{i}$. Since $r_{i} =\langle k_{\gamma}(X_{i}^{\mathrm{sel}},\cdot),\mathsf R_{i}\hat m^{(-i)}\rangle_{\mathcal{H}_{\gamma}}$ and $\mathsf R_{i}=-(C_{\gamma}^{(-i)}+\lambda I)^{-1}E_{i}(C_{\gamma}+\lambda I)^{-1}$, the rate of $r_{i}$ is governed by two factors: the covariance perturbation $E_{i}$ and the size of the leave-one-out embedding $\hat m^{(-i)}$.

\emph{Covariance perturbation.}
For $u\otimes v$, $\lVert u\otimes v\rVert_{\mathrm{HS}}=\lVert u\rVert\lVert v\rVert$;
by the reproducing property
$\lVert k_{\gamma}(X,\cdot)\rVert_{\mathcal{H}_{\gamma}}^{2}=k_{\gamma}(X,X)$,
hence
\[
    \lVert k_{\gamma}(X,\cdot)\otimes k_{\gamma}(X,\cdot)\rVert_{\mathrm{HS}}
    =k_{\gamma}(X,X)\le\bar k
\]
pointwise by Assumption~\ref{ass:B}(i), and by Jensen
$\lVert C_{\gamma}\rVert_{\mathrm{HS}}\le\mathbb{E}k_{\gamma}(X,X)\le\bar k$.
Consequently
\[
    \lVert\Delta_{j}\rVert_{\mathrm{HS}}\le2\bar k\quad\text{hence}\quad
    \lVert\Delta_{j}\rVert_{\mathrm{op}}\le2\bar k\ \ \text{a.s.,}
\]
with deterministic envelope $2\bar k$ uniform in $\gamma$.
Hilbert--Schmidt orthogonality of centered independent summands gives, for
each fixed $i$,
\[
    \mathbb{E}\lVert E_{i}\rVert_{\mathrm{op}}^{2}
    \le\mathbb{E}\lVert E_{i}\rVert_{\mathrm{HS}}^{2}
    =\frac{1}{(m-1)^{2}}\sum_{j\ne i}\mathbb{E}\lVert\Delta_{j}\rVert_{\mathrm{HS}}^{2}
    \le\frac{4\bar k^{2}}{m-1},
\]
so $\lVert E_{i}\rVert_{\mathrm{op}}=O_{p}(m^{-1/2})$ pointwise.
Uniformity in $i$ needs no union bound: with
$\bar\Delta:=m^{-1}\sum_{j=1}^{m}\Delta_{j}$,
$\lVert\bar\Delta\rVert_{\mathrm{op}}=O_{p}(m^{-1/2})$ by the same computation,
and
\[
    E_{i}=\frac{m}{m-1}\,\bar\Delta-\frac{1}{m-1}\,\Delta_{i},\qquad
    \lVert E_{i}\rVert_{\mathrm{op}}\le\frac{m}{m-1}\lVert\bar\Delta\rVert_{\mathrm{op}}
    +\frac{2\bar k}{m-1}=O_{p}(m^{-1/2})
\]
uniformly in $i$.  For uniformity in $\eta=(\gamma,\lambda)$ only the
$\gamma$-dependence matters.  For the Gaussian kernel,
$\partial_{\gamma}k_{\gamma}(x,y)=-\lVert x-y\rVert^{2}k_{\gamma}(x,y)$ satisfies
$|\partial_{\gamma}k_{\gamma}(x,y)|\le\sup_{t\ge0}t\,e^{-\gamma_{L}t}=1/(e\gamma_{L})$
for all $\gamma\ge\gamma_{L}$, the substitution $s=\gamma_{L}t$ reducing the supremum to
$\gamma_{L}^{-1}\sup_{s\ge0}s\,e^{-s}$, attained at $s=1$; the envelope is recorded for use in the
proof of Lemma~\ref{lem:loo_linearization} below. The $\gamma$-Lipschitz bound for the summands follows from the cross-kernel calculus of Lemma~\ref{lem:crosskernel}. With the transported unit features
$\tilde u_{\gamma}(x):=J_{\gamma}k_{\gamma}(x,\cdot)$, which have unit norm in the common
space $\mathcal{H}_{\gamma_0}$ by the unit normalization $k_{\gamma}(x,x)=1$ of the Gaussian
kernel, the closed form \eqref{eq:crosskernel_closed} gives at $x=y$ the inner
product $\langle\tilde u_{\gamma_{1}}(x),\tilde u_{\gamma_{2}}(x)\rangle_{\mathcal{H}_{\gamma_0}}=c_{12}$ with
\[
    c_{12}=\Bigl(1-\frac{(\gamma_{1}-\gamma_{2})^{2}}{(\gamma_{1}+\gamma_{2})^{2}}\Bigr)^{d/4}.
\]
For unit vectors $u,v$ in a Hilbert space,
$\lVert u\otimes u-v\otimes v\rVert_{\mathrm{HS}}^{2}
=\operatorname{tr}\bigl[(u\otimes u-v\otimes v)^{2}\bigr]
=2-2\langle u,v\rangle^{2}$
by $(u\otimes u)(v\otimes v)=\langle u,v\rangle\,u\otimes v$ and unit idempotence,
so at the anchor $X_{j}$,
\[
    \bigl\lVert\tilde u_{\gamma_{1}}(X_{j})\otimes\tilde u_{\gamma_{1}}(X_{j})
      -\tilde u_{\gamma_{2}}(X_{j})\otimes\tilde u_{\gamma_{2}}(X_{j})\bigr\rVert_{\mathrm{HS}}
    =\bigl(2(1-c_{12}^{2})\bigr)^{1/2}
    \le\tfrac{1}{2}\max\{\sqrt2,\sqrt d\}\,\gamma_{L}^{-1}\,\lvert\gamma_{1}-\gamma_{2}\rvert,
\]
by $1-c_{12}^{2}=1-(1-s^{2})^{d/2}$ with
$s:=\lvert\gamma_{1}-\gamma_{2}\rvert/(\gamma_{1}+\gamma_{2})$, the elementary inequality
$1-(1-u)^{\alpha}\le\max(1,\alpha)\,u$ for $u\in[0,1]$
(concavity at $\alpha\le1$, Bernoulli at $\alpha\ge1$) at $u=s^{2}$, $\alpha=d/2$, and
$s\le\lvert\gamma_{1}-\gamma_{2}\rvert/(2\gamma_{L})$. The population term is bounded on
the common space by Jensen's inequality for the Bochner integral: transported,
$C_{\gamma}$ is the Bochner integral
$\mathbb{E}[\tilde u_{\gamma}(X)\otimes\tilde u_{\gamma}(X)]$, so
\[
\begin{aligned}
    \bigl\lVert J_{\gamma_{1}}C_{\gamma_{1}}J_{\gamma_{1}}^{-1}
      -J_{\gamma_{2}}C_{\gamma_{2}}J_{\gamma_{2}}^{-1}\bigr\rVert_{\mathrm{HS}}
    &\le\mathbb{E}\bigl\lVert\tilde u_{\gamma_{1}}(X)\otimes\tilde u_{\gamma_{1}}(X)
      -\tilde u_{\gamma_{2}}(X)\otimes\tilde u_{\gamma_{2}}(X)\bigr\rVert_{\mathrm{HS}}\\
    &\le\tfrac{1}{2}\max\{\sqrt2,\sqrt d\}\,\gamma_{L}^{-1}\,\lvert\gamma_{1}-\gamma_{2}\rvert.
\end{aligned}
\]
The triangle inequality, with both summands of $\Delta_{j}$ transported to the
common space by the isometry $J_{\gamma}$, therefore gives
\[
    \lVert\Delta_{j}(\gamma_{1})-\Delta_{j}(\gamma_{2})\rVert_{\mathrm{HS}}
    \le\max\{\sqrt2,\sqrt d\}\,\gamma_{L}^{-1}\,\lvert\gamma_{1}-\gamma_{2}\rvert
\]
with a deterministic constant, while each $\Delta_{j}$ retains the a.s.\ envelope
$\lVert\Delta_{j}\rVert_{\mathrm{op}}\le2\bar k$ recorded above, so
$E_{i}=(m-1)^{-1}\sum_{j\ne i}\Delta_{j}$ is Lipschitz in $\gamma$, uniformly in $i$,
with the same a.s.\ envelope.  By the standing net argument,
\[
    \sup_{i}\sup_{\eta\in\Upsilon}\lVert E_{i}\rVert_{\mathrm{op}}=O_{p}(m^{-1/2}).
\]

\emph{Mean embedding.}
Write $\hat m:=m^{-1}\sum_{j}\varepsilon_{\perp,j}^{0}k_{\gamma}(X_{j},\cdot)$ and
\[
    \hat m^{(-i)}=\frac{m}{m-1}\,\hat m-\frac{1}{m-1}\,\varepsilon_{\perp,i}^{0}k_{\gamma}(X_{i},\cdot),
\]
so that
\[
    \lVert\hat m^{(-i)}\rVert_{\mathcal{H}_{\gamma}}
    \le\frac{m}{m-1}\lVert\hat m\rVert_{\mathcal{H}_{\gamma}}
    +\frac{\bar k^{1/2}\max_{j}|\varepsilon_{\perp,j}^{0}|}{m-1}.
\]
Conditioning on $\mathbf X$, under the convention of Remark~\ref{rem:cond_det},
write $\boldsymbol g_{m}:=\mathbb E[\boldsymbol\varepsilon_{0}\mid\mathbf X]$ for the $\mathbf X$-measurable
conditional mean of the true residual and $\boldsymbol v:=\boldsymbol\varepsilon_{0}-\boldsymbol g_{m}$
for the conditionally centered noise, and decompose
\[
    \hat m=\mathbb E[\hat m\mid\mathbf X]+\bar m,
    \qquad
    \mathbb E[\hat m\mid\mathbf X]=\frac{1}{m}\sum_{j}(\hat{\boldsymbol\Pi}\boldsymbol g_{m})_{j}\,k_{\gamma}(X_{j},\cdot),
    \qquad
    \bar m=\frac{1}{m}\sum_{j}(\hat{\boldsymbol\Pi}\boldsymbol v)_{j}\,k_{\gamma}(X_{j},\cdot).
\]
The centered part has conditional covariance
$\mathbb E[(\hat{\boldsymbol\Pi}\boldsymbol v)(\hat{\boldsymbol\Pi}\boldsymbol v)^{\top}\mid\mathbf X]
=\hat{\boldsymbol\Pi}D\hat{\boldsymbol\Pi}$ with
$D:=\operatorname{diag}(\sigma^{2}(X_{1}),\dots,\sigma^{2}(X_{m}))$, the noise being
conditionally independent and centered; expanding in the reproducing kernel and
writing $c_{jk}:=(\hat{\boldsymbol\Pi}D\hat{\boldsymbol\Pi})_{jk}$,
\[
\begin{aligned}
    \mathbb E\bigl[\lVert\bar m\rVert_{\mathcal{H}_{\gamma}}^{2}\mid\mathbf X\bigr]
    &=\frac{1}{m^{2}}\sum_{j,k}c_{jk}\,k_{\gamma}(X_{j},X_{k})
      \le\frac{\bar k}{m^{2}}\sum_{j}c_{jj}+\frac{\bar k}{m^{2}}\sum_{j\ne k}|c_{jk}|\\
    &\le\frac{\bar k\,\sigma_{\max}^{2}}{m}+\bar k\sup_{j\ne k}|c_{jk}|
    =O_{p}(m^{-1}).
\end{aligned}
\]
Here the diagonal is bounded by the Loewner chain
$D\le\sigma_{\max}^{2}\boldsymbol I$ (Assumption~\ref{ass:G}(iii)), hence
$\hat{\boldsymbol\Pi}D\hat{\boldsymbol\Pi}\le\sigma_{\max}^{2}\hat{\boldsymbol\Pi}$ by congruence and
idempotence, so that
$c_{jj}\le\sigma_{\max}^{2}\hat{\boldsymbol\Pi}_{jj}\le\sigma_{\max}^{2}$; the off-diagonal
entries obey $c_{jk}=O_{p}(m^{-1})$ uniformly for $j\ne k$, the
$c_{jk}$-analysis of Section~\ref{sec:uniform_limit}, which is $\eta$-free and
uses only $|k_{\gamma}|\le\bar k$; and
$\sum_{j\ne k}|c_{jk}|\le m(m-1)\sup_{j\ne k}|c_{jk}|$. The conditional second
moment $\mathbb E[\lVert\bar m\rVert_{\mathcal{H}_{\gamma}}^{2}\mid\mathbf X]$ is thus $\mathbf X$-measurable
and $O_{p}(m^{-1})$, and the conditional Markov inequality applied through the
bounded transform $\min(1,\cdot)$,
\[
    \mathbb{P}\bigl(\lVert\bar m\rVert_{\mathcal{H}_{\gamma}}>\varepsilon\bigr)
    \le\mathbb{E}\Bigl[\min\Bigl(1,\frac{\mathbb E[\lVert\bar m\rVert_{\mathcal{H}_{\gamma}}^{2}\mid\mathbf X]}{\varepsilon^{2}}\Bigr)\Bigr],
\]
together with the boundedness of $\min(1,\cdot)$, which turns the convergence in
probability of the conditional second moment into convergence of its expectation,
gives $\lVert\bar m\rVert_{\mathcal{H}_{\gamma}}=O_{p}(m^{-1/2})$. The conditional mean is
$\mathbf X$-measurable, and by the Cauchy--Schwarz inequality
$\sum_{j}|a_{j}|\le m^{1/2}\lVert a\rVert$ and
$\lVert\hat{\boldsymbol\Pi}\rVert_{\mathrm{op}}\le1$ (orthogonal projector),
\[
    \lVert\mathbb E[\hat m\mid\mathbf X]\rVert_{\mathcal{H}_{\gamma}}
    \le\frac{\bar k^{1/2}}{m}\sum_{j}\bigl|(\hat{\boldsymbol\Pi}\boldsymbol g_{m})_{j}\bigr|
    \le\bar k^{1/2}m^{-1/2}\lVert\hat{\boldsymbol\Pi}\boldsymbol g_{m}\rVert
    \le\bar k^{1/2}m^{-1/2}\lVert\boldsymbol g_{m}\rVert,
\]
which is $0$ under $H_{0}$, $O_{p}(m^{-1/2})$ under a local alternative with
signal $O(m^{-1/2})$, and $O_{p}(1)$ under a fixed alternative, since
$\lVert\boldsymbol g_{m}\rVert$ is $0$, $O_{p}(1)$, and $O_{p}(m^{1/2})$ in the three
regimes by the law of large
numbers, the embedding being then of the order of the population object
$m_{\Pi f}:=\mathbb{E}[(\Pi f)(X)k_{\gamma}(X,\cdot)]$. Hence
$\lVert\hat m\rVert_{\mathcal{H}_{\gamma}}=O_{p}(m^{-1/2})$ in the null and local regimes,
and $O_{p}(1)$ under a fixed alternative.
For the leftover term, the trivial bound
$\max_{j}|\varepsilon_{\perp,j}^{0}|
\le\bigl(\sum_{j}(\varepsilon_{\perp,j}^{0})^{2}\bigr)^{1/2}=O_{p}(m^{1/2})$
(no moments beyond those already assumed) puts the second term in the display for $\lVert\hat m^{(-i)}\rVert$ at $O_{p}(m^{-1/2})$, of the same order as the leading term and therefore absorbed by the triangle inequality, and
\[
    \lVert\hat m^{(-i)}\rVert_{\mathcal{H}_{\gamma}}=O_{p}(m^{-1/2})
    \ \text{uniformly in $i$}
\]
in the null and local regimes, and $O_{p}(1)$ under a fixed alternative;
$\eta$-uniformity holds by the standing net argument (the pointwise bound uses only $\lVert k_{\gamma}\rVert_{\infty}\le\bar k$ and the $\eta$-free $c_{jk}$ analysis).

\emph{Resolvent and evaluation at $X_{i}$.}
$\lVert(C_{\gamma}+\lambda I)^{-1}\rVert_{\mathrm{op}}\le\lambda_{L}^{-1}$,
$\lVert(C_{\gamma}^{(-i)}+\lambda I)^{-1}\rVert_{\mathrm{op}}\le\lambda_{L}^{-1}$,
so for $\mathsf R_{i}$ defined above
\[
    \lVert \mathsf R_{i}\rVert_{\mathrm{op}}\le\lambda_{L}^{-2}\lVert E_{i}\rVert_{\mathrm{op}}
    =O_{p}(m^{-1/2})
\]
uniformly in $i,\eta$, and
\[
    \lVert \mathsf R_{i}\hat m^{(-i)}\rVert_{\mathcal{H}_{\gamma}}
    =O_{p}(m^{-1})\ \text{(null/local)},\quad O_{p}(m^{-1/2})\ \text{(fixed)},
\]
uniformly in $i,\eta$.  The remainder at $X_{i}^{\mathrm{sel}}$ is a reproducing evaluation:
for $f=\mathsf R_{i}\hat m^{(-i)}$,
\[
    |r_{i}|=|f(X_{i}^{\mathrm{sel}})|
    =|\langle k_{\gamma}(X_{i}^{\mathrm{sel}},\cdot),f\rangle_{\mathcal{H}_{\gamma}}|
    \le\lVert k_{\gamma}(X_{i}^{\mathrm{sel}},\cdot)\rVert_{\mathcal{H}_{\gamma}}\lVert f\rVert_{\mathcal{H}_{\gamma}}
    \le\bar k^{1/2}\lVert \mathsf R_{i}\hat m^{(-i)}\rVert_{\mathcal{H}_{\gamma}},
\]
hence
\[
    |r_{i}|\le\bar k^{1/2}\lVert \mathsf R_{i}\hat m^{(-i)}\rVert_{\mathcal{H}_{\gamma}}
    =\begin{cases}
      O_{p}(m^{-1}),&\text{null and local alternatives},\\[2pt]
      O_{p}(m^{-1/2}),&\text{fixed alternative},
     \end{cases}
\]
uniformly in $i,\eta$, with constants depending only on
$(\bar k,\lambda_{L},\sigma_{\max})$ and the constants of the $c_{jk}$-orders of
Section~\ref{sec:uniform_limit}.  In particular the leading $g_{\eta}$-average
$L_{i}$ is $O_{p}(m^{-1/2})$ uniformly, while $r_{i}$ is an order smaller in
the null and local regimes. This proves the lemma.
\end{proof}

\begin{lemma}[Leave-one-out criterion decomposition]\label{lem:loo_linearization}
Adopt the setting and notation of Lemma~\ref{lem:loo_prediction}, with the linearized predictor $L_i$ and its remainder $r_i$ as defined there. Under $H_0$ and local alternatives the remainders $r_i$ contribute $o_p(m^{-1})$, and expanding the square of the leave-one-out error gives, uniformly in $\eta$, 
\[
    \mathrm{CV}_m(\eta)=\underbrace{\frac{1}{m}\sum_{i=1}^m(\varepsilon_{\perp,i}^{0})^{2}}_{(\mathrm{I})}
    \;-\;\underbrace{2\,U_m(\eta)}_{(\mathrm{II})}
    \;+\;\underbrace{Q_m(\eta)}_{(\mathrm{III})}
    \;+\;\underbrace{R_m(\eta)}_{(\mathrm{IV})}
    \;+\;o_p(m^{-1}),
\]
with the off-diagonal products of the cross term
\[
    U_m(\eta):=\frac{1}{m(m-1)}\sum_{i\ne j}\varepsilon_{\perp,i}^{0}\varepsilon_{\perp,j}^{0}\,g_\eta(X_i^{\mathrm{sel}},X_j^{\mathrm{sel}}),
\]
and the diagonal and off-diagonal contributions of the squared prediction,
\[
    Q_m(\eta):=\frac{1}{m(m-1)^2}\sum_{i\ne j}(\varepsilon_{\perp,j}^{0})^{2}\,g_\eta(X_i^{\mathrm{sel}},X_j^{\mathrm{sel}})^2,
\]
and
\[
    R_m(\eta):=\frac{2}{m(m-1)^2}\sum_{i=1}^m\sum_{\substack{j<k\\ j,k\ne i}}
      g_\eta(X_i^{\mathrm{sel}},X_j^{\mathrm{sel}})\,g_\eta(X_i^{\mathrm{sel}},X_k^{\mathrm{sel}})\,\varepsilon_{\perp,j}^{0}\varepsilon_{\perp,k}^{0};
\]
in $R_m$ each unordered pair $\{j,k\}$ is counted once, the factor $2$ coming from the two orders of the double sum. Under a fixed alternative the same decomposition holds at the weaker rate; the proof of Theorem~\ref{th:cv_uniform}(b) applies its algebra once to the full orthogonalized residual and closes the remainder at $O_p(m^{-1/2})$ uniformly in $\eta$, through the fixed-alternative clause of Lemma~\ref{lem:loo_prediction}.
\end{lemma}
\begin{proof}
The leave-one-out CV criterion is
$\mathrm{CV}_{m}(\eta)=m^{-1}\sum_{i=1}^{m}
(\varepsilon_{\perp,i}^{0}-\hat w^{(-i)}_{\eta}(X_{i}^{\mathrm{sel}}))^{2}$,
i.e.\ with $L_{i}$ and $r_{i}$ as in Lemma~\ref{lem:loo_prediction},
\[
    \mathrm{CV}_{m}(\eta)=\frac{1}{m}\sum_{i=1}^{m}
    (\varepsilon_{\perp,i}^{0}-L_{i}-r_{i})^{2}.
\]
Expanding the square algebraically,
\begin{align*}
    \frac{1}{m}\sum_{i}(\varepsilon_{\perp,i}^{0}-L_{i}-r_{i})^{2}
    &=\frac{1}{m}\sum_{i}(\varepsilon_{\perp,i}^{0})^{2}
      -\frac{2}{m}\sum_{i}\varepsilon_{\perp,i}^{0}L_{i}
      +\frac{1}{m}\sum_{i}L_{i}^{2}\\
    &\quad-\frac{2}{m}\sum_{i}(\varepsilon_{\perp,i}^{0}-L_{i})r_{i}
      +\frac{1}{m}\sum_{i}r_{i}^{2}.
\end{align*}
The three $\eta$-dependent sums identify $U_{m}$ together with the diagonal and off-diagonal parts
of the $L_{i}^{2}$ average that define $Q_{m}$ and $R_{m}$. Indeed,
\[
    \frac{1}{m}\sum_{i}\varepsilon_{\perp,i}^{0}L_{i}
    =\frac{1}{m(m-1)}\sum_{i\ne j}\varepsilon_{\perp,i}^{0}\varepsilon_{\perp,j}^{0}
      g_{\eta}(X_{i}^{\mathrm{sel}},X_{j}^{\mathrm{sel}})
    =:U_{m}(\eta),
\]
by definition of $L_{i}$.  Moreover
\[
    \frac{1}{m}\sum_{i}L_{i}^{2}
    =\frac{1}{m(m-1)^{2}}\sum_{i=1}^{m}\sum_{j\ne i}\sum_{k\ne i}
      \varepsilon_{\perp,j}^{0}\varepsilon_{\perp,k}^{0}\,
      g_{\eta}(X_{i}^{\mathrm{sel}},X_{j}^{\mathrm{sel}})
      g_{\eta}(X_{i}^{\mathrm{sel}},X_{k}^{\mathrm{sel}}).
\]
Split the double sum over $j,k\ne i$ into the diagonal $j=k$ and the
off-diagonal $j\ne k$:
\[
    \sum_{j\ne i}\sum_{k\ne i}\!\! =\!\!\sum_{j\ne i}\!(\,j=k\,)
      \;+\!\!\sum_{\substack{j\ne i,k\ne i\\ j\ne k}},
\]
so that
\[
\begin{aligned}
    \frac{1}{m}\sum_{i}L_{i}^{2}
    &=\underbrace{\frac{1}{m(m-1)^{2}}\sum_{i\ne j}
        (\varepsilon_{\perp,j}^{0})^{2}g_{\eta}(X_{i}^{\mathrm{sel}},X_{j}^{\mathrm{sel}})^{2}}_{=:Q_{m}(\eta)}\\
     &\quad+\underbrace{\frac{1}{m(m-1)^{2}}\sum_{i}\sum_{\substack{j\ne i,k\ne i\\ j\ne k}}
        g_{\eta}(X_{i}^{\mathrm{sel}},X_{j}^{\mathrm{sel}})
        g_{\eta}(X_{i}^{\mathrm{sel}},X_{k}^{\mathrm{sel}})
        \varepsilon_{\perp,j}^{0}\varepsilon_{\perp,k}^{0}}_{\text{ordered }j\ne k}.
\end{aligned}
\]
In the last ordered sum each unordered pair $\{j,k\}$ with $j<k$ appears
twice ($g_{ij}g_{ik}=g_{ik}g_{ij}$), whence
\begin{align*}
    &\frac{1}{m(m-1)^{2}}\sum_{i}\sum_{\substack{j\ne i,k\ne i\\ j\ne k}}
    g_{\eta}(X_{i}^{\mathrm{sel}},X_{j}^{\mathrm{sel}})
    g_{\eta}(X_{i}^{\mathrm{sel}},X_{k}^{\mathrm{sel}})
    \varepsilon_{\perp,j}^{0}\varepsilon_{\perp,k}^{0} \\
    &\qquad =\frac{2}{m(m-1)^{2}}\sum_{i=1}^{m}\sum_{\substack{j<k\\ j,k\ne i}}
      g_{\eta}(X_{i}^{\mathrm{sel}},X_{j}^{\mathrm{sel}})
      g_{\eta}(X_{i}^{\mathrm{sel}},X_{k}^{\mathrm{sel}})
      \varepsilon_{\perp,j}^{0}\varepsilon_{\perp,k}^{0}
    =:R_{m}(\eta),
\end{align*}
which is exactly the definition of $R_{m}(\eta)$ in the statement.
Collecting,
\[
    \frac{1}{m}\sum_{i}(\varepsilon_{\perp,i}^{0})^{2}
    -\frac{2}{m}\sum_{i}\varepsilon_{\perp,i}^{0}L_{i}
    +\frac{1}{m}\sum_{i}L_{i}^{2}
    =\frac{1}{m}\sum_{i}(\varepsilon_{\perp,i}^{0})^{2}
      -2U_{m}(\eta)+Q_{m}(\eta)+R_{m}(\eta)
\]
as an identity, and therefore
\[
    \mathrm{CV}_{m}(\eta)=\frac{1}{m}\sum_{i}(\varepsilon_{\perp,i}^{0})^{2}
    -2U_{m}(\eta)+Q_{m}(\eta)+R_{m}(\eta)
    -\frac{2}{m}\sum_{i}(\varepsilon_{\perp,i}^{0}-L_{i})r_{i}
    +\frac{1}{m}\sum_{i}r_{i}^{2}.
\]
The last two $r_{i}$-terms are the only remainder.

\emph{CV remainder: $r_{i}^{2}$ and the cross term.}
Moreover, $m^{-1}\sum_{i}r_{i}^{2}=O_{p}(m^{-2})=o_{p}(m^{-1})$ uniformly in
$\eta$ in the null/local regimes (it is $O_{p}(m^{-1})$ under a fixed alternative).
It remains to show the cross term is $o_{p}(m^{-1})$ uniformly in $\eta$
in the null/local regimes.  Write the algebraic identity for the remainder
from the proof of Lemma~\ref{lem:loo_prediction}:
$r_{i}=\langle k_{\gamma}(X_{i},\cdot),\mathsf R_{i}\hat m^{(-i)}\rangle
=(m-1)^{-1}\sum_{j\ne i}\varepsilon_{\perp,j}^{0}h_{ij}$ with
\[
    h_{ij}:=\big\langle k_{\gamma}(X_{i},\cdot),\mathsf R_{i}k_{\gamma}(X_{j},\cdot)\big\rangle,
    \qquad |h_{ij}|\le\bar k\,\lVert \mathsf R_{i}\rVert_{\mathrm{op}}=O_{p}(m^{-1/2})
\]
uniformly in $i,j,\eta$.  Both $(\varepsilon_{\perp,i}^{0}-L_{i})$ and $r_{i}$ are $\mathbf X$-measurable
linear forms in $\boldsymbol\varepsilon_\perp$; write them
$u_{i}^{\top}\boldsymbol\varepsilon_\perp$ and
$v_{i}^{\top}\boldsymbol\varepsilon_\perp$ with the design-measurable
coefficient rows
\[
    u_{i,l}:=\delta_{il}-\frac{1}{m-1}\,g_{\eta}(X_{i},X_{l})\,\mathbf 1\{l\ne i\},
    \qquad
    v_{i,l}:=\frac{1}{m-1}\,h_{il}\,\mathbf 1\{l\ne i\},
\]
where $\delta_{il}$ is the Kronecker delta; equivalently
$u_{i}=e_{i}-(m-1)^{-1}(g_{\eta}(X_{i},X_{l}))_{l\ne i}$ with $e_{i}\in\mathbb R^{m}$
the $i$-th canonical basis vector, to be distinguished from the population
eigenfunctions $e_{i}=\sqrt{\mu_{i}}\phi_{i}$ of
Section~\ref{sec:equivalence}, which do not appear here.  Write the
$L_{i}$-factor as $s_{i}^{\top}\boldsymbol\varepsilon_\perp$ with
$s_{i,l}:=(m-1)^{-1}g_{\eta}(X_{i},X_{l})\mathbf 1\{l\ne i\}$.  Two structural
facts drive the orders below: $v_{i,i}=0$, since $r_{i}$ carries no coefficient
on the excluded observation, and the envelopes
\[
    \sum_{l}|u_{i,l}|\le 1+\bar k\lambda_{L}^{-1}=O(1),\qquad
    \sum_{l}|s_{i,l}|\le\bar k\lambda_{L}^{-1}=O(1),\qquad
    \sum_{l}|v_{i,l}|=O_{p}(m^{-1/2}),
\]
together with $|v_{i,l}|=O_{p}(m^{-3/2})$ and the square sums
$\sum_{l\ne i}g_{il}^{2}\le(m-1)\bar k^{2}\lambda_{L}^{-2}=O(m)$ and
$\sum_{l\ne i}h_{il}^{2}\le(m-1)\bar k^{2}\lVert \mathsf R_{i}\rVert_{\mathrm{op}}^{2}=O_{p}(1)$,
all uniformly in $i,\eta$.

Since $\boldsymbol\varepsilon_{\perp}^{0}=\hat{\boldsymbol\Pi}\boldsymbol\varepsilon_{0}$ with $\hat{\boldsymbol\Pi}$ symmetric and $\mathbf X$-measurable, and the innovations are conditionally independent and centred with $\mathbb E[\boldsymbol\varepsilon_{0}\mid\mathbf X]=\boldsymbol 0$ and $\mathbb E[\boldsymbol\varepsilon_{0}\boldsymbol\varepsilon_{0}^{\top}\mid\mathbf X]=D$, each conditional expectation reduces in closed form to the second-order pairing
\[
    \mathbb E\bigl[(u^{\top}\boldsymbol\varepsilon_{\perp}^{0})(v^{\top}\boldsymbol\varepsilon_{\perp}^{0})\mid\mathbf X\bigr]
    =u^{\top}\hat{\boldsymbol\Pi}\,D\,\hat{\boldsymbol\Pi}\,v
    =\sum_{l}u_{l}v_{l}\,c_{ll}
      +\sum_{l\ne l'}u_{l}v_{l'}\,c_{ll'},
\]
the expansion recorded in Section~\ref{sec:uniform_limit} giving $c_{ll}:=(\hat{\boldsymbol\Pi}D\hat{\boldsymbol\Pi})_{ll}=\sigma^{2}(X_{l})+O_{p}(m^{-1})$ on the diagonal and $c_{ll'}=O_{p}(m^{-1})$ for $l\ne l'$: variance pairings $\sigma^{2}(X_{l})$ on matched indices and covariance pairings $c_{ll'}=O_{p}(m^{-1})$ otherwise. Under a local alternative at the total-sample rate (Section~\ref{sec:cv_selection}), $\boldsymbol\varepsilon_{0}=R/\sqrt{m+n}+\boldsymbol v$, the moment matrix gains $RR^{\top}/(m+n)$, and summing over $i$ with the prefactor gives
\[
    \frac{2}{m(m+n)}(\hat{\boldsymbol\Pi}R)^{\top}\Bigl(\sum_{i}u_{i}v_{i}^{\top}\Bigr)\hat{\boldsymbol\Pi}R,
\]
with $\bigl\lVert\sum_{i}u_{i}v_{i}^{\top}\bigr\rVert_{\mathrm{op}}\le C\,m^{-1/2}$: the $e_{i}$-anchored part $(m-1)^{-1}H$, $H:=(h_{ij})_{i,j}$, has row sums $O_{p}(m^{1/2})$ by $|h_{ij}|=O_{p}(m^{-1/2})$, and the remainder $(m-1)^{-2}\sum_{i}(g_{\eta}(X_{i},X_{l}))_{l\ne i}(h_{il'})_{l'\ne i}^{\top}$ has Hilbert--Schmidt norm at most $(m-1)^{-2}\lVert G\rVert_{F}\lVert H\rVert_{F}=O_{p}(m^{-1/2})$, from $|g_{\eta}|\le\bar k\lambda_{L}^{-1}$ (giving $\lVert G\rVert_{F}=O(m)$) and $\lVert H\rVert_{F}^{2}=O_{p}(m)$; while $\lVert\hat{\boldsymbol\Pi}R\rVert_{2}=O_{p}(m^{1/2})$.  The contribution is therefore $O_{p}\bigl(m^{-1/2}/(m+n)\bigr)\le O_{p}(m^{-3/2})$, below the pairing orders; the stated orders are unchanged.

\emph{The conditional mean.}  Instantiating the pairing identity at the row pairs $(u_{i},v_{i})$ and $(s_{i},v_{i})$:

\emph{The $\varepsilon_{\perp,i}^{0}r_{i}$ pairing.}  The diagonal sum vanishes exactly, $u_{i,i}v_{i,i}c_{ii}=v_{i,i}c_{ii}=0$ by $v_{i,i}=0$, so only the off-diagonal pairings survive:
\[
\begin{aligned}
    \mathbb E[\varepsilon_{\perp,i}^{0}r_{i}\mid\mathbf X]
    &=\sum_{j\ne i}v_{i,j}c_{ij}
    =\frac{1}{m-1}\sum_{j\ne i}h_{ij}c_{ij},\\
    \bigl|\mathbb E[\varepsilon_{\perp,i}^{0}r_{i}\mid\mathbf X]\bigr|
    &\le\sup_{j\ne i}|h_{ij}|\,\sup_{j\ne i}|c_{ij}|=O_{p}(m^{-3/2}),
\end{aligned}
\]
uniformly in $i$: the pairing is a pure off-diagonal covariance, nonzero only through the projection-induced correlations $c_{ij}=O_{p}(m^{-1})$.

\emph{The $L_{i}r_{i}$ pairing.}  Both parts of the identity contribute:
\[
    \mathbb E[L_{i}r_{i}\mid\mathbf X]
    =\frac{1}{(m-1)^{2}}\sum_{l\ne i}g_{il}h_{il}c_{ll}
    +\frac{1}{(m-1)^{2}}\sum_{\substack{l\ne l'\\ l,l'\ne i}}g_{il}h_{il'}c_{ll'}.
\]
The diagonal part is bounded by Cauchy--Schwarz, and the off-diagonal part by the product of the row sums:
\[
    \Bigl|\frac{1}{(m-1)^{2}}\sum_{l\ne i}g_{il}h_{il}c_{ll}\Bigr|
    \le\frac{\sigma_{\max}^{2}}{(m-1)^{2}}
      \Bigl(\sum_{l\ne i}g_{il}^{2}\Bigr)^{1/2}\Bigl(\sum_{l\ne i}h_{il}^{2}\Bigr)^{1/2}
    =O_{p}(m^{-3/2}),
\]
(the $O_{p}(m^{-1})$ correction of $c_{ll}$ contributing the smaller order $O_{p}(m^{-5/2})$), and
\[
    \Bigl|\frac{1}{(m-1)^{2}}\sum_{\substack{l\ne l'\\ l,l'\ne i}}g_{il}h_{il'}c_{ll'}\Bigr|
    \le\frac{1}{(m-1)^{2}}\,\sup_{l\ne l'}|c_{ll'}|\,
      \Bigl(\sum_{l}|g_{il}|\Bigr)\Bigl(\sum_{l'}|h_{il'}|\Bigr)
    =O_{p}(m^{-3/2}),
\]
using $\sum_{l}|g_{il}|\le(m-1)\bar k\lambda_{L}^{-1}=O(m)$ and $\sum_{l'}|h_{il'}|=O_{p}(m^{1/2})$.  All per-$i$ bounds are uniform in $i$, so summing over $i$ contributes the factor $m$, and the prefactor $2/m$ cancels it:
\[
    \Bigl|\frac{2}{m}\sum_{i}\mathbb{E}[\varepsilon_{\perp,i}^{0}r_{i}\mid\mathbf X]\Bigr|
    +\Bigl|\frac{2}{m}\sum_{i}\mathbb{E}[L_{i}r_{i}\mid\mathbf X]\Bigr|
    =O_{p}(m^{-3/2}).
\]

\emph{The fluctuation.}  Conditional on $\mathbf X$, write $w:=\boldsymbol\varepsilon_{0}$ for the innovations and $\hat u_{i}:=\hat{\boldsymbol\Pi}u_{i}$, $\hat v_{i}:=\hat{\boldsymbol\Pi}v_{i}$ for the projected coefficient rows, so that $\varepsilon_{\perp,i}^{0}-L_{i}=\hat u_{i}^{\top}w$ and $r_{i}=\hat v_{i}^{\top}w$; the cross term is the quadratic form
\[
    \frac{2}{m}\sum_{i}(\varepsilon_{\perp,i}^{0}-L_{i})r_{i}=w^{\top}B\,w,
    \qquad
    B:=\frac{1}{m}\sum_{i=1}^{m}\bigl(\hat u_{i}\hat v_{i}^{\top}
      +\hat v_{i}\hat u_{i}^{\top}\bigr).
\]
The key structural fact is the entrywise bound, uniform in $j,j'$,
\[
    B_{jj'}=O_{p}(m^{-5/2}),
\]
obtained by tracking the anchor of the $u$-sum, $u_{j,j}=1$ being its only $O(1)$ entry:
\[
    \sum_{i}u_{i,j}v_{i,j'}
    =v_{j,j'}-\frac{1}{(m-1)^{2}}\sum_{\substack{i\ne j\\ i\ne j'}}g_{ij}h_{ij'}
    =O_{p}(m^{-3/2})+O_{p}(m^{-3/2}),
\]
with $v_{j,j'}=(m-1)^{-1}h_{jj'}\,\mathbf 1\{j'\ne j\}$ and $\bigl|\sum_{i\ne j,\,i\ne j'}g_{ij}h_{ij'}\bigr|\le\bigl(\sum_{i}|g_{ij}|\bigr)\sup_{i}|h_{ij'}|=O_{p}(m^{1/2})$, the sign being immaterial to the magnitude; the symmetric partner is identical.  The projection corrections enter at the same order: the entries of $\hat{\mathbb G}M\hat{\mathbb G}^{\top}$, with $\hat{\mathbb G}$ the sieve-estimated score design matrix and $M:=(\hat{\mathbb G}^{\top}\hat{\mathbb G})^{-1}$ so that $\hat{\boldsymbol\Pi}=\boldsymbol I-\hat{\mathbb G}M\hat{\mathbb G}^{\top}$ (Section~\ref{sec:uniform_limit}), are $O_{p}(m^{-1})$ uniformly, with row sums $\sum_{l}|(\hat{\mathbb G}M\hat{\mathbb G}^{\top})_{jl}|=O_{p}(1)$ (the off-diagonal order recorded in Section~\ref{sec:uniform_limit} together with $\sum_{a}\hat{\boldsymbol\Pi}_{ja}^{2}=\hat{\boldsymbol\Pi}_{jj}\le1$), so $|\hat u_{i,j}-u_{i,j}|=O_{p}(m^{-1})$ contributes
\[
    \frac{1}{m}\sum_{i}|\hat u_{i,j}-u_{i,j}|\,|v_{i,j'}|=O_{p}(m^{-5/2}),
\]
while the $v$-corrections aggregate through the anchored inner sums just displayed, $\sum_{i}u_{i,j}v_{i,l}=O_{p}(m^{-3/2})$ uniformly in $(j,l)$:
\[
    \Bigl|\frac{1}{m}\sum_{i}u_{i,j}(\hat v_{i,j'}-v_{i,j'})\Bigr|
    =\frac{1}{m}\Bigl|\sum_{l}(\hat{\mathbb G}M\hat{\mathbb G}^{\top})_{j'l}\,
      \sum_{i}u_{i,j}v_{i,l}\Bigr|
    =O_{p}(m^{-5/2}).
\]
Consequently $\lVert B\rVert_{F}^{2}=O_{p}(m^{-3})$ with $\mathbb E\lVert B\rVert_{F}^{2}=O(m^{-3})$, the constants entering through $\mathbb E\lVert \mathsf R_{i}\rVert_{\mathrm{op}}^{2}\le\lambda_{L}^{-4}\,\mathbb E\lVert E_{i}\rVert_{\mathrm{HS}}^{2}=O(m^{-1})$ as in the proof of Lemma~\ref{lem:loo_prediction}, and \eqref{eq:quadvar_bound} of Lemma~\ref{lem:pairvar}, applied to the conditionally independent centered innovations $w$ (for which $d_{j}=\sigma^{2}(X_{j})\le\sigma_{\max}^{2}$ and $\sup_{j}\mathbb E[w_{j}^{2p}\mid X_{j}]\le\bar M_{2p}$ by Assumption~\ref{ass:G}(iii),(v)), gives
\[
    \operatorname{Var}\Bigl(\frac{2}{m}\sum_{i}(\varepsilon_{\perp,i}^{0}-L_{i})r_{i}
      \mid\mathbf X\Bigr)
    \le\bigl(2\sigma_{\max}^{4}+\bar M_{2p}^{2/p}\bigr)\lVert B\rVert_{F}^{2}
    =O_{p}(m^{-3}).
\]

This is the degenerate second-order structure: the first-order Hoeffding projection of the kernel $(\varepsilon_{\perp,i}^{0}-L_{i})r_{i}$ vanishes because both factors are conditionally centered linear forms with $r_{i}$ carrying no coefficient on $w_{i}$, so the fluctuation is governed by the second-order pairing term $\lVert B\rVert_{F}^{2}$ rather than by a first-order term of larger order.  Chebyshev's inequality with the unconditional moment bound $\mathbb E[\operatorname{Var}(\,\cdot\mid\mathbf X)]=O(m^{-3})$ gives fluctuation $O_{p}(m^{-3/2})$; the same holds for the centered part under a local alternative, while its non-centered mean is already covered above.
Thus the cross term is $O_{p}(m^{-3/2})=o_{p}(m^{-1})$ pointwise in $\eta$.
Both the conditional-mean bounds and the
$\operatorname{Var}(\,\cdot\mid\mathbf X)=O_{p}(m^{-3})$ bound are Lipschitz in
$\eta$ via $\partial_{\gamma}k_{\gamma}$ (the envelope recorded in the proof of
Lemma~\ref{lem:loo_prediction}) and the resolvent Lipschitz
Lemma~\ref{lem:lipschitz_common}(c) for $h_{ij}(\eta)$ and $g_{\eta}$
($\lVert g_{\eta_{1}}-g_{\eta_{2}}\rVert_{\infty}\le L_{g}\lVert\eta_{1}-\eta_{2}\rVert$),
with envelope $2\bar k\lambda_{L}^{-1}+2\bar k\lambda_{L}^{-2}L_{C}\operatorname{diam}(\Upsilon)$,
$L_{C}$ as in Lemma~\ref{lem:lipschitz_common}(b);
hence the standing net argument upgrades the bound to
\[
    \sup_{\eta}O_{p}(m^{-3/2})=o_{p}(m^{-1}).
\]
Consequently
\[
    \mathrm{CV}_{m}(\eta)
    =\frac{1}{m}\sum_{i}(\varepsilon_{\perp,i}^{0})^{2}
      -2U_{m}(\eta)+Q_{m}(\eta)+R_{m}(\eta)+o_{p}(m^{-1})
\]
uniformly in $\eta$, which is (I)--(IV) of the statement with $o_{p}(m^{-1})$ remainder
in the null and local regimes.  Under a fixed alternative the same
algebra holds with $r_{i}=O_{p}(m^{-1/2})$, giving only an
$O_{p}(m^{-1/2})$ remainder.
\end{proof}

\begin{lemma}[Conditional variance of bilinear and quadratic forms]\label{lem:pairvar}
Let $(x_{l},y_{l})_{l\le m}$ be conditionally independent pairs given $\mathbf X$ with $\mathbb E[x_{l}\mid\mathbf X]=\mathbb E[y_{l}\mid\mathbf X]=0$, $\mathbb E[x_{l}^{2}\mid\mathbf X]=s_{l}^{2}$, $\mathbb E[y_{l}^{2}\mid\mathbf X]=t_{l}^{2}$, $\mathbb E[x_{l}y_{l}\mid\mathbf X]=n_{l}$, and let $B$ be design-measurable. Then
\begin{equation}
    \mathbb E\bigl[x^{\top}By\mid\mathbf X\bigr]=\sum_{l=1}^{m}B_{ll}\,n_{l},
    \label{eq:pairmean}
\end{equation}
and
\begin{equation}
    \operatorname{Var}\bigl(x^{\top}By\mid\mathbf X\bigr)
    =\sum_{l\ne k}B_{lk}^{2}s_{l}^{2}t_{k}^{2}
     +\sum_{l\ne k}B_{lk}B_{kl}\,n_{l}n_{k}
     +\sum_{l=1}^{m}B_{ll}^{2}\,\mathbb E\bigl[(x_{l}y_{l}-n_{l})^{2}\mid\mathbf X\bigr],
    \label{eq:pairvar}
\end{equation}
the four-equal family being the variance of $x_{l}y_{l}$ around its conditional mean $n_{l}$, not a difference of factorized moments. The sibling of \eqref{eq:pairvar} for quadratic forms is classical: for conditionally independent centered $\zeta_{k}$ with $\mathbb E[\zeta_{k}^{2}\mid\mathbf X]=d_{k}$ and symmetric $A$,
\begin{equation}
    \operatorname{Var}\bigl(\zeta^{\top}A\zeta\mid\mathbf X\bigr)
    =2\bigl\lVert D_{d}^{1/2}AD_{d}^{1/2}\bigr\rVert_{F}^{2}
     +\sum_{k=1}^{m}A_{kk}^{2}\bigl(\mathbb E[\zeta_{k}^{4}\mid\mathbf X]-3d_{k}^{2}\bigr),
    \label{eq:quadvar}
\end{equation}
$D_{d}:=\operatorname{diag}(d_{1},\dots,d_{m})$. Moreover, if
$\max_{k}d_{k}\le\sigma_{\max}^{2}$ and
$\sup_{k}\mathbb E[\zeta_{k}^{2p}\mid\mathbf X]\le\bar M_{2p}$ for some $p\in(2,4]$,
then
\begin{equation}
    \operatorname{Var}\bigl(\zeta^{\top}A\zeta\mid\mathbf X\bigr)
    \le\bigl(2\sigma_{\max}^{4}+\bar M_{2p}^{2/p}\bigr)\lVert A\rVert_{\mathrm{F}}^{2}
    \label{eq:quadvar_bound}
\end{equation}
for every symmetric $\mathbf X$-measurable $A$.
\end{lemma}

\begin{proof}
Write $\mu_{j}:=\mathbb E[x_{j}^{2}y_{j}^{2}\mid\mathbf X]$. Throughout the proof, conditional independence across observations factorizes every conditional expectation over the multiset of observation indices appearing in the monomial, and an index appearing exactly once contributes the factor $\mathbb E[x_{j}\mid\mathbf X]=0$ or $\mathbb E[y_{j}\mid\mathbf X]=0$, so the monomial vanishes.

For \eqref{eq:pairmean}, expand over the entries of $B$ and apply the factorization with $\mathbb E[x_{l}y_{l}\mid\mathbf X]=n_{l}$:
\[
    \mathbb E\bigl[x^{\top}By\mid\mathbf X\bigr]
    =\sum_{l=1}^{m}\sum_{k=1}^{m}B_{lk}\,\mathbb E\bigl[x_{l}y_{k}\mid\mathbf X\bigr]
    =\sum_{l=1}^{m}B_{ll}\,n_{l},
\]
the off-diagonal terms vanishing because for $l\ne k$ the variables $x_{l}$ and $y_{k}$ are conditionally independent with zero conditional means.

For \eqref{eq:pairvar}, expand the square over four indices,
\[
    \mathbb E\bigl[(x^{\top}By)^{2}\mid\mathbf X\bigr]
    =\sum_{l=1}^{m}\sum_{k=1}^{m}\sum_{p=1}^{m}\sum_{q=1}^{m}
      B_{lk}B_{pq}\,\mathbb E\bigl[x_{l}y_{k}x_{p}y_{q}\mid\mathbf X\bigr],
\]
and classify the quadruples $(l,k,p,q)$ by the multiplicities of the observation indices in the monomial $x_{l}y_{k}x_{p}y_{q}$. A quadruple in which some index appears exactly once vanishes by the factorization rule above. The surviving quadruples either place all four slots at one observation or place two slots at each of two observations $j\ne j'$, and contribute as follows. The all-equal quadruple $(j,j,j,j)$ contributes $B_{jj}^{2}\mu_{j}$. The split quadruples $(j,j',j,j')$ and $(j',j,j',j)$, which place both $x$-slots at one observation and both $y$-slots at the other, contribute $B_{jj'}^{2}s_{j}^{2}t_{j'}^{2}$ and $B_{j'j}^{2}s_{j'}^{2}t_{j}^{2}$, the two orientations of the split double sum. The crossed quadruples $(j,j',j',j)$ and $(j',j,j,j')$ contribute $B_{jj'}B_{j'j}n_{j}n_{j'}$ each, since $x_{j}$ shares its index with $y_{j}$ and $x_{j'}$ with $y_{j'}$. The mixed quadruples $(j,j,j',j')$ and $(j',j',j,j)$, which pair the slots within each of the two factors $x^{\top}By$ separately, contribute $B_{jj}B_{j'j'}n_{j}n_{j'}$ each. Each of the crossed and mixed families contributes twice per unordered pair $\{j,j'\}$ with a summand invariant under $j\leftrightarrow j'$, so its total $2\sum_{j<j'}$ equals the ordered double sum $\sum_{l\ne k}$. Collecting,
\[
    \mathbb E\bigl[(x^{\top}By)^{2}\mid\mathbf X\bigr]
    =\sum_{j}B_{jj}^{2}\mu_{j}
     +\sum_{l\ne k}B_{lk}^{2}s_{l}^{2}t_{k}^{2}
     +\sum_{l\ne k}B_{lk}B_{kl}\,n_{l}n_{k}
     +\sum_{l\ne k}B_{ll}B_{kk}\,n_{l}n_{k}.
\]
Squaring \eqref{eq:pairmean},
\[
    \Bigl(\sum_{l}B_{ll}n_{l}\Bigr)^{2}
    =\sum_{j}B_{jj}^{2}n_{j}^{2}
     +\sum_{l\ne k}B_{ll}B_{kk}\,n_{l}n_{k},
\]
so the mixed family cancels against the square of the mean, and on the diagonal, by $\mathbb E[x_{j}y_{j}\mid\mathbf X]=n_{j}$,
\[
    \mathbb E\bigl[(x_{j}y_{j}-n_{j})^{2}\mid\mathbf X\bigr]
    =\mu_{j}-2n_{j}\,\mathbb E[x_{j}y_{j}\mid\mathbf X]+n_{j}^{2}
    =\mu_{j}-n_{j}^{2},
\]
which turns the all-equal family into $\sum_{j}B_{jj}^{2}\mathbb E[(x_{j}y_{j}-n_{j})^{2}\mid\mathbf X]$ and gives \eqref{eq:pairvar}.

For \eqref{eq:quadvar}, specialize \eqref{eq:pairvar} to $x=y=\boldsymbol\zeta$ and
$B=A$: the pairs are conditionally independent and centered given $\mathbf X$ with
$s_{l}^{2}=t_{l}^{2}=n_{l}=d_{l}$, the second family coincides with the first by
the symmetry of $A$, $A_{lk}A_{kl}=A_{lk}^{2}$, and the four-equal entry is
$\mathbb E[(\zeta_{l}^{2}-d_{l})^{2}\mid\mathbf X]=\mu_{l}-d_{l}^{2}$ with
$\mu_{l}=\mathbb E[\zeta_{l}^{4}\mid\mathbf X]$ as above, so
\[
    \operatorname{Var}\bigl(\zeta^{\top}A\zeta\mid\mathbf X\bigr)
    =2\sum_{l\ne k}A_{lk}^{2}d_{l}d_{k}
     +\sum_{l}A_{ll}^{2}\bigl(\mathbb E[\zeta_{l}^{4}\mid\mathbf X]-d_{l}^{2}\bigr).
\]
Adding and subtracting $2\sum_{l}A_{ll}^{2}d_{l}^{2}$ and using
$\lVert D_{d}^{1/2}AD_{d}^{1/2}\rVert_{F}^{2}=\sum_{l,k}A_{lk}^{2}d_{l}d_{k}$
gives
\[
    \operatorname{Var}\bigl(\zeta^{\top}A\zeta\mid\mathbf X\bigr)
    =2\bigl\lVert D_{d}^{1/2}AD_{d}^{1/2}\bigr\rVert_{F}^{2}
     +\sum_{l}A_{ll}^{2}\bigl(\mathbb E[\zeta_{l}^{4}\mid\mathbf X]-3d_{l}^{2}\bigr),
\]
which is \eqref{eq:quadvar}. For \eqref{eq:quadvar_bound}, the first family of
\eqref{eq:quadvar} is bounded entrywise,
$\lVert D_{d}^{1/2}AD_{d}^{1/2}\rVert_{F}^{2}
=\sum_{l,k}A_{lk}^{2}d_{l}d_{k}\le\sigma_{\max}^{4}\lVert A\rVert_{F}^{2}$ by
$\max_{k}d_{k}\le\sigma_{\max}^{2}$, and the diagonal family by Lyapunov norm
monotonicity at $q=4\le2p$,
$\mathbb E[\zeta_{k}^{4}\mid\mathbf X]\le\bigl(\mathbb E[\zeta_{k}^{2p}\mid\mathbf X]\bigr)^{2/p}
\le\bar M_{2p}^{2/p}$, so that, dropping the nonpositive $-3d_{k}^{2}$ and using
$\sum_{k}A_{kk}^{2}\le\lVert A\rVert_{F}^{2}$,
\[
    \sum_{k}A_{kk}^{2}\bigl(\mathbb E[\zeta_{k}^{4}\mid\mathbf X]-3d_{k}^{2}\bigr)
    \le\bar M_{2p}^{2/p}\sum_{k}A_{kk}^{2}
    \le\bar M_{2p}^{2/p}\lVert A\rVert_{F}^{2};
\]
summing the two families gives \eqref{eq:quadvar_bound}.
\end{proof}

\begin{lemma}[Replacement cost of the feasible residuals]\label{lem:repl_cost}
Let Assumptions~\ref{ass:A}--\ref{ass:F} hold, with Assumption~\ref{ass:D}(iii) and the rate condition of Assumption~\ref{ass:F}(iii) holding at the selection sample of size $m$, and let Assumptions~\ref{ass:G}(iii) and~\ref{ass:G}(v) hold. Suppose the null premise that the true residual is conditionally centered, $\mathbb E[\boldsymbol\varepsilon_{0}\mid\mathbf X]=\boldsymbol 0$, as holds under $H_{0}$, where $\boldsymbol\varepsilon_{0}$ is the noise $\boldsymbol v$; and the joint moment bound $\mathbb E[\varepsilon_{0}^{2}\,\sup_{\theta\in\Theta}\lVert\nabla_{\theta}^{2}\varepsilon(S;\theta)\rVert^{2}]<\infty$, automatic in the conditional-mean case by Assumptions~\ref{ass:E}(iii) and~\ref{ass:G}(iii). Let $\boldsymbol{\varepsilon}^{0}_{\perp}:=\hat{\boldsymbol\Pi}\boldsymbol\varepsilon_{0}$ and $\hat{\boldsymbol\varepsilon}_{\perp}:=\hat{\boldsymbol\Pi}\hat{\boldsymbol\varepsilon}$ denote the infeasible and feasible orthogonalized residuals on the selection sample. Then, for
$\mathrm{CV}^{\mathrm{dep}}_{m}(\eta;v):=-2U_{m}(\eta;v)+Q_{m}(\eta;v)+R_{m}(\eta;v)$,
the $\eta$-dependent part of the leave-one-out criterion evaluated at a residual vector $v$, namely the statistics $U_{m}$, $Q_{m}$, $R_{m}$ of Lemma~\ref{lem:loo_linearization} with the coefficient vector $\boldsymbol\varepsilon_{\perp}^{0}$ replaced by $v$,
\[
    \sup_{\eta\in\Upsilon}\bigl|m\,\mathrm{CV}^{\mathrm{dep}}_{m}(\eta;\hat{\boldsymbol\varepsilon}_{\perp})-m\,\mathrm{CV}^{\mathrm{dep}}_{m}(\eta;\boldsymbol{\varepsilon}^{0}_{\perp})\bigr|
    =O_{p}\bigl(m^{-\varrho}+m^{1/2-2\varrho}+m^{1-4\varrho}\bigr)=o_{p}(1),
\]
the three rates being the idiosyncratic-linear, smooth-second-order, and quadratic contributions. When the score is $X$-measurable the idiosyncratic increment of the proof is absent and the remainder is $O_{p}(m^{1/2-2\varrho}+m^{-2\varrho}+m^{1-4\varrho})$; rates faster than $o_{p}(1)$ require additional structure that removes the smooth conditional mean of the second-order term, as in the location model where $\boldsymbol r^{(2)}\equiv\boldsymbol 0$.
\end{lemma}

\begin{proof}
All bounds are uniform in $\eta\in\Upsilon$: only the envelopes $\lVert g_{\eta}\rVert_{\infty}\le\lambda_{L}^{-1}$ and $\lVert g_{\eta}\circ g_{\eta}\rVert_{\infty}\le\lambda_{L}^{-2}$ enter, obtained from $g_{\eta}(x,x)\le\lambda^{-1}k_{\gamma}(x,x)=\lambda^{-1}$ and the Cauchy--Schwarz inequality $|g_{\eta}(x,y)|\le\sqrt{g_{\eta}(x,x)g_{\eta}(y,y)}$, together with $|(g_{\eta}\circ g_{\eta})(x,y)|\le\mathbb E_{Z}[g_{\eta}(x,Z)g_{\eta}(Z,y)]$ and $|r_{\eta}|\le2\lambda_{L}^{-2}$; consequently every matrix $\Psi$ with $\lVert\Psi\rVert_{\max}\le\bar\psi$ obeys $\lVert\Psi v\rVert\le\lVert\Psi\rVert_{F}\lVert v\rVert\le m\bar\psi\,\lVert v\rVert$ together with the entrywise bound $\lVert\Psi v\rVert_{\infty}\le\bar\psi\lVert v\rVert_{1}$. Conditional moments given $\mathbf X$ treat $\hat\theta$, the fitted sieve, and $\hat{\boldsymbol\Pi}$ as conditionally deterministic, under the convention of Remark~\ref{rem:cond_det}, valid here since $\hat\theta$ is a smooth functional of sample averages; the $\boldsymbol\varepsilon_{0}$-dependence of these objects contributes only the higher-order corrections already recorded in Section~\ref{sec:uniform_limit} and in the proof of Lemma~\ref{lem:remainder_cv}.

\medskip\noindent\emph{Structural facts.}\;
Let $T:=\boldsymbol I-P_{Z}$, idempotent and symmetric with $|T_{jk}|\le\mathbf 1_{j=k}+Cm^{-1}$ (bounded sieve basis and $\lambda_{\min}(Z^{\top}Z/m)\ge c$, Assumption~\ref{ass:F}(ii)), and recall the recorded off-diagonal order $\max_{a\ne b}|\hat{\boldsymbol\Pi}_{ab}|=O_{p}(m^{-1})$ of Section~\ref{sec:uniform_limit}. Three facts are used throughout. First, the row sums
\[
    \max_{l}\sum_{i}\bigl|\hat{\boldsymbol\Pi}_{il}\bigr|=O_{p}(1),\qquad
    \max_{j}\sum_{i}|T_{ij}|=O(1),
\]
by idempotence ($\sum_{i}\hat{\boldsymbol\Pi}_{il}^{2}=\hat{\boldsymbol\Pi}_{ll}\le1$, $\sum_{i}T_{ij}^{2}=T_{jj}\le1$) together with the off-diagonal orders across $m$ terms. Second, diagonal dominance: any matrix $M$ with $M_{jj}=O_{p}(1)$ and $M_{jk}=O_{p}(m^{-1})$ for $j\ne k$ obeys $\bigl|\sum_{i,j}\psi_{ij}M_{ij}\bigr|=O_{p}(\bar\psi m)$ for every $\Psi$ with $\lVert\Psi\rVert_{\max}\le\bar\psi$, by $\bigl|\sum_{i,j}\psi_{ij}M_{ij}\bigr|\le\bar\psi\sum_{j}|M_{jj}|+\bar\psi\sum_{i\ne j}|M_{ij}|$. Third, with $V_{j}:=\mathbb E[\boldsymbol\eta_{0,j}\boldsymbol\eta_{0,j}^{\top}\mid X_{j}]$ and $\boldsymbol v_{j}:=\mathbb E[\varepsilon_{0,j}\boldsymbol\eta_{0,j}\mid X_{j}]$, the envelopes $\mathbb E\lVert V_{j}\rVert<\infty$, $\mathbb E\lVert V_{j}\rVert^{2}<\infty$ and $\mathbb E\lVert\boldsymbol v_{j}\rVert<\infty$ follow from Assumptions~\ref{ass:E}(iii) and~\ref{ass:G}(iii) by Jensen and Cauchy--Schwarz against the fourth-moment envelope of the raw score, so by Markov's inequality
\[
    \tfrac{1}{m}\sum_{l}\lVert V_{l}\rVert=O_{p}(1),\qquad
    \tfrac{1}{m}\sum_{l}\lVert V_{l}\rVert^{2}=O_{p}(1),\qquad
    \tfrac{1}{m}\sum_{l}\lVert\boldsymbol v_{l}\rVert=O_{p}(1).
\]
Consequently, for $B:=\hat{\boldsymbol\Pi}\Psi T$,
\[
\begin{aligned}
    |B_{ll}|&\le\bar\psi\Bigl(\sum_{i}|\hat{\boldsymbol\Pi}_{li}|\Bigr)\Bigl(\sum_{j}|T_{jl}|\Bigr)\le\bar\psi\,O_{p}(1),
    \qquad
    \sum_{l}B_{lk}^{2}=\lVert\hat{\boldsymbol\Pi}\Psi Te_{k}\rVert^{2}\le\lVert\Psi Te_{k}\rVert^{2}\\
    &\le m\bar\psi^{2}\,\Bigl(\max_{j}\sum_{i}|T_{ij}|\Bigr)^{2}=O(m\bar\psi^{2})
\end{aligned}
\]
uniformly in $k$, the last bound entrywise; this column-sum refinement replaces the Frobenius bound $\lVert B\rVert_{F}^{2}\le m^{2}\bar\psi^{2}$, which would inflate every order below by a factor $m$. Fourth, the conditional second moments of the infeasible orthogonalized residuals: under the null premise $\boldsymbol\varepsilon_{0}=\boldsymbol v$ and the convention of Remark~\ref{rem:cond_det},
\[
    c_{ij}:=\mathbb E\bigl[\varepsilon^{0}_{\perp,i}\varepsilon^{0}_{\perp,j}\mid\mathbf X\bigr]
    =(\hat{\boldsymbol\Pi}D\hat{\boldsymbol\Pi})_{ij},
    \qquad D:=\operatorname{diag}\bigl(\sigma^{2}(X_{1}),\dots,\sigma^{2}(X_{m})\bigr),
\]
and $D\le\sigma_{\max}^{2}\boldsymbol I$ in Loewner order (Assumption~\ref{ass:G}(iii))
implies by congruence and idempotence
$\hat{\boldsymbol\Pi}D\hat{\boldsymbol\Pi}\le\sigma_{\max}^{2}\hat{\boldsymbol\Pi}$, whence the diagonal
\[
    c_{jj}\le\sigma_{\max}^{2}\hat{\boldsymbol\Pi}_{jj}\le\sigma_{\max}^{2},
\]
while for $i\ne j$, using $|\hat{\boldsymbol\Pi}_{ii}|\le1$ (idempotence) and
$|\hat{\boldsymbol\Pi}_{ab}|\le\max_{a\ne b}|\hat{\boldsymbol\Pi}_{ab}|$ off the diagonal,
\[
    |c_{ij}|
    \le\sigma_{\max}^{2}\Bigl(|\hat{\boldsymbol\Pi}_{ii}|\,|\hat{\boldsymbol\Pi}_{ij}|
      +\max_{a\ne b}|\hat{\boldsymbol\Pi}_{ab}|\sum_{a\ne i}|\hat{\boldsymbol\Pi}_{aj}|\Bigr)
    \le\sigma_{\max}^{2}\max_{a\ne b}|\hat{\boldsymbol\Pi}_{ab}|
      \Bigl(1+\sum_{a}|\hat{\boldsymbol\Pi}_{aj}|\Bigr)
    =O_{p}(m^{-1}),
\]
every product carrying an off-diagonal entry (the factor $\hat{\boldsymbol\Pi}_{ij}$ at
$a=i$, the factor $\hat{\boldsymbol\Pi}_{ia}$ otherwise), the row sums being $O_{p}(1)$ by
the first fact.

\medskip\noindent\emph{The increment and its split.}\;
Write $\boldsymbol\delta:=\hat\theta-\theta_{0}$ and expand the feasible residuals at the selection sample exactly as in the proof of Lemma~\ref{lemma:estimation_remainder}, $\hat{\boldsymbol\varepsilon}=\boldsymbol\varepsilon_{0}+\mathbb G\boldsymbol\delta+\boldsymbol r^{(2)}$ with $\mathbb G$ the raw-score matrix (rows $\nabla_{\theta}\varepsilon(s_{j};\hat\theta)^{\top}$) and $\boldsymbol r^{(2)}=O_{p}(m^{-2\varrho})$ element-wise. The identities $\hat{\boldsymbol\Pi}\hat{\mathbb G}=\boldsymbol 0$ and $\hat{\boldsymbol\Pi}(\boldsymbol I-P_{Z})=(\boldsymbol I-P_{Z})$ of that proof give, for the increment $d:=\hat{\boldsymbol\varepsilon}_{\perp}-\boldsymbol{\varepsilon}^{0}_{\perp}=\hat{\boldsymbol\Pi}(\mathbb G\boldsymbol\delta+\boldsymbol r^{(2)})$,
\[
    d=(\boldsymbol I-P_{Z})\,\mathbb G\,\boldsymbol\delta+\hat{\boldsymbol\Pi}\boldsymbol r^{(2)}
      =a+b+c,
\]
with
\[
    a:=(\boldsymbol I-P_{Z})\boldsymbol\eta_{0}\boldsymbol\delta,\qquad
    b:=(\boldsymbol I-P_{Z})\boldsymbol g\boldsymbol\delta,\qquad
    c:=(\boldsymbol I-P_{Z})\tilde{\boldsymbol H}\boldsymbol\delta\boldsymbol\delta+\hat{\boldsymbol\Pi}\boldsymbol r^{(2)},
\]
where $\mathbb G=\boldsymbol g+\boldsymbol\eta_{0}+\tilde{\boldsymbol H}\boldsymbol\delta$ is the decomposition into conditional score, conditionally centered idiosyncratic score, and path-averaged Hessian rows used in that proof, and $(\boldsymbol I-P_{Z})\hat{\mathbb G}=\boldsymbol 0$. Write $c=u+\tilde c$ with $u:=\mathbb E[c\mid\mathbf X]$.

\medskip\noindent\emph{Orders and conditional moments.}\;
The observations are i.i.d.\ (Assumption~\ref{ass:C}(i)). With $d_{k}:=\boldsymbol\delta^{\top}V_{k}\boldsymbol\delta$,
\[
    \mathbb E\bigl[aa^{\top}\mid\mathbf X\bigr]
    =T\operatorname{diag}\bigl(d_{1},\dots,d_{m}\bigr)T,
\]
whose diagonal entries are $O_{p}(\lVert\boldsymbol\delta\rVert^{2})$ and whose off-diagonal entries are $O_{p}(\lVert\boldsymbol\delta\rVert^{2}m^{-1})$ uniformly: the diagonal is $T_{jj}d_{j}+\sum_{k\ne j}T_{jk}^{2}d_{k}$ with $\sum_{k\ne j}T_{jk}^{2}d_{k}\le(C/m)^{2}\sum_{k}d_{k}=O_{p}(m^{-1}\lVert\boldsymbol\delta\rVert^{2})$ by the off-diagonal bound $|T_{jk}|\le Cm^{-1}$ of Assumption~\ref{ass:F}(ii) and the first Markov average, and the off-diagonal entry $\sum_{k}T_{jk}T_{j'k}d_{k}$ splits into the two fixed-index terms $T_{jj'}d_{j}+T_{j'j}d_{j'}=O_{p}(m^{-1}\lVert\boldsymbol\delta\rVert^{2})$ and an $O(m^{-2})\sum_{k}d_{k}$ remainder, against $\frac{1}{m}\sum_{k}d_{k}=O_{p}(\lVert\boldsymbol\delta\rVert^{2})$ (the first Markov average). Moreover
\[
    \lVert a\rVert^{2}
    =\boldsymbol\delta^{\top}\boldsymbol\eta_{0}^{\top}T\boldsymbol\eta_{0}\boldsymbol\delta
    \le\sum_{j=1}^{m}\bigl(\boldsymbol\eta_{0,j}^{\top}\boldsymbol\delta\bigr)^{2}
    =O_{p}\bigl(m\lVert\boldsymbol\delta\rVert^{2}\bigr),
\]
the conditional expectation of the middle sum being $\sum_{j}d_{j}=O_{p}(m\lVert\boldsymbol\delta\rVert^{2})$, so $\lVert a\rVert=O_{p}(m^{1/2-\varrho})$; and $a_{j}=\sum_{k}T_{jk}\boldsymbol\eta_{0,k}^{\top}\boldsymbol\delta$ is a linear form in conditionally independent centered variables with $\sum_{k}T_{jk}^{2}=T_{jj}\le1$. For $b$: $T\boldsymbol g=T(\boldsymbol g-\hat{\boldsymbol g})$ since $P_{Z}\hat{\boldsymbol g}=\hat{\boldsymbol g}$, so $\lVert b\rVert\le\lVert\boldsymbol g-\hat{\boldsymbol g}\rVert_{F}\lVert\boldsymbol\delta\rVert=o_{p}(m^{1/2-3\varrho})$ by the sieve rate of Assumption~\ref{ass:F}(iii) at the selection sample (the sieve-error step of the proof of Lemma~\ref{lemma:estimation_remainder}). For the second-order component, $\mathbb E[c_{j}^{2}\mid\mathbf X]=O_{p}(m^{-4\varrho})$ uniformly in $j$ (Assumptions~\ref{ass:E}(iii) and~\ref{ass:D}(iii), Cauchy--Schwarz, and the $O(1)$ row sums of $T$ and $\hat{\boldsymbol\Pi}$), whence, by Jensen,
\[
    \lVert u\rVert^{2}=\sum_{j=1}^{m}u_{j}^{2}\le\sum_{j=1}^{m}\mathbb E\bigl[c_{j}^{2}\mid\mathbf X\bigr]=O_{p}\bigl(m^{1-4\varrho}\bigr),
\]
and
\[
    \mathbb E\bigl[\tilde c_{j}^{2}\mid\mathbf X\bigr]=O_{p}(m^{-4\varrho}),\qquad
    \mathbb E\bigl[\tilde c_{j}\tilde c_{j'}\mid\mathbf X\bigr]=O_{p}(m^{-1-4\varrho})\quad(j\ne j'),
\]
the off-diagonal order following from the same-index survival of the conditionally independent per-observation factors $\boldsymbol\delta^{\top}\tilde{\boldsymbol H}_{k}\boldsymbol\delta$ and $r_{k}^{(2)}$ together with the off-diagonal entries of $T$ and $\hat{\boldsymbol\Pi}$, exactly as in the enumeration of Lemma~\ref{lem:pairvar}.

\medskip\noindent\emph{The three differences.}\;
With $K^{\mathrm{od}}_{\eta}:=(g_{\eta}(X_{i},X_{j})\mathbf 1_{i\ne j})$, $g_{ij}:=g_{\eta}(X_{i},X_{j})$, and $s_{j}:=\sum_{i\ne j}g_{ij}^{2}\le(m-1)\lambda_{L}^{-2}$,
\[
    \Delta^{U}_{m}(\eta):=mU_{m}(\eta;\hat{\boldsymbol\varepsilon}_{\perp})-mU_{m}(\eta;\boldsymbol{\varepsilon}^{0}_{\perp})
    =\frac{2}{m-1}\,w^{\top}K^{\mathrm{od}}_{\eta}d+\frac{1}{m-1}\,d^{\top}K^{\mathrm{od}}_{\eta}d,
\]
\[
    \Delta^{Q}_{m}(\eta)
    =\frac{2}{(m-1)^{2}}\sum_{j=1}^{m}w_{j}d_{j}s_{j}
     +\frac{1}{(m-1)^{2}}\sum_{i\ne j}d_{j}^{2}g_{ij}^{2},
\]
and, since $v_{j}v_{k}\mapsto(w_{j}+d_{j})(w_{k}+d_{k})-w_{j}w_{k}=w_{j}d_{k}+w_{k}d_{j}+d_{j}d_{k}$ contributes both orientations of its linear part under the once-counted symmetric pair sum,
\[
    \Delta^{R}_{m}(\eta)
    =\frac{2}{(m-1)^{2}}\sum_{i=1}^{m}\sum_{\substack{j<k\\ j,k\ne i}}g_{ij}g_{ik}\bigl(w_{j}d_{k}+w_{k}d_{j}+d_{j}d_{k}\bigr).
\]
Splitting $g_{ij}g_{ik}=(g_{\eta}\circ g_{\eta})(X_{j},X_{k})+r_{\eta}(X_{i};X_{j},X_{k})$ and summing the population piece over the $m-2$ anchors, the linear part of $\Delta^{R}_{m}$ becomes
\[
    \Delta^{R}_{m,\mathrm{lin}}(\eta)
    =\frac{2(m-2)}{(m-1)^{2}}\,w^{\top}K^{\circ,\mathrm{od}}_{\eta}d
     +\frac{2}{(m-1)^{2}}\,w^{\top}\Psi^{r}_{\eta}d,
\]
where $K^{\circ,\mathrm{od}}_{\eta}:=((g_{\eta}\circ g_{\eta})(X_{j},X_{k})\mathbf 1_{j\ne k})$ and $(\Psi^{r}_{\eta})_{jk}:=\mathbf 1_{j\ne k}\sum_{i\ne j,k}r_{\eta}(X_{i};X_{j},X_{k})$ is a design-measurable symmetric matrix with $\lVert\Psi^{r}_{\eta}\rVert_{\max}\le2m\lambda_{L}^{-2}$.

\medskip\noindent\emph{Linear forms against the increment.}\;
Let $\Psi$ be any design-measurable symmetric matrix with $\lVert\Psi\rVert_{\max}\le\bar\psi$ and put $B:=\hat{\boldsymbol\Pi}\Psi T$.

\emph{(i) The idiosyncratic channel.} By \eqref{eq:pairmean} with $(x,y)=(\boldsymbol\varepsilon_{0},\boldsymbol\eta_{0}\boldsymbol\delta)$, the pairs being conditionally independent and centered given $\mathbf X$ with $n_{l}=\boldsymbol v_{l}^{\top}\boldsymbol\delta$,
\[
    \mathbb E\bigl[w^{\top}\Psi a\mid\mathbf X\bigr]
    =\sum_{l=1}^{m}B_{ll}\,\boldsymbol v_{l}^{\top}\boldsymbol\delta
    =O_{p}\bigl(\bar\psi\,m\lVert\boldsymbol\delta\rVert\bigr),
\]
the order by the diagonal bound on $B$ and the Markov average of the $\boldsymbol v_{l}$. No centering improves the rate: the criterion difference contains $w^{\top}\Psi a$ itself, whose conditional mean is of the same order as the fluctuation below. The variance is \eqref{eq:pairvar} with $s_{l}^{2}=\sigma^{2}(X_{l})$ (Assumption~\ref{ass:G}(iii)), $t_{k}^{2}=d_{k}$, and $n_{l}=\boldsymbol v_{l}^{\top}\boldsymbol\delta$. The product family is
$\sum_{l\ne k}B_{lk}^{2}\sigma^{2}(X_{l})d_{k}
\le\sigma_{\max}^{2}\bigl(\sum_{k}d_{k}\bigr)\max_{k}\sum_{l}B_{lk}^{2}
=O_{p}\bigl(\bar\psi^{2}m^{2}\lVert\boldsymbol\delta\rVert^{2}\bigr)$
by the column-sum refinement; the crossed family is
$\sum_{l\ne k}B_{lk}B_{kl}\boldsymbol v_{l}^{\top}\boldsymbol\delta\,\boldsymbol v_{k}^{\top}\boldsymbol\delta
\le\bigl(\max_{l,k}|B_{lk}|\bigr)^{2}\bigl(\sum_{l}|\boldsymbol v_{l}^{\top}\boldsymbol\delta|\bigr)^{2}
=O_{p}\bigl(\bar\psi^{2}m^{2}\lVert\boldsymbol\delta\rVert^{2}\bigr)$,
the entries of $B$ being $\bar\psi\,O_{p}(1)$ uniformly by the row sums; and the four-equal family is
$\sum_{l}B_{ll}^{2}\,\mathbb E\bigl[\bigl(\varepsilon_{0,l}\boldsymbol\eta_{0,l}^{\top}\boldsymbol\delta-\boldsymbol v_{l}^{\top}\boldsymbol\delta\bigr)^{2}\mid\mathbf X\bigr]
\le\bar\psi^{2}O_{p}(1)\sum_{l}\bigl(\mathbb E[\varepsilon_{0,l}^{4}\mid\mathbf X]\bigr)^{1/2}\bigl(\mathbb E[(\boldsymbol\eta_{0,l}^{\top}\boldsymbol\delta)^{4}\mid\mathbf X]\bigr)^{1/2}
=O_{p}\bigl(\bar\psi^{2}m\lVert\boldsymbol\delta\rVert^{2}\bigr)$
by Cauchy--Schwarz with $\mathbb E[\varepsilon_{0,l}^{4}\mid\mathbf X]\le\bar M_{2p}^{2/p}$ (the Lyapunov interpolation of the $(2p)$-clause of Assumption~\ref{ass:G}(v) at $q=4$, the null premise giving $\boldsymbol\varepsilon_{0}=\boldsymbol v$) and $\mathbb E[(\boldsymbol\eta_{0,l}^{\top}\boldsymbol\delta)^{4}\mid\mathbf X]\le C\lVert\boldsymbol\delta\rVert^{4}$ (Assumption~\ref{ass:E}(iii)). Collecting,
\[
    \operatorname{Var}\bigl(w^{\top}\Psi a\mid\mathbf X\bigr)=O_{p}\bigl(\bar\psi^{2}m^{2}\lVert\boldsymbol\delta\rVert^{2}\bigr),
\]
whence
\[
    \bigl|w^{\top}\Psi a\bigr|=O_{p}\bigl(\bar\psi\,m^{1-\varrho}\bigr).
\]

\emph{(ii) The smooth and second-order deterministic channels.} For design-measurable $e$, $\mathbb E[w\mid\mathbf X]=\hat{\boldsymbol\Pi}\,\mathbb E[\boldsymbol\varepsilon_{0}\mid\mathbf X]=\boldsymbol 0$ by the null premise, so $\mathbb E[w^{\top}\Psi e\mid\mathbf X]=0$ and
\[
    \operatorname{Var}\bigl(w^{\top}\Psi e\mid\mathbf X\bigr)
    =e^{\top}\Psi^{\top}\hat{\boldsymbol\Pi}D\hat{\boldsymbol\Pi}\Psi e
    \le\sigma_{\max}^{2}\lVert\Psi e\rVert^{2}
    \le\sigma_{\max}^{2}m^{2}\bar\psi^{2}\lVert e\rVert^{2},
\]
with $D:=\operatorname{diag}(\sigma^{2}(X_{1}),\dots,\sigma^{2}(X_{m}))$, whence $|w^{\top}\Psi b|=o_{p}(\bar\psi m^{3/2-3\varrho})$ and $|w^{\top}\Psi u|=O_{p}(\bar\psi m^{3/2-2\varrho})$, the bound being uniform over the smooth and conditionally centered components of $e$, so the smooth component of $u$ enters at full strength.

\emph{(iii) The centered second-order channel.} Write the centered second-order increment as a sum of conditionally centered per-observation factors, $\tilde c_{j}=\sum_{l}\tilde F^{(j)}_{l}$, and let $\gamma^{(j)}_{l}:=\mathbb E[\varepsilon_{0,l}\tilde F^{(j)}_{l}\mid X_{l}]$. Conditional independence across observations gives the exact mean
\[
    \mathbb E\bigl[w_{i}\tilde c_{j}\mid\mathbf X\bigr]
    =\sum_{l}\hat{\boldsymbol\Pi}_{il}\,\gamma^{(j)}_{l}
    =\hat{\boldsymbol\Pi}_{ij}\,\nu_{j}+\rho_{ij},
\]
with $\nu_{j}:=\gamma^{(j)}_{j}=\mathbb E[\varepsilon_{0,j}\tilde c_{j}\mid X_{j}]$ and $\rho_{ij}:=\sum_{l\ne j}\hat{\boldsymbol\Pi}_{il}\gamma^{(j)}_{l}$: only the same-index factor survives in $\nu_{j}$, and $|\nu_{j}|=O_{p}(m^{-2\varrho})$ by Cauchy--Schwarz and the element-wise orders under the stated moment bound. For $l\ne j$ the factor $\tilde F^{(j)}_{l}$ carries the off-diagonal weights $T_{jl}$ and $\hat{\boldsymbol\Pi}_{jl}$, both $O_{p}(m^{-1})$, so the same Cauchy--Schwarz argument under the joint moment bound gives $|\gamma^{(j)}_{l}|=O_{p}(m^{-1-2\varrho})$ and $|\rho_{ij}|=O_{p}(m^{-1-2\varrho})$. The mean is therefore
\[
    \mathbb E\bigl[w^{\top}\Psi\tilde c\mid\mathbf X\bigr]
    =\sum_{i,j}\Psi_{ij}\,\hat{\boldsymbol\Pi}_{ij}\nu_{j}+O_{p}(\bar\psi m^{1-2\varrho})
    =O_{p}(\bar\psi m^{1-2\varrho}),
\]
the first term by the row sums of $\hat{\boldsymbol\Pi}$, the cross-index contribution by
\[
    \sum_{j}\sum_{l\ne j}|\gamma^{(j)}_{l}|\,\Bigl(\sum_{i}|\Psi_{ij}\hat{\boldsymbol\Pi}_{il}|\Bigr)
    \le\bar\psi\,O_{p}(1)\sum_{j,l}|\gamma^{(j)}_{l}|
    =O_{p}(m^{1-2\varrho}),
\]
the $m(m-1)$ cross-index pairs each $O_{p}(m^{-1-2\varrho})$. Write $\mu_{ij}:=\mathbb E[w_{i}\tilde c_{j}\mid\mathbf X]$. The centered second moment pairs as in the idiosyncratic case, and centering at the true conditional means makes the crossed family the exact product of the crossed means,
\[
    \mathbb E\bigl[(w_{i}\tilde c_{j}-\mu_{ij})(w_{i'}\tilde c_{j'}-\mu_{i'j'})\mid\mathbf X\bigr]
    =c_{ii'}\,\mathbb E\bigl[\tilde c_{j}\tilde c_{j'}\mid\mathbf X\bigr]
     +\mu_{ij'}\,\mu_{i'j}
     +\kappa'_{ij,i'j'},
\]
with $\kappa'_{ij,i'j'}$ collecting the four-equal $\varepsilon$--factor configurations together with the same-observation and diagonal corrections left out of the two named families, $|\kappa'_{ij,i'j'}|\le Cm^{-4\varrho}\sum_{l}|\hat{\boldsymbol\Pi}_{il}\hat{\boldsymbol\Pi}_{i'l}|\,\bigl(|T_{jl}|+|\hat{\boldsymbol\Pi}_{jl}|\bigr)\bigl(|T_{j'l}|+|\hat{\boldsymbol\Pi}_{j'l}|\bigr)$ (element-wise orders, Cauchy--Schwarz under the stated moment bound on $\varepsilon_{0}^{2}\lVert\nabla_{\theta}^{2}\varepsilon\rVert^{2}$). The first family contributes $O_{p}(\bar\psi^{2}m^{2-4\varrho})$: its four $(i,i')\times(j,j')$ configurations pair the recorded diagonal and off-diagonal orders of $c_{ii'}$ and $\mathbb E[\tilde c_{j}\tilde c_{j'}\mid\mathbf X]$, the four products $O_{p}(m^{-4\varrho})$, $O_{p}(m^{-1-4\varrho})$, $O_{p}(m^{-1-4\varrho})$, $O_{p}(m^{-2-4\varrho})$ against $m^{2}$, $m^{3}$, $m^{3}$, $m^{4}$ index-quadruples respectively. The second family contributes $O_{p}(\bar\psi^{2}m^{2-4\varrho})$: the $i,i'$-sum is $\bar\psi^{2}O_{p}(1)$ per index pair by the row sums of $\hat{\boldsymbol\Pi}$, against $|\mu_{ij'}\mu_{i'j}|\le Cm^{-4\varrho}$ over the $m^{2}$ pairs, since $|\mu_{ij}|\le|\hat{\boldsymbol\Pi}_{ij}|\,|\nu_{j}|+|\rho_{ij}|=O_{p}(m^{-2\varrho})$ uniformly. The cumulant contributes $O_{p}(\bar\psi^{2}m^{1-4\varrho})$: its four-equal configurations force the per-observation factors of $\tilde c_{j}$ and $\tilde c_{j'}$ onto the common index $l$, so summing first over $j,j'$ costs $\bar\psi^{2}(\max_{j}\sum_{i}|T_{ij}|)^{2}$ per $l$, and the remaining $\sum_{l}(\sum_{i}|\hat{\boldsymbol\Pi}_{il}|)^{2}=O_{p}(m)$. Hence
\[
    \bigl|w^{\top}\Psi\tilde c\bigr|=O_{p}\bigl(\bar\psi m^{1-2\varrho}\bigr).
\]

\medskip\noindent\emph{Quadratic parts.}\;
The quadratic forms are controlled by two deterministic operator envelopes, uniform in $\eta$: $\lVert K^{\mathrm{od}}_{\eta}\rVert_{\mathrm{op}}\le(m-1)\bar k\lambda_{L}^{-1}$ by the row sums and the entrywise envelope $\lVert g_{\eta}\rVert_{\infty}\le\bar k\lambda_{L}^{-1}$, and $\lVert W'_{\eta}\rVert_{\mathrm{op}}\le2(m-1)^{2}\bar k^{2}\lambda_{L}^{-2}$ for $W'_{\eta}:=(K^{\mathrm{od},\top}_{\eta}K^{\mathrm{od}}_{\eta})-\operatorname{diag}(s_{\eta})$, the matrix of the $(Q,R)$-kernel, since the operator norm of $K^{\mathrm{od},\top}_{\eta}K^{\mathrm{od}}_{\eta}$ is $\lVert K^{\mathrm{od}}_{\eta}\rVert_{\mathrm{op}}^{2}$ and $\max_{j}s_{\eta,j}\le(m-1)\bar k^{2}\lambda_{L}^{-2}$. For components without $a$, the bounds are then pathwise, the operator norms acting on the Euclidean norms:
\[
    \tfrac{1}{m-1}\bigl|e^{\top}K^{\mathrm{od}}_{\eta}e'\bigr|
    \le\bar k\lambda_{L}^{-1}\lVert e\rVert\lVert e'\rVert,
    \qquad
    \tfrac{2}{(m-1)^{2}}\bigl|e^{\top}W'_{\eta}e'\bigr|
    \le4\bar k^{2}\lambda_{L}^{-2}\lVert e\rVert\lVert e'\rVert,
\]
so the $(b,b)$-pieces are $o_{p}(m^{1-6\varrho})$, and every piece involving $u$ or $\tilde c$ is $O_{p}(m^{1-4\varrho})$ in $\Delta^{U}_{m}$ and $O_{p}(m^{1-4\varrho})$ in $\Delta^{R}_{m}$. For the $(a,a)$-pieces, the conditional mean is
\[
    \mathbb E\Bigl[\tfrac{1}{m-1}a^{\top}K^{\mathrm{od}}_{\eta}a\Bigm|\mathbf X\Bigr]
    =\tfrac{1}{m-1}\sum_{i\ne j}g_{ij}\,\mathbb E[a_{i}a_{j}\mid\mathbf X]
    =O_{p}(\lVert\boldsymbol\delta\rVert^{2}),
\]
by diagonal dominance of $\mathbb E[aa^{\top}\mid\mathbf X]$ (with $\Psi=K^{\mathrm{od}}_{\eta}$), and the fluctuation is, by \eqref{eq:quadvar} with $\zeta_{k}:=\boldsymbol\eta_{0,k}^{\top}\boldsymbol\delta$ and $A:=TK^{\mathrm{od}}_{\eta}T$,
\[
    \operatorname{Var}\Bigl(\tfrac{1}{m-1}a^{\top}K^{\mathrm{od}}_{\eta}a\Bigm|\mathbf X\Bigr)
    =\tfrac{1}{(m-1)^{2}}\Bigl(2\sum_{k,l}A_{kl}^{2}d_{k}d_{l}+\sum_{k}A_{kk}^{2}\bigl(\mathbb E[\zeta_{k}^{4}\mid\mathbf X]-3d_{k}^{2}\bigr)\Bigr)
    =O_{p}\bigl(\lambda_{L}^{-2}\lVert\boldsymbol\delta\rVert^{4}\bigr),
\]
since $|A_{kl}|\le\lambda_{L}^{-1}C$ by the row sums of $T$, $\sum_{k}d_{k}=O_{p}(m\lVert\boldsymbol\delta\rVert^{2})$, and $\mathbb E[\zeta_{k}^{4}\mid\mathbf X]\le\lVert\boldsymbol\delta\rVert^{4}\mathbb E[\lVert\boldsymbol\eta_{0,k}\rVert^{4}\mid X_{k}]\le C\lVert\boldsymbol\delta\rVert^{4}$ (Assumption~\ref{ass:E}(iii)); hence the $(a,a)$-piece of $\Delta^{U}_{m}$ is $O_{p}(\lambda_{L}^{-1}\lVert\boldsymbol\delta\rVert^{2})$, and the $(a,a)$-pieces of $\Delta^{Q}_{m}$ and $\Delta^{R}_{m}$ are the same computation with the weight matrices $\operatorname{diag}(s)$ and $W'$, $W'_{jk}:=\sum_{i\ne j,k}g_{ij}g_{ik}$, $\lVert W'\rVert_{\max}\le m\lambda_{L}^{-2}$, giving $O_{p}(\lambda_{L}^{-2}\lVert\boldsymbol\delta\rVert^{2})$ after their prefactors. For the mixed pieces $(a,b)$ and $(a,u)$, the conditional mean vanishes, and expanding $a^{\top}\Psi e=(T\Psi e)^{\top}\boldsymbol\eta_{0}\boldsymbol\delta$ in the conditionally independent centered variables $\boldsymbol\eta_{0,k}^{\top}\boldsymbol\delta$,
\[
    \operatorname{Var}\bigl(a^{\top}\Psi e\mid\mathbf X\bigr)
    =\sum_{k=1}^{m}(T\Psi e)_{k}^{2}\,\boldsymbol\delta^{\top}V_{k}\boldsymbol\delta
    \le\lVert\Psi e\rVert\,\lVert\Psi e\rVert_{\infty}\,O_{p}\bigl(\lVert\boldsymbol\delta\rVert^{2}m^{1/2}\bigr),
\]
the inequality Cauchy--Schwarz in $\bigl((T\Psi e)_{k}^{2}\bigr)_{k}$ against $\bigl(\sum_{k}\lVert V_{k}\rVert^{2}\bigr)^{1/2}=O_{p}(m^{1/2})$, with $\lVert Tz\rVert\le\lVert z\rVert$ and $\lVert Tz\rVert_{\infty}\le(1+C)\lVert z\rVert_{\infty}$; hence $|a^{\top}\Psi e|=O_{p}(\lVert\boldsymbol\delta\rVert\bar\psi m\lVert e\rVert)$, giving $O_{p}(\bar\psi m^{1/2-3\varrho})$ in $\Delta^{U}_{m}$ and $O_{p}(\bar\psi^{2}m^{1/2-3\varrho})$ in $\Delta^{R}_{m}$, both $o_{p}(m^{-\varrho})$, after the prefactors; for $(a,\tilde c)$, Cauchy--Schwarz gives mean $O_{p}(\bar\psi m^{-3\varrho})$ (same-index survival of $\mathbb E[a_{j}\tilde c_{j}]$, of size $O_{p}(\lVert\boldsymbol\delta\rVert m^{-2\varrho})$ by the element-wise orders) and fluctuation, by the same variance expansion, with $\lVert\tilde c\rVert=O_{p}(m^{1/2-2\varrho})$ by the displayed second moments and the same averaging as for $\lVert u\rVert$, so that $\lVert K^{\mathrm{od}}_{\eta}\tilde c\rVert\le m\lambda_{L}^{-1}\lVert\tilde c\rVert=O_{p}(\lambda_{L}^{-1}m^{3/2-2\varrho})$ and $\lVert K^{\mathrm{od}}_{\eta}\tilde c\rVert_{\infty}\le\lambda_{L}^{-1}\lVert\tilde c\rVert_{1}=O_{p}(\lambda_{L}^{-1}m^{1-2\varrho})$,
\[
    \tfrac{2}{m-1}\bigl|a^{\top}K^{\mathrm{od}}_{\eta}\tilde c\bigr|
    =\tfrac{2}{m-1}\,O_{p}\bigl(\lVert\boldsymbol\delta\rVert\lambda_{L}^{-1}m^{3/2-2\varrho}\bigr)
    =O_{p}\bigl(\lambda_{L}^{-1}m^{1/2-3\varrho}\bigr)=o_{p}\bigl(m^{-\varrho}\bigr).
\]

\medskip\noindent\emph{Assembly.}\;
The linear channels collect through the four instantiations $\Psi=K^{\mathrm{od}}_{\eta}$, $K^{\circ,\mathrm{od}}_{\eta}$, $\Psi^{r}_{\eta}$, $\operatorname{diag}(s)$ with prefactors $\Gamma_{1}:=\tfrac{2}{m-1}$, $\Gamma_{2}:=\tfrac{2(m-2)}{(m-1)^{2}}$, $\Gamma_{3}:=\Gamma_{4}:=\tfrac{2}{(m-1)^{2}}$ and $\bar\psi$-values $\lambda_{L}^{-1}$, $\lambda_{L}^{-2}$, $2m\lambda_{L}^{-2}$, $(m-1)\lambda_{L}^{-2}$; each product $\Gamma_{\Psi}\bar\psi_{\Psi}$ is $O(m^{-1})$ up to $\lambda_{L}$-constants, so
\[
\begin{aligned}
a\text{-linear:}\quad
&\sum_{\Psi}\Gamma_{\Psi}\bar\psi_{\Psi}\,O_{p}(m^{1-\varrho})&&=O_{p}(m^{-\varrho}),\\[2pt]
u\text{-linear:}\quad
&\sum_{\Psi}\Gamma_{\Psi}\bar\psi_{\Psi}\,O_{p}(m^{3/2-2\varrho})&&=O_{p}(m^{1/2-2\varrho}),\\[2pt]
\tilde c\text{- and }b\text{-linear:}\quad
&\sum_{\Psi}\Gamma_{\Psi}\bar\psi_{\Psi}\,\bigl(O_{p}(m^{1-2\varrho})+o_{p}(m^{3/2-3\varrho})\bigr)&&=O_{p}(m^{-2\varrho})+o_{p}(m^{1/2-3\varrho}),
\end{aligned}
\]
with $O_{p}(m^{-3\varrho})$ added by the same-index means of the $(a,\tilde c)$-pairing. The quadratic parts collect as $O_{p}(m^{-2\varrho})$ from the $(a,a)$-pieces and $O_{p}(m^{1-4\varrho})$ from the $(e,e')$-pieces with $e,e'\in\{u,\tilde c\}$, the $(b,b)$-pieces being $o_{p}(m^{1-6\varrho})$, and the mixed $(a,e)$- and $(a,\tilde c)$-pieces are $o_{p}(m^{-\varrho})$. Since $\varrho>1/4$ every exponent displayed is negative, and the three statement rates correspond one-to-one to the idiosyncratic-linear, smooth-second-order, and quadratic channels; hence
\[
    \sup_{\eta\in\Upsilon}\bigl|\Delta^{U}_{m}(\eta)+\Delta^{Q}_{m}(\eta)+\Delta^{R}_{m}(\eta)\bigr|
    =O_{p}\bigl(m^{-\varrho}+m^{1/2-2\varrho}+m^{1-4\varrho}\bigr)=o_{p}(1),
\]
which is the displayed conclusion, the left-hand side being the difference of the two $m\,\mathrm{CV}^{\mathrm{dep}}_{m}$ processes by the linearity of the linearization displays in the residual vector.

\medskip\noindent When the score is $X$-measurable, $\nabla_{\theta}\varepsilon(s;\theta)=g_{\theta}(X)$ implies $\boldsymbol\eta_{0}=\boldsymbol 0$ and hence $a=\boldsymbol 0$, leaving the remainder recorded in the statement.
\end{proof}

\begin{lemma}[Bilinear representation of the criterion's quadratic statistics]\label{lem:bilinear_cv}
Let Assumption~\ref{ass:A} hold and let $\varepsilon_{\perp,1},\dots,\varepsilon_{\perp,m}$ be scalar residual coefficients on a sample of size $m$. The identities in (i) and (ii) are algebraic and hold for arbitrary scalar coefficients. Define the scaled mean embedding $M_{m}(\gamma):=m^{-1/2}\sum_{i=1}^{m}\varepsilon_{\perp,i}^{0}k_{\gamma}(X_{i},\cdot)\in\mathcal H_{\gamma}$, its transport $\widetilde M_{m}(\gamma):=J_{\gamma}M_{m}(\gamma)$ to the common space $\mathcal H_{\gamma_0}$, and the \emph{population} operators
\[
    A^{0}_{\eta}:=(C_{\gamma}+\lambda I)^{-1},
    \qquad
    B^{0}_{\eta}:=(C_{\gamma}+\lambda I)^{-1}C_{\gamma}(C_{\gamma}+\lambda I)^{-1},
\]
on $\mathcal H_{\gamma}$. Their reproducing kernels are $g_{\eta}$ and $(g_{\eta}\circ g_{\eta})$ on $L^{2}(\mathbb P_{X})$, where $(g_{\eta}\circ g_{\eta})(x,y):=\mathbb E_{Z}[g_{\eta}(x,Z)g_{\eta}(Z,y)]$: by the Mercer expansion of Assumption~\ref{ass:A},
$\langle k_{\gamma}(X_{i},\cdot),A^{0}_{\eta}k_{\gamma}(X_{j},\cdot)\rangle_{\mathcal H_{\gamma}}=g_{\eta}(X_{i},X_{j})$
and likewise with $B^{0}_{\eta}$ and $(g_{\eta}\circ g_{\eta})(X_{i},X_{j})$. Then:
\begin{enumerate}[label=(\roman*)]
\item \emph{Exact identity.} For every $\eta\in\Upsilon$,
\[
    mU_{m}(\eta)=\frac{m}{m-1}\bigl\langle M_{m}(\gamma),A^{0}_{\eta}M_{m}(\gamma)\bigr\rangle_{\mathcal H_{\gamma}}
      -\frac{1}{m-1}\sum_{i=1}^{m}(\varepsilon_{\perp,i}^{0})^{2}\,g_{\eta}(X_{i},X_{i}).
\]
\item \emph{Order-two reduction of $R_{m}$.} With $h_{\eta}^{(R)}(z_{1},z_{2}):=\varepsilon_{\perp,1}\varepsilon_{\perp,2}\,(g_{\eta}\circ g_{\eta})(x_{1},x_{2})$,
\[
\begin{aligned}
    mR_{m}(\eta)&=\frac{m(m-2)}{m-1}\,U_{m}(h_{\eta}^{(R)})+\mathcal R_{m}(\eta),\\
    \mathcal R_{m}(\eta)&:=\frac{2}{(m-1)^{2}}\sum_{i=1}^{m}\sum_{\substack{j<k\\ j,k\ne i}}r_{\eta}(X_{i};X_{j},X_{k})\,\varepsilon_{\perp,j}^{0}\varepsilon_{\perp,k}^{0},
\end{aligned}
\]
where $r_{\eta}(x;y,z):=g_{\eta}(x,y)g_{\eta}(x,z)-(g_{\eta}\circ g_{\eta})(y,z)$ satisfies $\mathbb E_{X}[r_{\eta}(X;y,z)]=0$, and
\[
    mU_{m}(h_{\eta}^{(R)})=\frac{m}{m-1}\bigl\langle M_{m}(\gamma),B^{0}_{\eta}M_{m}(\gamma)\bigr\rangle_{\mathcal H_{\gamma}}
      -\frac{1}{m-1}\sum_{i=1}^{m}(\varepsilon_{\perp,i}^{0})^{2}(g_{\eta}\circ g_{\eta})(X_{i},X_{i});
\]
The uniform probabilistic bound for $\mathcal R_{m}$ is stated in Lemma~\ref{lem:remainder_cv} below. The calligraphic $\mathcal R_{m}$ denotes the anchor-centered remainder of this reduction and is to be distinguished from the criterion component $R_{m}$ of Section~\ref{sec:uniform_limit}.
\item \emph{Deterministic moduli and boundedness.}\; \emph{(a)} For all $\eta,\eta'\in\Upsilon$,
\[
\begin{gathered}
    \sup_{x,y}\bigl|g_{\eta}(x,y)-g_{\eta'}(x,y)\bigr|\le C_{g}\,\lVert\eta-\eta'\rVert,
    \qquad
    \sup_{x,y,z}\bigl|r_{\eta}(x;y,z)-r_{\eta'}(x;y,z)\bigr|\le C''\,\lVert\eta-\eta'\rVert,\\
    \sup_{x,y}\bigl|(g_{\eta}\circ g_{\eta})(x,y)-(g_{\eta'}\circ g_{\eta'})(x,y)\bigr|\le C'_{g}\,\lVert\eta-\eta'\rVert,
\end{gathered}
\]
with $C'':=2\bar k\,\lambda_{L}^{-1}C_{g}+C'_{g}$ and all three constants deterministic, depending only on $(\gamma_{L},\gamma_{U},\lambda_{L},\lambda_{U},d,\bar k)$. Moreover, the three kernels are entrywise continuously differentiable in $\eta$ on $\Upsilon$, and, for $l\in\{\gamma,\lambda\}$,
\[
    \sup_{\eta,x,y}\bigl|\partial_{\eta_{l}}g_{\eta}(x,y)\bigr|\le C_{g},
    \qquad
    \sup_{\eta,x,y}\bigl|\partial_{\eta_{l}}(g_{\eta}\circ g_{\eta})(x,y)\bigr|\le C'_{g},
    \qquad
    \sup_{\eta,x,y,z}\bigl|\partial_{\eta_{l}}r_{\eta}(x;y,z)\bigr|\le C'',
\]
the derivative envelopes being bounded by the moduli above. \emph{(b)} If, in addition, Assumptions~\ref{ass:B}--\ref{ass:F} and \ref{ass:G}(iii), (v) hold, and under the local alternative the Pitman calibration of Lemma~\ref{lem:lipschitz_common}, then $\sup_{\gamma}\lVert\widetilde M_{m}(\gamma)\rVert_{\mathcal H_{\gamma_0}}=O_{p}(1)$.
\end{enumerate}
\end{lemma}

\begin{proof}
\emph{(i)} The claim is algebra alone: the coefficients $\varepsilon_{\perp,i}^{0}$ enter as indeterminates, and no distributional assumption is used. Since $M_{m}(\gamma)=m^{-1/2}\sum_{i=1}^{m}\varepsilon_{\perp,i}^{0}\,k_{\gamma}(X_{i},\cdot)$, bilinearity of the inner product and the reproducing-kernel identification of the statement give
\[
\begin{aligned}
    \bigl\langle M_{m}(\gamma),A^{0}_{\eta}M_{m}(\gamma)\bigr\rangle_{\mathcal H_{\gamma}}
    &=\frac{1}{m}\sum_{i=1}^{m}\sum_{j=1}^{m}\varepsilon_{\perp,i}^{0}\varepsilon_{\perp,j}^{0}\,
      \bigl\langle k_{\gamma}(X_{i},\cdot),A^{0}_{\eta}k_{\gamma}(X_{j},\cdot)\bigr\rangle_{\mathcal H_{\gamma}}\\
    &=\frac{1}{m}\sum_{i=1}^{m}\sum_{j=1}^{m}\varepsilon_{\perp,i}^{0}\varepsilon_{\perp,j}^{0}\,g_{\eta}(X_{i},X_{j}).
\end{aligned}
\]
The index set splits into the strict off-diagonal and the diagonal,
\[
    \frac{1}{m}\sum_{i,j}\varepsilon_{\perp,i}^{0}\varepsilon_{\perp,j}^{0}\,g_{\eta}(X_{i},X_{j})
    =\frac{1}{m}\sum_{i\ne j}\varepsilon_{\perp,i}^{0}\varepsilon_{\perp,j}^{0}\,g_{\eta}(X_{i},X_{j})
     +\frac{1}{m}\sum_{i=1}^{m}(\varepsilon_{\perp,i}^{0})^{2}\,g_{\eta}(X_{i},X_{i}),
\]
and the symmetry $g_{\eta}(x,y)=g_{\eta}(y,x)$ pairs the summand of each ordered pair $(i,j)$, $i\ne j$, with the equal summand of $(j,i)$, so
\[
    \frac{1}{m}\sum_{i\ne j}\varepsilon_{\perp,i}^{0}\varepsilon_{\perp,j}^{0}\,g_{\eta}(X_{i},X_{j})
    =\frac{2}{m}\sum_{1\le i<j\le m}\varepsilon_{\perp,i}^{0}\varepsilon_{\perp,j}^{0}\,g_{\eta}(X_{i},X_{j}).
\]
With the recorded normalization $U_{m}(\eta)=\frac{2}{m(m-1)}\sum_{i<j}\varepsilon_{\perp,i}^{0}\varepsilon_{\perp,j}^{0}\,g_{\eta}(X_{i},X_{j})$, the off-diagonal term is
\[
    \frac{2}{m}\sum_{i<j}\varepsilon_{\perp,i}^{0}\varepsilon_{\perp,j}^{0}\,g_{\eta}(X_{i},X_{j})
    =\frac{m-1}{m}\cdot\frac{2}{m-1}\sum_{i<j}\varepsilon_{\perp,i}^{0}\varepsilon_{\perp,j}^{0}\,g_{\eta}(X_{i},X_{j})
    =\frac{m-1}{m}\,mU_{m}(\eta),
\]
the first equality inserting $\tfrac{m-1}{m-1}$, the second the normalization. Collecting the two pieces,
\[
    \bigl\langle M_{m}(\gamma),A^{0}_{\eta}M_{m}(\gamma)\bigr\rangle_{\mathcal H_{\gamma}}
    =\frac{m-1}{m}\,mU_{m}(\eta)+\frac{1}{m}\sum_{i=1}^{m}(\varepsilon_{\perp,i}^{0})^{2}\,g_{\eta}(X_{i},X_{i}),
\]
multiplying by $\tfrac{m}{m-1}$ and simplifying the scalar coefficients,
\[
    \frac{m}{m-1}\bigl\langle M_{m}(\gamma),A^{0}_{\eta}M_{m}(\gamma)\bigr\rangle_{\mathcal H_{\gamma}}
    =mU_{m}(\eta)+\frac{1}{m-1}\sum_{i=1}^{m}(\varepsilon_{\perp,i}^{0})^{2}\,g_{\eta}(X_{i},X_{i}),
\]
and subtracting the diagonal term yields
\[
    mU_{m}(\eta)=\frac{m}{m-1}\bigl\langle M_{m}(\gamma),A^{0}_{\eta}M_{m}(\gamma)\bigr\rangle_{\mathcal H_{\gamma}}
      -\frac{1}{m-1}\sum_{i=1}^{m}(\varepsilon_{\perp,i}^{0})^{2}\,g_{\eta}(X_{i},X_{i}),
\]
which is (i).

\emph{(ii)} The proof is algebra alone: four exact identities, none involving the law of the residuals. The uniform probabilistic bound for the remainder $\mathcal R_{m}$ is established in the proof of Lemma~\ref{lem:remainder_cv} below.

\emph{The anchor-centered split.}\; The split $g_{\eta}(x,y)g_{\eta}(x,z)=(g_{\eta}\circ g_{\eta})(y,z)+r_{\eta}(x;y,z)$ restates the definition of $r_{\eta}$ in the statement; $r_{\eta}$ is symmetric in $(y,z)$, and $\mathbb E_{X}[r_{\eta}(X;y,z)]=\mathbb E_{X}[g_{\eta}(X,y)g_{\eta}(X,z)]-(g_{\eta}\circ g_{\eta})(y,z)=0$ by the definition of the composed kernel.

\emph{The reduction, by anchor counting.}\; Substituting the split into $mR_{m}(\eta)$ and collecting the two pieces of each summand,
\[
\begin{aligned}
    mR_{m}(\eta)
    &=\underbrace{\frac{2}{(m-1)^{2}}\sum_{i}\sum_{\substack{j<k\\ j,k\ne i}}(g_{\eta}\circ g_{\eta})(X_{j},X_{k})\,\varepsilon_{\perp,j}^{0}\varepsilon_{\perp,k}^{0}}_{\text{the population piece: independent of the anchor } i}\\
    &\quad+\underbrace{\frac{2}{(m-1)^{2}}\sum_{i}\sum_{\substack{j<k\\ j,k\ne i}}r_{\eta}(X_{i};X_{j},X_{k})\,\varepsilon_{\perp,j}^{0}\varepsilon_{\perp,k}^{0}}_{\text{the anchored remainder}}.
\end{aligned}
\]
By the definition in the statement, the second double sum is $\mathcal R_{m}(\eta)$. The population piece does not depend on the anchor $i$, and each unordered pair $\{j,k\}$ admits exactly $m-2$ anchors:
\[
    \frac{2}{(m-1)^{2}}\sum_{i}\sum_{\substack{j<k\\ j,k\ne i}}(g_{\eta}\circ g_{\eta})(X_{j},X_{k})\,\varepsilon_{\perp,j}^{0}\varepsilon_{\perp,k}^{0}
    =\frac{2(m-2)}{(m-1)^{2}}\sum_{j<k}(g_{\eta}\circ g_{\eta})(X_{j},X_{k})\,\varepsilon_{\perp,j}^{0}\varepsilon_{\perp,k}^{0}.
\]
The ordered double sum in $U_{m}(h_{\eta}^{(R)})$ counts each unordered pair twice:
\[
    \sum_{j<k}(g_{\eta}\circ g_{\eta})(X_{j},X_{k})\,\varepsilon_{\perp,j}^{0}\varepsilon_{\perp,k}^{0}
    =\frac{m(m-1)}{2}\,U_{m}\bigl(h_{\eta}^{(R)}\bigr),
\]
and the population piece equals
\[
    \frac{2(m-2)}{(m-1)^{2}}\cdot\frac{m(m-1)}{2}\,U_{m}\bigl(h_{\eta}^{(R)}\bigr)
    =\frac{m(m-2)}{m-1}\,U_{m}\bigl(h_{\eta}^{(R)}\bigr),
\]
which is the reduction displayed in the statement.

\emph{The composed kernel.}\; $A^{0}_{\eta}:=(C_{\gamma}+\lambda I)^{-1}$ is self-adjoint with $\lVert A^{0}_{\eta}\rVert_{\mathrm{op}}\le\lambda_{L}^{-1}$, and $B^{0}_{\eta}:=A^{0}_{\eta}C_{\gamma}A^{0}_{\eta}$ has the composed kernel as its reproducing kernel:
\[
\begin{aligned}
    \bigl\langle k_{\gamma}(x,\cdot),B^{0}_{\eta}k_{\gamma}(y,\cdot)\bigr\rangle_{\mathcal H_{\gamma}}
    &=\bigl\langle A^{0}_{\eta}k_{\gamma}(x,\cdot),\,C_{\gamma}\,A^{0}_{\eta}k_{\gamma}(y,\cdot)\bigr\rangle_{\mathcal H_{\gamma}}\\
    &=\mathbb E_{Z}\bigl[(A^{0}_{\eta}k_{\gamma}(x,\cdot))(Z)\,(A^{0}_{\eta}k_{\gamma}(y,\cdot))(Z)\bigr]\\
    &=\mathbb E_{Z}\bigl[g_{\eta}(x,Z)\,g_{\eta}(Z,y)\bigr]
    =(g_{\eta}\circ g_{\eta})(x,y),
\end{aligned}
\]
the first equality by self-adjointness of $A^{0}_{\eta}$, the second by the covariance property $\langle f,C_{\gamma}g\rangle_{\mathcal H_{\gamma}}=\mathbb E[f(X)g(X)]$ for $f,g\in\mathcal H_{\gamma}$, and the third by the reproducing property $(A^{0}_{\eta}k_{\gamma}(y,\cdot))(Z)=\langle k_{\gamma}(Z,\cdot),A^{0}_{\eta}k_{\gamma}(y,\cdot)\rangle_{\mathcal H_{\gamma}}=g_{\eta}(Z,y)$.

\emph{The bilinear representation.}\; With the scaled embedding
\[
    M_{m}(\gamma):=m^{-1/2}\sum_{i=1}^{m}\varepsilon_{\perp,i}^{0}\,k_{\gamma}(X_{i},\cdot),
\]
the split of the double sum into $i\ne j$ and $i=j$ gives
\[
\begin{aligned}
    \frac{m}{m-1}\bigl\langle M_{m},B^{0}_{\eta}M_{m}\bigr\rangle_{\mathcal H_{\gamma}}
    &=\frac{1}{m-1}\sum_{i,j}\varepsilon_{\perp,i}^{0}\varepsilon_{\perp,j}^{0}\,(g_{\eta}\circ g_{\eta})(X_{i},X_{j})\\
    &=mU_{m}\bigl(h_{\eta}^{(R)}\bigr)
      +\frac{1}{m-1}\sum_{i}(\varepsilon_{\perp,i}^{0})^{2}(g_{\eta}\circ g_{\eta})(X_{i},X_{i}),
\end{aligned}
\]
and solving for $mU_{m}$ gives the identity displayed in the statement. This is the algebra of (i) with $A^{0}_{\eta}$ replaced by $B^{0}_{\eta}$ and $g_{\eta}$ by its composition; like (i), it holds for arbitrary scalar coefficients.

\emph{(iii)} Transport to the common space, with $\widetilde A^{0}_{\eta}:=J_{\gamma}A^{0}_{\eta}J_{\gamma}^{-1}$ and $\widetilde B^{0}_{\eta}:=J_{\gamma}B^{0}_{\eta}J_{\gamma}^{-1}$: $g_{\eta}(x,y)=\bigl\langle\widetilde k_{\gamma}(x),\widetilde A^{0}_{\eta}\widetilde k_{\gamma}(y)\bigr\rangle_{\mathcal H_{\gamma_0}}$ with $\widetilde k_{\gamma}(x):=J_{\gamma}k_{\gamma}(x,\cdot)$ and $\sup_{x}\lVert\widetilde k_{\gamma}(x)\rVert\le\bar k^{1/2}$; the feature distance $\bigl\lVert\widetilde k_{\gamma}(x)-\widetilde k_{\gamma'}(x)\bigr\rVert_{\mathcal H_{\gamma_0}}\le L_{d}\,\lvert\gamma-\gamma'\rvert$ follows from Fact~\ref{fact:pairdiff}(i) by the mean value theorem. Adding and subtracting $\bigl\langle\widetilde k_{\gamma'}(x),\widetilde A^{0}_{\eta}\widetilde k_{\gamma}(y)\bigr\rangle$ and telescoping over the three factors,
\[
    \bigl|g_{\eta}(x,y)-g_{\eta'}(x,y)\bigr|
    \le 2\lambda_{L}^{-1}\bar k^{1/2}\,L_{d}\,|\gamma-\gamma'|
      +\bar k\,\bigl\lVert\widetilde A^{0}_{\eta}-\widetilde A^{0}_{\eta'}\bigr\rVert_{\mathrm{op}},
\]
the two feature-increment terms being contracted by $\lVert\widetilde A^{0}_{\eta}\rVert_{\mathrm{op}}\le\lambda_{L}^{-1}$ and each bounded by $\bar k^{1/2}L_{d}\lvert\gamma-\gamma'\rvert$; the operator increment obeys, by the resolvent identity,
$\lVert\widetilde A^{0}_{\eta}-\widetilde A^{0}_{\eta'}\rVert_{\mathrm{op}}\le\lambda_{L}^{-2}\bigl(\lVert\widetilde C_{\gamma}-\widetilde C_{\gamma'}\rVert_{\mathrm{op}}+\lvert\lambda-\lambda'\rvert\bigr)\le\lambda_{L}^{-2}\bigl(2\bar k^{1/2}L_{d}\,\lvert\gamma-\gamma'\rvert+\lvert\lambda-\lambda'\rvert\bigr)$,
the last step combining Fact~\ref{fact:pairdiff}(iii), $\lVert\widetilde C_{\gamma}-\widetilde C_{\gamma'}\rVert_{\mathrm{HS}}\le2\,\bar k^{1/2}L_{d}\,\lvert\gamma-\gamma'\rvert$, with $\lVert\cdot\rVert_{\mathrm{op}}\le\lVert\cdot\rVert_{\mathrm{HS}}$. This is the first modulus. The second follows by the same telescope with $\widetilde B^{0}$ in place of $\widetilde A^{0}$, using the product increment: on the common space $\widetilde B^{0}_{\eta}=\bigl(\widetilde A^{0}_{\eta}\bigr)^{2}\widetilde C_{\gamma}$, and the split
$\widetilde B^{0}_{\eta}-\widetilde B^{0}_{\eta'}=\bigl[\bigl(\widetilde A^{0}_{\eta}\bigr)^{2}-\bigl(\widetilde A^{0}_{\eta'}\bigr)^{2}\bigr]\widetilde C_{\gamma}+\bigl(\widetilde A^{0}_{\eta'}\bigr)^{2}\bigl(\widetilde C_{\gamma}-\widetilde C_{\gamma'}\bigr)$
gives
\[
    \bigl\lVert\widetilde B^{0}_{\eta}-\widetilde B^{0}_{\eta'}\bigr\rVert_{\mathrm{op}}
    \le 2\bar k\lambda_{L}^{-1}\bigl\lVert\widetilde A^{0}_{\eta}-\widetilde A^{0}_{\eta'}\bigr\rVert_{\mathrm{op}}
      +\lambda_{L}^{-2}\bigl\lVert\widetilde C_{\gamma}-\widetilde C_{\gamma'}\bigr\rVert_{\mathrm{op}},
\]
using $\bigl\lVert\bigl(\widetilde A^{0}_{\eta}\bigr)^{2}-\bigl(\widetilde A^{0}_{\eta'}\bigr)^{2}\bigr\rVert_{\mathrm{op}}\le2\lambda_{L}^{-1}\lVert\widetilde A^{0}_{\eta}-\widetilde A^{0}_{\eta'}\rVert_{\mathrm{op}}$ and $\lVert\widetilde C_{\gamma}\rVert_{\mathrm{op}}\le\operatorname{tr}\widetilde C_{\gamma}=\mathbb E[k_{\gamma}(X,X)]=1=\bar k$ for the unit-normalized Gaussian kernel.

\emph{The anchor modulus.}\; Adding and subtracting $g_{\eta'}(x,y)\,g_{\eta}(x,z)$ in the first product of the definition of $r_{\eta}$ gives the split, with $\delta g_{\eta}:=g_{\eta}-g_{\eta'}$ and $\delta(g\circ g)_{\eta}:=(g_{\eta}\circ g_{\eta})-(g_{\eta'}\circ g_{\eta'})$,
\[
    r_{\eta}(x;y,z)-r_{\eta'}(x;y,z)
    =\delta g_{\eta}(x,y)\,g_{\eta}(x,z)+g_{\eta'}(x,y)\,\delta g_{\eta}(x,z)-\delta(g\circ g)_{\eta}(y,z),
\]
so, by the sup-norm bound $|g_{\eta}|\le\bar k\,\lambda_{L}^{-1}$ and the two moduli above,
\[
    \bigl|r_{\eta}(x;y,z)-r_{\eta'}(x;y,z)\bigr|
    \le 2\bar k\lambda_{L}^{-1}C_{g}\lVert\eta-\eta'\rVert+C'_{g}\lVert\eta-\eta'\rVert
    =C''\lVert\eta-\eta'\rVert,
\]
the third modulus of (iii)(a), which proves that clause.

\emph{Differentiability.}\; On the common space, $\partial_{\lambda}\widetilde A^{0}_{\eta}=-(\widetilde A^{0}_{\eta})^{2}$ and $\partial_{\gamma}\widetilde A^{0}_{\eta}=-\widetilde A^{0}_{\eta}\bigl(\partial_{\gamma}\widetilde C_{\gamma}\bigr)\widetilde A^{0}_{\eta}$ by the derivative-of-the-inverse identity applied to $\widetilde C_{\gamma}+\lambda I$, $\gamma\mapsto\widetilde k_{\gamma}(x)$ is $C^{1}$ with $\sup_{\gamma,x}\bigl\lVert\partial_{\gamma}\widetilde k_{\gamma}(x)\bigr\rVert_{\mathcal H_{\gamma_{0}}}\le L_{d}$ by Fact~\ref{fact:pairdiff}(i), and $\gamma\mapsto\widetilde C_{\gamma}$ is $C^{1}$ in Hilbert--Schmidt norm with $\bigl\lVert\partial_{\gamma}\widetilde C_{\gamma}\bigr\rVert_{\mathrm{HS}}\le2\,\bar k^{1/2}L_{d}$ by Fact~\ref{fact:pairdiff}(iii); the product rule on the transported pairing $g_{\eta}(x,y)=\bigl\langle\widetilde k_{\gamma}(x),\widetilde A^{0}_{\eta}\widetilde k_{\gamma}(y)\bigr\rangle_{\mathcal H_{\gamma_{0}}}$ and dominated convergence over $\mathbb P_{Z}$ for the composition give entrywise continuous derivatives of $g_{\eta}$, $(g_{\eta}\circ g_{\eta})$, and $r_{\eta}$. Since $\Upsilon$ is a convex rectangle, the mean value theorem along segments applied to the moduli of the statement yields the three derivative envelopes.

\emph{(iii)(b): the boundedness clause.}\; The proof of Lemma~\ref{lem:lipschitz_common}(a) is residual- and sample-size-agnostic: it consumes only the Lipschitz and moment inputs recorded there, at a generic sample size. At the selection sample the Pitman shift is $f_{m}=\delta_{m}h$ with $\sqrt m\,\delta_{m}\to a<\infty$, so the calibration $N\lVert g\rVert_{L^{2}}^{2}=O(1)$ of that proof holds with $N=m$, and every display of the proof transfers at sample size $m$, giving $\sup_{\gamma}\lVert\widetilde M_{m}(\gamma)\rVert_{\mathcal H_{\gamma_0}}=O_{p}(1)$.
\end{proof}

\begin{lemma}[Uniform bound for the anchor-centered remainder]\label{lem:remainder_cv}
Let Assumption~\ref{ass:A} hold and let $\mathcal R_{m}$ be the anchor-centered remainder of Lemma~\ref{lem:bilinear_cv}(ii). The bound is stated for the projected residuals $\boldsymbol\varepsilon_{\perp}^{0}=\hat{\boldsymbol\Pi}\boldsymbol\varepsilon_{0}$ of the null regime: their innovations are conditionally centered, $\mathbb E[\boldsymbol\varepsilon_{0}\mid\mathbf X]=\boldsymbol 0$, conditionally independent given $\mathbf X$ with $\mathbb E[\varepsilon_{0,a}^{2}\mid X_{a}]=\sigma^{2}(X_{a})$ and $\bar M_{2p}:=\sup_{a}\mathbb E[\varepsilon_{0,a}^{2p}\mid X_{a}]<\infty$ for the exponent $p\in(2,4]$ of Assumption~\ref{ass:G}(iii), (v); $\hat{\boldsymbol\Pi}$-sandwiched coefficient matrices are taken as conditionally deterministic under the convention of Remark~\ref{rem:cond_det}. Then
\[
    \sup_{\eta\in\Upsilon}\bigl|\mathcal R_{m}(\eta)\bigr|=O_{p}\bigl(m^{-1/2}\bigr)
    \qquad\text{and}\qquad
    \sup_{\eta\in\Upsilon}\bigl|\mathbb E\bigl[\mathcal R_{m}(\eta)\bigm|\mathbf X\bigr]\bigr|=O_{p}\bigl(m^{-1/2}\bigr).
\]
\end{lemma}

\begin{proof}
\emph{Design of the bound.}\; The target is stronger than a pointwise rate: Chebyshev's inequality applied to the second-moment bound below gives $\mathcal R_{m}(\eta)=O_{p}(m^{-1/2})$ at each fixed $\eta$, whereas the criterion analysis subtracts the remainder uniformly over the whole tuning set, and a pointwise rate says nothing about the supremum. The proof is organized around three observations. First, the summands of $\mathcal R_{m}$ and of its $\eta$-gradients are bounded and Lipschitz in $\eta$ with deterministic, design-free envelopes; every later moment computation consumes these envelopes, and the path-regularity hypothesis of the Sobolev inequality is literally a deterministic gradient bound. Second, the rate is a normalization fact: the coefficient matrix of the matrix form below has Frobenius norm $O(m^{-1/2})$ in expectation, and this single estimate funds every later bound. Third, the projection $\hat{\boldsymbol\Pi}$ makes the entries of $\boldsymbol\varepsilon_{\perp}^{0}$ conditionally dependent given $\mathbf X$, so computations for independent products are unavailable; the workhorse is an order-$p$ moment bound for the sandwiched quadratic form, obtained by decoupling and conditional Rosenthal's inequality at the exponent $p\in(2,4]$ of Assumption~\ref{ass:G}(v), and the anchor-multiset count isolates the combinatorics of the coefficient increments so that the derivative family reuses it verbatim. No stochastic covering or chaining enters: the only two-point computation is the anchor-multiset count for the coefficient matrices, and the uniformity over $\Upsilon$ is assembled from within-cell oscillations on a mesh fixed independently of $m$; the uniform gradient moment is converted into those oscillations by the Sobolev maximal inequality of Lemma~\ref{lem:sobolev} at the exponent $p\in(2,4]$, $p>d=2$, the mesh levels are governed by the pointwise variances, and the suprema over cells are assembled by H\"older's inequality. Each paragraph below records the later paragraphs that consume it.

\emph{The matrix form of $\mathcal R_{m}$.}\; Let $A_{\eta}:=(A_{\eta,jk})$ be the symmetric matrix with entries $A_{\eta,jk}:=\frac{1}{(m-1)^{2}}\sum_{i\ne j,k}r_{\eta}(X_{i};X_{j},X_{k})$ for $j\ne k$ and zero diagonal. Symmetry of $r_{\eta}$ in its last two slots makes the ordered double sum twice the unordered one, so
\[
    \mathcal R_{m}(\eta)=\boldsymbol\varepsilon_{\perp}^{\top}A_{\eta}\boldsymbol\varepsilon_{\perp}^{0}
    =\boldsymbol\varepsilon_{0}^{\top}B_{\eta}\boldsymbol\varepsilon_{0},
    \qquad
    B_{\eta}:=\hat{\boldsymbol\Pi}A_{\eta}\hat{\boldsymbol\Pi},
\]
the first identity reproducing the definition of $\mathcal R_{m}(\eta)$ and the second inserting $\boldsymbol\varepsilon_{\perp}^{0}=\hat{\boldsymbol\Pi}\boldsymbol\varepsilon_{0}$ and using the idempotence and symmetry of $\hat{\boldsymbol\Pi}$. The algebra of the reduction of Lemma~\ref{lem:bilinear_cv}(ii) is complete: the third-order sum $R_{m}$ is a second-order U-statistic scaled by $m(m-2)/(m-1)$, itself a quadratic form of the mean embedding against a deterministic operator, plus the anchor-centered remainder $\mathcal R_{m}$ in matrix form.

The probabilistic input is the null-regime innovation structure: conditional on $\mathbf X$, the innovations are independent and centered with $\mathbb E[\varepsilon_{0,a}^{2}\mid X_{a}]=\sigma^{2}(X_{a})$, and the fourth-moment level consumed below is carried by the Lyapunov interpolation of the $(2p)$-clause of Assumption~\ref{ass:G}(v) at $q=4$, $\bar M_{4}:=\bar M_{2p}^{2/p}\ge\sup_{a}\mathbb E[\varepsilon_{0,a}^{4}\mid X_{a}]$, a deterministic constant depending only on $(p,\bar M_{2p})$ (Assumption~\ref{ass:G}(iii), (v)).

\emph{Deterministic and moment envelopes.}\; From $g_{\eta}(x,x)\le\lambda^{-1}k_{\gamma}(x,x)\le\bar k\lambda_{L}^{-1}$ and Cauchy--Schwarz,
\[
    \lVert g_{\eta}\rVert_{\infty}\le\bar k\lambda_{L}^{-1},\qquad
    \lVert g_{\eta}\circ g_{\eta}\rVert_{\infty}\le\bar k^{2}\lambda_{L}^{-2},\qquad
    \lvert r_{\eta}\rvert\le2\bar k^{2}\lambda_{L}^{-2},\qquad
    \mathbb E[r_{\eta}^{2}]\le4\bar k^{4}\lambda_{L}^{-4},
\]
and, by the anchor-modulus constant of part (iii)(a) of Lemma~\ref{lem:bilinear_cv}, $\lvert r_{\eta}(x;y,z)-r_{\eta'}(x;y,z)\rvert\le C''\lVert\eta-\eta'\rVert$ on $\mathcal X^{3}$, $C'':=2\bar k\lambda_{L}^{-1}C_{g}+C'_{g}$. For the centered innovation squares, at the exponent $p$ of Assumption~\ref{ass:G}(v),
\[
    \mathbb E\bigl[\bigl|\varepsilon_{0,a}^{2}-\sigma^{2}(X_{a})\bigr|^{p}\mid X_{a}\bigr]\le2^{p-1}\bigl(\bar M_{2p}+\sigma_{\max}^{2p}\bigr),
\]
by $|a-b|^{p}\le2^{p-1}(a^{p}+b^{p})$ and Jensen's inequality.

\emph{The size of the coefficient matrices.}\; Since $r_{\eta}$ is $E_{X}$-centered in its anchor argument and the $m-2$ anchors are independent,
\[
\begin{aligned}
    \mathbb E\lVert A_{\eta}\rVert_{F}^{2}
    &=\frac{1}{(m-1)^{4}}\sum_{j\ne k}\sum_{i,i'\ne j,k}\mathbb E[r_{\eta}(X_{i};X_{j},X_{k})r_{\eta}(X_{i'};X_{j},X_{k})]\\
    &=\frac{m(m-1)(m-2)}{(m-1)^{4}}\,\mathbb E[r_{\eta}^{2}]
    \le\frac{m(m-1)(m-2)}{(m-1)^{4}}\,4\bar k^{4}\lambda_{L}^{-4}
    =O(m^{-1}),
\end{aligned}
\]
the middle equality retaining only the diagonal $i=i'$ of the anchor double sum. Markov's inequality gives $\lVert A_{\eta}\rVert_{F}=O_{p}(m^{-1/2})$, and the contraction $\lVert\hat{\boldsymbol\Pi}\rVert_{\mathrm{op}}\le1$ gives $\lVert B_{\eta}\rVert_{F}\le\lVert A_{\eta}\rVert_{F}$.

\emph{Means and variances.}\; Let $D:=\operatorname{diag}\bigl(\sigma^{2}(X_{1}),\dots,\sigma^{2}(X_{m})\bigr)$ and $C:=\hat{\boldsymbol\Pi}D\hat{\boldsymbol\Pi}$. Under the null premise $\mathbb E[\boldsymbol\varepsilon_{0}\mid\mathbf X]=\boldsymbol 0$,
\[
    \mathbb E[\mathcal R_{m}(\eta)\mid\mathbf X]
    =\sum_{a}(B_{\eta})_{aa}\sigma^{2}(X_{a})
    =\operatorname{tr}(B_{\eta}D)
    =\operatorname{tr}(A_{\eta}C^{\mathrm{od}}),
\]
by cyclicity of the trace, the zero diagonal of $A_{\eta}$ annihilating the diagonal of $C$; Cauchy--Schwarz gives $|\mathbb E[\mathcal R_{m}\mid\mathbf X]|\le\lVert A_{\eta}\rVert_{F}\lVert C^{\mathrm{od}}\rVert_{F}$, with
\[
    \lVert C^{\mathrm{od}}\rVert_{F}^{2}=\sum_{j\ne k}c_{jk}^{2}\le3\bigl(2\sigma_{\max}^{4}d_{\theta}+\sigma_{\max}^{4}d_{\theta}\bigr)=9\sigma_{\max}^{4}d_{\theta}
\]
deterministically, through the expansion recorded in Section~\ref{sec:uniform_limit} of the main text: with $H:=\hat{\mathbb G}M\hat{\mathbb G}^{\top}$ the rank-$d_{\theta}$ hat projection of that expansion, $c_{jk}=-\sigma^{2}(X_{j})H_{jk}-\sigma^{2}(X_{k})H_{jk}+(HDH)_{jk}$, $\sum_{j\ne k}H_{jk}^{2}\le\operatorname{tr}(H^{2})=\operatorname{tr}(H)=d_{\theta}$, and $\lVert HDH\rVert_{F}\le\sigma_{\max}^{2}\lVert H\rVert_{F}=\sigma_{\max}^{2}\sqrt{d_{\theta}}$ ($\lVert D\rVert_{\mathrm{op}}\le\sigma_{\max}^{2}$, $\lVert H\rVert_{\mathrm{op}}\le1$, $H$ idempotent of rank $d_{\theta}$). Hence $|\mathbb E[\mathcal R_{m}\mid\mathbf X]|=O_{p}(m^{-1/2})$. For the variance, Lemma~\ref{lem:pairvar} applied to the conditionally independent, centered innovations $\varepsilon_{0,a}$ gives
\[
\begin{aligned}
    \operatorname{Var}\bigl(\mathcal R_{m}(\eta)\mid\mathbf X\bigr)
    &=2\bigl\lVert D^{1/2}B_{\eta}D^{1/2}\bigr\rVert_{F}^{2}
     +\sum_{a=1}^{m}(B_{\eta})_{aa}^{2}\,\kappa_{4}(a)
    \le\bigl(2\sigma_{\max}^{4}+2\bar M_{4}\bigr)\lVert A_{\eta}\rVert_{F}^{2},
\end{aligned}
\]
with $\kappa_{4}(a):=\mathbb E[\varepsilon_{0,a}^{4}\mid X_{a}]-3\sigma^{4}(X_{a})$, where $|\kappa_{4}(a)|\le2\bar M_{4}$: $\kappa_{4}(a)\le \bar M_{4}$ by the interpolated fourth moment of the matrix-form paragraph, $\kappa_{4}(a)\ge-2\sigma^{4}(X_{a})\ge-2\sigma_{\max}^{4}$ by Jensen's inequality $\mathbb E[\varepsilon_{0,a}^{4}\mid X_{a}]\ge\sigma^{4}(X_{a})$, and $\sigma_{\max}^{4}\le \bar M_{4}$ by the same Jensen bound; the remaining steps use $\lVert D^{1/2}B_{\eta}D^{1/2}\rVert_{F}\le\sigma_{\max}^{2}\lVert B_{\eta}\rVert_{F}$, $\sum_{a}(B_{\eta})_{aa}^{2}\le\lVert B_{\eta}\rVert_{F}^{2}$, and the contraction of the size paragraph. Collecting mean and variance,
\[
    \mathbb E\bigl[\mathcal R_{m}(\eta)^{2}\mid\mathbf X\bigr]
    \le\bigl(2\sigma_{\max}^{4}+2\bar M_{4}+\lVert C^{\mathrm{od}}\rVert_{F}^{2}\bigr)\lVert A_{\eta}\rVert_{F}^{2}
    =O_{p}(m^{-1}).
\]

\emph{Order-$p$ moments of projected-residual quadratic forms.}\
This paragraph isolates the order-$p$ moment bound for a generic coefficient matrix, at the exponent $p\in(2,4]$ of Assumption~\ref{ass:G}(v). Let $A$ be a design-measurable symmetric $m\times m$ matrix with zero diagonal and set $\widetilde A:=\hat{\boldsymbol\Pi}A\hat{\boldsymbol\Pi}$, which is again symmetric and, being $\hat\theta$-dependent through $\hat{\boldsymbol\Pi}$, conditionally deterministic under the convention of Remark~\ref{rem:cond_det}. Lemma~\ref{lem:quad_moment}, applied conditionally on $\mathcal G=\sigma(\mathbf X)$ to the conditionally independent, centered innovations and to the coefficient matrix $\widetilde A$, gives
\[
    \mathbb E\Bigl[\bigl|\boldsymbol\varepsilon_{0}^{\top}\widetilde A\boldsymbol\varepsilon_{0}-\operatorname{tr}(\widetilde A\,D)\bigr|^{p}\mid\mathbf X\Bigr]
    \le C_{p}(\sigma_{\max},\bar M_{2p})\,\lVert\widetilde A\rVert_{F}^{p},
\]
the bound \eqref{eq:quad4} at the selection sample; the contraction $\lVert\hat{\boldsymbol\Pi}A\hat{\boldsymbol\Pi}\rVert_{F}\le\lVert A\rVert_{F}$ of the idempotent projection converts it to the coefficient matrix $A$.
The conditional mean is absorbed by the zero diagonal of $A$: with $C:=\hat{\boldsymbol\Pi}D\hat{\boldsymbol\Pi}$ of the means-and-variances paragraph and $C_{0}:=9\sigma_{\max}^{4}d_{\theta}$,
\[
    \bigl|\operatorname{tr}(\widetilde A\,D)\bigr|
    =\bigl|\operatorname{tr}(A\,C)\bigr|
    =\bigl|\operatorname{tr}(A\,C^{\mathrm{od}})\bigr|
    \le\lVert A\rVert_{F}\,\lVert C^{\mathrm{od}}\rVert_{F}
    \le C_{0}^{1/2}\lVert A\rVert_{F},
\]
the first equality by cyclicity of the trace, the second discarding the diagonal of $C$ against the zero diagonal of $A$, the inequality Cauchy--Schwarz with the deterministic envelope $\lVert C^{\mathrm{od}}\rVert_{F}^{2}\le C_{0}$ of the means-and-variances paragraph. Assembling the centered and mean parts by $(x+y)^{p}\le2^{p-1}(x^{p}+y^{p})$ and the tower property,
\[
    \mathbb E\Bigl[\bigl|\boldsymbol\varepsilon_{0}^{\top}\hat{\boldsymbol\Pi}A\hat{\boldsymbol\Pi}\boldsymbol\varepsilon_{0}\bigr|^{p}\;\Big|\;\mathbf X\Bigr]
    \le C_{p}(\sigma_{\max},\bar M_{2p},C_{0})\,\lVert A\rVert_{F}^{p},
\]
the constant depending only on $(p,\sigma_{\max},\bar M_{2p},d_{\theta})$; this plain form is the one cited below. Alongside it, the anchor-multiset count for coefficient increments is recorded, whose pattern the process-level paragraph reuses for the derivative family.

\emph{The anchor-multiset count.}\;
Writing $\Delta A_{jk}=\frac{1}{(m-1)^{2}}\sum_{i\ne j,k}\delta r_{\eta}(X_{i};X_{j},X_{k})$ with $\delta r_{\eta}:=r_{\eta}-r_{\eta'}$, which is $\mathbb E_{X}$-centered and bounded by $C''\lVert\eta-\eta'\rVert$ on $\mathcal X^{3}$, and abbreviating $\delta r_{i}:=\delta r_{\eta}(X_{i};X_{j},X_{k})$ at fixed $(j,k)$,
\[
    \mathbb E\bigl[\Delta A_{jk}^{2}\,\Delta A_{j'k'}^{2}\bigr]
    =\frac{1}{(m-1)^{8}}\sum_{\substack{i_{1},i_{2}\in A(j,k)\\ i_{3},i_{4}\in A(j',k')}}\mathbb E\bigl[\delta r_{i_{1}}\,\delta r_{i_{2}}\,\delta r_{i_{3}}\,\delta r_{i_{4}}\bigr],
    \qquad A(j,k):=\{1,\dots,m\}\setminus\{j,k\}.
\]
The factors $\delta r_{i}$ are conditionally independent across distinct anchors outside the conditioning set $\{j,k,j',k'\}$ and $\mathbb E_{X}$-centered, so a four-factor product vanishes unless every anchor index outside $\{j,k,j',k'\}$ appears at least twice among $\{i_{1},i_{2},i_{3},i_{4}\}$; anchors falling inside $\{j,k,j',k'\}$ are conditioning-measurable and bounded by $C''\lVert\eta-\eta'\rVert$, and the configurations containing them add only $O(m)$ bounded summands per entry pair, absorbed in the constant below. The surviving multisets are the all-equal one, $i_{1}=i_{2}=i_{3}=i_{4}$, which requires $i\in A(j,k)\cap A(j',k')$, about $m$ choices, each summand at most $C''^{4}\lVert\eta-\eta'\rVert^{4}$, and the two-and-two ones, dominantly $i_{1}=i_{2}\in A(j,k)$ with $i_{3}=i_{4}\in A(j',k')$, about $m^{2}$ choices, each contributing $\mathbb E[\delta r^{2}]^{2}\le C''^{4}\lVert\eta-\eta'\rVert^{4}$, the cross-paired two-and-two configurations having about $m$ choices; hence
\[
    \mathbb E\bigl[\Delta A_{jk}^{2}\,\Delta A_{j'k'}^{2}\bigr]\le\frac{C''^{4}\lVert\eta-\eta'\rVert^{4}\,C\,(m^{2}+m)}{(m-1)^{8}}\le C\,\frac{\lVert\eta-\eta'\rVert^{4}}{m^{6}},
\]
and summing over the $m^{2}(m-1)^{2}$ index pairs,
\[
    \mathbb E\lVert\Delta A\rVert_{F}^{4}\le C\,\frac{\lVert\eta-\eta'\rVert^{4}}{m^{2}}.
\]
With $r_{\eta}$ itself in place of $\delta r_{\eta}$, likewise $\mathbb E_{X}$-centered and bounded by $K_{r}:=2\bar k^{2}\lambda_{L}^{-2}$, the same count gives $\mathbb E\lVert A_{\eta}\rVert_{F}^{4}\le C\,K_{r}^{4}/m^{2}$, consumed by the conditional-mean paragraph below.

\emph{The process level: gradient moments, fixed mesh, Sobolev--Morrey.}\;\;
Normalize $S_{m}(\eta):=\sqrt m\,\mathcal R_{m}(\eta)$. No stochastic covering or chaining enters: the only two-point computation is the anchor-multiset count for the coefficient matrices, and the uniformity over $\Upsilon$ is assembled from within-cell oscillations on a mesh fixed independently of $m$.

\emph{Path regularity and gradient envelopes.}\;
By the derivative-envelope clause of Lemma~\ref{lem:bilinear_cv}(iii)(a), the kernels $g_{\eta}$, $(g_{\eta}\circ g_{\eta})$, and $r_{\eta}$ are entrywise $C^{1}$ in $\eta$ on $\Upsilon$ with $\sup_{\eta,x,y}\lvert\partial_{\eta_{l}}g_{\eta}\rvert\le C_{g}$, $\sup_{\eta,x,y}\lvert\partial_{\eta_{l}}(g_{\eta}\circ g_{\eta})\rvert\le C'_{g}$, and $\sup_{\eta,x,y,z}\lvert\partial_{\eta_{l}}r_{\eta}\rvert\le C''$ for $l\in\{\gamma,\lambda\}$, all three bounds deterministic and design-free; centering is inherited, $\mathbb E_{X}[\partial_{\eta_{l}}r_{\eta}(X;y,z)]=\partial_{\eta_{l}}\,\mathbb E_{X}[r_{\eta}(X;y,z)]=0$, the interchange by dominated convergence with dominator $C''$. Consequently the matrix field $A_{\eta}$ of the matrix-form paragraph is entrywise $C^{1}$ with
\[
    \partial_{\eta_{l}}A_{\eta,jk}
    =\frac{1}{(m-1)^{2}}\sum_{i\ne j,k}\partial_{\eta_{l}}r_{\eta}(X_{i};X_{j},X_{k})
    \quad(j\ne k),
    \qquad
    \partial_{\eta_{l}}A_{\eta,jj}=0,
\]
so $S_{m}$ has $C^{1}$ paths and $\partial_{\eta_{l}}S_{m}(\eta)=\sqrt m\,\boldsymbol\varepsilon_{0}^{\top}\hat{\boldsymbol\Pi}(\partial_{\eta_{l}}A_{\eta})\hat{\boldsymbol\Pi}\,\boldsymbol\varepsilon_{0}$.

\emph{Gradient moments.}\;
The derivative family $\partial_{\eta_{l}}r_{\eta}$ is $\mathbb E_{X}$-centered and bounded by $C''$, the constant of the envelope paragraph. The anchor-multiset count above, applied verbatim with $\partial_{\eta_{l}}r_{\eta}$ in place of $\delta r_{\eta}$, gives
\[
    \mathbb E\bigl\lVert\partial_{\eta_{l}}A_{\eta}\rVert_{F}^{4}\le C\,\frac{C''^{4}}{m^{2}},
\]
and the size count of the size paragraph, applied with the same substitution, gives
\[
    \mathbb E\bigl\lVert\partial_{\eta_{l}}A_{\eta}\bigr\rVert_{F}^{2}
    =\frac{m(m-1)(m-2)}{(m-1)^{4}}\,\mathbb E\bigl[(\partial_{\eta_{l}}r_{\eta})^{2}\bigr]
    \le\frac{2\,C''^{2}}{m},
\]
using $\mathbb E[(\partial_{\eta_{l}}r_{\eta})^{2}]\le C''^{2}$, $m(m-2)\le(m-1)^{2}$, and $(m-1)^{-1}\le2\,m^{-1}$ for $m\ge2$. With $Z:=\lVert\partial_{\eta_{l}}A_{\eta}\rVert_{F}^{2}$ and $q:=p/2\in(1,2]$, Lyapunov's interpolation of moments, $\mathbb E Z^{q}\le\bigl(\mathbb E Z\bigr)^{2-q}\bigl(\mathbb E Z^{2}\bigr)^{q-1}$ for $1\le q\le2$, gives
\[
    \mathbb E\bigl\lVert\partial_{\eta_{l}}A_{\eta}\bigr\rVert_{F}^{p}
    =\mathbb E\bigl[Z^{q}\bigr]
    \le\Bigl(\frac{2\,C''^{2}}{m}\Bigr)^{2-q}\Bigl(C\,\frac{C''^{4}}{m^{2}}\Bigr)^{q-1}
    \le C_{p}\,\frac{C''^{p}}{m^{p/2}},
\]
the exponents of $C''$ and $m$ combining to $2(2-q)+4(q-1)=2q=p$ and $-(2-q)-2(q-1)=-q=-p/2$. The plain form of the order-$p$ paragraph, applied with $A=\partial_{\eta_{l}}A_{\eta}$ and the tower property, gives
\[
    \mathbb E\bigl|\partial_{\eta_{l}}S_{m}(\eta)\bigr|^{p}
    =m^{p/2}\,\mathbb E\Bigl[\bigl|\boldsymbol\varepsilon_{0}^{\top}\hat{\boldsymbol\Pi}(\partial_{\eta_{l}}A_{\eta})\hat{\boldsymbol\Pi}\,\boldsymbol\varepsilon_{0}\bigr|^{p}\Bigr]
    \le C_{p}\,m^{p/2}\,\mathbb E\bigl\lVert\partial_{\eta_{l}}A_{\eta}\bigr\rVert_{F}^{p}
    \le C_{p}\,C''^{p}.
\]
Since $\lVert\nabla_{\eta}S_{m}\rVert^{p}=\bigl(|\partial_{\gamma}S_{m}|^{2}+|\partial_{\lambda}S_{m}|^{2}\bigr)^{p/2}\le2^{p/2-1}\bigl(|\partial_{\gamma}S_{m}|^{p}+|\partial_{\lambda}S_{m}|^{p}\bigr)$ by $(a+b)^{q}\le2^{q-1}(a^{q}+b^{q})$ at $q=p/2$, the Sobolev input of Lemma~\ref{lem:sobolev} holds with $p\in(2,4]$, $p>d=2$:
\[
    \sup_{\eta\in\Upsilon}\mathbb E\bigl\lVert\nabla_{\eta}S_{m}(\eta)\bigr\rVert^{p}\le N_{p},
    \qquad
    N_{p}:=2^{p/2}\,C_{p}\,C''^{p},
\]
a constant independent of $m$.

\emph{The fixed mesh and the assembly.}\;
Partition $\Upsilon=[\gamma_{L},\gamma_{U}]\times[\lambda_{L},\lambda_{U}]$ by a grid of mesh width $\delta>0$, fixed once and independently of $m$: at most $J:=\bigl\lceil(\gamma_{U}-\gamma_{L})/\delta\bigr\rceil\bigl\lceil(\lambda_{U}-\lambda_{L})/\delta\bigr\rceil$ cells $Q$, each contained in a square of side $\delta$ contained in $\Upsilon$, and at most $(J+1)^{2}$ grid points, all $m$-free. The grid values are controlled unconditionally, because $\lVert C^{\mathrm{od}}\rVert_{F}^{2}\le C_{0}$ deterministically turns the means-and-variances bound into
\[
\begin{aligned}
    \mathbb E\bigl[S_{m}(\eta)^{2}\bigr]
    &=m\,\mathbb E\bigl[\mathcal R_{m}(\eta)^{2}\bigr]
    \le m\bigl(2\sigma_{\max}^{4}+2\bar M_{4}+C_{0}\bigr)\,\mathbb E\lVert A_{\eta}\rVert_{F}^{2}\\
    &\le m\bigl(2\sigma_{\max}^{4}+2\bar M_{4}+C_{0}\bigr)\frac{4\bar k^{4}\lambda_{L}^{-4}}{m-1}
    \le\bigl(2\sigma_{\max}^{4}+2\bar M_{4}+C_{0}\bigr)\,8\bar k^{4}\lambda_{L}^{-4}
    =:c_{*},
\end{aligned}
\]
using the size count, $\mathbb E\lVert A_{\eta}\rVert_{F}^{2}\le\frac{m(m-1)(m-2)}{(m-1)^{4}}\,4\bar k^{4}\lambda_{L}^{-4}\le\frac{4\bar k^{4}\lambda_{L}^{-4}}{m-1}$, the last step by $m(m-2)\le(m-1)^{2}$, and $\frac{m}{m-1}\le2$ for $m\ge2$. Markov's inequality and the union bound over the fixed grid give $\max_{\mathrm{grid}}|S_{m}|=O_{p}(1)$. On each cell $Q$, Lemma~\ref{lem:sobolev} with field $S_{m}$, dimension $d=2$, exponent $p\in(2,4]$, $p>d$, and $N_{p}$ yields
\[
    \mathbb E\Bigl[\sup_{t,s\in Q}\bigl|S_{m}(t)-S_{m}(s)\bigr|^{p}\Bigr]\le C_{2,p}\,N_{p}\,\delta^{p},
    \qquad
    \mathbb P\Bigl(\sup_{t,s\in Q}\bigl|S_{m}(t)-S_{m}(s)\bigr|\ge u\Bigr)\le C_{2,p}\,N_{p}\,\delta^{p}\,u^{-p},
\]
both bounds uniform in $(m,Q)$. The assembly over the fixed family of cells,
\[
    \sup_{\eta\in\Upsilon}\bigl|S_{m}(\eta)\bigr|
    \le\max_{\mathrm{grid}}\bigl|S_{m}\bigr|+\sum_{Q}\sup_{t,s\in Q}\bigl|S_{m}(t)-S_{m}(s)\bigr|,
\]
is a sum of $J+1$ terms, each $O_{p}(1)$ uniformly in $m$ by the displays above and Markov's inequality, the cell sum being governed through Hölder's inequality $(x_{1}+\dots+x_{J})^{p}\le J^{p-1}\sum_{Q}x_{Q}^{p}$; hence $\sup_{\eta\in\Upsilon}|S_{m}(\eta)|=O_{p}(1)$, that is,
\[
    \sup_{\eta\in\Upsilon}\bigl|\mathcal R_{m}(\eta)\bigr|=O_{p}\bigl(m^{-1/2}\bigr),
\]
the first conclusion of the lemma; it remains to record the uniform bound on the conditional mean.

\emph{The uniform conditional mean.}\;
Let $\mu_{m}(\eta):=\mathbb E\bigl[\mathcal R_{m}(\eta)\bigm|\mathbf X\bigr]$. By the means-and-variances paragraph, $\mu_{m}(\eta)=\operatorname{tr}\bigl(A_{\eta}C^{\mathrm{od}}\bigr)$ with $\lVert C^{\mathrm{od}}\rVert_{F}^{2}\le C_{0}$ deterministic, so Cauchy--Schwarz for the Frobenius inner product gives the pointwise bounds
\[
    \bigl|\mu_{m}(\eta)\bigr|\le C_{0}^{1/2}\,\lVert A_{\eta}\rVert_{F},
    \qquad
    \bigl|\partial_{\eta_{l}}\mu_{m}(\eta)\bigr|=\bigl|\operatorname{tr}\bigl((\partial_{\eta_{l}}A_{\eta})\,C^{\mathrm{od}}\bigr)\bigr|\le C_{0}^{1/2}\,\bigl\lVert\partial_{\eta_{l}}A_{\eta}\bigr\rVert_{F},
\]
the differentiation passing under the trace because $C^{\mathrm{od}}$ is $\eta$-free and $A_{\eta}$ is entrywise $C^{1}$ by the path-regularity paragraph. The moment chain of the coefficient field combines the size paragraph, the multiset count of the preceding paragraph with $r_{\eta}$ in place of $\delta r_{\eta}$, and Lyapunov's interpolation of moments at $q=p/2$, $\mathbb E Z^{q}\le(\mathbb E Z)^{2-q}(\mathbb E Z^{2})^{q-1}$ for $1\le q\le2$ applied to $Z:=\lVert A_{\eta}\rVert_{F}^{2}$:
\[
\begin{aligned}
    \mathbb E\lVert A_{\eta}\rVert_{F}^{2}\le\frac{4\bar k^{4}\lambda_{L}^{-4}}{m-1},
    &\qquad
    \mathbb E\lVert A_{\eta}\rVert_{F}^{4}\le C\,\frac{K_{r}^{4}}{m^{2}},\\
    \sup_{\eta}\mathbb E\lVert A_{\eta}\rVert_{F}^{p}
    &\le\Bigl(\frac{4\bar k^{4}\lambda_{L}^{-4}}{m-1}\Bigr)^{2-q}\Bigl(C\,\frac{K_{r}^{4}}{m^{2}}\Bigr)^{q-1}
    \le C_{p}\,m^{-p/2},
\end{aligned}
\]
with $K_{r}:=2\bar k^{2}\lambda_{L}^{-2}$, the exponents of $m$ combining as $-(2-q)-2(q-1)=-q=-p/2$ and the constants $2^{2-q}\bigl(4\bar k^{4}\lambda_{L}^{-4}\bigr)^{2-q}\bigl(CK_{r}^{4}\bigr)^{q-1}$ into $C_{p}$. Since $\lVert\nabla_{\eta}\mu_{m}\rVert^{2}=|\partial_{\gamma}\mu_{m}|^{2}+|\partial_{\lambda}\mu_{m}|^{2}$, the norm monotonicity $(a^{2}+b^{2})^{q}\le2^{q-1}(a^{q}+b^{q})$ at $q=p/2$, the pointwise gradient bound, and the gradient-moments paragraph, $\sup_{\eta}\mathbb E\lVert\partial_{\eta_{l}}A_{\eta}\rVert_{F}^{p}\le C_{p}C''^{p}m^{-p/2}$, give the Sobolev input
\[
    \sup_{\eta\in\Upsilon}\mathbb E\bigl\lVert\nabla_{\eta}\mu_{m}(\eta)\bigr\rVert^{p}\le2^{p/2}\,C_{0}^{p/2}\,C_{p}\,C''^{p}\,m^{-p/2}=:N'_{p}.
\]
The assembly mirrors the fixed-mesh paragraph on the same mesh: Markov's inequality with the pointwise bound at the $(J+1)^{2}$ grid points and $\sup_{\eta}\mathbb E\lVert A_{\eta}\rVert_{F}^{p}\le C_{p}m^{-p/2}$ gives $\max_{\mathrm{grid}}|\mu_{m}|=O_{p}(m^{-1/2})$; on each cell $Q$, Lemma~\ref{lem:sobolev} with field $\mu_{m}$, dimension $d=2$, exponent $p\in(2,4]$, $p>d$, and $N'_{p}$ gives $\mathbb E\bigl[\sup_{t,s\in Q}\bigl|\mu_{m}(t)-\mu_{m}(s)\bigr|^{p}\bigr]\le C_{2,p}\,N'_{p}\,\delta^{p}$; and H\"older's inequality over the fixed family of $J$ cells, $(x_{1}+\dots+x_{J})^{p}\le J^{p-1}\sum_{Q}x_{Q}^{p}$, assembles
\[
    \sup_{\eta\in\Upsilon}\bigl|\mu_{m}(\eta)\bigr|=O_{p}\bigl(m^{-1/2}\bigr),
\]
the second conclusion of the lemma. This completes the proof.

Under the local alternative, $\boldsymbol\varepsilon_{0}=\boldsymbol v+\boldsymbol f_{m}$ with the signal at the local scale $\lVert\hat{\boldsymbol\Pi}\boldsymbol f_{m}\rVert^{2}=O_{p}(1)$, the signal part of $\mathcal R_{m}$ contributes $\boldsymbol f_{m}^{\top}\hat{\boldsymbol\Pi}A_{\eta}\hat{\boldsymbol\Pi}\boldsymbol f_{m}$ and $2\boldsymbol f_{m}^{\top}\hat{\boldsymbol\Pi}A_{\eta}\hat{\boldsymbol\Pi}\boldsymbol v$, bounded by $\lVert A_{\eta}\rVert_{F}\lVert\hat{\boldsymbol\Pi}\boldsymbol f_{m}\rVert^{2}$ and, in conditional second moment, by $4\sigma_{\max}^{2}\lVert A_{\eta}\rVert_{F}^{2}\lVert\hat{\boldsymbol\Pi}\boldsymbol f_{m}\rVert^{2}$: with $c:=\hat{\boldsymbol\Pi}A_{\eta}\hat{\boldsymbol\Pi}\boldsymbol f_{m}$, $\operatorname{Var}\bigl(2\boldsymbol f_{m}^{\top}\hat{\boldsymbol\Pi}A_{\eta}\hat{\boldsymbol\Pi}\boldsymbol v\mid\mathbf X\bigr)=4\,c^{\top}Dc\le4\sigma_{\max}^{2}\lVert c\rVert^{2}\le4\sigma_{\max}^{2}\lVert A_{\eta}\rVert_{F}^{2}\lVert\hat{\boldsymbol\Pi}\boldsymbol f_{m}\rVert^{2}$, the conditional covariance of $\boldsymbol v$ being $D$, with $\lVert D\rVert_{\mathrm{op}}\le\sigma_{\max}^{2}$ and $\lVert\hat{\boldsymbol\Pi}A_{\eta}\hat{\boldsymbol\Pi}\rVert_{\mathrm{op}}\le\lVert A_{\eta}\rVert_{F}$. Both bounds are event-free: no truncation event enters, the deterministic constants of the null-regime displays being inflated by the design-measurable factor $\lVert\hat{\boldsymbol\Pi}\boldsymbol f_{m}\rVert^{2}=O_{p}(1)$. The uniform conclusion of the lemma is not extended to this regime; the gradient-moment input of the process level is noise-only by construction.

\end{proof}

\begin{lemma}[Chaos identification for the criterion limits]\label{lem:chaos_cv}
Let $J\ge1$, let $\eta_{1}=(\gamma_{1},\lambda_{1}),\dots,\eta_{J}=(\gamma_{J},\lambda_{J})\in\Upsilon$, and let $G_{\perp,\gamma_{1}},\dots,G_{\perp,\gamma_{J}}$ be the projected Gaussian elements of Corollary~\ref{cor:proj_null}; for $\eta=(\gamma,\lambda)\in\Upsilon$, write
\[
\begin{aligned}
    \mathcal C^{\perp}(\eta)&:=\bigl\langle G_{\perp,\gamma},A^{0}_{\eta}G_{\perp,\gamma}\bigr\rangle_{\mathcal H_{\gamma}}
      -\mathbb E\bigl\langle G_{\perp,\gamma},A^{0}_{\eta}G_{\perp,\gamma}\bigr\rangle_{\mathcal H_{\gamma}},\\
    \mathcal D^{\perp}(\eta)&:=\bigl\langle G_{\perp,\gamma},B^{0}_{\eta}G_{\perp,\gamma}\bigr\rangle_{\mathcal H_{\gamma}}
      -\mathbb E\bigl\langle G_{\perp,\gamma},B^{0}_{\eta}G_{\perp,\gamma}\bigr\rangle_{\mathcal H_{\gamma}},
\end{aligned}
\]
with $A^{0}_{\eta},B^{0}_{\eta}$ as in Lemma~\ref{lem:bilinear_cv}. Then $\bigl(\mathcal C^{\perp}(\eta_l),\mathcal D^{\perp}(\eta_l)\bigr)_{l\le J}$ is the image of a single isonormal process under the quadratic forms with kernels $(K^{\mathrm{sym},\perp}_{\eta},K^{D,\perp}_{\eta})$ (jointly in $l$), a centered second-order Gaussian chaos with, for all $\eta,\eta'\in\Upsilon$, covariance entries
$\operatorname{Cov}\bigl(\mathcal C^{\perp}(\eta),\mathcal C^{\perp}(\eta')\bigr)
=2\bigl\langle K^{\mathrm{sym},\perp}_{\eta},K^{\mathrm{sym},\perp}_{\eta'}\bigr\rangle_{\mathrm{HS}}$,
$\operatorname{Cov}\bigl(\mathcal C^{\perp}(\eta),\mathcal D^{\perp}(\eta')\bigr)
=2\bigl\langle K^{\mathrm{sym},\perp}_{\eta},K^{D,\perp}_{\eta'}\bigr\rangle_{\mathrm{HS}}$,
$\operatorname{Cov}\bigl(\mathcal D^{\perp}(\eta),\mathcal D^{\perp}(\eta')\bigr)
=2\bigl\langle K^{D,\perp}_{\eta},K^{D,\perp}_{\eta'}\bigr\rangle_{\mathrm{HS}}$,
with $K^{\mathrm{sym},\perp}_\eta(x,x')=\sigma(x)(\Pi g_\eta\Pi)(x,x')\sigma(x')$ and
$K^{D,\perp}_\eta(x,x')=\sigma(x)\bigl(\Pi(g\circ g)_\eta\Pi\bigr)(x,x')\sigma(x')$, exactly as in Theorem~\ref{th:cv_uniform}.
\end{lemma}

\begin{proof}
The proof proceeds in three paragraphs: it identifies the two centered pairings $\mathcal C^{\perp},\mathcal D^{\perp}$ as quadratic forms of a single isonormal process through the operator families of Lemma~\ref{lem:bilinear_cv}, identifies the integral kernels of those quadratic forms, and reads off all covariances from the Wick identity through the Hilbert--Schmidt inner products of the kernels.

\medskip\noindent\emph{The isonormal coupling.}\;
Work in the white-noise coupling of Appendix~\ref{app:limit_continuity}: the isonormal process $\zeta$ there is a centered Gaussian family $\{\zeta(f):f\in L^{2}(\mathbb P_{X})\}$ with $\mathbb E[\zeta(f)\,\zeta(g)]=\langle f,g\rangle_{L^{2}(\mathbb P_{X})}$ for all $f,g$, and the limit objects of the present lemma are defined in it as evaluations of the single noise $\zeta$: the directional coordinates $Z^{\perp}_{i}(\gamma)=\bigl\langle\zeta,\boldsymbol M_{\sigma}\boldsymbol\Pi\phi_{i}(\cdot;\gamma)\bigr\rangle_{L^{2}(\mathbb P_{X})}$ of the projected Gaussian elements, jointly over all eigen-indices $i\ge1$ and all bandwidths $\gamma\in[\gamma_{L},\gamma_{U}]$, and the score limit $W=\bigl\langle\zeta,\boldsymbol M_{\sigma}g_{\theta_{0}}\bigr\rangle_{L^{2}(\mathbb P_{X})}$; the cross-$\gamma$ covariances are the $L^{2}(\mathbb P_{X})$ inner products of the feature family $\{\boldsymbol M_{\sigma}\boldsymbol\Pi\phi_{i}(\cdot;\gamma)\}_{i,\gamma}\cup\{\boldsymbol M_{\sigma}g_{\theta_{0}}\}$, as recorded there, reducing to $\Sigma^{\perp}_{ij}(\gamma)$ at equal parameters. The argument is limit-object-only: no finite-sample central limit theorem enters, the joint Gaussianity of any finite subcollection being inherited from the linearity of the evaluations of $\zeta$. On the selection sample the same coupling underlies the convergence arguments: the embedding paragraph of part~(a) of the proof of Theorem~\ref{th:cv_uniform} records the joint convergence of the scaled embeddings with the score limit, so the statistics converging to the present pairings are driven, in the limit, by that single coupling.

\medskip\noindent\emph{Functions of a single isonormal process.}\;
In the eigenbasis of $C_{\gamma}$,
\[
    \bigl\langle G_{\perp,\gamma},A^{0}_{\eta}G_{\perp,\gamma}\bigr\rangle_{\mathcal H_{\gamma}}
    =\sum_{i\ge1}\frac{G_{i}^{2}}{\mu_{i}(\gamma)+\lambda}
    =\sum_{i\ge1}\kappa_{i}(\eta)\,\bigl(Z^{\perp}_{i}(\gamma)\bigr)^{2},
    \qquad
    \kappa_{i}(\eta):=\frac{\mu_{i}(\gamma)}{\mu_{i}(\gamma)+\lambda},
\]
with $G_{i}:=\langle G_{\perp,\gamma},e_{i}(\gamma)\rangle_{\mathcal H_{\gamma}}=\mu_{i}(\gamma)^{1/2}Z^{\perp}_{i}(\gamma)$. Substituting the coupling and writing $\psi_{i}:=\boldsymbol M_{\sigma}\boldsymbol\Pi\phi_{i}(\cdot;\gamma)$,
\[
    \sum_{i\ge1}\kappa_{i}(\eta)\,\langle\zeta,\psi_{i}\rangle^{2}
    =\bigl\langle\zeta,\boldsymbol K^{(C)}_{\eta}\zeta\bigr\rangle_{L^{2}},
    \qquad
    \boldsymbol K^{(C)}_{\eta}:=\sum_{i\ge1}\kappa_{i}(\eta)\,\psi_{i}\otimes\psi_{i}
    =\boldsymbol M_{\sigma}\boldsymbol\Pi P_{\eta}\boldsymbol\Pi\boldsymbol M_{\sigma},
\]
the last identity being the spectral expansion $P_{\eta}=\sum_{i}\kappa_{i}(\eta)\,\phi_{i}\otimes\phi_{i}$ of the smoothing operator $P_{\eta}$ of Appendix~\ref{app:limit_continuity}, $P_{\eta}:=\boldsymbol L_{\gamma}(\boldsymbol L_{\gamma}+\lambda\boldsymbol I)^{-1}$, whose kernel is $g_{\eta}$. Identically,
$\bigl\langle G_{\perp,\gamma},B^{0}_{\eta}G_{\perp,\gamma}\bigr\rangle_{\mathcal H_{\gamma}}
=\langle\zeta,\boldsymbol K^{(D)}_{\eta}\zeta\rangle$ with
$\boldsymbol K^{(D)}_{\eta}:=\boldsymbol M_{\sigma}\boldsymbol\Pi(P_{\eta})^{2}\boldsymbol\Pi\boldsymbol M_{\sigma}$, the kernel of $(P_{\eta})^{2}$ being $\mathbb E_{Z}[g_{\eta}(x,Z)g_{\eta}(Z,y)]=(g_{\eta}\circ g_{\eta})(x,y)$ by operator composition. This proves the first claim: both centered pairings are quadratic forms of the single isonormal process $\zeta$, jointly in $l$ through the coupling.

\medskip\noindent\emph{Integral kernels of the quadratic forms.}\;
\emph{Kernels of Hilbert--Schmidt operators.} A bounded operator $T$ on $L^{2}(\mathbb P_{X})$ is Hilbert--Schmidt iff it admits an integral kernel $k_{T}\in L^{2}(\mathbb P_{X}\otimes\mathbb P_{X})$ with $(Tu)(x)=\int k_{T}(x,y)\,u(y)\,\mathrm d\mathbb P_{X}(y)$; the map $T\mapsto k_{T}$ is an isometry from the Hilbert--Schmidt class, $\lVert T\rVert_{\mathrm{HS}}^{2}=\operatorname{tr}(T^{*}T)$, onto a closed subspace of $L^{2}(\mathbb P_{X}\otimes\mathbb P_{X})$, and $\operatorname{tr}(T_{1}T_{2})=\langle k_{T_{1}},k_{T_{2}}\rangle_{L^{2}(\mathbb P_{X}\otimes\mathbb P_{X})}$ for Hilbert--Schmidt $T_{1},T_{2}$. Both operators qualify: the partial sums of the spectral form above increase in the Loewner order, $\kappa_{i}(\eta)\ge0$, so the trace passes to the limit by monotone convergence for the trace on positive semidefinite operators,
$\operatorname{tr}\boldsymbol K^{(C)}_{\eta}=\sum_{i\ge1}\kappa_{i}(\eta)\,\lVert\psi_{i}(\gamma)\rVert_{L^{2}}^{2}
\le\sigma_{\max}^{2}\sum_{i\ge1}\kappa_{i}(\eta)\le\sigma_{\max}^{2}\lambda_{L}^{-1}$,
using $\lVert\psi_{i}(\gamma)\rVert_{L^{2}}^{2}=\Sigma^{\perp}_{ii}(\gamma)$ (the covariance of the directional coordinates, Appendix~\ref{app:limit_continuity}), $\Sigma^{\perp}_{ii}(\gamma)\le\sigma_{\max}^{2}$ (Corollary~\ref{cor:proj_null}), $\kappa_{i}(\eta)\le\mu_{i}(\gamma)\lambda_{L}^{-1}$, and the trace identity $\sum_{i}\mu_{i}(\gamma)=\operatorname{tr}\boldsymbol L_{\gamma}=1$; identically $\operatorname{tr}\boldsymbol K^{(D)}_{\eta}\le\sigma_{\max}^{2}\lambda_{L}^{-1}$ through $\kappa_{i}(\eta)^{2}\le\kappa_{i}(\eta)$, so both operators are trace class, hence Hilbert--Schmidt. The spectral forms above record the integral kernel $g_{\eta}$ of $P_{\eta}$; the composition rule for integral kernels, $(ST)(x,y)=\int S(x,z)\,T(z,y)\,\mathrm d\mathbb P_{X}(z)$, gives the composed kernel $(g_{\eta}\circ g_{\eta})(x,y)=\mathbb E_{Z}[g_{\eta}(x,Z)\,g_{\eta}(Z,y)]$ of $(P_{\eta})^{2}$; the $\boldsymbol\Pi$-sandwich composes the kernel with the projection in each argument, $(\Pi k\,\Pi)(x,x')$ denoting the kernel of $\boldsymbol\Pi T\boldsymbol\Pi$ for an operator $T$ with kernel $k$; and the multiplication sandwich $\boldsymbol M_{\sigma}\,\cdot\,\boldsymbol M_{\sigma}$ multiplies it by the factor $\sigma(x)\sigma(x')$. Hence the quadratic-form kernels have the integral kernels
\begin{equation}\label{eq:chaoscv-kernels}
    K^{\mathrm{sym},\perp}_{\eta}(x,x')=\sigma(x)\,(\Pi g_{\eta}\Pi)(x,x')\,\sigma(x'),
    \qquad
    K^{D,\perp}_{\eta}(x,x')=\sigma(x)\,\bigl(\Pi(g\circ g)_{\eta}\Pi\bigr)(x,x')\,\sigma(x').
\end{equation}
This proves the kernel identification.

\medskip\noindent\emph{The covariances.}\;
All three covariance families follow from one computation. The two operators share the single feature family $\{\psi_{i}(\gamma)\}_{i}$: with the spectral weights $\kappa^{(1)}_{i}(\eta):=\kappa_{i}(\eta)$ of $P_{\eta}$ and $\kappa^{(2)}_{i}(\eta):=\kappa_{i}(\eta)^{2}$ of $P_{\eta}^{2}$,
$\boldsymbol K^{(C)}_{\eta}=\sum_{i\ge1}\kappa^{(1)}_{i}(\eta)\,\psi_{i}(\gamma)\otimes\psi_{i}(\gamma)$
and
$\boldsymbol K^{(D)}_{\eta}=\sum_{i\ge1}\kappa^{(2)}_{i}(\eta)\,\psi_{i}(\gamma)\otimes\psi_{i}(\gamma)$.
Wick's identity for quadratic forms of an isonormal process,
$\operatorname{Cov}\bigl(\langle\zeta,K\zeta\rangle,\langle\zeta,K'\zeta\rangle\bigr)=2\operatorname{tr}(KK')$, reduces every covariance between two of the pairings to the single parametrized identity
\[
    2\operatorname{tr}\bigl(\boldsymbol K_{\eta}\boldsymbol K_{\eta'}\bigr)
    =2\sum_{i,j\ge1}\kappa^{(p)}_{i}(\eta)\,\kappa^{(q)}_{j}(\eta')\,
      \bigl\langle\psi_{i}(\gamma),\psi_{j}(\gamma')\bigr\rangle_{L^{2}(\mathbb P_{X})}^{2},
    \qquad (p,q)\in\{1,2\}^{2},
\]
the Gram matrix $\bigl(\langle\psi_{i}(\gamma),\psi_{j}(\gamma')\rangle\bigr)_{i,j}$ being free of $\lambda$, so that $\lambda$ enters only through the scalar spectral weights. By the trace--kernel isometry of the integral-kernels paragraph, $\operatorname{tr}(\boldsymbol K_{\eta}\boldsymbol K_{\eta'})$ is the $L^{2}(\mathbb P_{X}\otimes\mathbb P_{X})$ inner product of the integral kernels \eqref{eq:chaoscv-kernels}, so the three instances $(p,q)=(1,1)$, $(1,2)$, $(2,2)$ are exactly the three covariance families of the statement,
$\operatorname{Cov}\bigl(\mathcal C^{\perp}(\eta),\mathcal C^{\perp}(\eta')\bigr)
=2\bigl\langle K^{\mathrm{sym},\perp}_{\eta},K^{\mathrm{sym},\perp}_{\eta'}\bigr\rangle_{\mathrm{HS}}$,
$\operatorname{Cov}\bigl(\mathcal C^{\perp}(\eta),\mathcal D^{\perp}(\eta')\bigr)
=2\bigl\langle K^{\mathrm{sym},\perp}_{\eta},K^{D,\perp}_{\eta'}\bigr\rangle_{\mathrm{HS}}$,
$\operatorname{Cov}\bigl(\mathcal D^{\perp}(\eta),\mathcal D^{\perp}(\eta')\bigr)
=2\bigl\langle K^{D,\perp}_{\eta},K^{D,\perp}_{\eta'}\bigr\rangle_{\mathrm{HS}}$,
cross-covariances between distinct $\gamma_{l}$ covered by the same computation, the coupling being simultaneous in $\gamma$. This completes the proof.
\end{proof}

\subsubsection{Proof of Theorem~\ref{th:cv_uniform}: Uniform Limit Theory of the CV-Criterion Process}
\label{proof:cv_uniform}

\begin{proof}
Assume Assumption~\ref{ass:G}(iii) $0<\sigma_{\min}^{2}\le\sigma^{2}(X)\le\sigma_{\max}^{2}$ a.s., and (v) there exist $p\in(2,4]$ and $\bar M_{2p}<\infty$ with $\sup_{x}\mathbb{E}[v^{2p}\mid X=x]\le\bar M_{2p}$, and for part~(b) $\|\Pi f\|_{\infty}<\infty$; $\varrho>1/4$ is Assumption~\ref{ass:D}(iii). Throughout the proof of part (a) the conditional $(2p)$-moment clause $\sup_{a}\mathbb E[\varepsilon_{0,a}^{2p}\mid X_{a}]\le\bar M_{2p}<\infty$ for the exponent $p\in(2,4]$ of Assumption~\ref{ass:G}(v) is in force, as required by the projected-residual hypothesis of Lemma~\ref{lem:remainder_cv} and by the order-$p$ moment bounds below. We prove (a)--(c) through the bilinear representation of Lemma~\ref{lem:bilinear_cv}: the centered quadratic statistics are quadratic forms of the single scaled mean embedding against deterministic population operators, the finite-dimensional limits follow from the embedding CLT by continuous mapping with the chaos specification identified by Lemma~\ref{lem:chaos_cv}, and stochastic equicontinuity follows from uniform gradient moments through the shifted-grid containment and equicontinuity lemmas (Lemmas~\ref{lem:shifts} and~\ref{lem:equi}); no process-level entropy criterion enters.

\medskip\noindent\textit{Linearization.}\;
In the null and local regimes Lemma~\ref{lem:loo_linearization} gives, uniformly in $\eta\in\Upsilon$,
\[
\mathrm{CV}_{m}(\eta;\boldsymbol\varepsilon^{0}_{\perp})=\frac1m\sum_{i=1}^{m}(\varepsilon_{\perp,i}^{0})^{2}
    -2U_{m}(\eta)+Q_{m}(\eta)+R_{m}(\eta)+o_{p}(m^{-1}),
\]
with
\begin{align}
U_{m}(\eta)&:=\frac1{m(m-1)}\sum_{i\ne j}
        \varepsilon_{\perp,i}^{0}\varepsilon_{\perp,j}^{0}g_{\eta}(X_{i},X_{j}),\\
    Q_{m}(\eta)&:=\frac1{m(m-1)^{2}}\sum_{i\ne j}
        (\varepsilon_{\perp,j}^{0})^{2}g_{\eta}(X_{i},X_{j})^{2},\\
    R_{m}(\eta)&:=\frac2{m(m-1)^{2}}\sum_{i}\sum_{\substack{j<k\\j,k\ne i}}
        g_{\eta}(X_{i},X_{j})g_{\eta}(X_{i},X_{k})
        \varepsilon_{\perp,j}^{0}\varepsilon_{\perp,k}^{0},
\end{align}
$g_{\eta}(x,y)=\langle k_{\gamma}(x,\cdot),
(C_{\gamma}+\lambda I)^{-1}k_{\gamma}(y,\cdot)\rangle_{\mathcal H_{\gamma}}$,
and $F_{m}(\eta):=m[\mathrm{CV}_{m}(\eta)-m^{-1}\sum_{i}(\varepsilon_{\perp,i}^{0})^{2}]
=-2mU_{m}(\eta)+mQ_{m}(\eta)+mR_{m}(\eta)+o_{p}(1)$.
We call $F_{m}$ the centered CV criterion process, the criterion being evaluated at the infeasible projected residuals in the notation of Lemma~\ref{lem:repl_cost} in parts (a) and (c) below, the transfer to the feasible residuals being recorded in the Feasibility paragraph of part (a); part (b) works with the feasible criterion directly, through the pathwise leave-one-out bound available there.

 \medskip\noindent\textbf{Proof of (a): null limit (the local $O(m^{-1/2})$ case is part (c)).}\;
Work with $F_{m}$ above. Let $\varepsilon_{0}=\varepsilon(\theta_{0})$ be the infeasible residual at the true $\theta_{0}$, with $\hat\theta-\theta_{0}=O_{p}(m^{-\varrho})$ for some $\varrho>1/4$, and set $\varepsilon_{\perp}^{0}:=\hat{\boldsymbol\Pi}\varepsilon_{0}$. By Lemma~\ref{lem:repl_cost} the estimation-effect remainder is $o_{p}(1)$ uniformly in $\eta$, so the limit theory is developed on $\varepsilon_{\perp}^{0}$ and transferred back to the feasible residuals $\varepsilon_{\perp}=\hat{\boldsymbol\Pi}\hat\varepsilon$ at the end of the proof. The projection is not a perturbation to be removed: the embedding of $\varepsilon_{\perp}^{0}$ converges to the projected Gaussian element $G_{\perp,\gamma}$ of Corollary~\ref{cor:proj_null}, so the rank-$d_{\theta}$ projection is carried by the limit itself.

\medskip\noindent\textit{Finite-dimensional convergence: the embedding.}\;
Fix $\eta_{1},\dots,\eta_{J}\in\Upsilon$ and let $M_{m}(\gamma):=m^{-1/2}\sum_{i=1}^{m}w_{i}\,k_{\gamma}(X_{i},\cdot)$ be the scaled mean embedding built on the projected residual vector $w=\varepsilon_{\perp}^{0}$, with transport $\widetilde M_{m}(\gamma)$ to the common space. The joint convergence of the embeddings is that of Corollary~\ref{cor:proj_null}, transported to the selection sample: the Hilbert--space CLT and the projection correction are sample-size-agnostic, so
\[
    \bigl(\widetilde M_{m}(\gamma_{1}),\dots,\widetilde M_{m}(\gamma_{J})\bigr)
    \xrightarrow{\;d\;}
    \bigl(G_{\perp,\gamma_{1}},\dots,G_{\perp,\gamma_{J}}\bigr)
    \quad\text{in }\textstyle\bigoplus_{l\le J}\mathcal H_{\gamma_{l}}
\]
jointly with the score limit. The transfer to size $m$ is verbatim: the matrix averages behind Assumption~\ref{ass:E}(ii) are laws of large numbers in the sample size, Assumption~\ref{ass:E}(iii) supplies the moment conditions of the joint central limit theorem, and the sieve rate of Assumption~\ref{ass:F}(iii) holds at size $m$, since $\hat\theta$, the sieve, and $\hat{\boldsymbol\Pi}$ are formed on each subsample (Section~\ref{sec:cv_selection}). The embedding CLT is the only central-limit input of part (a); the remaining inputs are laws of large numbers, continuous mapping, and moment bounds.

\medskip\noindent\textit{Finite-dimensional convergence: the drifts.}\;
The deterministic part of $F_m$ is assembled from the conditional means of $mU_m$, $mQ_m$, $mR_m$ given $\mathbf X$. With $c_{ij}:=\mathbb E[\varepsilon_{\perp,i}^{0}\varepsilon_{\perp,j}^{0}\mid\mathbf X]$,
\[
    \nu_{m}(\eta):=\mathbb E[mU_{m}(\eta)\mid\mathbf X]
      =\frac{1}{m-1}\sum_{i\ne j}c_{ij}\,g_{\eta}(X_{i},X_{j}),
\]
\[
    \rho_{m}(\eta):=\mathbb E[mR_{m}(\eta)\mid\mathbf X]
      =\frac{2}{(m-1)^{2}}\sum_{i}\sum_{j<k,\,j,k\ne i}c_{jk}\,g_{\eta}(X_{i},X_{j})g_{\eta}(X_{i},X_{k}),
\]
\[
    \mathbb E[mQ_{m}(\eta)\mid\mathbf X]
      =\frac{1}{(m-1)^{2}}\sum_{i\ne j}c_{jj}\,g_{\eta}(X_{i},X_{j})^{2}.
\]
The route to the limits is explicit in the conditional second moments of $w=\hat{\boldsymbol\Pi}\varepsilon_{0}$. Under $H_{0}$ the innovations are conditionally independent and centered with $\mathbb E[\varepsilon_{0}\varepsilon_{0}^{\top}\mid\mathbf X]=D:=\operatorname{diag}(\sigma^{2}(X_{1}),\dots,\sigma^{2}(X_{m}))$, and $\hat{\boldsymbol\Pi}=\boldsymbol I_{m}-\hat{\mathbb G}M\hat{\mathbb G}^{\top}$ with $M:=(\hat{\mathbb G}^{\top}\hat{\mathbb G})^{-1}$ and $\hat g_{i}:=\hat g_{\hat\theta}(X_{i})$, so expanding $(\hat{\boldsymbol\Pi}D\hat{\boldsymbol\Pi})_{ij}$ gives, the sandwich being conditionally deterministic under the convention of Remark~\ref{rem:cond_det},
\[
    c_{ij}
    =\delta_{ij}\,\sigma^{2}(X_{i})
    -\sigma^{2}(X_{i})\,\hat g_{i}^{\top}M\hat g_{j}
    -\sigma^{2}(X_{j})\,\hat g_{i}^{\top}M\hat g_{j}
    +\hat g_{i}^{\top}M\bigl(\hat{\mathbb G}^{\top}D\hat{\mathbb G}\bigr)M\hat g_{j}.
\]
For $i\ne j$ each of the three terms is an $m$-normalized product of the design factors with the middle factor $mM$: $mM\xrightarrow{p}\Gamma_{\theta_{0}}^{-1}$ by Assumption~\ref{ass:E}(ii) at the selection sample, and $m^{-1}\hat{\mathbb G}^{\top}D\hat{\mathbb G}\xrightarrow{p}\mathbb E[\sigma^{2}(X)g_{\theta_{0}}(X)g_{\theta_{0}}(X)^{\top}]$ by the law of large numbers. The fitted factors are replaced by their population counterparts at $o_{p}(1)$ cost: splitting $\hat g_{i}-g_{\theta_{0}}(X_{i})=\bigl(g_{\hat\theta}(X_{i})-g_{\theta_{0}}(X_{i})\bigr)+\bigl(\hat g_{i}-g_{\hat\theta}(X_{i})\bigr)$, the $\hat\theta$-drift is $O_{p}(m^{-\varrho})$ element-wise by the moment bounds of Assumption~\ref{ass:E}(i) and the Lipschitz-in-$\theta$ structure of $g_{\theta}$ on the compact $\Theta$, and the sieve-fit deviation is $O_{p}(m^{-1/4})$ by the sup-norm side condition of Assumption~\ref{ass:F}(iii), the $\hat\theta$-drift being absorbed since $\varrho>1/4$, so that $\sup_{i}\lVert\hat g_{i}\rVert\le\bar G+O_{p}(m^{-1/4})$ in both regimes, the bounded-score clause of Assumption~\ref{ass:E}(i) giving $\sup_{i}\lVert g_{\theta_{0}}(X_{i})\rVert\le\bar G$ and the sup-norm side condition of Assumption~\ref{ass:F}(iii) covering the fitted factors. The uniformity in $(i,j)$ is licensed by these bounds together with the Loewner order: $\hat g_{i}^{\top}M\hat g_{j}$ is bounded in modulus by $(\bar G+O_{p}(m^{-1/4}))^{2}\lVert M\rVert_{\mathrm{op}}$, and $M(\hat{\mathbb G}^{\top}D\hat{\mathbb G})M\preceq\sigma_{\max}^{2}M$ by congruence from $D\preceq\sigma_{\max}^{2}\boldsymbol I_{m}$, so that $\hat g_{i}^{\top}M(\hat{\mathbb G}^{\top}D\hat{\mathbb G})M\hat g_{j}$ is bounded in modulus by $\sigma_{\max}^{2}(\bar G+O_{p}(m^{-1/4}))^{2}\lVert M\rVert_{\mathrm{op}}$, with $\lVert M\rVert_{\mathrm{op}}=O_{p}(m^{-1})$ by Assumption~\ref{ass:E}(ii). Hence $c_{ij}=m^{-1}c^{*}(X_{i},X_{j})+o_{p}(m^{-1})$ uniformly in $(i,j)$, with $c^{*}$ as recorded in Section~\ref{sec:uniform_limit}, while the diagonal keeps its leading value, $c_{jj}=\sigma^{2}(X_{j})+O_{p}(m^{-1})$ uniformly in $j$: at $i=j$ the two correction terms of the expansion are bounded by the two bounds just recorded, uniformly in $j$.

The three drift averages are therefore $m$-normalized averages of the recorded factors,
\[
    \nu_{m}(\eta)=\frac{m}{m-1}\cdot\frac{1}{m^{2}}\sum_{i\ne j}c^{*}(X_{i},X_{j})\,g_{\eta}(X_{i},X_{j})+o_{p}(1),
\]
\[
    \rho_{m}(\eta)=\frac{m^{2}}{(m-1)^{2}}\cdot\frac{2}{m^{3}}\sum_{i}\sum_{j<k,\,j,k\ne i}c^{*}(X_{j},X_{k})\,g_{\eta}(X_{i},X_{j})g_{\eta}(X_{i},X_{k})+o_{p}(1),
\]
\[
    \mathbb E[mQ_{m}(\eta)\mid\mathbf X]
    =\frac{m^{2}}{(m-1)^{2}}\cdot\frac{1}{m}\sum_{j}\sigma^{2}(X_{j})
      \Bigl(\frac{1}{m}\sum_{i\ne j}g_{\eta}(X_{i},X_{j})^{2}\Bigr)+O_{p}(m^{-1}).
\]
The averages converge by the $U$- and $V$-statistic-type laws of large numbers, uniformly in $\eta$ through the Lipschitz modulus of $g_{\eta}$ (Lemma~\ref{lem:bilinear_cv}(iii)(a)): the off-diagonal averages have the limits $\nu(\eta)$ and $\rho(\eta)$, the second using the symmetry $\mathbb E_{X_{1}}[g_{\eta}(X_{1},X_{2})g_{\eta}(X_{1},X_{3})]=(g_{\eta}\circ g_{\eta})(X_{2},X_{3})$; in the diagonal average, the inner sum $\frac{1}{m}\sum_{i\ne j}g_{\eta}(X_{i},X_{j})^{2}\xrightarrow{p}\mathbb E_{Z}\bigl[g_{\eta}(X_{j},Z)^{2}\bigr]$ uniformly in $(j,\eta)$ (bounded summands, Chebyshev with a union bound over $j$, Lipschitz modulus in $\eta$), which assembles the limit $V(\eta)$. Collecting the three limits,
\[
\sup_{\eta}\bigl|\nu_{m}(\eta)-\nu(\eta)\bigr|\xrightarrow{p}0,\qquad
\sup_{\eta}\bigl|\rho_{m}(\eta)-\rho(\eta)\bigr|\xrightarrow{p}0,\qquad
\sup_{\eta}\bigl|\mathbb E[mQ_{m}(\eta)\mid\mathbf X]-V(\eta)\bigr|\xrightarrow{p}0,
\]
where
\[
\nu(\eta)=\mathbb E\bigl[c^{*}(X_{1},X_{2})g_{\eta}(X_{1},X_{2})\bigr],\qquad
\rho(\eta)=\mathbb E\bigl[c^{*}(X_{1},X_{2})(g_{\eta}\circ g_{\eta})(X_{1},X_{2})\bigr],
\]
\[
V(\eta)=\mathbb E\bigl[\sigma^{2}(X_{2})g_{\eta}(X_{1},X_{2})^{2}\bigr],
\]
the first two being the recorded projection corrections of Section~\ref{sec:uniform_limit}. Hence the deterministic part of $F_{m}$ converges uniformly to $-2\nu+V+\rho=V+\Delta$ with $\Delta:=\rho-2\nu$.

\medskip\noindent\textit{Finite-dimensional convergence: the centered statistics.}\;
Write $\bar U_{m}(\eta):=mU_{m}(\eta)-\nu_{m}(\eta)$ and $\bar R_{m}(\eta):=mR_{m}(\eta)-\rho_{m}(\eta)$.

\emph{The $U$-coordinate.} By Lemma~\ref{lem:bilinear_cv}(i), $mU_{m}(\eta)=Q^{U}_{m}(\eta)-D_{m}(\eta)$ with
\[
    Q^{U}_{m}(\eta):=\frac{m}{m-1}\bigl\langle\widetilde M_{m}(\gamma),\widetilde A^{0}_{\eta}\widetilde M_{m}(\gamma)\bigr\rangle_{\mathcal H_{\gamma_0}},
    \qquad
    D_{m}(\eta):=\frac{1}{m-1}\sum_{i}w_{i}^{2}\,g_{\eta}(X_{i},X_{i}),
\]
so by linearity of the conditional expectation $\bar U_{m}(\eta)=\bigl[Q^{U}_{m}(\eta)-\mathbb E\bigl(Q^{U}_{m}(\eta)\mid\mathbf X\bigr)\bigr]-\bigl[D_{m}(\eta)-\mathbb E\bigl(D_{m}(\eta)\mid\mathbf X\bigr)\bigr]$. The map $x\mapsto\bigl(\langle x_{\gamma_{l}},\widetilde A^{0}_{\eta_{l}}x_{\gamma_{l}}\rangle\bigr)_{l\le J}$ on $\bigoplus_{l\le J}\mathcal H_{\gamma_{l}}$ is continuous and the operators $\widetilde A^{0}_{\eta_{l}}$ are deterministic, so the continuous mapping theorem applied to the joint convergence of the embeddings gives the joint convergence of the bilinear forms to $\bigl(\langle G_{\perp,\gamma_{l}},A^{0}_{\eta_{l}}G_{\perp,\gamma_{l}}\rangle_{\mathcal H_{\gamma_{l}}}\bigr)_{l\le J}$.

The conditional mean of $Q^{U}_{m}(\eta)$ reduces by the trace, the sandwich being conditionally deterministic under the convention of Remark~\ref{rem:cond_det}. Since $\widehat\Sigma_{m}(\gamma):=\mathbb E[\widetilde M_{m}(\gamma)\widetilde M_{m}(\gamma)^{*}\mid\mathbf X]=\frac{1}{m}\sum_{i,j}c_{ij}\,\widetilde k_{\gamma}(X_{i})\otimes\widetilde k_{\gamma}(X_{j})$ and $\langle\widetilde k_{\gamma}(X_{i}),\widetilde A^{0}_{\eta}\widetilde k_{\gamma}(X_{j})\rangle_{\mathcal H_{\gamma_0}}=g_{\eta}(X_{i},X_{j})$ (the reproducing-kernel identification in the statement of Lemma~\ref{lem:bilinear_cv}),
\[
\begin{aligned}
    \mathbb E\bigl(Q^{U}_{m}(\eta)\mid\mathbf X\bigr)
    &=\frac{m}{m-1}\operatorname{tr}\bigl(\widetilde A^{0}_{\eta}\widehat\Sigma_{m}(\gamma)\bigr)
    =\frac{m}{m-1}\cdot\frac{1}{m}\sum_{i,j}c_{ij}\,g_{\eta}(X_{i},X_{j})\\
    &=\nu_{m}(\eta)+\frac{1}{m-1}\sum_{i}c_{ii}\,g_{\eta}(X_{i},X_{i}),
\end{aligned}
\]
the last step by $\sum_{i\ne j}c_{ij}g_{\eta}(X_{i},X_{j})=(m-1)\nu_{m}(\eta)$. By the Drifts paragraph the diagonal average converges by the law of large numbers, $c_{ii}=\sigma^{2}(X_{i})+O_{p}(m^{-1})$ uniformly in $i$, so $\mathbb E\bigl(Q^{U}_{m}(\eta)\mid\mathbf X\bigr)\xrightarrow{p}a(\eta):=\nu(\eta)+\mathbb E\bigl[\sigma^{2}(X)g_{\eta}(X,X)\bigr]$ uniformly in $\eta$; while the conditional mean $\mathbb E[D_{m}(\eta)\mid\mathbf X]=\frac{1}{m-1}\sum_{i}c_{ii}\,g_{\eta}(X_{i},X_{i})$ converges to $d(\eta):=\mathbb E[\sigma^{2}(X)g_{\eta}(X,X)]$ by the law of large numbers and $c_{ii}=\sigma^{2}(X_{i})+O_{p}(m^{-1})$ uniformly in $i$, uniformly in $\eta$ through the Lipschitz modulus of $g_{\eta}$; the fluctuation $D_{m}(\eta)-\mathbb E[D_{m}(\eta)\mid\mathbf X]$ is the centered quadratic form at $A^{(D)}_{\eta}:=\frac{1}{m-1}\hat{\boldsymbol\Pi}\operatorname{diag}\bigl(g_{\eta}(X_{1},X_{1}),\dots,g_{\eta}(X_{m},X_{m})\bigr)\hat{\boldsymbol\Pi}$, symmetric and conditionally deterministic under Remark~\ref{rem:cond_det}, with $\lVert A^{(D)}_{\eta}\rVert_{\mathrm F}\le\frac{\sqrt m}{m-1}\,\bar k\,\lambda_{L}^{-1}\le2\,\bar k\,\lambda_{L}^{-1}\,m^{-1/2}$ for $m\ge2$, so \eqref{eq:quadvar_bound} gives the variance $O(m^{-1})$ uniformly in $\eta$ and the fluctuation is $O_{p}(m^{-1/2})$; so the second bracket of $\bar U_{m}(\eta)$ vanishes in the limit. Slutsky's theorem gives
\[
    \bar U_{m}(\eta_{l})\;\Rightarrow\;
    \bigl\langle G_{\perp,\gamma_{l}},A^{0}_{\eta_{l}}G_{\perp,\gamma_{l}}\bigr\rangle_{\mathcal H_{\gamma_{l}}}-a(\eta_{l}),
    \qquad l\le J\ \text{jointly}.
\]
Since $\mathbb E[\bar U_{m}(\eta)]=0$ for every $m$ while the second moments are uniformly bounded, the shifted limit has mean zero, $a(\eta)=\mathbb E\langle G_{\perp,\gamma},A^{0}_{\eta}G_{\perp,\gamma}\rangle_{\mathcal H_{\gamma}}$: $\bar U_{m}(\eta)$ is the centered quadratic form at the $\hat{\boldsymbol\Pi}$-sandwich of the design-measurable $A^{(g)}_{\eta,ij}:=\frac{1}{m-1}g_{\eta}(X_{i},X_{j})\mathbf 1_{i\ne j}$, symmetric, zero-diagonal, and conditionally deterministic under Remark~\ref{rem:cond_det} through the sandwich, so \eqref{eq:quadvar_bound} gives $\sup_{m\ge2,\eta}\mathbb E[\bar U_{m}(\eta)^{2}\mid\mathbf X]\le(2\sigma_{\max}^{4}+\bar M_{2p}^{2/p})\,2\,\bar k^{2}\lambda_{L}^{-2}$ through $\lVert A^{(g)}_{\eta}\rVert_{\mathrm F}^{2}\le\frac{m}{m-1}\bar k^{2}\lambda_{L}^{-2}$ and $\lVert\hat{\boldsymbol\Pi}A\hat{\boldsymbol\Pi}\rVert_{\mathrm F}\le\lVert A\rVert_{\mathrm F}$; the same display at the composed kernel, whose coefficient matrix has entries $\frac{m-2}{(m-1)^{2}}(g_{\eta}\circ g_{\eta})(X_{i},X_{j})\mathbf 1_{i\ne j}$ and Frobenius norm at most $\bar k^{2}\lambda_{L}^{-2}$, bounds the centered bilinear part of the $R$-coordinate below; and the limit is the centered chaos $\mathcal C^{\perp}(\eta_{l})$ of Lemma~\ref{lem:chaos_cv}.

\emph{The $R$-coordinate.} The same device applies with $\widetilde B^{0}_{\eta}$ in place of $\widetilde A^{0}_{\eta}$: by Lemma~\ref{lem:bilinear_cv}(ii),
\[
    mR_{m}(\eta)=\frac{m(m-2)}{m-1}\,U_{m}(h_{\eta}^{(R)})+\mathcal R_{m}(\eta),
\]
where $mU_{m}(h_{\eta}^{(R)})$ is the bilinear form of $\widetilde B^{0}_{\eta}$ with the composed-kernel diagonal correction, and $\sup_{\eta}|\mathcal R_{m}(\eta)|=O_{p}(m^{-1/2})$. The conditional mean of the bilinear form reduces by the same trace computation, the reproducing kernel of $\widetilde B^{0}_{\eta}$ being the composed $(g_{\eta}\circ g_{\eta})(X_{i},X_{j})$, and converges to $b(\eta):=\rho(\eta)+\mathbb E\bigl[\sigma^{2}(X)(g_{\eta}\circ g_{\eta})(X,X)\bigr]$ uniformly in $\eta$, the deterministic bookkeeping $\rho_{m}\to\rho$ cancelling against the off-diagonal part of the trace; the same mean-zero argument identifies $b(\eta)=\mathbb E\langle G_{\perp,\gamma},B^{0}_{\eta}G_{\perp,\gamma}\rangle_{\mathcal H_{\gamma}}$, and $\mathcal R_{m}$ is $o_{p}(1)$ pointwise. Hence $\bar R_{m}(\eta_{l})\Rightarrow\mathcal D^{\perp}(\eta_{l})$ jointly.

\emph{The $Q$-coordinate.} The fluctuation of $mQ_{m}$ vanishes at the limit. Regrouping by the squared residual,
\[
    \bar Q_{m}(\eta):=mQ_{m}(\eta)-\mathbb E\bigl[mQ_{m}(\eta)\mid\mathbf X\bigr]
    =\boldsymbol\varepsilon_{0}^{\top}B_{\eta}\boldsymbol\varepsilon_{0}-\operatorname{tr}(B_{\eta}D),
    \qquad
    B_{\eta}:=\frac{1}{(m-1)^{2}}\hat{\boldsymbol\Pi}\operatorname{diag}(s_{\eta})\hat{\boldsymbol\Pi},
\]
with $s_{\eta,j}:=\sum_{i\ne j}g_{\eta}(X_{i},X_{j})^{2}$, the second identity following from $w=\hat{\boldsymbol\Pi}\varepsilon_{0}$ together with the conditional independence and centering of the innovations. Lemma~\ref{lem:pairvar}, applied to $\zeta=\boldsymbol\varepsilon_{0}$ and $A=B_{\eta}$ (symmetric and conditionally deterministic under Remark~\ref{rem:cond_det}), gives by the variance bound \eqref{eq:quadvar_bound}, uniformly in $\eta$,
\[
    \mathbb E\bigl[\bar Q_{m}(\eta)^{2}\mid\mathbf X\bigr]
    \le\bigl(2\sigma_{\max}^{4}+\bar M_{2p}^{2/p}\bigr)\lVert B_{\eta}\rVert_{\mathrm F}^{2}
    \le\bigl(2\sigma_{\max}^{4}+\bar M_{2p}^{2/p}\bigr)\frac{m\bigl((m-1)\bar k^{2}\lambda_{L}^{-2}\bigr)^{2}}{(m-1)^{4}},
\]
the second inequality by $\lVert\hat{\boldsymbol\Pi}\operatorname{diag}(s_{\eta})\hat{\boldsymbol\Pi}\rVert_{\mathrm F}\le\lVert\hat{\boldsymbol\Pi}\rVert_{\mathrm{op}}^{2}\lVert\operatorname{diag}(s_{\eta})\rVert_{\mathrm F}$ (idempotence of $\hat{\boldsymbol\Pi}$, $\lVert\hat{\boldsymbol\Pi}\rVert_{\mathrm{op}}\le1$) and the envelope $s_{\eta,j}\le(m-1)\bar k^{2}\lambda_{L}^{-2}$ from $|g_{\eta}|\le\bar k\lambda_{L}^{-1}$. Hence $\mathbb E[\bar Q_{m}(\eta)^{2}\mid\mathbf X]=O(m^{-1})$ uniformly in $\eta$, and $\bar Q_{m}(\eta_{l})=o_{p}(1)$ pointwise by Chebyshev's inequality, so the $Q$-coordinate of the fidi limit is deterministic. Assembling the three coordinates,
\[
    \bigl(mU_{m}(\eta_{l}),mQ_{m}(\eta_{l}),mR_{m}(\eta_{l})\bigr)_{l\le J}
    \xrightarrow{d}
    \bigl(\nu(\eta_{l})+\mathcal C^{\perp}(\eta_{l}),\;V(\eta_{l}),\;
    \rho(\eta_{l})+\mathcal D^{\perp}(\eta_{l})\bigr)_{l\le J},
\]
the chaoses carrying the kernels $K^{\mathrm{sym},\perp},K^{D,\perp}$ of Lemma~\ref{lem:chaos_cv}.

\medskip\noindent\textit{Stochastic equicontinuity.}\;
Equicontinuity and boundedness follow from uniform gradient moments through the shifted-grid equicontinuity lemma (Lemma~\ref{lem:equi}), whose geometric input is the containment of short pairs by shifted grids (Lemma~\ref{lem:shifts}); the uniform remainder bound of Lemma~\ref{lem:remainder_cv} and deterministic drift moduli supply its additive-term hypotheses. Alongside the centered statistics $\bar U_{m}$ and $\bar Q_{m}$ of the finite-dimensional paragraph, write
\[
    mR^{\mathrm{bil}}_{m}(\eta):=\frac{m(m-2)}{m-1}\,U_{m}\bigl(h^{(R)}_{\eta}\bigr),
    \qquad
    \bar R^{\mathrm{bil}}_{m}(\eta):=mR^{\mathrm{bil}}_{m}(\eta)-\mathbb E\bigl[mR^{\mathrm{bil}}_{m}(\eta)\mid\mathbf X\bigr],
\]
so that $mR_{m}=mR^{\mathrm{bil}}_{m}+\mathcal R_{m}$ by Lemma~\ref{lem:bilinear_cv}(ii).

\emph{Step A: kernel-derivative envelopes.}\; By the derivative-envelope clause of Lemma~\ref{lem:bilinear_cv}(iii)(a), the kernels $g_{\eta}$, $(g_{\eta}\circ g_{\eta})$, and $r_{\eta}$ are entrywise $C^{1}$ on $\Upsilon$ with, for $l\in\{\gamma,\lambda\}$,
\[
    \sup_{\eta,\,x,y}\bigl|\partial_{\eta_{l}}g_{\eta}(x,y)\bigr|\le C_{g}=:G_{1},
    \qquad
    \sup_{\eta,\,x,y}\bigl|\partial_{\eta_{l}}(g_{\eta}\circ g_{\eta})(x,y)\bigr|\le C'_{g}=:G_{2},
\]
and $|\partial_{\eta_{l}}(g_{\eta}^{2})(x,y)|\le2\bar k\,\lambda_{L}^{-1}G_{1}$ by the product rule and $|g_{\eta}|\le\bar k\,\lambda_{L}^{-1}$; the constants are deterministic and $\mathbf X$-free. The same term-by-term differentiation applied to the finite sums defining the centered statistics shows that $\bar U_{m}$, $\bar R^{\mathrm{bil}}_{m}$, and $\bar Q_{m}$ have $C^{1}$ paths on $\Upsilon$, with gradients the formal term-by-term derivatives; entrywise boundedness uniform in $(\eta,x,y)$ is all that Step B uses, and no second derivative enters.

\emph{Step B: gradient moments, conditional on $\mathbf X$.}\;
Throughout the step, $w=\varepsilon^{0}_{\perp}$, $D=\operatorname{diag}\bigl(\sigma^{2}(X_{1}),\dots,\sigma^{2}(X_{m})\bigr)$, and the entries of $\hat{\boldsymbol\Pi}D\hat{\boldsymbol\Pi}$ are $c_{ij}=\mathbb E[w_{i}w_{j}\mid\mathbf X]$. \emph{The $U$-piece.} The definition of $U_{m}$ gives $mU_{m}(\eta)=\frac{1}{m-1}\sum_{i\ne j}w_{i}w_{j}\,g_{\eta}(X_{i},X_{j})$, the off-diagonal form of Lemma~\ref{lem:bilinear_cv}(i), and term-by-term differentiation of the finite sums gives $\nabla_{\eta}\nu_{m}(\eta)=\frac{1}{m-1}\sum_{i\ne j}c_{ij}\,\partial_{\eta}g_{\eta}(X_{i},X_{j})$; hence
\[
    \nabla_{\eta}\bar U_{m}(\eta)
    =\frac{1}{m-1}\sum_{i\ne j}\bigl(w_{i}w_{j}-c_{ij}\bigr)\,\partial_{\eta}g_{\eta}(X_{i},X_{j})
    =\boldsymbol\varepsilon_{0}^{\top}\hat{\boldsymbol\Pi}A^{(U)}_{\eta}\hat{\boldsymbol\Pi}\boldsymbol\varepsilon_{0}
      -\operatorname{tr}\bigl(\hat{\boldsymbol\Pi}A^{(U)}_{\eta}\hat{\boldsymbol\Pi}D\bigr),
\]
where $A^{(U)}_{\eta,ij}:=\frac{1}{m-1}\,\partial_{\eta}g_{\eta}(X_{i},X_{j})\,\mathbf 1_{i\ne j}$ is symmetric, zero-diagonal and design-measurable, the identity holding componentwise in $(\gamma,\lambda)$, the second identity inserting $w=\hat{\boldsymbol\Pi}\boldsymbol\varepsilon_{0}$ and computing the conditional mean as $\operatorname{tr}\bigl(\hat{\boldsymbol\Pi}A\hat{\boldsymbol\Pi}D\bigr)=\operatorname{tr}\bigl(A\hat{\boldsymbol\Pi}D\hat{\boldsymbol\Pi}\bigr)=\frac{1}{m-1}\sum_{i\ne j}c_{ij}\,\partial_{\eta}g_{\eta}(X_{i},X_{j})$ by cyclicity of the trace. By Step A,
\[
    \bigl\lVert A^{(U)}_{\eta}\bigr\rVert_{\mathrm F}^{2}
    =\frac{1}{(m-1)^{2}}\sum_{i\ne j}\bigl(\partial_{\eta}g_{\eta}\bigr)^{2}(X_{i},X_{j})
    \le\frac{m(m-1)\,G_{1}^{2}}{(m-1)^{2}}
    \le2\,G_{1}^{2}
    \qquad\text{for }m\ge2,
\]
deterministically. The order-$p$ moment bound \eqref{eq:quad4} (Lemma~\ref{lem:quad_moment}), instantiated at the coefficient matrix $\hat{\boldsymbol\Pi}A^{(U)}_{\eta}\hat{\boldsymbol\Pi}$, conditionally deterministic under Remark~\ref{rem:cond_det}, gives
\[
\begin{aligned}
    \mathbb E\Bigl[\bigl\lvert\boldsymbol\varepsilon_{0}^{\top}\hat{\boldsymbol\Pi}A^{(U)}_{\eta}\hat{\boldsymbol\Pi}\boldsymbol\varepsilon_{0}
      &-\operatorname{tr}\bigl(\hat{\boldsymbol\Pi}A^{(U)}_{\eta}\hat{\boldsymbol\Pi}D\bigr)\bigr\rvert^{p}\Bigm|\mathbf X\Bigr]\\
    &\le C_{p}\,\bigl\lVert\hat{\boldsymbol\Pi}A^{(U)}_{\eta}\hat{\boldsymbol\Pi}\bigr\rVert_{\mathrm F}^{p}
    \le C_{p}\,\bigl\lVert A^{(U)}_{\eta}\bigr\rVert_{\mathrm F}^{p}
    \le C_{p}\,(2G_{1})^{p}=:C_{U},
\end{aligned}
\]
the middle step by $\bigl\lVert\hat{\boldsymbol\Pi}A\hat{\boldsymbol\Pi}\bigr\rVert_{\mathrm F}\le\lVert A\rVert_{\mathrm F}$ ($\lVert\hat{\boldsymbol\Pi}\rVert_{\mathrm{op}}\le1$) and the last by $\lVert A^{(U)}_{\eta}\rVert_{\mathrm F}\le2\,G_{1}$ from the envelope displayed above, where and below $C_{p} := C_{p}(\sigma_{\max},\bar M_{2p})$ denotes the constant of \eqref{eq:quad4}. Hence $\sup_{\eta\in\Upsilon}\mathbb E\bigl[\lvert\nabla_{\eta}\bar U_{m}(\eta)\rvert^{p}\mid\mathbf X\bigr]\le C_{U}$ uniformly in $m\ge2$, and the bound being $\mathbf X$-free it holds unconditionally. \emph{The $R$-piece.} The same computation applies to the bilinear part: by the off-diagonal expansion of $U_{m}$ at the composed kernel (clause (i) of Lemma~\ref{lem:bilinear_cv} with $(g_{\eta}\circ g_{\eta})$ in place of $g_{\eta}$), $mU_{m}\bigl(h^{(R)}_{\eta}\bigr)=\frac{1}{m-1}\sum_{i\ne j}w_{i}w_{j}\,(g\circ g)_{\eta}(X_{i},X_{j})$, so
\[
    \nabla_{\eta}\bar R^{\mathrm{bil}}_{m}(\eta)
    =\boldsymbol\varepsilon_{0}^{\top}\hat{\boldsymbol\Pi}A^{(R)}_{\eta}\hat{\boldsymbol\Pi}\boldsymbol\varepsilon_{0}
      -\operatorname{tr}\bigl(\hat{\boldsymbol\Pi}A^{(R)}_{\eta}\hat{\boldsymbol\Pi}D\bigr),
    \qquad
    A^{(R)}_{\eta,ij}:=\frac{m-2}{(m-1)^{2}}\,\partial_{\eta}(g\circ g)_{\eta}(X_{i},X_{j})\,\mathbf 1_{i\ne j},
\]
the prefactor being $\frac{m(m-2)}{m-1}\cdot\frac{1}{m(m-1)}=\frac{m-2}{(m-1)^{2}}$, the anchor-count prefactor of the reduction of Lemma~\ref{lem:bilinear_cv}(ii) times the $U_{m}$ normalization; and $\bigl\lVert A^{(R)}_{\eta}\bigr\rVert_{\mathrm F}\le\frac{(m-2)\sqrt{m(m-1)}}{(m-1)^{2}}\,G_{2}\le4\,G_{2}$ for $m\ge2$; the same order-$p$ enumeration gives
\[
    \sup_{\eta\in\Upsilon}\mathbb E\bigl[\lvert\nabla_{\eta}\bar R^{\mathrm{bil}}_{m}(\eta)\rvert^{p}\mid\mathbf X\bigr]
    \le C_{p}(\sigma_{\max},\bar M_{2p})\,(4G_{2})^{p}=:C_{R}.
\]
\emph{The $Q$-piece.} The centered diagonal statistic is regrouped by the squared residual, from the matrix form $B_{\eta}=\frac{1}{(m-1)^{2}}\hat{\boldsymbol\Pi}\operatorname{diag}(s_{\eta})\hat{\boldsymbol\Pi}$ of the finite-dimensional paragraph and the conditional centering of the innovations,
\[
    \bar Q_{m}(\eta)=\frac{1}{(m-1)^{2}}\sum_{i\ne j}g_{\eta}(X_{i},X_{j})^{2}\,\xi_{j},
    \qquad
    \xi_{j}:=w_{j}^{2}-\mathbb E\bigl[w_{j}^{2}\mid\mathbf X\bigr],
\]
and the coefficients $\xi_{j}$ are free of $\eta$, so
\[
    \nabla_{\eta}\bar Q_{m}(\eta)=\sum_{j=1}^{m}a_{\eta,j}\,\xi_{j},
    \qquad
    a_{\eta,j}:=\frac{1}{(m-1)^{2}}\sum_{i\ne j}\partial_{\eta}\bigl(g_{\eta}^{2}\bigr)(X_{i},X_{j}),
\]
with $\bigl\lvert\partial_{\eta}(g_{\eta}^{2})\bigr\rvert=2\lvert g_{\eta}\rvert\,\lvert\partial_{\eta}g_{\eta}\rvert\le2\bar k\,\lambda_{L}^{-1}G_{1}$, whence $\lvert a_{\eta,j}\rvert\le\frac{2\bar k\,\lambda_{L}^{-1}G_{1}}{m-1}$ and
\[
    \sum_{j=1}^{m}a_{\eta,j}^{2}
    \le\frac{4\,m\,\bar k^{2}\lambda_{L}^{-2}G_{1}^{2}}{(m-1)^{2}}
    \le\frac{16\,\bar k^{2}\lambda_{L}^{-2}G_{1}^{2}}{m}
    \qquad\text{for }m\ge2.
\]
Inserting $w=\hat{\boldsymbol\Pi}\boldsymbol\varepsilon_{0}$ gives the sandwich reduction
\[
    \sum_{j}a_{\eta,j}\xi_{j}
    =\boldsymbol\varepsilon_{0}^{\top}\hat{\boldsymbol\Pi}\operatorname{diag}(a_{\eta})\hat{\boldsymbol\Pi}\boldsymbol\varepsilon_{0}
    -\operatorname{tr}\bigl(\operatorname{diag}(a_{\eta})\,\hat{\boldsymbol\Pi}D\hat{\boldsymbol\Pi}\bigr),
\]
turning $\sum_{j}a_{\eta,j}\xi_{j}$ into the centered quadratic form at the diagonal matrix $\operatorname{diag}(a_{\eta})$, where the centered form requires no zero-diagonal condition, the diagonal entries shifting only the subtracted conditional mean; the same order-$p$ enumeration gives
\[
    \sup_{\eta\in\Upsilon}\mathbb E\bigl[\lvert\nabla_{\eta}\bar Q_{m}(\eta)\rvert^{p}\mid\mathbf X\bigr]
    \le C_{p}\Bigl(\sum_{j}a_{\eta,j}^{2}\Bigr)^{p/2}
    \le C_{p}\,\frac{\bigl(16\,\bar k^{2}\lambda_{L}^{-2}G_{1}^{2}\bigr)^{p/2}}{m^{p/2}}
    =:C_{Q}\,m^{-p/2}.
\]
All three gradient bounds hold conditionally on $\mathbf X$ with $\mathbf X$-free constants, hence unconditionally.

\emph{Drift modulus.}\;
The drift $\mathrm{dr}_{m}:=-2\nu_{m}+\mathbb E[mQ_{m}\mid\mathbf X]+\rho_{m}$ is the additive term of the decomposition of Step C and enters the equicontinuity argument through hypothesis (ii) of Lemma~\ref{lem:equi}, so it is controlled by a Lipschitz modulus and uniform convergence rather than by a gradient moment. The drifts are bounded and Lipschitz: the recorded expansion of the Drifts paragraph, $c_{jk}=-\sigma^{2}(X_{j})H_{jk}-\sigma^{2}(X_{k})H_{jk}+(HDH)_{jk}$ with the idempotent $H=\hat{\mathbb G}M\hat{\mathbb G}^{\top}$, gives $|H_{jk}|^{2}\le H_{jj}H_{kk}\le1$ (the Gram bound for positive semidefinite matrices, $H_{jj}\le1$ by idempotence) and $\lvert(HDH)_{jk}\rvert\le\sigma^{2}_{\max}\sqrt{H_{jj}H_{kk}}$ (Cauchy--Schwarz with $\lVert H_{j\cdot}\rVert^{2}=(H^{2})_{jj}=H_{jj}$), so $\lvert c_{jk}\rvert\le3\sigma^{2}_{\max}$ deterministically and hence $\lvert c^{*}\rvert\le3\sigma^{2}_{\max}$ a.s. through the uniform convergence of the Drifts paragraph; consequently the limit drifts are Lipschitz with $\mathbf X$-free constants,
\[
\begin{aligned}
    \max\Bigl(\bigl\lvert\nu(\eta)-\nu(\eta')\bigr\rvert,\,\bigl\lvert\rho(\eta)-\rho(\eta')\bigr\rvert,\,\bigl\lvert V(\eta)-V(\eta')\bigr\rvert\Bigr)
    &\le L_{\mathrm{dr}}\,\lVert\eta-\eta'\rVert,\\
    L_{\mathrm{dr}}&:=\bigl(3\sigma^{2}_{\max}C_{g}\bigr)\vee\bigl(3\sigma^{2}_{\max}C'_{g}\bigr)\\
    &\quad\vee\bigl(2\bar k\,\lambda_{L}^{-1}\sigma^{2}_{\max}C_{g}\bigr),
\end{aligned}
\]
and bounded, $\lVert\nu\rVert_{\infty}\le3\sigma^{2}_{\max}\bar k\,\lambda_{L}^{-1}$, $\lVert\rho\rVert_{\infty}\le3\sigma^{2}_{\max}\bar k^{2}\lambda_{L}^{-2}$, $\lVert V\rVert_{\infty}\le\sigma^{2}_{\max}\bar k^{2}\lambda_{L}^{-2}$; the three uniform convergences displayed in the Drifts paragraph give $r^{\mathrm{dr}}_{m}:=\sup_{\eta}\bigl(\lvert\nu_{m}-\nu\rvert+\lvert\rho_{m}-\rho\rvert+\bigl\lvert\mathbb E[mQ_{m}\mid\mathbf X]-V\bigr\rvert\bigr)=o_{p}(1)$.

\emph{Step C: equicontinuity and tightness.}\;
Fix $\eta_{0}\in\Upsilon$. Write $\bar F_{m}:=-2\,\bar U_{m}+\bar Q_{m}+\bar R^{\mathrm{bil}}_{m}$ and $\Delta S_{m}:=S_{m}(\eta)-S_{m}(\eta')$ for the increment of a field $S_{m}$ at a pair of tuning parameters. Collecting the decompositions $mU_{m}=\bar U_{m}+\nu_{m}$, $mQ_{m}=\bar Q_{m}+\mathbb E[mQ_{m}\mid\mathbf X]$ and $mR_{m}=\bar R^{\mathrm{bil}}_{m}+\rho_{m}+\bigl(\mathcal R_{m}-\mathbb E[\mathcal R_{m}\mid\mathbf X]\bigr)$ with $\mathrm{dr}_{m}:=-2\nu_{m}+\mathbb E[mQ_{m}\mid\mathbf X]+\rho_{m}$,
\[
    F_{m}=\bar F_{m}+\mathrm{dr}_{m}+\bigl(\mathcal R_{m}-\mathbb E[\mathcal R_{m}\mid\mathbf X]\bigr)+o_{p}(1),
\]
the last term being the uniform linearization remainder of Lemma~\ref{lem:loo_linearization}. The criterion and the three centered pieces have $C^{1}$ paths on $\Upsilon$: $\mathrm{CV}_{m}$ is a rational function of the leave-one-out Gram resolvents $\bigl(\boldsymbol K^{(-i)}_{\gamma}+\lambda(m-1)\boldsymbol I\bigr)^{-1}$ with the real-analytic kernel entries $e^{-\gamma\lVert X_{i}-X_{j}\rVert^{2}}$, and $\bar U_{m}$, $\bar R^{\mathrm{bil}}_{m}$ and $\bar Q_{m}$ have $C^{1}$ paths by Step A.

The key estimate is the gradient moment of the centered part. Since $\lVert a+b+c\rVert^{p}\le3^{p-1}\bigl(\lVert a\rVert^{p}+\lVert b\rVert^{p}+\lVert c\rVert^{p}\bigr)$,
\[
\begin{aligned}
    \bigl\lVert\nabla_{\eta}\bar F_{m}(\eta)\bigr\rVert^{p}
    &\le3^{p-1}\Bigl(2^{p}\,\bigl\lVert\nabla_{\eta}\bar U_{m}(\eta)\bigr\rVert^{p}
      +\bigl\lVert\nabla_{\eta}\bar R^{\mathrm{bil}}_{m}(\eta)\bigr\rVert^{p}
      +\bigl\lVert\nabla_{\eta}\bar Q_{m}(\eta)\bigr\rVert^{p}\Bigr),\\
    \bigl\lVert\nabla_{\eta}S_{m}(\eta)\bigr\rVert^{p}
    &\le2^{p/2-1}\Bigl(\bigl\lvert\partial_{\gamma}S_{m}(\eta)\bigr\rvert^{p}
      +\bigl\lvert\partial_{\lambda}S_{m}(\eta)\bigr\rvert^{p}\Bigr),
\end{aligned}
\]
the first from $\lVert-2\nabla_{\eta}\bar U_{m}+\nabla_{\eta}\bar R^{\mathrm{bil}}_{m}+\nabla_{\eta}\bar Q_{m}\rVert\le2\lVert\nabla_{\eta}\bar U_{m}\rVert+\lVert\nabla_{\eta}\bar R^{\mathrm{bil}}_{m}\rVert+\lVert\nabla_{\eta}\bar Q_{m}\rVert$ and the $\ell^{1}$--$\ell^{p}$ bound on three terms, the second from $\lVert\nabla_{\eta}S_{m}\rVert^{2}=\lvert\partial_{\gamma}S_{m}\rvert^{2}+\lvert\partial_{\lambda}S_{m}\rvert^{2}$ and the norm monotonicity $(a^{2}+b^{2})^{p/2}\le2^{p/2-1}(a^{p}+b^{p})$ on $\mathbb R^{2}$, and it holds for each of the three centered pieces. Taking expectations and inserting the componentwise bounds of Step B,
\[
    \sup_{\eta\in\Upsilon}\mathbb E\bigl\lVert\nabla_{\eta}\bar F_{m}(\eta)\bigr\rVert^{p}
    \le3^{p-1}\bigl(2^{3p/2}C_{U}+2^{p/2}C_{R}+2^{p/2}C_{Q}\bigr)\\
    \le3^{p-1}2^{p/2}\bigl(2^{p}C_{U}+C_{R}+C_{Q}\bigr)=:N_{p},
\]
the second row collecting terms and using $m^{-p/2}\le1$ for $m\ge2$.

The additive terms are excluded from the moment route by construction and are priced by the level. The drift $\mathrm{dr}_{m}$ is measurable with respect to $\mathbf X$, and the Drift modulus paragraph gives
\[
    \bigl\lvert\mathrm{dr}_{m}(\eta)-\mathrm{dr}_{m}(\eta')\bigr\rvert\le\Lambda\,\lVert\eta-\eta'\rVert+2\,r_{m},
    \qquad
    \Lambda:=4L_{\mathrm{dr}},\quad r_{m}:=2\,r^{\mathrm{dr}}_{m}=o_{p}(1):
\]
by the triangle inequality $\lvert\mathrm{dr}_{m}(\eta)-\mathrm{dr}_{m}(\eta')\rvert\le\lvert\mathrm{dr}_{m}(\eta)-\mathrm{dr}(\eta)\rvert+\lvert\mathrm{dr}(\eta)-\mathrm{dr}(\eta')\rvert+\lvert\mathrm{dr}(\eta')-\mathrm{dr}_{m}(\eta')\rvert$, the middle term is at most $(2+1+1)L_{\mathrm{dr}}\lVert\eta-\eta'\rVert$ by the Lipschitz display of the Drift modulus paragraph applied to $\mathrm{dr}:=-2\nu+\rho+V$, and $\sup_{\eta}\lvert\mathrm{dr}_{m}(\eta)-\mathrm{dr}(\eta)\rvert\le2\,r^{\mathrm{dr}}_{m}$, since $\mathrm{dr}_{m}-\mathrm{dr}=-2(\nu_{m}-\nu)+\bigl(\mathbb E[mQ_{m}\mid\mathbf X]-V\bigr)+(\rho_{m}-\rho)$ and $r^{\mathrm{dr}}_{m}$ is the supremum of the sum of the three absolute values. The remainder term obeys
\[
    \sup_{\eta\in\Upsilon}\bigl\lvert F_{m}(\eta)-\bar F_{m}(\eta)-\mathrm{dr}_{m}(\eta)\bigr\rvert
    \le2\sup_{\eta}\bigl\lvert\mathcal R_{m}(\eta)\bigr\rvert+2\sup_{\eta}\bigl\lvert\mathbb E[\mathcal R_{m}\mid\mathbf X]\bigr\rvert+o_{p}(1)=o_{p}(1),
\]
the two suprema being $O_{p}(m^{-1/2})$, the first by the uniform remainder bound of Lemma~\ref{lem:remainder_cv} and the second by the uniform conditional-mean bound of Lemma~\ref{lem:remainder_cv}, and the $o_{p}(1)$ being the uniform linearization remainder.

Lemma~\ref{lem:equi} now applies to $S_{m}:=F_{m}$ with $S^{c}_{m}:=\bar F_{m}$, $a_{m}:=\mathrm{dr}_{m}$ and $e_{m}:=F_{m}-\bar F_{m}-\mathrm{dr}_{m}$: hypothesis (i) is the gradient-moment display with $N_{p}=3^{p-1}2^{p/2}(2^{p}C_{U}+C_{R}+C_{Q})$, hypothesis (ii) is the drift display with $\Lambda=4L_{\mathrm{dr}}$ and $r_{m}=2r^{\mathrm{dr}}_{m}$, hypothesis (iii) is the remainder display, and the pointwise bound $F_{m}(\eta_{0})=O_{p}(1)$ is the finite-dimensional paragraph, whose display gives $\bigl(mU_{m}(\eta_{0}),mQ_{m}(\eta_{0}),mR_{m}(\eta_{0})\bigr)=O_{p}(1)$ and which combines with $F_{m}=-2mU_{m}+mQ_{m}+mR_{m}+o_{p}(1)$ at the uniform linearization remainder. Its conclusions are asymptotic equicontinuity,
\[
    \lim_{\delta\downarrow0}\ \limsup_{m\to\infty}\ \mathbb P\Bigl(\sup_{\lVert\eta-\eta'\rVert\le\delta}\bigl\lvert\Delta F_{m}\bigr\rvert>\varepsilon\Bigr)=0
    \qquad\text{for every }\varepsilon>0,
\]
and $\sup_{\eta\in\Upsilon}\lvert F_{m}(\eta)\rvert=O_{p}(1)$. The exponent surplus $p-d=p-2>0$ is exactly what drives the first display to zero as $\delta\downarrow0$, uniformly in $m$: the per-cell bound of Lemma~\ref{lem:sobolev} at $p>d=2$ is $C_{2,p}N_{p}h^{p}$ with $h=2\delta$, and the union runs over the $\lvert\mathcal G_{\delta}\rvert\le C_{\Upsilon}\delta^{-2}$ cells of Lemma~\ref{lem:shifts}, so that $\lvert\mathcal G_{\delta}\rvert h^{p}\le2^{p}\,C_{\Upsilon}\delta^{p-2}$.

Finite-dimensional convergence together with asymptotic equicontinuity and boundedness, both supplied by Lemma~\ref{lem:equi} in Step C, gives weak convergence in $\ell^{\infty}(\Upsilon)$ \citep[Th.~1.5.4]{vaart1997weak},

\begin{equation}
F_{m}\rightsquigarrow\mathbb H
  =-2\mathcal C^{\perp}+V+\mathcal D^{\perp}+\Delta
  \quad\text{in }\ell^{\infty}(\Upsilon),
  \tag{limit-null}\label{eq:limit-null}
\end{equation}
which is part (a) of Theorem~\ref{th:cv_uniform}: the projection is already accounted for in the limit element $G_{\perp,\gamma}$, the chaoses carrying the kernels $K^{\mathrm{sym},\perp},K^{D,\perp}$ by Lemma~\ref{lem:chaos_cv}, and the drift being $\Delta=\rho-2\nu$; the limit $\mathbb H$ has continuous sample paths by the sample-path continuity of the limit fields established in Appendix~\ref{app:limit_continuity}, which covers $\mathbb H$ explicitly.

\medskip\noindent\textit{Feasibility and bookkeeping.}\;
The limit \eqref{eq:limit-null} is stated at the infeasible level: Lemma~\ref{lem:repl_cost} transfers the $\eta$-dependent part, $\sup_{\eta}|m\,\mathrm{CV}^{\mathrm{dep}}_{m}(\eta;\hat{\boldsymbol\varepsilon}_{\perp})-m\,\mathrm{CV}^{\mathrm{dep}}_{m}(\eta;\boldsymbol\varepsilon^{0}_{\perp})|=o_{p}(1)$. The level difference between the two centerings,
\[
    \sum_{i=1}^{m}\bigl[\hat\varepsilon_{\perp,i}^{2}-(\varepsilon^{0}_{\perp,i})^{2}\bigr]
    =2\boldsymbol\varepsilon_{0}^{\top}\hat{\boldsymbol\Pi}\boldsymbol r_{\mathrm{est}}+\lVert\hat{\boldsymbol\Pi}\boldsymbol r_{\mathrm{est}}\rVert^{2},
    \qquad
    \boldsymbol r_{\mathrm{est}}:=\hat{\boldsymbol\varepsilon}-\boldsymbol\varepsilon_{0},
\]
is free of $\eta$ but not covered by Lemma~\ref{lem:repl_cost}: in the sieve regime of Assumption~\ref{ass:D}(iii) it is $O_{p}(m^{1/2-\varrho}\vee m^{1-2\varrho})$, vanishing for $\varrho>1/2$, the linear term being $O_{p}(\lVert\hat{\boldsymbol\Pi}\boldsymbol r_{\mathrm{est}}\rVert)$ by Chebyshev's inequality against the conditionally centered noise (Remark~\ref{rem:cond_det}); when the scores are $\mathbf X$-measurable it vanishes at the second-order rate of Fact~\ref{fact:residual_split}, since $\hat{\boldsymbol\Pi}\hat{\mathbb G}=\boldsymbol 0$ annihilates the linear part of $\boldsymbol r_{\mathrm{est}}$ exactly. Two bookkeeping points complete part (a). First, the limit \eqref{eq:limit-null} is computed directly on the projected residuals, the embedding CLT delivering $G_{\perp,\gamma}$ for $\varepsilon_{\perp}^{0}$, so the projection fluctuation $\hat{\boldsymbol\Pi}-\boldsymbol\Pi$ needs no separate bookkeeping: it is part of the limit and is reflected in the projected kernels $K^{\mathrm{sym},\perp},K^{D,\perp}$ and in the drift $\Delta$; there is no $\eta$-constant shift to subtract. Second, the level term $\sum_{i}(\varepsilon^{0}_{\perp,i})^{2}$ removed by the centering is constant in $\eta$: $\eta$-free additive shifts leave the argmin location and every argmin-based conclusion unchanged, though in the sieve regime they displace the centered process itself by a nonvanishing amount; the conclusions drawn from \eqref{eq:limit-null}, in particular the selector's argmin over $\Upsilon$, therefore never depend on the level at which the centering is taken. This proves part (a).

\medskip\noindent\textbf{Proof of (b): fixed alternative.}\;

\emph{Road map.} The proof proceeds in four moves and records three lemmas. (1)~\emph{One linearization:} the fixed-alternative clause of Lemma~\ref{lem:loo_prediction} bounds the leave-one-out remainder by $r_{i}=O_{p}(m^{-1/2})$, uniformly in $(i,\eta)$, so the algebra of Lemma~\ref{lem:loo_linearization}, applied once to the full orthogonalized residual, represents $\mathrm{CV}_{m}(\eta)$ as the level $m^{-1}\sum_{i}\varepsilon_{\perp,i}^{2}$ plus the three quadratic statistics with a remainder $O_{p}(m^{-1/2})$ uniform in $\eta$. (2)~\emph{One channel ledger:} the split $\varepsilon_{\perp,i}=s(X_{i})+r(X_{i})+v_{i}$ turns each statistic into a sum of pairings of the signal, estimation-error, and noise components; the ledger routes every pairing to a deterministic constant collected into $R_{\infty}(\eta)$, to a centered bounded average collected into the pooled family $\psi_{\eta}$ of \eqref{eq:pooled-psi-b}, to the drift term $\mathrm{Drift}_{m}(\eta)$, or to an $O_{p}(m^{-1})$ remainder, and yields the collection \eqref{eq:level-collect} of the criterion around its risk. (3)~\emph{Two uniform laws of large numbers:} Lemma~\ref{lem:fam_psi} fixes the deterministic centering, moment, and smoothness constants of the pooled family, and Lemma~\ref{lem:coda_unif} converts them into $\sup_{\eta}\bigl\lvert m^{-1}\sum_{i}\psi_{\eta}(X_{i},v_{i})\bigr\rvert=O_{p}(m^{-1/2})$ by a fixed mesh, Chebyshev's inequality, the Marcinkiewicz--Zygmund inequality, and the Sobolev maximal inequality of Lemma~\ref{lem:sobolev}; no central limit argument is consumed at this stage. (4)~\emph{Assembly:} the triangle inequality over \eqref{eq:level-collect}, the drift bound \eqref{eq:driftm}, and the replacement cost gives the level statement \eqref{eq:level-final}. Lemma~\ref{lem:pooled_clt}, recorded after Lemma~\ref{lem:coda_unif}, upgrades the same two inputs to the process-level fluctuation law; it is not consumed by the level statement and is stated for use by the proof of part~(c) and by the remark at the end of the proof.

\emph{Setup.} Write the residual at the true parameter as in the Introduction,
\(\varepsilon_{0}=f(X)+v\), with \(f(X):=\mathbb{E}[\varepsilon_{0}\mid X]\)
deterministic, \(f\neq0\), \(\|f\|_{L^{2}}<\infty\), and
\(v:=\varepsilon_{0}-f(X)\), \(\mathbb{E}[v\mid X]=0\). The
orthogonalized residual has effective signal \(\Pi f\). For part (b)
assume \(\|\Pi f\|_{\infty}<\infty\). Throughout part (b) the noise moments are those of Assumption~\ref{ass:G}(v), the noise being $v$: $\sup_{x}\mathbb E[v^{2p}\mid X=x]\le\bar M_{2p}$; write $\bar M_{4}:=\bar M_{2p}^{2/p}$ for the interpolated conditional fourth moment, $\sup_{x}\mathbb E[(v^{\rm raw})^{4}\mid X=x]\le\bar M_{4}$, finite and deterministic by the Lyapunov interpolation of the $(2p)$-clause at $q=4$ (as in the proof of Lemma~\ref{lem:remainder_cv}), whose Marcinkiewicz--Zygmund linear-form bound passes the bound to the projected component $v_{k}=\sum_{a}\hat{\boldsymbol\Pi}_{ka}v^{\rm raw}_{a}$ at $q=4$; every fourth-moment constant below depends only on $(p,\bar M_{2p})$. Write
$r:=\hat{\boldsymbol\Pi}\boldsymbol f-\boldsymbol\Pi f$ for the difference between the empirical and population projections of the signal, with $\boldsymbol f:=(f(X_{1}),\dots,f(X_{m}))^{\top}$ the raw-signal evaluations and $\hat{\boldsymbol\Pi}=\boldsymbol I_{m}-\hat{\mathbb G}M\hat{\mathbb G}^{\top}$, $M:=(\hat{\mathbb G}^{\top}\hat{\mathbb G})^{-1}$, so that
$\varepsilon_{\perp,i}=\Pi f(X_{i})+r(X_{i})+v^{\rm raw}_{i}+\mathbf d_{i}=\Pi f(X_{i})+r(X_{i})+v_{i}$ uniformly in $i$,
where $v^{\rm raw}_{i}$ is the raw noise, $v_{i}=[\hat{\boldsymbol\Pi}\boldsymbol v^{\rm raw}]_{i}$ is the noise component of the orthogonalized residual, and $\mathbf d:=\boldsymbol v-\boldsymbol v^{\rm raw}=-(\boldsymbol I_{m}-\hat{\boldsymbol\Pi})\boldsymbol v^{\rm raw}$ is the replacement defect; $\lVert r\rVert_{L^{2}(\mathbb P_{m})}=O_{p}(m^{-1/2})$ by Section~\ref{sec:uniform_limit}. Put \(w^{*}_{\eta}=P_{\eta}(\Pi f)\),
\(e_{\eta}=\Pi f-w^{*}_{\eta}\), \(A(\eta)=\|e_{\eta}\|_{L^{2}}^{2}\),
\(\bar\sigma^{2}=\mathbb{E}[v^{2}]\), \(R_{\infty}=\bar\sigma^{2}+A\). Put \(a_{\eta}=e_{\eta}\), \(b_{\eta}=P_{\eta}a_{\eta}\), so

\begin{equation}
a_{\eta}-b_{\eta}=(I-P_{\eta})^{2}\Pi f,\qquad
(a_{\eta}-b_{\eta})_{\ell}
  =\frac{\lambda^{2}}{(\mu_{\ell}+\lambda)^{2}}\,
   (\Pi f)_{\ell}(\gamma),
\label{eq:ab}
\end{equation}
with
\((\Pi f)_{\ell}(\gamma)=\langle\Pi f,\phi_{\ell}(\cdot;\gamma)\rangle\)
and \(0\le\lambda^{2}/(\mu_{\ell}+\lambda)^{2}\le1\). Two further population objects enter. The anchor of the composed kernel, in the sense of \eqref{eq:g_eta}, is
\[
\kappa_{\eta}(x):=\mathbb E_{Z}\bigl[g_{\eta}(x,Z)\,w^{*}_{\eta}(Z)\bigr]
=\bigl(P_{\eta}^{2}\,\Pi f\bigr)(x),
\qquad
(\kappa_{\eta})_{\ell}=\frac{\mu_{\ell}^{2}}{(\mu_{\ell}+\lambda)^{2}}\,(\Pi f)_{\ell}(\gamma):
\]
the first equality inserts the kernel of $S_{\lambda}$ and exchanges the expectations, $\mathbb E[w^{*}_{\eta}(Z)\phi_{\ell}(Z)]=\frac{\mu_{\ell}}{\mu_{\ell}+\lambda}(\Pi f)_{\ell}(\gamma)$, the second reads off the multiplier $t_{\ell}^{2}$ of $S^{2}_{\lambda}$, with $t_{\ell}:=\frac{\mu_{\ell}(\gamma)}{\mu_{\ell}(\gamma)+\lambda}$. The anchors are uniformly bounded: $\lvert w^{*}_{\eta}(x)\rvert\le\bar k^{1/2}\lVert w^{*}_{\eta}\rVert_{\mathcal H_{\gamma}}\le(4\lambda_{L})^{-1/2}\bar k^{1/2}\lVert\Pi f\rVert_{L^{2}}$ and $\lVert\kappa_{\eta}\rVert_{\mathcal H_{\gamma}}\le\lVert w^{*}_{\eta}\rVert_{\mathcal H_{\gamma}}$, the filter $P_{\eta}$ being an $\mathcal H_{\gamma}$-contraction and the spectral norm having multiplier $t^{2}_{\ell}/\mu_{\ell}=\mu_{\ell}/(\mu_{\ell}+\lambda)^{2}\le(4\lambda)^{-1}\le(4\lambda_{L})^{-1}$, so $\lVert w^{*}_{\eta}\rVert_{\mathcal H_{\gamma}}\le(4\lambda_{L})^{-1/2}\lVert\Pi f\rVert_{L^{2}(\mathbb P_{X})}$; likewise $\lvert b_{\eta}(x)\rvert\le\lVert a_{\eta}\rVert_{L^{2}}\lVert g_{\eta}(x,\cdot)\rVert_{L^{2}}\le\bar k^{1/2}\lVert\Pi f\rVert_{L^{2}}$. \emph{Replacement of the noise component.} Every pairing of the projected noise $v_{i}$ below survives, at $o_{p}(1)$ cost after the $\sqrt m$ normalization and uniformly in $\eta$, the replacement of $v_{i}$ by the raw noise $v^{\rm raw}_{i}$: $\lVert\mathbf d\rVert^{2}=(\boldsymbol v^{\rm raw})^{\top}(\boldsymbol I_{m}-\hat{\boldsymbol\Pi})\boldsymbol v^{\rm raw}$ has $\mathbb E[\lVert\mathbf d\rVert^{2}\mid\mathbf X]=\operatorname{tr}\{\boldsymbol\Sigma(\boldsymbol I_{m}-\hat{\boldsymbol\Pi})\}\le\sigma^{2}_{\max}d_{\theta}$, with $\boldsymbol\Sigma:=\operatorname{diag}(\sigma^{2}(X_{1}),\dots,\sigma^{2}(X_{m}))$, so $\lVert\mathbf d\rVert=O_{p}(1)$; in particular $\boldsymbol v^{\top}\mathbf d=0$ by $\boldsymbol v=\hat{\boldsymbol\Pi}\boldsymbol v^{\rm raw}$ and $\hat{\boldsymbol\Pi}(\boldsymbol I_{m}-\hat{\boldsymbol\Pi})=\boldsymbol 0$. The three replacement pairings are bounded sharply. The level square is exact, $m^{-1}\sum_{i}(v_{i}^{2}-v^{\rm raw\,2}_{i})=m^{-1}\bigl(\lVert\boldsymbol v\rVert^{2}-\lVert\boldsymbol v^{\rm raw}\rVert^{2}\bigr)=-\lVert\mathbf d\rVert^{2}/m=O_{p}(m^{-1})$. The signal cross-pairing closes through the exact identity $(\boldsymbol I_{m}-\hat{\boldsymbol\Pi})(\boldsymbol s+\boldsymbol r)=\boldsymbol 0$, where $\boldsymbol s:=(\Pi f(X_{1}),\dots,\Pi f(X_{m}))^{\top}$ and $\boldsymbol s+\boldsymbol r=\hat{\boldsymbol\Pi}\boldsymbol f$, which gives $m^{-1}\boldsymbol s^{\top}\mathbf d=-m^{-1}(\boldsymbol P_{\hat S}\boldsymbol s)^{\top}\boldsymbol v^{\rm raw}=m^{-1}(\boldsymbol P_{\hat S}\boldsymbol r)^{\top}\boldsymbol v^{\rm raw}$; this is conditionally centered with conditional variance $m^{-2}(\boldsymbol P_{\hat S}\boldsymbol r)^{\top}\boldsymbol\Sigma(\boldsymbol P_{\hat S}\boldsymbol r)\le\sigma^{2}_{\max}m^{-2}\lVert\boldsymbol r\rVert^{2}\le\sigma^{2}_{\max}m^{-2}\,m\lVert\boldsymbol r\rVert^{2}_{L^{2}(\mathbb P_{m})}=O_{p}(m^{-2})$, hence $m^{-1}\boldsymbol s^{\top}\mathbf d=O_{p}(m^{-1})$ by Chebyshev's inequality, within the unit-scale budget below; the quadratic defect $\lVert\mathbf d\rVert^{2}/m=O_{p}(m^{-1})$ is absorbed at the same order through the trace bound above, and no second-order term survives the split. The linear $v$-borne pairings close conditionally at conditional mean zero and conditional variance $m^{-2}\boldsymbol h_{\eta}^{\top}(\boldsymbol I_{m}-\hat{\boldsymbol\Pi})\boldsymbol\Sigma(\boldsymbol I_{m}-\hat{\boldsymbol\Pi})\boldsymbol h_{\eta}\le\sigma^{2}_{\max}m^{-2}\lVert(\boldsymbol I_{m}-\hat{\boldsymbol\Pi})\boldsymbol h_{\eta}\rVert^{2}=O_{p}(m^{-2})$ uniformly in $\eta$: writing $\boldsymbol P_{\hat S}:=\boldsymbol I_{m}-\hat{\boldsymbol\Pi}$ and $\boldsymbol P_{S}:=\boldsymbol I_{m}-\boldsymbol\Pi$ for the empirical and population score-span projections; here $\boldsymbol P_{S}\boldsymbol g$ denotes the design evaluation of the population projection $g\mapsto g-\Pi g$, so that $\boldsymbol P_{S}\boldsymbol h_{\eta}=\boldsymbol 0$ exactly, $h_{\eta}\in S^{\perp}$ making $P_{S}h_{\eta}=0$ pointwise, and $\boldsymbol r=(\boldsymbol P_{S}-\boldsymbol P_{\hat S})\boldsymbol f$ follows from $\boldsymbol s=(\boldsymbol I_{m}-\boldsymbol P_{S})\boldsymbol f$ together with the exact identity $(\boldsymbol I_{m}-\hat{\boldsymbol\Pi})(\boldsymbol s+\boldsymbol r)=\boldsymbol 0$; the free insertion of $\Pi$ displayed below gives $h_{\eta}\in S^{\perp}$, so $\lVert(\boldsymbol I_{m}-\hat{\boldsymbol\Pi})\boldsymbol h_{\eta}\rVert_{L^{2}(\mathbb P_{m})}\le\lVert\boldsymbol P_{\hat S}-\boldsymbol P_{S}\rVert_{\mathrm{op}}\,\lVert h_{\eta}\rVert_{L^{2}(\mathbb P_{m})}$ with $\boldsymbol r=(\boldsymbol P_{S}-\boldsymbol P_{\hat S})\boldsymbol f$ and $\lVert\boldsymbol P_{\hat S}-\boldsymbol P_{S}\rVert_{\mathrm{op}}=O_{p}(m^{-1/2})$, the projection-perturbation rate induced by $\hat{\mathbb G}^{\top}\hat{\mathbb G}/m\xrightarrow{p}\Gamma_{\theta_{0}}$ (Assumption~\ref{ass:E}(ii)), consistent with the rate of $\lVert\boldsymbol r\rVert_{L^{2}(\mathbb P_{m})}$. We may therefore read $v_{i}$ as $v^{\rm raw}_{i}$ wherever the i.i.d.\ structure of the summands is invoked, and we drop the distinction from the notation. One further convention. Whenever a statistic $T$ satisfies $\mathbb E[T^{2}\mid\mathbf X]\le W$ for an $\mathbf X$-measurable $W$ with $W=O_{p}(s_{m}^{2})$ at a deterministic scale $s_{m}$, then $T=O_{p}(s_{m})$: split on the event $\{W\le Cs_{m}^{2}\}$, apply conditional Chebyshev's inequality at the level $Ks_{m}$ to get $\mathbb P(\lvert T\rvert>Ks_{m}\mid\mathbf X)\le C/K^{2}$, let $K\uparrow\infty$ and then $C\uparrow\infty$. This variance-event splitting is invoked below as such.

\emph{Leave-one-out remainder.}\;
This paragraph closes the linearization remainder at $O_{p}(m^{-1/2})$ uniformly in $\eta$, which is what allows the proof to run on a single linearization. The coefficient averages of the orthogonalized residual are bounded,
\[
\frac{1}{m}\sum_{j=1}^{m}\lvert\varepsilon_{\perp,j}\rvert=O_{p}(1),
\]
by Jensen's inequality against $\frac{1}{m}\sum_{j}\varepsilon_{\perp,j}^{2}=O_{p}(1)$, which follows from the split $\varepsilon_{\perp,i}=s(X_{i})+r(X_{i})+v_{i}$ of the setup, the bound $\lvert\Pi f\rvert\le\lVert\Pi f\rVert_{\infty}$, the perturbation rate $\lVert r\rVert_{L^{2}(\mathbb P_{m})}=O_{p}(m^{-1/2})$, and $\frac{1}{m}\sum_{j}v_{j}^{2}=O_{p}(1)$ read at the raw noise. The fixed-alternative clause of Lemma~\ref{lem:loo_prediction}, stated there at the oracle residual and extending verbatim to the feasible one, its stochastic inputs being the design concentration of the leave-one-out covariance perturbation and the bounded coefficient averages just displayed, therefore gives
\[
\max_{i\le m}\,\sup_{\eta\in\Upsilon}\,\lvert r_{i}\rvert=O_{p}(m^{-1/2}).
\]
Applying Lemma~\ref{lem:loo_linearization} once to the full orthogonalized residual, the algebra of its proof, expanding the square and collecting the three $\eta$-dependent sums into $U_{m},Q_{m},R_{m}$ and leaving the remainder $-\frac{2}{m}\sum_{i}(\varepsilon_{\perp,i}-L_{i})r_{i}+\frac{1}{m}\sum_{i}r_{i}^{2}$, is an identity at every realized sample, and the two remainder sums are pathwise
\[
\begin{aligned}
\Bigl\lvert\frac{2}{m}\sum_{i=1}^{m}(\varepsilon_{\perp,i}-L_{i})\,r_{i}\Bigr\rvert
&\le2\Bigl(\max_{j}\lvert L_{j}\rvert+\frac{1}{m}\sum_{j}\lvert\varepsilon_{\perp,j}\rvert\Bigr)\max_{i}\lvert r_{i}\rvert
=O_{p}(m^{-1/2}),\\
\frac{1}{m}\sum_{i=1}^{m}r_{i}^{2}&\le\max_{i}r_{i}^{2}=O_{p}(m^{-1}),
\end{aligned}
\]
by the triangle and Cauchy--Schwarz inequalities, with
$\lvert L_{j}(\eta)\rvert\le\bar k\lambda_{L}^{-1}\,\frac{1}{m-1}\sum_{i\ne j}\lvert\varepsilon_{\perp,i}\rvert=O_{p}(1)$
uniformly in $(j,\eta)$. Hence, uniformly in $\eta$,
\begin{equation}
\mathrm{CV}_{m}(\eta)
=\frac{1}{m}\sum_{i=1}^{m}\varepsilon_{\perp,i}^{2}
-2U_{m}(\eta)+Q_{m}(\eta)+R_{m}(\eta)+O_{p}(m^{-1/2}),
\label{eq:cv-lin-b}
\end{equation}
where $U_{m}(\eta)=\frac{1}{m(m-1)}\sum_{i\ne j}\varepsilon_{\perp,i}\varepsilon_{\perp,j}g_{\eta}(X_{i},X_{j})$ and $Q_{m},R_{m}$ are as displayed in Section~\ref{sec:uniform_limit}. The weaker order carries through unchanged: the assembly below consumes the collection only at the $O_{p}(m^{-1/2})$ scale.

\emph{Anchor decomposition of the signal-side channels.}\;
For a symmetric kernel $\mathfrak g_{\eta}\in\{g_{\eta},(g_{\eta}\circ g_{\eta})_{\eta}\}$ write
$U^{(\mathfrak g)}_{m}(\eta):=\frac{1}{m(m-1)}\sum_{i\ne j}\Pi f(X_{i})\,\Pi f(X_{j})\,\mathfrak g_{\eta}(X_{i},X_{j})$,
$\vartheta_{\mathfrak g}(x):=\mathbb E_{Y}\bigl[\mathfrak g_{\eta}(x,Y)\Pi f(Y)\bigr]$, and
$\theta_{\mathfrak g}(\eta):=\mathbb E\bigl[\Pi f(X)\vartheta_{\mathfrak g}(X)\bigr]$. The anchors are exactly the two population objects above: $\vartheta_{g}=w^{*}_{\eta}$, since $\mathbb E_{Y}[g_{\eta}(x,Y)\Pi f(Y)]$ applies the kernel of $S_{\lambda}$ to the coordinate sequence of $\Pi f$; and $\vartheta_{g\circ g}=\kappa_{\eta}$, since
$\mathbb E_{Y}[(g_{\eta}\circ g_{\eta})(x,Y)\Pi f(Y)]=\mathbb E_{Z}\bigl[g_{\eta}(x,Z)\,\mathbb E_{Y}[g_{\eta}(Z,Y)\Pi f(Y)]\bigr]=\mathbb E_{Z}[g_{\eta}(x,Z)w^{*}_{\eta}(Z)]$.
Introduce the centered, doubly degenerate kernel
\[
\hat h_{\eta}(x,y):=\mathfrak g_{\eta}(x,y)\,\Pi f(x)\,\Pi f(y)
-\Pi f(x)\,\vartheta_{\mathfrak g}(x)-\Pi f(y)\,\vartheta_{\mathfrak g}(y)+\theta_{\mathfrak g}(\eta),
\qquad
\mathbb E_{X}\bigl[\hat h_{\eta}(x,X)\bigr]=0 :
\]
the first and second terms of the expectation cancel by the definition of $\vartheta_{\mathfrak g}$, the third and fourth by the definition of $\theta_{\mathfrak g}$. The exact decomposition, by counting,
\begin{equation}
U^{(\mathfrak g)}_{m}(\eta)-\theta_{\mathfrak g}(\eta)
=\frac{2}{m}\sum_{i=1}^{m}\Bigl[\Pi f(X_{i})\,\vartheta_{\mathfrak g}(X_{i})-\theta_{\mathfrak g}(\eta)\Bigr]
+\underbrace{\frac{1}{m(m-1)}\sum_{i\ne j}\hat h_{\eta}(X_{i},X_{j})}_{=:\ \mathrm D^{(\mathfrak g)}_{m}(\eta)},
\label{eq:anchor-id}
\end{equation}
follows since $\frac{1}{m(m-1)}\sum_{i\ne j}\Pi f(X_{i})\vartheta_{\mathfrak g}(X_{i})=\frac{1}{m}\sum_{i}\Pi f(X_{i})\vartheta_{\mathfrak g}(X_{i})$, the index $i$ contributing $m-1$ terms, and likewise with $i$ in the second argument, while $\frac{1}{m(m-1)}\sum_{i\ne j}\theta_{\mathfrak g}=\theta_{\mathfrak g}$. All summands are bounded: $\lvert\hat h_{\eta}\rvert\le C^{(0)}_{\mathrm{an}}$, where $C^{(0)}_{\mathrm{an}}:=\bar k\,\lambda_{L}^{-1}\lVert\Pi f\rVert_{\infty}^{2}+2\,(4\lambda_{L})^{-1/2}\bar k^{1/2}\lVert\Pi f\rVert_{\infty}\lVert\Pi f\rVert_{L^{2}}+\lVert\Pi f\rVert_{L^{2}}^{2}$ for $\mathfrak g_{\eta}=g_{\eta}$, and $C^{(0)}_{\mathrm{an}}$ with $\bar k\,(4\lambda_{L})^{-1}$ in place of $\bar k\lambda_{L}^{-1}$ for the composed kernel: the composed kernel is that of $R_{\lambda}C_{\gamma}R_{\lambda}$, an operator of norm $\sup_{\ell}\mu_{\ell}/(\mu_{\ell}+\lambda)^{2}\le(4\lambda_{L})^{-1}$, so $\lvert(g_{\eta}\circ g_{\eta})(x,y)\rvert\le\bar k\,(4\lambda_{L})^{-1}$ against sections of norm $\bar k^{1/2}$. The degenerate remainder is $O_{p}(m^{-1})$ uniformly in $\eta$, in three steps.

First, the pointwise second moment. Since $\mathbb E[\hat h_{\eta}(x,X_{j})\mid X_{j}]=0$, a product $\mathbb E[\hat h_{ij}\hat h_{i'j'}]$ vanishes whenever one of the four indices is private to its factor, the expectation over that coordinate being taken against a centered anchor; the survivors are the diagonal $i=i',j=j'$ and the transposed $i=j',j=i'$, whence
\[
\mathbb E\bigl[\mathrm D^{(\mathfrak g)}_{m}(\eta)^{2}\bigr]
=\frac{2\,\mathbb E\bigl[\hat h_{\eta}(X_{1},X_{2})^{2}\bigr]}{m(m-1)}
\le\frac{2\,C^{(0)\,2}_{\mathrm{an}}}{m(m-1)} ,
\qquad
\max_{\eta_{l}\in\mathcal M}\bigl\lvert\mathrm D^{(\mathfrak g)}_{m}(\eta_{l})\bigr\rvert=O_{p}(m^{-1})
\]
at the points of any fixed mesh $\mathcal M$ by Chebyshev's inequality.

Second, the paths and the gradient envelope. Term-by-term differentiation is justified by dominated convergence through deterministic envelopes: Step A of the proof of part (a) gives $\sup_{\eta,x,y}\lvert\partial_{\eta}\mathfrak g_{\eta}(x,y)\rvert\le G_{1}\vee G_{2}$ entrywise, and the anchor functions, represented as $\vartheta_{\mathfrak g}(x)=\mathbb E_{Y}[\mathfrak g_{\eta}(x,Y)\,\Pi f(Y)]$, inherit the envelope through dominated convergence and Cauchy--Schwarz,
\[
\begin{aligned}
\bigl\lvert\nabla_{\eta}\vartheta_{\mathfrak g}(x)\bigr\rvert
&\le\mathbb E_{Y}\bigl[\lvert\nabla_{\eta}\mathfrak g_{\eta}(x,Y)\rvert\,\lvert\Pi f(Y)\rvert\bigr]
\le(G_{1}\vee G_{2})\,\lVert\Pi f\rVert_{L^{2}},\\
\bigl\lVert\nabla_{\eta}\vartheta_{\mathfrak g}\bigr\rVert_{L^{2}}&\le(G_{1}\vee G_{2})\,\lVert\Pi f\rVert_{L^{2}},
\end{aligned}
\]
the second bound by quadrature of the first over $x$. The population means differentiate under the same envelopes, $\lvert\nabla_{\eta}\theta_{\mathfrak g}(\eta)\rvert\le\lVert\Pi f\rVert_{L^{2}}\lVert\nabla_{\eta}\vartheta_{\mathfrak g}\rVert_{L^{2}}\le(G_{1}\vee G_{2})\,\lVert\Pi f\rVert_{L^{2}}^{2}$. Hence $\nabla_{\eta}\hat h_{\eta}$ exists with $\lvert\nabla_{\eta}\hat h_{\eta}(x,y)\rvert\le C_{\mathrm{an}}$ for a deterministic $C_{\mathrm{an}}$, the gradient inherits the double degeneracy by dominated convergence,
\[
\nabla_{\eta}\mathrm D^{(\mathfrak g)}_{m}(\eta)
=\frac{1}{m(m-1)}\sum_{i\ne j}\nabla_{\eta}\hat h_{\eta}(X_{i},X_{j}),
\qquad
\mathbb E_{X}\bigl[\nabla_{\eta}\hat h_{\eta}(x,X)\bigr]=\nabla_{\eta}\,\mathbb E_{X}\bigl[\hat h_{\eta}(x,X)\bigr]=0 ,
\]
and $\mathrm D^{(\mathfrak g)}_{m}$ has $C^{1}$ paths on $\Upsilon$.

Third, the uniform upgrade by the Sobolev route. In the fourfold product $\mathbb E\prod_{l=1}^{4}\nabla_{\eta}\hat h_{\eta}(X_{i_{l}},X_{j_{l}})$, every design index must occur at least twice among the eight slots, or the factor carrying it privately is centered by the double degeneracy and kills the expectation; the induced partition of the eight slots has all blocks of size at least two, there are at most $\mathrm B_{8}=4140$ such partition patterns, each carries at most $m^{4}$ index assignments, and each assignment contributes at most $C^{4}_{\mathrm{an}}$. The same count with $\hat h_{\eta}$ in place of $\nabla_{\eta}\hat h_{\eta}$ gives $\mathbb E[\mathrm D^{(\mathfrak g)}_{m}(\eta)^{4}]\le C\,m^{-4}$. Therefore, with $C_{\mathrm D}:=16\,\mathrm B_{8}\,C^{4}_{\mathrm{an}}$,
\[
\sup_{\eta\in\Upsilon}\mathbb E\bigl[\bigl\lvert\nabla_{\eta}\mathrm D^{(\mathfrak g)}_{m}(\eta)\bigr\rvert^{4}\bigr]\le\frac{C_{\mathrm D}}{m^{4}}
\qquad\text{for }m\ge2,
\]
the constant being $\mathbf X$-free and $m$-uniform since $(m-1)^{4}\ge m^{4}/16$. Lemma~\ref{lem:sobolev} with $p=4>d=2$ on each cell of a fixed mesh gives cell oscillations $\max_{K}\sup_{K}\lvert\Delta\mathrm D^{(\mathfrak g)}_{m}\rvert=O_{p}(\delta\,m^{-1})$, and the mesh level is $O_{p}(m^{-1})$ by the pointwise second moment, so
\begin{equation}
\sup_{\eta\in\Upsilon}\Bigl\lvert\mathrm D^{(g\circ g)}_{m}(\eta)-2\,\mathrm D^{(g)}_{m}(\eta)\Bigr\rvert=O_{p}(m^{-1}).
\label{eq:Dm-unif}
\end{equation}

\emph{The three channels.}\;
Collect \eqref{eq:cv-lin-b} channel by channel, writing $s$ for $\Pi f$. The split $\varepsilon_{\perp,i}=s(X_{i})+r(X_{i})+v_{i}$ turns each statistic of \eqref{eq:cv-lin-b} into a sum of \emph{pairings}, and this subsection routes every pairing to one of four destinations. The level is the square of a trinomial: it contributes the three self-pairings $s^{2}$, $r^{2}$, $v^{2}$ and the three cross-pairings $2sr$, $2sv$, $2rv$, each summed over $i$ with weight $m^{-1}$. The statistics $-2U_{m}$ and $R_{m}$ are bilinear, and by the exact reduction of Lemma~\ref{lem:bilinear_cv}(ii) stated below, $R_{m}$ equals $\frac{m-2}{m-1}\,U^{(g\circ g)}_{m}$ up to a remainder; each bilinear summand carries one copy of each of the nine coefficient pairings $a_{i}b_{j}$ with $a,b\in\{s,r,v\}$, through the kernel $g_{\eta}$ for $U_{m}$ and the composed kernel $(g_{\eta}\circ g_{\eta})$ for $U^{(g\circ g)}_{m}$:
\[
\underbrace{(s,s)}_{\text{mean, linear}}
\quad
\underbrace{(s,v),\,(v,s)}_{\text{linear in }v}
\quad
\underbrace{(s,r),\,(r,s)}_{\text{linear in }r}
\quad
\underbrace{(v,v),\,(v,r),\,(r,v),\,(r,r)}_{\text{die at }\sqrt{m}}
\]
and $Q_{m}$ carries the squared coefficients $\varepsilon_{\perp,j}^{2}$ and dies by normalization alone. The routing rule, applied to this list and used without further comment below, reads: \emph{(i)} a pairing with a nonzero mean contributes a deterministic constant, collected into $R_{\infty}(\eta)$; \emph{(ii)} a pairing linear in the noise is an average of $m$ centered bounded terms of size $O_{p}(m^{-1/2})$ and enters the pooled family \eqref{eq:pooled-psi-b}; \emph{(iii)} a pairing linear in the estimation error has the same size and is collected into the drift term $\mathrm{Drift}_{m}$ of \eqref{eq:level-collect}; \emph{(iv)} the remaining pairings are $O_{p}(m^{-1})$ and die, since $\sqrt m\,O_{p}(m^{-1})=o_{p}(1)$: $(v,v)$ is centered once the projected noise is read as raw noise, at replacement cost $O_{p}(m^{-1})$, $(v,r)$ and $(r,v)$ carry the rate $\lVert r\rVert_{L^{2}(\mathbb P_{m})}=O_{p}(m^{-1/2})$ against their normalizations, and $(r,r)$ carries $\lVert r\rVert^{2}_{L^{2}(\mathbb P_{m})}=O_{p}(m^{-1})$ against an $m^{-2}$ prefactor. The pairing ledger, class by class; for each class, the carriers, the mechanism, the order at the criterion scale, and the destination. \emph{Class $(s,s)$.} Carriers: $U^{(g)}_{m}$ and $R_{m}$'s $\theta$-anchored part $\frac{m-2}{m-1}\,U^{(g\circ g)}_{m}$, both handled by the anchor decomposition displayed above (the doubly degenerate kernel $\hat h_{\eta}$, remainders $\mathrm D^{(\mathfrak g)}_{m}=O_{p}(m^{-1})$, linear projections into the drift term); and $\mathcal R_{m}$'s own $(s,s)$ pairing, the only one that does not die, whose first Hoeffding projection \eqref{eq:Rcal-ss-final} is routed to the pooled family \eqref{eq:pooled-psi-b} while its two chaos components are $O_{p}(m^{-1})$. \emph{Class $(s,v)+(v,s)$.} Carriers: $U_{m}$, contributing $-\frac{4}{m}\sum_{j}v_{j}\,\frac{1}{m-1}\sum_{i\ne j}s(X_{i})\,g_{\eta}(X_{i},X_{j})$ through the kernel $g_{\eta}$; $R_{m}$'s $\theta$-anchored part, contributing $+\frac{2}{m}\sum_{j}v_{j}\,\frac{1}{m-1}\sum_{i\ne j}s(X_{i})\,(g_{\eta}\circ g_{\eta})(X_{i},X_{j})$; and the level's $2sv$. The anchor-fluctuation count displayed in the signal--noise channel replaces the empirical anchors by the population anchors $w^{*}_{\eta}(X_{j})$, $\kappa_{\eta}(X_{j})$ at cost $O_{p}(m^{-1})$, and the collapse merges the bracket with the level's $2sv$ into the noise-borne linear process \eqref{eq:mixed-b}: $O_{p}(m^{-1/2})$, pooled family. $\mathcal R_{m}$'s $(s,v)+(v,s)$, the $\Gamma^{(m)}\boldsymbol s$ pairing, is $O_{p}(m^{-1})$ and dies. \emph{Class $(s,r)+(r,s)$.} Carriers: the level, $2\langle s,r\rangle_{L^{2}(\mathbb P_{m})}$; $U_{m}$, leaving $-4\langle w^{*}_{\eta},r\rangle_{L^{2}(\mathbb P_{m})}+O_{p}(m^{-1})$; $R_{m}$'s $\theta$-anchored part, leaving $+2\langle\kappa_{\eta},r\rangle_{L^{2}(\mathbb P_{m})}+O_{p}(m^{-1})$; and $\mathcal R_{m}$, $O_{p}(1)$ at the $\mathcal R_{m}$ scale and $O_{p}(m^{-1})$ at the criterion scale. All fold, through the collapse, into the drift term $\mathrm{Drift}_{m}(\eta)$ of \eqref{eq:level-collect}, whose surviving part is $2\langle a_{\eta}-b_{\eta},\,r\rangle_{L^{2}(\mathbb P_{m})}+O_{p}(m^{-1})$: $O_{p}(m^{-1/2})$ by Cauchy--Schwarz against $C_{\rm dr}$ and the perturbation rate, dead at the budget, no central limit theorem consumed. \emph{Class $(v,v)$.} $U_{m}$ and $R_{m}$'s $\theta$-anchored part combine as $-2U^{(v,v)}_{m}+U^{(v,v,g\circ g)}_{m}$: centered after the raw-noise replacement at cost $O_{p}(m^{-1})$, both pairings $O_{p}(m^{-1})$ uniformly in $\eta$ and dead at the $\sqrt m$ scale; $\mathcal R^{(v,v)}_{m}/m=O_{p}(m^{-1})$. The level's $v^{2}$ is kept, read as raw noise at cost $O_{p}(m^{-1})$; $Q_{m}$ is dead in its entirety. \emph{Class $(v,r)+(r,v)$.} The level's $2rv$ is $O_{p}(m^{-1})$ by centering; $U_{m}$ and $R_{m}$'s $\theta$-anchored part are $O_{p}(m^{-1})$ by conditional Cauchy--Schwarz against the perturbation rate; $\mathcal R_{m}$'s pair is $O_{p}(m^{-3/2})$. \emph{Class $(r,r)$.} The level's $r^{2}$ is $O_{p}(m^{-1})$; $U_{m}$ and $R_{m}$'s $\theta$-anchored part are $O_{p}(m^{-1})$ by $(\sum_{i}\lvert r(X_{i})\rvert)^{2}\le m^{2}\lVert r\rVert^{2}_{L^{2}(\mathbb P_{m})}$ against the $m^{-2}$ prefactor; $\mathcal R_{m}$'s pairing is $O_{p}(m^{-3/2})$.
The routing applies to the pairings of $-2U_{m}+Q_{m}+R_{m}$ proper; the reduction remainder $\mathcal R_{m}$ is routed separately in the signal--noise channel, whose $(s,s)$ pairing carries a surviving first Hoeffding projection \eqref{eq:Rcal-ss-final}. The subsection implements the route in order: the level is sorted in place; $Q_{m}$ and the noise--noise pairing die; the signal--noise pairings merge with the level's $2sv$ term into the single process \eqref{eq:mixed-b}; the signal--signal pairing supplies the constants completing $A(\eta)$ and the linear brackets of \eqref{eq:pooled-psi-b}; the surviving $r$-pairings form the drift term collected at the end of the subsection. \emph{The level (one-point classes).} Inserting the split,
\[
\frac{1}{m}\sum_{i=1}^{m}\varepsilon_{\perp,i}^{2}
=\bar\sigma^{2}
+\frac{1}{m}\sum_{i=1}^{m}\Bigl[\bigl(v_{i}^{2}-\bar\sigma^{2}\bigr)+2\,v_{i}\,s(X_{i})+s(X_{i})^{2}\Bigr]
+2\langle s,r\rangle_{L^{2}(\mathbb P_{m})}
+\mathcal R^{\mathrm{lev}}_{m},
\]
where $\mathcal R^{\mathrm{lev}}_{m}:=\frac{2}{m}\sum_{i}v_{i}r(X_{i})+\frac{1}{m}\sum_{i}r(X_{i})^{2}$: the first piece has conditional mean zero and conditional variance $\frac{4\sigma^{2}_{\max}}{m^{2}}\sum_{i}r(X_{i})^{2}\le\frac{4\sigma^{2}_{\max}}{m}\lVert r\rVert_{L^{2}(\mathbb P_{m})}^{2}=O_{p}(m^{-2})$, and the second is $O_{p}(m^{-1})$ by Cauchy--Schwarz, so $\mathcal R^{\mathrm{lev}}_{m}=O_{p}(m^{-1})$; $2\langle s,r\rangle_{L^{2}(\mathbb P_{m})}$ is the level's contribution to the drift term, collected below; and $\frac{1}{m}\sum_{i}s(X_{i})^{2}=\lVert s\rVert_{L^{2}}^{2}+O_{p}(m^{-1/2})$ by the law of large numbers, the summand being bounded. The $2\,v_{i}s(X_{i})$ term of the bracket is the level's copy of the cross-pairing $(s,v)$; it is deferred to the signal--noise channel, with which it merges in \eqref{eq:mixed-b} below.

\emph{The $Q$-channel.} Grouping by the squared index,
\[
Q_{m}(\eta)=\frac{1}{m}\sum_{j=1}^{m}\varepsilon_{\perp,j}^{2}\;
\frac{1}{(m-1)^{2}}\sum_{i\ne j}g_{\eta}(X_{i},X_{j})^{2},
\qquad
\frac{1}{(m-1)^{2}}\sum_{i\ne j}g_{\eta}(X_{i},X_{j})^{2}\le\frac{\bar k^{2}\lambda_{L}^{-2}}{m-1},
\]
so $\sup_{\eta}\lvert Q_{m}(\eta)\rvert\le\frac{\bar k^{2}\lambda_{L}^{-2}}{m-1}\,\frac{1}{m}\sum_{j}\varepsilon_{\perp,j}^{2}=O_{p}(m^{-1})$: the $Q$-channel dies at the $\sqrt m$ scale by normalization alone, uniformly in $\eta$. 

\emph{Class $(v,v)$ (noise--noise pairing) dies by a conditional-moment count after the raw-noise switch.} By the replacement convention of the setup the pairings $-2U^{(v,v)}_{m}+U^{(v,v,g\circ g)}_{m}$ are read at the raw noise $\boldsymbol u:=\boldsymbol v^{\rm raw}$, and the switch costs $O_{p}(m^{-1})$ uniformly in $\eta$. Let $C$ be either kernel matrix, $C_{ij}=g_{\eta}(X_{i},X_{j})$ or $C_{ij}=(g_{\eta}\circ g_{\eta})(X_{i},X_{j})$ for $i\ne j$, symmetrized with zero diagonal, at the common envelope $\lvert C_{ij}\rvert\le\bar c:=\bar k\lambda_{L}^{-1}$, and write $H:=\boldsymbol I_{m}-\hat{\boldsymbol\Pi}=\boldsymbol P_{\hat S}$, symmetric idempotent of rank $d_{\theta}$. Inserting $v_{i}=[(\boldsymbol I_{m}-H)\boldsymbol u]_{i}$ and expanding the product gives $v_{i}v_{j}-u_{i}u_{j}=-u_{i}[Hu]_{j}-[Hu]_{i}u_{j}+[Hu]_{i}[Hu]_{j}$, the two cross terms contributing equally by the symmetry of $C$, so the switch of the pairing $\frac{1}{m(m-1)}\sum_{i\ne j}C_{ij}v_{i}v_{j}$ is the quadratic form
\[
\frac{1}{m(m-1)}\sum_{i\ne j}C_{ij}\bigl(v_{i}v_{j}-u_{i}u_{j}\bigr)
=\frac{\boldsymbol u^{\top}\bigl(HCH-2\,CH\bigr)\boldsymbol u}{m(m-1)} .
\]
Its conditional mean is $\frac{1}{m(m-1)}\sum_{i\ne j}C_{ij}\,[\hat{\boldsymbol\Pi}D\hat{\boldsymbol\Pi}]_{ij}$, the conditional covariance of $v$ being $[\hat{\boldsymbol\Pi}D\hat{\boldsymbol\Pi}]_{ij}$, and the same expansion at the covariance level gives, off the diagonal, $[\hat{\boldsymbol\Pi}D\hat{\boldsymbol\Pi}]_{ij}=-H_{ij}\bigl(\sigma^{2}(X_{i})+\sigma^{2}(X_{j})\bigr)+[HDH]_{ij}$, whence Cauchy--Schwarz against $\sum_{i\ne j}H^{2}_{ij}\le\operatorname{tr}(H^{2})=d_{\theta}$ for the first term, and the row-norm bound $\sum_{i\ne j}\lvert[HDH]_{ij}\rvert\le\sigma^{2}_{\max}\sum_{k}\lVert H_{\cdot k}\rVert^{2}_{1}\le\sigma^{2}_{\max}\,m\,d_{\theta}$ for the second, give
\[
\Bigl\lvert\mathbb E\Bigl[\frac{1}{m(m-1)}\sum_{i\ne j}C_{ij}\bigl(v_{i}v_{j}-u_{i}u_{j}\bigr)\Bigm|\mathbf X\Bigr]\Bigr|
\le\frac{\sigma^{2}_{\max}\bar c\,\bigl(2\sqrt{d_{\theta}}+d_{\theta}\bigr)}{m-1}
=O(m^{-1}),
\]
deterministically in $\mathbf X$ and uniformly in $\eta$. The centered remainder of the switch dies at second-moment level. Its cross part is the quadratic form $\frac{2}{m(m-1)}\boldsymbol u^{\top}CH\boldsymbol u$ minus its conditional mean, and the conditional-variance clause \eqref{eq:quadvar_bound} of Lemma~\ref{lem:pairvar}, applied to the raw noise at the symmetrized matrix, bounds its conditional variance by
\[
\frac{4\,\bigl(2\sigma^{4}_{\max}+\bar M_{4}\bigr)\,\lVert CH\rVert^{2}_{\mathrm F}}{m^{2}(m-1)^{2}}
\le\frac{4\,\bigl(2\sigma^{4}_{\max}+\bar M_{4}\bigr)\,\bar c^{2}}{(m-1)^{2}}
=O\bigl(m^{-2}\bigr),
\qquad\lVert CH\rVert_{\mathrm F}\le\lVert C\rVert_{\mathrm F}\le m\bar c,
\]
so the cross part is $O_{p}(m^{-1})$ by the variance-event splitting of the setup. The quadratic part is bounded pathwise,
\[
\Bigl\lvert\frac{\boldsymbol u^{\top}HCH\boldsymbol u}{m(m-1)}\Bigr\rvert
=\frac{\bigl\lvert(Hu)^{\top}C(Hu)\bigr\rvert}{m(m-1)}
\le\frac{\bar c}{m(m-1)}\Bigl(\sum_{i}\lvert[Hu]_{i}\rvert\Bigr)^{2}
\le\frac{\bar c\,\boldsymbol u^{\top}H\boldsymbol u}{m-1},
\]
by the entrywise envelope and Cauchy--Schwarz, and $\mathbb E[\boldsymbol u^{\top}H\boldsymbol u\mid\mathbf X]=\operatorname{tr}(HD)\le\sigma^{2}_{\max}d_{\theta}$, $H$ being an orthogonal projector, so Markov's inequality gives the quadratic part as $O_{p}(m^{-1})$. Assembling the two kernel matrices with their prefactors $2$ and $\frac{m-2}{m-1}\le1$, the switch is $O_{p}(m^{-1})$ uniformly in $\eta$. After the switch, both pairings are quadratic forms $\frac{1}{m(m-1)}\sum_{i\ne j}C_{ij}u_{i}u_{j}$ in the conditionally i.i.d.\ raw noise: conditionally centered given $\mathbf X$, with conditional second moments $O(1)$ at the $mU_{m}$ scale by Lemma~\ref{lem:pairvar} at the covariance $D$ and the fourth-moment bound $\bar M_{4}=\bar M^{2/p}_{2p}$, both pairings are $O_{p}(m^{-1})$ uniformly in $\eta$ and die at the $\sqrt m$ scale. 

\emph{Class $(s,v)+(v,s)$ (the signal--noise channel).} Pairings $(s,v)$ and $(v,s)$, together with the level cross-pairing $2\,v_{i}s(X_{i})$ deferred above, are linear in the noise and merge into a single process. By symmetry of the sum indices and of $g_{\eta}$,
\[
-2\cdot\frac{1}{m(m-1)}\sum_{i\ne j}\bigl[s(X_{i})v_{j}+s(X_{j})v_{i}\bigr]g_{\eta}(X_{i},X_{j})
=-\frac{4}{m}\sum_{j=1}^{m}v_{j}\,\frac{1}{m-1}\sum_{i\ne j}s(X_{i})\,g_{\eta}(X_{i},X_{j});
\]
likewise, by the exact reduction of Lemma~\ref{lem:bilinear_cv}(ii), valid for arbitrary coefficients,
$R_{m}=\frac{m-2}{m-1}\,U^{(g\circ g)}_{m}+\mathcal R_{m}/m$ with
$U^{(g\circ g)}_{m}:=\frac{1}{m(m-1)}\sum_{i\ne j}\varepsilon_{\perp,i}\varepsilon_{\perp,j}(g_{\eta}\circ g_{\eta})(X_{i},X_{j})$,
so the bilinear part contributes $+\frac{2}{m}\sum_{j}v_{j}\frac{1}{m-1}\sum_{i\ne j}s(X_{i})(g_{\eta}\circ g_{\eta})(X_{i},X_{j})$ at the same scale, the prefactor $\frac{m-2}{m-1}$ differing from $1$ by $O(m^{-1})$. The remainder of the reduction is controlled under the fixed alternative by the
anchor-centeredness $\mathbb E_{X}[r_{\eta}(X;y,z)]=0$, with one pairing that does
not die and is treated first. Reindexing the remainder by its pair slots,
\begin{equation}
\mathcal R_{m}(\eta)
=\frac{2}{(m-1)^{2}}\sum_{j<k}\Gamma^{(m)}_{jk}(\eta)\;\varepsilon_{\perp,j}\varepsilon_{\perp,k},
\qquad
\Gamma^{(m)}_{jk}(\eta):=\sum_{i\ne j,k}r_{\eta}(X_{i};X_{j},X_{k}),
\label{eq:Rcal-gamma}
\end{equation}
(the coefficient matrix of the matrix-form paragraph in the proof of
Lemma~\ref{lem:remainder_cv} has entries $\Gamma^{(m)}_{jk}/(m-1)^{2}$), the
envelope $\lvert r_{\eta}\rvert\le K_{r}$, $K_{r}:=2\bar k^{2}\lambda_{L}^{-2}$,
and the anchor-centeredness give
$\mathbb E[\Gamma^{(m)}_{jk}\mid X_{j},X_{k}]=0$,
$\mathbb E[(\Gamma^{(m)}_{jk})^{2}\mid X_{j},X_{k}]\le(m-2)K_{r}^{2}$, and
$\mathbb E\lVert\Gamma^{(m)}\rVert_{\mathrm F}^{2}\le m(m-1)(m-2)K_{r}^{2}$. The pairings $(v,r)$, $(r,v)$, $(r,r)$ are $O_{p}(m^{-1/2})$ and $(v,v)$ is $O_{p}(1)$, all at the $\mathcal R_{m}$ scale and uniformly in $\eta$, each $O_{p}(m^{-1})$ at the criterion scale. For $(v,v)$, introduce the symmetric zero-diagonal matrix $A$ with entries $A_{jk}:=\Gamma^{(m)}_{jk}/(m-1)^{2}$ for $j\ne k$, so that the pairing reads $T_{(v,v)}=\boldsymbol v^{\top}A\boldsymbol v$. The components of $v$ are conditionally centered but not conditionally independent, so the count is run after the sandwich $T_{(v,v)}=\boldsymbol u^{\top}A'\boldsymbol u$ with $A':=\hat{\boldsymbol\Pi}A\hat{\boldsymbol\Pi}$ and $\boldsymbol u:=\boldsymbol v^{\rm raw}$: the components of $\boldsymbol u$ are conditionally i.i.d., centered, with conditional covariance $D$, and this is the operational content of applying Lemma~\ref{lem:pairvar} at the conditional covariance $\hat{\boldsymbol\Pi}D\hat{\boldsymbol\Pi}$. Lemma~\ref{lem:pairvar}, applied to $\boldsymbol u$ at $A'$ through its conditional-variance clause \eqref{eq:quadvar_bound} (the kurtosis-excess terms bounded through the interpolated fourth moment $\bar M_{4}=\bar M_{2p}^{2/p}$), gives
\[
\operatorname{Var}\bigl(T_{(v,v)}\mid\mathbf X\bigr)
\le\bigl(2\sigma^{4}_{\max}+\bar M^{2/p}_{2p}\bigr)\lVert A\rVert_{\mathrm F}^{2},
\]
using $\lVert D^{1/2}A'D^{1/2}\rVert_{\mathrm F}\le\sigma^{2}_{\max}\lVert A'\rVert_{\mathrm F}$ and $\lVert A'\rVert_{\mathrm F}\le\lVert A\rVert_{\mathrm F}$ (the Frobenius contraction of $\hat{\boldsymbol\Pi}\otimes\hat{\boldsymbol\Pi}$). The conditional mean $\mathbb E[T_{(v,v)}\mid\mathbf X]=\operatorname{tr}(A'\,D)=\operatorname{tr}(A\,C_{v})$ with $C_{v}:=\hat{\boldsymbol\Pi}D\hat{\boldsymbol\Pi}$ does not vanish, and
\[
\bigl\lvert\mathbb E[T_{(v,v)}\mid\mathbf X]\bigr\rvert
\le\lVert A\rVert_{\mathrm F}\,\lVert C_{v}\rVert_{\mathrm F},
\qquad
\begin{aligned}[t]
\lVert C_{v}\rVert_{\mathrm F}^{2}
&=\operatorname{tr}\bigl(C_{v}^{2}\bigr)
\le\sigma^{2}_{\max}\operatorname{tr}\bigl(C_{v}\bigr)\\
&=\sigma^{2}_{\max}\operatorname{tr}\bigl(\hat{\boldsymbol\Pi}D\bigr)
\le\sigma^{4}_{\max}\,(m-d_{\theta}):
\end{aligned}
\]
the Loewner bound $C_{v}^{2}\le\lVert C_{v}\rVert_{\mathrm{op}}C_{v}$ with $\lVert C_{v}\rVert_{\mathrm{op}}\le\sigma^{2}_{\max}$, the cyclicity $\operatorname{tr}(\hat{\boldsymbol\Pi}D\hat{\boldsymbol\Pi})=\operatorname{tr}(D\hat{\boldsymbol\Pi})$ through the idempotence of $\hat{\boldsymbol\Pi}$, and $\operatorname{tr}(\hat{\boldsymbol\Pi}D)\le\sigma^{2}_{\max}\operatorname{rank}(\hat{\boldsymbol\Pi})$ with $\operatorname{rank}(\hat{\boldsymbol\Pi})=m-d_{\theta}$: the rank-$d_{\theta}$ object of the problem is the complementary projector $\boldsymbol I_{m}-\hat{\boldsymbol\Pi}=\boldsymbol P_{\hat S}$, which carries the replacement defect of the setup, whereas $C_{v}$ lives on $\operatorname{range}(\hat{\boldsymbol\Pi})$ and carries the full noise mass. Hence $\lVert C_{v}\rVert_{\mathrm F}\le\sigma^{2}_{\max}\sqrt m$ and the conditional mean is $O_{p}(1)$ at the $\mathcal R_{m}$ scale. It remains to bound $\lVert A\rVert_{\mathrm F}$. Expanding the square over anchor pairs, the summands with $i\ne i'$ have conditional expectation $\mathbb E\bigl[r_{\eta}(X_{i};y,z)\,r_{\eta}(X_{i'};y,z)\bigm|y,z\bigr]=\bigl(\mathbb E_{X}[r_{\eta}(X;y,z)]\bigr)^{2}=0$ at the frozen realized pair $(y,z)=(X_{j},X_{k})$ (anchor-centeredness), so only the $i=i'$ diagonal survives in expectation and
\[
\mathbb E\bigl[\lVert A\rVert_{\mathrm F}^{2}\bigr]
\le\frac{m(m-1)(m-2)\,K_{r}^{2}}{(m-1)^{4}}
=O\bigl(m^{-1}\bigr),
\qquad K_{r}:=2\bar k^{2}\lambda_{L}^{-2},
\]
whence $\lVert A\rVert_{\mathrm F}^{2}=O_{p}(m^{-1})$ by Markov's inequality. Collecting,
$\mathbb E\bigl[T_{(v,v)}^{2}\mid\mathbf X\bigr]
\le\bigl(2\sigma^{4}_{\max}+\bar M^{2/p}_{2p}\bigr)\lVert A\rVert_{\mathrm F}^{2}+\lVert A\rVert_{\mathrm F}^{2}\,\lVert C_{v}\rVert_{\mathrm F}^{2}=O_{p}(1)$,
and the variance-event splitting of the setup gives $T_{(v,v)}=O_{p}(1)$ at the $\mathcal R_{m}$ scale, that is, $O_{p}(m^{-1})$ at the criterion scale, uniformly in $\eta$. The pairings $(v,r)$ and $(r,v)$ are linear in $v$ with coefficient vector $\Gamma^{(m)}\boldsymbol r$ and die by conditional Chebyshev; the pairing $(r,r)$ dies by pathwise Cauchy--Schwarz. For the first two, symmetry of $\Gamma^{(m)}_{jk}$ in its pair slots (the symmetry of $r_{\eta}(x;y,z)$ in $(y,z)$) combines the two cross terms into one ordered sum,
\[
T_{(v,r),(r,v)}
=\frac{2}{(m-1)^{2}}\sum_{j\ne k}\Gamma^{(m)}_{jk}\,v_{j}\,r_{k}
=\frac{2}{(m-1)^{2}}\,\boldsymbol v^{\top}\bigl(\Gamma^{(m)}\boldsymbol r\bigr):
\]
linear in $v$, with the design-only coefficient vector $\Gamma^{(m)}\boldsymbol r$ (the diagonal convention being immaterial at the final order). The conditional mean vanishes,
$\mathbb E\bigl[T_{(v,r),(r,v)}\mid\mathbf X\bigr]
=\frac{2}{(m-1)^{2}}\,(\Gamma^{(m)}\boldsymbol r)^{\top}\,\mathbb E[\boldsymbol v\mid\mathbf X]=0$
(noise-centeredness), and the conditional variance is
\[
\operatorname{Var}\bigl(T_{(v,r),(r,v)}\mid\mathbf X\bigr)
=\frac{4}{(m-1)^{4}}\,(\Gamma^{(m)}\boldsymbol r)^{\top}\hat{\boldsymbol\Pi}D\hat{\boldsymbol\Pi}\,(\Gamma^{(m)}\boldsymbol r)
\le\frac{4\,\sigma^{2}_{\max}}{(m-1)^{4}}\,\lVert\Gamma^{(m)}\boldsymbol r\rVert^{2}
=O_{p}\bigl(m^{-1}\bigr),
\]
the inequality being the Loewner bound $\hat{\boldsymbol\Pi}D\hat{\boldsymbol\Pi}\le\sigma^{2}_{\max}\hat{\boldsymbol\Pi}\le\sigma^{2}_{\max}\boldsymbol I_{m}$. The order of the coefficient vector: Cauchy--Schwarz gives
$\lVert\Gamma^{(m)}\boldsymbol r\rVert\le\lVert\Gamma^{(m)}\rVert_{\mathrm F}\,\lVert\boldsymbol r\rVert$
with $\lVert\boldsymbol r\rVert=\sqrt m\,\lVert r\rVert_{L^{2}(\mathbb P_{m})}=O_{p}(1)$ and
$\lVert\Gamma^{(m)}\rVert_{\mathrm F}^{2}=O_{p}(m^{3})$ by the anchor-centeredness count
(expanding $\sum_{j\ne k}(\Gamma^{(m)}_{jk})^{2}$ over anchor pairs $i,i'$, the terms $i\ne i'$ have conditional expectation
$\mathbb E\bigl[r_{\eta}(X_{i};y,z)\,r_{\eta}(X_{i'};y,z)\bigm|y,z\bigr]
=\bigl(\mathbb E_{X}[r_{\eta}(X;y,z)]\bigr)^{2}=0$
at the frozen realized pair $(y,z)=(X_{j},X_{k})$, leaving the $i=i'$ diagonal
$O(m^{3}K_{r}^{2})$), so $\lVert\Gamma^{(m)}\boldsymbol r\rVert=O_{p}(m^{3/2})$. The variance-event splitting of the setup then gives $T_{(v,r),(r,v)}=O_{p}(m^{-1/2})$ uniformly in $\eta$, that is, $O_{p}(m^{-3/2})$ at the criterion scale. For $(r,r)$ no centering is available ($r$ is not centered), so the route is pathwise:
\[
\bigl\lvert T_{(r,r)}\bigr\rvert
\le\frac{2}{(m-1)^{2}}\Bigl(\sum_{j<k}\bigl(\Gamma^{(m)}_{jk}\bigr)^{2}\Bigr)^{1/2}
\Bigl(\sum_{j<k}r_{j}^{2}r_{k}^{2}\Bigr)^{1/2}
\le\frac{2}{(m-1)^{2}}\,\lVert\Gamma^{(m)}\rVert_{\mathrm F}\cdot\tfrac{1}{\sqrt2}\,\lVert\boldsymbol r\rVert^{2}
=O_{p}\bigl(m^{-1/2}\bigr)
\]
uniformly in $\eta$, the second factor using $\lVert\boldsymbol r\boldsymbol r^{\top}\rVert_{\mathrm F}=\lVert\boldsymbol r\rVert^{2}=m\,\lVert r\rVert_{L^{2}(\mathbb P_{m})}^{2}=O_{p}(1)$, both factors being $O_{p}(m^{3/2})$ and $O_{p}(1)$ from the anchor count and the perturbation rate. The pairings $(s,r)$ and $(r,s)$ are $O_{p}(1)$ by pathwise Cauchy--Schwarz: no centering is available (both factors are design-measurable and non-centered), so the route is a two-norm bound,
\[
\lvert T_{(s,r),(r,s)}\rvert
\le\frac{2}{(m-1)^{2}}
\Bigl(\sum_{j<k}\bigl(\Gamma^{(m)}_{jk}\bigr)^{2}\Bigr)^{1/2}
\Bigl(\sum_{j<k}\bigl(s_{j}r_{k}+r_{j}s_{k}\bigr)^{2}\Bigr)^{1/2}.
\]
The kernel factor: $\bigl(\sum_{j<k}(\Gamma^{(m)}_{jk})^{2}\bigr)^{1/2}\le\tfrac{1}{\sqrt{2}}\,\lVert\Gamma^{(m)}\rVert_{\mathrm F}$ (unordered pairs are half the ordered ones), with the deterministic envelope
$\lVert\Gamma^{(m)}\rVert_{\mathrm F}^{2}\le m(m-1)(m-2)^{2}K_{r}^{2}=O(m^{3}K_{r}^{2})$
from $\lvert\Gamma^{(m)}_{jk}\rvert\le(m-2)K_{r}$, $K_{r}:=2\bar k^{2}\lambda_{L}^{-2}$. The residual factor: by $(a+b)^{2}\le2a^{2}+2b^{2}$ and the symmetry bookkeeping $\sum_{j<k}(s_{j}^{2}r_{k}^{2}+r_{j}^{2}s_{k}^{2})\le(\sum_{j}s_{j}^{2})(\sum_{k}r_{k}^{2})$,
\[
\Bigl(\sum_{j<k}\bigl(s_{j}r_{k}+r_{j}s_{k}\bigr)^{2}\Bigr)^{1/2}
\le\sqrt{2}\,\lVert\boldsymbol s\rVert\,\lVert\boldsymbol r\rVert
=\sqrt{2m}\,\lVert s\rVert_{\infty}\cdot O_{p}(1),
\]
using $\lVert\boldsymbol s\rVert\le m^{1/2}\lVert s\rVert_{\infty}$ (the part (b) standing assumption) and $\lVert\boldsymbol r\rVert=\sqrt m\,\lVert r\rVert_{L^{2}(\mathbb P_{m})}=O_{p}(1)$ (the perturbation rate). Assembling,
\[
\lvert T_{(s,r),(r,s)}\rvert
\le\frac{2}{(m-1)^{2}}\cdot O\bigl(m^{3/2}K_{r}\bigr)\cdot O\bigl(m^{1/2}\lVert s\rVert_{\infty}\bigr)\cdot O_{p}(1)
=O\bigl(K_{r}\lVert s\rVert_{\infty}\bigr)\cdot O_{p}(1)
=O_{p}(1),
\]
the printed $\tfrac{4}{(m-1)^{2}}\lVert\Gamma^{(m)}\rVert_{\mathrm F}\lVert\boldsymbol s\rVert\,\lVert\boldsymbol r\rVert$ bound with the constant stated generously. No rate is available or needed: $O_{p}(1)$ at the $\mathcal R_{m}$ scale is $O_{p}(m^{-1})$ at criterion scale, inside the budget.

The pairings $(s,v)$ and $(v,s)$ are linear in $v$ with coefficient vector $\Gamma^{(m)}\boldsymbol s$, and die by the noise centering. Symmetry of $\Gamma^{(m)}_{jk}$ combines them into one ordered sum,
\[
T_{(s,v),(v,s)}
=\frac{2}{(m-1)^{2}}\sum_{j\ne k}\Gamma^{(m)}_{jk}\,v_{j}\,s_{k}
=\frac{2}{(m-1)^{2}}\,\boldsymbol v^{\top}\bigl(\Gamma^{(m)}\boldsymbol s\bigr),
\]
linear in $v$ with the design-only coefficient vector $\boldsymbol c:=\Gamma^{(m)}\boldsymbol s$. The conditional mean vanishes,
$\mathbb E\bigl[T_{(s,v),(v,s)}\mid\mathbf X\bigr]
=\frac{2}{(m-1)^{2}}\,\boldsymbol c^{\top}\,\mathbb E[\boldsymbol v\mid\mathbf X]=0$
(noise-centeredness), and the conditional variance is bounded through the same Loewner chain as the $(v,r),(r,v)$ pairing,
\[
\operatorname{Var}\bigl(T_{(s,v),(v,s)}\mid\mathbf X\bigr)
=\frac{4}{(m-1)^{4}}\,\boldsymbol c^{\top}\hat{\boldsymbol\Pi}D\hat{\boldsymbol\Pi}\,\boldsymbol c
\le\frac{4\,\sigma^{2}_{\max}}{(m-1)^{4}}\,\lVert\boldsymbol c\rVert^{2}
\le\frac{4\,\sigma^{2}_{\max}}{(m-1)^{4}}\,\lVert\Gamma^{(m)}\rVert_{\mathrm F}^{2}\,\lVert\boldsymbol s\rVert^{2},
\]
with $\lVert\boldsymbol c\rVert\le\lVert\Gamma^{(m)}\rVert_{\mathrm F}\lVert\boldsymbol s\rVert$ by Cauchy--Schwarz and $\lVert\boldsymbol s\rVert^{2}\le m\,\lVert s\rVert_{\infty}^{2}$; the anchor-centeredness count recorded above gives $\mathbb E\lVert\Gamma^{(m)}\rVert_{\mathrm F}^{2}\le m(m-1)(m-2)K_{r}^{2}=O(m^{3})$, so
\[
\mathbb E\Bigl[\operatorname{Var}\bigl(T_{(s,v),(v,s)}\mid\mathbf X\bigr)\Bigr]
\le\frac{4\,\sigma^{2}_{\max}\,m\,\lVert s\rVert_{\infty}^{2}\,m(m-1)(m-2)K_{r}^{2}}{(m-1)^{4}}
=O(1),
\]
and Markov's inequality, followed by the variance-event splitting of the setup at the unit scale, gives $T_{(s,v),(v,s)}=O_{p}(1)$ uniformly in $\eta$, that is, $O_{p}(m^{-1})$ at the criterion scale. All eight pairings are therefore $O_{p}(1)$ at the $\mathcal
R_{m}$ scale and $O_{p}(m^{-1})$ at the criterion scale, within the budget of the collection \eqref{eq:level-collect} below.

The pairing $(s,s)$ does not die: its centering lives only in the anchor slot of
$r_{\eta}$, and its first Hoeffding projection survives. Write
\[
T_{m}(\eta):=\frac{2}{(m-1)^{2}}\sum_{i=1}^{m}\sum_{\substack{j<k\\ j,k\ne i}}
r_{\eta}(X_{i};X_{j},X_{k})\,s(X_{j})s(X_{k}),
\qquad
u_{\eta}:=w^{*2}_{\eta},
\]
and recall that $\theta_{g\circ g}(\eta)=\lVert w^{*}_{\eta}\rVert_{L^{2}}^{2}$.
The conditional mean given one coordinate is read off from the position of the
frozen index, and the position is not fixed by the conditioning: freezing the value
$X_{1}=x$ pins the value at site $1$ and leaves the index combinatorics to place it,
so the terms of $T_{m}$ fall into three classes. Index $1$ as the anchor ($i=1$,
$\binom{m-1}{2}$ inner pairs); index $1$ in a pair slot ($i\ne1$ with $j=1$ or $k=1$,
$(m-1)(m-2)$ terms, more numerous than the anchor class); index $1$ away from all
three slots; the conditional mean is the sum of the three class expectations.
\emph{Anchor class.} For a fresh pair $\{j,k\}$ the expectation factors by
independence, each factor being $\mathbb E[s(X_{j})g_{\eta}(x,X_{j})]
=w^{*}_{\eta}(x)$ and $\mathbb E[s(X_{k})g_{\eta}(x,X_{k})]=w^{*}_{\eta}(x)$, so the
product contributes $w^{*2}_{\eta}(x)=u_{\eta}(x)$; the constant term is disposed of
by inserting the definition of the composed kernel and swapping the expectations,
$\mathbb E[(g_{\eta}\circ g_{\eta})(X_{j},X_{k})s(X_{j})s(X_{k})]
=\mathbb E_{W}[w^{*2}_{\eta}(W)]=\lVert w^{*}_{\eta}\rVert_{L^{2}}^{2}$. The class
contributes $\binom{m-1}{2}\bigl[u_{\eta}(x)-\theta_{g\circ g}(\eta)\bigr]$.
\emph{Slot class.} For $i\ne1$ with $j=1$, the summand is
$r_{\eta}(X_{i};x,X_{k})s(x)s(X_{k})$ with fresh anchor and fresh partner, and its
two pieces are individually nonzero: the product piece carries the conditional
expectation $\mathbb E[g_{\eta}(X_{i},x)g_{\eta}(X_{i},X_{k})s(X_{k})\mid X_{i},x]
=\mathbb E_{X_{i}}[g_{\eta}(x,X_{i})w^{*}_{\eta}(X_{i})]=\kappa_{\eta}(x)$, by the
first anchor identity inside and the second outside; the constant piece carries
$\mathbb E[(g_{\eta}\circ g_{\eta})(x,X_{k})\,s(X_{k})\mid x]=\kappa_{\eta}(x)$; the bracket
cancels as $\kappa_{\eta}(x)-\kappa_{\eta}(x)=0$, a cancellation funded by the anchor
identities rather than an absence of centering. The class contributes zero.
\emph{Away class.} With $1\notin\{i,j,k\}$ all three indices are fresh, and the
defining centering $\mathbb E[r_{\eta}(X;X_{j},X_{k})\mid X_{j},X_{k}]=0$ kills the
class. Hence
\begin{equation}
\mathbb E\bigl[T_{m}(\eta)\bigm|X_{1}=x\bigr]
=\frac{m-2}{m-1}\bigl[u_{\eta}(x)-\theta_{g\circ g}(\eta)\bigr],
\qquad\text{whence}\qquad
\mathbb E\bigl[T_{m}(\eta)\bigr]=0,
\label{eq:Rcal-ss-proj}
\end{equation}
the coefficient $\frac{m-2}{m-1}$ being $\frac{2}{(m-1)^{2}}\binom{m-1}{2}$, the
fresh-pair count of the anchor class; the mean vanishing because
$\mathbb E[u_{\eta}(X)]=\lVert w^{*}_{\eta}\rVert_{L^{2}}^{2}$: there is no
deterministic drift, the raw masses of size $m^{3}$ in $T_{m}$ cancelling through
$\theta_{g\circ g}(\eta)=\lVert w^{*}_{\eta}\rVert_{L^{2}}^{2}$.

Each summand of $T_{m}$ involves at most three design points, so M\"obius inversion
into the Hoeffding components is exact, with no remainder term of any order. Fold
the double sum by letting each member of an unordered triple take the anchor role
once: with $h(x;y,z):=r_{\eta}(x;y,z)s(y)s(z)$ and the symmetric kernel
$\tilde h(x,y,z):=h(x;y,z)+h(y;x,z)+h(z;x,y)$,
$T_{m}(\eta)=\frac{2}{(m-1)^{2}}\sum_{\{a,b,c\}\subset[m]}\tilde h(a,b,c)$.
The three projections are
\[
\tilde g_{1}(x)=\mathbb E[\tilde h(x,Y,Z)]=u_{\eta}(x)-\theta_{g\circ g}(\eta),
\]
the anchor computation above with its two slot terms zero;
\[
\tilde g_{2}(x,y):=\mathbb E_{Z}[\tilde h(x,y,Z)]-\tilde g_{1}(x)-\tilde g_{1}(y)
=s(y)\Phi_{\eta}(x,y)+s(x)\Phi_{\eta}(y,x)-\bigl[u_{\eta}(x)+u_{\eta}(y)-2\,\theta_{g\circ g}(\eta)\bigr]
=\varphi_{\eta}(x,y),
\]
with $\Phi_{\eta}(x,y):=g_{\eta}(x,y)\,w^{*}_{\eta}(x)-\kappa_{\eta}(y)$ the
depth-two reduction of the anchor identities and the slot-slot projection dying by
anchor-centeredness; and $\tilde g_{3}:=\tilde h-\sum_{\text{pairs}}\tilde g_{2}
-\sum_{\text{points}}\tilde g_{1}$, which is the printed $g_{3,\eta}$, its last
pair sum being $\sum_{\text{pairs}}h(c;a,b)=\tilde h(x,y,z)$. The components are
centered by construction, $\mathbb E[\tilde h]=0$ being the mean-zero statement of
\eqref{eq:Rcal-ss-proj}. Summing the pointwise identity over the $\binom{m}{3}$
triples, each point lies in $\binom{m-1}{2}$ triples and each pair in $m-2$, and
the prefactor $\frac{2}{(m-1)^{2}}$ converts these counts into the coefficients of
\begin{equation}
\begin{split}
T_{m}(\eta)=&\frac{m-2}{m-1}\sum_{i=1}^{m}\bigl[u_{\eta}(X_{i})
-\theta_{g\circ g}(\eta)\bigr]
+\frac{2(m-2)}{(m-1)^{2}}\sum_{i<j}\varphi_{\eta}(X_{i},X_{j})\\
&+\frac{2}{(m-1)^{2}}\sum_{i<j<l}g_{3,\eta}(X_{i},X_{j},X_{l}),
\end{split}
\label{eq:Rcal-ss-decomp}
\end{equation}
with the centered two-point kernel
\[
\begin{aligned}
\varphi_{\eta}(x,y):={}&-\bigl[u_{\eta}(x)+u_{\eta}(y)-2\,\theta_{g\circ g}(\eta)\bigr]
+g_{\eta}(x,y)\bigl[s(y)\,w^{*}_{\eta}(x)+s(x)\,w^{*}_{\eta}(y)\bigr]\\
&-s(y)\,\kappa_{\eta}(y)-s(x)\,\kappa_{\eta}(x),
\qquad
\mathbb E_{X}\bigl[\varphi_{\eta}(x,X)\bigr]=0 ,
\end{aligned}
\]
and the three-point kernel, centered in each coordinate separately,
\[
\begin{aligned}
g_{3,\eta}(x,y,z):={}&\bigl[u_{\eta}(x)+u_{\eta}(y)+u_{\eta}(z)
-3\,\theta_{g\circ g}(\eta)\bigr]\\
&-\sum_{\{a,b\}\subset\{x,y,z\}}\Bigl[s(b)\bigl(g_{\eta}(a,b)\,w^{*}_{\eta}(a)
-\kappa_{\eta}(b)\bigr)+s(a)\bigl(g_{\eta}(a,b)\,w^{*}_{\eta}(b)
-\kappa_{\eta}(a)\bigr)\Bigr]\\
&+\sum_{\{a,b\}\subset\{x,y,z\}}\Bigl[g_{\eta}(c,a)\,g_{\eta}(c,b)
-(g_{\eta}\circ g_{\eta})(a,b)\Bigr]s(a)s(b),
\end{aligned}
\]
where $c$ denotes the remaining point of $\{x,y,z\}$, and the last sum runs over the
three unordered pairs $\{a,b\}$ with $c$ the third point. All three kernels are bounded by deterministic, $\mathbf X$-free
envelopes: $\lvert u_{\eta}\rvert\le U_{w}^{2}$, $\lvert\varphi_{\eta}\rvert\le
M_{\varphi}$, $\lvert g_{3,\eta}\rvert\le M_{3}$, with
$U_{w}:=(4\lambda_{L})^{-1/2}\bar k^{1/2}\lVert s\rVert_{L^{2}}$ and
$M_{\varphi},M_{3}$ depending only on $(\lVert s\rVert_{\infty},\lVert
s\rVert_{L^{2}},\bar k,\lambda_{L})$ through the anchor envelopes of the setup.
The degenerate sums are controlled by their exact second moments, and the control
is uniform in $\eta$ because both ingredients are: the centering identities hold at
every fixed tuning, the anchor identities carrying no $\eta$-dependence, and the
envelopes are $\eta$-free. The centering kills every cross term of the chaos
squares, pairs of pairs sharing an index by conditioning on the shared index,
disjoint pairs by independence, so
$\operatorname{Var}\bigl(\sum_{i<j}\varphi_{\eta}(X_{i},X_{j})\bigr)=\binom
m2\,\mathbb E[\varphi_{\eta}^{2}(X_{1},X_{2})]\le\binom m2 M_{\varphi}^{2}$,
likewise $\operatorname{Var}\bigl(\sum_{i<j<l}g_{3,\eta}\bigr)=\binom
m3\,\mathbb E[g_{3,\eta}^{2}]\le\binom m3 M_{3}^{2}$, and
$\sum_{i<j}\varphi_{\eta}(X_{i},X_{j})=O_{p}(m)$,
$\sum_{i<j<l}g_{3,\eta}=O_{p}(m^{3/2})$. Collecting, the linear coefficient of $T_{m}/m$ rounds
as $\frac{m-2}{m(m-1)}=\frac{1}{m}-\frac{1}{m(m-1)}$, a deterministic $O(m^{-1})$ on the
linear sum because $\lvert u_{\eta}-\theta_{g\circ g}\rvert\le
U_{w}^{2}+\theta_{g\circ g}$ bounds the sum pathwise; the pairwise prefactor is
$O(m^{-2})$ against $O_{p}(m)$, giving $O_{p}(m^{-1})$; the triple prefactor is
$O(m^{-3})$ against $O_{p}(m^{3/2})$, giving $O_{p}(m^{-3/2})$, absorbed. Hence
\begin{equation}
\frac{\mathcal R^{(s,s)}_{m}(\eta)}{m}
=\frac{1}{m}\sum_{i=1}^{m}\bigl[w^{*2}_{\eta}(X_{i})-\theta_{g\circ g}(\eta)\bigr]
+O_{p}(m^{-1})
\qquad\text{uniformly in }\eta,
\label{eq:Rcal-ss-final}
\end{equation}
the coefficient $\frac{m-2}{m-1}$ differing from $1$ by $O(m^{-1})$. The first
Hoeffding projection of the remainder is a centered average of bounded summands
and is routed to the pooled family below, where it becomes the sixth bracket of
\eqref{eq:pooled-psi-b}; it is centered because $\mathbb E[w^{*2}_{\eta}(X)]
=\lVert w^{*}_{\eta}\rVert_{L^{2}}^{2}=\theta_{g\circ g}(\eta)$, and
Gaussian-borne because it is built from $w^{*}_{\eta}$, the kernel ridge
smoothing of the signal in $\mathcal H_{\gamma}$, the two higher chaos
components having died at $O_{p}(m^{-1})$ and $O_{p}(m^{-3/2})$. At the
$\sqrt m$ scale the entire remainder channel contributes only this
Gaussian-borne linear term. The pairings $(s,v)+(v,s)$ of $-2U_{m}+Q_{m}+R_{m}$
are linear in the noise, and $Q_{m}$ carries none, its coefficients
$\varepsilon_{\perp,j}^{2}$ being squared; the two carriers are those displayed
at the start of the signal--noise channel above, $-\frac{4}{m}\sum_{j=1}^{m}v_{j}\,
\frac{1}{m-1}\sum_{i\ne j}s(X_{i})\,g_{\eta}(X_{i},X_{j})$ through the kernel
$g_{\eta}$ and $+\frac{2}{m}\sum_{j=1}^{m}v_{j}\,\frac{1}{m-1}\sum_{i\ne
j}s(X_{i})\,(g_{\eta}\circ g_{\eta})(X_{i},X_{j})$ through the composed kernel,
and each inner leave-one-out average is an empirical plug-in of a population
anchor in the sense of \eqref{eq:anchor-id}: for the generic kernel
$\mathfrak g_{\eta}\in\{g_{\eta},(g_{\eta}\circ g_{\eta})_{\eta}\}$ the anchor is
$\vartheta_{\mathfrak g}(x):=\mathbb E_{Y}[\mathfrak g_{\eta}(x,Y)\,\Pi f(Y)]$, and the
anchor identities give $\vartheta_{g}=w^{*}_{\eta}$ and
$\vartheta_{g\circ g}=\kappa_{\eta}$ respectively. For each of the two inner
factors of the mixed channel, the empirical fluctuation of the anchor around
this population limit is $L^{2}$-small: with $\hat\vartheta_{j}:=\frac{1}{m-1}\sum_{i\ne j}s(X_{i})\,\mathfrak g_{\eta}(X_{i},X_{j})$ and $\zeta_{ij}:=s(X_{i})\mathfrak g_{\eta}(X_{i},X_{j})-\vartheta_{\mathfrak g}(X_{j})$,
\[
\hat\vartheta_{j}-\vartheta_{\mathfrak g}(X_{j})=\frac{1}{m-1}\sum_{i\ne j}\zeta_{ij},
\qquad
\mathbb E\Bigl[\frac{1}{m}\sum_{j=1}^{m}\bigl(\hat\vartheta_{j}-\vartheta_{\mathfrak g}(X_{j})\bigr)^{2}\Bigr]
=\frac{\mathbb E\bigl[\zeta_{12}^{2}\bigr]}{m-1}\le\frac{C_{\zeta}}{m-1},
\]
the identity inserting the definition of the $\zeta_{ij}$ into the leave-one-out
average, the second moment being computed conditionally on $X_{j}$: the
summands $\zeta_{ij}$, $i\ne j$, are bounded and centered over $i$ for each
fixed $j$, $\mathbb E[\zeta_{ij}\mid X_{j}]=\mathbb E[s(X)\mathfrak g_{\eta}(X,X_{j})]
-\vartheta_{\mathfrak g}(X_{j})=0$ by the defining property of the anchor, and
conditionally independent, being functions of the independent fresh draws
$X_{i}$, so that the cross terms $i\ne i'$ vanish in the inner second moment,
$\mathbb E[\zeta_{ij}\zeta_{i'j}\mid X_{j}]
=\mathbb E[\zeta_{ij}\mid X_{j}]\,\mathbb E[\zeta_{i'j}\mid X_{j}]=0$, and only
the diagonal $\frac{1}{(m-1)^{2}}\sum_{i\ne j}\mathbb E[\zeta_{ij}^{2}]$
survives; $C_{\zeta}$ depends only on the entrywise envelopes, $\lvert
s\rvert\le\lVert s\rVert_{\infty}$, $\lvert g_{\eta}\rvert\le\bar
k\lambda_{L}^{-1}$, $\lvert(g_{\eta}\circ g_{\eta})\rvert\le\bar
k(4\lambda_{L})^{-1}$, with the anchor bounds of the setup, so the anchors
fluctuate at rate $m^{-1/2}$ in mean square, uniformly in $\eta$; the same
computation applies with $g_{\eta}$. Inserting
$\hat\vartheta_{j}=\vartheta_{\mathfrak g}(X_{j})+(\hat\vartheta_{j}
-\vartheta_{\mathfrak g}(X_{j}))$ into the two carriers, each anchor fluctuation
enters as $\frac{c}{m}\sum_{j}v_{j}(\hat\vartheta_{j}
-\vartheta_{\mathfrak g}(X_{j}))$ with $c\in\{-4,+2\}$; the pairing with $v$ is
conditionally centered with conditional variance
$\frac{c^{2}}{m^{2}}\sum_{j}\sigma^{2}(X_{j})\bigl(\hat\vartheta_{j}
-\vartheta_{\mathfrak g}(X_{j})\bigr)^{2}\le\frac{16\,\sigma^{2}_{\max}}{m^{2}}
\sum_{j}\bigl(\hat\vartheta_{j}-\vartheta_{\mathfrak g}(X_{j})\bigr)^{2}
=O_{p}(m^{-2})$ by the displayed bound, the factor $16=4^{2}$ taking the worse
of the two coefficients and the $O_{p}(m^{-2})$ following by Markov's
inequality since $\mathbb E[\sum_{j}(\hat\vartheta_{j}
-\vartheta_{\mathfrak g}(X_{j}))^{2}]\le m\,C_{\zeta}/(m-1)=O(1)$, so conditional
Chebyshev's inequality turns each anchor fluctuation into an
$O_{p}(m^{-1})$ term, uniformly in $\eta$; and the diagonal corrections are handled by aligning the leave-one-out inner average with the full-sample one: with $\bar\vartheta_{j}:=\frac{1}{m}\sum_{i=1}^{m}s(X_{i})\,\mathfrak g_{\eta}(X_{i},X_{j})$,
\[
\hat\vartheta_{j}=\frac{m}{m-1}\,\bar\vartheta_{j}-\frac{1}{m-1}\,s(X_{j})\,\mathfrak g_{\eta}(X_{j},X_{j}),
\]
so that, at the prefactors $c\in\{-4,+2\}$ of the two carriers, the residual diagonal term enters as $-\frac{c}{m(m-1)}\sum_{j}v_{j}\,s(X_{j})\,\mathfrak g_{\eta}(X_{j},X_{j})$, is conditionally centered given $\mathbf X$, and has conditional variance
\[
\frac{c^{2}}{m^{2}(m-1)^{2}}\sum_{j}\sigma^{2}(X_{j})\,s(X_{j})^{2}\,\mathfrak g_{\eta}(X_{j},X_{j})^{2}
\le\frac{16\,\sigma^{2}_{\max}\,\bar k^{2}\lambda_{L}^{-2}\,\lVert s\rVert^{2}_{\infty}}{m\,(m-1)^{2}}
=O(m^{-3}),
\]
hence $O_{p}(m^{-1})$ after the average over $j$, uniformly in $\eta$. Collecting,
\[
\Bigl[-2U_{m}+Q_{m}+R_{m}\Bigr]_{\text{mixed}}
=\frac{2}{m}\sum_{j=1}^{m}v_{j}\Bigl[-2\,w^{*}_{\eta}(X_{j})+\kappa_{\eta}(X_{j})\Bigr]+O_{p}(m^{-1})
\quad\text{uniformly in }\eta ,
\]
the bracket collecting the two coefficients as
$-4\,w^{*}_{\eta}+2\,\kappa_{\eta}=2\,(-2\,w^{*}_{\eta}+\kappa_{\eta})$. The
coordinate-wise identity $-2t_{\ell}+t_{\ell}^{2}=(1-t_{\ell})^{2}-1$, with
\eqref{eq:ab} and the multiplier $t_{\ell}=\mu_{\ell}(\gamma)
/(\mu_{\ell}(\gamma)+\lambda)$ of the setup, rewrites the bracket as the
degenerate combination shifted by the signal,
\[
-2\,w^{*}_{\eta}+\kappa_{\eta}=(a_{\eta}-b_{\eta})-s:
\]
both sides have spectral coordinates, $w^{*}_{\eta}$ and $\kappa_{\eta}$
reading at the multipliers $t_{\ell}$ and $t_{\ell}^{2}$ against
$(\Pi f)_{\ell}(\gamma)$ and $a_{\eta}-b_{\eta}$ at $(1-t_{\ell})^{2}$ by
\eqref{eq:ab}, so the left side carries $-2t_{\ell}+t_{\ell}^{2}$ against
$(\Pi f)_{\ell}(\gamma)$ and the right side $(1-t_{\ell})^{2}-1$ against the
same, equal by the binomial identity. The signal shift is supplied by the
level, whose displayed split pairs $v$ against $s$ at
$+\frac{2}{m}\sum_{j}v_{j}\,s(X_{j})$, so the mixed channel and the level
pairing collect into the total noise-borne linear channel
$\frac{2}{m}\sum_{j}v_{j}(a_{\eta}-b_{\eta})(X_{j})$: in coordinates, the total
$v$-coefficient is $2s-4w^{*}_{\eta}+2\kappa_{\eta}
=2\,(s-2w^{*}_{\eta}+\kappa_{\eta})=2\,(a_{\eta}-b_{\eta})$, the binomial
identity rearranged as $1-2t_{\ell}+t_{\ell}^{2}=(1-t_{\ell})^{2}$, applied
twice. The insertion of $\Pi$ is exactly free: the pooled family reads the
channel at $h_{\eta}:=\Pi(a_{\eta}-b_{\eta})$, the component of
$a_{\eta}-b_{\eta}$ orthogonal to the score span, and the discarded component
$s^{\mathrm{sp}}_{\eta}:=(I-\Pi)(a_{\eta}-b_{\eta})
=\tilde c_{\eta}^{\top}g_{\theta_{0}}$ lies in the score span by construction,
$I-\Pi$ being the projection onto the span of the $d_{\theta}$ score functions,
so that, evaluated at the design, its vector is $\hat{\mathbb G}\tilde
c_{\eta}$ and
\[
\frac{1}{m}\sum_{i=1}^{m}v_{i}\,s^{\mathrm{sp}}_{\eta}(X_{i})
=\frac{1}{m}\tilde c_{\eta}^{\top}\hat{\mathbb G}^{\top}v
=0
\]
identically, because $v=\hat{\boldsymbol\Pi}v^{\rm raw}$ is the noise component
of the orthogonalized residual and $\hat{\mathbb G}^{\top}\hat{\boldsymbol\Pi}
=(\hat{\boldsymbol\Pi}\hat{\mathbb G})^{\top}=\mathbf 0$ (the finite-sample
orthogonality of Section~\ref{sec:uniform_limit}; $\hat{\boldsymbol\Pi}$
symmetric and idempotent with $\hat{\boldsymbol\Pi}\hat{\mathbb G}
=\boldsymbol 0$, the columns of $\hat{\mathbb G}$ lying in the span that
$\hat{\boldsymbol\Pi}$ projects off). The vanishing is an algebraic identity at
every tuning and every realized design, so the insertion is free uniformly in
$\eta$. The total noise-borne linear channel of the criterion is therefore the
single empirical process
\begin{equation}
\Bigl[-2U_{m}+Q_{m}+R_{m}\Bigr]_{\text{mixed}}
+\frac{2}{m}\sum_{j=1}^{m}v_{j}\,s(X_{j})
=\frac{2}{m}\sum_{j=1}^{m}v_{j}\,h_{\eta}(X_{j})+O_{p}(m^{-1}),
\qquad
h_{\eta}:=\Pi(a_{\eta}-b_{\eta}),
\label{eq:mixed-b}
\end{equation}
uniformly in $\eta$, and the process enters the pooled family
\eqref{eq:pooled-psi-b} as its $2\,v\,h_{\eta}(x)$ bracket. 

\emph{The signal--signal channel.} Pairing $(s,s)$ is the only bilinear pairing with a nonzero mean, and its centered linear part survives beside it. By the anchor decomposition \eqref{eq:anchor-id} applied to the two kernels,
\[
\begin{aligned}
-2\,U^{(g)}_{m}(\eta)+U^{(g\circ g)}_{m}(\eta)
&=-2\,\theta_{g}(\eta)+\theta_{g\circ g}(\eta)\\
&\quad+\frac{1}{m}\sum_{i=1}^{m}\Bigl[-4\bigl(\Pi f(X_{i})w^{*}_{\eta}(X_{i})-\theta_{g}(\eta)\bigr)\Bigr.\Bigr.\\
&\qquad\qquad\qquad\qquad\Bigl.\Bigl.+2\bigl(\Pi f(X_{i})\kappa_{\eta}(X_{i})-\theta_{g\circ g}(\eta)\bigr)\Bigr]
+O_{p}(m^{-1}),
\end{aligned}
\]
the last estimate being \eqref{eq:Dm-unif}. The deterministic constants collect with the signal level: since $P_{\eta}$ is self-adjoint on $L^{2}(\mathbb P_{X})$ and $\theta_{g}=\langle s,w^{*}_{\eta}\rangle_{L^{2}}$, $\theta_{g\circ g}=\langle s,\kappa_{\eta}\rangle_{L^{2}}$,
\[
\lVert s\rVert_{L^{2}}^{2}-2\,\langle s,P_{\eta}s\rangle_{L^{2}}+\langle s,P_{\eta}^{2}s\rangle_{L^{2}}
=\bigl\langle s,(I-P_{\eta})^{2}s\bigr\rangle_{L^{2}}
=\lVert(I-P_{\eta})s\rVert_{L^{2}}^{2}
=A(\eta).
\]
Assembling the level, whose $2v\,s$ pairing entered \eqref{eq:mixed-b}, the three channels, and $R_{\infty}=\bar\sigma^{2}+A$, the centered, bounded pooled family of the criterion is
\begin{equation}
\begin{aligned}
\psi_{\eta}(x,v):={}&\bigl[v^{2}-\bar\sigma^{2}\bigr]
+\bigl[s(x)^{2}-\lVert s\rVert_{L^{2}}^{2}\bigr]
-4\bigl[s(x)\,w^{*}_{\eta}(x)-\theta_{g}(\eta)\bigr]\\
&+2\bigl[s(x)\,\kappa_{\eta}(x)-\theta_{g\circ g}(\eta)\bigr]
+2\,v\,h_{\eta}(x)
+\bigl[w^{*2}_{\eta}(x)-\theta_{g\circ g}(\eta)\bigr],
\end{aligned}
\label{eq:pooled-psi-b}
\end{equation}
whose bracket means are zero by construction, by $\mathbb E[v\mid X]=0$, and by
$\mathbb E[w^{*2}_{\eta}(X)]=\lVert w^{*}_{\eta}\rVert_{L^{2}}^{2}
=\theta_{g\circ g}(\eta)$, and whose components are bounded or carry bounded
fourth moments: $\lvert v^{2}-\bar\sigma^{2}\rvert$ and $\lvert v\,h_{\eta}(x)
\rvert$ through $\bar M_{4}$ and $\mathbb E[h_{\eta}^{2}(X)\sigma^{2}(X)]\le\sigma^{2}_{\max}\lVert h_{\eta}\rVert^{2}_{L^{2}}\le\sigma^{2}_{\max}\lVert s\rVert^{2}_{L^{2}}$, the $L^{2}(\mathbb P_{X})$-contractions of $(I-P_{\eta})^{2}$ and $\Pi$, the $X$-measurable components through the
envelopes above, the last bracket through $\lvert w^{*2}_{\eta}
-\theta_{g\circ g}\rvert\le 2\,U_{w}^{2}$ with $U_{w}:=(4\lambda_{L})^{-1/2}\bar
k^{1/2}\lVert s\rVert_{L^{2}}$. The six brackets carry the routing labels of the
map: the first is the diagonal of the pairing $(v,v)$, the second the diagonal of
$(s,s)$, the third and fourth are the pairings $(s,v)$, $(v,s)$ with the anchors
$w^{*}_{\eta}$ and $\kappa_{\eta}$ in place of the raw inner factors, the fifth
is the merged noise-borne channel \eqref{eq:mixed-b}, and the sixth is the first
Hoeffding projection of the $(s,s)$ pairing of the reduction remainder
\eqref{eq:Rcal-ss-final}.

\emph{Drift term.}\;
The square completion is performed around the population fit $w^{*}_{\eta}$, and the surviving $r$-pairings (classes $(s,r)+(r,s)$, $(r,v)+(v,r)$, $(r,r)$) are collected into a single pairing against $r$. The $r$-content of the criterion is design-borne and noise-free: $r$ is fixed given the design, so each such pairing measures the estimation error against the signal, against the noise, or against itself, and the term collected from them is the drift of the criterion away from its risk.
\emph{The level.} In the level the $r$-pairings are $2\langle s,r\rangle_{L^{2}(\mathbb P_{m})}$ plus $\mathcal R^{\mathrm{lev}}_{m}=O_{p}(m^{-1})$, as displayed at the level above: the $v\cdot r$ piece is conditionally centered, $\mathbb E[v_{i}\mid\mathbf X]=0$ with $r$ design-fixed, at conditional variance $\frac{4\sigma^{2}_{\max}}{m^{2}}\sum_{i}r(X_{i})^{2}=\frac{4\sigma^{2}_{\max}}{m}\lVert r\rVert^{2}_{L^{2}(\mathbb P_{m})}=O_{p}(m^{-2})$, and the $r^{2}$ piece equals $\lVert r\rVert^{2}_{L^{2}(\mathbb P_{m})}=O_{p}(m^{-1})$ by the perturbation rate. The cross term survives as the level's copy of the pairing $(s,r)$.
\emph{In $-2U_{m}$.} The pairings with one $r$-factor are $-\frac{4}{m(m-1)}\sum_{i\ne j}s(X_{i})\,r(X_{j})\,g_{\eta}(X_{i},X_{j})$, the orderings $(s,r)$ and $(r,s)$ giving equal double sums by the symmetry of $g_{\eta}$. The pairings with two $v$-factors are centered after the raw-noise switch of the noise--noise channel above, at its $O_{p}(m^{-1})$ replacement cost, and die by the second-moment count of Lemma~\ref{lem:pairvar} at the raw-noise covariance $D$. The pairings with a $v$- and an $r$-factor die by conditional Cauchy--Schwarz: with $r$ fixed given the design and $v$ conditionally centered, $\mathbb E\bigl\lVert\bigl(\frac{1}{m-1}\sum_{i\ne j}v_{i}g_{\eta}(X_{i},X_{j})\bigr)_{j}\bigr\rVert^{2}_{\mathbb P_{m}}\le\frac{\sigma^{2}_{\max}\bar k^{2}\lambda_{L}^{-2}}{m-1}$, the square over the anchor indices $i\ne i'$ losing its cross terms to $\mathbb E[v\mid X]=0$ and its diagonal to the kernel envelope $\lvert g_{\eta}\rvert\le\bar k\lambda_{L}^{-1}$, so the second factor is $O_{p}(m^{-1/2})$ uniformly in $\eta$ and the product with the rate $\lVert r\rVert_{L^{2}(\mathbb P_{m})}=O_{p}(m^{-1/2})$ is $O_{p}(m^{-1})$. The $r,r$-pairings are $O_{p}(m^{-1})$: with the kernel envelope, they are bounded by $\frac{4\bar k\lambda_{L}^{-1}}{m(m-1)}\bigl(\sum_{i}\lvert r(X_{i})\rvert\bigr)^{2}$, and the Cauchy--Schwarz inequality $\bigl(\sum_{i}\lvert r(X_{i})\rvert\bigr)^{2}\le m\sum_{i}r(X_{i})^{2}=m^{2}\lVert r\rVert^{2}_{L^{2}(\mathbb P_{m})}$ against the $m^{-2}$ prefactor leaves $O(\lVert r\rVert^{2}_{L^{2}(\mathbb P_{m})})=O_{p}(m^{-1})$, the quadratic pairing dying at the square of the perturbation rate. The inner factor of the surviving pairing is the empirical anchor $\hat\vartheta^{(g)}_{j}:=\frac{1}{m-1}\sum_{i\ne j}s(X_{i})\,g_{\eta}(X_{i},X_{j})$ of the signal--noise channel; the anchor fluctuation is $L^{2}$-small in mean square uniformly in $\eta$ as displayed there, and Cauchy--Schwarz against the design-fixed $r$ costs $\bigl\lvert\frac{4}{m}\sum_{j}r(X_{j})\bigl(\hat\vartheta^{(g)}_{j}-w^{*}_{\eta}(X_{j})\bigr)\bigr\rvert\le4\,\lVert r\rVert_{L^{2}(\mathbb P_{m})}\bigl(\frac{1}{m}\sum_{j}\bigl(\hat\vartheta^{(g)}_{j}-w^{*}_{\eta}(X_{j})\bigr)^{2}\bigr)^{1/2}=O_{p}(m^{-1/2})\,O_{p}(m^{-1/2})=O_{p}(m^{-1})$, so the pairing equals $-4\langle w^{*}_{\eta},r\rangle_{L^{2}(\mathbb P_{m})}+O_{p}(m^{-1})$.
\emph{In $R_{m}$.} The exact reduction $R_{m}=\frac{m-2}{m-1}\,U^{(g\circ g)}_{m}+\mathcal R_{m}/m$ runs the same computation with the composed kernel, whose anchor is $\vartheta_{g\circ g}=\kappa_{\eta}$ by the anchor identities of the setup, and gives $+2\langle\kappa_{\eta},r\rangle_{L^{2}(\mathbb P_{m})}+O_{p}(m^{-1})$. Collecting the three surviving pairings and inserting the coordinate identity $s-2w^{*}_{\eta}+\kappa_{\eta}=a_{\eta}-b_{\eta}$ of \eqref{eq:ab}, the spectral collapse $1-2t_{\ell}+t_{\ell}^{2}=(1-t_{\ell})^{2}$ already used by the signal--noise channel,
\begin{equation}
\mathrm{CV}_{m}(\eta)-R_{\infty}(\eta)
=\frac{1}{m}\sum_{i=1}^{m}\psi_{\eta}(X_{i},v_{i})
+2\bigl\langle a_{\eta}-b_{\eta},\,r\bigr\rangle_{L^{2}(\mathbb P_{m})}
+O_{p}(m^{-1/2})
\qquad\text{uniformly in }\eta,
\label{eq:level-collect}
\end{equation}
the drift term being denoted $\mathrm{Drift}_{m}(\eta):=2\langle a_{\eta}-b_{\eta},\,r\rangle_{L^{2}(\mathbb P_{m})}$ in the sequel. The drift is the only surviving coupling of the criterion to the estimation error: it is linear in $r$, carries no noise, and every other $r$-bearing pairing has died above at $O_{p}(m^{-1})$. For its bound, the contraction $\lVert(I-P_{\eta})^{2}\rVert_{\mathrm{op}}\le1$ behind $a_{\eta}-b_{\eta}=(I-P_{\eta})^{2}s$ holds on $L^{2}(\mathbb P_{X})$ and not on the empirical metric $L^{2}(\mathbb P_{m})$ in which the pairing is taken: the contractivity of $P_{\eta}$ is spectral, $P_{\eta}$ being a function of the population covariance with multipliers $t_{\ell}\in[0,1]$, and the empirical measure induces a different geometry, in which no bound of $\lVert(I-P_{\eta})^{2}\boldsymbol u\rVert_{L^{2}(\mathbb P_{m})}$ by $\lVert\boldsymbol u\rVert_{L^{2}(\mathbb P_{m})}$ follows. The route is instead a deterministic sup-norm envelope of the coefficient function, followed by Cauchy--Schwarz in $L^{2}(\mathbb P_{m})$. The setup supplies, uniformly in $(\eta,x)$, $\lvert s(x)\rvert\le\lVert s\rVert_{\infty}$ and $\lvert w^{*}_{\eta}(x)\rvert\le U_{w}$ with $U_{w}:=(4\lambda_{L})^{-1/2}\bar k^{1/2}\lVert s\rVert_{L^{2}}$, the anchor bound being the RKHS argument of the setup, $\lvert w^{*}_{\eta}(x)\rvert\le\bar k^{1/2}\lVert w^{*}_{\eta}\rVert_{\mathcal H_{\gamma}}$ with $\lVert w^{*}_{\eta}\rVert^{2}_{\mathcal H_{\gamma}}=\sum_{\ell}\frac{\mu_{\ell}}{(\mu_{\ell}+\lambda)^{2}}\,(\Pi f)_{\ell}^{2}\le(4\lambda_{L})^{-1}\lVert s\rVert^{2}_{L^{2}}$ by the AM--GM bound $\mu+\lambda\ge2\sqrt{\mu\lambda}$ on each spectral coordinate; and $\lvert b_{\eta}(x)\rvert\le\bar k^{1/2}\lVert a_{\eta}\rVert_{L^{2}}\le\bar k^{1/2}\lVert s\rVert_{L^{2}}$, the last step being the $L^{2}(\mathbb P_{X})$-contraction of $P_{\eta}$, legitimate here because the anchor bound is a population statement; the triangle inequality $\lvert a_{\eta}(x)\rvert\le\lvert s(x)\rvert+\lvert w^{*}_{\eta}(x)\rvert$ and $\lvert(a_{\eta}-b_{\eta})(x)\rvert\le\lvert a_{\eta}(x)\rvert+\lvert b_{\eta}(x)\rvert$ give
\begin{equation}
C_{\mathrm{dr}}:=\lVert s\rVert_{\infty}+U_{w}+\bar k^{1/2}\lVert s\rVert_{L^{2}}
\qquad\Longrightarrow\qquad
\sup_{x\in\mathcal X}\bigl\lvert(a_{\eta}-b_{\eta})(x)\bigr\rvert\le C_{\mathrm{dr}}
\quad\text{for all }\eta\in\Upsilon .
\label{eq:driftenv}
\end{equation}
the constant $C_{\mathrm{dr}}$ being deterministic and $\eta$-free, so the envelope is uniform over $\Upsilon$ and no stochastic equicontinuity enters the drift bound. Cauchy--Schwarz in $L^{2}(\mathbb P_{m})$, $\bigl\lvert\frac{2}{m}\sum_{i}(a_{\eta}-b_{\eta})(X_{i})\,r(X_{i})\bigr\rvert\le2\bigl(\frac{1}{m}\sum_{i}(a_{\eta}-b_{\eta})(X_{i})^{2}\bigr)^{1/2}\lVert r\rVert_{L^{2}(\mathbb P_{m})}$, followed by \eqref{eq:driftenv} and the perturbation rate $\lVert r\rVert_{L^{2}(\mathbb P_{m})}=O_{p}(m^{-1/2})$ of the setup, gives
\begin{equation}
\sup_{\eta\in\Upsilon}\bigl\lvert\mathrm{Drift}_{m}(\eta)\bigr\rvert\le2\,C_{\mathrm{dr}}\,\lVert r\rVert_{L^{2}(\mathbb P_{m})}=O_{p}\bigl(m^{-1/2}\bigr).
\label{eq:driftm}
\end{equation}
The bound is deliberately crude, and deliberately so: the assembly's triangle over \eqref{eq:level-collect} consumes only $O_{p}(m^{-1/2})$ for the drift, so no influence-function expansion of the estimation error is ever taken; the exact form of $\mathrm{Drift}_{m}$ is on record in \eqref{eq:level-collect}, whose $\sqrt m$-scaled form is reused at the end of the proof.

\begin{lemma}[Moment and smoothness bookkeeping of the pooled family]\label{lem:fam_psi}
Read with the raw noise, $Z_{i}:=(X_{i},v^{\rm raw}_{i})$, the family $\psi_{\eta}$ of \eqref{eq:pooled-psi-b}, viewed as a function $\psi_{\eta}(x,u)$ of the design point and a real argument, satisfies:
\emph{(i)} $\sup_{\eta\in\Upsilon}\bigl\lvert\mathbb E\,[\psi_{\eta}(Z_{1})]\bigr\rvert\le c_{\sigma}/m$ for the deterministic constant $c_{\sigma}:=\sigma^{2}_{\max}d_{\theta}$;
\emph{(ii)} $\sup_{\eta\in\Upsilon}\mathbb E\,[\psi_{\eta}(Z_{1})^{2}]\le M_{2}$;
\emph{(iii)} the paths $\eta\mapsto\psi_{\eta}(x,u)$ are $C^{1}$ on $\Upsilon$ for every $(x,u)$, the pointwise gradient envelopes $\lvert\nabla_{\eta}w^{*}_{\eta}(x)\rvert\le(G_{1}\vee G_{2})\lVert s\rVert_{L^{2}}$ and $\lvert\nabla_{\eta}h_{\eta}(x)\rvert\le C^{(\nabla)}_{h}$ hold for all $(\eta,x)$, and $\sup_{\eta\in\Upsilon}\mathbb E\,[\lVert\nabla_{\eta}\psi_{\eta}(Z_{1})\rVert^{4}]\le N_{4}$;
\emph{(iv)} the replacement cost at the criterion scale is $\sup_{\eta\in\Upsilon}\bigl\lvert\frac{1}{m}\sum_{i=1}^{m}\bigl[\psi_{\eta}(X_{i},v_{i})-\psi_{\eta}(X_{i},v^{\rm raw}_{i})\bigr]\bigr\rvert=O_{p}(m^{-1})$.
The constants $M_{2}$, $N_{4}$, $c_{\sigma}$ are deterministic and free of $(m,\eta,\mathbf X)$.
\end{lemma}

\begin{proof}
\emph{Means.} The bracket means are zero by construction, by $\mathbb E[v\mid X]=0$, and by $\mathbb E[w^{*2}_{\eta}(X)]=\lVert w^{*}_{\eta}\rVert_{L^{2}}^{2}=\theta_{g\circ g}(\eta)$, as displayed after \eqref{eq:pooled-psi-b}; that display is read with the noise component $v$, and the only bracket whose mean distinguishes $v^{\rm raw}$ from $v$ is the first, since $\mathbb E[v^{\rm raw}h_{\eta}(X)]=\mathbb E\bigl[h_{\eta}(X)\,\mathbb E[v^{\rm raw}\mid X]\bigr]=0$ like its $v$-counterpart. Now $\mathbb E[(v^{\rm raw})^{2}]=\mathbb E[\sigma^{2}(X)]$, while $\bar\sigma^{2}=\mathbb E[v^{2}]$ with $\mathbb E[v_{i}^{2}\mid\mathbf X]=[\hat{\boldsymbol\Pi}\boldsymbol\Sigma\hat{\boldsymbol\Pi}]_{ii}$, so that $\bar\sigma^{2}=m^{-1}\mathbb E[\operatorname{tr}(\hat{\boldsymbol\Pi}\boldsymbol\Sigma)]$ by the cyclic trace formula and the idempotence of $\hat{\boldsymbol\Pi}$; hence
\[
\bigl\lvert\mathbb E[\sigma^{2}(X)]-\bar\sigma^{2}\bigr\rvert
=\frac{1}{m}\,\bigl\lvert\mathbb E\bigl[\operatorname{tr}\{(\boldsymbol I_{m}-\hat{\boldsymbol\Pi})\boldsymbol\Sigma\}\bigr]\bigr\rvert
\le\frac{\sigma^{2}_{\max}d_{\theta}}{m},
\]
the last step being the replacement identity of the setup, which gives \emph{(i)}.

\emph{Second moment.} Write $B_{1},\dots,B_{6}$ for the six brackets of \eqref{eq:pooled-psi-b}. Minkowski's inequality in $L^{2}$ and the envelope bookkeeping of the setup give $\lVert B_{1}\rVert_{L^{2}}\le(2\bar M_{4}+2\sigma^{4}_{\max})^{1/2}$ by $(a-b)^{2}\le2a^{2}+2b^{2}$, $\mathbb E[(v^{\rm raw})^{4}\mid X]\le\bar M_{4}$, and $\bar\sigma^{2}\le\mathbb E[\sigma^{2}(X)]\le\sigma^{2}_{\max}$; $\lVert B_{2}\rVert_{L^{2}}\le\lVert s\rVert_{\infty}^{2}+\lVert s\rVert_{L^{2}}^{2}$ pointwise; $\lVert B_{3}\rVert_{L^{2}}\le4\bigl(\lVert s\rVert_{\infty}U_{w}+\lVert s\rVert_{L^{2}}^{2}\bigr)$ and $\lVert B_{4}\rVert_{L^{2}}\le2\bigl(\lVert s\rVert_{\infty}U_{w}\bar k^{1/2}+\lVert s\rVert_{L^{2}}^{2}\bigr)$, the two constants being bounded by $\lvert\theta_{g}\rvert,\lvert\theta_{g\circ g}\rvert\le\lVert s\rVert_{L^{2}}^{2}$ through Cauchy--Schwarz and the $L^{2}(\mathbb P_{X})$-contraction of $P_{\eta}$; $\lVert B_{5}\rVert_{L^{2}}\le2\,\sigma_{\max}\lVert s\rVert_{L^{2}}$ by $\mathbb E[h_{\eta}^{2}(X)\,(v^{\rm raw})^{2}]=\mathbb E[h_{\eta}^{2}(X)\,\sigma^{2}(X)]\le\sigma^{2}_{\max}\lVert h_{\eta}\rVert^{2}_{L^{2}}\le\sigma^{2}_{\max}\lVert s\rVert^{2}_{L^{2}}$, the $L^{2}(\mathbb P_{X})$-contraction of $\Pi$ behind the last step; and $\lVert B_{6}\rVert_{L^{2}}\le2\,U_{w}^{2}$ as displayed after \eqref{eq:pooled-psi-b}. Hence $\sup_{\eta}\mathbb E[\psi_{\eta}(Z_{1})^{2}]\le\bigl(\sum_{k=1}^{6}\lVert B_{k}\rVert_{L^{2}}\bigr)^{2}=:M_{2}$, which is \emph{(ii)}.

\emph{Gradient.} Term-by-term differentiation in $\eta$ is legitimate by dominated convergence through the deterministic envelopes, Step A of the proof of part (a) for the kernel entries and the anchor gradients of the anchor-decomposition paragraph. For the $X$-measurable components, $\nabla_{\eta}\bigl[-4(sw^{*}_{\eta}-\theta_{g})+2(s\kappa_{\eta}-\theta_{g\circ g})\bigr](x)$ and $\nabla_{\eta}\bigl[w^{*2}_{\eta}-\theta_{g\circ g}\bigr](x)$ have coordinates bounded by $C_{\nabla}:=(13+2\,(4\lambda_{L})^{-1/2}\bar k^{1/2})(G_{1}\vee G_{2})\lVert s\rVert_{\infty}\lVert s\rVert_{L^{2}}$, the product terms contributing $4+2$ and $2\,U_{w}\le2\,(4\lambda_{L})^{-1/2}\bar k^{1/2}\lVert s\rVert_{\infty}$ for the sixth bracket, and the mean-gradient terms $\le(G_{1}\vee G_{2})\lVert s\rVert_{L^{2}}^{2}\le(G_{1}\vee G_{2})\lVert s\rVert_{\infty}\lVert s\rVert_{L^{2}}$ each; for the noise-borne component $\nabla_{\eta}[2uh_{\eta}]=2u\,\nabla_{\eta}h_{\eta}$ has coordinates bounded by $2\lvert u\rvert\,C^{(\nabla)}_{h}$, where, with $C_{\Pi}$ the score-span sup-norm constant, $\lVert\Pi g\rVert_{\infty}\le C_{\Pi}\lVert g\rVert_{\infty}$ for every bounded $g$ by the bounded score functions of Assumption~\ref{ass:E}(i), the constant re-derived by Riesz--Thorin in part~(c) below, the mean-value form of the anchor representations gives $\lvert\nabla_{\eta}h_{\eta}(x)\rvert\le C_{\Pi}\lVert\nabla_{\eta}(a_{\eta}-b_{\eta})\rVert_{\infty}\le C_{\Pi}\,(2+\bar k^{1/2})(G_{1}\vee G_{2})\lVert s\rVert_{L^{2}}=:C^{(\nabla)}_{h}$, the bound on $\lVert\nabla_{\eta}(a_{\eta}-b_{\eta})\rVert_{\infty}$ including the term $\lVert(I-P_{\eta})(\nabla_{\eta}P_{\eta})s\rVert_{\infty}\le(1+\bar k^{1/2})(G_{1}\vee G_{2})\lVert s\rVert_{L^{2}}$ alongside the two resolvent-identity terms. With $(a+b)^{4}\le8a^{4}+8b^{4}$, the coordinate conversion $\lVert(g_{1},g_{2})\rVert^{4}\le2g_{1}^{4}+2g_{2}^{4}$, and $\mathbb E[(v^{\rm raw})^{4}\mid X]\le\bar M_{4}$,
\[
\sup_{\eta\in\Upsilon}\mathbb E\bigl[\lVert\nabla_{\eta}\psi_{\eta}(Z_{1})\rVert^{4}\bigr]
\le512\,\bigl(C_{\nabla}^{4}+(C^{(\nabla)}_{h})^{4}\,\bar M_{4}\bigr)=:N_{4},
\]
which is \emph{(iii)}, the pathwise differentiability holding for every fixed $(x,u)$.

\emph{Replacement.} Only $B_{1}$ and $B_{5}$ distinguish $v$ from $v^{\rm raw}$. For $B_{1}$ the setup's exact identity gives $m^{-1}\sum_{i}(v_{i}^{2}-v^{\rm raw\,2}_{i})=-\lVert\mathbf d\rVert^{2}/m=O_{p}(m^{-1})$; for $B_{5}$ the pairing $2m^{-1}\sum_{i}h_{\eta}(X_{i})(v_{i}-v^{\rm raw}_{i})=-2m^{-1}\boldsymbol h_{\eta}^{\top}\mathbf d$ has conditional mean zero and conditional variance $4m^{-2}\boldsymbol h_{\eta}^{\top}(\boldsymbol I_{m}-\hat{\boldsymbol\Pi})\boldsymbol\Sigma(\boldsymbol I_{m}-\hat{\boldsymbol\Pi})\boldsymbol h_{\eta}\le4\sigma^{2}_{\max}m^{-2}\lVert(\boldsymbol I_{m}-\hat{\boldsymbol\Pi})\boldsymbol h_{\eta}\rVert^{2}=O_{p}(m^{-2})$ uniformly in $\eta$ by the setup's replacement table, whence conditional Chebyshev's inequality gives $O_{p}(m^{-1})$ uniformly in $\eta$. This is \emph{(iv)}.
\end{proof}

\begin{lemma}[Uniform law of large numbers for smooth families]\label{lem:coda_unif}
Let $Z_{1},\dots,Z_{m}$ be i.i.d.\ in a measurable space and let $\{\phi_{\eta}\}_{\eta\in\Upsilon}$ be real-valued measurable functions on it such that, for deterministic constants and a deterministic sequence:
\emph{(i)} $\sup_{\eta\in\Upsilon}\bigl\lvert\mathbb E[\phi_{\eta}(Z_{1})]\bigr\rvert\le\mu_{m}$;
\emph{(ii)} $\sup_{\eta\in\Upsilon}\mathbb E[\phi_{\eta}(Z_{1})^{2}]\le M_{2}$;
\emph{(iii)} the paths $\eta\mapsto\phi_{\eta}(z)$ are $C^{1}$ on $\Upsilon$ for every $z$, and $\sup_{\eta\in\Upsilon}\mathbb E[\lVert\nabla_{\eta}\phi_{\eta}(Z_{1})\rVert^{4}]\le N_{4}$.
Then
\[
\sup_{\eta\in\Upsilon}\Bigl\lvert\frac{1}{m}\sum_{i=1}^{m}\phi_{\eta}(Z_{i})\Bigr\rvert=O_{p}\bigl(m^{-1/2}\bigr),
\]
the $O_{p}$-constant depending only on $(M_{2},N_{4},\Upsilon)$, with $\mu_{m}=O(m^{-1})$ absorbed.
\end{lemma}

\begin{proof}
The argument is unconditional: the only probabilistic structure consumed is the i.i.d.\ law of $(Z_{i})$, and all moments invoked are unconditional. Fix, once and for all, a partition of the compact rectangle $\Upsilon$ into cells $\{K\}$, each contained in a square of side $\delta$ contained in $\Upsilon$, and let $\mathcal M_{\delta}$ be the set of corner points, with $\lvert\mathcal M_{\delta}\rvert\le N(\delta)\le C_{\Upsilon}\,\delta^{-2}$; $\delta$ is fixed throughout and enters only constants.

\emph{Mesh level.} Write $\bar\phi_{\eta}:=m^{-1}\sum_{i}\phi_{\eta}(Z_{i})$. The variables $\phi_{\eta'}(Z_{i})-\mathbb E[\phi_{\eta'}(Z_{1})]$ are i.i.d., centered, with second moment at most $M_{2}$, so $\mathbb E\bigl[(\bar\phi_{\eta'}-\mathbb E[\phi_{\eta'}(Z_{1})])^{2}\bigr]=m^{-1}\,\mathbb E[(\phi_{\eta'}(Z_{1})-\mathbb E[\phi_{\eta'}(Z_{1})])^{2}]\le M_{2}/m$, and Chebyshev's inequality together with the deterministic offset of \emph{(i)} gives, for every $\eta'\in\mathcal M_{\delta}$,
\[
\mathbb P\Bigl(\lvert\bar\phi_{\eta'}\rvert>\varepsilon m^{-1/2}+\mu_{m}\Bigr)
\le\mathbb P\Bigl(\lvert\bar\phi_{\eta'}-\mathbb E[\phi_{\eta'}(Z_{1})]\rvert>\varepsilon m^{-1/2}\Bigr)
\le\frac{M_{2}}{\varepsilon^{2}} ,
\]
and the union bound over $\mathcal M_{\delta}$ multiplies the right side by $N(\delta)$.

\emph{Cell oscillation.} Since $\mathbb E[\bar\phi_{\eta}]=\mathbb E[\phi_{\eta}(Z_{1})]$ for every $\eta$, the oscillation of $\bar\phi_{\eta}$ over a cell equals that of the centered field $\tilde\phi_{\eta}:=\bar\phi_{\eta}-\mathbb E[\phi_{\eta}(Z_{1})]$, whose gradient is $\nabla_{\eta}\tilde\phi_{\eta}=\frac{1}{m}\sum_{i}\bigl(\nabla_{\eta}\phi_{\eta}(Z_{i})-\bar g_{\eta}\bigr)$ with the deterministic vector $\bar g_{\eta}:=\mathbb E[\nabla_{\eta}\phi_{\eta}(Z_{1})]$; centering of the gradient is never assumed, it is inserted by hand. Lyapunov's inequality gives $\lVert\bar g_{\eta}\rVert^{2}\le\mathbb E\lVert\nabla_{\eta}\phi_{\eta}(Z_{1})\rVert^{2}\le N_{4}^{1/2}$, so $\mathbb E\lVert\nabla_{\eta}\phi_{\eta}(Z_{1})-\bar g_{\eta}\rVert^{4}\le8\,(N_{4}+N_{4})=16\,N_{4}$ by $(a+b)^{4}\le8a^{4}+8b^{4}$. Each coordinate of the centered gradient average is a sum of i.i.d.\ centered summands with bounded fourth moments, so the Marcinkiewicz--Zygmund inequality in the form $\mathbb E\bigl\lvert\sum_{i=1}^{m}\xi_{i}\bigr\rvert^{4}\le C_{4}\,m^{2}\,\mathbb E[\xi_{1}^{4}]$ for independent centered $(\xi_{i})$, followed by the coordinate conversion $\lVert(g_{1},g_{2})\rVert^{4}\le2g_{1}^{4}+2g_{2}^{4}$, gives
\[
\sup_{\eta\in\Upsilon}\mathbb E\bigl\lVert\nabla_{\eta}\tilde\phi_{\eta}\bigr\rVert^{4}\le C_{\mathrm{MZ}}\,N_{4}\,m^{-2}
\]
for a numerical constant $C_{\mathrm{MZ}}$. The paths $t\mapsto\tilde\phi_{t}(\omega)$ are $C^{1}$ for every $\omega$, with gradient the average of the summand gradients, so Lemma~\ref{lem:sobolev} with $p=4>d=2$ on each cell $K$ gives $\mathbb E\bigl[\sup_{t,s\in K}\lvert\tilde\phi_{t}-\tilde\phi_{s}\rvert^{4}\bigr]\le C_{2,4}\,C_{\mathrm{MZ}}\,N_{4}\,\delta^{4}\,m^{-2}$, and Markov's inequality at the level $\varepsilon\,\delta\,m^{-1/2}$ gives a bound $\le C_{2,4}\,C_{\mathrm{MZ}}\,N_{4}/\varepsilon^{4}$ that is free of $m$; the union bound over the cells preserves $m$-freeness and multiplies by $N(\delta)\le C_{\Upsilon}\delta^{-2}$.

\emph{Assembly.} Every $\eta\in\Upsilon$ lies within distance $\delta$ of a corner $\eta'\in\mathcal M_{\delta}$, so, by the triangle inequality,
\[
\sup_{\eta\in\Upsilon}\lvert\bar\phi_{\eta}\rvert
\le\max_{\eta'\in\mathcal M_{\delta}}\bigl\lvert\bar\phi_{\eta'}-\mathbb E[\phi_{\eta'}(Z_{1})]\bigr\rvert
+\max_{K}\ \sup_{t,s\in K}\lvert\tilde\phi_{t}-\tilde\phi_{s}\rvert
+\mu_{m} .
\]
On the intersection of the two good events, $m^{1/2}\sup_{\eta}\lvert\bar\phi_{\eta}\rvert\le\varepsilon(1+\delta)+m^{1/2}\mu_{m}\le\varepsilon(1+\delta)+c_{0}$ for all large $m$, the values at finitely many small $m$ being absorbed into $c_{0}$, since $\mu_{m}=O(m^{-1})$; the complement has probability at most $C_{\Upsilon}\delta^{-2}\bigl(M_{2}/\varepsilon^{2}+C_{2,4}C_{\mathrm{MZ}}N_{4}/\varepsilon^{4}\bigr)$, which tends to $0$ as $\varepsilon\to\infty$. This is the claim.
\end{proof}

\begin{lemma}[Pooled central limit argument]\label{lem:pooled_clt}
Let $Z_{1},\dots,Z_{m}$ be i.i.d.\ and let $\{\phi_{\eta}\}_{\eta\in\Upsilon}$ satisfy hypotheses \emph{(i)--(iii)} of Lemma~\ref{lem:coda_unif}. Then
\[
\Bigl(m^{-1/2}\sum_{i=1}^{m}\phi_{\eta}(Z_{i})\Bigr)_{\eta\in\Upsilon}
\;\rightsquigarrow\;\mathbb G^{\mathrm{pooled}}
\quad\text{in }\ell^{\infty}(\Upsilon),
\]
where $\mathbb G^{\mathrm{pooled}}$ is a mean-zero Gaussian process with continuous paths and covariance
\[
\operatorname{Cov}\bigl(\mathbb G^{\mathrm{pooled}}(\eta),\mathbb G^{\mathrm{pooled}}(\eta')\bigr)
=\mathbb E\bigl[\phi_{\eta}(Z_{1})\,\phi_{\eta'}(Z_{1})\bigr]
\qquad\text{for }(\eta,\eta')\in\Upsilon^{2}.
\]
The conclusion persists for a triangular array $\phi_{\eta,m}=\phi_{\eta}+c_{m}$ with $\sup_{\eta}\lvert c_{m}\rvert=O(m^{-1})$ and $\{\phi_{\eta}\}$ an $m$-free family satisfying the same hypotheses: the finite-dimensional limits and the covariances are those of $\phi_{\eta}$, the recentering being the first paragraph of the proof.
\end{lemma}

\begin{proof}
\emph{Finite-dimensional limits.} For $\eta_{1},\dots,\eta_{L}\in\Upsilon$ and weights $(t_{1},\dots,t_{L})$ put $\xi:=\sum_{l}t_{l}\,\phi_{\eta_{l}}(Z_{1})$: the summands $\xi_{i}:=\xi(Z_{i})$ are i.i.d., with mean $\mu_{\xi}=\sum_{l}t_{l}\,\mathbb E[\phi_{\eta_{l}}(Z_{1})]$ satisfying $\lvert\mu_{\xi}\rvert\le\lVert t\rVert_{1}\,\mu_{m}$ and second moment $\sigma^{2}_{\xi}\le\mathbb E[\xi^{2}]\le\lVert t\rVert_{1}^{2}M_{2}$ by Cauchy--Schwarz. The Lindeberg condition is trivial; two forms are recorded, since both are consumed below. Under a finite-fourth-moment hypothesis on the family at hand, $\sup_{\eta}\mathbb E[\phi_{\eta}^{4}(Z_{1})]\le M_{4}$, H\"older's inequality and $\lvert\xi\rvert\le\lVert t\rVert_{1}\max_{l}\lvert\phi_{\eta_{l}}\rvert$ give
\[
\mathbb E\bigl[\xi^{2}\,1(\lvert\xi\rvert>\varepsilon\sqrt m)\bigr]
\le\frac{1}{\varepsilon^{2}m}\,\mathbb E[\xi^{4}]
\le\frac{\lVert t\rVert_{1}^{4}\,M_{4}}{\varepsilon^{2}\,m}
\;\longrightarrow\;0 ;
\]
under the second-moment hypothesis \emph{(ii)} for an $m$-free family, the expectation is the tail of the convergent integral $\mathbb E[\xi^{2}]<\infty$ and vanishes by dominated convergence; the application below consumes an $m$-dependent family and reduces it by recentering, $\phi_{\eta,m}=\bar\psi_{\eta}+c_{m}$ with $\bar\psi_{\eta}:=\psi_{\eta}+\bar\sigma^{2}-\mathbb E[\sigma^{2}(X)]$ free of $m$ with exact zero means and $\sup_{\eta}\lvert c_{m}\rvert=\sup_{\eta}\lvert\psi_{\eta}-\bar\psi_{\eta}\rvert\le\sigma^{2}_{\max}d_{\theta}/m$, so that $\mathbb E[\xi_{m}^{2}1(\lvert\xi_{m}\rvert>\varepsilon\sqrt m)]\le2\,\mathbb E[\bar\xi^{2}1(\lvert\bar\xi\rvert>\tfrac{\varepsilon}{2}\sqrt m)]+2\lVert t\rVert_{1}^{2}\sigma^{4}_{\max}d_{\theta}^{2}m^{-2}$, again vanishing. The Lindeberg--L\'evy central limit theorem applied to $m^{-1/2}\sum_{i}(\xi_{i}-\mu_{\xi})$, whose mean shift $m^{1/2}\mu_{\xi}$ is $O(m^{-1/2})\to0$, gives $m^{-1/2}\sum_{i}\xi_{i}\Rightarrow N(0,\sigma^{2})$ with limit variance $\sigma^{2}=\mathbb E[\xi^{2}]-\lim_{m}\mu_{\xi}^{2}=\sum_{l,l'}t_{l}t_{l'}\,\mathbb E[\phi_{\eta_{l}}(Z_{1})\phi_{\eta_{l'}}(Z_{1})]$ for an $m$-free family, the formula holding verbatim for the recentered family $\bar\psi_{\eta}$ above, whose covariances are $\mathbb E[\bar\psi_{\eta}(Z_{1})\bar\psi_{\eta'}(Z_{1})]$, and Cram\'er--Wold over arbitrary finite $(\eta_{l},t_{l})$ yields jointly Gaussian finite-dimensional limits with the displayed covariances and zero means.

\emph{Equicontinuity.} The centered field $\tilde\phi_{\eta}:=m^{-1/2}\sum_{i}\bigl(\phi_{\eta}(Z_{i})-\mathbb E[\phi_{\eta}(Z_{1})]\bigr)$ has the same oscillations as the field of the statement, its mean being the deterministic null field, and its gradient is the $\sqrt m$-rescaled gradient of the field of Lemma~\ref{lem:coda_unif}; the argument of that proof therefore runs verbatim at the new scaling. The gradient fourth moment $C_{\mathrm{MZ}}N_{4}\,m^{-2}$ becomes $C_{\mathrm{MZ}}N_{4}$, the $m^{-2}$ of the scaling cancelling the $m^{2}$ of the Marcinkiewicz--Zygmund inequality; Lemma~\ref{lem:sobolev} with $p=4>d=2$ and Markov's inequality at the level $\varepsilon\delta$ give cell oscillations $O_{p}(\delta)$ at an $m$-free level; the union bound over the cells sends $N(\delta)\,\delta^{4}\le C_{\Upsilon}\,\delta^{2}$ to zero as $\delta\downarrow0$, and the mesh level is $O_{p}(1)$ by Chebyshev's inequality under \emph{(ii)} with a union bound over $\mathcal M_{\delta}$: the field is asymptotically equicontinuous.

\emph{Continuity of the limit.} The $C^{1}$ paths and Jensen's inequality over the segment give $\mathbb E\bigl[(\phi_{\eta}-\phi_{\eta'})(Z_{1})^{2}\bigr]\le\lVert\eta-\eta'\rVert^{2}\int_{0}^{1}\mathbb E\lVert\nabla_{\eta}\phi_{\eta+s(\eta'-\eta)}(Z_{1})\rVert^{2}ds\le N_{4}^{1/2}\lVert\eta-\eta'\rVert^{2}$, so the covariance function of $\mathbb G^{\mathrm{pooled}}$ is Lipschitz in $(\eta,\eta')$ at that square order; the Gaussian fourth-moment formula $\mathbb E[(\mathbb G^{\mathrm{pooled}}(\eta)-\mathbb G^{\mathrm{pooled}}(\eta'))^{4}]=3\,\mathbb E[(\mathbb G^{\mathrm{pooled}}(\eta)-\mathbb G^{\mathrm{pooled}}(\eta'))^{2}]^{2}\le3\,N_{4}\lVert\eta-\eta'\rVert^{4}$ with the Kolmogorov--Chentsov criterion, $p=4>\dim\Upsilon$, yields continuous sample paths.

\emph{Lift.} Boundedness in probability of the mesh level, asymptotic equicontinuity, and the finite-dimensional convergence give weak convergence in $\ell^{\infty}(\Upsilon)$ \citep[Th.~1.5.4]{vaart1997weak}.
\end{proof}

\emph{Assembly.}\;
By Lemma~\ref{lem:fam_psi}\emph{(iv)} and Lemma~\ref{lem:coda_unif} applied to the family $\psi_{\eta}$ through Lemma~\ref{lem:fam_psi}\emph{(i)--(iii)} with $\mu_{m}=c_{\sigma}/m$,
\[
\sup_{\eta\in\Upsilon}\Bigl\lvert\frac{1}{m}\sum_{i=1}^{m}\psi_{\eta}(X_{i},v_{i})\Bigr\rvert
\le\sup_{\eta\in\Upsilon}\Bigl\lvert\frac{1}{m}\sum_{i=1}^{m}\psi_{\eta}(X_{i},v^{\rm raw}_{i})\Bigr\rvert
+O_{p}(m^{-1})
=O_{p}\bigl(m^{-1/2}\bigr),
\]
and the triangle inequality over the collection \eqref{eq:level-collect} and the drift bound \eqref{eq:driftm} gives
\begin{equation}
\sup_{\eta\in\Upsilon}\bigl\lvert\mathrm{CV}_{m}(\eta)-R_{\infty}(\eta)\bigr\rvert=O_{p}\bigl(m^{-1/2}\bigr).
\label{eq:level-final}
\end{equation}
Since $R_{\infty}-\bar\sigma^{2}=A$ is deterministic continuous on compact $\Upsilon$, $\operatorname{dist}(\hat\eta_{m},\arg\min_{\Upsilon}A)\xrightarrow{p}0$ for any random minimizer $\hat\eta_{m}$ of $\mathrm{CV}_{m}$, by part~(i) of Lemma~\ref{lem:argmin}. This proves (b).

\begin{remark}[Fluctuation of the criterion, demoted]\label{rem:cv_fluct_demoted}
The proof consumes none of the fluctuation structure beyond the level statement, and no downstream result invokes a limit law for $\sqrt m\,[\mathrm{CV}_{m}-R_{\infty}]$; two ingredients are recorded for completeness. First, \eqref{eq:level-final} is equivalent to the process boundedness $\sup_{\eta\in\Upsilon}\bigl\lvert\sqrt m\,[\mathrm{CV}_{m}(\eta)-R_{\infty}(\eta)]\bigr\rvert=O_{p}(1)$. Second, \eqref{eq:level-collect} gives, uniformly in $\eta$,
\[
\sqrt m\,\bigl[\mathrm{CV}_{m}(\eta)-R_{\infty}(\eta)\bigr]
=\frac{1}{\sqrt m}\sum_{i=1}^{m}\psi_{\eta}(X_{i},v_{i})
+\sqrt m\,\mathrm{Drift}_{m}(\eta)
+O_{p}(1)
\qquad\text{uniformly in }\eta,
\]
where, by the triangular-array clause of Lemma~\ref{lem:pooled_clt} applied to the recentered family $\bar\psi_{\eta}:=\psi_{\eta}+\bar\sigma^{2}-\mathbb E[\sigma^{2}(X)]$, which is free of $m$ with exact zero means and agrees with $\psi_{\eta}$ up to $\sup_{\eta}\lvert\psi_{\eta}-\bar\psi_{\eta}\rvert\le\sigma^{2}_{\max}d_{\theta}/m$ (so the transfer to the sums of $\psi_{\eta}$ themselves costs $\sqrt m\sup_{\eta}\lvert\psi_{\eta}-\bar\psi_{\eta}\rvert\to0$), the pooled sums have mean-zero Gaussian finite-dimensional limits with covariances $\mathbb E[\bar\psi_{\eta}(Z_{1})\bar\psi_{\eta'}(Z_{1})]$. The drift term is $O_{p}(1)$ and pathwise Lipschitz in $\eta$: the bound \eqref{eq:driftm}, the envelope \eqref{eq:driftenv}, and the pointwise Lipschitz moduli of the anchor families behind Lemma~\ref{lem:fam_psi}\emph{(iii)}. No process-level limit law for the criterion $\sqrt m[\mathrm{CV}_{m}-R_{\infty}]$ is claimed, and none is consumed elsewhere in the paper.
\end{remark}

\medskip\noindent\textbf{Proof of (c): local alternative.}\;

Let \(f_{m}=\delta_{m}h\) with \(\|h\|_{L^{2}}<\infty\),
\(\mathbb E[h^{4}(X)]<\infty\),
\(\sqrt m\,\delta_{m}\to a\in[0,\infty)\); the direction \(h\) is effective, \(h=\boldsymbol\Pi h\): a
score-span component of \(h\) is carried asymptotically by the projected-score piece of the residual
decomposition, adding the \(\eta\)-free shift \(a^{2}\,\mu_{h}^{\top}\Gamma_{\theta_0}^{-1}\mu_{h}\) with
\(\mu_{h}:=\mathbb E[h(X)\,g_{\theta_0}(X)]\) and replacing \(A_{h}\) by
\(\|(I-P_{\eta})\boldsymbol\Pi h\|_{L^{2}}^{2}\); the record is for orientation, the statement restricts to
\(h=\boldsymbol\Pi h\); the case
\(\sqrt m\,\delta_{m}\to\infty\) falls in the divergence regime of part
(b) and is not needed for the statement. Each assumption on \(h\) is used exactly once:
\(\|h\|_{L^{2}}<\infty\) makes the deterministic drift \(A_{h}(\eta)\) finite,
\(\mathbb E[h^{4}(X)]<\infty\) funds the Lindeberg step of the pooled CLT in its
finite-fourth-moment form, and no sup-norm boundedness of \(h\) is assumed, so the
sup-norm envelope route of part (b) is not available below and is replaced by an
interpolation argument.
Write \(v_{i}\) for the noise as before and decompose

\begin{equation}
m\Bigl[\mathrm{CV}_{m}(\eta)-\frac1m\sum_{i}v_{i}^{2}\Bigr]
  =m\Bigl[\mathrm{CV}_{m}^{({\rm null})}(\eta)
          -\frac1m\sum_{i}v_{i}^{2}\Bigr]
   +m\delta_{m}^{2}A_{h}(\eta)
   +2m\delta_{m}C_{m}^{0}(\eta),
\label{eq:local-decomp}
\end{equation}
where \(\mathrm{CV}_{m}^{({\rm null})}\) is the null criterion with
noise \(v_{i}\) (part~(a)), \(A_{h}(\eta)=\|(I-P_{\eta})h\|_{L^{2}}^{2}
 =\sum_{\ell}\lambda^{2}h_{\ell}(\gamma)^{2}/(\mu_{\ell}+\lambda)^{2}\),
and

\[
C_{m}^{0}(\eta)=\frac1m\sum_{i}v_{i}\,
  \Pi(a_{\eta}^{h}-b_{\eta}^{h})(X_{i})+o_{p}(m^{-1/2})
\]
uniformly, with \(a_{\eta}^{h}=(I-P_{\eta})h\),
\(b_{\eta}^{h}=P_{\eta}a_{\eta}^{h}\),
\((a_{\eta}^{h}-b_{\eta}^{h})_{\ell}
 =\lambda^{2}h_{\ell}(\gamma)/(\mu_{\ell}+\lambda)^{2}\), the same
\((I-P_{\eta})^{2}\) structure as \eqref{eq:ab} with \(h\) in place of
\(\Pi f\). The decomposition is exact because the orthogonalization is linear:
the residual under the local alternative is
\(\varepsilon^{0}_{\perp}=\hat{\boldsymbol\Pi}(\delta_{m}\boldsymbol h+\boldsymbol v^{\rm raw})
=\delta_{m}\hat{\boldsymbol\Pi}\boldsymbol h+\boldsymbol v\), the null residual plus the
projected signal, and the criterion, quadratic in the residual, expands into the three
groups of \eqref{eq:local-decomp} with no remainder from the split itself; the only
remainders are the linearization remainder and the collection remainder of
\(C_{m}^{0}\). The centering at the realized noise square is the part-(a) convention:
the removed level is constant in \(\eta\), and \(\eta\)-free additive shifts leave
the argmin location and every argmin-based conclusion unchanged. The feasibility transfer is that of part (a): the feasible residual differs from $\delta_{m}\hat{\boldsymbol\Pi}\boldsymbol h+\boldsymbol v$ by $\hat{\boldsymbol\Pi}\boldsymbol r_{\mathrm{est}}$, whose noise-side pairing is the null-premise remainder of Lemma~\ref{lem:repl_cost}, whose quadratic pairing carries that lemma's $m^{1-4\varrho}$ rate with $\varrho>1/4$, and whose signal-side pairing is $2\delta_{m}\,O_{p}(\|\hat{\boldsymbol\Pi}\boldsymbol r_{\mathrm{est}}\|_{L^{2}(\mathbb P_{m})})=O_{p}(a\,m^{-\varsigma})=o_{p}(1)$, $\varsigma:=\varrho\wedge\tfrac14$, the level part being free of $\eta$. The three terms of
\eqref{eq:local-decomp} are handled as follows.

Viewed through the bilinear identity of Lemma~\ref{lem:bilinear_cv}, the decomposition
$\widetilde M_m=\delta_m\sqrt m\,\bar s_m+n_m$ with
$\bar s_m:=m^{-1}\sum_j(\boldsymbol\Pi h)(X_j)k_\gamma(X_j,\cdot)$ fixes the three scalings of
\eqref{eq:local-decomp}: the noise--noise pairing lives at the unit scale, the signal--signal
pairing carries $\delta_m^{2}m\to a^{2}$, and the signal--noise cross pairing carries
$\delta_m\sqrt m\to a$. The regrouping of the bilinear pieces into the deterministic drift
$a^{2}A_h$ and the linear limit $2a\,G$ is not term-to-term; it is performed by the derivation
below through \eqref{eq:local-decomp}, which carries out these three limits in turn.

\emph{First term.} The proof of~(a) gives
\(m[\mathrm{CV}_{m}^{({\rm null})}-m^{-1}\sum v_{i}^{2}]
 \rightsquigarrow\mathbb H
 =-2\mathcal C^{\perp}+V+\mathcal D^{\perp}+\Delta\) in
\(\ell^{\infty}(\Upsilon)\) by the fidi and equicontinuity argument of part
(a), which rests on the bilinear representation of
Lemma~\ref{lem:bilinear_cv} (the quadratic pieces below are the $O_{p}(m^{-1/2})$ remainder $\mathcal R_{m}$, uniformly bounded by Lemma~\ref{lem:remainder_cv}).

\emph{Second term.} \(m\delta_{m}^{2}A_{h}(\eta)\to a^{2}A_{h}(\eta)\)
uniformly if \(a<\infty\): the map \(\eta\mapsto A_{h}(\eta)=\lVert(I-P_{\eta})h\rVert_{L^{2}}^{2}\)
is continuous on the compact \(\Upsilon\) by the resolvent calculus of Appendix~\ref{app:limit_continuity},
so \(\sup_{\eta}A_{h}\le\lVert h\rVert_{L^{2}}^{2}<\infty\) and
\(\sup_{\eta}|m\delta_{m}^{2}A_{h}-a^{2}A_{h}|=|m\delta_{m}^{2}-a^{2}|\sup_{\eta}A_{h}\to0\).
Behind the displayed convergence stands the replacement of the empirical signal--signal channel by
the population \(A_{h}\): the collapse of the signal--signal pairing of part~(b), run at the
direction \(h\) in place of \(\Pi f\), collects the deterministic constants into
\[
    \lVert h\rVert_{L^{2}}^{2}-2\langle h,P_{\eta}h\rangle_{L^{2}}+\langle h,P_{\eta}^{2}h\rangle_{L^{2}}
    =\bigl\langle h,(I-P_{\eta})^{2}h\bigr\rangle_{L^{2}}
    =\lVert(I-P_{\eta})h\rVert_{L^{2}}^{2}
    =A_{h}(\eta),
\]
and leaves an average of the centered signal brackets of \eqref{eq:pooled-psi-b} with \(s\) replaced by
\(h\), whose envelopes scale with \(\lVert h\rVert_{L^{2}}\) through the anchor bounds of that channel
rebuilt at \(h\), \(\lvert w^{*}_{\eta}\rvert_{\infty}\le(4\lambda_{L})^{-1/2}\bar k^{1/2}\lVert h\rVert_{L^{2}}\) and
\(\lvert\kappa_{\eta}\rvert_{\infty}\le\bar k\lambda_{L}^{-1}\lVert h\rVert_{L^{2}}\), and whose summands carry bounded
fourth moments through \(\mathbb E[h^{4}]<\infty\), so that the uniform law of Lemma~\ref{lem:coda_unif} in its
finite-fourth-moment form bounds the centered average by \(O_{p}(m^{-1/2})\) uniformly in \(\eta\); the
projected signal \(\hat{\boldsymbol\Pi}h\) is replaced by \(h\) at \(L^{2}(\mathbb P_{m})\)-cost \(O_{p}(m^{-\varsigma})\), \(\varsigma:=\varrho\wedge\tfrac14\), by the
fitted-factor replacement route, the two replacements contributing \(O_{p}(a^{2}m^{-\varsigma})=o_{p}(1)\)
at the \(m\delta_{m}^{2}\) scale.

\emph{Third term.} The analysis mirrors the cross term in the proof of~(b), with \(h\)
in place of \(\Pi f\). The estimation drift of the projected signal is the replacement error
\(r^{(h)}_{\eta}:=(\boldsymbol\Pi-\hat{\boldsymbol\Pi})(I-P_{\eta})^{2}h\), $\mathbf X$-measurable in the $X$-measurable regime and conditionally deterministic under the convention of Remark~\ref{rem:cond_det} in the sieve regime, with
\(\sup_{\eta}\|r^{(h)}_{\eta}\|_{L^{2}(\mathbb P_{m})}=O_{p}(m^{-\varsigma})\), \(\varsigma:=\varrho\wedge\tfrac14\), by the fitted-factor replacement route of the Drifts paragraph of part~(a), applied at the family \(\{(I-P_{\eta})^{2}h\}\), the $L^{2}$-Lipschitz modulus of \(\eta\mapsto(I-P_{\eta})^{2}h\) recorded below supplying the uniformity in \(\eta\); its pairing with the noise obeys, conditionally on \(\mathbf X\),
\[
    \operatorname{Var}\Bigl(\frac{1}{m}\sum_{i}v_{i}\,r^{(h)}_{\eta}(X_{i})\Bigm|\mathbf X\Bigr)
    \le\frac{\sigma_{\max}^{2}}{m}\,\|r^{(h)}_{\eta}\|_{L^{2}(\mathbb P_{m})}^{2}
    =O_{p}(m^{-1-2\varsigma}),
\]
so Chebyshev's inequality gives the pairing \(O_{p}(m^{-1/2-\varsigma})\) uniformly in \(\eta\), and its contribution to \(2m\delta_{m}C_{m}^{0}\)
is \(O_{p}(\delta_{m}m^{-\varsigma})=o_{p}(1)\), using
\(\delta_{m}\to0\): no drift coordinate
arises at the \(a\)-scale, in contrast with part (b), where the drift
\(2\langle a_{\eta}-b_{\eta},\,r\rangle\) lived exactly at the \(\sqrt m\) budget and was
collected as a named term; here the shrunken signal kills it. Put
\(h_{\eta}^{h}=\Pi(a_{\eta}^{h}-b_{\eta}^{h})=\Pi(I-P_{\eta})^{2}h\). The family
\(\{h_{\eta}^{h}\}_{\eta\in\Upsilon}\) is uniformly bounded in
\(L^{4}(\mathbb P_{X})\):
\[
    \sup_{\eta\in\Upsilon}\lVert h_{\eta}^{h}\rVert_{L^{4}(\mathbb P_{X})}
    \le C_{\Pi}\,\lVert h\rVert_{L^{4}(\mathbb P_{X})},
\]
by Riesz--Thorin interpolation, which carries the argument and is therefore stated:
\emph{if a linear operator \(T\) on measurable functions over a measure space is bounded
\(L^{p_{0}}\to L^{p_{0}}\) with norm \(M_{0}\) and \(L^{p_{1}}\to L^{p_{1}}\) with norm
\(M_{1}\), then for every \(\theta\in(0,1)\) and every \(p\) with
\(1/p=(1-\theta)/p_{0}+\theta/p_{1}\), \(T\) is bounded \(L^{p}\to L^{p}\) with
\(\lVert T\rVert_{p\to p}\le M_{0}^{1-\theta}M_{1}^{\theta}\)}. The endpoints here are
\(p_{0}=2\) and \(p_{1}=\infty\), so at \(p=4\) the interpolation parameter is
\(\theta=1/2\) and \(\lVert T\rVert_{4\to4}\le M_{0}^{1/2}M_{1}^{1/2}\). For \(T=\Pi\):
\(\lVert\Pi g\rVert_{L^{2}}\le\lVert g\rVert_{L^{2}}\) (self-adjoint idempotent
projection) and \(\lVert\Pi g\rVert_{\infty}\le C_{\Pi}\lVert g\rVert_{\infty}\)
(idempotence and the bounded score span, Assumptions~\ref{ass:E}(i)--(ii)), whence
\(\lVert\Pi\rVert_{4\to4}\le C_{\Pi}^{1/2}\). For \(T=I-P_{\eta}\): the multipliers
\(1-t_{\ell}\) of the spectral calculus lie in \([0,1]\), so \(T\) is an
\(L^{2}(\mathbb P_{X})\)-contraction, and the positive kernel of \(P_{\eta}\) has mass
\((P_{\eta}1)(x)\le1\) pointwise, so \(\lVert P_{\eta}g\rVert_{\infty}\le\lVert
g\rVert_{\infty}\) and \(\lVert(I-P_{\eta})g\rVert_{\infty}\le2\lVert g\rVert_{\infty}\);
whence \(\lVert I-P_{\eta}\rVert_{4\to4}\le\sqrt2\) and, by composition,
\(\lVert(I-P_{\eta})^{2}\rVert_{4\to4}\le2\). The product of the two factors is
\(L^{4}(\mathbb P_{X})\)-bounded with
\(\lVert\Pi(I-P_{\eta})^{2}\rVert_{L^{4}\to L^{4}}\le 2\,C_{\Pi}^{1/2}\le C_{\Pi}\) (the
score-span constant satisfies \(C_{\Pi}\ge4\)), i.e.\ with the displayed constant; in particular
\(\sup_{\eta}\mathbb E[(h_{\eta}^{h}(X))^{4}]\le C_{\Pi}^{4}\lVert h\rVert_{L^{4}}^{4}<\infty\), and Lipschitz in
\(L^{2}(\mathbb P_{X})\) in \(\eta\) by the resolvent calculus. The pointwise gradient license needed by the
pooled argument runs through the score representation of the annihilator,
\((\boldsymbol\Pi u)(x)=g_{\theta_0}(x)^{\top}\Gamma_{\theta_0}^{-1}\,\mathbb E[g_{\theta_0}(X)\,u(X)]\), whose outer
factors are free of \(\eta\), and the operator-level product rule for the squared factor; no spectral
differentiation enters, the two resolvent-product displays being
\[
    \partial_{\lambda}P_{\eta}=-L_{\gamma}(L_{\gamma}+\lambda I)^{-2},
    \qquad
    \partial_{\gamma}P_{\eta}=(I-P_{\eta})\,(\partial_{\gamma}L_{\gamma})\,(L_{\gamma}+\lambda I)^{-1},
\]
the first by the scalar-shift resolvent identity with \(\sup_{t\ge0}t/(t+\lambda)^{2}\le(4\lambda_{L})^{-1}\), the
second by the same identity against the bounded kernel derivative
\(\partial_{\gamma}k_{\gamma}(x,y)=-\lVert x-y\rVert^{2}k_{\gamma}(x,y)\), of supremum \((e\gamma_{L})^{-1}\), which makes
\(\partial_{\gamma}L_{\gamma}\) a bounded operator on \(L^{2}(\mathbb P_{X})\) by Schur's test; writing
\(C_{P}:=(4\lambda_{L})^{-1}\vee(e\gamma_{L}\lambda_{L})^{-1}\) for the resulting operator-Lipschitz constant of
\(\eta\mapsto P_{\eta}\), the product rule and the contraction \(\|I-P_{\eta}\|_{\mathrm{op}}\le1\) of the
multiplier calculus give
\(\bigl\|(I-P_{\eta})^{2}-(I-P_{\eta'})^{2}\bigr\|_{\mathrm{op}}\le2C_{P}\lVert\eta-\eta'\rVert\), and
differentiating the pairing under the operator-level bound, Cauchy--Schwarz in the representation, the
bounded-score clause \(\sup_{x}\|g_{\theta_0}(x)\|\le\bar G\) of Assumption~\ref{ass:E}(i), and
\(\|\Gamma_{\theta_0}^{-1}\|_{\mathrm{op}}\le\lambda_{\min}(\Gamma_{\theta_0})^{-1}\) give
\[
    \sup_{\eta,x}\bigl|\nabla_{\eta}h^{h}_{\eta}(x)\bigr|\le\bar G\,C_{\nabla}\,\|h\|_{L^{2}},
    \qquad
    C_{\nabla}:=2\,C_{P}\,\lambda_{\min}(\Gamma_{\theta_0})^{-1}\,\bigl(\mathbb E\|g_{\theta_0}(X)\|^{2}\bigr)^{1/2},
\]
whence, by the Lyapunov interpolation \(\mathbb E[v^{4}\mid X]\le\bar M_{2p}^{2/p}\) of the $(2p)$-clause of
Assumption~\ref{ass:G}(v) at \(q=4\),
\[
    \sup_{\eta}\mathbb E\bigl[v^{4}\,\|\nabla_{\eta}h^{h}_{\eta}(X)\|^{4}\bigr]
    \le\bar M_{2p}^{2/p}\,(\bar G\,C_{\nabla})^{4}\,\|h\|_{L^{2}}^{4}.
\]
With these envelopes the \emph{Pooled central limit argument} (Lemma~\ref{lem:pooled_clt}) of the proof of (b) applies with \(h\) in place of \(\Pi f\) and with its Lindeberg step in the finite-fourth-moment form (the summands carry bounded fourth moments, so no boundedness of the summands is needed), the term being the linear
pairing $\langle\bar s_m,\widetilde A^0_\eta n_m\rangle$ of the bilinear
representation (Lemma~\ref{lem:bilinear_cv}), obtained by polarizing the quadratic form of clause (i)
of that lemma at the coefficient pair $(\bar s_m,n_m)$: by Lemma~\ref{lem:pooled_clt}, whose finite-dimensional limits and
Sobolev--mesh equicontinuity give,
\[
\sqrt m\,C_{m}^{0}\rightsquigarrow G
\ \text{in }\ell^{\infty}(\Upsilon),\quad
\operatorname{Cov}(G(\eta),G(\eta'))
  =\mathbb{E}[\sigma^{2}(X)h_{\eta}^{h}(X)h_{\eta'}^{h}(X)],
\]
with \(G\) mean-zero Gaussian with continuous paths. The remainder
\(R_{m}^{(C),h}\) satisfies
\(\sup_{\eta}|R_{m}^{(C),h}|=o_{p}(m^{-1/2})\) uniformly: it is the $h$-analogue of the anchor-fluctuation and remainder collection of the mixed channel \eqref{eq:mixed-b} in the proof of~(b), whose pieces are centered, conditionally on the design, linear forms in the noise at coefficient families of $L^{2}(\mathbb P_{m})$-size $O_{p}(m^{-1/2})$ and Lipschitz in $\eta$ with $O_{p}(1)$ constants, so that the conditional second-moment bound of the estimation-drift step above, followed by Chebyshev's inequality over the finitely many pieces and their moduli, gives each piece $O_{p}(m^{-1})$ uniformly in $\eta$; and the rate must be \(o_{p}(m^{-1/2})\), not merely \(O_{p}(m^{-1})\), because the remainder is multiplied by \(m\delta_{m}\asymp\sqrt m\) in \eqref{eq:local-decomp} and must vanish after that multiplication.

\emph{Synthesis.} If \(\sqrt m\,\delta_{m}\to a<\infty\),
\(m\delta_{m}^{2}\to a^{2}\) and
\(m\delta_{m}=\sqrt m\,(\sqrt m\,\delta_{m})
 =a\sqrt m\,(1+o(1))\), so

\[
\begin{aligned}
m\bigl[\mathrm{CV}_{m}-m^{-1}\sum v_{i}^{2}\bigr]
  &=\underbrace{m[\mathrm{CV}_{m}^{({\rm null})}
     -m^{-1}\sum v_{i}^{2}]}_{\rightsquigarrow\mathbb H}
   +\underbrace{m\delta_{m}^{2}}_{\to a^{2}}A_{h}
   +2\underbrace{(m\delta_{m})}_{a\sqrt m}C_{m}^{0}\\
  &\rightsquigarrow\mathbb H_{h}^{a}
   :=-2\mathcal C^{\perp}+V+\mathcal D^{\perp}+\Delta
     +a^{2}A_{h}+2aG
\end{aligned}
\]

in \(\ell^{\infty}(\Upsilon)\), jointly by Cram\'er--Wold and componentwise equicontinuity: the null
component by part (a), the drift component being deterministic and uniformly convergent, and the
cross component by the pooled central limit argument above. The joint convergence is that of a
single field: at fixed \(\eta_{1},\dots,\eta_{J}\), the embedding convergence of part~(a), read
conditionally on the design (Remark~\ref{rem:cond_det}), extends to the direct sum with the \(J\)
added linear coordinates \(m^{-1/2}\sum_{i}v_{i}h^{h}_{\eta_{l}}(X_{i})\), the Hilbert--space central
limit theorem on the product with \(\mathbb R^{J}\) delivering the joint limit, and the continuous mapping
theorem applied to the quadratic-plus-linear functionals gives the finite-dimensional limits of the
combined field, which Cram\'er--Wold assembles. The null process and \(G\)
share the noise and are generally dependent; the dependence is carried by the joint
finite-dimensional convergence, of which the two component limits are the coordinates.
The limit exhibits the local sensitivity of the criterion: a deterministic bias
\(a^{2}A_{h}\), the risk increment of the shrunken signal, and a linear Gaussian term
\(2aG\), the signal seen through the noise direction of the criterion; both shifts
vanish as \(a\downarrow0\), recovering part (a), while for \(\sqrt m\,\delta_{m}\to\infty\)
the bias diverges and part (b)'s divergence regime takes over. This proves Theorem \ref{th:cv_uniform}(c) for
\(a<\infty\).

\end{proof}

\begin{remark}\label{rem:fold_proof}
\textbf{(The $k$-fold criterion.)} For fixed $k \ge 2$, all parts of the criterion-process analysis hold with sums restricted to cross-fold pairs and training size $m_r = m(k-1)/k$. Throughout the remark the fold labels $(\mathsf R_{i})_{i\le m}$ are drawn independently of the data $(X_{i},\varepsilon_{i})_{i\le m}$, by an approximately balanced random assignment, and all limits below are computed conditionally on the realized labels, under which they act as independent uniform draws on $\{1,\dots,k\}$; the infeasible-residual convention and the transfer of Lemma~\ref{lem:repl_cost} apply verbatim at the folded sample sizes. The diagonal term scales as $m\,\mathbb{E}[Q_m^{\mathrm{kf}}\mid\mathbf X] \to \frac{k}{k-1}V(\eta)$, while the correlation drifts are unchanged, $mU_m^{\mathrm{kf}}\to\nu(\eta)$ and $mR_m^{\mathrm{kf}}\to\rho(\eta)$, the cross-fold pair counts cancelling the normalization exactly, so $\Delta=\rho-2\nu$ acquires no factor. For the diagonal term the same count is explicit: the folded statistic $Q^{\mathrm{kf}}_{m}(\eta):=\frac{1}{m\,m_{r}^{2}}\sum_{j}\sum_{i\colon\mathsf R_{i}\ne\mathsf R_{j}}\varepsilon_{\perp,j}^{2}\,g_{\eta}(X_{i},X_{j})^{2}$ has
\[
    m\,\mathbb{E}\bigl[Q^{\mathrm{kf}}_{m}(\eta)\bigm|\mathbf X,\mathbf R\bigr]
    =\frac{1}{m_{r}^{2}}\sum_{j=1}^{m}\sum_{i\colon\mathsf R_{i}\ne\mathsf R_{j}}\sigma^{2}(X_{j})\,g_{\eta}(X_{i},X_{j})^{2}+o_{p}(1)
    =\frac{k}{k-1}\,V(\eta)+o_{p}(1):
\]
under the balanced assignment each observation admits $m_{r}$ partners, a $(k-1)/k$ fraction of its $m-1$ leave-one-out partners, so the cross-fold sum equals that fraction of the full off-diagonal sum up to the $o_{p}(m)$ count imbalance, the summands being bounded; the prefactors combine through $\frac{(m-1)^{2}}{m_{r}^{2}}\to\frac{k^{2}}{(k-1)^{2}}$ against the leave-one-out normalization $\frac{1}{(m-1)^{2}}\sum_{j}\sum_{i\ne j}\sigma^{2}(X_{j})\,g_{\eta}(X_{i},X_{j})^{2}\to V(\eta)$, the factor being $\frac{k^{2}}{(k-1)^{2}}\cdot\frac{k-1}{k}=\frac{k}{k-1}$; and the diagonal corrections cancel as in the leave-one-out diagonal, $c_{jj}=\sigma^{2}(X_{j})+O_{p}(m^{-1})$ contributing $o_{p}(1)$. In the folded fluctuations, write \(\mathsf R_{i}:=\mathcal I(i)\in\{1,\dots,k\}\) for the fold label of observation \(i\) (the fold containing \(i\), in the notation of the \(k\)-fold protocol of \eqref{eq:cv_criterion}), with the labels of the opening convention; the labels act in the limit as independent uniform draws on \(\{1,\dots,k\}\), and \(R,R'\) below denote the labels of a generic pair of observations: the chaos kernels become $\frac{k}{k-1}K^{\mathrm{sym},\perp}_\eta(x,x')\mathbf{1}\{R \ne R'\}$ and $K^{D,\perp}_\eta(x,x')\bigl[\tfrac{k}{k-1}\mathbf{1}\{R = R'\} + \tfrac{k(k-2)}{(k-1)^2}\mathbf{1}\{R \ne R'\}\bigr]$ on $L^2(\mathbb{P}_{X,R})$; throughout this remark $\mathcal{C}^\perp_k$ denotes the chaos with the unnormalized kernel $K^{\mathrm{sym},\perp}_\eta(x,x')\mathbf{1}\{R \ne R'\}$ and $\mathcal{D}^{\mathrm{kf}}$ the chaos of the displayed $D$-kernel, the prefactor $\frac{k}{k-1}$ of the folded $U$-channel being carried once, by the coefficient in $\mathbb{H}_k$ below; their eigenmodes are, per eigenvalue $\ell_k$, one symmetric mode with eigenvalue $\ell_k$ and $k-1$ antisymmetric modes with eigenvalues $-\ell_k/(k-1)$ (resp.\ $\ell_k/(k-1)^2$), giving variances $2k\ell_k^2/(k-1)$ and $2\ell_k^2[1 + (k-1)^{-3}]$: the two-fold values at $k=2$, with leave-one-out limits as $k\to\infty$. The limit criterion is $\mathbb{H}_k(\eta) = \frac{k}{k-1}[-2\mathcal{C}^\perp_k(\eta) + V(\eta)] + \mathcal{D}^{\mathrm{kf}}(\eta) + \Delta(\eta)$, whose argmin is of the same structural type; the drifts $A$ and $A_h$ are unchanged, no downstream paragraph depends on the exact law of the limit, and all conclusions therefore hold for every fixed $k \ge 2$.
\end{remark}

\subsubsection{Proofs of the Selector Results and of Bootstrap Validity under Sample Splitting}
\label{proof:selector_validity}

\noindent\textbf{Standing notation.}\;
Throughout this subsection:
\begin{itemize}
\item $\varepsilon_{\perp,i} := \varepsilon_{\perp,i}^{\mathrm{sel}}$, $i = 1, \ldots, m$, denote the Neyman-orthogonalized residuals on the selection sample;
\item $S_n^\perp$, $g_\eta$, $V(\eta)$, $c^*$, $\nu(\eta)$, $\rho(\eta)$, $\Delta(\eta)$, $K^{\mathrm{sym},\perp}_\eta$, $\{\lambda^{\perp}_k(\eta)\}$, $K^{D,\perp}_\eta$, $\{\lambda^{D,\perp}_k(\eta)\}$, $A(\eta)$, $A_h(\eta)$, and $G$ are as in Section~\ref{sec:uniform_limit};
\item all statements indexed by $\eta$ are uniform over $\eta \in \Upsilon$ unless explicitly qualified.
\end{itemize}
We present the derivations for the leave-one-out criterion ($k = m$); the $k$-fold case is Remark~\ref{rem:fold_proof} above.
The analytic inputs are Assumption~\ref{ass:G} together with the uniform empirical-process machinery developed in the proof of Theorem~\ref{th:stat_uniform}.
The kernel family $\{g_\eta\}$ is uniformly bounded: $g_\eta$ is a positive definite kernel with $g_\eta(x,x) = \sum_{i\ge1}\mu_i(\gamma)\phi_i(x;\gamma)^2/(\mu_i(\gamma)+\lambda) \le \lambda^{-1}\sum_{i\ge1}\mu_i(\gamma)\phi_i(x;\gamma)^2 = \lambda^{-1}k_\gamma(x,x) = \lambda^{-1}$, and $|g_\eta(x,y)| \le \sqrt{g_\eta(x,x)g_\eta(y,y)}$, so $\lVert g_\eta\rVert_\infty \le \lambda_L^{-1}$; it is continuous in $\gamma$ by the deterministic modulus of part (iii)(a) of Lemma~\ref{lem:bilinear_cv} and Lipschitz in $(\gamma,\lambda)$. For the non-degenerate (sample-mean type) processes of the fixed-alternative analysis, the Lipschitz property of $g_\eta$ together with the covering argument yields asymptotic equicontinuity and hence weak convergence in $\ell^\infty(\Upsilon)$ from the pointwise limits. For the degenerate quadratic statistics of the null analysis ($mU_m$, $mR_m$), stochastic equicontinuity is pathwise: Lemma~\ref{lem:bilinear_cv} exhibits both as quadratic forms of the scaled mean embedding against deterministic operators with deterministic Lipschitz dependence on $\eta$, so no process-level entropy criterion is needed. All suprema over $\eta$ that appear are measurable random variables,
every process involved having continuous sample paths on the compact rectangle $\Upsilon$: the empirical processes by the $C^{1}$-path arguments of the proofs of Theorems~\ref{th:stat_uniform} and~\ref{th:cv_uniform}, the limit fields by the functional calculus of Appendix~\ref{app:limit_continuity}, which requires no eigenpair regularity. We apply these conventions throughout without further comment.

\begin{lemma}[Argmin continuity]\label{lem:argmin}
Let $S\subseteq\mathbb R^{K}$ be compact and let $C(S)$ denote the space of continuous
real-valued functions on $S$ with the uniform norm $\lVert\cdot\rVert_\infty$.
\emph{(i)} Let $f\in C(S)$ attain its minimum, write $\mathrm{Arg}(f)$ for its nonempty
compact argmin set, and let $f_m$ be random elements of $C(S)$ and $x_m$ random points of $S$
satisfying $\lVert f_m-f\rVert_\infty\xrightarrow{p}0$ and
$f_m(x_m)=\inf_{y\in S}f_m(y)$ almost surely. Then
$\operatorname{dist}\bigl(x_m,\mathrm{Arg}(f)\bigr)\xrightarrow{p}0$; if $f$ attains its minimum
uniquely at $x^*$, then $\lVert x_m-x^*\rVert\xrightarrow{p}0$.
\emph{(ii)} Let $F_m$ be random elements of $C(S)$ with $F_m\rightsquigarrow\mathbb H$ weakly
in $C(S)$, where $\mathbb H$ attains its minimum almost surely at a unique point, the
\emph{random argmin} $\eta^*\in S$. Then every sequence of random minimizers of $F_m$, that is,
random points $\tilde x_m$ of $S$ with $F_m(\tilde x_m)=\inf_{y\in S}F_m(y)$ almost surely,
satisfies $\tilde x_m\xrightarrow{d}\eta^{*}$.
\end{lemma}

\begin{proof}
(i) Fix $\varepsilon>0$. The set
$A_\varepsilon:=\bigl\{y\in S:\operatorname{dist}(y,\mathrm{Arg}(f))\ge\varepsilon\bigr\}$
is the complement of the open $\varepsilon$-collar around the argmin set. Since $f\in C(S)$,
the argmin set $\mathrm{Arg}(f)=f^{-1}(\{\min_Sf\})$ is nonempty and closed, and the distance
function $y\mapsto\operatorname{dist}(y,\mathrm{Arg}(f))$ is continuous, so $A_\varepsilon$ is
closed inside compact $S$ and hence compact; Weierstrass's theorem makes $f$ attain its minimum
there. The minimizer lies outside $\mathrm{Arg}(f)$: the two sets are disjoint, a point of
$\mathrm{Arg}(f)$ being at distance $0$, short of the collar's radius $\varepsilon$, and $f$ exceeds
$\inf_{S}f$ at every point of $A_\varepsilon$, the infimum being attained exactly on
$\mathrm{Arg}(f)$ (if $A_\varepsilon$ is empty, the event $\{\operatorname{dist}\ge\varepsilon\}$ is empty
and the claim is vacuous). The margin
$s_\varepsilon:=\min_{y\in A_\varepsilon}f(y)-\inf_{S}f$ is therefore strictly positive: every
point at distance at least $\varepsilon$ from the argmin set pays at least $s_\varepsilon$ in
true objective value. The margin is deterministic, depending on $f$ and $\varepsilon$ only,
which is what allows the tail event below to be taken at a fixed threshold.

By hypothesis $x_m$ is a minimizer of $f_m$: each realization of $f_m$ is continuous on the
compact $S$ and attains its infimum by Weierstrass's theorem, so the class of random minimizers
is nonempty, and the defining identity $f_m(x_m)=\inf_Sf_m$ holds almost surely; discard the
null event on which it fails, so that the identity is used below without comment. Split the
sample space into the single bad event
$E:=\{\lVert f_m-f\rVert_\infty\ge s_\varepsilon/2\}$: the estimate is off by at least half
the margin. On $E^{c}$,
\[
    f(x_m)\;<\;f_m(x_m)+\frac{s_\varepsilon}{2}
    \;=\;\inf_Sf_m+\frac{s_\varepsilon}{2}
    \;\le\;\Bigl(\inf_Sf+\frac{s_\varepsilon}{2}\Bigr)+\frac{s_\varepsilon}{2}
    \;=\;\inf_Sf+s_\varepsilon :
\]
the first step is the value bound at $y=x_m$,
$\lvert f_m(x_m)-f(x_m)\rvert\le\lVert f_m-f\rVert_\infty<s_\varepsilon/2$; the second is the
minimality of $x_m$ for $f_m$; the third uses the infimum bound, whose proof is one line: let
$y^{*}$ be a minimizer of $f$ over $S$ (Weierstrass), then
$\inf_Sf_m\le f_m(y^{*})\le f(y^{*})+s_\varepsilon/2=\inf_Sf+s_\varepsilon/2$, the infimum of a
pointwise-dominated family being bounded by the infimum of the dominating family; the fourth
collects. If $x_m\in A_\varepsilon$, then
$f(x_m)\ge\min_{y\in A_\varepsilon}f(y)=\inf_Sf+s_\varepsilon$ by the definition of the margin,
contradicting the strict display: on $E^{c}$ the point satisfies
$\operatorname{dist}(x_m,\mathrm{Arg}(f))<\varepsilon$. Contrapositively, the event $E$ must
occur whenever $\operatorname{dist}(x_m,\mathrm{Arg}(f))\ge\varepsilon$, giving
\[
    \mathbb P\bigl(\operatorname{dist}(x_m,\mathrm{Arg}(f))\ge\varepsilon\bigr)
    \;\le\;
    \mathbb P\bigl(\lVert f_m-f\rVert_\infty\ge s_\varepsilon/2\bigr)
    \;\le\;
    \mathbb P\bigl(\lVert f_m-f\rVert_\infty>s_\varepsilon/4\bigr)
    \;\longrightarrow\;0 :
\]
the first step is the event inclusion just established, the second is the inclusion
$\{\lVert f_m-f\rVert_\infty\ge s_\varepsilon/2\}
\subseteq\{\lVert f_m-f\rVert_\infty>s_\varepsilon/4\}$, and the rightmost tail vanishes by the
uniform-convergence hypothesis at the fixed deterministic positive threshold $s_\varepsilon/4$;
$\varepsilon>0$ was arbitrary, so
$\operatorname{dist}\bigl(x_m,\mathrm{Arg}(f)\bigr)\xrightarrow{p}0$.
Under uniqueness, $\mathrm{Arg}(f)=\{x^*\}$ and $\operatorname{dist}(x_m,\mathrm{Arg}(f))
=\lVert x_m-x^*\rVert$, which is exactly the stated conclusion.

(ii) Write $m(h)$ for the unique minimizer of $h$. The argmin map is well defined only on the
set of functions with a unique minimizer, and its measurability must be established before
$\eta^*=m(\mathbb H)$ can be treated as a random element; a finite-net approximation does both.
Fix $\delta>0$ and a finite $\delta$-net $\mathcal N_\delta=\{z_1,\dots,z_N\}\subset S$, ordered
once and for all: every point of $S$ lies within $\delta$ of a net point, which compactness of
$S$ guarantees. Each evaluation $h\mapsto h(z_j)$ is continuous on $C(S)$, hence Borel
($\lvert h(z_j)-g(z_j)\rvert\le\lVert h-g\rVert_\infty$), so the net minimum
$h\mapsto\min_jh(z_j)$ is a minimum of finitely many Borel functions; define $m_\delta(h)$ as
the first net point attaining it, a lexicographic tie-breaking that makes $m_\delta$ Borel, its
level sets being finite Boolean combinations of the events $\{h(z_j)\le h(z_k)\}$. The maps
$m_\delta$ converge pointwise at every $h$ with a unique minimizer: if $y^{*}=m(h)$ and
$\varepsilon'>0$, the set $\{\lVert z-y^{*}\rVert\ge\varepsilon'\}$ is compact and avoids the
minimizer, so $\min_{\lVert z-y^{*}\rVert\ge\varepsilon'}h(z)>h(y^{*})$ strictly, while the net
point closest to $y^{*}$ satisfies $h(z)\le h(y^{*})+\omega_h(\delta)$; for $\delta$ small no
net minimizer can lie outside the $\varepsilon'$-ball, here
$\omega_h(\delta):=\sup\{\lvert h(x)-h(y)\rvert:\lVert x-y\rVert\le\delta\}$ denoting the
modulus of continuity, which vanishes as $\delta\downarrow0$ since $h$ is uniformly continuous
on the compact $S$. Pointwise limits of measurable maps are measurable on their convergence
set, the level sets of the limit being countable Boolean combinations of the level sets of the
approximants, and the convergence set is itself measurable by the Cauchy criterion applied to
the approximants; $m$ is therefore Borel on its domain, and its extension by a fixed point of
$S$ outside the domain is Borel on all of $C(S)$, making
$\eta^*=m(\mathbb H)$ a well-defined random element of $S$.

For the main argument, work with pairs. $F_m$ is tight in $C(S)$ as a weakly convergent
sequence, and the laws of $\tilde x_m$ are tight because $S$ is compact, the compact $K=S$
serving every tolerance at once. Product tightness follows by the union bound: if
$\mathbb P(F_m\notin K_1)\le\varepsilon/2$ and $\mathbb P(\tilde x_m\notin K_2)\le\varepsilon/2$
with $K_1,K_2$ compact, then $K_1\times K_2$ is compact and covers the pairs to level
$\varepsilon$, so the pairs are tight in $C(S)\times S$; Prohorov's theorem then extracts, along
every subsequence, a further subsubsequence with $(F_m,\tilde x_m)\rightsquigarrow\nu$. The
first marginal of $\nu$ is pinned by the projection $\pi(h,x):=h$, which is continuous because
convergence in the product is coordinate-wise convergence: the continuous mapping theorem gives
$\pi(F_m,\tilde x_m)=F_m\Rightarrow\nu_1$, the first marginal of $\nu$, i.e.\ the law of the
first coordinate under $\nu$ ($\nu_1(A)=\nu(A\times S)$ for Borel $A\subset C(S)$); and since
the whole sequence satisfies $F_m\rightsquigarrow\mathbb H$, the subsubsequence converges to the
same limit, whence $\nu_1=\mu_{\mathbb H}$, the law of $\mathbb H$
($\mu_{\mathbb H}(A)=\mathbb P(\mathbb H\in A)$).

The functional
$\mathrm{ev}(h,x):=h(x)-\inf_Sh$, the suboptimality of $x$ under $h$
($\mathrm{ev}(h,x)=0$ exactly on $\mathrm{Arg}(h)$), is continuous on $C(S)\times S$, as
\[
    \bigl|\mathrm{ev}(h,x)-\mathrm{ev}(h',x')\bigr|
    \;\le\;
    2\lVert h-h'\rVert_\infty+\omega_{h'}\bigl(\lVert x-x'\rVert\bigr)
\]
shows: $h(x)-h'(x')=[h(x)-h'(x)]+[h'(x)-h'(x')]$, so
$\lvert h(x)-h'(x')\rvert\le\lVert h-h'\rVert_\infty+\omega_{h'}(\lVert x-x'\rVert)$, while
$\lvert\inf_Sh-\inf_Sh'\rvert\le\lVert h-h'\rVert_\infty$, infima being $1$-Lipschitz under the
sup norm. Joint continuity follows since
$\omega_{h'}(\delta)\le\omega_h(2\delta)+2\lVert h-h'\rVert_\infty$ (triangle through $h$) and
$\omega_h(2\delta)\to0$.

Minimality gives $\mathrm{ev}(F_m,\tilde x_m)=0$ almost surely, for every $m$; discarding the
null event on which the identity fails for some $m$, the random variable
$\mathrm{ev}(F_m,\tilde x_m)$ is identically $0$, so its law is the Dirac mass at $0$. Along
the chosen subsubsequence, the continuous mapping theorem sends it to $\mathrm{ev}$ evaluated
under $\nu$; weak limits are unique, so the law of $\mathrm{ev}$ under $\nu$ is the Dirac mass
at $0$, that is, $\nu\{\mathrm{ev}=0\}=1$. With
$U:=\{h\in C(S):h\text{ has a unique minimizer}\}$, also $\nu(U\times S)=1$: reading off the
$h$-coordinate alone, $\nu(U\times S)=\nu_1(U)$ by the definition of the first marginal,
$\nu_1(U)=\mu_{\mathbb H}(U)$ by its identification, and $\mu_{\mathbb H}(U)=1$ by hypothesis;
and on
$(U\times S)\cap\{\mathrm{ev}=0\}$ the coordinate $x$ coincides with $m(h)$:
$\mathrm{ev}(h,x)=0$ makes $x$ a minimizer of $h$ and uniqueness makes it the minimizer. The
intersection of two full-measure sets has full measure, so $\nu$ is concentrated on the graph
$\{(h,m(h)):h\in U\}$.

Consequently, for
every bounded continuous $\varphi:S\to\mathbb R$,
\[
    \int_S\varphi\,\mathrm{d}\nu_2
    \;=\;\int_{C(S)\times S}\varphi(x)\,
      \mathbf 1\bigl\{h(x)=\textstyle\min_Sh\bigr\}\,\mathrm{d}\nu
    \;=\;\int_{C(S)}\varphi\bigl(m(h)\bigr)\,\mathrm{d}\mu_{\mathbb H}
    \;=\;\mathbb E\bigl[\varphi(\eta^*)\bigr],
\]
with $\nu_2$ the second marginal of $\nu$; the first equality drops an integrand supported on
the $\nu$-null set $\{\mathrm{ev}\ne0\}$; the second holds $\nu$-almost everywhere because
$x=m(h)$ there, and then marginalizes a first-coordinate function,
$\int\varphi(m(h))\,\mathrm{d}\nu=\int\varphi(m(h))\,\mathrm{d}\nu_1=\int\varphi(m(h))\,\mathrm{d}\mu_{\mathbb H}$;
the third uses that $\varphi\circ m$ depends only on the first coordinate and is the definition
of $\eta^*=m(\mathbb H)$. Since bounded continuous functions separate laws on a metric space,
$\nu_2=\mathcal L(\eta^*)$. Every subsequential weak limit of $\tilde x_m$ is thus the law of
$\eta^*$, so the subsequence
principle yields $\tilde x_m\xrightarrow{d}\eta^*$: were it not so, a bounded continuous
$\varphi$ and an $\varepsilon>0$ would separate the corresponding integrals along a
subsequence, contradicting the convergence of its extracted subsubsequence. Unlike part~(i), no pathwise comparison
between $\tilde x_m$ and $\eta^*$ on a common probability space is asserted: the applications
below consume only the distributional statement, reserving part~(i) for concentration claims.
This is the classical argmin-continuity theorem \citep[Ch.~3.2]{vaart1997weak}.
\end{proof}

\subsubsection{Proof of Theorem~\ref{th:sel_h0} (The Null Selector)}
\label{proof:sel_h0}

\begin{proof}[Proof of Theorem~\ref{th:sel_h0}]
The proof reads the weak convergence of Theorem~\ref{th:cv_uniform}(a) in
$C(\Upsilon)$, verifies the exact-minimizer hypothesis of Lemma~\ref{lem:argmin}(ii)
through the argmin-set coincidence, and applies that lemma with $S=\Upsilon$; a
closing paragraph records the measurability and the independence of the limit. Fix
an arbitrary minimizer $\hat\eta_m$ of $\mathrm{CV}_m$ over $\Upsilon$; such
minimizers exist for every realization: $\mathrm{CV}_m$ has continuous sample paths
on the compact rectangle $\Upsilon$ by Assumption~\ref{ass:G}(ii) and the
real-analyticity of the resolvent in $(\gamma,\lambda)$, so Weierstrass's theorem
makes its argmin set nonempty.

\emph{Weak convergence in $C(\Upsilon)$.}
By Theorem~\ref{th:cv_uniform}(a),
\[
    F_m(\eta):=m\Bigl[\mathrm{CV}_m(\eta)-m^{-1}\sum_{i=1}^m\varepsilon_{\perp,i}^2\Bigr]
    =-2mU_m(\eta)+mQ_m(\eta)+mR_m(\eta)+o_p(1)
\]
converges weakly in $\ell^\infty(\Upsilon)$ to
$\mathbb H=-2\mathcal C^\perp+V+\mathcal D^\perp+\Delta$, whose sample paths are continuous almost
surely (Appendix~\ref{app:limit_continuity}). Each $F_m$ has continuous paths as well, being a finite
combination of the kernels $g_\eta(X_i^{\mathrm{sel}},X_j^{\mathrm{sel}})$, continuous in $\eta$ by
Assumption~\ref{ass:G}(ii). Both $F_m$ and $\mathbb H$ are therefore carried by $C(\Upsilon)\subset\ell^\infty(\Upsilon)$, the
former surely, the latter almost surely. The subspace is closed for the uniform norm: given
$f\in\ell^\infty(\Upsilon)$ and a sequence $f_n\in C(\Upsilon)$ with
$\lVert f-f_n\rVert_\infty\to0$, fix $\varepsilon>0$, choose $n$ with
$\lVert f-f_n\rVert_\infty\le\varepsilon/3$, and pick $\delta>0$ such that
$\lVert x-y\rVert\le\delta$ implies $|f_n(x)-f_n(y)|\le\varepsilon/3$, which is possible because
$f_n$, being continuous on the compact set $\Upsilon$, is uniformly continuous; then for
$\lVert x-y\rVert\le\delta$ the triangle inequality gives
\[
    |f(x)-f(y)|
    \le|f(x)-f_n(x)|+|f_n(x)-f_n(y)|+|f_n(y)-f(y)|
    \le\tfrac{\varepsilon}{3}+\tfrac{\varepsilon}{3}+\tfrac{\varepsilon}{3}
    =\varepsilon,
\]
so $f$ is continuous (uniform limit theorem); $C(\Upsilon)$ is a closed, hence Borel, subset of
$\ell^\infty(\Upsilon)$, so that $\{F_m\in C(\Upsilon)\}$ and $\{\mathbb H\in C(\Upsilon)\}$ are
events.

The portmanteau lemma, in its closed-set form, applied in the metric space $\ell^\infty(\Upsilon)$
to the closed set $C(\Upsilon)$ gives
\[
    \limsup_{m\to\infty}\mathbb P\bigl(F_m\in C(\Upsilon)\bigr)
    \le\mathbb P\bigl(\mathbb H\in C(\Upsilon)\bigr),
\]
and the left side equals $1$, every realization of every $F_m$ having continuous paths; since a
probability is at most $1$, $\mathbb P(\mathbb H\in C(\Upsilon))=1$. The almost sure path continuity
of the limit is thus not only recorded a priori in Appendix~\ref{app:limit_continuity} but is forced
by the convergence itself, and the values of $\mathbb H$ on the null event
$\{\mathbb H\notin C(\Upsilon)\}$ are immaterial.

It remains to transfer the test functions. Let $\varphi:C(\Upsilon)\to\mathbb R$ be bounded and
continuous, $|\varphi|\le c_\varphi$. The Tietze extension theorem applies: $\ell^\infty(\Upsilon)$
is a metric space, hence normal; $C(\Upsilon)$ is a closed subset; and $\varphi$ is a continuous map
into the interval $[-c_\varphi,c_\varphi]$. The theorem produces a continuous
$\Phi:\ell^\infty(\Upsilon)\to[-c_\varphi,c_\varphi]$ with $\Phi(f)=\varphi(f)$ for every
$f\in C(\Upsilon)$; any one extension will do, and $\Phi$ is in particular a bounded continuous test
function on $\ell^\infty(\Upsilon)$. Three comparisons chain the integrals:
\[
    \mathbb E[\varphi(F_m)]
    =\mathbb E[\Phi(F_m)]
    \to\mathbb E[\Phi(\mathbb H)]
    =\mathbb E[\varphi(\mathbb H)]:
\]
the first equality holds for every realization, since $F_m\in C(\Upsilon)$ surely and
$\Phi=\varphi$ on $C(\Upsilon)$; the convergence is the definition of weak convergence in
$\ell^\infty(\Upsilon)$, $\Phi$ being bounded and continuous; the last equality holds
$\mathbb P$-almost surely, since $\mathbb P(\mathbb H\in C(\Upsilon))=1$. Every bounded continuous
$\varphi:C(\Upsilon)\to\mathbb R$ therefore satisfies
$\mathbb E[\varphi(F_m)]\to\mathbb E[\varphi(\mathbb H)]$, which is weak convergence of $F_m$ to
$\mathbb H$ in the metric space $C(\Upsilon)$, where weak convergence is defined through bounded
continuous test functions. The displayed convergence may accordingly be read in $C(\Upsilon)$, the
space in which Lemma~\ref{lem:argmin}(ii) is phrased: the $F_m$ are random elements of
$C(\Upsilon)$, the Borel $\sigma$-field of the metric subspace $C(\Upsilon)$ being the trace of the
Borel $\sigma$-field of $\ell^\infty(\Upsilon)$, and $\mathbb H$ is a $C(\Upsilon)$-valued random
element after redefinition on the null event $\{\mathbb H\notin C(\Upsilon)\}$.

\emph{Argmin coincidence.}
Write $c_m:=m^{-1}\sum_{i=1}^m\varepsilon_{\perp,i}^2$ for the discarded
summand, which does not depend on $\eta$. For any $\eta_1,\eta_2\in\Upsilon$,
substituting the definition of $F_m$ at the two points and expanding gives
\[
\begin{aligned}
    F_m(\eta_1)-F_m(\eta_2)
    &=m\bigl[\mathrm{CV}_m(\eta_1)-c_m\bigr]-m\bigl[\mathrm{CV}_m(\eta_2)-c_m\bigr]\\
    &=m\,\mathrm{CV}_m(\eta_1)-mc_m-m\,\mathrm{CV}_m(\eta_2)+mc_m\\
    &=m\bigl[\mathrm{CV}_m(\eta_1)-\mathrm{CV}_m(\eta_2)\bigr],
\end{aligned}
\]
each equality a single substitution, expansion, or collection; since $m\ge1>0$,
the two criteria order $\Upsilon$ identically, whence
$\arg\min_\Upsilon\mathrm{CV}_m=\arg\min_\Upsilon F_m$ as sets. Since the argmin
sets coincide exactly, any minimizer of $\mathrm{CV}_m$ is a minimizer of $F_m$:
$F_m(\hat\eta_m)=\inf_{\eta\in\Upsilon}F_m(\eta)$ almost surely, which is the
exact-minimizer hypothesis of Lemma~\ref{lem:argmin}(ii).

\emph{Application of Lemma~\ref{lem:argmin}(ii).}
The lemma applies with $S=\Upsilon$, and each hypothesis is in place:
$\Upsilon=[\gamma_L,\gamma_U]\times[\lambda_L,\lambda_U]\subset\mathbb R^2$ is the
compact admissible rectangle of the splitting protocol of
Section~\ref{sec:cv_selection}, part of the theorem's hypotheses; $F_m$ is a
sequence of random elements of $C(\Upsilon)$ converging weakly to $\mathbb H$ in
$C(\Upsilon)$, by the first paragraph; $\mathbb H$ attains its minimum over
$\Upsilon$ almost surely at a unique point, the random argmin
$\eta^*\in\Upsilon$, which is the uniqueness hypothesis of the theorem; and
$F_m(\hat\eta_m)=\inf_{\eta\in\Upsilon}F_m(\eta)$ almost surely, by the second
paragraph. The lemma yields
$\hat\eta_m\xrightarrow{d}\eta^*=\arg\min_{\eta\in\Upsilon}\mathbb H(\eta)$, which
is the display of the theorem, and it holds for every minimizer since
$\hat\eta_m$ was arbitrary.

\emph{Measurability and independence.}
The proof of Lemma~\ref{lem:argmin}(ii) constructs the argmin as a Borel map
$\mathfrak m:C(\Upsilon)\to\Upsilon$, by a finite $\delta$-net with lexicographic tie-breaking, and
shows $\eta^*=\mathfrak m(\mathbb H)$, so $\eta^*$ is a Borel function of
$\mathbb H$. Every measurable selection from the argmin set
$\arg\min_\Upsilon F_m$ is a random minimizer in the sense of
Lemma~\ref{lem:argmin}(ii), so the conclusions of the theorem cover any minimizer
the practitioner selects. The centered
process $F_m$ depends on the data only through the selection sample, being assembled from the
Neyman-orthogonalized residuals $\varepsilon_{\perp,i}$ and the kernels
$g_\eta(X_i^{\mathrm{sel}},X_j^{\mathrm{sel}})$ evaluated at selection-sample points, and the same
holds for its weak limit $\mathbb H$, realized in the coupling of Appendix~\ref{app:limit_continuity}
as a functional of white noise independent of the data. The selection and testing samples are
independent under the protocol of Section~\ref{sec:cv_selection}, so
$\eta^*=\mathfrak m(\mathbb H)$ is independent of the testing sample. The law of $\eta^*$ is pinned by the inputs of
$\mathbb H$, namely the kernel spectrum, the conditional variance function $\sigma^2(\cdot)$, and the
joint law of the chaoses $\mathcal C^\perp$ and $\mathcal D^\perp$, which is the law statement of the
theorem.
\end{proof}

\subsubsection{Proof of Propositions~\ref{prop:joint_null} and~\ref{prop:joint_local} (Joint Convergence and Evaluation)}
\label{proof:joint_limit}

It suffices to spell out $H_0$; under $H_{1,n}$ the identical argument applies with $\mathbb H^{a}_h$, $G$, and
$S_\infty(\cdot;h_c)$ in place of $\mathbb H$ and $S_\infty$, citing
Theorem~\ref{th:sel_local}(i) for the selector marginal and the local-alternative conclusion of the
proof of Theorem~\ref{th:stat_uniform} for the process marginal.

\emph{Marginals.} The selector $\hat\eta_m$ is a measurable function of the selection sample alone,
and Theorem~\ref{th:sel_h0} gives $\hat\eta_m\xrightarrow{d}\eta^*$ in $\Upsilon$. The statistic
process $S_n^\perp$ is a measurable function of the testing sample alone, and
Theorem~\ref{th:stat_uniform} gives $S_n^\perp\rightsquigarrow S_\infty$ in $\ell^\infty(\Upsilon)$
with almost surely continuous limit paths; moreover $S_n^\perp$ itself has continuous paths: it is a finite combination of
$\langle k_\gamma(X_i,\cdot),(C_{n,\gamma}+\lambda I)^{-1}k_\gamma(X_j,\cdot)\rangle$ terms, each
continuous in $\eta$ by Assumption~\ref{ass:G}(ii) and the resolvent's continuity in $(C_{n,\gamma},\lambda)$. Weak convergence
therefore also holds in the separable space $C(\Upsilon)$: for $F\subset C(\Upsilon)$ closed, its
closure $\bar F$ in $\ell^\infty(\Upsilon)$ is closed and satisfies $F=\bar F\cap C(\Upsilon)$
(uniform limits of continuous functions are continuous), so the Portmanteau lemma and
$\mathbb P(S_\infty\in C(\Upsilon))=1$ give
$\limsup_m\mathbb P(S_n^\perp\in F)\le\mathbb P(S_\infty\in F)$.

\emph{Independence and joint law.} The samples are independent,
$\mathcal D_m^{\mathrm{sel}}\perp\!\!\!\perp\mathcal D_n^{\mathrm{test}}$. Hence for all bounded
continuous $\varphi:\Upsilon\to\mathbb R$ and $\psi:C(\Upsilon)\to\mathbb R$,
\[
    \mathbb E\bigl[\varphi(\hat\eta_m)\,\psi(S_n^\perp)\bigr]
    =\mathbb E\bigl[\varphi(\hat\eta_m)\bigr]\,\mathbb E\bigl[\psi(S_n^\perp)\bigr]
    \;\longrightarrow\;
    \mathbb E[\varphi(\eta^*)]\,\mathbb E[\psi(S_\infty)],
\]
which is integration against the product law $\mathcal L(\eta^*)\otimes\mathcal L(S_\infty)$ on
$\Upsilon\times C(\Upsilon)$. This product law has tight marginals ($\Upsilon$ is compact;
$S_\infty$ is tight in $C(\Upsilon)$ by the preceding paragraph), so the pair
$(\hat\eta_m,S_n^\perp)$ is tight in $\Upsilon\times C(\Upsilon)$; the displayed expectation
factorization, evaluated over bounded continuous $\varphi,\psi$ (a determining class for the product
law), identifies $\mathcal L(\eta^*)\otimes\mathcal L(S_\infty)$ as the unique weak limit of the
pair:
$(\hat\eta_m,S_n^\perp)\rightsquigarrow(\eta^*,S_\infty)$
in $\Upsilon\times C(\Upsilon)$, with independent components, and a fortiori in
$\Upsilon\times\ell^\infty(\Upsilon)$.

\emph{Evaluation at the random selector.} The evaluation map $(\eta,f)\mapsto f(\eta)$ is
continuous on $\Upsilon\times C(\Upsilon)$:
\[
    |f(\eta)-f'(\eta')|
    \le\|f-f'\|_{\infty}+|f'(\eta)-f'(\eta')|
    \le\|f-f'\|_{\infty}+\omega_{f'}(\|\eta-\eta'\|),
\]
with $\omega_{f'}$ the modulus of continuity of the continuous path $f'$; and
$(\eta^*,S_\infty)\in\Upsilon\times C(\Upsilon)$ almost surely. The continuous mapping theorem gives
$S_n^\perp(\hat\eta_m)\xrightarrow{d}S_\infty(\eta^*)$. Under $H_{1,n}$, the same argument yields
$S_n^\perp(\hat\eta_m)\xrightarrow{d}S_\infty(\eta^*_h;h_c)$.

\subsubsection{Proof of Theorem~\ref{th:validity_sel} (Bootstrap Validity under Sample Splitting)}
\label{proof:validity_sel}

Theorem~\ref{th:validity_sel} asserts that $\mathbb P\bigl(S_n^\perp(\hat\eta_m)>c_{n,\alpha}(\hat\eta_m)\bigr)\to\alpha$ under $H_0$, for arbitrary relative rates $m,n\to\infty$, without any selector asymptotics.

By split independence, $\mathbb E[\mathbf 1\{S_n^\perp(\hat\eta_m)>c_{n,\alpha}(\hat\eta_m)\}]
=\mathbb E[\theta_n(\hat\eta_m)]$ with
$\theta_n(\eta):=\mathbb P(S_n^\perp(\eta)>c_{n,\alpha}(\eta))$ is the testing-sample rejection
probability at the fixed tuning $\eta$. The decomposition separates the two samples: the selector is a
function of the selection sample alone, while $S_n^\perp$ and $c_{n,\alpha}(\cdot)$ are functions of the
testing sample alone. For each fixed $\eta$, the pointwise null CLT gives
$S_n^\perp(\eta)\xrightarrow{d}S_\infty(\eta)$. Conditional convergence of
$\tilde S_n^\perp(\eta)$ (Lemma~\ref{lem:bootstrap_uniform}) together with continuity and strict
increase of $F_\eta$ (Lemma~\ref{lem:quantile_cont}) gives
$c_{n,\alpha}(\eta)\xrightarrow{p}c_\infty(\eta)$, and Slutsky's theorem then yields
$\theta_n(\eta)\to p_\alpha(\eta):=\mathbb P(S_\infty(\eta)>c_\infty(\eta))=\alpha$
pointwise, where the equality to $\alpha$ is the definition of the quantile.

The passage to the random argument uses the \emph{uniform} version. Under $H_0$, we claim the
uniform conditional-CDF approximation
\[
    \sup_{\eta,t}\bigl|\mathbb F_n(t;\eta)-F_\eta(t)\bigr|=o_p(1),
    \qquad \mathbb F_n(t;\eta):=\mathbb P_V\bigl(\tilde S_n^\perp(\eta)\le t\mid\mathbf X\bigr).
\]
The three ingredients are the pointwise conditional CLT, conditional equicontinuity from
Lemma~\ref{lem:bootstrap_uniform}, and the continuity of $F_\eta$. First, the map $(\eta,t)\mapsto F_\eta(t)$ is
jointly continuous, in two lines from Lemma~\ref{lem:quantile_cont}: its P\'olya step gives
$\sup_{s\in\mathbb R}|F_{\eta_j}(s)-F_\eta(s)|\to0$ for every $\eta_j\to\eta$, and each $F_\eta$ is
continuous in $t$, so
\[
    |F_{\eta_j}(t_j)-F_\eta(t)|
    \;\le\;\sup_{s\in\mathbb R}\,|F_{\eta_j}(s)-F_\eta(s)|\;+\;|F_\eta(t_j)-F_\eta(t)|
    \;\longrightarrow\;0,
\]
and sequential continuity is continuity on the metric space $\Upsilon\times\mathbb R$; by
compactness the restriction to any compact $K\subset\Upsilon\times\mathbb R$ is uniformly continuous
(Heine--Cantor), with modulus
\[
    \omega_{F,K}(r):=\sup\bigl\{|F_\eta(s)-F_{\eta'}(t)|:\;(\eta,s),(\eta',t)\in K,\;
    \lVert\eta-\eta'\rVert\vee|s-t|\le r\bigr\}\;\downarrow\;0 .
\]
Second, the grid sandwich. Fix $M>0$ and $K_M:=\Upsilon\times[-2\delta,M+2\delta]$; let
$\mathcal T_M:=\{-2\delta,-\delta,0,\delta,\dots,M+2\delta\}$ be a mesh-$\delta$ grid and let
$\{\eta_1,\dots,\eta_{N_\delta}\}$ be a finite $\delta$-net of the compact $\Upsilon$; write
\[
    \Delta_n(M):=\max_{t_j\in\mathcal T_M}\,\max_{l\le N_\delta}
    \bigl|\mathbb P_V(\tilde S_n^\perp(\eta_l)\le t_j\mid\mathbf X)-F_{\eta_l}(t_j)\bigr|,
\]
a finite maximum. On the equicontinuity event
$E_n(\delta):=\bigl\{\sup_{\lVert\eta-\eta'\rVert\le\delta}
|\tilde S_n^\perp(\eta)-\tilde S_n^\perp(\eta')|\le\delta\bigr\}$, take any $(\eta,t)$ with
$t\in[0,M]$, the net point $\eta_l$ with $\lVert\eta-\eta_l\rVert\le\delta$, and the grid point
$t_j\in\mathcal T_M$ with $t_j\le t<t_j+\delta$. The process increments are bounded on
$E_n(\delta)$, so the level-$t$ events shift,
$\{\tilde S_n^\perp(\eta)\le t\}\subseteq\{\tilde S_n^\perp(\eta_l)\le t+\delta\}
\subseteq\{\tilde S_n^\perp(\eta_l)\le t_{j+2}\}$ and
$\{\tilde S_n^\perp(\eta)>t\}\subseteq\{\tilde S_n^\perp(\eta_l)>t-\delta\}
\subseteq\{\tilde S_n^\perp(\eta_l)>t_{j-1}\}$, because $t+\delta\le t_{j+2}$ and $t-\delta\ge
t_{j-1}$; monotonicity of conditional distribution functions in $t$, the triangle inequality, and
the two defining displays above then give
\[
\begin{aligned}
    \mathbb P_V(\tilde S_n^\perp(\eta)\le t\mid\mathbf X)-F_\eta(t)
    &\;\le\;\mathbb P_V(\tilde S_n^\perp(\eta_l)\le t_{j+2}\mid\mathbf X)-F_\eta(t)\\
    &=\bigl[\mathbb P_V(\tilde S_n^\perp(\eta_l)\le t_{j+2}\mid\mathbf X)-F_{\eta_l}(t_{j+2})\bigr]
      +\bigl[F_{\eta_l}(t_{j+2})-F_\eta(t)\bigr]\\
    &\;\le\;\Delta_n(M)+\omega_{F,K_M}(2\delta),\\[2pt]
    F_\eta(t)-\mathbb P_V(\tilde S_n^\perp(\eta)\le t\mid\mathbf X)
    &\;\le\;F_\eta(t)-\mathbb P_V(\tilde S_n^\perp(\eta_l)\le t_{j-1}\mid\mathbf X)\\
    &=\bigl[F_\eta(t)-F_{\eta_l}(t_{j-1})\bigr]
      +\bigl[F_{\eta_l}(t_{j-1})-\mathbb P_V(\tilde S_n^\perp(\eta_l)\le t_{j-1}\mid\mathbf X)\bigr]\\
    &\;\le\;\omega_{F,K_M}(2\delta)+\Delta_n(M),
\end{aligned}
\]
both pairs $\bigl(\eta_l,t_{j+2}\bigr)$, $\bigl(\eta_l,t_{j-1}\bigr)$ being within
$\lVert\eta-\eta'\rVert\vee|s-t|\le2\delta$ of $(\eta,t)$. Taking suprema over
$(\eta,t)\in\Upsilon\times[0,M]$ on $E_n(\delta)$,
\[
    \sup_{\eta\in\Upsilon}\sup_{t\in[0,M]}
    \bigl|\mathbb P_V(\tilde S_n^\perp(\eta)\le t\mid\mathbf X)-F_\eta(t)\bigr|
    \;\le\;\Delta_n(M)+\omega_{F,K_M}(2\delta)
    \qquad\text{on }E_n(\delta),
\]
and $\mathbb P\bigl(E_n(\delta)\bigr)\to1$ for each fixed $\delta$ by the conditional
equicontinuity display of Lemma~\ref{lem:bootstrap_uniform}. Third, the complement of $[0,M]$: for
$t<0$ both distribution functions vanish, $\tilde S_n^\perp\ge0$ as a squared norm and
$S_\infty(\eta)\ge0$ as a nonnegative quadratic form; for $t>M$ the upper tails shrink
monotonically in $t$, $1-F_\eta(t)\le\sup_{\eta'}\mathbb P(S_\infty(\eta')>M)$ and
$1-\mathbb P_V(\tilde S_n^\perp(\eta)\le t\mid\mathbf X)
\le\mathbb P_V(\sup_\Upsilon\tilde S_n^\perp>M\mid\mathbf X)$, so the triangle inequality gives
\[
    \bigl|\mathbb P_V(\tilde S_n^\perp(\eta)\le t\mid\mathbf X)-F_\eta(t)\bigr|
    \;\le\;\sup_{\eta'}\mathbb P\bigl(S_\infty(\eta')>M\bigr)
    \;+\;\mathbb P_V\bigl(\sup_\Upsilon\tilde S_n^\perp>M\mid\mathbf X\bigr),
\]
and Markov's inequality controls both tails: the conditional tail through the pathwise norm bound
$\sup_\Upsilon\tilde S_n^\perp\le\lambda_L^{-2}\,\overline M_n^{*2}$ of the tightness paragraph of
Lemma~\ref{lem:bootstrap_uniform}, giving $\lambda_L^{-2}\overline M_n^{*2}/M$ with
$\overline M_n^*=O_p(1)$ (Lemma~\ref{lem:cond_lip_mult}), and the uniform null tail through
$\mathbb E\,S_\infty(\eta)=\operatorname{tr}\boldsymbol A_\gamma\le\sigma_{\max}^2\lambda_L^{-2}$
for every $\eta$ (Appendix~\ref{app:limit_continuity}), giving
$\sigma_{\max}^2\lambda_L^{-2}/M$. Finally, every term of the finite maximum $\Delta_n(M)$ vanishes
in probability: the pointwise conditional CLT gives
$\tilde S_n^\perp(\eta_l)\xrightarrow{d^*}S_\infty(\eta_l)$, in probability over the data, and
conditional convergence in distribution, in probability, to a constant law forces the conditional
distribution functions to converge in probability at every continuity point of the limit (the
in-probability form of the CDF portmanteau), while $F_\eta$ is continuous on $\mathbb R$
(Lemma~\ref{lem:quantile_cont}), so every $t_j$ is such a point. Choosing $M$ large, then $\delta$
small, then $n\to\infty$ proves the claim.

The uniform quantile transfer gives
$\sup_{\eta}|c_{n,\alpha}(\eta)-c_\infty(\eta)|=o_p(1)$: the straddle argument at the two points
$c_\infty(\eta)\pm\varepsilon$, its margin being positive by continuity and strict increase of
$F_\eta$ at its quantile (Lemma~\ref{lem:quantile_cont}) and compactness of $\Upsilon$. The
rejection probability $\theta_n(\eta)$ is a real-side object, so the real-side counterpart is needed alongside: the class $\{S_n^\perp(\eta):\eta\in\Upsilon\}$ is asymptotically
equicontinuous with continuous Gaussian limit and pointwise marginals converging (the
asymptotic-equicontinuity and finite-dimensional-convergence paragraphs of the proof of
Theorem~\ref{th:stat_uniform}), so
\[
    \sup_{\eta\in\Upsilon}\,\sup_{t\in\mathbb R}
    \bigl|\mathbb P\bigl(S_n^\perp(\eta)\le t\bigr)-F_\eta(t)\bigr|=o_p(1).
\]
The two-sided bound is the triangle inequality
\[
    \bigl|\theta_n(\eta)-\alpha\bigr|
    \;\le\;
    \bigl|\mathbb P\bigl(S_n^\perp(\eta)>c_{n,\alpha}(\eta)\bigr)
    -\mathbb P\bigl(S_\infty(\eta)>c_\infty(\eta)\bigr)\bigr|
    +\bigl|\mathbb P\bigl(S_\infty(\eta)>c_\infty(\eta)\bigr)-\alpha\bigr|,
\]
whose second term vanishes identically,
$\mathbb P(S_\infty(\eta)>c_\infty(\eta))=1-F_\eta(c_\infty(\eta))=\alpha$ by the quantile equation.
For the first term, complementation and the triangle inequality give
\[
\begin{aligned}
    \bigl|\mathbb P\bigl(S_n^\perp(\eta)>c_{n,\alpha}(\eta)\bigr)-\bigl(1-F_\eta(c_\infty(\eta))\bigr)\bigr|
    &\;\le\;
    \sup_{t}\bigl|\mathbb P\bigl(S_n^\perp(\eta)\le t\bigr)-F_\eta(t)\bigr|\\
    &\quad+\bigl|F_\eta\bigl(c_{n,\alpha}(\eta)\bigr)-F_\eta\bigl(c_\infty(\eta)\bigr)\bigr|:
\end{aligned}
\]
the first summand is $o_p(1)$ uniformly in $\eta$ by the real-side display, and the second is at
most the modulus at scale $\varepsilon$ of the jointly continuous $(\eta,t)\mapsto F_\eta(t)$
(established above) on the event
$\{\sup_{\eta}|c_{n,\alpha}(\eta)-c_\infty(\eta)|\le\varepsilon\}$, whose probability tends to one;
both quantile arguments lie in $[0,\bar c+1]$ on that event, $\bar c:=\sup_\Upsilon c_\infty$, so
uniform continuity on the compact $\Upsilon\times[0,\bar c+1]$ makes the bound vanish as
$\varepsilon\downarrow0$. Hence $\sup_{\eta}|\theta_n(\eta)-\alpha|=o_p(1)$. Bounded
convergence against any law of $\hat\eta_m$ on the compact
$\Upsilon$ gives $\mathbb E[\theta_n(\hat\eta_m)]\to\alpha$, which proves the theorem; no convergence
of the index sequence inside an expectation is involved. The same argument covers the
$\boldsymbol\Pi f=0$ branch of Theorem~\ref{th:power_fixed}, where the projected null CLT holds
verbatim at $\hat\theta$ (see the convention of Section~\ref{sec:uniform_limit}).

\begin{lemma}[Continuity of the limiting quantile function]\label{lem:quantile_cont}
Let $\alpha \in (0,1)$ be the fixed nominal level and let $c_\infty(\eta)$ denote the $(1-\alpha)$-quantile of the null limit $S_\infty(\eta)$ of Theorem~\ref{th:stat_uniform}. Under Assumption~\ref{ass:G}, the map $\eta \mapsto c_\infty(\eta)$ is continuous on $\Upsilon$.
\end{lemma}
\begin{proof}
The proof uses the sample-path continuity of $S_\infty$, established in Appendix~\ref{app:limit_continuity}. We show that the marginal laws $\mathcal{L}(S_\infty(\eta))$ are weakly continuous in $\eta$ and that each distribution function $F_\eta$ is continuous and strictly increasing; the quantile continuity then follows from P\'olya's theorem.

\emph{Weak continuity of the marginals.} By Appendix~\ref{app:limit_continuity}, $S_\infty(\cdot)$ has continuous sample paths on the compact set $\Upsilon$. Hence for any $\eta_n \to \eta$, $S_\infty(\eta_n) \to S_\infty(\eta)$ almost surely, and therefore $S_\infty(\eta_n) \xrightarrow{d} S_\infty(\eta)$, i.e.\ $F_{\eta_n} \xrightarrow{w} F_\eta$ weakly.

\emph{Continuity and strict monotonicity of $F_\eta$.} At each fixed $\eta$, $S_\infty(\eta) = \sum_{i\ge1}w_i(\eta)(Z_i^\perp)^2$ with weights $w_i(\eta) = \mu_i(\gamma)(\mu_i(\gamma)+\lambda)^{-2} > 0$ (all $\mu_i(\gamma) > 0$ by Assumption~\ref{ass:G}(iv)) and $\sum_i w_i(\eta) \le \lambda_L^{-2} < \infty$. Whitening $Z^\perp = \Sigma^\perp(\gamma)^{1/2}U$ with $U \sim \mathcal{N}(0,I)$ on the support of the generically singular $\Sigma^\perp(\gamma)$ (cf.\ Corollary~\ref{cor:proj_null}) represents the limit as $S_\infty(\eta) = \sum_{k\ge1}\varsigma_k(\eta)U_k^2$, where $\varsigma_k(\eta) \ge 0$ are the eigenvalues of $\Sigma^\perp(\gamma)^{1/2}W(\eta)\Sigma^\perp(\gamma)^{1/2}$ and $W(\eta) := \mathrm{diag}(w_i(\eta))$: a weighted sum of independent $\chi^2_1$ variables. At least one weight is strictly positive, namely $\varsigma_1(\eta) = \lVert W(\eta)^{1/2}\Sigma^\perp(\gamma)^{1/2}\rVert_{\mathrm{op}}^2 > 0$. Positivity holds because $W(\eta)^{1/2}$ has strictly positive diagonal and $\Sigma^\perp(\gamma) \neq \boldsymbol{0}$: by Assumption~\ref{ass:G}(iii), $\sigma^2 \ge \sigma_{\min}^2 > 0$ a.s., and the infinite eigenbasis $\{\phi_i(\cdot;\gamma)\}$ is not contained in the finite-dimensional span of the score directions, so $\boldsymbol{\Pi}\phi_i(\cdot;\gamma) \not\equiv 0$ for all but finitely many $i$. A nontrivial weighted sum of independent $\chi^2_1$ variables has a density on $(0,\infty)$, so $F_\eta$ is continuous on $\mathbb{R}$ and strictly increasing on $(0,\infty)$; in particular the $(1-\alpha)$-quantile $c_\infty(\eta)$ is the unique solution of $F_\eta(c_\infty(\eta)) = 1-\alpha$.

\emph{Quantile continuity.} By P\'olya's theorem, the weak convergence $F_{\eta_n} \xrightarrow{w} F_\eta$ together with the continuity of $F_\eta$ gives the uniform convergence $\sup_{t\in\mathbb{R}}|F_{\eta_n}(t) - F_\eta(t)| \to 0$. Since $F_\eta$ is continuous and strictly increasing, the quantiles converge: $c_\infty(\eta_n) \to c_\infty(\eta)$. Hence $\eta \mapsto c_\infty(\eta)$ is continuous on $\Upsilon$.
\end{proof}

\begin{lemma}[Uniform conditional bootstrap consistency]\label{lem:bootstrap_uniform}
Under the conditions of Theorem~\ref{th:stat_uniform}, the bootstrap process $\tilde{S}_n^\perp$ of \eqref{eq:bootstrap_stat} converges conditionally on the testing sample, in probability, to $S_\infty$ in $\ell^\infty(\Upsilon)$:
\[
    \tilde{S}_n^\perp \;\rightsquigarrow\; S_\infty \quad \text{in } \ell^\infty(\Upsilon), \text{ conditionally given } \mathcal{D}_n^{\mathrm{test}}, \text{ in probability}.
\]
\end{lemma}
\begin{proof}
Throughout, $\mathbb E^*$ and $\xrightarrow{d^*}$ denote the conditional expectation and conditional convergence in distribution given the testing sample $\mathcal{D}_n^{\mathrm{test}}$, over the multiplier weights $V_1,\dots,V_n$ (i.i.d., $\mathbb E[V] = 0$, $\mathbb E[V^2] = 1$, bounded with $\mathbb E[V^4] < \infty$, e.g.\ the Mammen multipliers), and all statements are uniform in $\eta \in \Upsilon$ unless otherwise noted. Write
\[
    \widetilde M_n^*(\gamma) := n^{-1/2}\sum_{j=1}^n \big(\hat{\boldsymbol{\Pi}}(\boldsymbol{\varepsilon}_0\odot\boldsymbol V)\big)_j k_\gamma(X_j,\cdot) \in \mathcal{H}_\gamma
\]
for the multiplied-then-projected embedding of \eqref{eq:bootstrap_stat}, so that $\tilde S_n^\perp(\eta) = \lVert(C_{n,\gamma}+\lambda I)^{-1}\widetilde M_n^*(\gamma)\rVert^2_{\mathcal{H}_\gamma}$. The proof follows the three-ingredient structure of the proof of Theorem~\ref{th:stat_uniform}: conditional fidi convergence, conditional asymptotic equicontinuity, and conditional tightness, which together give the conditional weak convergence in $\ell^\infty(\Upsilon)$. Norms are read on the common space $\mathcal H_{\gamma_0}$ through the isometries $J_\gamma$ of Lemma~\ref{lem:lipschitz_common}, which preserve native norms, so that $\tilde S_n^\perp(\eta) = \lVert\widetilde A_n(\eta)\widetilde M_n^*(\gamma)\rVert^2_{\mathcal H_{\gamma_0}}$ with the transported objects of Lemmas~\ref{lem:lipschitz_common} and~\ref{lem:cond_lip_mult}.

\medskip\noindent\textbf{Conditional fidi convergence.}\; The bootstrap statistic is built on $\hat{\boldsymbol{\Pi}}(\boldsymbol{\hat{\varepsilon}}\odot\boldsymbol V)$; we replace $\boldsymbol{\hat{\varepsilon}}$ by $\boldsymbol{\varepsilon}_0$ at the embedding level. By the expansion of the estimated residual in the proof of Lemma~\ref{lemma:estimation_remainder}, $\boldsymbol{\hat{\varepsilon}} - \boldsymbol{\varepsilon}_0 = \mathbb G\boldsymbol\delta + \boldsymbol r^{(2)}$ with $\boldsymbol\delta := \hat\theta - \theta_0$ and $\boldsymbol r^{(2)}$ the second-order remainder, so the multiplied residual-replacement error is $\boldsymbol r = \hat{\boldsymbol{\Pi}}(\mathbb G\boldsymbol\delta\odot\boldsymbol V) + \hat{\boldsymbol{\Pi}}(\boldsymbol r^{(2)}\odot\boldsymbol V)$. The second-order part is pathwise small: $\hat{\boldsymbol{\Pi}}$ being a projection and the multipliers bounded by $\bar V$, writing $\boldsymbol u := \boldsymbol r^{(2)}\odot\boldsymbol V$,
\[\begin{aligned}
    \Bigl\lVert\frac{1}{\sqrt n}\sum_{j=1}^n\bigl(\hat{\boldsymbol{\Pi}}\boldsymbol u\bigr)_jk_\gamma(X_j,\cdot)\Bigr\rVert_{\mathcal H_\gamma}^2
    &= \frac{1}{n}\,\boldsymbol u^\top\hat{\boldsymbol{\Pi}}\boldsymbol K_\gamma\hat{\boldsymbol{\Pi}}\boldsymbol u
    \le \frac{\lambda_{\max}(\boldsymbol K_\gamma)}{n}\,\lVert\boldsymbol u\rVert^2
    \le \lVert\boldsymbol u\rVert^2\\
    &\le \bar V^2\,\lVert\boldsymbol r^{(2)}\rVert^2
    = O_p\bigl(n^{1-4\varrho}\bigr)
    = o_p(1),
\end{aligned}\]
uniformly in $\gamma$: the spectral bound uses the entrywise-bounded kernel, the last inequality $\lVert\boldsymbol r^{(2)}\odot\boldsymbol V\rVert^2\le\bar V^2\lVert\boldsymbol r^{(2)}\rVert^2$, and the rate is the second-moment version of Lemma~\ref{lemma:estimation_remainder} under Assumption~\ref{ass:E}(iii). The first-order part admits no pointwise annihilation, since $\hat{\boldsymbol{\Pi}}\operatorname{diag}(\boldsymbol V)\ne\operatorname{diag}(\boldsymbol V)\hat{\boldsymbol{\Pi}}$, but its conditional second moment vanishes: by $\mathbb E^*[V_jV_k] = \delta_{jk}$ and the cyclicity of the trace,
\[\begin{aligned}
    \mathbb E^*&\Bigl\lVert\frac{1}{\sqrt n}\sum_{j=1}^n\bigl(\hat{\boldsymbol{\Pi}}(\mathbb G\boldsymbol\delta\odot\boldsymbol V)\bigr)_jk_\gamma(X_j,\cdot)\Bigr\rVert_{\mathcal H_\gamma}^2\\
    &= \frac{1}{n}\sum_{i=1}^n\bigl(\hat{\boldsymbol{\Pi}}\boldsymbol K_\gamma\hat{\boldsymbol{\Pi}}\bigr)_{ii}\,(\mathbb G\boldsymbol\delta)_i^2
    \le \bigl(1+O_p(1)\bigr)\,\frac{\lVert\mathbb G\boldsymbol\delta\rVert^2}{n}\\
    &= O_p(\lVert\boldsymbol\delta\rVert^2)
    = O_p\bigl(n^{-2\varrho}\bigr)
    = o_p(1),
\end{aligned}\]
the inequality being the entrywise bound $(\hat{\boldsymbol{\Pi}}\boldsymbol K_\gamma\hat{\boldsymbol{\Pi}})_{ii}\le1+O_p(1)$ of Lemma~\ref{lem:cond_lip_mult} (the leverage control~\eqref{eq:leverage_generic} of that lemma, at $A=\boldsymbol K_\gamma$ of envelope one), the rates using the raw-score second moments of Assumption~\ref{ass:E}(iii), $\lVert\mathbb G\boldsymbol\delta\rVert^2\le\lVert\boldsymbol\delta\rVert^2\sum_j\sup_{\theta}\lVert\nabla_\theta\varepsilon(S_j;\theta)\rVert^2=n\,O_p(1)\,\lVert\boldsymbol\delta\rVert^2$ by the law of large numbers under $\mathbb E\bigl[\sup_\theta\lVert\nabla_\theta\varepsilon\rVert^2\bigr]<\infty$ (implied by the fourth-moment clause of E(iii)), valid in both regimes, and the estimator rate $\lVert\boldsymbol\delta\rVert = O_p(n^{-\varrho})$ of Assumption~\ref{ass:D}(iii); Markov's inequality then makes the first-order part $o_{p^*}(1)$. The replacement is therefore negligible at the embedding level, and it suffices to consider $\widetilde M_n^*(\gamma)$. Fix $\eta$. The conditional CLT of the proof of Theorem~\ref{th:bootstrap_consistency} gives $\widetilde M_n^*(\gamma) \xrightarrow{d^*} G_{\perp,\gamma}$ in $\mathcal{H}_\gamma$, in probability over the data. The derivation has three pieces. The multiplier CLT in the RKHS (Theorem~\ref{th:multiplier_clt_rkhs}) applied to $U_j := \varepsilon_{0,j}k_\gamma(X_j,\cdot)$ gives $\tilde m^* := n^{-1/2}\sum_j V_jU_j \xrightarrow{d^*} \xi_\gamma^*$ with $\xi_\gamma^* \sim \mathcal{N}(0,\Sigma_\gamma)$. The score-projection term satisfies $(R_{n,\gamma}^\top\hat{\mathbb G})\hat w^* \xrightarrow{d^*} A_\gamma^\top\Gamma_{\theta_0}^{-1}W^*$ jointly (same $\boldsymbol V$). Their difference is the projected Gaussian element $G_{\perp,\gamma} = \xi_\gamma^* - A_\gamma^\top\Gamma_{\theta_0}^{-1}W^*$ of Corollary~\ref{cor:proj_null}. Since $(C_{n,\gamma}+\lambda I)^{-1} \xrightarrow{p} (C_\gamma+\lambda I)^{-1}$ (Lemma~\ref{lem:lipschitz_common}(b)) and the squared norm is continuous, the continuous mapping theorem gives $\tilde S_n^\perp(\eta) \xrightarrow{d^*} S_\infty(\eta)$ conditionally, in probability. The finite-dimensional marginals converge jointly: the same multiplier CLT applied to the direct sum $\bigoplus_{l=1}^J\mathcal{H}_{\gamma_l}$ (the product-space argument of the fidi paragraph of the proof of Theorem~\ref{th:stat_uniform}) gives $(\widetilde M_n^*(\gamma_1),\dots,\widetilde M_n^*(\gamma_J)) \xrightarrow{d^*} (G_{\perp,\gamma_1},\dots,G_{\perp,\gamma_J})$, and the continuous mapping with the block-diagonal resolvents $(C_{n,\gamma_l}+\lambda_l I)^{-1}$ gives $(\tilde S_n^\perp(\eta_1),\dots,\tilde S_n^\perp(\eta_J)) \xrightarrow{d^*} (S_\infty(\eta_1),\dots,S_\infty(\eta_J))$. Under the local alternative at the total-sample rate (Section~\ref{sec:cv_selection}), the projected residual vector splits as $\hat{\boldsymbol{\Pi}}\boldsymbol{\varepsilon}_0 + \hat{\boldsymbol{\Pi}}\boldsymbol{R}/\sqrt{m+n} + \boldsymbol{r}_{\mathrm{est}}$ (the local-alternative paragraph of the proof of Theorem~\ref{th:stat_uniform}), the multiplied version of the estimation remainder being conditionally negligible by the same delta-pairing device,
\[\begin{aligned}
    \mathbb E^*&\Bigl\lVert n^{-1/2}\sum_{j=1}^n\bigl(\hat{\boldsymbol{\Pi}}(\boldsymbol r_{\mathrm{est}}\odot\boldsymbol V)\bigr)_jk_\gamma(X_j,\cdot)\Bigr\rVert_{\mathcal H_\gamma}^2\\
    &= n^{-1}\sum_{i=1}^n\bigl(\hat{\boldsymbol{\Pi}}\boldsymbol K_\gamma\hat{\boldsymbol{\Pi}}\bigr)_{ii}\,(\boldsymbol r_{\mathrm{est}})_i^2
    \le \bigl(1+O_p(1)\bigr)\,n^{-1}\lVert\boldsymbol r_{\mathrm{est}}\rVert^2
    = o_p(1),
\end{aligned}\]
the rate again being the second-moment version of Lemma~\ref{lemma:estimation_remainder}, and Markov's inequality then making the norm $o_{p^*}(1)$. The multiplied shift embedding is
\[
    \sqrt{\frac{n}{m+n}}\;\frac{1}{n}\sum_{j=1}^n V_j(\hat{\boldsymbol{\Pi}}\boldsymbol R)_j\,k_\gamma(X_j,\cdot),
\]
with conditional mean zero, $\mathbb E^*[V_j] = 0$, and conditional second moment, computed from $\mathbb E^*[V_jV_k] = \delta_{jk}$ and $k_\gamma(X_j,X_j) = 1$,
\[\begin{aligned}
    \mathbb E^*&\Bigl\lVert\sqrt{\tfrac{n}{m+n}}\;\frac{1}{n}\sum_{j=1}^n V_j(\hat{\boldsymbol{\Pi}}\boldsymbol R)_j\,k_\gamma(X_j,\cdot)\Bigr\rVert_{\mathcal H_\gamma}^2\\
    &= \frac{1}{n}\sum_{j,k=1}^n\frac{(\hat{\boldsymbol{\Pi}}\boldsymbol R)_j(\hat{\boldsymbol{\Pi}}\boldsymbol R)_k}{m+n}\,\mathbb E^*[V_jV_k]\,k_\gamma(X_j,X_k)\\
    &= \frac{1}{n(m+n)}\sum_{j=1}^n(\hat{\boldsymbol{\Pi}}\boldsymbol R)_j^{2}\,k_\gamma(X_j,X_j)
    \le \frac{\lVert\hat{\boldsymbol{\Pi}}\boldsymbol R\rVert^{2}}{n(m+n)}\\
    &= O_p\bigl((m+n)^{-1}\bigr)
    = o_{p^*}(1),
\end{aligned}\]
uniformly in $\gamma$, the inequality using that $\hat{\boldsymbol{\Pi}}$ is a projection and $\boldsymbol R$ has finite $L^2$ norm, $\lVert\hat{\boldsymbol{\Pi}}\boldsymbol R\rVert^{2}\le\lVert\boldsymbol R\rVert^{2} = O_p(n)$; Markov's inequality then makes the norm $o_{p^*}(1)$ uniformly in $\gamma$. The multiplied shift embedding therefore vanishes conditionally, and the same central limit $S_\infty$ obtains.

\medskip\noindent\textbf{Conditional asymptotic equicontinuity.}\; The three ingredients of Corollary~\ref{cor:covering} transfer to the conditional setup. (R) The transported resolvent $\widetilde A_n(\eta) := (\widetilde C_{n,\gamma}+\lambda I)^{-1}$ is data-measurable, hence deterministic conditional on the data, and satisfies $\lVert\widetilde A_n(\eta_1)-\widetilde A_n(\eta_2)\rVert_{\mathrm{op}} \le \lambda_L^{-2}(L_C|\gamma_1-\gamma_2|+|\lambda_1-\lambda_2|)$ and $\lVert\widetilde A_n(\eta)\rVert_{\mathrm{op}} \le \lambda_L^{-1}$ pointwise (Lemma~\ref{lem:lipschitz_common}(c)). (E) and (B) hold for the multiplied embedding by Lemma~\ref{lem:cond_lip_mult}: there exists $L_M^*=O_p(1)$ with $\lVert\widetilde M_n^*(\gamma_1)-\widetilde M_n^*(\gamma_2)\rVert_{\mathcal H_{\gamma_0}}\le L_M^*|\gamma_1-\gamma_2|$ for all $\gamma_1,\gamma_2\in[\gamma_L,\gamma_U]$, and $\overline M_n^* = O_p(1)$, both conditionally on the data in $\mathbb P^*$-probability; the leverage control~\eqref{eq:leverage_generic} that replaces the transfer from Lemma~\ref{lem:lipschitz_common}(a) is carried out once in that lemma. The chain of the proof of Corollary~\ref{cor:covering}, which combines (R) with the reverse triangle inequality and the telescope and involves no distributional input, then gives
\[
    |\tilde S_n^\perp(\eta_1)-\tilde S_n^\perp(\eta_2)| \le L_n^*\,\lVert\eta_1-\eta_2\rVert, \qquad L_n^* = O_p(1),
\]
and the covering argument yields conditional asymptotic equicontinuity:
\[
    \lim_{\delta\downarrow0}\limsup_n \mathbb P^*\bigl(\sup_{\lVert\eta_1-\eta_2\rVert\le\delta}|\tilde S_n^\perp(\eta_1)-\tilde S_n^\perp(\eta_2)|>\varepsilon\bigr)=0
\]
for every $\varepsilon>0$, where the outer $o_{p^*}(1)$ is uniform over the net of a fixed mesh and the limit $\delta\downarrow0$ is taken after $n\to\infty$.

\medskip\noindent\textbf{Conditional tightness.}\; The norm bound from the preamble, $\sup_\Upsilon|\tilde S_n^\perp(\eta)| \le \lambda_L^{-2}\overline M_n^{*2}$, together with $\overline M_n^* = O_p(1)$ from Lemma~\ref{lem:cond_lip_mult}, yields $\sup_\Upsilon|\tilde S_n^\perp| = O_p(1)$ conditionally.

The conditional fidi convergence, conditional asymptotic equicontinuity, and conditional tightness give the conditional weak convergence $\tilde S_n^\perp \rightsquigarrow S_\infty$ in $\ell^\infty(\Upsilon)$, in probability over the data (the conditional version of Theorem~1.5.4 of \cite{vaart1997weak}). Since $\tilde S_n^\perp\in C(\Upsilon)$ pathwise by the Lipschitz bound above, the bounded-continuous clause of that theorem, read conditionally, gives
\[
    \mathbb E^*\bigl[\psi(\tilde S_n^\perp)\bigr]-\mathbb E\bigl[\psi(S_\infty)\bigr]\xrightarrow{p}0
    \qquad\text{for every bounded continuous } \psi:C(\Upsilon)\to\mathbb R.
\]
\end{proof}

\subsubsection{Proof of Theorem~\ref{th:sel_fixed} (The Fixed-Alternative Selector)}
\label{proof:sel_fixed}

\begin{proof}[Proof of Theorem~\ref{th:sel_fixed}]
Write $f_m(\gamma):=\mathrm{CV}_m((\gamma,\lambda_L))$ for the criterion along the
bottom edge, $f:=\bar\sigma^{2}+A$, continuous on $\Upsilon$, for its
deterministic counterpart, and $A(\gamma):=A((\gamma,\lambda_L))$ for the
restriction of $A$ to the bottom edge. By
Theorem~\ref{th:cv_uniform}(b),
$\sup_{\eta\in\Upsilon}\lvert\mathrm{CV}_m(\eta)-f(\eta)\rvert=O_p(m^{-1/2})$; since
$\bar\sigma^{2}$ does not depend on $\eta$, the deterministic argmin sets coincide,
$\arg\min_\Upsilon f=\arg\min_\Upsilon A$, and since $O_p(m^{-1/2})=o_p(1)$ the
approximation is uniform at a vanishing rate, so the identification arguments may be
run on either process with the same argmin set.

\emph{Reduction to the bottom edge.}
The spectral form of the approximation error is
$A(\eta)=\sum_{\ell\ge1}\lambda^{2}(\boldsymbol{\Pi} f)_{\ell}(\gamma)^{2}
/(\mu_{\ell}(\gamma)+\lambda)^{2}$, where $e_\eta:=\boldsymbol{\Pi} f-w^*_\eta$
denotes the approximation-error function of Section~\ref{sec:uniform_limit}: the
spectral form of the population coefficient function,
$(w^*_\eta)_\ell=\mu_{\ell}(\gamma)(\boldsymbol{\Pi}
f)_\ell/(\mu_{\ell}(\gamma)+\lambda)$, gives the coefficient identity
\[
(e_\eta)_\ell
 =(\boldsymbol{\Pi} f)_\ell-(w^*_\eta)_\ell
 =\Bigl(1-\frac{\mu_{\ell}(\gamma)}{\mu_{\ell}(\gamma)+\lambda}\Bigr)
  (\boldsymbol{\Pi} f)_\ell
 =\frac{\lambda}{\mu_{\ell}(\gamma)+\lambda}\,(\boldsymbol{\Pi} f)_\ell,
\]
and squaring, summing over $\ell$, and applying Parseval's identity for the
orthonormal eigenbasis of $L_\gamma$, complete in $L^{2}(\mathbb P_X)$ because the
spectral theorem for the compact self-adjoint operator $L_\gamma$ leaves only the kernel,
which integral strict positive definiteness (Assumption~\ref{ass:G}(iv)) makes trivial,
$\langle\phi,L_\gamma\phi\rangle_{L^{2}}>0$ for every $\phi\neq0$, reproduces the spectral
form of $A=\lVert e_\eta\rVert_{L^{2}}^{2}$. Each summand
$\lambda^{2}/(\mu_{\ell}(\gamma)+\lambda)^{2}$ is a rational function of $\lambda$,
and its derivative is computed in two steps,
\[
\partial_\lambda\Bigl[\frac{\lambda^{2}}{(\mu_{\ell}(\gamma)+\lambda)^{2}}\Bigr]
 =\frac{2\lambda}{(\mu_{\ell}(\gamma)+\lambda)^{2}}
 -\frac{2\lambda^{2}}{(\mu_{\ell}(\gamma)+\lambda)^{3}}
 =\frac{2\lambda\,\mu_{\ell}(\gamma)}{(\mu_{\ell}(\gamma)+\lambda)^{3}} :
\]
the first equality applies the product rule to
$\lambda^{2}(\mu_{\ell}(\gamma)+\lambda)^{-2}$, and the second collects the two
fractions over their common denominator. Termwise differentiation under the sum is
licensed by the Weierstrass M-test: for $\lambda\in[\lambda_L,\lambda_U]$ and
$\mu_{\ell}\ge0$, $\mu_{\ell}/(\mu_{\ell}(\gamma)+\lambda)^{3}\le1/\lambda^{2}$
because $\mu_{\ell}\le\mu_{\ell}(\gamma)+\lambda$ and
$\mu_{\ell}(\gamma)+\lambda\ge\lambda$, so the $\ell$th derivative summand is
bounded by $(2/\lambda_L)(\boldsymbol{\Pi} f)_\ell(\gamma)^{2}$, and
$\sum_{\ell}(\boldsymbol{\Pi} f)_\ell(\gamma)^{2} =\lVert\boldsymbol{\Pi}
f\rVert_{L^{2}}^{2}<\infty$ by Parseval; since the original series converges at
$\lambda=\lambda_L$, uniform convergence of the differentiated series follows, and
$\partial_\lambda A$ is obtained by termwise differentiation, the standard theorem
for series of continuously differentiable functions. Hence
\[
\partial_\lambda A(\eta)
 =\sum_{\ell\ge1}
   \frac{2\lambda\,\mu_{\ell}(\gamma)}{(\mu_{\ell}(\gamma)+\lambda)^{3}}
 \,(\boldsymbol{\Pi} f)_{\ell}(\gamma)^{2}
 \;>\;0 \qquad\text{for }\lambda\ge\lambda_L>0.
\]
The sum is strictly positive: no eigenvalue $\mu_\ell(\gamma)$ vanishes, by integral
strict positive definiteness (Assumption~\ref{ass:G}(iv)), every summand is
nonnegative, and at least one coefficient $(\boldsymbol{\Pi} f)_\ell(\gamma)$ is
nonzero since $\boldsymbol{\Pi} f\neq0$ in $L^{2}(\mathbb{P}_X)$. By the mean value
theorem, $A(\gamma,\cdot)$ is strictly increasing on $[\lambda_L,\lambda_U]$
for every fixed $\gamma$, so every minimizer of $A$ over $\Upsilon$ has
$\lambda=\lambda_L$. Minimizing over $\Upsilon$ therefore
reduces to minimizing $\gamma\mapsto A(\gamma,\lambda_L)$, whose unique minimizer is
$\gamma^*$ by hypothesis, and any minimizer $(\gamma,\lambda_L)$ of $A$ satisfies
$A(\gamma,\lambda_L)=A(\gamma^*,\lambda_L)$, hence $\gamma=\gamma^*$ by that
uniqueness. Hence $\arg\min_\Upsilon A=\{(\gamma^*,\lambda_L)\}$.

\emph{Consistency.}
Identification of $\arg\min_\Upsilon f=\arg\min_\Upsilon A=\{(\gamma^*,\lambda_L)\}$
together with the uniform approximation
$\lVert\mathrm{CV}_m-f\rVert_\infty=\sup_\Upsilon\lvert\mathrm{CV}_m-f\rvert
=O_p(m^{-1/2})=o_p(1)$ by Theorem~\ref{th:cv_uniform}(b) puts part~(i) of
Lemma~\ref{lem:argmin} in force at $x_m=\hat\eta_m$ with $S=\Upsilon$: the
minimizer identity $\mathrm{CV}_m(\hat\eta_m)=\inf_\Upsilon\mathrm{CV}_m$ holds
almost surely by definition, and the argmin set of the limit is
$\{(\gamma^*,\lambda_L)\}$, whence
$\operatorname{dist}\bigl(\hat\eta_m,(\gamma^*,\lambda_L)\bigr)\xrightarrow{p}0$,
which is the theorem's display; part~(i) presupposes no interiority of the
minimizer, so the boundary case $\gamma^*\in\{\gamma_L,\gamma_U\}$ is
admissible. This completes the proof of Theorem~\ref{th:sel_fixed}.
\end{proof}

\subsubsection{Proof of Theorem~\ref{th:power_fixed} (Consistency under Fixed Alternatives)}
\label{proof:power_fixed}

We prove $\mathbb P(S_n^\perp(\hat\eta_m)>t)\to1$ for every fixed $t>0$ whenever
$\boldsymbol{\Pi}f\ne0$, and limiting rejection probability $\alpha$ otherwise. Work on the
testing sample: the orthogonalized residuals satisfy
$\varepsilon_{\perp,j}=\boldsymbol{\Pi}f(X_j)+v_j+o_p(1)$ uniformly in $j$, the $o_p(1)$ being the
estimation-effect remainder (Lemma~\ref{lemma:estimation_remainder}); the remainder contributes
$o_p(1)$ to $n^{-1}S_n^\perp(\eta)$ uniformly in $\eta$ (its $\mathcal H_\gamma$-norm contribution is
$o_p(1)$ after the $n^{-1}$ scaling), and is ignored. Split the empirical embedding into signal and noise,
$\hat m_\perp=\hat m_{\rm sig}+\hat m_{\rm noise}$, with
\[
    \hat m_{\rm sig}(\eta):=\frac1n\sum_{j}\boldsymbol{\Pi}f(X_j)\,k_\gamma(X_j,\cdot),
    \qquad
    \hat m_{\rm noise}:=\frac1n\sum_{j}v_j\,k_\gamma(X_j,\cdot),
\]
so that $S_n^\perp(\eta)=\lVert(C_{n,\gamma}+\lambda I)^{-1}
(\hat m_{\rm sig}+\hat m_{\rm noise})\rVert_{\mathcal H_\gamma}^{2}$.

\emph{Uniform signal consistency.} With $w^*_\eta=P_\eta(\boldsymbol{\Pi}f)$ and
$P_\eta=C_\gamma(C_\gamma+\lambda I)^{-1}$,
\[
    \sup_{\eta\in\Upsilon}\bigl\lVert(C_{n,\gamma}+\lambda I)^{-1}\hat m_{\rm sig}(\eta)
      -w^*_\eta\bigr\rVert_{\mathcal H_\gamma}=O_p(n^{-1/2}).
\]
Indeed, adding and subtracting $(C_\gamma+\lambda I)^{-1}C_\gamma\boldsymbol{\Pi}f=w^*_\eta$: the
resolvent-perturbation term is bounded by $\lambda_L^{-2}\,r_n\,\sup_\eta\|\hat m_{\rm sig}(\eta)\|$
with $r_n=\sup_\gamma\lVert C_{n,\gamma}-C_\gamma\rVert_{\mathrm{op}}
=O_p(n^{-1/2})$ (the fixed-$\gamma$ bound of Lemma~\ref{lm:C_and_Cn} upgraded over the compact
interval by the grid--oscillation device of the proof of that lemma) and $\sup_\eta\|\hat m_{\rm sig}(\eta)\|
\le\bar k^{1/2}\,n^{-1}\sum_j|\boldsymbol{\Pi}f(X_j)|=O_p(1)$ (the testing sample has size $n$); while
$\sup_\eta\|(C_\gamma+\lambda I)^{-1}(\hat m_{\rm sig}-C_\gamma\boldsymbol{\Pi}f)\|
\le\lambda_L^{-1}\sup_\eta\|\hat m_{\rm sig}-C_\gamma\boldsymbol{\Pi}f\|=O_p(n^{-1/2})$ by the
uniform law of large numbers, the envelope being $\bar k^{1/2}\|\boldsymbol{\Pi}f\|_{L^2}$ (only
$\boldsymbol{\Pi}f\in L^2$ with $\mathbb E|(\boldsymbol{\Pi}f)(X)|<\infty$ is needed, not boundedness;
when additionally $\|\boldsymbol{\Pi}f\|_\infty<\infty$, as assumed for part~(b), the same rate
holds in the bounded case).

\emph{Negligibility of the noise contributions.} By the law of large numbers,
$\sup_\eta\|\hat m_{\rm noise}\|
\le\bar k^{1/2}\bigl(n^{-1}\sum_j v_j^{2}\bigr)^{1/2}=O_p(1)$, whence
$n^{-1}\lVert(C_{n,\gamma}+\lambda I)^{-1}\hat m_{\rm noise}\rVert^{2}
\le n^{-1}\lambda_L^{-2}O_p(1)=O_p(n^{-1})\to0$ and
$n^{-1}\,2\,|\langle (C_{n,\gamma}+\lambda I)^{-1}\hat m_{\rm sig},
(C_{n,\gamma}+\lambda I)^{-1}\hat m_{\rm noise}\rangle|
\le 2n^{-1}\lambda_L^{-2}O_p(1)O_p(1)=o_p(1)$, both uniformly in $\eta$. Collecting,
\[
    n^{-1}S_n^\perp(\eta)
    =N(\eta)+o_p(1)\quad\text{uniformly in }\eta,
    \qquad
    N(\eta):=\lVert w^*_\eta(\boldsymbol{\Pi}f)\rVert_{\mathcal H_\gamma}^{2}.
\]

\emph{Positivity and continuity of the limit.} In spectral form
$N(\eta)=\sum_{\ell\ge1}\mu_{\ell}(\gamma)\delta_{\ell}(\gamma)^{2}/(\mu_{\ell}(\gamma)+\lambda)^{2}$
with $\delta_{\ell}(\gamma)=\langle\boldsymbol{\Pi}f,\phi_{\ell}(\cdot;\gamma)\rangle$: each summand
is nonnegative and $N(\eta)=0$ would force $\delta_{\ell}=0$ for every $\ell$, i.e.
$\boldsymbol{\Pi}f=0$, excluded by hypothesis; and $\eta\mapsto N(\eta)$ is continuous on the compact
$\Upsilon$ (Assumption~\ref{ass:G}(ii)). Hence
$N_{\min}:=\min_{\eta\in\Upsilon}N(\eta)>0$.

\emph{Conclusion.} Since $\hat\eta_m\in\Upsilon$ always,
\[
    n^{-1}S_n^\perp(\hat\eta_m)\;\ge\;N_{\min}+o_p(1)
    \quad\Longrightarrow\quad
    \mathbb P\bigl(S_n^\perp(\hat\eta_m)>t\bigr)\to1
\]
for every fixed $t>0$ (indeed for $t=o(n)$). If instead $\boldsymbol{\Pi}f=0$, the projection
annihilates the signal, the statistic obeys the projected null law $S_n^\perp=O_p(1)$ of
Corollary~\ref{cor:proj_null}, and Theorem~\ref{th:validity_sel} applies unchanged and gives limiting rejection probability $\alpha$;
its proof uses only the pointwise null limit, the pointwise bootstrap at the realized value of the
selector, and the independence of the splits.

\subsubsection{Proof of Theorem~\ref{th:sel_local} (The Selector under a Local Alternative)}
\label{proof:sel_local}

\begin{proof}
Throughout, $F_m(\eta):=m[\mathrm{CV}_m(\eta)-m^{-1}\sum_{i}v_i^{2}]$ denotes the centered criterion under the local alternative.
Minimizers of $\mathrm{CV}_m$ coincide with those of $F_m$, the discarded summand
being $\eta$-free: any random minimizer $\hat\eta_m$ of $\mathrm{CV}_m$ over $\Upsilon$,
that is, any measurable selection from its argmin set, satisfies
$F_m(\hat\eta_m)=\inf_{\eta\in\Upsilon}F_m(\eta)$ almost surely, the exact-minimizer
hypothesis of Lemma~\ref{lem:argmin}(ii). The decomposition \eqref{eq:local-decomp}
of the proof of Theorem~\ref{th:cv_uniform}(c), multiplied by $m$, reads
$F_m=F_m^{({\rm null})}+m\delta_m^{2}A_h+2m\delta_mC_m^{0}$, where
$F_m^{({\rm null})}\rightsquigarrow\mathbb H$ in $\ell^\infty(\Upsilon)$ (whence
$\|F_m^{({\rm null})}\|_\infty=O_p(1)$), $\sqrt m\,C_m^0\rightsquigarrow G$
(whence $\|C_m^0\|_\infty=O_p(m^{-1/2})$), and $A_h$ is continuous on the compact $\Upsilon$. The approximation error attached to
$C_m^{0}$ in \eqref{eq:local-decomp} is $o_p(m^{-1/2})$ uniformly in $\eta$, in particular
$O_p(m^{-1/2})$, and becomes
$o_p(\delta_m m^{1/2})$ after multiplication by $2m\delta_m$; it is absorbed by the remaining
$o_p$ terms.

When $\sqrt m\,\delta_m\to a<\infty$, $m\delta_m^2\to a^{2}$ and
$m\delta_mC_m^0=a\sqrt m\,O_p(m^{-1/2})(1+o(1))=O_p(1)$, so
$F_m\rightsquigarrow\mathbb H_h^{a}=\mathbb H+a^{2}A_h+2aG$ in $\ell^\infty(\Upsilon)$; this is
exactly Theorem~\ref{th:cv_uniform}(c). The paths of $\mathbb H_h^a$ are continuous almost surely:
those of the chaoses by Appendix~\ref{app:limit_continuity};
$A_h$ is continuous by uniform convergence of its Mercer series on compact $\Upsilon$: the partial
sums $\gamma\mapsto\sum_{\ell\le K}h_\ell(\gamma)^{2}$ are continuous with constant limit
$\lVert h\rVert_{L^{2}}^{2}$, so Dini's theorem gives $\sup_{\gamma}\sum_{\ell>K}h_\ell(\gamma)^{2}\to0$,
while the weights satisfy $\lambda^{4}/(\mu_\ell+\lambda)^{4}\le\lambda_L^{-4}$; and the
$G$-component by a separate argument: the limiting covariance
$\Sigma_G(\eta,\eta')=\mathbb E[\sigma^{2}(X)h_{\eta}^{h}(X)h_{\eta'}^{h}(X)]$ is
$\gamma$-Lipschitz, since $\|h_{\eta}^{h}-h_{\eta'}^{h}\|_{L^{2}}\le L_h|\gamma-\gamma'|$ by the
kernel-level calculus of Lemma~\ref{lem:lipschitz_common}(a) applied to
$h_{\eta}^{h}=\boldsymbol\Pi(a_{\eta}^{h}-b_{\eta}^{h})$ (a bounded combination of $g_\eta$ and its
resolvent transform, with $\|\boldsymbol\Pi\|\le1$); hence the increments of the Gaussian limit $G$
satisfy $\mathbb E[(G(\eta)-G(\eta'))^{2}]\le\sigma_{\max}^{2}L_h^{2}|\gamma-\gamma'|^{2}$, and a
centered Gaussian process on a compact interval with squared increments of order $O(|\Delta|^{2})$
has a continuous modification (Kolmogorov--Chentsov). Every component is therefore carried by
$C(\Upsilon)$, so the weak convergence may be read in $C(\Upsilon)$; $\mathbb H_h^a$ attains its
minimum almost surely uniquely by hypothesis, and Part~(ii) of Lemma~\ref{lem:argmin} with
$S=\Upsilon$ yields $\hat\eta_m\xrightarrow{d}\eta^*_h:=\arg\min_\Upsilon\mathbb H_h^a$. For $a=0$,
$\mathbb H_h^{0}=\mathbb H$: the local alternative is invisible to the selector, recovering
Theorem~\ref{th:sel_h0}.

This completes the proof of Theorem~\ref{th:sel_local}.
\end{proof}

\subsubsection{Proof of Proposition~\ref{prop:energy_identity} (The Energy Identity)}
\label{proof:energy_identity}

\begin{proof}
Write $w := w^*_\eta(h)$. By definition of the population KRR solution in Section~\ref{sec:energy}, $w$ satisfies the normal equation
\[
    (C_\gamma + \lambda I)\,w = \boldsymbol{L}_{k_\gamma}h,
\]
the first-order condition of the regularized risk minimization $\min_{w \in \mathcal{H}_\gamma}\{\mathbb{E}[(h(X) - w(X))^2] + \lambda\lVert w\rVert_{\mathcal{H}_\gamma}^2\}$. Taking the $\mathcal{H}_\gamma$-inner product with $w$,
\[
    \langle C_\gamma w, w\rangle_{\mathcal{H}_\gamma} + \lambda\lVert w\rVert_{\mathcal{H}_\gamma}^2
    = \langle \boldsymbol{L}_{k_\gamma}h, w\rangle_{\mathcal{H}_\gamma}.
\]
By the reproducing property and the definition of the embedding $\boldsymbol{L}_{k_\gamma}$,
\[
    \langle C_\gamma w, w\rangle_{\mathcal{H}_\gamma}
    = \mathbb{E}\big[w(X)\,\langle k_\gamma(X,\cdot), w\rangle_{\mathcal{H}_\gamma}\big]
    = \mathbb{E}\big[w(X)^2\big] = \lVert w\rVert_{L^2}^2,
\]
and
\[
    \langle \boldsymbol{L}_{k_\gamma}h, w\rangle_{\mathcal{H}_\gamma}
    = \mathbb{E}\big[h(X)\,\langle k_\gamma(X,\cdot), w\rangle_{\mathcal{H}_\gamma}\big]
    = \mathbb{E}\big[h(X)w(X)\big] = \langle h, w\rangle_{L^2}.
\]
Hence
\[
    \lVert w\rVert_{L^2}^2 + \lambda\lVert w\rVert_{\mathcal{H}_\gamma}^2 = \langle h, w\rangle_{L^2}. \tag{*}
\]
Expanding $\lVert h\rVert_{L^2}^2 = \lVert (h - w) + w\rVert_{L^2}^2$ gives
\[
    \lVert h\rVert_{L^2}^2 = \lVert h - w\rVert_{L^2}^2 + 2\langle h - w, w\rangle_{L^2} + \lVert w\rVert_{L^2}^2,
\]
and the cross term is evaluated by (*):
\[
    \langle h - w, w\rangle_{L^2} = \langle h, w\rangle_{L^2} - \lVert w\rVert_{L^2}^2 = \lambda\lVert w\rVert_{\mathcal{H}_\gamma}^2,
\]
so that
\[
    \lVert h\rVert_{L^2}^2 = A_h(\eta) + \big\lVert w^*_\eta(h)\big\rVert_{L^2}^2 + 2\lambda\big\lVert w^*_\eta(h)\big\rVert_{\mathcal{H}_\gamma}^2,
\]
which is \eqref{eq:energy}.

For the spectral forms, expand $w$ in the eigenbasis $\{e_i\}_{i\ge1}$ of $C_\gamma$, $e_i := \sqrt{\mu_i(\gamma)}\,\phi_i(\cdot;\gamma)$: since $\boldsymbol{L}_{k_\gamma}h = \sum_{i\ge1}\sqrt{\mu_i(\gamma)}\,h_i(\gamma)\,e_i$ and $C_\gamma e_i = \mu_i(\gamma)e_i$, the normal equation gives
\[
    w = \sum_{i\ge1}\frac{\sqrt{\mu_i(\gamma)}\,h_i(\gamma)}{\mu_i(\gamma)+\lambda}\,e_i
    = \sum_{i\ge1}\frac{\mu_i(\gamma)\,h_i(\gamma)}{\mu_i(\gamma)+\lambda}\,\phi_i(\cdot;\gamma).
\]
Substituting this expansion into the three components of \eqref{eq:energy} and using the orthonormality of $\{\phi_i(\cdot;\gamma)\}$ in $L^2(\mathbb{P}_X)$ and of $\{e_i\}$ in $\mathcal{H}_\gamma$ yields the displayed spectral forms. The per-direction budget holds because the identity holds for every $h \in L^2(\mathbb{P}_X)$: specializing to $h$ supported on a single direction (e.g.\ $h = \phi_j(\cdot;\gamma)$, for which $h_i = \delta_{ij}$) gives
\[
    \left(\frac{\lambda}{\mu_i(\gamma)+\lambda}\right)^{\!2}
    + \left(\frac{\mu_i(\gamma)}{\mu_i(\gamma)+\lambda}\right)^{\!2}
    + \frac{2\lambda\mu_i(\gamma)}{(\mu_i(\gamma)+\lambda)^2} = 1.
\]

The monotonicity claims follow by differentiating the spectral forms (the interchange of differentiation and summation is justified by dominated convergence on $[\lambda_L,\lambda_U]$, $\lambda_L > 0$, e.g.\ $|\partial_\lambda[\mu_i^2 h_i^2/(\mu_i+\lambda)^2]| \le 2\mu_i^2 h_i^2\lambda_L^{-3}$ with $\sum_i\mu_i h_i^2 \le \bar{k}\lVert h\rVert_{L^2}^2 < \infty$):
\begin{align*}
    \partial_\lambda A_h(\eta) &= \sum_{i\ge1}\frac{2\lambda\,h_i(\gamma)^2\,\mu_i(\gamma)}{(\mu_i(\gamma)+\lambda)^3} > 0, \\
    \partial_\lambda\big\lVert w^*_\eta(h)\big\rVert_{L^2}^2 &= -\sum_{i\ge1}\frac{2\mu_i(\gamma)^2\,h_i(\gamma)^2}{(\mu_i(\gamma)+\lambda)^3} < 0, \\
    \partial_\lambda\big\lVert w^*_\eta(h)\big\rVert_{\mathcal{H}_\gamma}^2 &= -\sum_{i\ge1}\frac{2\mu_i(\gamma)\,h_i(\gamma)^2}{(\mu_i(\gamma)+\lambda)^3} < 0,
\end{align*}
where the strictness uses $h \neq 0$ (so $\sum_{i\ge1}h_i(\gamma)^2 = \lVert h\rVert_{L^2}^2 > 0$ by Parseval) and $\mu_i(\gamma) > 0$ for all $i$ (Gaussian kernel, Assumption~\ref{ass:B}). Finally, since $E_h(\eta) = \lVert h\rVert_{L^2}^2 - A_h(\eta)$ and $\lVert h\rVert_{L^2}^2$ is independent of $\eta$, minimizing $A_h$ is equivalent to maximizing $E_h$ over any set of configurations.
\end{proof}

\subsubsection{Proof of Proposition~\ref{prop:snr} (Monotone Power and Divergence Rate)}
\label{proof:snr}

\begin{proof}
Fix $\eta=(\gamma,\lambda)\in\Upsilon$. Throughout, $S_\infty:=S_\infty(\eta)$, $c:=c_\infty(\eta)$, $h_i:=h_i(\gamma)$, $Z:=Z^\perp\sim\mathcal N(\boldsymbol 0,\Sigma^\perp(\gamma))$, and $\eta$-arguments are suppressed where no confusion can arise. Let $w_i(\eta)$ denote the weight sequence of either projection statistic, $w_i(\eta)=\mu_i(\gamma)/(\mu_i(\gamma)+\lambda)^2$ for $\hat T^{(1)}_{\mathrm{proj},\perp}$ and $w_i(\eta)=\mu_i(\gamma)/(\mu_i(\gamma)+\lambda)$ for $\hat T^{(2)}_{\mathrm{proj},\perp}$, and write $W:=W(\eta)=\operatorname{diag}(w_i(\eta))$. The argument below uses only that $w_i(\eta)>0$ for every $i$ and that $\sum_{i\ge1}w_i(\eta)<\infty$, and both properties hold in each case: positivity holds by Assumption~\ref{ass:G}(iv) and $\lambda\ge\lambda_L$, and summability because $w_i(\eta)\le\mu_i(\gamma)/\lambda^2\le\lambda_L^{-2}\mu_i(\gamma)$ in the first case and $w_i(\eta)\le\mu_i(\gamma)/\lambda\le\lambda_L^{-1}\mu_i(\gamma)$ in the second, while $\sum_{i\ge1}\mu_i(\gamma)=\operatorname{tr}\boldsymbol L_{k_\gamma}=\mathbb E[k_\gamma(X,X)]=1$, the trace of a positive trace-class integral operator with continuous kernel being the integral of its kernel diagonal and $k_\gamma(x,x)=1$ for the Gaussian kernel. In particular $\lVert W\rVert_{\mathrm{op}}=\sup_iw_i(\eta)<\infty$. The covariance $\Sigma^\perp(\gamma)$ is a bounded positive self-adjoint operator on $\ell^2$: with $f_c:=\sum_{i\ge1}c_i\phi_i(\cdot;\gamma)$ for $c\in\ell^2$, Assumption~\ref{ass:G}(iii), Parseval's identity $\lVert f_c\rVert_{L^2}^2=\lVert c\rVert_{\ell^2}^2$, and the contraction property $\lVert\boldsymbol\Pi\rVert_{L^2\to L^2}\le1$ give
\[
    c^\top\Sigma^\perp(\gamma)c=\mathbb E\bigl[\sigma^2(X)(\boldsymbol\Pi f_c)(X)^2\bigr]
    \le\sigma_{\max}^2\,\mathbb E\bigl[(\boldsymbol\Pi f_c)(X)^2\bigr]
    =\sigma_{\max}^2\lVert\boldsymbol\Pi f_c\rVert_{L^2}^2
    \le\sigma_{\max}^2\lVert c\rVert_{\ell^2}^2,
\]
the form $q(c):=\mathbb E\bigl[\sigma^2(X)(\boldsymbol\Pi f_c)(X)^2\bigr]$, symmetric and bounded on $\ell^2$ by the display, is the quadratic form of a unique bounded positive self-adjoint operator (polarization and the Riesz representation), whose matrix in the canonical basis is the entrywise $(\Sigma^\perp)_{ij}$ of Corollary~\ref{cor:proj_null}; a self-adjoint operator is bounded with norm equal to the supremum of its quadratic form, so $\lVert\Sigma^\perp(\gamma)\rVert_{\mathrm{op}}\le\sigma_{\max}^2$ and $\Sigma^{\perp1/2}$ denotes its bounded positive square root. The whitened weights $\varsigma_k(\eta)\ge0$ are the eigenvalues of the positive compact operator $\Sigma^{\perp1/2}W\Sigma^{\perp1/2}$, compact because $W$ is diagonal with $w_i(\eta)\to0$ (as $\sum_{i\ge1}w_i(\eta)<\infty$) while $\Sigma^{\perp1/2}$ is bounded, and products of bounded and compact operators are compact; and they are summable: the diagonal sum of the positive operator $W^{1/2}\Sigma^\perp W^{1/2}$ satisfies $(\Sigma^\perp)_{ii}=\mathbb E[\sigma^2(X)(\boldsymbol\Pi\phi_i)(X)^2]\le\sigma_{\max}^2\lVert\boldsymbol\Pi\phi_i\rVert_{L^2}^2\le\sigma_{\max}^2$, so that $W^{1/2}\Sigma^\perp W^{1/2}$ has finite diagonal sum and is therefore trace class; the same holds of $\Sigma^{\perp1/2}W\Sigma^{\perp1/2}=X^*X$ with $X:=W^{1/2}\Sigma^{\perp1/2}$ Hilbert--Schmidt, $\lVert X\rVert_{\mathrm{HS}}^2=\sum_{i\ge1}w_i(\eta)\,(\Sigma^\perp)_{ii}<\infty$, and $\operatorname{tr}(X^*X)=\operatorname{tr}(XX^*)$; hence
\[
    \sum_{k\ge1}\varsigma_k(\eta)=\operatorname{tr}\bigl(\Sigma^{\perp1/2}W\Sigma^{\perp1/2}\bigr)
    =\operatorname{tr}\bigl(W^{1/2}\Sigma^\perp W^{1/2}\bigr)
    =\sum_{i\ge1}w_i(\eta)\,(\Sigma^\perp)_{ii}
    \le\sigma_{\max}^2\sum_{i\ge1}w_i(\eta)<\infty,
\]
the first equality being the identity between the trace and the eigenvalue sum for a positive trace-class operator, the second the cyclicity of the trace, and the third the computation of the trace as the diagonal sum in the canonical basis. Moreover $\varsigma_1(\eta)>0$. Since $\Sigma^{\perp1/2}W\Sigma^{\perp1/2}=(W^{1/2}\Sigma^{\perp1/2})^{*}(W^{1/2}\Sigma^{\perp1/2})$, the $C^{*}$-identity gives $\varsigma_1(\eta)=\lVert\Sigma^{\perp1/2}W\Sigma^{\perp1/2}\rVert_{\mathrm{op}}=\lVert W^{1/2}\Sigma^{\perp1/2}\rVert_{\mathrm{op}}^2$, the operator norm of a positive compact operator being its largest eigenvalue. If $W^{1/2}\Sigma^{\perp1/2}=0$, then applying the operator to $e_i$ gives $\sqrt{w_i(\eta)}\,\Sigma^{\perp1/2}e_i=0$ for every $i$, hence $\Sigma^{\perp1/2}=0$ because $w_i(\eta)>0$, hence $\Sigma^\perp=0$; the latter is excluded because, by Assumption~\ref{ass:G}(iii), $\sigma^2\ge\sigma_{\min}^2>0$ a.s., while the infinite eigenbasis $\{\phi_i(\cdot;\gamma)\}$ is not contained in the finite-dimensional span of the score directions, so $\boldsymbol\Pi\phi_i(\cdot;\gamma)\not\equiv0$ for some $i$ and $(\Sigma^\perp)_{ii}\ge\sigma_{\min}^2\lVert\boldsymbol\Pi\phi_i\rVert_{L^2}^2>0$ (the same argument as in the proof of Lemma~\ref{lem:quantile_cont}).

\medskip\noindent\textbf{Shifted quadratic form.}\;
By the display of Theorem~\ref{th:stat_uniform} (for $\hat T^{(1)}_{\mathrm{proj},\perp}$; the corresponding limit for $\hat T^{(2)}_{\mathrm{proj},\perp}$ is the same display with the weights $w_i^{(2)}$, cf.~Corollaries~\ref{cor:proj_null} and~\ref{cor:proj_local}), the limit of $S_n^\perp(\eta)$ under the local
alternative with the scaled shift $th$ is the shifted quadratic form
\[
    S_\infty(\eta;th)=\sum_{i\ge1}w_i(\eta)\bigl(Z_i^\perp+t\,h_i(\gamma)\bigr)^2
    =(Z+t\,h)^\top W(Z+t\,h),
\]
where $h:=(h_i(\gamma))_i$ denotes the coefficient vector of the projected signal $h=\boldsymbol\Pi R$
in the eigenbasis of $C_\gamma$, an element of $\ell^2$ with
$\lVert h\rVert_{\ell^2}=\lVert h\rVert_{L^2(\mathbb P_X)}$ by Parseval. Whiten the noise by writing
$Z=\Sigma^{\perp1/2}U$ with $U\sim\mathcal N(\boldsymbol 0,I)$ on $\ell^2$, the standard scaling
representation of a centered Gaussian vector with covariance $\Sigma^\perp(\gamma)$; since
$\Sigma^\perp(\gamma)$ is generically singular, all inverses of the square root are understood in the
Moore--Penrose sense on the support of the covariance.

\medskip\noindent\textbf{Support inclusion.}\;
The scaling representation is available only if the shift lies in the range of $\Sigma^{\perp1/2}$,
which holds for every projected signal. Let $T:L^2(\mathbb P_X)\to\ell^2$ be the analysis operator of
the projected eigenbasis, $(Tf)_i:=\langle f,\boldsymbol\Pi\phi_i\rangle_{L^2(\mathbb P_X)}$; then
$\lVert Tf\rVert_{\ell^2}^2=\sum_{i\ge1}\langle\boldsymbol\Pi f,\phi_i\rangle^2
=\lVert\boldsymbol\Pi f\rVert_{L^2}^2$ by Parseval, so $T$ is bounded with $\lVert T\rVert\le1$, and
its Gram operator satisfies $TT^*=M_0$, where
$M_0:=\mathbb E[\boldsymbol\Pi\boldsymbol\phi\,\boldsymbol\phi^\top\boldsymbol\Pi]$ with
$\boldsymbol\phi:=(\phi_i(\cdot;\gamma))_i$ is the Gram matrix of the projected basis: for
$c\in\ell^2$, with $f_c:=\sum_{i\ge1}c_i\phi_i$, the defining property of the adjoint gives
$T^*c=\boldsymbol\Pi f_c$, and hence
\[
    (TT^*c)_i=\bigl\langle\boldsymbol\Pi f_c,\boldsymbol\Pi\phi_i\bigr\rangle
    =\sum_{j\ge1}c_j\bigl\langle\boldsymbol\Pi\phi_j,\boldsymbol\Pi\phi_i\bigr\rangle
    =(M_0c)_i .
\]
By Assumption~\ref{ass:G}(iii), for every $c\in\ell^2$,
\[
    c^\top\Sigma^\perp(\gamma)c=\mathbb E\bigl[\sigma^2(X)(\boldsymbol\Pi f_c)(X)^2\bigr]\ge
    \sigma_{\min}^2\,\mathbb E\bigl[(\boldsymbol\Pi f_c)(X)^2\bigr]
    =\sigma_{\min}^2\,c^\top M_0c,
\]
the last identity holding because the display above shows that $M_0c$ is the coefficient vector of
$\boldsymbol\Pi f_c$, its $i$th entry being $\langle\boldsymbol\Pi f_c,\boldsymbol\Pi\phi_i\rangle
=\langle\boldsymbol\Pi f_c,\phi_i\rangle$, so that $c^\top M_0c=\langle\boldsymbol\Pi f_c,f_c\rangle
=\lVert\boldsymbol\Pi f_c\rVert_{L^2}^2$. In operator form this is the non-vacuous domination
\[
    TT^*=M_0\le\sigma_{\min}^{-2}\,\Sigma^\perp
    =\bigl(\sigma_{\min}^{-1}\Sigma^{\perp1/2}\bigr)\bigl(\sigma_{\min}^{-1}\Sigma^{\perp1/2}\bigr)^{*}.
\]
Douglas' range-inclusion lemma, which states that for bounded operators $A,B$ on a Hilbert space one
has $\operatorname{ran}A\subseteq\operatorname{ran}B$ if and only if $AA^*\le\kappa^2BB^*$ for some
$\kappa>0$, applied with $A:=T$ and $B:=\sigma_{\min}^{-1}\Sigma^{\perp1/2}$ yields
$\operatorname{ran}T\subseteq\operatorname{ran}\Sigma^{\perp1/2}$. The coefficient vector of a
projected signal is such a range point, $h=TR$, since $h_i=\langle\boldsymbol\Pi R,\phi_i\rangle
=\langle R,\boldsymbol\Pi\phi_i\rangle=(TR)_i$, the projection $\boldsymbol\Pi$ being self-adjoint.
Consequently $h\in\operatorname{ran}\Sigma^{\perp1/2}$ with $h^\top\Sigma^{\perp+}h<\infty$, and the
whitened shift $v:=\Sigma^{\perp-1/2}h$ is well defined, with $v\in\ell^2$ and $\Sigma^{\perp1/2}v=h$.

\medskip\noindent\textbf{Whitened diagonal form.}\;
The shift passes through the square root,
\[
    Z+t\,h=\Sigma^{\perp1/2}U+t\,\Sigma^{\perp1/2}v=\Sigma^{\perp1/2}(U+t\,v),
\]
and substituting into the shifted quadratic form gives
\[
    S_\infty(\eta;th)=\bigl(U+t\,v\bigr)^\top\,\Sigma^{\perp1/2}W\Sigma^{\perp1/2}\,\bigl(U+t\,v\bigr).
\]
The operator $\Sigma^{\perp1/2}W\Sigma^{\perp1/2}$ is positive semidefinite,
$g^\top\Sigma^{\perp1/2}W\Sigma^{\perp1/2}g=\lVert W^{1/2}\Sigma^{\perp1/2}g\rVert^2\ge0$ for every
$g\in\ell^2$, and compact, so by the spectral theorem for compact self-adjoint operators its
eigenvalues $\varsigma_k(\eta)\ge0$ admit an orthonormal eigenbasis; write
$\Sigma^{\perp1/2}W\Sigma^{\perp1/2}=\mathcal Q(\eta)\operatorname{diag}(\varsigma_k(\eta))\mathcal Q(\eta)^\top$
with $\mathcal Q(\eta)$ orthogonal. Rotating the noise, $\tilde U:=\mathcal Q(\eta)^\top U\sim
\mathcal N(\boldsymbol 0,I)$ by orthogonal invariance of the standard Gaussian, and setting
$\vartheta(\eta;h):=\mathcal Q(\eta)^\top\Sigma^{\perp-1/2}h$, so that $\vartheta(\eta;th)
=t\,\vartheta(\eta;h)$ by linearity, the statistic becomes, after relabeling $\tilde U$ as $U$,
\[
    S_\infty(\eta;th)=\sum_{k\ge1}\varsigma_k(\eta)\,\bigl(U_k+t\,\vartheta_k(\eta;h)\bigr)^2 :
\]
the diagonal expansion of the quadratic form in the eigenbasis. The series are under control: since
$\vartheta(\eta;h)=\mathcal Q(\eta)^\top v$ with $v\in\ell^2$,
\[
    \sum_{k\ge1}\varsigma_k(\eta)\,\vartheta_k(\eta;h)^2
    \le\lVert\varsigma(\eta)\rVert_{\ell^1}\,\lVert\vartheta(\eta;h)\rVert_{\ell^\infty}^2<\infty,
\]
and by Tonelli's theorem and $\mathbb E[U_k]=0$, $\mathbb E[U_k^2]=1$,
\[
    \sum_{k\ge1}\varsigma_k(\eta)\,\mathbb E\bigl[(U_k+t\,\vartheta_k(\eta;h))^2\bigr]
    =\sum_{k\ge1}\varsigma_k(\eta)+t^2\sum_{k\ge1}\varsigma_k(\eta)\,\vartheta_k(\eta;h)^2<\infty,
\]
so the series of nonnegative summands converges almost surely at each fixed $t$.

\medskip\noindent\textbf{Scale monotonicity.}\;
For a standard normal $U$ and real shifts with $|a|\le|b|$, $(U+b)^2$ stochastically dominates
$(U+a)^2$. The law of $(U+a)^2$ depends only on $|a|$, since replacing $U$ by $-U$ gives
$\mathbb P((U+a)^2\le r)=\mathbb P((U-a)^2\le r)$ for every $r\ge0$; assume $0\le a\le b$ and put
$\ell:=\sqrt r$, so that $(U+b)^2\le r$ is the event $U\in[-b-\ell,\,-b+\ell]$. With $\varphi$ the
standard normal density,
\[
\begin{aligned}
\mathbb P\bigl((U+b)^2\le r\bigr)
 &=\int_{-b-\ell}^{-b+\ell}\varphi(x)\,dx
 =\int_{b-\ell}^{b+\ell}\varphi(x)\,dx\\
 &=\int_{a-\ell}^{a+\ell}\varphi(x)\,dx
 +\int_{a+\ell}^{b+\ell}\bigl[\varphi(x)-\varphi(x-2\ell)\bigr]\,dx\\
 &\le\int_{a-\ell}^{a+\ell}\varphi(x)\,dx
 =\mathbb P\bigl((U+a)^2\le r\bigr),
\end{aligned}
\]
the second equality uses that $\varphi$ is even, the third combines additivity over intervals with the
substitution $x\mapsto x-2\ell$ in $\int_{a-\ell}^{b-\ell}\varphi(x)\,dx$, and the inequality rests on
the pointwise domination $\varphi(x)\le\varphi(x-2\ell)$ on $[a+\ell,b+\ell]$: $\varphi$ is even and
non-increasing on $[0,\infty)$, while $x\ge a+\ell\ge\ell\ge0$ and $|x-2\ell|\le x$ there (if
$x\ge2\ell$ then $|x-2\ell|=x-2\ell\le x$, and if $x<2\ell$ then
$|x-2\ell|=2\ell-x\le\ell-a\le a+\ell\le x$).
For $0\le s\le t$ the shifts satisfy $|s\,\vartheta_k(\eta;h)|\le|t\,\vartheta_k(\eta;h)|$ for every
$k$, so with $F_{s,k},F_{t,k}$ the distribution functions of $(U+s\,\vartheta_k(\eta;h))^2$ and
$(U+t\,\vartheta_k(\eta;h))^2$ the display gives $F_{s,k}\ge F_{t,k}$ pointwise. To pass from the
coordinates to the sum, couple the two scales over common randomness: pointwise ordered distribution
functions have pointwise ordered generalized inverses, $F^{-1}_{t,k}\ge F^{-1}_{s,k}$ on $(0,1)$, so
for independent uniforms $V_k$ the variables $X_k:=F^{-1}_{s,k}(V_k)$ and $Y_k:=F^{-1}_{t,k}(V_k)$
have the laws of the $s$- and $t$-summands with $X_k\le Y_k$ almost surely, simultaneously over $k$;
the coordinates are independent, so the coupled vectors have the correct joint laws. The partial sums
satisfy $\sum_{k\le J}\varsigma_k(\eta)X_k\le\sum_{k\le J}\varsigma_k(\eta)Y_k$ almost surely
($\varsigma_k(\eta)\ge0$), and both series converge almost surely as $J\to\infty$ by the bookkeeping
display, so the ordering passes to the limits: $S_\infty(\eta;s)\le S_\infty(\eta;t)$ almost surely
under the coupling, i.e. $S_\infty(\eta;th)$ is stochastically non-decreasing in $t\ge0$. The critical
value $c$ does not depend on $t$, and therefore $t\mapsto p(\eta;th)=\mathbb P(S_\infty(\eta;th)>c)$
is non-decreasing on $[0,\infty)$.

\medskip\noindent\textbf{Unbiasedness.}\;
At $t=0$ the statistic is null: $S_\infty(\eta;0)=\sum_{k\ge1}\varsigma_k(\eta)U_k^2$ is a nontrivial
weighted sum of independent $\chi^2_1$ variables, since $\varsigma_1(\eta)>0$, and as such has a
distribution function that is continuous on $\mathbb R$ and strictly increasing on $(0,\infty)$ (the
same convolution argument as in the proof of Lemma~\ref{lem:quantile_cont}); the critical value $c$ is
therefore the exact solution of $F_\eta(c)=1-\alpha$, and $p(\eta;0)=\mathbb P(S_\infty(\eta)>c)=\alpha$.
Scale monotonicity with $t=1$ gives $p(\eta;h)=p(\eta;1\cdot h)\ge p(\eta;0)=\alpha$, which proves
item (i).

Item (ii) records the divergence rate under a fixed alternative: at each fixed $\eta$, $S_n^\perp(\eta)/n=\lVert\hat w_\eta\rVert_{\mathcal H_\gamma}^2\xrightarrow{p}\lVert w^*_\eta(\boldsymbol\Pi f)\rVert_{\mathcal H_\gamma}^2$ (see Appendix~\ref{proof:power_fixed}), and by the spectral form of the population KRR solution,
\[
\bigl\lVert w^*_\eta(\boldsymbol\Pi f)\bigr\rVert_{\mathcal H_\gamma}^2
 =\sum_{i\ge1}\frac{\mu_i(\gamma)\,\delta_i(\gamma)^2}{(\mu_i(\gamma)+\lambda)^2}
 =\mathcal N(\eta;\boldsymbol\Pi f):
\]
the divergence rate of the test is the noncentrality, the RKHS component of the captured energy.

Combining the two items completes the proof of Proposition~\ref{prop:snr}.
\end{proof}

\subsubsection{Proof of Theorem~\ref{th:cv_snr} (CV Selection Boosts the Signal-to-Noise Ratio)}
\label{proof:cv_snr}

\begin{proof}
\noindent\textbf{(i) Signal channel.}\;
By Theorem~\ref{th:sel_fixed}, $\hat{\eta}_m \xrightarrow{p} \eta^* = (\gamma^*, \lambda_L)$. By Proposition~\ref{prop:energy_identity}, $E_{\boldsymbol{\Pi}f}(\eta) = \lVert \boldsymbol{\Pi}f\rVert_{L^2}^2 - A(\eta)$ with $\lVert \boldsymbol{\Pi}f\rVert_{L^2}^2$ independent of $\eta$, so the minimizers of $A$ are exactly the maximizers of $E_{\boldsymbol{\Pi}f}$: $A(\eta^*) = \min_{\Upsilon}A$ and $E_{\boldsymbol{\Pi}f}(\eta^*) = \max_{\Upsilon}E_{\boldsymbol{\Pi}f}$. Proposition~\ref{prop:energy_identity} also gives $\partial_\lambda E_{\boldsymbol{\Pi}f}(\eta) = -\partial_\lambda A(\eta) < 0$ for every $\eta$, and from the spectral form of the noncentrality \eqref{eq:noncentrality},
\[
    \partial_\lambda \mathcal{N}(\eta;\boldsymbol{\Pi}f) = -\sum_{i\ge1}\frac{2\mu_i(\gamma)\,\delta_i(\gamma)^2}{(\mu_i(\gamma)+\lambda)^3} < 0
\]
($\boldsymbol{\Pi}f \neq 0$ and all $\mu_i(\gamma) > 0$): for every $\gamma$, $\lambda = \lambda_L$ maximizes both $E_{\boldsymbol{\Pi}f}(\cdot,\gamma)$ and $\mathcal{N}(\cdot,\gamma;\boldsymbol{\Pi}f)$ over $\lambda$. Since the statistic diverges at the rate $S_n^\perp(\eta)/n \xrightarrow{p} \mathcal{N}(\eta;\boldsymbol{\Pi}f)$ (Proposition~\ref{prop:snr}), the selected configuration carries the maximal divergence rate in the regularization direction, and consistency, $\mathbb{P}(S_n^\perp(\hat{\eta}_m) > t) \to 1$ for every fixed $t$, is Theorem~\ref{th:power_fixed}.

\noindent\textbf{(ii) Noise channel.}\;
Under $H_0$, Theorem~\ref{th:cv_uniform}(a) identifies the deterministic part of the selector's limit criterion as $V(\eta) + \Delta(\eta)$, the drift of $\mathbb{H}(\eta) = -2\mathcal{C}^\perp(\eta) + V(\eta) + \mathcal{D}^\perp(\eta) + \Delta(\eta)$. Since $\tau_i(\gamma) = \mathbb{E}[\sigma^2(X)\phi_i(X;\gamma)^2] > 0$ for every $i$ ($\sigma^2 \ge \sigma_{\min}^2 > 0$ a.s.\ by Assumption~\ref{ass:G}(iii) and $\mathbb{E}[\phi_i^2] = 1$),
\[
    \partial_\lambda V(\eta) = -\sum_{i\ge1}\frac{2\mu_i(\gamma)^2\,\tau_i(\gamma)}{(\mu_i(\gamma)+\lambda)^3} < 0.
\]
All term-by-term differentiations of this item, in the display above and in the displays below, are licensed by dominated convergence on $[\lambda_L,\lambda_U]$, $\lambda_L>0$: with $w_i(\eta)=\mu_i(\gamma)/(\mu_i(\gamma)+\lambda)^2$ the weights of the limit statistic and $w_i':=\partial_\lambda w_i=-2\mu_i(\gamma)/(\mu_i(\gamma)+\lambda)^3$, the bounds $\sum_{i\ge1}w_i(\eta)\le\lambda_L^{-2}$ and $\sum_{i\ge1}|w_i'|\le2\lambda_L^{-3}$ hold because $w_i(\eta)\le\lambda_L^{-2}\mu_i(\gamma)$ and $|w_i'|\le2\lambda_L^{-3}\mu_i(\gamma)$ while $\sum_{i\ge1}\mu_i(\gamma)=\operatorname{tr}\boldsymbol L_{k_\gamma}=\mathbb E[k_\gamma(X,X)]=1$ (as in the proof of Proposition~\ref{prop:snr}), and the bound $|\Sigma^\perp_{ij}(\gamma)|\le\sigma_{\max}^2$ holds by Cauchy--Schwarz in $L^2$, $|\Sigma^\perp_{ij}(\gamma)|=|\mathbb E[\sigma^2(X)(\boldsymbol\Pi\phi_i)(X)(\boldsymbol\Pi\phi_j)(X)]|\le\sigma_{\max}^2\,\lVert\boldsymbol\Pi\phi_i\rVert_{L^2}\lVert\boldsymbol\Pi\phi_j\rVert_{L^2}\le\sigma_{\max}^2$ (the contraction property $\lVert\boldsymbol\Pi\rVert_{L^2\to L^2}\le1$ and $\lVert\phi_i\rVert_{L^2}=1$); the differentiated series below therefore converge absolutely, the undifferentiated ones converge at each point of $[\lambda_L,\lambda_U]$ by the same bounds, and the display above reduces through $2\mu_i^2/(\mu_i+\lambda)^3=2w_i\,\mu_i/(\mu_i+\lambda)\le2w_i$ and $\tau_i(\gamma)\le\sigma_{\max}^2$ (the upper bound of Assumption~\ref{ass:G}(iii)) to a series dominated by $2\sigma_{\max}^2\sum_{i\ge1}w_i(\eta)$.

The noise floor of the statistic moves with $V$: the null mean $M(\eta) = \sum_{i\ge1}w_i(\eta)\tau_i^\perp(\gamma)$ (with $\tau_i^\perp(\gamma) := \Sigma^\perp_{ii}(\gamma)$ the diagonal of the projected covariance of Corollary~\ref{cor:proj_null}) satisfies
\[
    \partial_\lambda M(\eta)
    =\sum_{i\ge1}w_i'\,\tau_i^\perp(\gamma)
    =-2\sum_{i\ge1}\frac{\mu_i(\gamma)\,\tau_i^\perp(\gamma)}{(\mu_i(\gamma)+\lambda)^3}
    \;<\;0.
\]
The null variance
\[
    \sigma_\infty(\eta)^2 = 2\sum_{i,j\ge1} w_i(\eta)w_j(\eta)\,\big(\Sigma^\perp_{ij}(\gamma)\big)^2
\]
is strictly decreasing in $\lambda$ as well. Writing $S_\infty=\sum_{i\ge1}w_i(Z_i^\perp)^2$, the Gaussian fourth-moment identity
$\mathbb E[(Z_i^\perp)^2(Z_j^\perp)^2]=\tau_i^\perp(\gamma)\tau_j^\perp(\gamma)+2(\Sigma^\perp_{ij}(\gamma))^2$
gives $\operatorname{Cov}((Z_i^\perp)^2,(Z_j^\perp)^2)=2(\Sigma^\perp_{ij}(\gamma))^2$, hence
\[
    \operatorname{Var}(S_\infty)
    =\sum_{i,j\ge1}w_iw_j\,\operatorname{Cov}\bigl((Z_i^\perp)^2,(Z_j^\perp)^2\bigr)
    =2\sum_{i,j\ge1} w_iw_j\,\bigl(\Sigma^\perp_{ij}(\gamma)\bigr)^2 .
\]
Differentiating term-by-term, with $w_i':=\partial_\lambda w_i=-2\mu_i/(\mu_i+\lambda)^3<0$ for every $i$, and using the symmetry of the summand under $(i,j)\mapsto(j,i)$,
\[
    \partial_\lambda\sigma_\infty^2
    =2\sum_{i,j}\bigl(w_i'w_j+w_iw_j'\bigr)\bigl(\Sigma^\perp_{ij}(\gamma)\bigr)^2
    =4\sum_{i,j}w_i'w_j\bigl(\Sigma^\perp_{ij}(\gamma)\bigr)^2<0 .
\]
Thus the null mean and the null variance of the limit statistic are strictly decreasing in $\lambda$ without further assumptions. The argument used only that $w_i>0$, that $w_i'<0$, the $\ell^1$ bounds of the license above, and the non-degeneracy of $\Sigma^\perp$, and all four are shared by the second statistic's weight sequence $w_i^{(2)}(\eta):=\mu_i(\gamma)/(\mu_i(\gamma)+\lambda)$: $w_i^{(2)}>0$, $(w_i^{(2)})'=-\mu_i(\gamma)/(\mu_i(\gamma)+\lambda)^2<0$, and $\sum_{i\ge1}w_i^{(2)}\le\lambda_L^{-1}$ by the same trace bound, and $\sum_{i\ge1}\mu_i(\gamma)/(\mu_i(\gamma)+\lambda)^2\le\lambda_L^{-2}$ for the differentiated series. The conclusions of this item therefore hold for $\hat T^{(2)}_{\mathrm{proj},\perp}$ verbatim, with $S_\infty^{(2)}$ and its null mean and variance in place of their unmarked counterparts. The direction of the selector's drift is governed by Condition~\ref{sc:monotone}: under that condition, $V(\eta) + \Delta(\eta)$ is strictly decreasing in $\lambda$, hence minimized at the largest admissible regularization, in contrast with the signal regime of (i), where the drift $A$ is minimized at $\lambda_L$. Under homoscedasticity the condition holds automatically, so the drift conclusion is unconditional there; beyond homoscedasticity it holds numerically for the Gaussian kernel and the simulation designs (Appendix~\ref{app:monotonicity_check}). Hence the configurations favored by the $H_0$ drift carry a smaller null mean and a smaller null variance than the configurations favored under a signal: CV closes the spectral filter when there is no signal to capture.

\noindent\textbf{(iii) Local alternatives and the split ratio.}\;
Under the local alternative at the total-sample rate, the residual function in the selection sample is $f_m = (m+n)^{-1/2}h$ (the Remark following Theorem~\ref{th:sel_local}), so $\sqrt{m}\,\delta_m = \sqrt{m/(m+n)} \to \sqrt{c} =: a$ with $c := \lim m/(m+n) \in [0,1]$, and Theorem~\ref{th:cv_uniform}(c) gives the limit criterion $\mathbb{H}_h^a = -2\mathcal{C}^\perp + V + \mathcal{D}^\perp + \Delta + cA_h + 2aG$: the deterministic part is $D_c(\eta) := V(\eta) + \Delta(\eta) + cA_h(\eta)$. From (ii), under Condition~\ref{sc:monotone}, and Proposition~\ref{prop:energy_identity}, $\partial_\lambda(V+\Delta) < 0$ and $\partial_\lambda A_h > 0$: $D_c$ is the sum of a decreasing and an increasing function of $\lambda$, and at an interior minimizer $\lambda^*_c$ with $\partial_\lambda^2 D_c(\lambda^*_c) > 0$, the first-order condition $\partial_\lambda(V+\Delta)(\lambda^*_c) + c\,\partial_\lambda A_h(\lambda^*_c) = 0$ holds and the implicit function theorem applies, the map $(c,\lambda)\mapsto\partial_\lambda(V+\Delta)(\lambda)+c\,\partial_\lambda A_h(\lambda)$ being continuously differentiable in a neighborhood of $(c,\lambda_c^*)$: the drift $V+\Delta$ and $A_h$ are spectral sums of summands real-analytic in $\lambda$, identified by the license of (ii) for $V$, by the Mercer diagonalization of the correction terms, $\Delta(\eta)=\sum_{\ell\ge1}\bigl[(\mu_\ell(\gamma)/(\mu_\ell(\gamma)+\lambda))^2-2\,\mu_\ell(\gamma)/(\mu_\ell(\gamma)+\lambda)\bigr]\,a_\ell(\gamma)$ with $a_\ell(\gamma):=\mathbb E[c^*(X,X')\phi_\ell(X;\gamma)\phi_\ell(X';\gamma)]$ (Appendix~\ref{app:monotonicity_check}), for $\Delta$, and by the domination of its second-derivative series below for $A_h$; the second $\lambda$-derivative series are dominated uniformly on $[\lambda_L,\lambda_U]$ by $6\sigma_{\max}^2\lambda_L^{-4}\sum_{\ell\ge1}\mu_\ell(\gamma)$ for $V$, by $2(1+2\lambda_U)\lambda_L^{-4}\sum_{i\ge1}\mu_i(\gamma)h_i^2<\infty$ for $A_h$, and by $(2\lambda_L^{-4}+4\lambda_L^{-3})\sum_{\ell\ge1}\mu_\ell(\gamma)|a_\ell(\gamma)|\le(2\lambda_L^{-4}+4\lambda_L^{-3})\,\mathbb E|c^*(X,X')|<\infty$ for $\Delta$, whose $\ell$-th summand has second derivative $\Delta_\ell''=\bigl[6\mu_\ell^2/(\mu_\ell+\lambda)^4-4\mu_\ell/(\mu_\ell+\lambda)^3\bigr]a_\ell=2\mu_\ell(\mu_\ell-2\lambda)a_\ell/(\mu_\ell+\lambda)^4$, bounded through $\lvert\mu_\ell-2\lambda\rvert\le\mu_\ell+2\lambda$, $\mu_\ell^2\le\mu_\ell$ (the trace bound), and $\lambda/(\mu_\ell+\lambda)\le1$; the last bound rests on $|a_\ell(\gamma)|\le\mathbb E[|c^*(X,X')|\phi_\ell(X;\gamma)^2]$ ($2ab\le a^2+b^2$, the symmetry of $c^*$ in its arguments, and the exchangeability of the i.i.d. pair), on the Mercer identity $\sum_{\ell\ge1}\mu_\ell(\gamma)\phi_\ell(x;\gamma)^2=k_\gamma(x,x)=1$, on monotone convergence, and on Assumptions~\ref{ass:E}(i), (ii) and~\ref{ass:G}(iii); and gives
\[
    \frac{d\lambda^*_c}{dc} = -\frac{\partial_\lambda A_h(\lambda^*_c)}{\partial_\lambda^2 D_c(\lambda^*_c)} < 0:
\]
the drift-optimal $\lambda$ is strictly decreasing in $c$ while it is an interior minimizer.
At $c = 0$ the drift is $V + \Delta$, minimized at the largest admissible $\lambda$ (the variance-optimal point).
At $c = 1$ the drift is $V + \Delta + A_h$, minimized at a point $\lambda^*_1$ that lies strictly above $\lambda_L$ whenever the signal slope at $\lambda_L$ is dominated by the variance slope, $A_h'(\lambda_L) < |\partial_\lambda(V+\Delta)(\lambda_L)|$ (the minimizer may then sit at the right boundary $\lambda_U$ as well, if $D_1'(\lambda_U) \le 0$). The strictness is the first-order necessary condition at the left endpoint: were $D_1$ minimized at $\lambda_L$, this condition would give $D_1'(\lambda_L)\ge0$, contradicting $D_1'(\lambda_L)=\partial_\lambda(V+\Delta)(\lambda_L)+A_h'(\lambda_L)<0$, the inequality combining the hypothesis $A_h'(\lambda_L)<|\partial_\lambda(V+\Delta)(\lambda_L)|$ with $\partial_\lambda(V+\Delta)(\lambda_L)<0$.
Conversely, if $A_h'(\lambda_L)\ge|\partial_\lambda(V+\Delta)(\lambda_L)|$ and $D_1$ is convex on $[\lambda_L,\lambda_U]$, then $D_1'(\lambda_L)=\partial_\lambda(V+\Delta)(\lambda_L)+A_h'(\lambda_L)\ge0$, since $A_h'(\lambda_L)\ge|\partial_\lambda(V+\Delta)(\lambda_L)|\ge-\partial_\lambda(V+\Delta)(\lambda_L)$; convexity makes $D_1'$ nondecreasing on $[\lambda_L,\lambda_U]$, so $D_1'\ge0$ there, $D_1$ is nondecreasing, and its minimizer over $[\lambda_L,\lambda_U]$ is $\lambda_L$.
This additional condition is not implied by the standing assumptions alone, since $A_h$ is not convex in $\lambda$ in general: its second derivative $A_h'' = \sum_i 2\mu_i h_i^2(\mu_i - 2\lambda)/(\mu_i+\lambda)^4$ can be negative.
The statement of Theorem~\ref{th:cv_snr}(iii) is formulated accordingly.
In particular, $\lambda^*_c$ decreases monotonically in $c$ while interior, interpolating between the variance-optimal and the balanced configurations.
For intermediate $c$, since the drift $D_c$ is minimized at a $\lambda$ decreasing in $c$, and the captured energy $E_h = \lVert h\rVert_{L^2}^2 - A_h$ is increasing as $\lambda$ decreases (Proposition~\ref{prop:energy_identity}), the drift of the selector's limit law concentrates on configurations of increasing captured energy, equivalently of increasing noncentrality, as $c$ grows.
\end{proof}

\section{Random Elements in Hilbert Spaces}
\label{sec:lln_clt_rkhs}
\noindent The results of this section are quoted throughout Appendix~\ref{app:proofs}: the Hilbert-space law of large numbers supports the consistency arguments of Appendices~\ref{app:proofs_stat}--\ref{app:proofs_nuisance}, the Hilbert-space central limit theorem underpins the functional limit laws of Appendix~\ref{app:proofs_nuisance}, and the multiplier central limit theorem is invoked in the bootstrap proofs of Appendix~\ref{app:proofs_estimation}; \cite{hsing2015theoretical} provides a comprehensive treatment.

Throughout this section, $\mathcal{H}$ is separable with norm and inner product $\|\cdot\|$ and $\langle\cdot,\cdot\rangle$; random elements are understood with respect to $\mathcal B(\mathcal H)$, which coincides with the $\sigma$-field generated by the evaluation maps $x\mapsto\langle x,f\rangle$, so an $\mathcal H$-valued random element is measurable as soon as all coordinates $\langle X,f\rangle$ are measurable, and its law is determined by these one-dimensional marginals \citep[Ch.~7]{hsing2015theoretical}. For $\mathbb E\|X\|<\infty$, the mean element is the Bochner integral $\mathbb E(X)=\int X\,dP$, characterized by $\langle\mathbb E(X),f\rangle=\mathbb E\langle X,f\rangle$; for $\mathbb E\|X\|^{2}<\infty$, the covariance operator $\Sigma=\mathbb E[(X-\mu)\otimes(X-\mu)]=\mathbb E[X\otimes X]-\mu\otimes\mu$ is positive, self-adjoint, and trace-class.

Let $X_1, X_2, \dots$ be i.i.d. random elements in $\mathcal{H}$ with mean element $\mathbb{E}(X)$, and let $\bar{X}_n = \frac{1}{n}\sum_{i=1}^n X_i$ denote the sample mean. The following result records the strong law of large numbers and the central limit theorem for random elements in Hilbert spaces.

\begin{theorem}
    \label{th:lln_clt_rkhs}
    If $\mathbb{E}(\| X \|) < \infty$, then
    \[
        \bar{X}_n \xrightarrow{a.s.} \mathbb{E}(X).
    \]
    If $\mathbb{E}(\| X \|^2) < \infty$, then
    \[
        \sqrt{n}(\bar{X}_n - \mathbb{E}(X)) \xrightarrow{d} \xi,
    \]
    where $\xi$ is a Gaussian random element in $\mathcal{H}$ with mean zero and covariance operator
    \[
        \Sigma = \mathbb{E}\left[(X - \mathbb{E}(X))\otimes (X - \mathbb{E}(X))\right].
    \]
\end{theorem}

The multiplier central limit theorem stated next is a key tool for establishing the consistency of the multiplier bootstrap procedure.

\begin{theorem}
    \label{th:multiplier_clt_rkhs}
    Let $\mathcal{H}$ be a separable Hilbert space (such as an RKHS). Let $X, X_1, X_2, \dots$ be i.i.d. random elements in $\mathcal{H}$ with $\mathbb{E}[\|X\|^2] < \infty$, mean element $\mu = \mathbb{E}[X]$, and covariance operator $\Sigma = \mathbb{E}[(X-\mu)\otimes(X-\mu)]$.

    Let $V, V_1, V_2, \dots$ be i.i.d. real-valued random variables independent of $X, X_1, X_2, \dots$, satisfying $\mathbb{E}[V] = 0$ and $\mathbb{E}[V^2] = 1$. Define the multiplier process
    $$ S_n = \frac{1}{\sqrt{n}} \sum_{i=1}^n V_i (X_i - \mu). $$
    Then, the conditional distribution of $S_n$ given the data $\mathbf{X}_n = (X_1, X_2, \dots, X_n)$ converges weakly to a Gaussian random element $\xi \sim \mathcal{N}(0, \Sigma)$ in $\mathcal{H}$ almost surely. Furthermore, the same conclusion holds if the true mean $\mu$ is replaced by the sample mean $\bar{X}_n = \frac{1}{n}\sum_{i=1}^n X_i$, yielding the bootstrap multiplier process $\tilde{S}_n = \frac{1}{\sqrt{n}} \sum_{i=1}^n V_i (X_i - \bar{X}_n)$.
\end{theorem}

\begin{proof}[Proof sketch]
Conditional fidi convergence follows from the conditional characteristic function
$\hat\phi_n(tf)=\prod_{j\le n}\phi_V(tY_j/\sqrt n)$ with $Y_j=\langle X_j-\mu,f\rangle$: the SLLN gives
$n^{-1}\sum_jY_j^{2}\to\sigma_f^{2}=\langle\Sigma f,f\rangle$ and $\max_jY_j^{2}/n\to0$ almost surely, so
$\log\hat\phi_n(tf)\to-\tfrac12t^{2}\sigma_f^{2}$ a.s.\ for each $t$ and each $f$ in a countable dense set of
directions, which suffices by separability. Conditional tightness follows from
$\mathbb E[\|(I-P_K)S_n\|^{2}\mid\mathbf X_n]=n^{-1}\sum_i\|(I-P_K)(X_i-\mu)\|^{2}
\to\operatorname{tr}((I-P_K)\Sigma)\to0$ as $K\to\infty$, together with the $K=0$ boundedness bound;
Prokhorov's theorem then identifies the conditional limit with $\mathcal N(0,\Sigma)$, Gaussian measures on
a separable Hilbert space being determined by their characteristic functionals. The sample-mean version
follows from $\tilde S_n=S_n-\bar V_n\sqrt n(\bar X_n-\mu)$ with
$\mathbb E[\|\bar V_n\sqrt n(\bar X_n-\mu)\|^{2}\mid\mathbf X_n]=\|\bar X_n-\mu\|^{2}=o_{a.s.}(1)$ and
Slutsky's theorem in $\mathcal H$. See \cite[Th.~2.9.2]{vaart1997weak} for this multiplier scheme in
generality.
\end{proof}

\section{Singular Value Decomposition of $\boldsymbol{R}_n$}
\label{app:svd_restriction}
\noindent The propositions of this section are quoted by the rotated-eigenbasis proof of Appendix~\ref{app:proofs_stat} and by the auxiliary lemmas of Appendix~\ref{app:proofs_nuisance}.
The spectral analysis of $\boldsymbol{R}_n$ rests on the following two propositions from \cite{rosasco2010regularization}.

\begin{proposition}\label{prop:population_operators}
    The following facts hold for the population operators.
\begin{enumerate}
    \item The operators $\boldsymbol{L}_k$ and $C$ are positive, self-adjoint, compact, and trace-class. Consequently, $\sigma(\boldsymbol{L}_k), \sigma(C) \subset [0,\bar{k}]$, where $\sigma(\cdot)$ denotes the spectrum.
    \item The non-zero spectra of $\boldsymbol{L}_k$ and $C$ coincide. If $\mu$ is a non-zero eigenvalue with associated eigenfunctions $\phi$ (in $L^2(\rho)$) and $e$ (in $\mathcal{H}$), normalized to unit norm, then
    \begin{equation*}
        \phi(x) = \mu^{-1/2} e(x) \text{ for } \mathbb{P}_X\text{-a.e. } x \in \mathcal{X}, \quad e(x) = \mu^{-1/2} \int_{\mathcal{X}} k(x,x') \phi(x') \, d\rho(x') \text{ for all } x \in \mathcal{X}.
    \end{equation*}
    \item The spectral decompositions are given by:
    \begin{equation*}
        \boldsymbol{L}_k = \sum_{i \ge 1} \mu_i \phi_i \otimes \phi_i, \quad C = \sum_{i \ge 1} \mu_i e_i \otimes e_i,
    \end{equation*}
    where $\{\phi_i\}_{i \ge 1}$ forms an orthonormal basis (ONB) of $\operatorname{Ker}(\boldsymbol{L}_k)^\perp$, and $\{e_i\}_{i \ge 1}$ forms an ONB of $\operatorname{Ker}(C)^\perp$.
\end{enumerate}
\end{proposition}

\begin{remark}
    Under Assumption~\ref{ass:A}, item 3 strengthens to: $\{e_i\}_{i \ge 1}$ is an ONB of all of
    $\HK$ and $\kerop(C)=\{0\}$. Indeed, $\{\sqrt{\mu_i}\phi_i\}_{i=1}^\infty$ is a complete orthonormal
    basis of $\HK$, and $C(\sqrt{\mu_i}\phi_i)=\sqrt{\mu_i}\,\boldsymbol{L}_k\phi_i
    =\mu_i(\sqrt{\mu_i}\phi_i)$, so each basis vector is a $C$-eigenvector with eigenvalue $\mu_i>0$;
    expanding any $f\in\kerop(C)$ in this basis gives $a_j\mu_j=0$ for every $j$, hence $f=0$. In
    particular, when $\HK$ is infinite-dimensional (e.g.\ under the universality of Assumption~\ref{ass:B}(i)), $C$ has infinitely many strictly positive eigenvalues.
\end{remark}

\begin{proposition}
    \label{svd_empirical}
The following facts hold for the empirical operators.
\begin{enumerate}
    \item The operator $C_n$ is finite rank, self-adjoint, and positive semi-definite. The matrix $K/n$ is symmetric and positive semi-definite. The spectrum $\sigma(C_n)$ contains finitely many non-zero elements, all in $[0,\bar{k}]$.
    \item The non-zero spectra of $K/n$ and $C_n$ coincide. If $\hat{\mu}$ is a non-zero eigenvalue with associated eigenvector $\boldsymbol{u}$ (in $\mathbb{R}^n$) and eigenfunction $\hat{e}$ (in $\mathcal{H}$), normalized to unit norm, 
    \begin{equation*}
        \boldsymbol{u} = (n\hat{\mu})^{-1/2} \bigl(\hat{e}(x_1), \dots, \hat{e}(x_n)\bigr)^\top, \quad \hat{e} = (n\hat{\mu})^{-1/2}  \sum_{i=1}^{n} \boldsymbol{u}_i k(x_i,\cdot),
    \end{equation*}
    where $\boldsymbol{u}_i$ is the $i$-th component of $\boldsymbol{u}$.
    \item The spectral decompositions are given by $K/n = \sum_{i=1}^{r} \nu_i \boldsymbol{u}_i \boldsymbol{u}_i^\top$ and $C_n = \sum_{i=1}^{r} \nu_i \hat{e}_i \otimes \hat{e}_i$,
    where $r = \operatorname{rank}(K)$, $\{\boldsymbol{u}_i\}_{i=1}^{r}$ forms an ONB of $\operatorname{Ker}(K)^\perp \subset \mathbb{R}^n$, and $\{\hat{e}_i\}_{i=1}^{r}$ forms an ONB of $\operatorname{Ker}(C_n)^\perp \subset \mathcal{H}$.
\end{enumerate}
\end{proposition}

\begin{remark}
From Proposition~\ref{svd_empirical}, under the scaled inner product $\langle a,b\rangle_{\mathbb{R}^n} = n^{-1}a^\top b$ on $\mathbb{R}^n$, the restriction operator and its adjoint admit finite-dimensional SVD representations: $\boldsymbol{R}_n = \sum_{i=1}^{r} \sigma_i \, \boldsymbol{u}_i  \hat{e}_i$ and $\boldsymbol{R}_n^\top = \sum_{i=1}^{r} (\sigma_i/n) \, \hat{e}_i  \boldsymbol{u}_i$, where $\sigma_i^2$ are the eigenvalues of the Gram matrix $\boldsymbol K$ (equivalently $\sigma_i = \sqrt{n\nu_i}$, so that $\nu_i = \sigma_i^2/n$ are the eigenvalues of $\boldsymbol K/n$ as in the main text) and $\boldsymbol u_i$ are the Euclidean-unit eigenvectors.
The factor $1/n$ in $\boldsymbol R_n^\top$ reflects the adjoint with respect to the scaled inner product; the $\sigma_i$ are the singular values of $\boldsymbol R_n$ with respect to the \emph{Euclidean} inner product on $\mathbb{R}^n$ (with respect to the scaled inner product the singular values are $\sqrt{\nu_i} = \sigma_i/\sqrt{n}$ and the left singular vectors are $\sqrt{n}\,\boldsymbol u_i$).
\end{remark}

\section{Useful Non-asymptotic Results}
\noindent The three results of this section are quoted repeatedly across Appendix~\ref{app:proofs}: the spectral-projection bound supports the orthogonal-transformations analysis of the next section; the bound on the empirical operators is invoked in the consistency, equivalence, bootstrap, and uniform-convergence proofs of Appendices~\ref{app:proofs_stat}--\ref{app:proofs_nuisance}; and the eigenvalue bounds enter the rotated-eigenbasis proof of Appendix~\ref{app:proofs_stat}.
\subsection{Finite Sample Probability Bound for Spectral Projection Estimation}
\begin{lemma}
\label{bound_projections}
Let $V$ be an integer and $v$ be the sum of the multiplicities of the first $V$ distinct eigenvalues of $\boldsymbol{L}_k$, such that $\mu^{(1)} \geq \mu^{(2)} \geq \dots \geq \mu^{(v)} > \mu^{(v+1)}$. Let $r = \operatorname{rank}(K)$ and $\{\boldsymbol{u}_i\}_{i=1}^{r}$ be the eigenvectors corresponding to the non-zero eigenvalues of $K$ in decreasing order. Denote by $\{\hat{e}_i\}_{i=1}^{r} \in \HK$ the corresponding Nystr\"om extensions (Proposition \ref{svd_empirical}).

Given $\tau > 0$, if the sample size $n$ satisfies $n > 128 \bar{k}^2 \tau/(\mu^{(v)} - \mu^{(v+1)})^2$, then with probability at least $1 - 2e^{-\tau}$, we have
\begin{equation*}
\left\|\hat{P}_v - P_v\right\|_{\mathrm{HS}}^2 \le \frac{32 \bar{k}^2 \tau}{n (\mu^{(v)} - \mu^{(v+1)})^2},
\end{equation*}
where $P_v = \sum_{i=1}^v e_i \otimes e_i$, $\hat{P}_v = \sum_{i=1}^v \hat{e}_i \otimes \hat{e}_i$, and $\| \cdot \|_{\mathrm{HS}}$ denotes the Hilbert--Schmidt norm.
\end{lemma}
\begin{proof}
    See Theorem 12 of \cite{rosasco2010regularization}.
\end{proof}

\subsection{Bound on the Second-Moment Operator}
\begin{lemma}
    \label{lm:C_and_Cn}
    The operators $C$ and $C_n$ are Hilbert-Schmidt. In addition, with confidence $1-2\exp(-\tau)$, we have
    \[ \|C - C_n\|_{HS} \le \frac{2\sqrt{2}\bar{k}\sqrt{\tau}}{\sqrt{n}}. \]
\end{lemma}
\begin{proof}
    See Theorem 7 of \cite{rosasco2010regularization}.
\end{proof}

\subsection{Bounds on the Eigenvalues}
\begin{lemma}
    \label{bound_eigenvalues}
    Let $\{\mu_i\}_{i \ge 1}$ and $\{\sigma^2_i\}_{i \ge 1}$ be the at most countable families of eigenvalues of $\boldsymbol{L}_k$ and $K$, respectively, where $\sigma^2_i = 0$ for $i > r$. Then, for any $\tau > 0$, with probability at least $1 - 2e^{-\tau}$, we have
    \begin{align*}
        \sum_{i \ge 1} \left(\mu_i - \sigma^2_i/n\right)^2 &\leq \frac{8 \bar{k}^2 \tau}{n}, \\
        \sup_{i \ge 1} \left|\mu_i - \sigma^2_i/n \right| &\leq \frac{2\sqrt{2}\bar{k}\sqrt{\tau}}{\sqrt{n}}.
    \end{align*}
\end{lemma}
\begin{proof}
    See Propositions 10 and 11 of \cite{rosasco2010regularization}.
\end{proof}

\section{Orthogonal Transformations}
\label{orthogonal_transformations}
\noindent The random orthogonal matrices constructed here are exactly those appearing in the rotated-eigenbasis expansion of Lemma~\ref{lemma:rotatedbasis}, proved in Appendix~\ref{app:proofs_stat}.
The alignment between the empirical and population subspaces ($\hat{E}$ and $E$) can be characterized explicitly via an orthogonal transformation. Define the $v \times v$ matrix $M$ by $M_{ij} = \langle e_i, \hat e_j\rangle_{\HK}$ (the inner product of the $i$-th population eigenfunction with the $j$-th empirical eigenfunction), and consider two functions $f \in \hat{E} \subset \HK$ and $g \in E \subset \HK$ expressed in the respective bases,
\[
    f = \sum_{i=1}^v v_i\,\hat e_i, \qquad g = \sum_{j=1}^v u_j\,e_j,
\]
normalized so that $\lVert v\rVert_2 = \lVert u\rVert_2 = 1$. The cosine of the angle between the two functions is their inner product, $\langle f,g\rangle_{\HK} = \sum_{i,j=1}^v v_iu_j\langle e_j,\hat e_i\rangle_{\HK} = \sum_{j=1}^v u_j(Mv)_j = u^\top M v$; hence the singular values $\{\tilde\sigma_i\}_{i=1}^v$ of $M$ are the cosines of the principal angles between $E$ and $\hat{E}$, and $\tilde\sigma_i^2 \in [0,1]$.

Let $P_v$ and $\hat P_v$ denote the orthogonal projections onto $E$ and $\hat E$. The matrix of the operator $\hat P_vP_v$ restricted to $\hat E$, in the empirical basis, is
\[
    \langle \hat e_i, \hat P_vP_v\,\hat e_j\rangle_{\HK} = \langle \hat e_i, P_v\hat e_j\rangle_{\HK} = \sum_{k=1}^v \langle \hat e_i, e_k\rangle_{\HK}\langle e_k,\hat e_j\rangle_{\HK} = (M^\top M)_{ij},
\]
so the eigenvalues of $M^\top M$ are precisely the squared singular values $\{\tilde\sigma_i^2\}_{i=1}^v$ of $M$. The squared Hilbert--Schmidt norm of the difference of the projections expands as
\begin{align*}
    \lVert P_v - \hat P_v\rVert_{\mathrm{HS}}^2
    &= \operatorname{Tr}\bigl((P_v-\hat P_v)^2\bigr)
     = \operatorname{Tr}(P_v) + \operatorname{Tr}(\hat P_v) - 2\operatorname{Tr}(P_v\hat P_v) \\
    &= 2v - 2\operatorname{Tr}(M^\top M) = 2\sum_{i=1}^v(1-\tilde\sigma_i^2),
\end{align*}
and Lemma~\ref{bound_projections} gives $\lVert P_v-\hat P_v\rVert_{\mathrm{HS}}^2 = O_p(n^{-1})$, so $\sum_{i=1}^v\tilde\sigma_i^2 = v - O_p(n^{-1})$.

Let $M = \boldsymbol U\Sigma\boldsymbol V^\top$ be the singular value decomposition, with $\boldsymbol U,\boldsymbol V \in \mathbb{R}^{v\times v}$ orthogonal and $\Sigma = \operatorname{diag}(\tilde\sigma_1,\ldots,\tilde\sigma_v)$. Construct the random orthogonal matrix $O_n := \boldsymbol U\boldsymbol V^\top$, so that $O_n^\top$ is the orthogonal Procrustes solution ($\arg\max_{O \in O(v)}\operatorname{Tr}(MO) = \boldsymbol V\boldsymbol U^\top$), and define the rotated empirical eigenfunctions $\tilde e_i := \sum_{j=1}^v(O_n)_{ij}\,\hat e_j$. Then
\[
    \langle e_i,\tilde e_j\rangle_{\HK} = \sum_{k=1}^v(O_n)_{jk}\langle e_i,\hat e_k\rangle_{\HK} = (MO_n^\top)_{ij},
\]
and $MO_n^\top = \boldsymbol U\Sigma\boldsymbol U^\top$ is symmetric positive semidefinite, so
\begin{align*}
    \sum_{i=1}^v\lVert e_i - \tilde e_i\rVert_{\HK}^2
    &= 2\bigl(v - \operatorname{Tr}(\boldsymbol U\Sigma\boldsymbol U^\top)\bigr)
     = 2\Bigl(v - \sum_{i=1}^v\tilde\sigma_i\Bigr) \\
    &\le 2\Bigl(v - \sum_{i=1}^v\tilde\sigma_i^2\Bigr) = O_p(n^{-1}),
\end{align*}
using $\tilde\sigma_i \in [0,1]$. Since $v$ is finite, $\lVert e_i-\tilde e_i\rVert_{\HK}^2 = O_p(n^{-1})$ for each $i$: the rotated empirical eigenfunctions are consistent estimators of the population eigenfunctions at the parametric rate. Finally, orthogonal transformations preserve the Euclidean norm, so the squared norm of the projection of any $f \in \HK$ onto the empirical subspace is invariant to the rotation:
\[
    \sum_{i=1}^v\langle f,\hat e_i\rangle_{\HK}^2 = \sum_{i=1}^v\langle f,\tilde e_i\rangle_{\HK}^2.
\]

\section{Additional Simulation Designs and Results}
\label{app:additional_sim}

Table~\ref{tab:size} collects the empirical sizes of the simulation study of Section~\ref{sec:simulation} at the $10\%$, $5\%$, and $1\%$ levels.

\begin{table}[htbp]
    \centering
    \footnotesize
    \caption{Empirical sizes: rejection rates under the null DGPs at the $10\%$, $5\%$, and $1\%$ levels}
    \label{tab:size}
    \begin{tabular}{@{}ll*{6}{c}@{}}
      \toprule
      & & \multicolumn{3}{c}{$d = 10$ ($\mathrm{DGP}_0$)} & \multicolumn{3}{c}{$d = 20$ ($\mathrm{DGP}_0$)} \\
      \cmidrule(lr){3-5} \cmidrule(lr){6-8}
      Method & $n$ & 10\% & 5\% & 1\% & 10\% & 5\% & 1\% \\
      \midrule
      \multicolumn{8}{c}{\textit{Projection tests under sample splitting}} \\
      $\hat T^{(1)}_{\mathrm{proj},\perp}$ & 200 & 0.095 & 0.035 & 0.005 & 0.217 & 0.071 & 0.005 \\
      $\hat T^{(2)}_{\mathrm{proj},\perp}$ & 200 & 0.096 & 0.037 & 0.005 & 0.212 & 0.071 & 0.005 \\
      $\hat T^{(1)}_{\mathrm{proj},\perp}$ & 400 & 0.098 & 0.037 & 0.006 & 0.105 & 0.042 & 0.008 \\
      $\hat T^{(2)}_{\mathrm{proj},\perp}$ & 400 & 0.101 & 0.039 & 0.006 & 0.106 & 0.045 & 0.009 \\
      \multicolumn{8}{c}{\textit{Projection tests, oracle protocol}} \\
      $\hat T^{(1)}_{\mathrm{proj},\perp}$ & 400$^{\dagger}$ & 0.097 & 0.039 & 0.007 & 0.135 & 0.053 & 0.012 \\
      $\hat T^{(2)}_{\mathrm{proj},\perp}$ & 400$^{\dagger}$ & 0.101 & 0.040 & 0.009 & 0.133 & 0.054 & 0.012 \\
      \multicolumn{8}{c}{\textit{KCM test}} \\
      $\gamma = 2$ & 200 & 0.000 & 0.000 & 0.000 & 0.000 & 0.000 & 0.000 \\
      $\gamma = 2$ & 400 & 0.000 & 0.000 & 0.000 & 0.000 & 0.000 & 0.000 \\
      $\gamma = 1$ & 200 & 0.000 & 0.000 & 0.000 & 0.000 & 0.000 & 0.000 \\
      $\gamma = 1$ & 400 & 0.000 & 0.000 & 0.000 & 0.000 & 0.000 & 0.000 \\
      $\gamma = 0.5$ & 200 & 0.000 & 0.000 & 0.000 & 0.000 & 0.000 & 0.000 \\
      $\gamma = 0.5$ & 400 & 0.000 & 0.000 & 0.000 & 0.000 & 0.000 & 0.000 \\
      $\gamma = 0.2$ & 200 & 0.055 & 0.018 & 0.001 & 0.003 & 0.000 & 0.000 \\
      $\gamma = 0.2$ & 400 & 0.057 & 0.016 & 0.002 & 0.000 & 0.000 & 0.000 \\
      $\gamma = 0.1$ & 200 & 0.100 & 0.042 & 0.008 & 0.077 & 0.007 & 0.000 \\
      $\gamma = 0.1$ & 400 & 0.074 & 0.033 & 0.006 & 0.026 & 0.009 & 0.000 \\
      $\gamma = 0.05$ & 200 & 0.105 & 0.057 & 0.009 & 0.164 & 0.045 & 0.003 \\
      $\gamma = 0.05$ & 400 & 0.089 & 0.033 & 0.011 & 0.096 & 0.034 & 0.006 \\
      $\gamma = 0.02$ & 200 & 0.105 & 0.056 & 0.011 & 0.198 & 0.076 & 0.012 \\
      $\gamma = 0.02$ & 400 & 0.088 & 0.038 & 0.013 & 0.112 & 0.047 & 0.006 \\
      median & 200 & 0.045 & 0.011 & 0.000 & 0.011 & 0.000 & 0.000 \\
      median & 400 & 0.047 & 0.011 & 0.001 & 0.004 & 0.000 & 0.000 \\
      \bottomrule
    \end{tabular}
    \begin{tablenotes}
        \footnotesize
        \item \textit{Note:} Rejection rates over 1000 replications. Under sample splitting, $n$ is the total sample size ($15\%$ selection, $85\%$ testing); $^{\dagger}$ oracle protocol with an additional selection sample of $200$ observations and a testing sample of $400$. The KCM test is evaluated at the kernel parameters $\gamma \in \{2, 1, 0.5, 0.2, 0.1, 0.05, 0.02\}$ and the median rule.
    \end{tablenotes}
\end{table}

\subsection{Sensitivity to the Sample-Splitting Ratio}
\label{app:split_ratio}

The main protocol fixes the split at $15\%$ selection / $85\%$ testing (Section~\ref{sec:cv_selection}). This subsection varies the splitting ratio while holding the \emph{testing} sample size fixed at $400$, so that the comparison across ratios isolates the effect of the split on the test itself rather than on its testing budget. For each testing share $t \in \{0.5, 0.6, 0.7, 0.8, 0.85, 0.9\}$, the total sample is the $N$ closest to $400/t$ and the selection sample has $N - 400$ observations, namely $400, 267, 171, 100, 71, 44$. The configuration $t = 0.85$ is the main protocol of Section~\ref{sec:simulation} and is marked by a dotted vertical line in Figure~\ref{fig:split_ratio}.

Figure~\ref{fig:split_ratio} separates the fixed alternatives ($\mathrm{DGP}_1$--$\mathrm{DGP}_5$), whose misspecification does not vanish with the sample size, from the local alternative ($\mathrm{DGP}_6$), whose misspecification $D = 1.5\,\lVert X\rVert^2/\sqrt{N}$ implements the total-rate Pitman shift of Section~\ref{sec:cv_selection}. For the local designs, the effective signal at testing share $t$ is $h_c = \sqrt{t}\,h$, attenuated by the testing share as in Section~\ref{sec:energy}. As $t \to 0$ the effective signal vanishes; as $t \to 1$ the test sees the full signal while the selector, asymptotically blind to the alternative, follows its null law. The limiting power is of order one throughout. Both endpoints lie outside the tested grid; panels (c) and (d) verify the order-one character and the intermediate behavior.

\paragraph{Size.}
The empirical size is robust to the splitting ratio. At the $5\%$ level the rejection rates under the null design range from $0.038$ to $0.055$ at $d = 10$ and from $0.039$ to $0.069$ at $d = 20$ ($\mathrm{DGP}_0$), within two to three simulation standard errors of the nominal level at every ratio. The mild conservatism at the largest testing shares ($t \ge 0.85$), where the CV choice rests on $44$--$71$ selection observations, mirrors the small-selection-sample effect documented for $n = 200$ in the main simulations; the $d = 20$ cell at $t = 0.5$ is mildly liberal ($0.069$). The invariance of size to the allocation of the sample is the empirical counterpart of Theorem~\ref{th:validity_sel}, which requires no selector limit and holds for any split.

\paragraph{Power under fixed alternatives.}
With the testing budget held fixed, the power of the fixed-alternative designs is largely insensitive to the split (panels (a) and (b) of Figure~\ref{fig:split_ratio}). The saturated designs are flat at every ratio ($\mathrm{DGP}_3$: $1.000$ at $d = 10$ and at least $0.996$ at $d = 20$), and the intermediate designs move by at most a few points between $t = 0.5$ (selection sample of $400$) and $t = 0.9$ (selection sample of $44$): $\mathrm{DGP}_1$: $0.950 \to 0.933$ at $d = 10$ and $0.397 \to 0.357$ at $d = 20$; $\mathrm{DGP}_2$: $0.195 \to 0.179$ at $d = 20$, while at $d = 10$ it is the only design whose power \emph{rises} with the testing share ($0.316 \to 0.398$); $\mathrm{DGP}_4$: $0.282 \to 0.283$ at $d = 10$ and $0.569 \to 0.522$ at $d = 20$. The designs for which the tuning choice matters most show the largest declines. The tuning is cheap in data: with $44$ selection observations the rejection rates are within a few points of the oracle protocol of Section~\ref{sec:simulation}, which spends $200$ selection observations at the same testing size ($\mathrm{DGP}_1$: $0.933$ vs.\ $0.915$ at $d = 10$; $\mathrm{DGP}_5$: $0.737$ vs.\ $0.747$ at $d = 10$ and $0.372$ vs.\ $0.387$ at $d = 20$), with one exception: $\mathrm{DGP}_4$ at $d = 20$, where the larger selection sample raises power from $0.522$ to $0.596$.

\paragraph{Power under local alternatives.}
The local design $\mathrm{DGP}_6$ implements the total-rate Pitman shift $D = 1.5\,\lVert X\rVert^2/\sqrt{N}$; the amplitude $1.5$ was chosen so that the limit power at the CV-selected configuration is non-degenerate and the split-ratio dependence is empirically resolvable without saturating the test. The observed $5\%$ rejection rates at the largest testing shares exceed those at $t = 0.5$: from $0.193$ at $t = 0.5$ to $0.247$ at $t = 0.9$ for $\mathrm{DGP}_6$ at $d = 10$, and from $0.205$ to $0.253$ at $d = 20$ ($\hat T^{(1)}_{\mathrm{proj},\perp}$). Across the grid the rates range over $0.19$--$0.30$ (panels (c) and (d) of Figure~\ref{fig:split_ratio}), moving in the direction of the signal attenuation: the effective signal $h_c = \sqrt{t}\,h$ increases in $t$, and by Proposition~\ref{prop:snr}(i) the pointwise power is non-decreasing in the captured-signal scale at fixed tuning. The rates at $t = 0.8$--$0.9$ exceed those at $t = 0.5$ in every cell, and the values are consistent with the oracle protocol at the same testing size ($\mathrm{DGP}_6$: $0.215$ at $d = 10$ and $0.251$ at $d = 20$).

\begin{figure}[htbp]
    \centering
    \includegraphics[width=\textwidth]{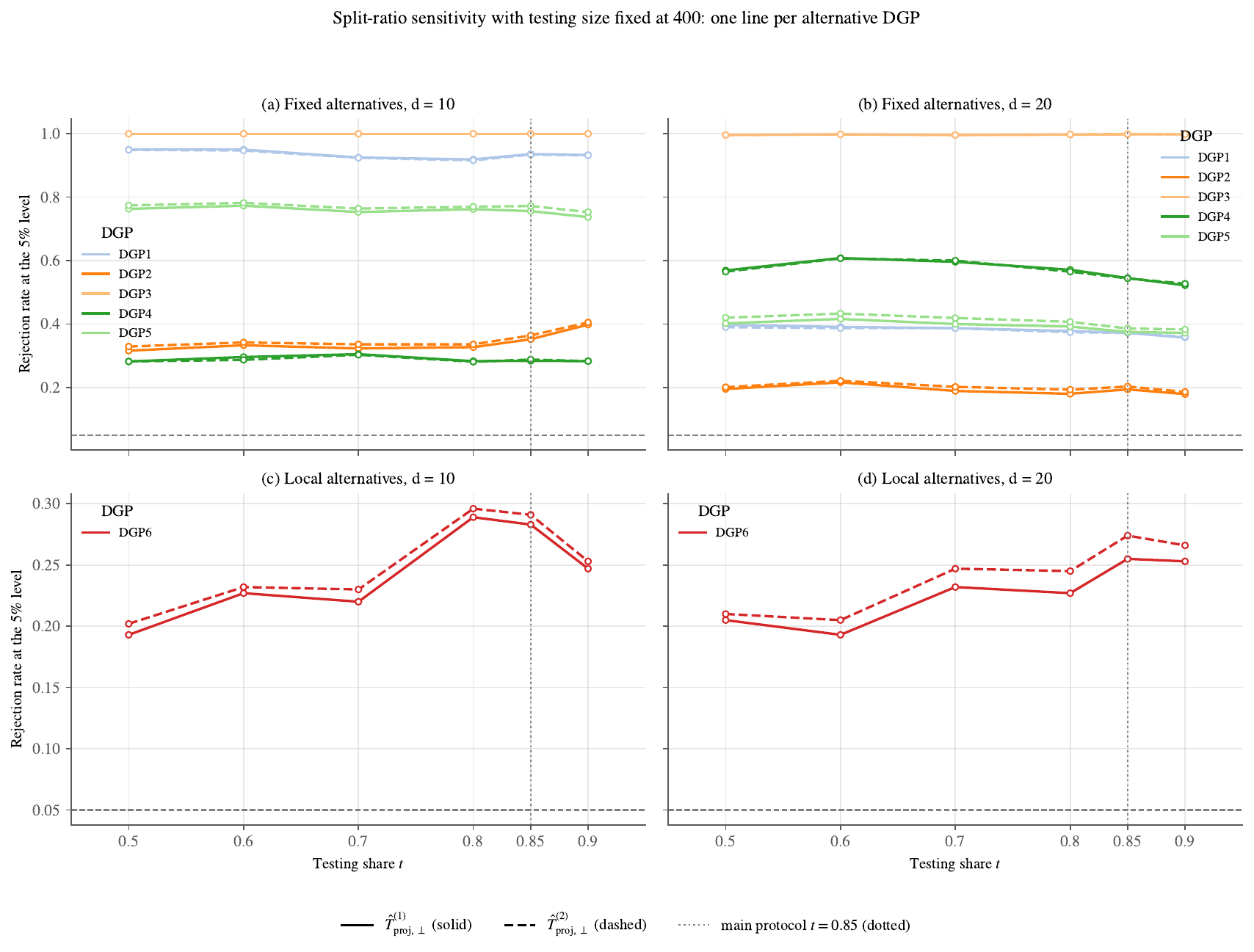}
    \caption{Rejection rates at the $5\%$ level for the projection tests $\hat T^{(1)}_{\mathrm{proj},\perp}$ (solid lines) and $\hat T^{(2)}_{\mathrm{proj},\perp}$ (dashed lines) with the testing sample size fixed at $400$ and the testing share $t$ on the horizontal axis, one line per alternative DGP: fixed alternatives ($\mathrm{DGP}_1$--$\mathrm{DGP}_5$) in panels (a) and (b), the local alternative ($\mathrm{DGP}_6$) in panels (c) and (d). The dotted vertical line marks the main protocol $t = 0.85$; the dashed horizontal line marks the nominal $5\%$ level.}
    \label{fig:split_ratio}
\end{figure}

\subsection{Numerical Verification of the Standing Monotonicity Condition}
\label{app:monotonicity_check}

Condition~\ref{sc:monotone} of Section~\ref{sec:energy} imposes a per-direction clause on the noise component of the selector's limit criterion, $D(\eta) := V(\eta) + \Delta(\eta)$; the strict decrease of $D$ in $\lambda$ on $\Upsilon$ is the consequence consumed by Theorem~\ref{th:cv_snr}(ii). This subsection derives the exact per-direction drift formula of which the clause of Condition~\ref{sc:monotone} is the termwise sufficient sign condition, proves that the clause implies the strict decrease consumed by Theorem~\ref{th:cv_snr}, exhibits the automaticity of the clause under homoscedasticity, and verifies the clause numerically for the Gaussian kernel against the simulation designs of Section~\ref{sec:simulation}, by empirical spectral approximation of the integral operator of $k_\gamma$ against $P_X = \mathcal{N}(0, I_d)$. On an iid sample $X_1, \ldots, X_N$ with Gram matrix eigendecomposition $K = U\operatorname{diag}(\nu)U^\top$, the empirical analogues of the drift components are
\[
    V(\eta) = \sum_{\ell \ge 1} w_\ell^2\,\tau_\ell, \qquad
    \Delta(\eta) = \sum_{\ell \ge 1} \big(w_\ell^2 - 2w_\ell\big)\,a_\ell,
    \qquad w_\ell := \frac{\mu_\ell}{\mu_\ell + \lambda},\ \ \mu_\ell := \frac{\nu_\ell}{N},
\]
with $\tau_\ell := u_\ell^\top\operatorname{diag}(\sigma^2(X_i))\,u_\ell$ and $a_\ell := u_\ell^\top C^* u_\ell/N$, where $C^*_{ij} := c^*(X_i, X_j)$ is the projection-correction kernel of Section~\ref{sec:nuisance}; the sample sizes are $N = 1200$ ($d = 10$) and $N = 1000$ ($d = 20$). The grid mirrors the selection protocol of Section~\ref{sec:cv_selection}: 20 values of $\gamma$ in $[\gamma_{\max}/100, \gamma_{\max}]$ with $\gamma_{\max} = 2/\mathrm{median}\{\lVert X_i - X_j\rVert^2\}$, and $\lambda \in [1, 10]$, checked on both the 10-point grid of the simulations and a 41-point refinement. The simulation designs of Section~\ref{sec:simulation} are conditional-mean models
\[
    Y = \alpha + X^\top\beta + D(X) + \sigma\,e, \qquad
    X \sim \mathcal{N}(0, I_d), \qquad e \perp\!\!\!\perp X, \qquad
    e \sim \mathcal{N}(0, 1),
\]
with $\alpha = 1$, $\beta = 0.5\,\mathbf{1}_d$, and $\sigma = 1$. The null design takes $D \equiv 0$ (the linear design $\mathrm{DGP}_0$ of Section~\ref{sec:simulation}), so that $\varepsilon_0 := Y - \alpha - X^\top\beta$ obeys $\mathbb{E}[\varepsilon_0 \mid X] = 0$ with conditional variance $\sigma^2(X) \equiv 1$, and the score of the residual with respect to $\theta = (\alpha, \beta^\top)^\top$ is $\nabla_\theta\varepsilon = -(1, X^\top)$, an $X$-measurable direction with second-moment matrix $\Gamma_{\theta_0} = \operatorname{diag}(1, I_d)$. The alternative designs add a misspecification $D$ (e.g.\ $D(x) = 1.5\exp(-(x^\top\beta)^2)$ or $D(x) = 2\cos(1.2\pi\lVert x\rVert)$); these enter only the alternative drift $A(\eta)$ of Theorem~\ref{th:cv_snr}(i) and play no role in the null drift of (ii) under verification. For the affine score, the projection-correction kernel of Section~\ref{sec:nuisance} reduces to the closed forms
\[
    c^*(x,x') = -(1 + x^\top x') \quad (\sigma^2 \equiv 1), \qquad
    c^*(x,x') = -\big(\sigma^2(x) + \sigma^2(x')\big)(1 + x^\top x') + 1 + \tfrac{d+1}{d}\,x^\top x',
\]
with
\[
  \sigma^2(x) = \tfrac{1 + \lVert x\rVert^2/d}{2},
\]
and for the location-model score $g_{\theta_0}(x) = -1$ simply $c^* \equiv -\sigma^2$ under homoscedasticity.

The drift $V(\eta) + \Delta(\eta)$ depends on the design only through $P_X$, $\sigma^2(\cdot)$, and the score span $g_{\theta_0}$; the check therefore covers eight configurations: $d \in \{10, 20\}$ with homoscedastic noise ($\sigma^2 \equiv 1$, the simulation setting) and heteroscedastic noise ($\sigma^2(x) = (1 + \lVert x\rVert^2/d)/2$), and four score spans: the simulation score $g_{\theta_0}(x) = -(1, x^\top)$, the location-model score $g_{\theta_0}(x) = -1$, the quadratic span $g_{\theta_0}(x) = -\big(1,\, x^\top,\, (\lVert x\rVert^2 - d)/\sqrt{2d}\,\big)$, and the no-nuisance case ($\Delta \equiv 0$).

\medskip\noindent\textit{An exact derivative formula and a structural condition.}\; The correction terms of Section~\ref{sec:uniform_limit} diagonalize in the Mercer basis: expanding $g_\eta(x,y) = \sum_{\ell\ge1}w_\ell(\lambda)\phi_\ell(x;\gamma)\phi_\ell(y;\gamma)$, $w_\ell(\lambda) := \mu_\ell(\gamma)/(\mu_\ell(\gamma)+\lambda)$, gives $\nu(\eta) = \sum_{\ell\ge1}w_\ell a_\ell$ and $\rho(\eta) = \sum_{\ell\ge1}w_\ell^2 a_\ell$ with $a_\ell(\gamma) := \mathbb{E}[c^*(X,X')\phi_\ell(X;\gamma)\phi_\ell(X';\gamma)]$ (the $X_1$-integration in $\rho$ is the identity, $\mathbb{E}[\phi_\ell(X_1)\phi_k(X_1)] = \delta_{\ell k}$), so that, together with $V(\eta) = \sum_{\ell\ge1}w_\ell^2\tau_\ell(\gamma)$,
\[
    D(\eta) = \sum_{\ell\ge1}\bigl[w_\ell^2\tau_\ell + (w_\ell^2-2w_\ell)a_\ell\bigr],
    \qquad
    \partial_\lambda D(\eta) = -2\sum_{\ell\ge1}\frac{\mu_\ell^2\tau_\ell - \lambda\mu_\ell a_\ell}{(\mu_\ell+\lambda)^3},
\]
where the two $a_\ell$-terms of the derivative combine into $2\lambda\mu_\ell a_\ell/(\mu_\ell+\lambda)^3$, giving the displayed numerator. Condition~\ref{sc:monotone} is the termwise sufficient condition for the sign of this display: if $a_\ell(\gamma) \le \mu_\ell(\gamma)\tau_\ell(\gamma)/\lambda_U$ for every $\ell$ with $a_\ell(\gamma) > 0$, then on $[\lambda_L, \lambda_U]$ each numerator $\mu_\ell(\gamma)\bigl(\mu_\ell(\gamma)\tau_\ell(\gamma) - \lambda\,a_\ell(\gamma)\bigr) \ge 0$, so $\partial_\lambda D(\eta) \le 0$ on $[\lambda_L, \lambda_U]$ with strict inequality on $[\lambda_L, \lambda_U)$ whenever some $\tau_\ell(\gamma) > 0$; hence $D(\eta)$ is strictly decreasing in $\lambda$ on $[\lambda_L, \lambda_U]$ for every $\gamma$, the minimizer conclusion of Theorem~\ref{th:cv_snr}(ii) holds, and the boundary comparison at $\lambda_L$ in Theorem~\ref{th:cv_snr}(iii) carries the strict sign. The per-direction clause is the structural form of this sufficient condition: the per-direction requirement $w_\ell\tau_\ell \ge (1-w_\ell)a_\ell$, binding at or before the upper endpoint $\lambda_U$ where $w_\ell$ is smallest. In words, the projection-correlation overlap $a_\ell$ of each Mercer direction with itself must not exceed the direction's own noise variance $\mu_\ell\tau_\ell$ divided by the largest admissible regularization. This per-direction condition is conservative: it is sufficient but not necessary, because directions with $a_\ell < 0$ contribute negative numerators $-\lambda\mu_\ell a_\ell$ that can offset directions with $a_\ell > 0$. The condition is automatic under homoscedasticity, $\sigma^2 \equiv \bar\sigma^2$: then $c^*(x,x') = -\bar\sigma^2 g_{\theta_0}(x)^\top\Gamma_{\theta_0}^{-1}g_{\theta_0}(x')$, hence $a_\ell = -\bar\sigma^2\,\mathbb{E}[\phi_\ell(X;\gamma)g_{\theta_0}(X)]^\top\Gamma_{\theta_0}^{-1}\mathbb{E}[\phi_\ell(X;\gamma)g_{\theta_0}(X)] \le 0$ for every $\ell$, and $D(\gamma,\cdot)$ is strictly decreasing for every kernel, every design, and every dimension. In the location model, $g_{\theta_0} \equiv \pm 1$, $c^*(x,x') = \bar\sigma^2 - \sigma^2(x) - \sigma^2(x')$ and $a_\ell = m_\ell(\bar\sigma^2 m_\ell - 2\mathbb{E}[\sigma^2\phi_\ell])$ with $m_\ell := \mathbb{E}[\phi_\ell]$: here $a_\ell$ can be positive, e.g.\ when $\sigma^2$ is small where $\phi_\ell$ is large, so the monotonicity is not automatic and the structural condition above is the operative sufficient condition. The large-$\lambda$ asymptotics of the derivative follow from the exact formula: by the Mercer expansion $\sum_{\ell\ge1}\mu_\ell(\gamma)a_\ell(\gamma) = \mathbb{E}[c^*(X_1,X_2)\,k_\gamma(X_1,X_2)]$, so $\partial_\lambda D(\eta) = 2\lambda^{-2}\,\mathbb{E}\big[c^*(X_1,X_2)\,k_\gamma(X_1,X_2)\big] + O(\lambda^{-3})$. Under homoscedasticity, $c^*(x,x') = -\bar\sigma^2 g_{\theta_0}(x)^\top\Gamma_{\theta_0}^{-1}g_{\theta_0}(x')$ is a negative semidefinite kernel (a Gram matrix of $\Gamma_{\theta_0}^{-1/2}g_{\theta_0}$), so $\mathbb{E}[c^*(X_1,X_2)k_\gamma(X_1,X_2)] \le 0$ for every score span and every dimension, and the leading term is nonpositive. In the location model $g_{\theta_0}\equiv\pm1$ the coefficient specializes to $\bar\sigma^2\mathbb{E}[p] - 2\mathbb{E}[\sigma^2 p]$, giving $\partial_\lambda D(\eta) = -2\bigl(2\mathbb{E}[\sigma^2 p] - \bar\sigma^2\mathbb{E}[p]\bigr)/\lambda^2 + O(\lambda^{-3})$ with $p(x) := \mathbb{E}[k_\gamma(x,X)] \ge 0$, which has a nonpositive leading term whenever $\bar\sigma^2 \le 2\min\sigma^2$ (as $2\mathbb{E}[\sigma^2 p] \ge 2\min\sigma^2\,\mathbb{E}[p]$); all configurations of Table~\ref{tab:monotonicity_check} satisfy $\bar\sigma^2 \le 2\min\sigma^2$. The threshold concerns only the leading term of the large-$\lambda$ expansion and does not by itself control $\partial_\lambda D$ on a finite interval $[\lambda_L, \lambda_U]$, where the $O(\lambda^{-3})$ remainder and the intermediate-$\lambda$ per-direction terms could in principle change the sign; the finite-$\lambda$ behavior for the designs considered is covered by the numerical checks of Table~\ref{tab:monotonicity_check}.

Table~\ref{tab:monotonicity_check} reports, for each configuration, the minimal relative gap between consecutive grid values,
\[
    \min_{\gamma,\,j} \frac{D(\gamma, \lambda_{j+1}) - D(\gamma, \lambda_j)}{D(\gamma, \lambda_j)},
\]
computed over all 20 grid values of $\gamma$ on the 41-point grid. In every configuration the drift is strictly decreasing in $\lambda$ at every $\gamma$: the minimal relative gap ranges from $-5.5\%$ to $-10.7\%$ (attained at the largest grid value of $\gamma$), the 10-point grid of the simulations yields gaps of at least $22\%$ in magnitude, and $D$ loses between $82\%$ and $99\%$ of its value over $\lambda \in [1, 10]$. Figure~\ref{fig:monotonicity_check} displays the drift as a function of $\lambda$ at the smallest, middle, and largest grid values of $\gamma$ for four representative configurations.

\begin{table}[htbp]
    \centering
    \footnotesize
    \caption{Numerical check of the standing monotonicity condition, in its per-direction clause form and through the strict decrease of $D(\gamma, \cdot)$ in $\lambda$ that it implies: minimal relative gap between consecutive $\lambda$-grid values over all 20 grid values of $\gamma$, on a 41-point refinement of $[1, 10]$.}
    \label{tab:monotonicity_check}
    \begin{tabular}{@{}lccclc@{}}
      \toprule
      Configuration & $d$ & $\sigma^2$ & Score span & Min.\ rel.\ gap & Strictly decreasing \\
      \midrule
      linear score (sims) & 10 & homosc.\ & $g(x) = (1, x)$ & $-5.5\%$ & yes \\
      linear score (sims) & 20 & homosc.\ & $g(x) = (1, x)$ & $-5.6\%$ & yes \\
      no nuisance & 10 & homosc.\ & none ($\Delta \equiv 0$) & $-10.7\%$ & yes \\
      no nuisance & 20 & homosc.\ & none ($\Delta \equiv 0$) & $-10.7\%$ & yes \\
      linear score & 10 & heterosc.\ & $g(x) = (1, x)$ & $-5.6\%$ & yes \\
      linear score & 20 & heterosc.\ & $g(x) = (1, x)$ & $-5.6\%$ & yes \\
      location model & 10 & homosc.\ & $g(x) = 1$ & $-5.5\%$ & yes \\
      quadratic span & 10 & homosc.\ & $g(x) = (1, x, \|x\|^2)$ & $-5.5\%$ & yes \\
      \bottomrule
    \end{tabular}
    \begin{tablenotes}
        \footnotesize
        \item \textit{Note:} $D(\eta) = V(\eta) + \Delta(\eta)$ with $\eta = (\gamma, \lambda)$; $\sigma^2$ homosc.\ means $\sigma^2 \equiv 1$, heterosc.\ means $\sigma^2(x) = (1 + \lVert x\rVert^2/d)/2$. ``Strictly decreasing'' holds if $D(\gamma, \lambda_{j+1}) < D(\gamma, \lambda_j)$ at every consecutive pair $(j, \gamma)$ on the grid. The spectral sample sizes are $N = 1200$ ($d = 10$) and $N = 1000$ ($d = 20$).
    \end{tablenotes}
\end{table}

\begin{figure}[htbp]
    \centering
    \includegraphics[width=\textwidth]{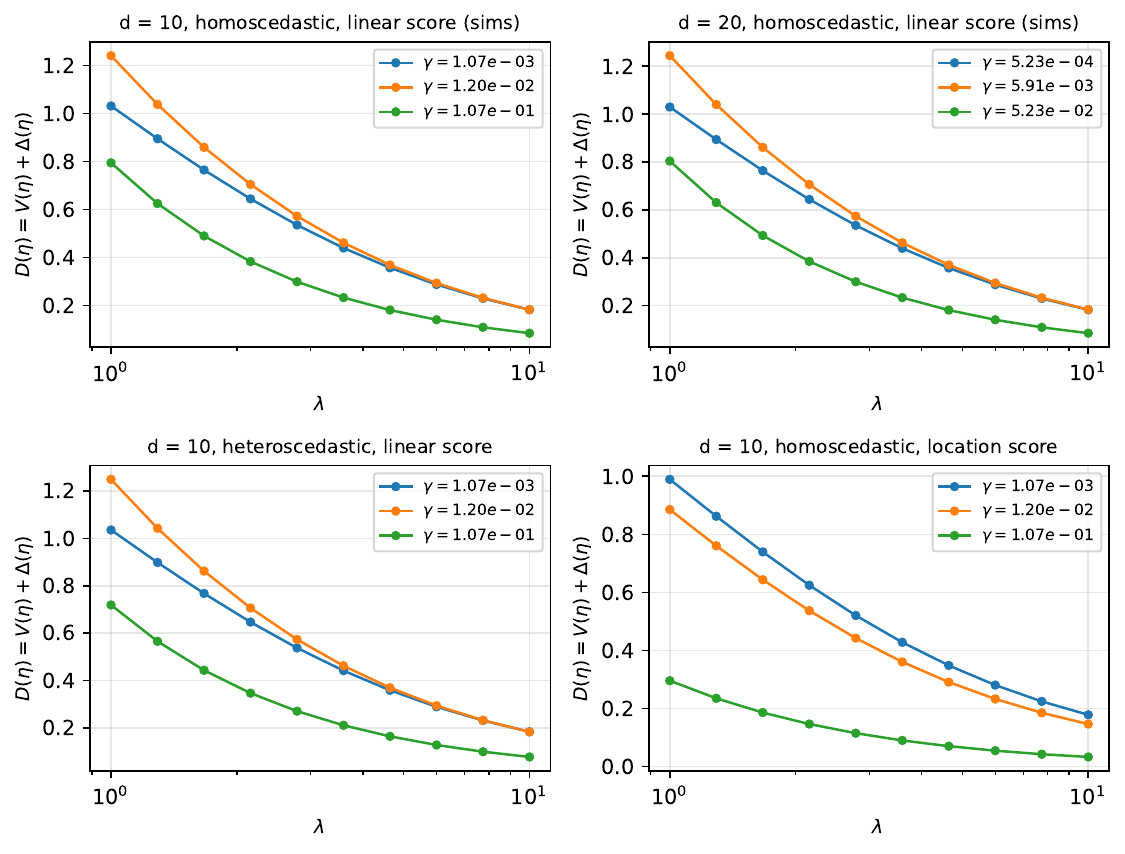}
    \caption{The drift $D(\eta) = V(\eta) + \Delta(\eta)$ as a function of $\lambda$ at the smallest, middle, and largest grid values of $\gamma$, for four configurations: $d = 10, 20$ with homoscedastic noise and the linear-regression score (the simulation designs), $d = 10$ with heteroscedastic noise, and the location model. The drift is strictly decreasing in $\lambda$ at every $\gamma$ in all configurations.}
    \label{fig:monotonicity_check}
\end{figure}

\end{document}